\documentclass[11pt]{article}
\usepackage{latexsym}
\usepackage{mathrsfs}
\usepackage{amsmath,amsthm}
\usepackage{graphicx,epstopdf,epsfig,multirow,epic,bm}
\usepackage{color,fancyvrb}
\usepackage{colortbl}
\usepackage{multicol}
\usepackage{amssymb}
\usepackage{CJK}

\usepackage{makeidx}
\usepackage{showidx}
\usepackage{textcomp}
\usepackage{paralist}

\theoremstyle{definition}
\newtheorem{thm}{Theorem}
\newtheorem{cor}[thm]{Corollary}
\newtheorem{lem}[thm]{Lemma}
\newtheorem{prop}[thm]{Proposition}
\theoremstyle{definition}
\newtheorem{defn}{Definition}
\theoremstyle{definition}
\newtheorem{rem}{Remark}

\theoremstyle{definition}
\newtheorem{problem}{Problem}

\theoremstyle{definition}
\newtheorem{conj}{Conjecture}
\theoremstyle{definition}
\newtheorem{example}{Example}
\theoremstyle{definition}

\DeclareGraphicsExtensions{.eps,.eps.gz}
\normalsize \rm
\usepackage{titlesec}

\newcommand{\qqed}{\hfill $\square$}

\usepackage[nottoc]{tocbibind} 

\begin{document}


\thispagestyle{empty}

\begin{center}
{\huge \textbf{Graph Set-colorings And Hypergraphs\\[12pt]
In Topological Coding}}\\[14pt]
{\Large \textbf{Bing Yao \footnote{~College of Mathematics and Statistics, Northwest Normal University, Lanzhou, 730070, CHINA \\ email: yybb918@163.com} and Fei Ma \footnote{~School of Computer Science, Northwestern Polytechnical University, Xi'an, 710072, CHINA \\ email: mafei123987@163.com}}\\[12pt]
(\today)}
\end{center}

\vskip 1cm

\begin{quote}
\textbf{Abstract:} We define new set-colorings: parameterized set-coloring, hyperedge-set coloring, distinguishing set-coloring, hypergraph-group coloring, \emph{etc.} We try to study parameterized hypergraph, hypergraph homomorphism, graphic groups based on hypergraphs, and to construct hypergraphs. We strength algebraic means in researching set-colorings and hypergraphs: (i) topological groups including graphic groups, graphic group homomorphisms, matrix groups, string groups, mixed-graphic groups, hypergraph groups, pan-groups, \emph{etc.}; (ii) topological lattices; and (iii) topological homomorphisms. We believe that exploring hypernetworks and their applications is important in mathematical theory and practical application. And we are aiming to apply the techniques of topology code theory in this article to (1) encrypt a network as a whole (homomorphic topology encryption and asymmetric topology cryptograph); (2) seek solutions for some difficult problems of graph theory; (3) investigate graph networks from DeepMind and GoogleBrain, since which generalizes and extends various approaches for neural networks that operate on graphs, and provides a straightforward interface for manipulating structured knowledge and producing structured behaviors for artificial intelligence.
\end{quote}

\newpage

\pagenumbering{roman}
\tableofcontents

\newpage

\setcounter{page}{1}
\pagenumbering{arabic}

\section{Introduction}

\subsection{Researching background}

In the coming \emph{Quantum Computer Era}, we will face with the following information security challenges:
\begin{asparaenum}[$\bullet$ ]
\item The Shor algorithm can completely destroy the encryption mechanism based on RSA and elliptic curve cryptography as long as the quantum computer has enough logical qubits to perform operations. As known, Shor algorithm can effectively attack RSA, EIGamal, ECC public-key cryptography and DH key agreement protocols which are widely used at present. This indicates that RSA, EIGamal, ECC public-key cryptography and DH key agreement protocols will no longer be secure in the quantum computing environment.
\item There is also an algorithm called Grover, which can completely weaken AES encryption from 128 bits to 64 bits, and then it can be cracked by ordinary computer algorithms.
\item We are in a digital age, and are facing important researching topics of coming quantum computation, lattices and cryptography, privacy computation, privacy computation and hypergraphs and hypernetworks.
\item ChatGPT-series and Sora (AGI) by OpenAI, \emph{etc.} The attack of artificial intelligence equipped with quantum computing on information security will become even crazier.
\item Artificial intelligences occupy every corner of the world. Demis Hassabis, CEO of Google DeepMind, said: In the coming years, artificial intelligence - ultimately general artificial intelligence - may become one of the driving forces behind the greatest social, economic, and scientific changes in history.
\end{asparaenum}

\subsubsection{Information security in the era of quantum computers}

In 2016 the National Institute of Standards and Technology has initiated a standardization procedure for post-quantum cryptosystems. Such cryptosystems are usually based on NP-complete problems for two reasons: NP-complete problems are at least as hard as the hardest problems in NP, but solutions of such problems can be verified efficiently. As research on quantum computers advances, the cryptographic community is searching for cryptosystems that will survive attacks on quantum computers. This area of research is called \emph{post-quantum cryptography}. The main candidates for post-quantum cryptography are:
\begin{asparaenum}[$\bullet$ ]
\item \textbf{Code-based cryptography} is based on the NP-complete problem of decoding a random linear code.
\item \textbf{Lattice-based cryptography} is based on Conjectured security against quantum attacks; Algorithmic simplicity, efficiency, and parallelism; Strong security guarantees from worst-case hardness; NP-complete problems of finding the shortest vector.
\item \textbf{Multivariate cryptography} is based on the NP-complete problem of solving multivariate quadratic equations defined over some finite field.
\item \textbf{Isogeny-based cryptography} is based on finding the isogeny map between two super-singular elliptic curves.
\end{asparaenum}

Notice that the lattice difficulty problem is not only a classical number theory, but also an important research topic of computational complexity theory. Researchers have found that lattice theory has a wide range of applications in cryptanalysis and design. Many difficult problems in lattice have been proved to be NP-hard. So, this kind of cryptosystems are generally considered to have the characteristics of quantum attack resistance (Ref. \cite{Wang-Xiao-Yun-Liu-2014}).

Peikert, in \cite{Chris-Peikert-decade}, pointed: ``\emph{Lattice-based ciphers have the following advantages: Conjectured security against quantum attacks; Algorithmic simplicity, efficiency, and parallelism; Strong security guarantees from worst-case hardness.}''

\subsubsection{Hypergraphs in the era of post-quantum encryption}

As known, all things of high-dimensional data sets are interrelated and interact on each other, we need to study the complex structures of high-dimensional data sets, and the interaction between high-dimensional data sets, one of research tools is \textbf{\emph{hypergraph}}. Since the hypergraph theory was proposed systematically by Claude Berge in 1973, more and more attention has been paid to the research of hypergraph theory and its application.

Hypergraphs can tease out of big data sets proposed in ``\emph{How Big Data Carried Graph Theory Into New Dimensions}'' by Stephen Ornes \cite{how-big-data-graph-theory-20210819}. And large data sets in practical application show that the impact of groups often exceeds that of individuals, thereby, it is more and more important to study hypergraphs.

As a subset system of a finite set, hypergraph is the most general discrete structure, which is widely used in information science, life science and other fields. However, hypergraphs are more difficult to draw on paper than graphs, there are methods for the visualization of hypergraphs, such as Venn diagram, PAOH \emph{etc}. Professor Wang, in his book tilted ``\emph{Information Hypergraph Theory}'' \cite{Jianfang-Wang-Hypergraphs-2008}, has investigated the structure of vertex-intersected graphs of hypergraphs. He said: ``\emph{When computers become very powerful, the security theory of information science requires hypergraphs to procedure and protect information data}.''

\textbf{Some applications of hypergraphs are:}

Undirected hypergraphs are useful in modelling such things as satisfiability problems \cite{Uriel-Kim-Han-Eran-2006}, databases \cite{Beeri-Fagin-Maier-Yannakakis-1983}, machine learning \cite{Huang-Zhang-Yu-Jeffrey-Xu-2015}, and Steiner tree problems \cite{Brazil-M-Zachariasen-2015}. They have been extensively used in machine learning tasks as the data model and classifier regularization (mathematics) \cite{Zhou-Dengyong-Huang-Jiayuan-Scholkopf-Bernhard-2006}. The applications include recommender system (communities as hyperedges) \cite{Tan-Bu-Chen-Xu-Wang-He-2013}, image retrieval (correlations as hyperedges) \cite{Liu-Qingshan-Huang-Metaxas-Dimitris-2013}, and bioinformatics (biochemical interactions as hyperedges) \cite{Patro-Rob-Kingsoford-Carl-2013}. Representative hypergraph learning techniques include hypergraph spectral clustering that extends the spectral graph theory with hypergraph Laplacian \cite{Gao-Wang-Zha-Shen-Li-Wu-Xindong-2013}, and hypergraph semi-supervised learning that introduces extra hypergraph structural cost to restrict the learning results \cite{Tian-Hwang-TaeHyun-Kuang-Rui-2009}. For large scale hypergraphs, a distributed framework \cite{Tan-Bu-Chen-Xu-Wang-He-2013} built using Apache Spark is also available.

Directed hypergraphs can be used to model things including telephony applications \cite{Goldstein-A-1982}, detecting money laundering \cite{Ranshous-Stephen-more-2017}, operations research \cite{Ausiello-Giorgio-Laura-Luigi-2017}, and transportation planning. They can also be used to model Horn-satisfiability \cite{Gallo-Longo-Pallottino-Nguyen-1993}. Directed hypergraphs can be used to model things including telephony applications, detecting money laundering, operations research, and transportation planning, and can also be used to model Horn-satisfiability.

\subsubsection{An example of topology code theory}

Topological coding is a mathematical subbranch based on graph theory, algebra, probability and combinatorics, \emph{etc.} Many techniques of topology code theory can be used in asymmetric cryptography and anti-quantum computing. As a brief introduction to the encryption of topology encoding, we show an example as follows:

\begin{example}\label{exa:topological-signatures-set-coloring}
There are four set-colored graphs (also, four \emph{topological signatures} admitting set-colorings) in Fig.\ref{fig:11-example}, in which the set-colored graph $G_1$ admits a set-coloring $\theta_1$, and it corresponds its own Topcode-matrix $T_{code}(G_1,\theta_1)$ shown in Eq.(\ref{eqa:g1-topcode-matrices-set-colorings}), where the topological structure of each $G_i$ is $H_1$ shown in Fig.\ref{fig:00-example}(a).
\begin{equation}\label{eqa:g1-topcode-matrices-set-colorings}
\centering
{
\begin{split}
T_{code}(G_1,\theta_1)=\left(
\begin{array}{ccccccccc}
\{1\} & \{2,4\} & \{3,6\} & \{2,4\} & \theta_1(x_1) & \{3,6\} & \theta_1(x_1) & \theta_1(x_1) & \theta_1(x_1)\\
\{1\} & \{2\} & \{3\} & \{4\} & \{5\} & \{6\} & \{7\} & \{8\} & \{9\}\\
\theta_1(y_1) & \theta_1(y_1) & \theta_1(y_1) & \{4,7\} & \theta_1(y_1) & \{6,8\} & \{4,7\} & \{6,8\} & \{9\}
\end{array}
\right)
\end{split}}
\end{equation}
where $\theta_1(x_1)=\{5,7,8,9\}$ and $\theta_1(y_1)=\{1,2,3,5\}$ (Ref. Definition \ref{defn:total-coloring-Topcode-matrixs}).

\begin{figure}[h]
\centering
\includegraphics[width=16.4cm]{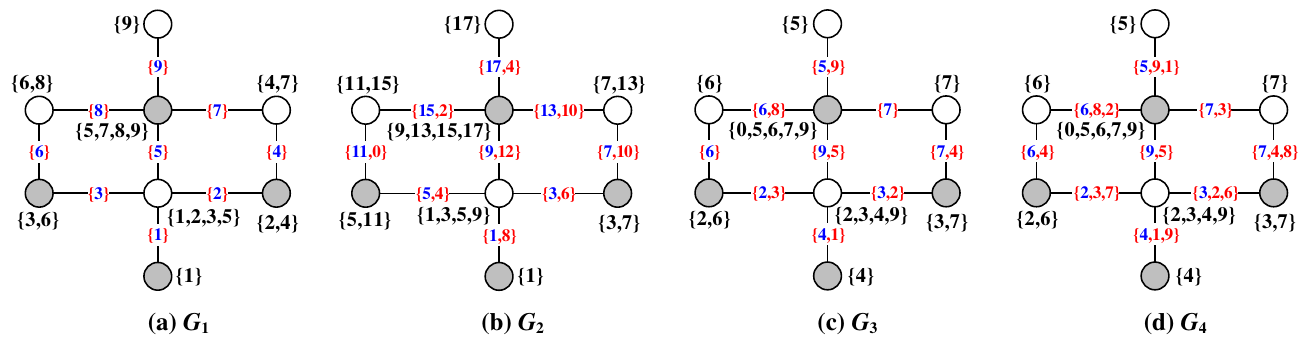}\\
\caption{\label{fig:11-example}{\small Four topological signatures made by the graphs admitting set-colorings based on $H_1$ shown in Fig.\ref{fig:00-example}(a).}}
\end{figure}

Let $R\,^i_{est}(R_{i,1},R_{i,2},\dots ,R_{i,k})$ be a constraint set with $k$ constraints. Moreover, about examples shown in Fig.\ref{fig:11-example}, we have:

\textbf{A.} Let
$${
\begin{split}
\mathcal{E}_1=&\big \{\{5,7,8,9\}, \{1,2,3,5\},\{1\}, \{2\}, \{3\}, \{4\}, \{5\}, \{6\},\\
& \, \{7\}, \{8\}, \{9\},\{2,4\}, \{3,6\}, \{2,4\}, \{6,8\}, \{4,7\}\big \}
\end{split}}
$$ be a hyperedge set based on a consecutive integer set $\Lambda_1=\{1,2,\dots ,9\}=[1,9]$. The connected $(8,9)$-graph $G_1$ shown in Fig.\ref{fig:11-example}(a) admits a \emph{graceful intersection set-coloring} $\theta_1:V(G_1)\cup E(G_1)\rightarrow \mathcal{E}_1$ subject to a constraint set $R\,^1_{est}(c_{A,1},c_{A,2})$ containing the following constraints:

$\mathbf{c_{A,1}}$: Each edge $uv\in E(G_1)$ is colored with an edge color set $\theta_1(uv)=\theta_1(u)\cap \theta_1(v)\neq \emptyset$.

$\mathbf{c_{A,2}}$: The edge color set $\theta_1(E(G_1))=[1,9]$, that is \emph{the graceful constraint}.

The set $\mathcal{E}_1$ is a hyperedge set holding $[1,9]=\bigcup _{e_i\in \mathcal{E}_1}e_i$ true, so $G_1$ is a vertex-intersected graph of a hypergraph $\mathcal{H}_{yper}=([1,9],\mathcal{E}_1)$ (Ref. Definition \ref{defn:set-hyperedge-sets}).

\textbf{B.} Let
$${
\begin{split}
\mathcal{E}_2=&\big \{\{17\}, \{11,15\},\{9,13,15,17\}, \{7,13\}, \{5,11\}, \{1,3,5,9\}, \{3,7\}, \{1\},\\
& \, \{17,4\}, \{15,2\}, \{13,10\},\{11,0\}, \{9,12\}, \{7,10\}, \{5,4\}, \{3,6\}, \{1,8\}\big \}
\end{split}}
$$ be a hyperedge set based on a consecutive integer set $\Lambda_2=\{1,2,\dots ,17\}\setminus \{12,14,16\}$. The connected $(8,9)$-graph $G_2$ shown in Fig.\ref{fig:11-example}(b) admits an \emph{odd-graceful intersection set-coloring} $\theta_2:V(G_2)\cup E(G_2)\rightarrow \mathcal{E}_2$ subject to a constraint set $R\,^2_{est}(c_{B,1},c_{B,2},c_{B,3})$ containing the following constraints:

$\mathbf{c_{B,1}}$: Each edge $xy\in E(G_2)$ is colored with an edge color set $\theta_2(xy)\supseteq \theta_2(x)\cap \theta_2(y)\neq \emptyset$;

$\mathbf{c_{B,2}}$: There is an odd-integer set $\{1,3,5, \dots ,17\}=[1,17]^o$ from edge color set $\theta_2(xy)$ for each edge $xy\in E(G_2)$, that is \emph{the odd-graceful constraint}.

$\mathbf{c_{B,3}}$: There are integers $a_{xy}\in \theta_2(xy)$, $a_{x}\in \theta_2(x)$ and $a_{y}\in \theta_2(y)$ holding each integer $a_{xy}=|a_{x}-a_{y}|$ to be odd.

\textbf{C.} Let
$${
\begin{split}
\mathcal{E}_3=&\big \{\{5\}, \{6\},\{0,5,6,7,9\}, \{7\}, \{2,6\}, \{2,3,4,9\}, \{3,7\}, \{4\},\\
& \, \{5,9\}, \{6,8\}, \{7\},\{6\}, \{9,5\}, \{7,4\}, \{2,3\}, \{3,2\}, \{4,1\}\big \}
\end{split}}
$$ be a hyperedge set based on a consecutive integer set $\Lambda_3=\{0,1,2,\dots ,9\}=[0,9]$. The connected $(8,9)$-graph $G_3$ shown in Fig.\ref{fig:11-example}(c) admits an \emph{edge-magic total intersection set-coloring} $\theta_3$ from $V(G_3)\cup E(G_3)$ to $\mathcal{E}_3$ subject to a constraint set $R\,^3_{est}(c_{C,1},c_{C,2})$ containing the following constraints:

$\mathbf{c_{C,1}}$: Each edge $uv\in E(G_3)$ is colored with an edge color set $\theta_3(uv)\supseteq \theta_3(u)\cap \theta_3(v)\neq \emptyset$.

$\mathbf{c_{C,2}}$: There are some integers $b_{uv}\in \theta_3(uv)$, $b_{u}\in \theta_3(u)$ and $b_{v}\in \theta_3(v)$ for each edge $uv\in E(G_3)$, such that the \emph{edge-magic constraint} $b_{u}+b_{uv}+b_{v}=14$ holds true.

\textbf{D.} Let
$${
\begin{split}
\mathcal{E}_4=&\big \{\{5\}, \{6\},\{0,5,6,7,9\}, \{7\}, \{2,6\}, \{2,3,4,9\}, \{3,7\}, \{4\},\\
& \, \{5,9,1\}, \{6,8,2\}, \{7,3\},\{6,4\}, \{9,5\}, \{7,4,8\}, \{2,3,7\}, \{3,2,6\}, \{4,1,9\}\big \}
\end{split}}
$$ be a hyperedge set based on a consecutive integer set $\Lambda_4=\{0,1,2,\dots ,9\}=[0,9]$. The connected $(8,9)$-graph $G_4$ shown in Fig.\ref{fig:11-example}(d) admits an \emph{edge-magic total intersection set-coloring} $\theta_4$ from $V(G_4)\cup E(G_4)$ to $\mathcal{E}_4$ subject to a constraint set $R\,^4_{est}(c_{D,1},c_{D,2},c_{D,3},c_{D,4})$ consisted of the following constraints:

$\mathbf{c_{D,1}}$: The color set of each edge $uv\in E(G_4)$ holds $\theta_4(uv)\supseteq \theta_4(u)\cap \theta_4(v)\neq \emptyset$.

$\mathbf{c_{D,2}}$: Each edge $uv\in E(G_4)$ holds the \emph{edge-magic constraint} $R_{u}+R_{uv}+R_{v}=14$ for some integers $R_{uv}\in \theta_4(uv)$, $R_{u}\in \theta_4(u)$ and $R_{v}\in \theta_4(v)$.

$\mathbf{c_{D,3}}$: Each edge $uv\in E(G_4)$ holds the \emph{felicitous-difference constraint} $\big ||d_{u}-d_{v}|-d_{uv}\big |=4$ for some integers $d_{uv}\in \theta_4(uv)$, $d_{u}\in \theta_4(u)$ and $d_{v}\in \theta_4(v)$.

$\mathbf{c_{D,4}}$: $\{\textrm{some }e_i\in \theta_4(uv): uv\in E(G_4)\}=[1,9]$.\qqed
\end{example}

\begin{rem}\label{rem:333333}
The examples shown in Fig.\ref{fig:11-example} enable us to obtain:

\textbf{(1) Number-based strings.} The Topcode-matrix $T_{code}(G_1,\theta_1)$ shown in Eq.(\ref{eqa:g1-topcode-matrices-set-colorings}) can produce the following number-based strings
\begin{equation}\label{eqa:two-number-based-strings}
{
\begin{split}
s_{1,public}&=1243624~5789~36~578957895789~123456789~123512351235~47~1235~6847689\\
s_{1,private}&=9867486~5321~74~532153215321~987654321~987598759875~63~9875~4263421
\end{split}}
\end{equation}
as a pair of keys. We can get
$$(27!)\times 294912=3211258267361780000000000000000000~(>2^{111})$$ number-based strings generated from the Topcode-matrix $T_{code}(G_1,\theta_1)$, like two number-based strings $s_{1,public}$ and $s_{1,private}$ shown in Eq.(\ref{eqa:two-number-based-strings}), each of these number-based strings has 57 bytes.

\begin{thm}\label{thm:666666}
$^*$ The number-based strings generated from the Topcode-matrix $T_{code}$ can be classified into two kinds $S_{public}$ and $S_{private}$, such that each number-based string $s\in S_{public}$ corresponds to a unique number-based string $s\,'\in S_{private}$, and they have the same cardinality $|S_{public}|=|S_{private}|$.
\end{thm}

\textbf{(2) The uniqueness of topological signatures.} There are five topological structures $H_1,H_2,H_3,H_4,H_5$ shown in Fig.\ref{fig:00-example}, such that each graph $H_i$ of them corresponds its own Topcode-matrix $T_{code}(H_i)=T_{code}$ shown in Eq.(\ref{eqa:general-topology-code-matrices}).

\begin{equation}\label{eqa:general-topology-code-matrices}
\centering
{
\begin{split}
T_{code}=\left(
\begin{array}{ccccccccc}
x_4 & x_3 & x_2 & x_3 & x_1 & x_2 & x_1 & x_1 & x_1\\
x_4y_1 & x_3y_1 & x_2y_1 & x_3y_2 & x_1y_1 & x_2y_3 & x_1y_2 & x_1y_3 & x_1y_4\\
y_1 & y_1 & y_1 & y_2 & y_1 & y_3 & y_2 & y_3 & y_4
\end{array}
\right)
\end{split}}
\end{equation}

\begin{figure}[h]
\centering
\includegraphics[width=16.4cm]{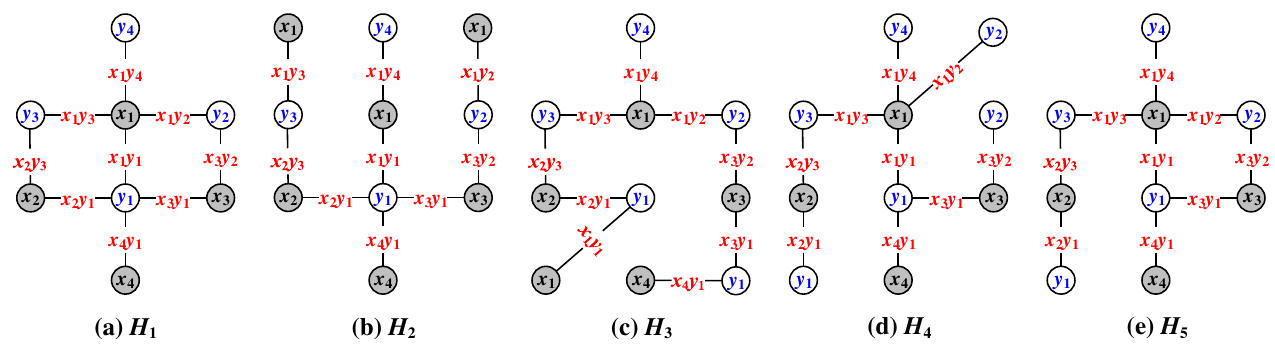}\\
\caption{\label{fig:00-example}{\small Different topological structures.}}
\end{figure}

However, two topological structures $H_i$ and $H_j$ with $i\neq j$ are not isomorphic from each other, namely, $H_i\not\cong H_j$, since two adjacent matrices $A(H_i)$ and $A(H_j)$ are not similar from each other, even two adjacent matrices $A(H_i)$ and $A(H_j)$ have the same order, but there is no a reversible matrix $P$ holding $A(H_i)=PA(H_i)P^{-1}$. This property is just the uniqueness of topological signatures for ensuring the security and uniqueness of identity authentication in practical application scenarios.

\begin{equation}\label{eqa:adjacent-matrices}
\centering
{
\begin{split}
A(H_1)=\left(
\begin{array}{c|cccccccc}
* & x_1 & x_2 & x_3 & x_4 & y_1 & y_2 & y_3 & y_4\\
\hline
x_1 & 0 & 0 & 0 & 0 & 1 & 1 & 1 & 1\\
x_2 & 0 & 0 & 0 & 0 & 1 & 0 & 1 & 0\\
x_3 & 0 & 0 & 0 & 0 & 1 & 1 & 0 & 0\\
x_4 & 0 & 0 & 0 & 0 & 1 & 0 & 0 & 0\\
y_1 & 1 & 1 & 1 & 1 & 0 & 0 & 0 & 0\\
y_2 & 1 & 0 & 1 & 0 & 0 & 0 & 0 & 0\\
y_3 & 1 & 1 & 0 & 0 & 0 & 0 & 0 & 0\\
y_4 & 1 & 0 & 0 & 0 & 0 & 0 & 0 & 0
\end{array}
\right)=\left(
\begin{array}{ccccccccc}
0 & 0 & 0 & 0 & 1 & 1 & 1 & 1\\
0 & 0 & 0 & 0 & 1 & 0 & 1 & 0\\
0 & 0 & 0 & 0 & 1 & 1 & 0 & 0\\
0 & 0 & 0 & 0 & 1 & 0 & 0 & 0\\
1 & 1 & 1 & 1 & 0 & 0 & 0 & 0\\
1 & 0 & 1 & 0 & 0 & 0 & 0 & 0\\
1 & 1 & 0 & 0 & 0 & 0 & 0 & 0\\
1 & 0 & 0 & 0 & 0 & 0 & 0 & 0
\end{array}
\right)_{8\times 8}
\end{split}}
\end{equation}
Other adjacent matrices are $A(H_2)_{10\times 10}$, $A(H_3)_{10\times 10}$, $A(H_4)_{10\times 10}$ and $A(H_5)_{9\times 9}$.

\vskip 0.4cm

Each set-colored graph $G_i$ with $i\in [1,4]$ shown in Fig.\ref{fig:11-example} corresponds four set-colored graphs $G_{i,j}$ holding $G_{i,j}\cong H_j$ for $j\in [2,5]$ shown in Fig.\ref{fig:00-example}, such that $T_{code}(G_i,\theta_i)=T_{code}(G_{i,j},\theta_{i,j})$ for $i\in [1,4]$ and $j\in [2,5]$, although two topological signatures $G_{i,j}\not\cong G_{i,s}$ for $j,s\in [2,5]$ and $j\neq s$.

\textbf{(3) The mixed set-colorings.} By Fig.\ref{fig:11-example} and Fig.\ref{fig:00-example}, notice that $H_1\cong G_i$ for each $i\in [1,4]$, we define a \emph{set-set-coloring} $F$ for the graph $H_1$ shown in Fig.\ref{fig:11-example}(a) as: $F(x)=\{\theta_i(x):~i\in [1,4]\}$ for each vertex $x\in V(H_1)=V(G_i)$ with $i\in [1,4]$, and $F(xy)=\{\theta_i(xy):~i\in [1,4]\}$ for each edge $xy\in E(H_1)=E(G_i)$ with $i\in [1,4]$.

A set-set-coloring is a \emph{compound set-coloring} (Ref. Definition \ref{defn:set-set-coloring}).\qqed
\end{rem}

In Fig.\ref{fig:topological-encryption}, we show a diagram for the algorithmic programming of the asymmetric topology encryption.

\begin{figure}[h]
\centering
\includegraphics[width=16.5cm]{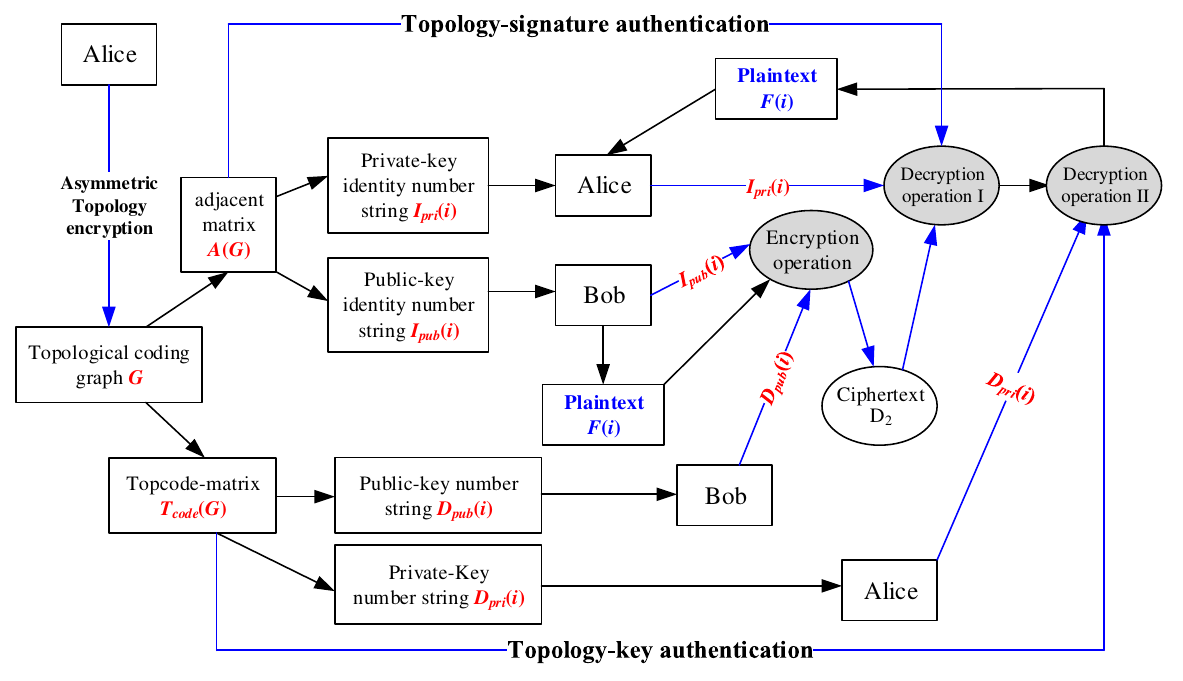}
\caption{\label{fig:topological-encryption}{\small A scheme for the asymmetric topology encryption, where Alice and Bob are a pair of communication users in a network.}}
\end{figure}

We propose the following FCGSC-problem (The problem of finding a set-colored graph admitting a $W$-constraint set-coloring):
\begin{quote}
\textbf{FCGSC-problem}: For a given $[0,9]$-string $s=c_1c_2\cdots c_n$ with $c_i\in [0,9]=\{0,1,2,\dots ,9\}$, \textbf{find} a set-colored graph $G$ admitting a $W$-constraint set-coloring $f$, such that the Topcode-matrix $T_{code}(G,f)$ deduces just the given number-based string $s$.
\end{quote}

\begin{rem}\label{rem:FCGSC-problem-difficulties}
To solve the FCGSC-problem, those people who attacking asymmetric topological encryption will encounter the following difficult points:
\begin{asparaenum}[\textrm{\textbf{Diff}}-1. ]
\item \textbf{Rewriting number-based string.} First, we need to write a number-based string $s=c_1c_2\cdots c_n$ with $c_i\in [0,9]$ as another number-based string $s^*=a_1a_2\cdots a_{3q}$, where each segment string $a_j=a_{j,1}a_{j,2}\cdots a_{j,b_j}$ for $j\in [1,3q]$, such that each $c_i$ of the number-based string $s$ appears in one and only one segment string $a_j$. As of now, there is no polynomial method reported to complete this task, so this is related with NP-type problems.
\item \textbf{Constructing a Topcode-matrix holding a $W$-constraint.} Using the number-based string $s^*=a_1a_2\cdots a_{3q}$ writes a Topcode-matrix $T_{code}$ of order $3\times q$ (Ref. Definition \ref{defn:topcode-matrix-definition} and Definition \ref{defn:total-coloring-Topcode-matrixs}), such that the number-based string $s^*$ is a permutation of the elements of the Topcode-matrix $T_{code}$, and $s^*$ holds some $W$-constraint $R^i_{est}(c_1,c_2,\cdots ,c_k)$ (Ref. Example \ref{exa:topological-signatures-set-coloring}). Unfortunately, there is no guarantee that there is no another Topcode-matrix $T\,'_{code}$, which can deduces the number-based string $s^*$.
\item \textbf{Subgraph Isomorphic NP-complete Problem.} Find a colored graph $G$ of $p$ vertices and $q$ edges admitting a $W$-constraint set-coloring $f$, such that the colored graph $G$ has its own Topcode-matrix $T_{code}(G,f)=T_{code}$.

\quad \textbf{Finding} the colored graph $G$ will meet the \textbf{Subgraph Isomorphic NP-complete Problem}, and moreover, for the number $n_p$ of graphs having $p$ vertices, we have two numbers $n_{23}$ and $n_{24}$ of different topological structures of graphs on 23 vertices and 24 vertices as follows
\begin{equation}\label{eqa:number-graphs-23-24-vertices}
{
\begin{split}
n_{23}&=559946939699792080597976380819462179812276348458981632\approx 2^{179}\\
n_{24}&=195704906302078447922174862416726256004122075267063365754368\approx 2^{197}
\end{split}}
\end{equation}
computed by Harary and Palmer \cite{Harary-Palmer-1973}.

\quad \textbf{Finding} the wanted colored graph $G$ is a terrible computational task for supercomputers, even for quantum computers as graphs have numerous vertices.

\quad \textbf{Finding} the $W$-constraint set-coloring $f$ of the colored graph $G$ also is sharp-P-hard, since the number of colorings of graphs is changing everyday, and the $W$-constraint $R^i_{est}(c_1,c_2,\cdots ,c_k)$ with $k$ constraints is related with a large number of unresolved mathematical conjectures.

\item Proposition \ref{prop:translated-number-based-string} tells us: A number-based string can be generated by the Topcode-matrices of two colored graphs $G$ and $H$, such that $G\not\cong H$. So, \textbf{Finding} the wanted colored graph $G$ is sharp-P-hard.\qqed
\end{asparaenum}
\end{rem}

By Remark \ref{rem:FCGSC-problem-difficulties}, we are able to claim:

(i) Using a given number-based string to find out a colored graph $G$ having its own Topcode-matrix $T_{code}(G,f)$ based on a $W$-constraint set-coloring $f$ is also NP-complete.

(ii) Due to the irreversibility and difficulty in cracking topological number-based strings, the absence of polynomial algorithms, and the presence of a large number of unresolved mathematical conjectures and challenges, any algorithm designed by using the topology encoding technology of this article has \emph{computable security} and \emph{provable security}, without the need for every practical application algorithm to undergo the demonstrations of computable security and provable security.

\subsubsection{Main topics in this article}

Graph set-colorings have been investigated by many researchers. Bollob\'{a}s and Thomason \cite{Bollobas-Thomason-2006} researched: ``\emph{An $r$-set coloring of a graph $G$ is an assignment of $r$ distinct colors to each vertex of the graph $G$ so that the sets of colors assigned to adjacent vertices are disjoint}.'' The \emph{sumset-labelling} was discussed in \cite{N-K-Sudev-2015}. Balister, Gy\H{o}i and Schelp \cite{Balister-Gyoi-Schelp-2011} discussed the strongly set-colorable graphs. Hegde \cite{S-M-Hegde-2009} introduced another type of set-coloring: ``\emph{A set-coloring of the graph $G$ is an assignment (function) of distinct subsets of a finite set $X$ of colors to the vertices of the graph, where the colors of the edges are obtained as the symmetric differences of the sets assigned to their end vertices which are also distinct.}''

Set-colorings serve to make more complex number-based strings from topology code theory for defending against the \emph{intelligent attacks} equipped with quantum computing and providing effective protection technology for the age of quantum computing. Graphs can be quickly converted into graphs admitting set-colorings, which will produce more complex number-based strings for serving information security, and moreover graph colorings and labelings are special set-colorings (Ref. \cite{Bing-Yao-arXiv:2207-03381}, \cite{Bing-et-al-arXiv-asymmetric-4520331}, \cite{Yao-Ma-arXiv-2201-13354v1}, \cite{Yao-Sun-Zhang-Mu-Sun-Wang-Su-Zhang-Yang-Yang-2018arXiv}, \cite{Yao-Wang-2106-15254v1}, \cite{Yao-Zhang-Sun-Mu-Sun-Wang-Wang-Ma-Su-Yang-Yang-Zhang-2018arXiv} and \cite{Yao-Sun-Zhang-Li-Yan-Zhang-Wang-ITOEC-2017}).

As known, graph colorings and labelings are special set-colorings, and more important is: Simplicity complex connects topology and graph theory, and provides a visual angle to observe hypergraphs through topological structure. Due to the limited number of visualization techniques for hypergraphs, it is difficult to master completely them.

\vskip 0.2cm

We, in this paper, have kept the many contents of the article \cite{Yao-Ma-arXiv-2201-13354v1}, and moreover add recent researching contents. We will focus on the following researching topics:
\begin{asparaenum}[\textrm{\textbf{Point}}-1.]
\item We will design vertex/edge-intersected graphs as a visualization of hypergraphs, and apply graph set-colorings for investigating hypergraphs.
\item Use set-type Topcode-matrices as an algebraic visualization of hypergraphs.
\item Set-colored graphs can be regarded as a non-direct or partial visualization of hypergraphs.
\item Study hyperedge sets of hypergraphs.
\item Researching set-colored graphs admitting set-colorings, especially related with hypergraphs.
\item Investigate set-colored graph homomorphism, hypergraph homomorphism.
\item Apply the theory of hypergraphs to information encryption.
\item Research various topological groups by means of algebraic methods.
\item Explore hypernetworks and hypernetworks lattices.
\item Try researching graph networks proposed by DeepMind and GoogleBrain.
\item Try more algebraic methods in studying hypergraphs.
\end{asparaenum}

\subsection{Terminology and notation}

Standard terminology and notation of graphs used here are cited from \cite{Bang-Jensen-Gutin-digraphs-2007}, \cite{Bondy-2008} and \cite{Gallian2022}, and all graphs mentioned here are \emph{simple} (namely, no \emph{multiple-edges} and \emph{loops}), unless otherwise stated, and \emph{graph colorings} and \emph{labelings} mentioned here are in \cite{Gallian2022} and \cite{Yao-Wang-2106-15254v1} if no definitions for them. We will employ the following notation and terminology:

\begin{asparaenum}[$\bullet$]
\item The \emph{number} (also, \emph{cardinality}) of elements of a set $X$ is denoted as $|X|$.
\item All non-negative integers are collected in the set $Z^0$, and all integers are in the set $Z$, so the positive integer set $Z^+=Z^0\setminus \{0\}$.
\item A $(p,q)$-graph $G$ is a \emph{topological structure} having $p$ vertices and $q$ edges, and $G$ has no multiple edge, loop and directed-edge, such that its own vertex set $V(G)$ holds the cardinality $|V(G)|=p$ and its own edge set $E(G)$ holds the cardinality $|E(G)|=q$. And the \emph{complementary graph} of the $(p,q)$-graph $G$ is denoted as $\overline{G}$, such that $|V(G)|=p=|V(\overline{G})|$ and $E(G)\cup E(\overline{G})=E(K_p)$, where $K_p$ is a \emph{complete graph} of $p$ vertices.
\item The \emph{degree} of a vertex $x$ in a $(p,q)$-graph $G$ is denoted as $\textrm{deg}_G(x)=|N_{ei}(x)|$, where $N_{ei}(x)$ is the set of \emph{neighbors} of the vertex $x$, such that each edge $xy\in E(G)$ for each neighbor vertex $y\in N_{ei}(x)$, also, we call $N_{ei}(x)$ \emph{adjacent neighbor set}.
\item A vertex $x$ in a graph is called a \emph{leaf} if its degree $\textrm{deg}_G(x)=1$.
\item The symbol $[a,b]$ stands for a consecutive integer set $\{a,a+1,a+2,\dots, b\}$ with two integers $a,b$ subject to $0\leq a<b$, and the notation $[\alpha,\beta]^o$ denotes an \emph{odd-set} $\{\alpha,\alpha+2,\dots, \beta\}$ with odd integers $\alpha,\beta$ holding $1\leq \alpha<\beta$ true.
\item Let $\Lambda$ be a set. The set of all subsets of $\Lambda$ is denoted as $\Lambda^2=\{X:~X\subseteq \Lambda\}$, called the \emph{power set} of $\Lambda$, and the power set $\Lambda^2$ contains no empty set at all. For example, for a given set $\Lambda=\{a,b,c,d,e\}$, the power set $\Lambda^2$ has its own elements $\{a\}$, $\{b\}$, $\{c\}$, $\{d\}$, $\{e\}$, $\{a,b\}$, $\{a,c\}$, $\{a,d\}$, $\{a,e\}$, $\{b,c\}$, $\{b,d\}$, $\{b,e\}$, $\{c, d\}$, $\{c, e\}$, $\{d,e\}$, $\{a,b,c\}$, $\{a,b,d\}$, $\{a,b,e\}$, $\{a,c,d\}$, $\{a,c,e\}$, $\{a,d,e\}$, $\{b,c,d\}$, $\{b,c,e\}$, $\{a,b,c,d\}$, $\{a,b,c,e\}$, $\{a,c,d,e\}$, $\{b,c,d,e\}$ and $\Lambda$ itself.

\qquad Moreover, the integer set $[1,4]=\{1,2,3,4\}$ induces a power set $[1,4]^2=\big \{\{1\},\{2\},\{3\},\{4\}$, $\{1,2\}$, $\{1,3\},\{1,4\}$, $\{2,3\}$, $\{2,4\}$, $\{3,4\}$, $\{1,2,3\}$, $\{1,2,4\}$, $\{1,3,4\}$, $\{2,3,4\}$, $[1,4]\big \}$, the number of subsets of the integer set $[1,4]$ is $\big |[1,4]^2\big |=2\,^4-1=15$, in total.
\item A \emph{sequence} $\textbf{\textrm{d}}=(m_1,m_2,\dots,m_n)=\{m_k\}^n_{k=1}$ consists of positive integers $m_1, m_2, \dots , m_n$. If a graph $G$ has its \emph{degree-sequence} ${\textrm{deg}}(G)=\textbf{\textrm{d}}$, then $\textbf{\textrm{d}}$ is \emph{graphical} by Lemma \ref{thm:basic-degree-sequence-lemma}, and we call $\textbf{\textrm{d}}$ \emph{degree-sequence}, and each $m_i$ \emph{degree component}, and $n=L_{\textrm{ength}}(\textbf{\textrm{d}})$ \emph{length} of $\textbf{\textrm{d}}$.
\item For integers $r,k\geq 0$ and $s,d\geq 1$, we define a \emph{parameterized set} as follows
\begin{equation}\label{eqa:two-parameterized-sets11}
S_{s,k,r,d}=\big \{k+rd,k+(r+1)d,\dots ,k+(r+s)d\big \}
\end{equation} and define a \emph{parameterized odd-set} as follows
\begin{equation}\label{eqa:two-parameterized-sets22}
O_{s,k,r,d}=\big \{k+[2(r+1)-1]d,k+[2(r+2)-1]d,\dots ,k+[2(r+s)-1]d\big \}
\end{equation}
\item A \emph{tree} is a connected and acyclic graph. A \emph{caterpillar} $T$ is a tree, if removing all of leaves from the caterpillar $T$, the remainder is just a \emph{path}. A \emph{lobster} is a tree too, and the deletion of all leaves of the lobster produces a caterpillar.
\item A topological isomorphism $G\cong H$ is a configuration identity on two graphs $G$ and $H$, it is independent of the colorings and drawing methods of these two graphs.
\end{asparaenum}

\begin{lem} \label{thm:basic-degree-sequence-lemma}
\cite{Bondy-2008} (Erd\"{o}s-Galia Theorem) A sequence $\textbf{\textrm{d}}=(m_k)^n_{k=1}$ with $m_{i}\geq m_{i+1}\geq 0$ to be \emph{degree-sequence} of a $(p,q)$-graph graph $G$ if and only if the sum $\sum^n_{i=1}m_i=2q$ and
\begin{equation}\label{eqa:basic-degree-sequence-lemma}
\sum^k_{i=1}m_i\leq k(k-1)+\sum ^n_{j=k+1}\min\{k,m_j\},~ 1\leq k\leq n
\end{equation}
\end{lem}

\begin{defn}\label{defn:topcode-matrix-definition}
\cite{Yao-Sun-Zhao-Li-Yan-2017} A \emph{Topcode-matrix} (or \emph{topology code theory matrix}) is defined as
\begin{equation}\label{eqa:Topcode-matrix}
\centering
{
\begin{split}
T_{code}= \left(
\begin{array}{ccccc}
x_{1} & x_{2} & \cdots & x_{q}\\
e_{1} & e_{2} & \cdots & e_{q}\\
y_{1} & y_{2} & \cdots & y_{q}
\end{array}
\right)_{3\times q}=
\left(\begin{array}{c}
X\\
E\\
Y
\end{array} \right)=(X,~E,~Y)^{T}
\end{split}}
\end{equation}\\
where \emph{v-vector} $X=(x_1, x_2, \cdots, x_q)$, \emph{e-vector} $E=(e_1$, $e_2 $, $ \cdots $, $e_q)$, and \emph{v-vector} $Y=(y_1, y_2, \cdots, y_q)$ consist of non-negative integers $e_i$, $x_i$ and $y_i$ for $i\in [1,q]$. We say $T_{code}$ to be \emph{$W$-constraint} if there exists an equation $W$ such that $W[x_i, e_i, y_i]=0$ for $i\in [1,q]$, and call $x_i$ and $y_i$ to be the \emph{ends} of $e_i$, as well as $q$ is the \emph{size} of $T_{code}$.\qqed
\end{defn}

A Topcode-matrix $T_{code}$ defined in Definition \ref{defn:topcode-matrix-definition} is \emph{graphable} if there is a $(p,q)$-graph $G$ having its own edge set $E(G)=\{e_{1},e_{2},\dots ,e_{q}\}$ holding $e_i=x_iy_i$ for each $i\in [1,q]$. Lemma \ref{thm:basic-degree-sequence-lemma} can help us to determine whether a Topcode-matrix is graphable.

\subsection{Labelings and colorings of graphs}

Many of graph colorings and labelings of graph theory were introduced in \cite{Gallian2022} and \cite{Yao-Wang-2106-15254v1}.

\begin{defn}\label{defn:totally-normal-labeling}
\cite{su-yan-yao-2018} \textbf{Distinctiveness of labeling and coloring.} Suppose that a $(p,q)$-graph $G$ admits a coloring $f:V(G)\rightarrow [m,n]$ or a total coloring $f: V(G)\cup E(G)\rightarrow [m,n]$, we denote the set of vertex colors of the graph $G$ as
\begin{equation}\label{eqa:555555}
f(V(G))=\{f(x):x\in V(G)\}
\end{equation} and the set of edge colors of the graph $G$ by
\begin{equation}\label{eqa:555555}
f(E(G))=\{f(uv):uv\in E(G)\}
\end{equation} and the total color set by $f(V(G)\cup E(G))=f(V(G))\cup f(E(G))$.

(i) If $|f(V(G))|=p$, then $f$ is called \emph{vertex labeling} of the graph $G$, otherwise \emph{vertex coloring};

(ii) When as the cardinality $|f(E(G))|=q$, $f$ is called \emph{edge labeling} of the graph $G$, otherwise \emph{edge coloring}; and

(iii) If $|f(V(G)\cup E(G))|=p+q$, we call $f$ \emph{total labeling}, otherwise \emph{total coloring}.\qqed
\end{defn}

\begin{defn} \label{defn:basic-W-type-labelings}
\cite{Gallian2022, Yao-Sun-Zhang-Mu-Sun-Wang-Su-Zhang-Yang-Yang-2018arXiv, Zhou-Yao-Chen-Tao2012} Suppose that a connected $(p,q)$-graph $G$ admits a coloring $\theta:V(G)\rightarrow \{0,1,2,\dots ,M\}$. For each edge $xy\in E(G)$, the induced edge color is defined as $\theta(xy)=|\theta(x)-\theta(y)|$. Write the vertex color set by $\theta(V(G))=\{\theta(u):u\in V(G)\}$, and the edge color set by
$\theta(E(G))=\{\theta(xy):xy\in E(G)\}$. There are the following constraints:

B-1. $|\theta(V(G))|=p$;

B-2. $\theta(V(G))\subseteq [0,q]$, $\min \theta(V(G))=0$;

B-3. $\theta(V(G))\subset [0,2q-1]$, $\min \theta(V(G))=0$;

B-4. $\theta(E(G))=\{\theta(xy):xy\in E(G)\}=[1,q]$;

B-5. $\theta(E(G))=\{\theta(xy):xy\in E(G)\}=[1,2q-1]^o$;

B-6. $G$ is a bipartite graph with the bipartition $(X,Y)$ such that $\max\{\theta(x):x\in X\}< \min\{\theta(y):y\in Y\}$ ($\max \theta(X)<\min \theta(Y)$ for short);

B-7. $G$ is a tree having a perfect matching $M$ such that $\theta(x)+\theta(y)=q$ for each matching edge $xy\in M$; and

B-8. $G$ is a tree having a perfect matching $M$ such that $\theta(x)+\theta(y)=2q-1$ for each matching edge $xy\in M$.

\noindent \textbf{Then}:
\begin{asparaenum}[\textbf{\textrm{Lac}}-1.]
\item A \emph{graceful labeling} $\theta$ satisfies B-1, B-2 and B-4 at the same time.
\item A \emph{set-ordered graceful labeling} $\theta$ holds B-1, B-2, B-4 and B-6 true.
\item A \emph{strongly graceful labeling} $\theta$ holds B-1, B-2, B-4 and
B-7 true.
\item A \emph{set-ordered strongly graceful labeling} $\theta$ holds B-1, B-2, B-4, B-6 and B-7 true.
\item An \emph{odd-graceful labeling} $\theta$ holds B-1, B-3 and B-5 true.
\item A \emph{set-ordered odd-graceful labeling} $\theta$ abides B-1, B-3, B-5 and B-6.
\item A \emph{strongly odd-graceful labeling} $\theta$ holds B-1, B-3, B-5 and B-8, simultaneously.
\item A \emph{set-ordered strongly odd-graceful labeling} $\theta$ holds B-1, B-3, B-5, B-6 and B-8 true.\qqed
\end{asparaenum}
\end{defn}

\begin{defn} \label{defn:2020arXiv-gracefully-total-coloring}
\cite{Bing-Yao-2020arXiv} Suppose that a connected $(p,q)$-graph $G$ ($\neq K_p$) admits a total coloring $f:V(G)\cup E(G)\rightarrow [1,M]$, and it is allowed $f(x)=f(y)$ for some vertices $x,y\in V(G)$ and $xy\not \in E(G)$. If $|f(V(G))|< p$, and the edge color set
$$
f(E(G))=\big \{f(uv)=|f(u)-f(v)|:uv\in E(G)\big \}=[1,q]
$$ then we call $f$ \emph{gracefully total coloring}.\qqed
\end{defn}

\begin{defn} \label{defn:e-odd-graceful-v-matching-labeling}
\cite{Yao-Ma-arXiv-2201-13354v1} If each $(p_i,q_i)$-graph $G_i$ admits a labeling $f_i$ such that $f_i(x)\neq f_i(y)$ for distinct vertices $x,y\in V(G_i)$, and each edge color set
$$
f_i(E(G_i))=\big \{f_i(u_jv_j)=|f_i(u_j)-f_i(v_j)|:~u_jv_j\in E(G_i)\big \}=\big [1,2q_i-1\big ]^o
$$ and $\bigcup ^m_{i=1}f_i(V(G_i))=[0,M]$ with $m\geq 2$, then we say that the \emph{edge-odd-graceful graph base} $\textbf{B}=(G_1,G_2,\dots , G_m)$ admits an \emph{edge-odd-graceful vertex-matching labeling} $F$ defined by $F=\uplus ^m_{i=1}f_i$, and $|f_i(V(G_i))|=p_i$ for $i\in [1,m]$, since each $f_i$ is a vertex labeling (Ref. Definition \ref{defn:totally-normal-labeling}).\qqed
\end{defn}

\begin{defn} \label{defn:general-definition-set-colorings}
$^*$ Suppose that a graph $G$ admits a total set-coloring
$$
\theta:V(G)\cup E(G)\rightarrow S_{et}=\{e_{1},e_{2},\dots ,e_{m}\}
$$ where $S_{et}$ is the set of sets $e_{1},e_{2},\dots ,e_{m}$. There are the following various set-type colorings:
\begin{asparaenum}[(i) ]
\item If the vertex color set $\theta(V(G))=\emptyset$ and the edge color set $|\theta(E(G))|\geq 1$, we call $\theta$ \emph{edge set-coloring} of the graph $G$.
\item If the edge color set $\theta(E(G))=\emptyset$ and the vertex color set $|\theta(V(G))|\geq 1$, we call $\theta$ \emph{vertex set-coloring} of the graph $G$.
\item If the vertex color set $|\theta(V(G))|\geq 1$ and the edge color set $|\theta(E(G))|\geq 1$, we call $\theta$ \emph{total set-coloring} of the graph $G$.
\end{asparaenum}

Moreover, as $\theta$ is a total set-coloring of the graph $G$, we have:
\begin{asparaenum}[\textbf{Par}-1. ]
\item If there is a $W$-constraint equation $W[\theta(u),\theta(uv),\theta(v)]=0$ for each edge $uv\in E(G)$ holding true, we call $\theta$ \emph{$W$-constraint total set-coloring} of the graph $G$.
\item If the set $S_{et}=\{e_{1},e_{2},\dots ,e_{m}\}$ holds $|e_{i}|=1$ and $e_{i}\subset Z^0$ for $i\in [1,m]$, then $\theta$ is a \emph{popular $W$-constraint coloring/labeling} as the $W$-constraint equation $W[\theta(u),\theta(uv),\theta(v)]=0$ for each edge $uv\in E(G)$ holding true (Ref. \cite{Gallian2022, Yao-Wang-2106-15254v1}).\qqed
\item If the set $S_{et}=\{e_{1},e_{2},\dots ,e_{m}\}$ is a hypergraph set $\mathcal{E}$ of a hypergraph $\mathcal{H}_{yper}=(\Lambda,\mathcal{E})$, we call $\theta$ \emph{total hyperedge set-coloring} of the graph $G$. Moreover, $\theta$ is a \emph{total intersected-hyperedge set-coloring} if each edge $uv\in E(G)$ holds $\theta(u)\cap \theta(v)\subseteq \theta(uv)$ with $\theta(u)\cap \theta(v)\neq \emptyset$.
\end{asparaenum}
\end{defn}

\begin{defn}\label{defn:total-coloring-Topcode-matrixs}
$^*$ Let $E(G)=\{e_i=x_iy_i:~i\in [1,q]\}$ be the edge set of a $(p,q)$-graph $G$, and let $f:V(G)\cup E(G)\rightarrow S_{pan}$ be a \emph{total pan-coloring} subject to a constraint set $R_{est}(c_1,c_2,\dots, c_m)$ with $m\geq 1$, where $S_{pan}$ is a \emph{pan-set}. Then we call the following matrix
\begin{equation}\label{eqa:topcode-matrix-total-coloring}
{
\begin{split}
T_{code}(G,f)= \left(
\begin{array}{cccccccccc}
f(x_1) & f(x_2) & \cdots & f(x_q)\\
f(e_1) & f(e_2) & \cdots & f(e_q)\\
f(y_1) & f(y_2) & \cdots & f(y_q)
\end{array}
\right)_{3\times q}=(f(X),f(E),f(Y))^T_{3\times q}
\end{split}}
\end{equation} \emph{Topcode-matrix}, where \emph{v-vector} $f(X)=(f(x_1)$, $f(x_2)$, $\dots$, $f(x_q))$, \emph{e-vector} $f(E)=(f(e_1),f(e_2)$, $\dots $, $f(e_q))$ and \emph{v-vector} $f(Y)=(f(y_1),f(y_2),\dots $, $f(y_q))$, such that each constraint $c_i$ of the constraint set $R_{est}(c_1,c_2,\dots, c_m)$ holds true.\qqed
\end{defn}

\begin{rem}\label{rem:topcode-matrix-homomorphisms-isomor}
About Definition \ref{defn:total-coloring-Topcode-matrixs} we point out:

(i) The total pan-coloring $f$ of the graph $G$ defined in Definition \ref{defn:total-coloring-Topcode-matrixs}, often, is a popular coloring/labeling introduced in \cite{Gallian2022} and \cite{Yao-Wang-2106-15254v1}, or a pan-coloring, or a set-coloring, or a graphic coloring, or a graphic group coloring, or a matrix coloring, or a hyperedge set-coloring, or a thing-coloring. Correspondingly, the pan-set $S_{pan}$ is a number set, or a coloring set, or a set-set, or a graph set, or a matrix set, or a hyperedge set, or any thing set, \emph{etc}.

(ii) The constraint set $R_{est}(c_1,c_2,\dots, c_m)$ consists of a $W$-constraint, or a group of constraints.

(iii) There are more colored graphs $H$ corresponding to the Topcode-matrix $T_{code}(G,f)$ shown in Eq.(\ref{eqa:topcode-matrix-total-coloring}), such that each colored graph $H$ is graph homomorphism to $G$, and $G\not\cong H$. We collect these colored graphs in to the graph set $G_{raph}(T_{code})$, where $T_{code}=T_{code}(G,f)$, such that each graph $H\in G_{raph}(T_{code})$ corresponds its own Topcode-matrix $T_{code}(H,g)=T_{code}$.\qqed
\end{rem}

Techniques of Topcode-matrices and Remark \ref{rem:topcode-matrix-homomorphisms-isomor} enable us to obtain the following results:
\begin{prop}\label{prop:translated-number-based-string}
$^*$ (i) Each simple graph can be translated into a number-based string.

(ii) A number-based string can be generated by the Topcode-matrices of two colored graphs $G$ and $H$ with $G\not\cong H$.
\end{prop}
This is just ``\emph{codes are related with graphs, conversely, graphs are as codes}'' proposed by many researchers of computer and information security.

\subsection{Graph operations}

Many network problems in reality are composed of small block (modular) networks. Graph just organically combines these small blocks into a whole, which is also the most natural and reasonable technical means. By splitting and refining the network, the minimal structural features have been obtained. The minimal structural features of networks can help us to understand the structure and topological properties of networks.

\emph{Graph operations are the soul of topological structures of graphs}.

\subsubsection{Graph operations by adding or removing vertices and edges}

There are some simple graph operations as follows:

\begin{asparaenum}[$\bullet$ ]
\item Removing an edge $uv$ from a graph $G$ produces an \emph{edge-removed graph}, denoted as $G-uv$.
\item adding a new edge $xy\not \in E(G)$ to the graph $G$ makes an \emph{edge-added graph}, written as $G+xy$.
\item $G-w$ is a \emph{vertex-removed graph} after deleting the vertex $w$ from $G$, and removing those edges with one end to be this vertex $w$.
\item A \emph{ve-added graph} $G+\{y,x_1y,x_2y,\dots ,x_sy\}$ is obtained by adding a new vertex $y$ to a graph $G$, and join $y$ with vertices $x_1,x_2,\dots ,x_s$ of the graph $G$ by new edges $x_1y,x_2y,\dots ,x_sy$, respectively.
\item By those edge-removed graph $G-uv$ and vertex-removed graph $G-w$, we have a \emph{vertex-set-removed graph} $G-S$ for a vertex proper subset $S\subset V(G)$, as well as an \emph{edge-set-removed graph} $G-E^*$ for an edge subset $E^*\subset E(G)$.
\item Particularly, an \emph{edge-set-added graph} $G+E\,'$ is obtained by adding each edge of an edge set $E\,'\subset E(\overline{G})$ to a graph $G$, where $\overline{G}$ is the complement of the graph $G$.
\end{asparaenum}

\subsubsection{Vertex-splitting and vertex-coinciding operations}

\begin{defn} \label{defn:vertex-split-coinciding-operations}
\cite{Yao-Zhang-Sun-Mu-Sun-Wang-Wang-Ma-Su-Yang-Yang-Zhang-2018arXiv, Yao-Sun-Zhang-Mu-Wang-Jin-Xu-2018} \textbf{Vertex-splitting operation.} We vertex-split a vertex $u$ of a graph $G$ with $\textrm{deg}(u)\geq 2$ into two vertices $u\,'$ and $u\,''$, such that the neighbor set $N_{ei}(u)=N_{ei}(u\,')\cup N_{ei}(u\,'')$ with $|N_{ei}(u\,')|\geq 1$, $|N_{ei}(u\,'')|\geq 1$ and $N_{ei}(u\,')\cap N_{ei}(u\,'')=\emptyset$, the resultant graph is denoted as $G\wedge u$, called \emph{vertex-split graph} (see an example shown in Fig.\ref{fig:11-vertex-split-coin} (a)$\rightarrow$(b)). Moreover, we select randomly a proper subset $S$ of vertex set $V(G)$, and implement the vertex-splitting operation to each vertex of the proper subset $S$, the resultant graph is denoted as $G\wedge S$.

\textbf{Vertex-coinciding operation.} (Also, called the \emph{non-common neighbor vertex-coinciding operation}) A vertex-coinciding operation is the \emph{inverse} of a vertex-splitting operation, and vice versa. If two vertices $u\,'$ and $u\,''$ of a graph $H$ holds $|N_{ei}(u\,')|\geq 1$, $|N_{ei}(u\,'')|\geq 1$ and $N_{ei}(u\,')\cap N_{ei}(u\,'')=\emptyset$ true, we vertex-coincide $u\,'$ and $u\,''$ into one vertex $u=u\,'\bullet u\,''$, such that the neighbor set $N_{ei}(u)=N_{ei}(u\,')\cup N_{ei}(u\,'')$, the resultant graph is denoted as $G=H(u\,'\bullet u\,'')$, called \emph{vertex-coincided graph} (see a scheme shown in Fig.\ref{fig:11-vertex-split-coin} (b)$\rightarrow$(a)).\qqed
\end{defn}

\begin{figure}[h]
\centering
\includegraphics[width=16.4cm]{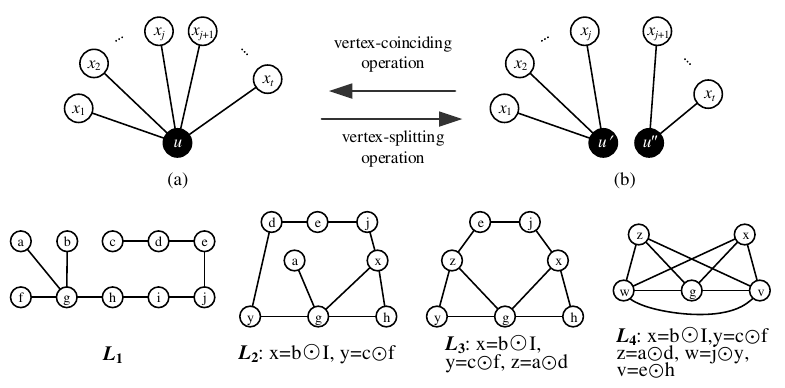}\\
\caption{\label{fig:11-vertex-split-coin}{\small The vertex-coinciding and the vertex-splitting operations defined in Definition \ref{defn:vertex-split-coinciding-operations}.}}
\end{figure}

\begin{rem}\label{rem:vertex-coinciding-operations}
If two vertex-disjoint colored graphs $G$ and $H$ have $k$ pairs of vertices with each pair of vertices is colored with the same color, then we, by the vertex-coinciding operation defined in Definition \ref{defn:vertex-split-coinciding-operations}, vertex-coincide each pair of vertices from the colored graph $G$ and the colored graph $H$ into one, the resultant graph is denoted as $G[\bullet _k]H$, called \emph{vertex-coincided graph} hereafter, and moreover we have two cardinalities
$$
|V(G[\bullet _k]H)|=|V(G)|+|V(H)|-k,~|E(G[\bullet _k]H)|=|E(G)|+|E(H)|
$$ for the vertex set and edge set of $G[\bullet _k]H$, respectively. Clearly, the vertex-coincided graph $G[\bullet _k]H$ holds for the case of two vertex-disjoint uncolored graphs $G$ and $H$ too.\qqed
\end{rem}

\begin{problem}\label{question:trees-same-number-leaves}
Let $S_{plit}(G,leaf)$ be the set of trees with the same number of leaves after vertex-splitting a connected graph $G$ (see Fig.\ref{fig:G-splt-trees-same-leaves}). \textbf{Determine} $S_{plit}(G,leaf)$.
\end{problem}

\begin{figure}[h]
\centering
\includegraphics[width=16.4cm]{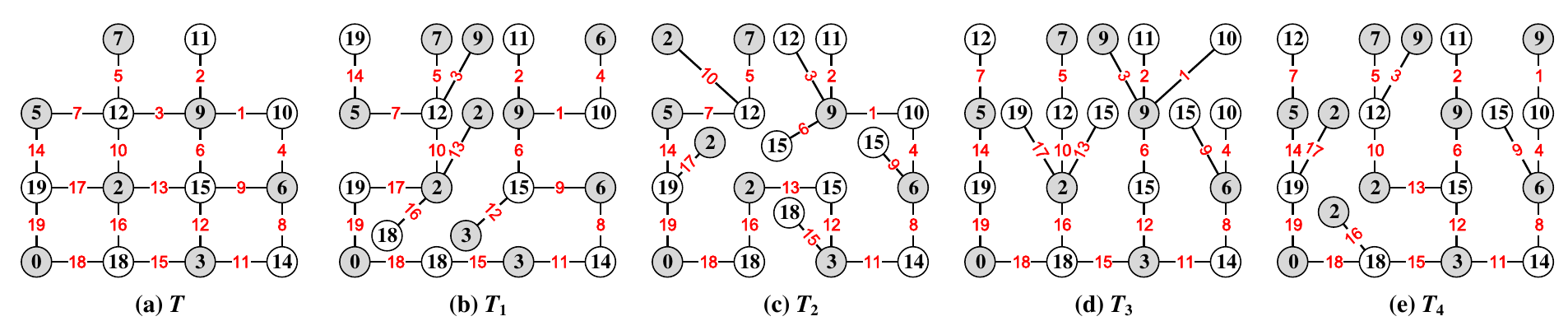}\\
\caption{\label{fig:G-splt-trees-same-leaves}{\small A scheme for illustrating Problem \ref{question:trees-same-number-leaves}.}}
\end{figure}

\begin{thm}\label{thm:666666}
A connected graph $G$ can be vertex-split into two edge-disjoint graphs $G_1$ and $G_2$ holding each maximal degree $\Delta(G_i)\leq \frac{1}{2}[\Delta(G)+1]$ for $i=1,2$, and moreover three total chromatic numbers hold
\begin{equation}\label{eqa:555555}
\chi\,''(G)\leq \chi\,''(G_1)+\chi\,''(G_2)
\end{equation} by the vertex-coinciding operation.
\end{thm}

\begin{rem}\label{rem:333333}
Let $\Delta(G)$ be the maximum degree of a graph $G$, and let $K(G)$ be the maximum clique number of the graph $G$. For the chromatic number $\chi(G)$ and the total chromatic number $\chi\,''(G)$ of a graph $G$, there are two longstanding conjectures:
$${
\begin{split}
&\textrm{Reed's conjecture: }~\chi(G)\leq \left \lceil \frac{\Delta(G)+1+K(G)}{2} \right \rceil \\
&\textrm{Behzad and Vizing's conjecture: } ~\chi\,''(G)\leq \Delta(G)+2
\end{split}}
$$ proposed by Bruce Reed (1998), Behzad (1965), Vizing (1964), respectively.\qqed
\end{rem}

\begin{defn} \label{defn:set-colored-graphs-split-coincide}
$^*$ Let $S_{et}$ be a set of sets, and let $G$ be a connected graph.

\textbf{A. The vertex-splitting operation of set-colored graphs.} Suppose that a connected graph $G$ admits a total set-coloring $F:V(G)\cup E(G)\rightarrow S_{et}$, such that $F(uv)=F(u)\cap F(v)$ for each edge $uv\in E(G)$. By the vertex-splitting operation of Definition \ref{defn:vertex-split-coinciding-operations}, we vertex-split a vertex $u$ of the connected graph $G$ if degree $\textrm{deg}_G(u)\geq 2$ into two vertices $u\,',v\,'$ holding $N_{ei}(u)=N_{ei}(u\,')\cup N_{ei}(v\,')$ and $|N_{ei}(u\,')|\geq 1$ and $|N_{ei}(v\,')|\geq 1$, and $N_{ei}(u\,')\cap N_{ei}(v\,')=\emptyset$, and define a new total set-coloring $F\,'$ for the \emph{vertex-split graph} $G\wedge u$ as:

(A-1) $F\,'(u\,')\cup F\,'(v\,')=F(u)$,

(A-2) $F\,'(u\,'x)=F\,'(u\,')\cap F\,'(x)=F\,'(u\,')\cap F(x)$ for each vertex $x\in N_{ei}(u\,')$,

(A-3) $F\,'(v\,'y)=F\,'(v\,')\cap F\,'(y)=F\,'(v\,')\cap F(y)$ for each vertex $y\in N_{ei}(v\,')$,

(A-4) each element
$$
w\in \big (V(G\wedge u)\cup E(G\wedge u)\big )\setminus \big (\{u\,',v\,'\}\cup \{x, u\,'x:x\in N_{ei}(u\,')\}\cup \{y, v\,'y:y\in N_{ei}(v\,')\}\big )
$$ holding $w\in V(G)\cup E(G)$ is colored with $F\,'(w)=F(w)$.

\textbf{B. The vertex-coinciding operation of set-colored graphs.} Suppose that a graph $G$ admits a total set-coloring $f:V(G)\cup E(G)\rightarrow S_{et}$, such that $f(xy)=f(x)\cap f(y)$ for each edge $xy\in E(G)$. If there are two vertices $a$ and $b$ holding two neighbor sets $N_{ei}(a)\cap N_{ei}(b)=\emptyset$, by Definition \ref{defn:vertex-split-coinciding-operations}, we vertex-coincide these two vertices into one vertex $w=a\bullet b$ and $N_{ei}(w)=N_{ei}(a)\cup N_{ei}(b)$, and defined a new total set-coloring $g$ for the \emph{vertex-coincided graph} $H=G(a\bullet b)$ as follows:

(B-1) $g(w)=f(a)\cup f(b)$,

(B-2) $g(wx)=f(ax)$ for each vertex $x\in N_{ei}(a)\subset N_{ei}(w)$,

(B-3) $g(wy)=f(by)$ for each vertex $y\in N_{ei}(b)\subset N_{ei}(w)$,

(B-4) $g(z)=f(z)$ for each element $z\in \big (V(H)\cup E(H)\big )\setminus \{w,wz;z\in N_{ei}(w)\}$.\\
Two vertex sets holds $|V(H)|=|V(G)|-1$ and, two edge sets holds $|E(H)|=|E(G)|$.\qqed
\end{defn}

\subsubsection{Edge-coinciding and edge-splitting operations}

\begin{defn} \label{defn:edge-split-coinciding-operations}
\cite{Yao-Zhang-Sun-Mu-Sun-Wang-Wang-Ma-Su-Yang-Yang-Zhang-2018arXiv} \textbf{Edge-splitting operation.} For an edge $uv$ of a graph $G$ with $\textrm{deg}(u)\geq 2$ and $\textrm{deg}(v)\geq 2$, we remove the edge $uv$ from $G$ first, next we vertex-split, respectively, two end vertices $u$ and $v$ of the edge $uv$ into vertices $u\,'$ and $u\,''$, $v\,'$ and $v\,''$. And then we add a new edge $u\,'v\,'$ to join two vertices $u\,'$ and $v\,'$ together, and add another new edge $u\,''v\,''$ to join two vertices $u\,''$ and $v\,''$ together, respectively. The resultant graph is denoted as $G\wedge uv$, see Fig.\ref{fig:11-edge-leaf-split-coin} (a)$\rightarrow$(c). We call the procedure of obtaining $G\wedge uv$ \emph{edge-splitting operation}.

Here, it is allowed that two adjacent neighbor sets $|N_{ei}(u\,')|=1$ and $|N_{ei}(v\,'')|=1$, see Fig.\ref{fig:11-edge-leaf-split-coin}(a)$\rightarrow$(b), in this case, we call the procedure of obtaining $G\wedge uv$ \emph{leaf-splitting operation}, or \emph{train-hook splitting operation}, they are \emph{particular cases} of the edge-splitting operation. The \emph{inverse} of a train-hook splitting operation is called \emph{train-hook coinciding operation}.

\textbf{Edge-coinciding operation.} For two edges $u\,'v\,'$ and $u\,''v\,''$ of a graph $H$, if the adjacent neighbor sets $N_{ei}(u\,')\cap N_{ei}(u\,'')=\emptyset$ and $N_{ei}(v\,')\cap N_{ei}(v\,'')=\emptyset$, we edge-coincide two edges $u\,'v\,'$ and $u\,''v\,''$ into one edge $uv=u\,'v\,'\ominus u\,''v\,''$ with $u=u\,'\bullet u\,''$ and $v=v\,'\bullet v\,''$. The resultant graph $H(u\,'v\,'\ominus u\,''v\,'')$ is the result of doing the \emph{edge-coinciding operation} to $H$, see Fig.\ref{fig:11-edge-leaf-split-coin}(c)$\rightarrow$(a). Also, $H(u\,'v\,'\ominus u\,''v\,'')$ is the result of doing the \emph{leaf-coinciding operation} to $H$ as $|N_{ei}(u\,'')|=1$ and $|N_{ei}(v\,')|=1$ (see Fig.\ref{fig:11-edge-leaf-split-coin}(b)$\rightarrow$(a)).\qqed
\end{defn}

\begin{figure}[h]
\centering
\includegraphics[width=16.4cm]{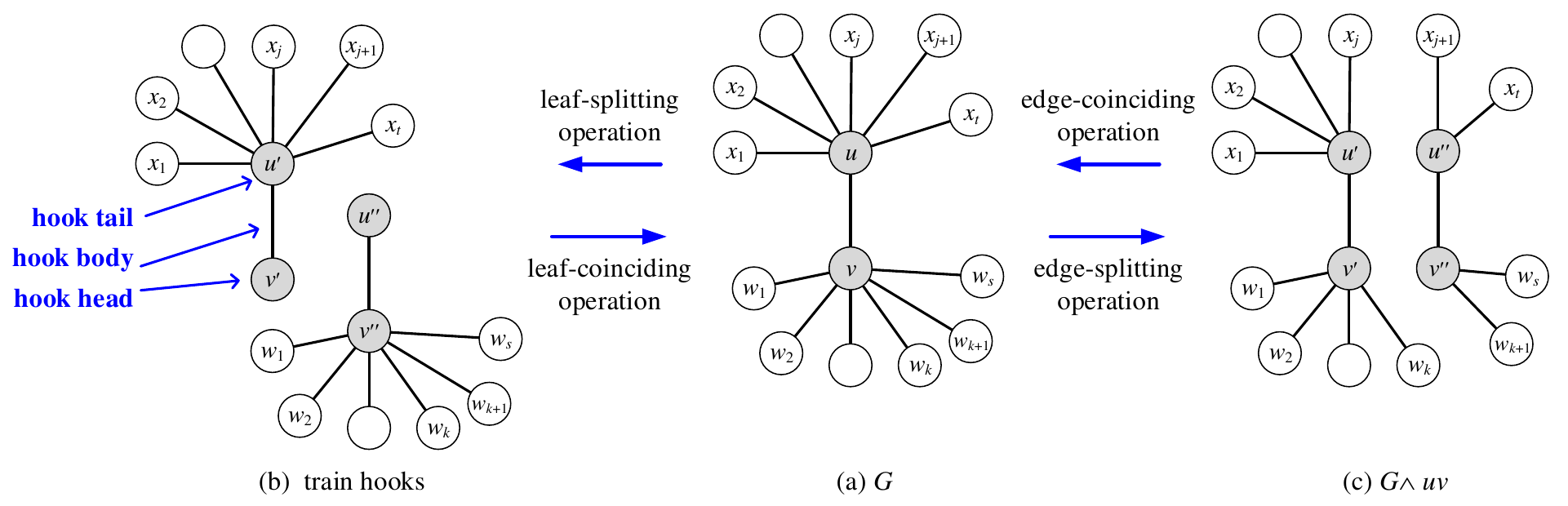}\\
\caption{\label{fig:11-edge-leaf-split-coin}{\small The edge-coinciding and the edge-splitting operations defined in Definition \ref{defn:edge-split-coinciding-operations}, cited from \cite{Yao-Zhang-Sun-Mu-Sun-Wang-Wang-Ma-Su-Yang-Yang-Zhang-2018arXiv}.}}
\end{figure}

\subsubsection{$\Omega$-coinciding and $\Omega$-splitting operations}

\begin{defn} \label{defn:W-splitting-coinciding-operation}
\cite{Yao-Su-Ma-Wang-Yang-arXiv-2202-03993v1} Let $\Omega$ be a proper subgraph of a graph $G$. We do a $\Omega$-splitting operation to $G$ in the following way \cite{Yao-Su-Sun-Wang-Graph-Operations-2021}:

(i) Removing the edges of $E(\Omega)$ from the proper subgraph $\Omega$;

(ii) Vertex-split each vertex $x_i\in V(\Omega)=\{x_i:i\in [1,m]\}$ into two vertices $x\,'_i$ and $x\,''_i$, such that $N_{ei}(x_i)\setminus V(\Omega)=N_{ei}(x\,'_i)\cup N_{ei}(x\,''_i)$ with $N_{ei}(x\,'_i)\cap N_{ei}(x\,''_i)=\emptyset$;

(iii) adding new edges to the vertex set $\big \{x\,'_i:i\in [1,m]\big \}$ produces a graph $H_1$ holding $H_1\cong \Omega$ true, and then adding new edges to the vertex set $\big \{x\,''_i:i\in [1,m]\big \}$ makes another graph $H_2$ holding $H_2\cong \Omega$ true, such that each edge $x_ix_j\in E(\Omega)$ corresponds an edge $x\,'_ix\,'_j\in E(H_1)$ and an edge $x\,''_ix\,''_j\in E(H_2)$, and vice versa.

The resultant graph is written as $G\wedge \Omega$, and it has the following properties:

(i) Both $\Omega$-type graphs $H_1$ and $H_2$ are two vertex disjoint isomorphic subgraphs of the $\Omega$-split graph $G\wedge \Omega$, namely, $V(H_1)\cap V(H_2)=\emptyset$;

(ii) each $H_i$ is joined with a vertex $w_i\in V(G\wedge \Omega)\setminus \big [V(H_1)\cup V(H_2)\big ]$ for $i=1,2$; and

(iii) no a common vertex $u^*\in V(G\wedge \Omega)\setminus \big [V(H_1)\cup V(H_2)\big ]$ holds $u^*x_i\in E(H_i)$ for $i=1,2$.

We call the process of obtaining the \emph{$\Omega$-split graph} $G\wedge \Omega$ \emph{$\Omega$-splitting operation}. Conversely, the process of obtaining $G$ from $G\wedge \Omega$ by the vertex-coinciding operation and the edge-coinciding operation defined in Definition \ref{defn:vertex-split-coinciding-operations} and Definition \ref{defn:edge-split-coinciding-operations} is called \emph{$\Omega$-coinciding operation}, since $N_{ei}(x\,'_i)\cap N_{ei}(x\,''_i)=\emptyset$ for $i\in [1,m]$. \qqed
\end{defn}

\begin{rem}\label{rem:edge-split-coinciding-operation}
In Definition \ref{defn:W-splitting-coinciding-operation}, if $G\wedge \Omega$ is disconnected, so $G$ has two vertex-disjoint components $G_1$ and $G_2$ holding $\Omega=H_i\subset G_i$ for $i=1,2$, we then write $G=G_1\big [\ominus ^{\Omega}_k\big ]G_2$.

For vertex disjoint graphs $T_s$ with $s\in [1,n]$, if each graph $T_s$ contains a subgraph $\Omega$ of $k$ vertices, we get an $\Omega$-coincided graph denoted as
\begin{equation}\label{eqa:555555}
\big [\ominus ^{\Omega}_k\big ]^n_{s=1}T_s=\Big (\cdots \big (T_1\big [\ominus ^{\Omega}_k\big ]T_2\big )\big [\ominus ^{\Omega}_k\big ]T_3\cdots \Big )\big [\ominus ^{\Omega}_k\big ]T_n
\end{equation} For a permutation $T_{j_1}, T_{j_2}, \dots ,T_{j_s}$ of $T_1,T_2,\dots ,T_n$, we have $\big [\ominus ^{\Omega}_k\big ]^n_{r=1}T_{j_r}$. So, there are $n!$ $\Omega$-coincided graphs in total.

About Definition \ref{defn:W-splitting-coinciding-operation}, we have the following particular cases:

\textbf{Case 1.} If $\Omega$ is a cycle $C$ of $k$ vertices, we write $T_1\big [\ominus ^{cyc}_k\big ]T_2$ by ``$\big [\ominus ^{cyc}_k\big ]$'' replacing ``$\big [\ominus ^{\Omega}_k\big ]$'', similarly, $T_1\big [\ominus ^{path}_k\big ]T_2$ if $\Omega$ is a path of $k$ vertices, and $T_1\big [\ominus ^{tree}_k\big ]T_2$ if $\Omega$ is a tree of $k$ vertices, since cycles, paths and trees are \emph{linear-type graphs} in various applications.

\textbf{Case 2.} If $\Omega$ is a complete graph $K_n$ of $n$ vertices, we have $T_1\big [\ominus ^{K}_n\big ]T_2$ if $K_n$ is a subgraph of two vertex disjoint graphs $T_i$ for $i=1,2$.

\textbf{Case 3.} If $\Omega$ is a cycle $C$ of $k$ vertices in a maximal planar graph $G$, the cycle-split graph $G\wedge C$ has just two vertex disjoint components $G^C_{out}$ and $G^C_{in}$, called \emph{semi-maximal planar graphs}, where $G^C_{out}$ is in the \emph{infinite plane}, and $G^C_{in}$ is inside of the graph $G$. Thereby, we write $G=G^C_{out}\big [\ominus^C_k\big ]G^C_{in}$ hereafter (Ref. \cite{Jin-Xu-Maximal-Science-Press-2019}).

\textbf{Case 4.} If $\Omega$ is a complete graph $K_1$ of one vertex, we write $T_1\big [\ominus ^{K}_1\big ]T_2=T_1[\bullet]T_2$, that is, the graph $\Omega$ shrinks to a vertex.\qqed
\end{rem}

\begin{problem}\label{problem:xxxxxxxxx}
\cite{Yao-Su-Sun-Wang-Graph-Operations-2021} \textbf{Characterize} the following particular cycle-coincided graphs:
\begin{asparaenum}[\textbf{\textrm{Planep}}-1. ]
\item A graph $G$ can be expressed as $G=H_i[\ominus ^{cyc}_{k_i}]G_i$ for $i\in [1,m]$ with $m\geq 1$ by the cycle-coinciding operation, where two vertex disjoint graphs $H_i$ and $G_i$ for $i\in [1,m]$ contain cycles with the same length $k_i$.

\item A cycle-coincided graph $L_i=H^*\big [\ominus ^{cyc}_{k_i}\big ]G_i$ for $i\in [1,m]$, where any pair of vertex disjoint graphs $H^*$ and each $G_i$ contain cycles with the same length $k_i$, where $H^*$ is like a fixed ``point'' under the cycle-coinciding operation. Furthermore, we get a cycle-coincided graph

\begin{equation}\label{eqa:555555}
[\ominus ^{cyc}]^m_{j=1}L_j=\big [\ominus ^{cyc}\big ]^m_{j=1}\Big (H^*[\ominus ^{cyc}_{k_j}]G_j \Big )=\Big (\cdots \big (H^*\big [\ominus ^{cyc}_{k_1}\big ]G_1\big )[\ominus ^{cyc}_{k_2}]G_2\cdots \Big )\big [\ominus ^{cyc}_{k_m}\big ]G_m
\end{equation} also, called \emph{kaleidoscope}.

\item A cycle-coincided graph $B=H^*\big [\ominus ^{cyc}_{n}\big ]^m_{i=1}G_i$ is like a ``\emph{super book}'', where $H^*$ and each $G_i$ contain cycles with the same length $n$, so each $G_i$ is a \emph{book page} and $H^*$ is the \emph{book back} of the super book.
\item If $\Omega$ is a path of $k$ vertices, $H^*\big [\ominus ^{path}_k\big ]^m_{i=1}G_i$ is ``\emph{topological-page book}'', where the book back $H^*$ and each topological-page $G_i$ contain paths of $k$ vertices.
\end{asparaenum}
\end{problem}

\begin{problem}\label{qeu:uniquely-4-colorable-mpgs}
\cite{Yao-Su-Ma-Wang-Yang-arXiv-2202-03993v1} Let $C$ be a $k$-cycle of a maximal planar graph $G$ with $k\geq 3$, so $G=G^{C}_{out}[\ominus^{cyc}_k]G^{C}_{in}$, and write $G^{C}_{out}=G_{out}$ (as a \emph{public-key}) and $G^{C}_{in}=G_{in}$ (as a \emph{private-key}) if there is no confusion. We have:
\begin{asparaenum}[\textbf{\textrm{MPG}}-1. ]
\item For each triangle $C=K_3$, $G=G_{out}\big [\ominus^{cyc}_3\big ]G_{in}$ holds $G_{out}=G$ and $G_{in}=K_3$, we call $G$ a \emph{no-$3$-cycle split maximal planar graph}.
\item For each $k$-cycle $C$ with $4\leq k <|V(G)|$, if the edge-removed graph $G_{in}-E(C)$ in $G=G_{out}\big [\ominus^{cyc}_k\big ]G_{in}$ is a tree $T$, we call $G_{in}$ a \emph{cycle-chord semi-maximal planar graph} if $V(C)=V(T)$, $G_{in}$ a \emph{tree-pure semi-maximal planar graph} if $|V(C)|<|V(T)|$, refer to \cite{Jin-Xu-Maximal-Science-Press-2019}.
\item For a maximal planar graph $G\neq K_4$, if $G=G_{out}(1)\big [\ominus^{cyc}_3\big ]G_{in}(1)$ with $G_{out}(1)$ is a maximal planar graph being not $K_4$ and $G_{in}(1)=K_4$, we have $G_{out}(1)=G_{out}(2)\big [\ominus^{cyc}_3\big ]G_{in}(2)$ with $G_{out}(2)$ is a maximal planar graph being not $K_4$ and $G_{in}(2)=K_4$, go on in this way, we get $G_{out}(k-1)=G_{out}(k)\big [\ominus^{cyc}_3\big ]G_{in}(k)$ with $G_{out}(k)$ is a maximal planar graph being not $K_4$ and $G_{in}(k)=K_4$ for $k\in [1,m]$, where $G=G_{out}(0)$, $G_{out}(m-1)=G_{out}(m)\big [\ominus^{cyc}_3\big ]G_{in}(m)$ with $G_{out}(m)=G_{in}(m)=K_4$. So, $G$ is a \emph{recursive maximal planar graph} and admits a proper vertex $4$-coloring $f$, such that $V(G)=\bigcup^4_{k=1} V_k(G)$ and $f(x)=k$ for $x\in V_k(G)$ with $k\in [1,4]$. Uniquely 4-colorable Maximal Planar Graph Conjecture \cite{Greenwell-Kronk-Uniquely-4c-1973}: A recursive maximal planar graph $G$ is uniquely 4-colorable, that is, each set $V_k(G)$ in $V(G)=\bigcup^4_{k=1} V_k(G)$ is not changed by any two $4$-colorings of the recursive maximal planar graph $G$.
\end{asparaenum}
\end{problem}

Now, we define the so-called \emph{$\Omega$-type graph-split connectivity} for a connected graph $G$:

\begin{defn} \label{defn:W-type-graph-split-connectivity}
\cite{Yao-Su-Ma-Wang-Yang-arXiv-2202-03993v1} Let $H$ be a $\Omega$-type proper subgraph of a connected graph $G$. If the $\Omega$-split graph $G\wedge \Omega$ is disconnected, we call the following parameter
$$\min \{|V(H)|:~G\wedge \Omega\textrm{ is disconnected},~\textrm{and }H~ \textrm{is a $\Omega$-type proper subgraph of}~G\}$$ \emph{$\Omega$-type graph-split connectivity} of the connected graph $G$, denoted as $\kappa_{W}(G)$.\qqed
\end{defn}

\begin{thm}\label{thm:vertex-splitting-connectivity-traditional}
\cite{Wang-Su-Yao-2021-computer-science} The \emph{vertex-splitting connectivity} of a connected graph is equivalent to its own \emph{vertex
connectivity}.
\end{thm}

\begin{rem}\label{rem:333333}
About Definition \ref{defn:W-type-graph-split-connectivity}, we have:

(i) ``$\Omega$-type'' may be one of path, cycle, complete graph, tree, bipartite complete graph, particular graph, and so on.

(ii) If $H$ is a graph consisted of edges, then the $\Omega$-type graph-split connectivity $\kappa_{W}(G)=\kappa\,'(G)$, the traditional \emph{edge connectivity} of graphs; and if $H$ is a graph consisted of vertices and edges, then the $\Omega$-type graph-split connectivity $\kappa_{W}(G)=\kappa(G)$ or $\kappa_{W}(G)=\kappa\,''(G)$ for the traditional \emph{vertex connectivity} $\kappa(G)$, or the traditional \emph{total connectivity} $\kappa\,''(G)$. Notice that the $\Omega$-split graph $G\wedge \Omega$ differs from the vertex-removed graph $G-V(H)$, since $G\wedge \Omega$ keeps all information of the original graph $G$.\qqed
\end{rem}

\begin{problem}\label{qeu:444444}
\cite{Yao-Su-Ma-Wang-Yang-arXiv-2202-03993v1} \textbf{Determine} $\Omega$-type graph-split connectivities $\kappa_{path}$, $\kappa_{cycle}$ and $\kappa_{tree}$ for connected graphs, where $\kappa_w(G)$ is defined in Definition \ref{defn:W-type-graph-split-connectivity}, and $W=$path, cycle, tree.
\end{problem}

\begin{rem}\label{rem:333333}
Many network problems in reality are composed of small block (modular) networks. Graph just organically combines them into a whole, which is also the most natural and reasonable technical means. By splitting and refining the network, the minimal structural features have been obtained. The minimal structural features of networks can help us to understand the structure and topological properties of networks.

Because our vertex-splitting connectivity is equivalent to the traditional vertex connectivity (Ref. Definition \ref{defn:W-type-graph-split-connectivity} and Theorem \ref{defn:W-type-graph-split-connectivity}), so the reliability of topology code theory has been proved.\qqed
\end{rem}

\subsubsection{Operations on graph homomorphisms}

Homomorphic encryption is a cryptographic technique based on computational complexity theory of mathematical puzzles in cloud computing, e-commerce, Internet of Things, mobile code \emph{etc}. The homomorphic encrypted data is processed to get an output, and the output is decrypted to get the same output as the unencrypted raw data processed in the same way, in other word, homomorphic encryption is required to achieve data security.

\begin{defn}\label{defn:definition-graph-homomorphism}
\cite{Bondy-2008} A \emph{graph homomorphism} $G\rightarrow H$ from a graph $G$ into another graph $H$ is a coloring $f: V(G) \rightarrow V(H)$ such that each edge $f(u)f(v)\in E(H)$ if and only if each edge $uv\in E(G)$.\qqed
\end{defn}

\begin{example}\label{exa:8888888888}
In Fig.\ref{fig:11-vertex-split-coin}, there are three graph homomorphisms $L_i\rightarrow _{\textrm{v-coin}}L_{i+1}$ for $i\in [1,3]$ obtained by the vertex-coinciding operation defined in Definition \ref{defn:vertex-split-coinciding-operations}, and we have three \emph{graph anti-homomorphisms} $L_k\rightarrow _{\textrm{v-split}}L_{k-1}$ for $k\in [2,4]$ obtained by the vertex-splitting operation defined in Definition \ref{defn:vertex-split-coinciding-operations}.\qqed
\end{example}

\begin{defn}\label{defn:11-topo-auth-faithful}
\cite{Gena-Hahn-Claude-Tardif-1997} A graph homomorphism $\varphi: G \rightarrow H$ is called \emph{faithful} if $\varphi(G)$ is an induced subgraph of the graph $H$, and called \emph{full} if $uv\in E(G)$ if and only if $\varphi(u)\varphi(v)\in E(H)$.\qqed
\end{defn}

\begin{thm}\label{thm:bijective-graph-homomorphism}
\cite{Gena-Hahn-Claude-Tardif-1997} A faithful bijective graph homomorphism $\varphi: G \rightarrow H$ is $G \cong H$.
\end{thm}

\begin{thm}\label{thm:homomorphic-two-more-graphs}
$^*$ A graph can be graph homomorphic to two or more graphs that are not isomorphic to each other.
\end{thm}

\begin{problem}\label{question:444444}
In Theorem \ref{thm:homomorphic-two-more-graphs}, a graph can be graph homomorphic to each graph of a graph set $H_{omo}(G)$, conversely, each graph $L\in H_{omo}(G)$ can be vertex-split into $G$, also, graph anti-homomorphisms. \textbf{Determine} the graph set $H_{omo}(G)$ for a connected graph $G$.
\end{problem}

\begin{defn}\label{defn:gracefully-graph-homomorphism}
\cite{Bing-Yao-Hongyu-Wang-arXiv-2020-homomorphisms, Bing-Yao-Hongyu-Wang-graph-homomorphisms-2020} Let $G\rightarrow H$ be a graph homomorphism from a $(p,q)$-graph $G$ to another $(p\,',q\,')$-graph $H$ based on a coloring $\alpha: V(G) \rightarrow V(H)$ such that each edge $\alpha(u)\alpha(v)\in E(H)$ if and only if each edge $uv\in E(G)$. The graph $G$ admits a total coloring $f$, and the graph $H$ admits a total coloring $g$, so $G\rightarrow H$ is a \emph{totally-colored graph homomorphism}. Write $f(E(G))=\{f(uv):uv\in E(G)\}$, $g(E(H))=\{g(\alpha(u)\alpha(v)):\alpha(u)\alpha(v)\in E(H)\}$. There are constraints as follows:
\begin{asparaenum}[\textbf{\textrm{C}}-1. ]
\item \label{bipartite} the vertex set $V(G)=X\cup Y$ with $X\cap Y=\emptyset$, each edge $uv\in E(G)$ holds $u\in X$ and $v\in Y$ true; and the vertex set $V(H)=X_H\cup Y_H$ with $X_H\cap Y_H=\emptyset$, each edge $\alpha(u)\alpha(v)\in E(G)$ holds $\alpha(u)\in X_H$ and $\alpha(v)\in Y_H$ true;
\item \label{edge-difference} each edge color $f(uv)=|f(u)-f(v)|$ for each edge $uv\in E(G)$, and the edge set $g(\alpha(u)\alpha(v))=|g(\alpha(u))-g(\alpha(v))|$ for each edge $\alpha(u)\alpha(v)\in E(H)$;
\item \label{edge-homomorphism}each edge color $f(uv)=g(\alpha(u)\alpha(v))$ for each edge $uv\in E(G)$;
\item \label{vertex-color-set} the vertex colors $f(x)\in [1,q+1]$ for $x\in V(G)$ and $g(y)\in [1,q\,'+1]$ with $y\in V(H)$;
\item \label{odd-vertex-color-set} the vertex colors $f(x)\in [1,2q+2]$ for $x\in V(G)$ and $g(y)\in [1,2q\,'+2]$ with $y\in V(H)$;
\item \label{grace-color-set} the edge color set $[1,q]=f(E(G))=g(E(H))=[1,q\,']$;
\item \label{odd-grace-color-set} the edge color set $[1,2q-1]^o=f(E(G))=g(E(H))=[1,2q\,'-1]^o$; and
\item \label{set-ordered} the set-ordered constraint $\max f(X)<\min f(Y)$ and $\max g(X_H)<\min g(Y_H)$.
\end{asparaenum}
\noindent \textbf{We say the graph homomorphism $G\rightarrow H$ to be}:
\begin{asparaenum}[(i) ]
\item \emph{bipartite graph homomorphism} if \textbf{C}-\ref{bipartite} holds true.
\item \emph{graceful graph homomorphism} if \textbf{C}-\ref{edge-difference}, \textbf{C}-\ref{edge-homomorphism}, \textbf{C}-\ref{vertex-color-set} and \textbf{C}-\ref{grace-color-set} hold true.
\item \emph{set-ordered graceful graph homomorphism} if \textbf{C}-\ref{bipartite}, \textbf{C}-\ref{edge-difference}, \textbf{C}-\ref{edge-homomorphism}, \textbf{C}-\ref{grace-color-set}, \textbf{C}-\ref{vertex-color-set} and \textbf{C}-\ref{set-ordered} hold true.
\item \emph{odd-graceful graph homomorphism} if \textbf{C}-\ref{bipartite}, \textbf{C}-\ref{edge-difference}, \textbf{C}-\ref{edge-homomorphism}, \textbf{C}-\ref{odd-vertex-color-set} and \textbf{C}-\ref{odd-grace-color-set} hold true.
\item \emph{set-ordered odd-graceful graph homomorphism} if \textbf{C}-\ref{bipartite}, \textbf{C}-\ref{edge-difference}, \textbf{C}-\ref{edge-homomorphism}, \textbf{C}-\ref{odd-vertex-color-set}, \textbf{C}-\ref{odd-grace-color-set} and \textbf{C}-\ref{set-ordered} hold true.\qqed
\end{asparaenum}
\end{defn}

\begin{defn} \label{defn:W-constraint-coloring-graph-homomorphism}
$^*$ A \emph{$W$-constraint colored graph homomorphism} $G\rightarrow _{color}H$ is defined as: A graph $G$ admits a $W$-constraint coloring $F$ and another graph $H$ admits a $W$-constraint coloring $F^*$, and there is a graph homomorphism $\varphi:V(G)\rightarrow V(H)$, such that the $W$-constraint $W[F(u),F(uv),F(v)]=0$ holds true if and only if the $W$-constraint $W[F^*(\varphi(u))$, $F^*(\varphi(u)\varphi(v))$, $F^*(\varphi(v))]=0$ holds true.\qqed
\end{defn}

\begin{thm}\label{thm:graph-set-graph-isomorphism}
$^*$ Each totally colored and connected graph $H$ corresponds a totally colored graph set $G_{raph}(H)$, such that each totally colored graph $T\in G_{raph}(H)$ is totally colored graph isomorphism to $H$, namely, $T\rightarrow _{color}H$.
\end{thm}

\begin{thm}\label{thm:2222222222222}
$^*$ Suppose that a graph $G$ admits a $W$-constraint coloring $f$ and another graph $H$ admits a $W$-constraint coloring $h$, and there is a graph homomorphism $\varphi:V(G)\rightarrow V(H)$, such that $G\rightarrow _{color}H$. If there is another graph homomorphism $\phi:V(H)\rightarrow V(G)$, such that $H\rightarrow _{color}G$, then $G\cong H$ with $V(G)=V(H)$ and $E(H)=E(G)$, such that $f(x)=h(x)$ for each vertex $x\in V(H)=V(G)$ and $f(uv)=h(uv)$ for each edge $uv\in E(H)=E(G)$.
\end{thm}

\section{Colorings And Labelings Based On Sets}

\begin{defn} \label{defn:matching-set-set-colorings}
$^*$ Let $S$ be a set, and its elements are all sets, so we call $S$ \emph{set-set}. A graph $G$ admits a \emph{set-coloring} $\alpha: X\rightarrow S$ to be \emph{full} if the color set $\alpha(X)=S$, where $X$ is a subset of the total set $V(G)\cup E(G)$.

If $\alpha$ is not full, namely, $\alpha(X)\subset S$, and there is another graph $H$ admitting a set-coloring $\beta: Y\rightarrow S$ with $Y\subset V(H)\cup E(H)$ and its color set $\beta(Y)\subset S$, such that two color sets $\alpha(X)\cup \beta(Y)=S$, then two set-colorings $\alpha$ and $\beta$ are a matching of colorings based on the set-set $S$, and two graphs $G$ (as a \emph{private topological signature}) and $H$ (as a \emph{public topological signature}) are matching from each other based on the set-sets.\qqed
\end{defn}

\subsection{Set-colorings}

\begin{defn}\label{defn:55-set-labeling}
\cite{Yao-Sun-Zhang-Mu-Sun-Wang-Su-Zhang-Yang-Yang-2018arXiv} Let $G$ be a $(p,q)$-graph, and $[0,p+q]^2$ be the power set of the integer set $[0,p+q]$.

(i) A total set-coloring $F: V(G)\cup E(G)\rightarrow [0, p+q]^2$ is called \emph{total set-labeling} of the graph $G$ if two sets $F(x)\neq F(y)$ for distinct elements $x,y\in V(G)\cup E(G)$.

(ii) A vertex set-coloring $F: V(G) \rightarrow [0, p+q]^2$ is called \emph{vertex set-labeling} of the graph $G$ if two sets $F(x)\neq F(y)$ for distinct vertices $x,y\in V(G)$.

(iii) An edge set-coloring $F: E(G) \rightarrow [0, p+q]^2$ is called \emph{edge set-labeling} of the graph $G$ if two sets $F(uv)\neq F(xy)$ for distinct edges $uv, xy\in E(G)$.

(iv) A vertex set-coloring $F: V(G) \rightarrow [0, p+q]^2$ and a proper edge coloring $g: E(G) \rightarrow [a, b]$ are called \emph{v-set e-proper labeling $(F,g)$} of the graph $G$ if two sets $F(x)\neq F(y)$ for distinct vertices $x,y\in V(G)$ and two edge colors $g(uv)\neq g(wz)$ for distinct edges $uv, wz\in E(G)$.

(v) An edge set-coloring $F: E(G) \rightarrow [0, p+q]^2$ and a proper vertex coloring $f: V(G) \rightarrow [a,b]$ are called \emph{e-set v-proper labeling $(F,f)$} of the graph $G$ if two edges sets $F(uv)\neq F(wz)$ for distinct edges $uv, wz\in E(G)$ and two vertex colors $f(x)\neq f(y)$ for distinct vertices $x,y\in V(G)$.\qqed
\end{defn}

\begin{defn}\label{defn:new-set-colorings}
\cite{Yao-Wang-2106-15254v1} Let $G$ be a $(p,q)$-graph, and let ``$W$-constraint'' be one of constraints on the existing graph colorings and graph labelings of graph theory, and the set $[0,p+q]^2$ be the \emph{power set} of subsets of the consecutive integer set $[0,p+q]$.

(i) A \emph{$W$-constraint ve-set-coloring} $F$ of the graph $G$ holds $F: V(G)\cup E(G)\rightarrow [0, p+q]^2$ such that two sets $F(x)\neq F(y)$ for two adjacent or incident elements $x,y\in V(G)\cup E(G)$ holding the $W$-constraint $W[F(u), F(uv), F(v)]=0$ for each edge $uv\in E(G)$.

(ii) A \emph{$W$-constraint v-set-coloring} $F$ of the graph $G$ holds $F: V(G) \rightarrow [0, p+q]^2$ such that two sets $F(x)\neq F(y)$ for each edge $xy\in E(G)$ holding the $W$-constraint $W[F(u), F(v)]=0$ for each edge $uv\in E(G)$.

(iii) A \emph{$W$-constraint e-set-coloring} $F$ of the graph $G$ holds $F: E(G) \rightarrow [0, p+q]^2$ such that two sets $F(uv)\neq F(uw)$ for two adjacent edges $uv, uw\in E(G)$ holding the $W$-constraint.

(iv) An \emph{e-proper $W$-constraint v-set-coloring} $(F,g)$ of the graph $G$ is consisted of a vertex set-coloring $F: V(G) \rightarrow [0, p+q]^2$ and a proper edge coloring $g: E(G) \rightarrow [a, b]$ such that two sets $F(x)\neq F(y)$ for each edge $xy\in E(G)$ and two adjacent edge colors $g(uv)\neq g(uw)$ for two adjacent edges $uv, uw\in E(G)$ holding the $W$-constraint $W[F(u), g(uv), F(v)]=0$ for each edge $uv\in E(G)$.

(v) A \emph{v-proper $W$-constraint e-set-coloring} $(F,f)$ of the graph $G$ is consisted of an edge set-coloring $F: E(G) \rightarrow [0, p+q]^2$ and a proper vertex coloring $f: V(G) \rightarrow [a,b]$, such that two sets $F(uv)\neq F(uw)$ for two adjacent edges $uv, uw\in E(G)$, and $f(x)\neq f(y)$ for each edge $xy\in E(G)$, as well as the $W$-constraint $W[f(u), F(uv), f(v)]=0$ for each edge $uv\in E(G)$.\qqed
\end{defn}

\begin{rem}\label{rem:333333}
Suppose that a $(p,q)$-graph $G$ admits a $W$-constraint ve-set-coloring $F: V(G)\cup E(G)\rightarrow [0, p+q]^2$ defined in Definition \ref{defn:new-set-colorings}. $T_{code}(G,F)$ derives $(3q)!$ set-based strings $S_i=C_{i,1}C_{i,2}\cdots C_{i,3q}$ with $i\in [1,(3q)!]$, where each $C_{i,j}$ is a number-based set $C_{i,j}=\{a_{i,j,1},a_{i,j,2}$, $\dots $, $a_{i,j,b(i,j)}\}$ with $b(i,j)\geq 1$. A set-based strings $S_i$ exports $\prod^{3q}_{k=1}b(i,j)!$ number-based strings. Thereby, the set-coloring Topcode-matrix $T_{code}(G,F)$ derives $(3q)!\prod^{3q}_{k=1}b(i,j)!$ number-based strings, in total. \qqed
\end{rem}

\begin{defn} \label{defn:55-v-set-e-proper-more-labelings}
\cite{Yao-Wang-2106-15254v1} A \emph{v-set e-proper $W$-constraint labeling} (resp. $\varepsilon$-\emph{coloring}) of a $(p,q)$-graph $G$ is a total coloring $f:V(G)\cup E(G)\rightarrow \Omega$, where $\Omega$ consists of numbers and sets, such that $f(u)$ is a set for each vertex $u\in V(G)$, and the edge color $f(xy)$ for each edge $xy\in E(G)$ is a number, and the edge color set $f(E(G))$ satisfies the given $W$-constraint $W[f(u), f(uv), f(v)]=0$ for each edge $uv\in E(G)$.\qqed
\end{defn}

\subsection{Set-labeling}

\begin{defn}\label{defn:graceful-intersection}
\cite{Yao-Zhang-Sun-Mu-Sun-Wang-Wang-Ma-Su-Yang-Yang-Zhang-2018arXiv} Suppose that a $(p,q)$-graph $G$ admits a set-labeling $F:V(G)\rightarrow [1,q]^2$~(resp. $[1,2q-1]^2)$, and induces an edge set-color $F(uv)=F(u)\cap F(v)$ for each edge $uv\in E(G)$. If we select a \emph{representative} $a_{uv}\in F(uv)$ for each edge color set $F(uv)$ such that $\{a_{uv}:~uv\in E(G)\}=[1,q]$ (resp. $[1,2q-1]^o$), then $F$ is called \emph{graceful-intersection (resp. odd-graceful-intersection) total set-labeling} of the graph $G$.\qqed
\end{defn}

\begin{thm}\label{thm:total-set-labelings}
\cite{Yao-Zhang-Sun-Mu-Sun-Wang-Wang-Ma-Su-Yang-Yang-Zhang-2018arXiv} Each tree $T$ admits a \emph{graceful-intersection (resp. an odd-graceful-intersection) total set-labeling} (see an example shown in Fig.\ref{fig:0-graceful-intersection}(a)).
\end{thm}

\begin{thm}\label{thm:rainbow-total-set-labelings}
\cite{Yao-Zhang-Sun-Mu-Sun-Wang-Wang-Ma-Su-Yang-Yang-Zhang-2018arXiv} Each tree $T$ of $q$ edges admits a \emph{regular rainbow intersection total set-labeling} based on a \emph{regular rainbow set-sequence} $\{[1,k]\}^{q}_{k=1}$ (see an example shown in Fig.\ref{fig:0-graceful-intersection}(b)).
\end{thm}

\begin{figure}[h]
\centering
\includegraphics[width=16cm]{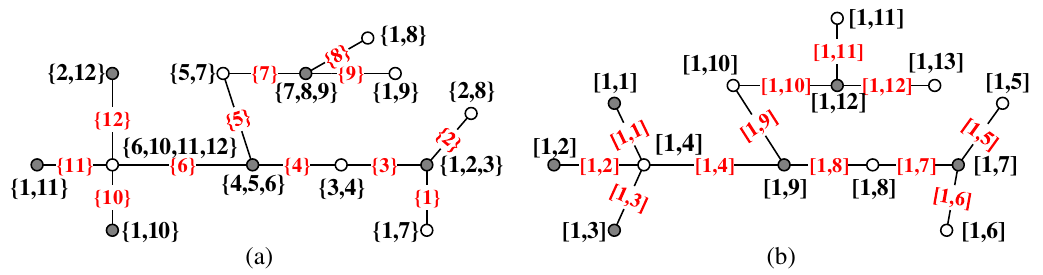}\\
\caption{\label{fig:0-graceful-intersection}{\small Left tree admits a graceful-intersection total set-labeling for illustrating Theorem \ref{thm:total-set-labelings}; Right tree admits a regular rainbow intersection total set-labeling for illustrating Theorem \ref{thm:rainbow-total-set-labelings}, cited from \cite{Yao-Zhang-Sun-Mu-Sun-Wang-Wang-Ma-Su-Yang-Yang-Zhang-2018arXiv}.}}
\end{figure}

\begin{defn} \label{defn:set-coloring-definitions}
\cite{Yao-Sun-Zhang-Li-Yan-Zhang-Wang-ITOEC-2017} Let a $(p,q)$-graph $G$ with integers $q\geq p-1\geq 2$ admit a set-coloring $F: X\rightarrow S$, where $X$ is a subset of $V(G)\cup E(G)$, $S$ is a subset of the power set $[0,pq]^2$ of a consecutive integer set $[0,pq]$, and let $R_{est}(c_1,c_2,\dots, c_m)$ be a constraint set. There are the following constraints:
\begin{asparaenum}[(a)]
\item \label{vertex-set} $X=V(G)$;
\item \label{edge-set} $X=E(G)$;
\item \label{total-set} $X=V(G)\cup E(G)$;
\item \label{adjacent-vertex-labeling} $F(u)\not =F(v)$ if each edge $uv\in E(G)$ (it may happen $F(u)\cap F(v)\neq \emptyset$);
\item \label{adjacent-edge-labeling} $F(uv)\not =F(uw)$ for any pair of adjacent edges $uv$ and $uw$ of the graph $G$ (it may happen $F(uv)\cap F(uw)\neq \emptyset$);
\item \label{vertex-labeling} $|F(V(G))|=p$, also, $F(x)\not =F(y)$ for any pair of vertices $x$ and $y$ of the graph $G$;
\item \label{edge-labeling} $|F(E(G))|=q$, so $F(xy)\not =F(uv)$ for distinct edges $uv$ and $xy$ of the graph $G$;
\item \label{edge-induced} An edge coloring $F\,': E(G)\rightarrow S$ is induced by $F$ subject to a constraint set $R_{est}(c_1,c_2,\dots, c_m)$, that is, each edge $uv\in E(G)$ is colored by the set $F\,'(uv)$ such that each $c\in F\,'(uv)$ is generated by some $a\in F(u)$, $b\in F(v)$ and holds one constraint or more constraints of $R_{est}(c_1,c_2,\dots, c_m)$;
\item \label{induced-edge-labeling} $|F\,'(E(G))|=q$ by the definition of (\ref{edge-induced}).
\end{asparaenum}

\noindent \textbf{We call}:
\begin{asparaenum}[(1)]
\item $F$ \emph{strong vertex set-labeling} of the graph $G$ if both (\ref{vertex-set}) and (\ref{vertex-labeling}) hold true.
\item $F$ \emph{strong edge-set-labeling} of the graph $G$ if both (\ref{edge-set}) and (\ref{edge-labeling}) hold true.
\item $F\,'$ \emph{strongly induced edge-set-labeling} of the graph $G$ if both (\ref{edge-labeling}) and (\ref{edge-induced}) hold true.
\item $F$ \emph{strongly total set-labeling} of the graph $G$ if (\ref{total-set}), (\ref{vertex-labeling}) and (\ref{edge-labeling}) hold true.
\item $(F,F\,')$ \emph{strong set-coloring} subject to a constraint set $R_{est}(c_1,c_2,\dots, c_m)$ if (\ref{vertex-set}), (\ref{vertex-labeling}), (\ref{edge-induced}) and (\ref{induced-edge-labeling}) hold true.
\end{asparaenum}

\begin{asparaenum}[(1')]
\item $F$ \emph{set-labeling} of the graph $G$ if it satisfies (\ref{vertex-set}) and (\ref{adjacent-vertex-labeling}) simultaneously.
\item $F$ an \emph{edge-set-labeling} of the graph $G$ if it satisfies (\ref{edge-set}) and (\ref{adjacent-edge-labeling}) simultaneously.
\item $F$ \emph{total set-coloring} of the graph $G$ if it satisfies (\ref{total-set}), (\ref{adjacent-vertex-labeling}) and (\ref{adjacent-edge-labeling}) simultaneously.
\item $(F,F\,')$ \emph{set-coloring} subject to the constraint set $R_{est}(c_1,c_2,\dots, c_m)$ if (\ref{vertex-set}), (\ref{adjacent-vertex-labeling}), (\ref{adjacent-edge-labeling}) and (\ref{edge-induced}) are true simultaneously.
\end{asparaenum}

\begin{asparaenum}[(1'')]
\item $F$ \emph{pseudo-vertex set-labeling} of the graph $G$ if it holds (\ref{vertex-set}), but not (\ref{adjacent-vertex-labeling}).
\item $F$ \emph{pseudo-edge set-labeling} of the graph $G$ if it holds (\ref{edge-set}), but not (\ref{adjacent-edge-labeling}).
\item $F$ \emph{pseudo-total set-coloring} of the graph $G$ if it holds (\ref{total-set}), but not (\ref{adjacent-vertex-labeling}), or but not (\ref{adjacent-edge-labeling}), or not both (\ref{adjacent-vertex-labeling}) and (\ref{adjacent-edge-labeling}).\qqed
\end{asparaenum}
\end{defn}

Hereafter, we say ``a set-coloring $(F,F\,')$ subject to the constraint set $R_{est}(c_1,c_2,\dots, c_m)$'' defined in Definition \ref{defn:set-coloring-definitions}, and say ``a total set-coloring $\psi$ subject to the constraint set $R_{est}(c_1,c_2,\dots, c_m)$'' defined in Definition \ref{defn:set-coloring-definitions}.

\begin{defn} \label{defn:111111}
\cite{Yao-Ma-arXiv-2201-13354v1} Let $v_s=\min \{|F(x)|:x\in V(G)\}$ and $v_l=\max \{|F(y)|:y\in V(G)\}$ in Definition \ref{defn:set-coloring-definitions}. The coloring $F$ is called \emph{$\alpha$-uniformly vertex set-labeling} of the graph $G$ if $v_s=v_l=\alpha$. Similarly, there are two parameters $e_s=\min \{|F\,'(uv)|:~uv\in E(G)\}$ and $e_l=\max \{|F\,'(xy)|:~xy\in E(G)\}$. The coloring $F\,'$ is called \emph{$\beta$-uniformly edge set-labeling} of the graph $G$ if $e_s=e_l=\beta$. As $\alpha=\beta=1$ above, $(F,F\,')$ is just a popular labeling of graph theory (Ref. \cite{Gallian2022}). For another group of parameters
\begin{equation}\label{eqa:555555}
t_s=\min \{|\psi(x)|:x\in V(G)\cup E(G)\},~t_l=\max \{|\psi(y)|:~y\in V(G)\cup E(G)\}
\end{equation} from Definition \ref{defn:set-coloring-definitions}, and we call $\psi$ \emph{$k$-uniformly total set-coloring} if $k=t_s=t_l$.\qqed
\end{defn}

\begin{rem}\label{rem:more-set-labeling-colorings}
\cite{Yao-Wang-2106-15254v1} For a (strongly) total set-labeling $\psi$ subject to the constraint set $R_{est}(c_1,c_2,\dots, c_m)$ defined in Definition \ref{defn:set-coloring-definitions}, we point out that three numbers $a\in \psi(u)$, $b\in \psi(v)$ and $c\in \psi(uv)$ correspond a constraint $c_i\in R_{est}(c_1,c_2,\dots, c_m)$, by graph colorings (resp. labelings), such that $c_i$ holds one of the following constraints:
\begin{asparaenum}[(a)]
\item \label{odd-graceful} the form $|a-b|=c$ inducing \emph{graceful labelings, or odd-graceful labelings, or odd-elegant labelings, or vertex (distinguishing) coloring if $c\neq 0$}.

\item \label{edge-magic-total} the edge-magic constraint $a+b+c=k$ inducing \emph{edge-magic total labelings} for $k\geq 1$.

\item \label{filicitous} the form $a+b=c$ ($\bmod~\eta$) inducing \emph{felicitous labelings, or harmonious labelings}.

\item \label{edge-magic-graceful} the felicitous-difference constraint $|a+b-c|=k$ inducing \emph{felicitous-difference graceful labelings}.

\item \label{parameters-edge-magic-graceful} the form $|a+b-\lambda c|=k$ inducing \emph{$(k,\lambda)$-edge magic graceful labelings}, or \emph{$(k,\lambda)$-odd-magic graceful labelings}.

\item \label{parameters-edge-magic-total} the form $a+b=k+\lambda c$ inducing \emph{$(k,\lambda)$-magic total labelings}, or \emph{$(k,\lambda)$-odd-magic total labelings}.

\item \label{total-coloring} $a\neq b$, $b\neq c$ and $c\neq a$ induce \emph{total colorings, vertex distinguishing total colorings, list-colorings}.

\item \label{edge-coloring} $c\in \psi(uv)$ and $c\,'\in \psi(uv')$ hold $c\neq c\,'$ inducing \emph{edge colorings}.

\item \label{couple-edge-magic-total} the form $a+b+c=k^+$, or the form $|a+b-c|=k^-$ inducing \emph{$(k^+,k^-)$-couple edge-magic total labelings}.

\item \label{two-magics-magic-graceful} the form $|a+b-c|=k_1$, or the form $|a+b-c|=k_2$ inducing \emph{$(k_1,k_2)$-edge-magic graceful labelings}.\qqed
\end{asparaenum}
\end{rem}

\begin{example}\label{exa:8888888888}
The \emph{first example} is about a \emph{strong set-coloring} $(F,F\,')$ in which the graphical structure is shown in Fig.\ref{fig:set-password-22}(a). We color each vertex $u$ with a set $F(u)$ such that $F(x)\not =F(y)$ for any pair of vertices $x,y$; there are some $a\in F(u)$ and $b\in F(v)$ to hold the unique constraint $|a-b|=c$ subject to $R_{est}(c_1)$ that induces the edge set $F\,'(uv)$ with $c\in F\,'(uv)$ such that $F\,'(uv)\not =F\,'(xy)$ for any pair of edges $uv$ and $xy$.

A \emph{strongly total set-labeling} $\psi$, as the \emph{second example}, is shown in Fig.\ref{fig:set-password-22}(b) with $\psi(x)\not =\psi(y)$ for any two elements $x,y\in V(G)\cup E(G)$, and for an edge $uv$, each $c\in \psi(uv)$ corresponds some $a\in \psi(u)$ and $b\in \psi(v)$ such that at least one of two constraints $|a-b|=c$ and $|a+b-c|=4$ holds true.

The \emph{third example} on a \emph{strongly total set-labeling} $\theta$ subject to the constraint set $R^*_{est}(c_1,c_2,c_3)$ is shown in Fig.\ref{fig:set-password-22}(c), where

$c_1:~|a-b|=c$ for $c\in \theta(uv)$, $a\in \theta(u)$ and $b\in \theta(v)$;

$c_2:~|a\,'+b\,'-c\,'|=4$ for $c\,'\in \theta(uv)$, $a\,'\in \theta(u)$ and $b\,'\in \theta(v)$; and

$c_3:~a\,''+b\,''=c\,''~(\bmod~6)$ for $c\,''\in \theta(uv)$, $a\,''\in \theta(u)$ and $b\,''\in \theta(v)$.

Thereby, each number $c\in \theta(uv)$ corresponds some numbers $a\in \theta(u)$ and $b\in \theta(v)$ such that they hold at least one of three constraints of $R^*_{est}(c_1,c_2,c_3)$.

A \emph{strongly total set-labeling} $(\phi,\phi\,')$ subject to the constraint set $R^*_{est}(c_1,c_2,c_3)$ shown in Fig.\ref{fig:set-password-33} holds: $\theta(x)=\phi (x)$ for any $x\in V(G)$ and $\theta(uv)\subseteq \phi\,'(uv)$ for each edge $uv\in E(G)$, where $\theta$ is defined in Fig.\ref{fig:set-password-22}(c). We say $(\phi,\phi\,')$ to be the \emph{maximally strong set-coloring} subject to the constraint set $R^*_{est}(c_1,c_2,c_3)$.\qqed
\end{example}

\begin{figure}[h]
\centering
\includegraphics[width=15cm]{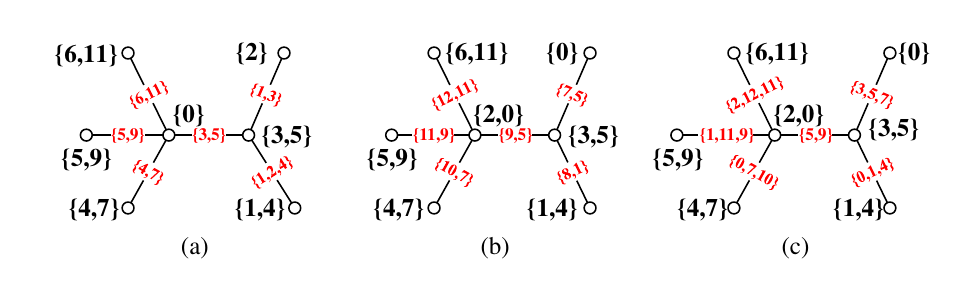}\\
{\small \caption{\label{fig:set-password-22} (a) A strong set-coloring $(F,F\,')$; (b) a strongly total set-labeling $\psi$; (c) another strongly total set-labeling $\theta$, cited from \cite{Yao-Zhang-Sun-Mu-Sun-Wang-Wang-Ma-Su-Yang-Yang-Zhang-2018arXiv}.}}
\end{figure}

\begin{problem}\label{qeu:444444}
By Remark \ref{rem:more-set-labeling-colorings}, suppose that a $(p,q)$-graph $G$ admits a set-labeling $F:V(G)\rightarrow X$, so $F(V(G))=F(V_{=1})\cup F(V_{\geq 2})$ with $V(G)=V_{=1}\cup V_{\geq 2}$ and $V_{=1}\cap V_{\geq 2}=\emptyset $, where
$$
F(V_{=1})=\big \{|F(u)|=1:u\in V_{=1}\big \},\quad F(V_{\geq 2})=\big \{|F(w)|\geq 2:w\in V_{\geq 2}\big \}
$$ \textbf{Find} a $W$-constraint set-coloring $F$ for a graph $G$ holding $|F(V_{=1})|\geq |g(V_{=1})|$ and $|F(V_{\geq 2})|\leq |g(V_{\geq 2})|$ for each $W$-constraint set-labeling $g$ of the graph $G$, here $W$-constraint $\in \{$graceful, odd-graceful, elegant, odd-elegant, edge-magic total, \emph{etc.}$\}$.
\end{problem}

\begin{figure}[h]
\centering
\includegraphics[width=10.6cm]{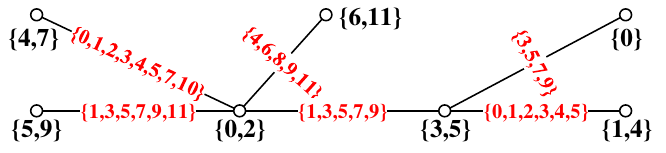}\\
{\small \caption{\label{fig:set-password-33} A strongly set-coloring $(\phi,\phi\,')$ subject to $R^*_{est}(c_1,c_2,c_3)$, cited from \cite{Yao-Zhang-Sun-Mu-Sun-Wang-Wang-Ma-Su-Yang-Yang-Zhang-2018arXiv}.}}
\end{figure}

Since each simple and connected $(p,q)$-graph $G$ can be vertex-split into a tree of $q+1$ vertices, so we have the following results:

\begin{thm} \label{them:v-set-e-proper-W-type-colorings}
\cite{Wang-Wang-Yao2019-Euler-Split} Each simple and connected $(p,q)$-graph $G$ can be vertex-split into a tree $T$ of $q+1$ vertices by the vertex-splitting operation, and admits a \emph{v-set e-proper $W$-constraint coloring} (Ref. Definition \ref{defn:55-set-labeling} and Definition \ref{defn:55-v-set-e-proper-more-labelings}) if $T$ admits a \emph{$W$-constraint coloring}.
\end{thm}

\begin{thm}\label{thm:pseudo-v-set-e-proper-graceful}
\cite{Yao-Mu-Sun-Sun-Zhang-Wang-Su-Zhang-Yang-Zhao-Wang-Ma-Yao-Yang-Xie2019} Every connected graph admits a \emph{v-set e-proper graceful labeling} defined in Definition \ref{defn:55-v-set-e-proper-more-labelings}.
\end{thm}

\begin{thm} \label{them:graph-edge-graceful-set-labelings}
\cite{Yao-Sun-Zhang-Li-Zhao2017} Each simple and connected $(p,q)$-graph $G$ admits a \emph{v-set e-proper graceful coloring} $f: V(G)\rightarrow [0,q]^2$ defined in Definition \ref{defn:55-v-set-e-proper-more-labelings} and Definition \ref{defn:55-set-labeling}, such that each edge $uv$ is colored with an induced edge color $f(uv)=|a_u-b_v|$ for some $a_u\in f(u)$ and $b_v\in f(v)$, and the edge color set $f(E(G))$ $=\{f(uv):uv\in E(G)\}=[1,q]$.
\end{thm}

\begin{defn}\label{defn:total-set-coloring}
\cite{Yao-Sun-Zhang-Li-Zhao2017} A \emph{strongly total set-labeling} $\psi$ of a $(p,q)$-graph $G$ subject to the constraint set $R_{est}(c_1,c_2,\dots, c_m)$ is a total coloring $\psi:V(G)\cup E(G)\rightarrow S$, where $R_{est}(c_1,c_2,\dots, c_m)$ is a given constraint set and $S$ is a set of subsets of the set $[0,pq]$, such that $\psi(x)\not =\psi(y)$ for any pair of elements $x,y\in V(G)\cup E(G)$, and each element of the label set $\psi(uv)$ for each edge $uv\in E(G)$ holds at least one constraint $c_i\in R_{est}(c_1,c_2,\dots, c_m)$.\qqed
\end{defn}

\begin{thm} \label{them:each-graph-set-labelings}
\cite{Yao-Sun-Zhang-Li-Zhao2017} Each simple and connected $(p,q)$-graph $G$ admits a strongly total set-labeling $F$ such that
\begin{equation}\label{eqa:theorem-11}
\max \big \{|F(x)|:~x\in V(G)\cup E(G)\big \}=\Delta(G)+1
\end{equation}
\begin{equation}\label{eqa:theorem-22}
\max \big \{c:~c\in F(x),x\in V(G)\cup E(G)\big \}\in [1,p+q].
\end{equation}
\end{thm}

\begin{thm} \label{them:tree-edge-magic-total-labeling}
\cite{Yao-Sun-Zhang-Li-Zhao2017} If a tree admits a super edge-magic total labeling, then it admits a $2$-uniform strongly total set-coloring.
\end{thm}

\subsection{Connections between colorings and set-colorings}

\begin{defn} \label{defn:111111}
\cite{Yao-Sun-Zhang-Li-Yan-Zhang-Wang-ITOEC-2017} Let $f$ be a proper edge coloring of a graph $G$, and $C_{ne}(x)=\{f(xy):y\in N_{ei}(x)\}$ be the set of colored edges incident with the vertex $x\in V(G)$. If $C_{ne}(x)\neq C_{ne}(w)$ for any pair of distinct vertices $x,w\in V(G)$, we then call $f$ \emph{vertex distinguishing edge coloring} of the graph $G$.\qqed
\end{defn}

\begin{lem} \label{components-cycles}
\cite{Yao-Sun-Zhang-Li-Yan-Zhang-Wang-ITOEC-2017} Each vertex distinguishing edge coloring of a graph $G$ induces a strong set-labeling $F$ with \begin{equation}\label{eqa:degrees}
\Delta(G)=\max \big \{|F(x)|:~x\in V(G)\big \},\quad \delta(G)=\min \big \{|F(x)|:~x\in V(G)\big \}.
\end{equation}
\end{lem}

\begin{thm} \label{thm:theorem11}
\cite{Yao-Sun-Zhang-Li-Yan-Zhang-Wang-ITOEC-2017} If a set-labeling $F$ of a graph $G$ holds: $|F(x)|\geq \textrm{deg}_G(x)$ for each vertex $x\in V(G)$, and $|F(u)\cap F(v)|=1$ for each edge $uv\in E(G)$, and $F(u)\cap F(v)\neq F(u)\cap F(w)$ for two adjacent edges $uv,uw\in E(G)$, then $F$ induces a \emph{proper edge coloring} of the graph $G$.
\end{thm}

\begin{defn} \label{defn:adjacent-1-common-edge-coloring}
\cite{Yao-Sun-Zhang-Li-Yan-Zhang-Wang-ITOEC-2017} An \emph{adjacent $1$-common edge coloring} $f$ of the graph $G$ satisfies: $f$ is a proper edge coloring,
$$
\big |\big \{f(ux_i):x_i\in N_{ei}(u)\big \}\cap \big \{f(vy_j):y_j\in N_{ei}(v)\big \}\big |=1
$$ for each edge $uv\in E(G)$, and $C_{ne}(u)\cup C_{ne}(v)\neq C_{ne}(u)\cup C_{ne}(w)$, also,
$$
\big \{f(ux_i):x_i\in N_{ei}(u)\big \}\cap \big \{f(vy_j):y_j\in N_{ei}(v)\big \}\neq \big \{f(ux_i):x_i\in N_{ei}(u)\big \}\cap \big \{f(wz_k):z_k\in N_{ei}(w)\big \}
$$ for any pair of adjacent edges $uv,uw\in E(G)$. \qqed
\end{defn}

By Theorem \ref{thm:theorem11} we obtain the following result:
\begin{thm} \label{thm:theorem22}
\cite{Yao-Sun-Zhang-Li-Yan-Zhang-Wang-ITOEC-2017} A set-labeling $F$ of a graph $G$ holds: $|F(x)|\geq \textrm{deg}_G(x)$ for each vertex $x\in V(G)$, and $|F(u)\cap F(v)|=1$ for each edge $uv\in E(G)$, and $F(u)\cap F(v)\neq F(u)\cap F(w)$ for any pair of adjacent edges $uv,uw\in E(G)$, then the set-labeling $F$ induces an adjacent $1$-common edge coloring of the graph $G$.
\end{thm}

We, by Definition \ref{defn:adjacent-1-common-edge-coloring}, define a new parameter
\begin{equation}\label{eqa:555555}
\chi\,'_{set}(G)=\min_{f}\max\{f(uv):~uv\in E(G)\}
\end{equation} over all adjacent $1$-common edge colorings of the graph $G$. It is not hard to show that $\chi\,'_{set}(K_n)=\frac{1}{2}n(n-1)$, and
\begin{equation}\label{eqa:555555}
\chi\,'_{set}(G)=\max \big \{\textrm{deg}_G(u)+\textrm{deg}_G(v)-1:~uv\in E(G)\big \}
\end{equation} if the graph $G$ is a tree.

A vertex $w$ of a graph $G_k$ admitting a set-labeling $F_k:V(G_k)\rightarrow S$, where $S$ is a set of subsets of a consecutive integer set $[1,M]$, has the family of $F(w_1),F(w_2),\dots ,F(w_{M(w)})$ with $M(w)=\textrm{deg}_{G_k}(w)$ and $w_i\in N_{ei}(w)$, which satisfies $|X|\leq \big |\bigcup _{i\in X} F(w_i)\big |$ for every subset $X\subset [1,M(w)]$. Excellently, according to the famous Philip Hall's theorem (1935, see Bollob\'{a}s' book \cite{Bela-Bollobas}), there exists a \emph{system of distinct representative} $R_{ep}(w)=\{z_1,z_2,\dots ,z_{M(w)}\}$ with pairwise distinct $z_j\in F(w_j)$ of the family for $j\in [1,M(w]$.

\begin{example}\label{exa:8888888888}
We take a consecutive integer set $X=[1,6]$, and let
$A_1=A_2=\{1,2\}$, $A_3=\{2,3\}$ and $A_4=\{1,4,5,6\}$. For $\mathcal{A}=\{A_1, A_2, A_3,A_4\}$, the set $[1, 4]$ is a
system of distinct representative. If we have a new family $\mathcal{B}$ by adding another set $A_5 =\{2, 3\}$ to $\mathcal{A}$, then it is impossible to find a \emph{transversal} for the family $\mathcal{B}$. Notice that $S=\{1,2,3,5\}\subset [1,5]$. But, $|S|> \big |\bigcup _{i\in S} A_i\big |=|\{1,2,3\}|=3$.\qqed
\end{example}

We show a result as follows:

\begin{thm} \label{thm:box}
\cite{Yao-Sun-Zhang-Li-Yan-Zhang-Wang-ITOEC-2017} Suppose that a graph $G_1$ admits a set-labeling $F_1: V(G_1)\rightarrow S$, where $S$ is a set of subsets of $[1,M]$, and $F_1(u)\cap F_1(v)\neq \emptyset$ for each edge $uv\in E(G_1)$. And there are graphs $G_k=G_{k-1}-u_{k-1}$ for $u_{k-1}\in V(G_{k-1})$ with $k\geq 2$, and for each graph $G_k$ with $k\in [2,|G_1|-2]$, there exists a set-labeling $F_k(x)=F_{k-1}(x)$ if $x\not\in N_{ei}(u_{k-1})$, and
$$F_k(y)=F_{k-1}(y)\setminus [F_{k-1}(y)\cap F_{k-1}(u_{k-1})]
$$ if $y\in N_{ei}(u_{k-1})$; as well as $F_k(u)\cap F_k(v)\neq \emptyset$ for each edge $uv\in E(G_k)$. If two \emph{systems of distinct representatives} for each edge $uv\in E(G_k)$ holds $|R_{ep}(u)\cap R_{ep}(v)|=1$ with $k\in [2,|G_1|-2]$, then $F_1$ induces a proper edge coloring of the graph $G_1$.
\end{thm}

\begin{defn}\label{defn:set-coloring}
\cite{Yao-Ma-arXiv-2201-13354v1} A vertex set-labeling $F$ of a $(p,q)$-graph $G$ is a coloring $F:V(G)\rightarrow S$ such that $F(u)\not =F(v)$ for any pair of vertices $u,v$ of the graph $G$, where $S$ is a set of subsets of the set $[0,pq]$, and $F(x)\not =F(y)$ for any pair of vertices $x$ and $y$ of the graph $G$. An edge set-labeling $F\,'$ induced by the vertex set-labeling $F$ is subjected to a constraint set $R_{est}(c_1,c_2,\dots, c_m)$ based on $F$, such that the edge set-labeling $F\,'(uv)$ of each edge $uv\in E(G)$ holds the constraint set $R_{est}(c_1,c_2,\dots, c_m)$ true, and $F\,'(xy)\not =F\,'(uv)$ for any pair of distinct edges $uv$ and $xy$ of the graph $G$. We call $(F,F\,')$ a \emph{strong set-coloring} subject to the constraint set $R_{est}(c_1,c_2,\dots, c_m)$, and $F\,'$ \emph{induced edge-set-labeling} over $F$.\qqed
\end{defn}

\begin{problem}\label{question:444444}
The set-labelings and set-colorings defined in Definition \ref{defn:set-coloring-definitions} can be optimal in this way: $S$ is the power set of a consecutive integer set $\{0,1,\dots ,\chi_\epsilon (G)\}=[0,\chi_\epsilon (G)]$ such that $G$ admits a set-labeling or a set-coloring defined in Definition \ref{defn:set-coloring-definitions}, and but $S$ is not the power set of any consecutive integer set $[0,M]$ if $M<\chi_\epsilon (G)$, where $\epsilon$ is a combinatoric of some conditions of (a)-(i) stated in Definition \ref{defn:set-coloring-definitions}, and $\chi_\epsilon (G)$ is called an \emph{$\epsilon$-chromatic number} of the graph $G$. For example, $\epsilon=\{$(a), (d)$\}$ if only about a \emph{set-labeling} of the graph $G$. So, we determine the $\epsilon$-chromatic number $\chi_\epsilon (G)$ for a fixed $\epsilon$. As known, there are many long-standing conjectures in graph colorings and graph labelings, so we believe that there are new open problems on the set-colorings, or set-labelings defined in Definition \ref{defn:set-coloring-definitions}.
\end{problem}

\begin{problem}\label{question:444444}
\textbf{Define} mixed set-colorings, or set-labelings of graphs in order to design more complicated topological codes.
\end{problem}

\begin{problem}\label{question:444444}
Construction of lager scale of graphs admitting set-colorings, or set-labelings by smaller size of graphs admitting the same type of set-colorings, or set-labelings. Trees are first object for constructing such graphs.
\end{problem}

\begin{problem}\label{question:444444}
Notice that a new-type of matrices defined by set-colorings, or set-labelings of graphs goes into sight of our research, although we do not know more properties about such matrices, called \emph{set-matrices}.
Thereby, we define a set-matrix for a simple $(p,q)$-graph $G$ admitting an edge-set-labeling $F\,'$ as:
$S_e(G)=(A_{ij})_{q\times q}$ such that
\begin{equation}\label{eqa:555555}
A_{ij}=\left\{
\begin{array}{ll}
F\,'(u_iu_j),& u_iu_j\in E(G);\\
\emptyset,& \textrm{otherwise}.
\end{array}
\right.
\end{equation}

Suppose that a simple $(p,q)$-graph $G$ admits a set-labeling $F$ defined on its vertex set $V(G)$; we define an operation ``$[\bullet]$'' on two sets, and then define a set-matrix $S_v(G)=(B_{ij})_{q\times q}$ of the graph $G$ based on the set-labeling $F$ as:
\begin{equation}\label{eqa:555555}
B_{ij}=\left\{
\begin{array}{ll}
F(u_i)[\bullet]F(u_j),& u_iu_j\in E(G);\\
\emptyset,& \textrm{otherwise}.
\end{array}
\right.
\end{equation} where the result of each operation $F(u_i)[\bullet]F(u_j)$ is still a set.\qqed
\end{problem}

\begin{problem}\label{question:444444}
We change the condition in Definition \ref{defn:adjacent-1-common-edge-coloring} by the following one:
$$
\{f(uu_i):u_i\in N_{ei}(u)\}\cap \{f(vv_j):v_j\in N_{ei}(v)\}\neq \{f(xx_i):x_i\in N_{ei}(x)\}\cap \{f(yy_j):y_j\in N_{ei}(y)\}
$$ for any two edges $uv,xy\in E(G)$ and keep other conditions in original, then we obtain the \emph{vertex 1-common-edge-coloring} of a graph $G$. By the distinguishing total colorings introduced in \cite{Yang-Yao-Ren-2016}, we can define a \emph{set-set-coloring} $F^*$ of a graph $G$ such that each vertex $u$ of the graph $G$ is colored by a set $F^*(u)$ consisted of sets $F_i(u)$ with $i\in [1,u_k]$. Studying set-set-colorings of graphs is a new topic in hypergraphs.
\end{problem}

\begin{problem}\label{question:444444}
In \cite{N-K-Sudev-2015}, Sudev defined a set-coloring $(F,F\,')$ of a graph $H$ such that $F\,'(uv)=F\,'(xy)$ for any pair of edges $uv$ and $xy$ of the graph $H$. It may be interesting to find graphs admitting the set-coloring $(F,F\,')$ mentioned above. Obviously, finding such graphs is a challenge with many unknown parts.
\end{problem}

\begin{problem}\label{question:444444}
We focus on particular set-colorings, or set-labelings of a graph $G$, such as:

(i) No two edges $uv$ and $xy$ of the graph $G$ hold $|F\,'(uv)|=|F\,'(xy)|$ in a set-coloring $(F,F\,')$ of the graph $G$.

(ii) No two vertices $x$ and $y$ of the graph $G$ hold $|F(x)|=|F(y)|$ in a set-labeling $F$ of the graph $G$.

(iii) No two elements $\alpha$ and $\beta$ of $V(G)\cup E(G)$ hold $|F(\alpha)|=|F(\beta)|$ in a total set-coloring $F$.

(iv) A graph $G$ admits a set-coloring $(F,F\,')$ subject to two different constraint sets $R_{set}(m_1)$ and $R_{set}(m_2)$, respectively, or more constraint sets. Conversely, the graph $G$ admits two different set-colorings $(F_1,F\,'_1)$ and $(F_2,F\,'_2)$ subject to a constraint set $R_{est}(c_1,c_2,\dots, c_m)$ only.
\end{problem}

\begin{problem}\label{question:444444}
For a simple and connected graph $G$ admitting adjacent $1$-common edge colorings, \textbf{determine} the parameter $\chi\,'_{set}(G)=\min_{f}\max\{f(uv):~uv\in E(G)\}$ over all adjacent $1$-common edge colorings of the connected graph $G$ (Ref. Definition \ref{defn:adjacent-1-common-edge-coloring}). It may be possible to consider $\chi\,'_{set}(G)$ if $G$ is a planar graph.
\end{problem}

\subsection{Parameterized colorings}

\subsubsection{Traditional parameterized colorings}

\begin{defn} \label{defn:kd-w-type-colorings}
\cite{Yao-Su-Wang-Hui-Sun-ITAIC2020} Let $G$ be a bipartite and connected $(p,q)$-graph, then its vertex set $V(G)=X\cup Y$ with $X\cap Y=\emptyset$ such that each edge $uv\in E(G)$ holds $u\in X$ and $v\in Y$. Let integers $a,k,m\geq 0$, $d\geq 1$ and $q\geq 1$. We have two parameterized sets
\begin{equation}\label{eqa:555555}
{
\begin{split}
S_{m,k,a,d}=& \big \{k+ad,k+(a+1)d,\dots ,k+(a+m)d\big \},\\
O_{2q-1,k,d}=& \big \{k+d,k+3d,\dots ,k+(2q-1)d\big \}
\end{split}}
\end{equation} with two \emph{cardinalities} $|S_{m,k,a,d}|=m+1$ and $|O_{2q-1,k,d}|=q$. Suppose that the bipartite and connected $(p,q)$-graph $G$ admits a coloring
\begin{equation}\label{eqa:555555}
f:X\rightarrow S_{m,0,0,d}=\big \{0,d,\dots ,md\big \},~f:Y\cup E(G)\rightarrow S_{n,k,0,d}=\big \{k,k+d,\dots ,k+nd\big \}
\end{equation} with integers $k\geq 0$ and $d\geq 1$, here it is allowed $f(x)=f(y)$ for some distinct vertices $x,y\in V(G)$. Let $c$ be a non-negative integer. We define the following \emph{parameterized colorings}:
\begin{asparaenum}[\textbf{\textrm{Ptol}}-1. ]
\item If edge color $f(uv)=|f(u)-f(v)|$ for each edge $uv\in E(G)$, and two color sets
\begin{equation}\label{eqa:555555}
f(E(G))=S_{q-1,k,0,d},\quad f(V(G)\cup E(G))\subseteq S_{m,0,0,d}\cup S_{q-1,k,0,d}
\end{equation} then $f$ is called a \emph{$(k,d)$-gracefully total coloring}; and moreover $f$ is called a \emph{strongly $(k,d)$-graceful total coloring} if $f(x)+f(y)=k+(q-1)d$ for each matching edge $xy$ of a matching $M$ of the graph $G$.
\item If edge color $f(uv)=|f(u)-f(v)|$ for each edge $uv\in E(G)$,
$$f(E(G))=O_{2q-1,k,d},~f(V(G)\cup E(G))\subseteq S_{m,0,0,d}\cup S_{2q-1,k,0,d}$$ then $f$ is called a \emph{$(k,d)$-odd-gracefully total coloring}; and moreover $f$ is called a \emph{strongly $(k,d)$-odd-graceful total coloring} if $f(x)+f(y)=k+(2q-1)d$ for each matching edge $xy$ of a matching $M$ of the graph $G$.
\item If there is a color set
$$\big \{f(u)+f(uv)+f(v):uv\in E(G)\big \}=\big \{2k+2ad,2k+2(a+1)d,\dots ,2k+2(a+q-1)d\big \}
$$ with $a\geq 0$ and the total color set $f(V(G)\cup E(G))\subseteq S_{m,0,0,d}\cup S_{2(a+q-1),k,a,d}$, then $f$ is called a \emph{$(k,d)$-edge antimagic total coloring}.
\item If edge color $f(uv)=f(u)+f(v)~(\bmod^*qd)$ defined by
\begin{equation}\label{eqa:555555}
f(uv)-k=\big [f(u)+f(v)-k\big ]~(\bmod ~qd),~uv\in E(G)
\end{equation} and the edge color set $f(E(G))=S_{q-1,k,0,d}$, then we call $f$ \emph{$(k,d)$-harmonious total coloring}.
\item If edge color $f(uv)=f(u)+f(v)~(\bmod^*qd)$ defined by $f(uv)-k=[f(u)+f(v)-k](\bmod ~qd)$ for each edge $uv\in E(G)$, and the edge color set $f(E(G))=O_{2q-1,k,d}$, then we call $f$ \emph{$(k,d)$-odd-elegant total coloring}.
\item If \emph{edge-magic constraint} $f(u)+f(uv)+f(v)=c$ for each edge $uv\in E(G)$, the edge color set $f(E(G))=S_{q-1,k,0,d}$, and the vertex color set $f(V(G))\subseteq S_{m,0,0,d}\cup S_{q-1,k,0,d}$, then $f$ is called \emph{strongly edge-magic $(k,d)$-total coloring}; and moreover $f$ is called \emph{edge-magic $(k,d)$-total coloring} if the cardinality $|f(E(G))|\leq q$ and $f(u)+f(uv)+f(v)=c$ for each edge $uv\in E(G)$.
\item If \emph{edge-difference constraint} $f(uv)+|f(u)-f(v)|=c$ for each edge $uv\in E(G)$ and the edge color set $f(E(G))=S_{q-1,k,0,d}$, then $f$ is called \emph{strongly edge-difference $(k,d)$-total coloring}; and moreover $f$ is called \emph{edge-difference $(k,d)$-total coloring} if the cardinality $|f(E(G))|\leq q$ and $f(uv)+|f(u)-f(v)|=c$ for each edge $uv\in E(G)$.
\item If \emph{felicitous-difference constraint} $|f(u)+f(v)-f(uv)|=c$ for each edge $uv\in E(G)$ and the edge color set $f(E(G))=S_{q-1,k,0,d}$, then $f$ is called \emph{strongly felicitous-difference $(k,d)$-total coloring}; and moreover we call $f$ \emph{felicitous-difference $(k,d)$-total coloring} if the cardinality $|f(E(G))|\leq q$ and $\big |f(u)+f(v)-f(uv)\big |=c$ for each edge $uv\in E(G)$.
\item If \emph{graceful-difference constraint} $\big ||f(u)-f(v)|-f(uv)\big |=c$ for each edge $uv\in E(G)$ and the edge color set $f(E(G))=S_{q-1,k,0,d}$, then we call $f$ to be \emph{strongly graceful-difference $(k,d)$-total coloring}; and we call $f$ \emph{graceful-difference $(k,d)$-total coloring} if the cardinality $|f(E(G))|\leq q$ and $\big ||f(u)-f(v)|-f(uv)\big |=c$ for each edge $uv\in E(G)$.\qqed
\end{asparaenum}
\end{defn}

\begin{figure}[h]
\centering
\includegraphics[width=16.4cm]{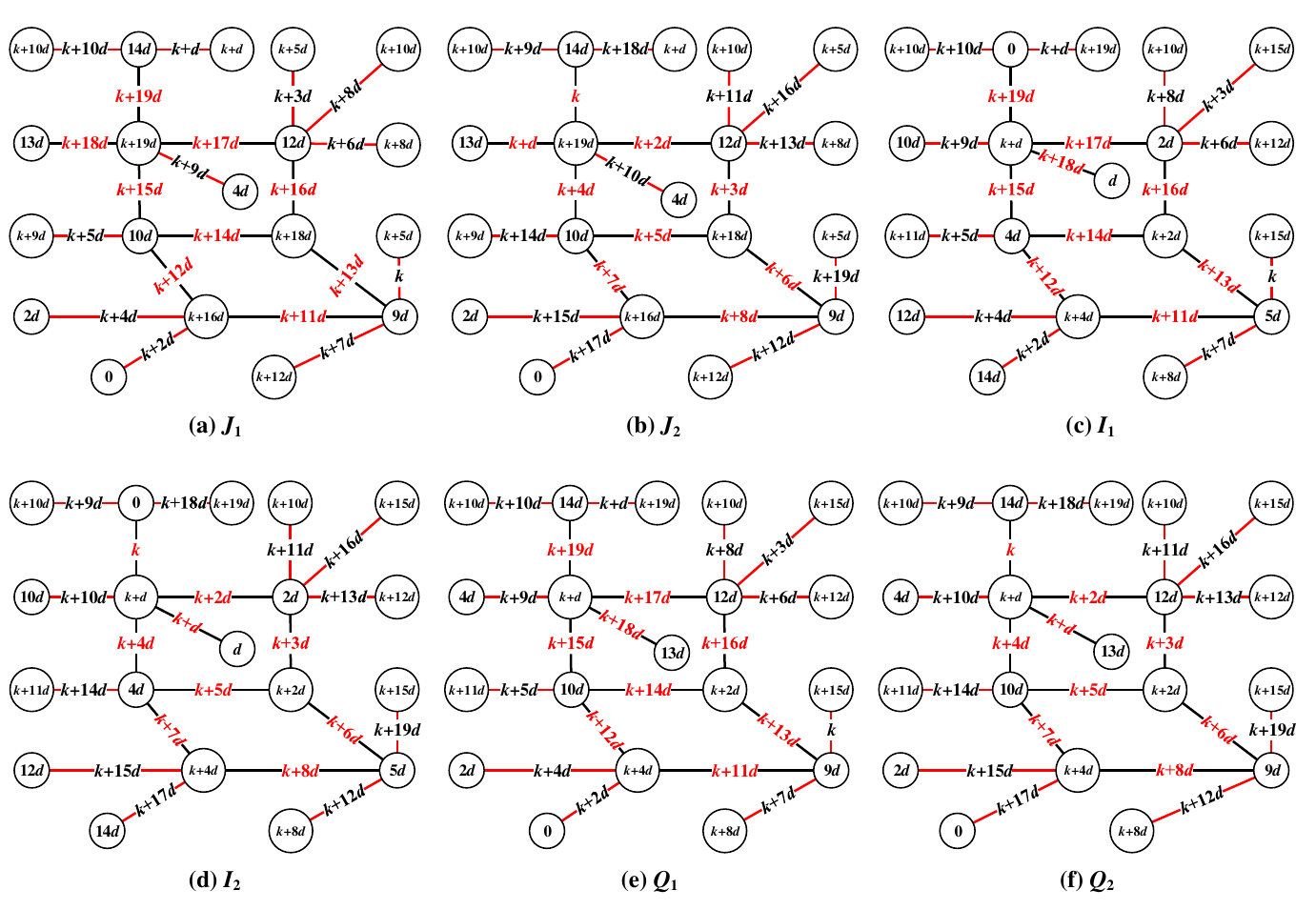}\\
\caption{\label{fig:66-parameter-labelings}{\small A scheme for understanding some parameterized colorings and labelings defined in Definition \ref{defn:kd-w-type-colorings}.}}
\end{figure}

\begin{example}\label{exa:8888888888}
By Definition \ref{defn:kd-w-type-colorings} and Fig.\ref{fig:66-parameter-labelings}, we have

(i) The $(k,d)$-total colored graph $J_1$ admits a \emph{felicitous-difference $(k,d)$-total coloring} $f_1$ holding $\big |f_1(u)+f_1(v)-f_1(uv)\big |=14d$ for each edge $uv\in E(J_1)$.

(ii) The $(k,d)$-total colored graph $J_2$ admits an \emph{edge-magic $(k,d)$-total coloring} $f_2$ holding $f_2(u)+f_2(uv)+f_2(v)=2k+33d$ for each edge $uv\in E(J_2)$.

(iii) The $(k,d)$-total colored graph $I_1$ admits an \emph{edge-magic $(k,d)$-total coloring} $h_1$ holding $h_1(u)+h_1(uv)+h_1(v)=2k+20d$ for each edge $uv\in E(I_1)$.

(iv) The $(k,d)$-total colored graph $I_2$ admits a \emph{felicitous-difference $(k,d)$-total coloring} $h_2$ holding $\big |h_2(u)+h_2(v)-h_2(uv)\big |=d$ for each edge $uv\in E(I_2)$.

(v) The $(k,d)$-total colored graph $Q_1$ admits an \emph{edge-difference $(k,d)$-total coloring} $g_1$ holding $g_1(uv)+\big |g_1(u)-g_1(v)\big |=2k+6d$ for each edge $uv\in E(Q_1)$.

(vi) The $(k,d)$-total colored graph $Q_2$ admits a \emph{pan-edge-difference $(k,d)$-total coloring} $g_2$.
\qqed
\end{example}

\begin{rem}\label{rem:333333}
In Definition \ref{defn:kd-w-type-colorings}, we have four \emph{magic-constraints}: the edge-magic constraint, the edge-difference constraint, the felicitous-difference constraint and the graceful-difference constraint.\qqed
\end{rem}

\begin{thm}\label{thm:tree-graceful-total-coloringss}
\cite{Yao-Su-Wang-Hui-Sun-ITAIC2020} Each tree admits a $(k,d)$-gracefully total coloring defined in Definition \ref{defn:kd-w-type-colorings}, also, a set-ordered gracefully total coloring as $(k,d)=(1,1)$, and a set-ordered odd-gracefully total coloring as $(k,d)=(1,2)$.
\end{thm}

\subsubsection{Parameterized Topcode-matrices}

If there is no confusion, we omit ``$3\times q$ order'' in the following discussion, or add a sentence `` the Topcode-matrices $I$, $T_{code}(G)$ and $P_{(k,d)}(G)$ have the same order''. For \emph{bipartite graphs}, especially, we define the \emph{unite Topcode-matrix} as follows
\begin{equation}\label{eqa:unit-Topcode-matrix}
{
\begin{split}
I\,^0=\left(
\begin{array}{ccccc}
0 & 0 & \cdots & 0\\
1 & 1 & \cdots & 1\\
1 & 1 & \cdots & 1
\end{array}
\right)_{3\times q}=(X\,^0,~E\,^0,~Y\,^0)^T
\end{split}}
\end{equation} with two vertex-vectors $X\,^0=(0, 0, \dots ,0)_{1\times q}$ and $Y\,^0=(1, 1, \dots ,1)_{1\times q}$, and the edge-vector $E\,^0=(1, 1, \dots ,1)_{1\times q}$.

\begin{defn} \label{defn:bipartite-parameterized-topcode-matrix}
\cite{Bing-Yao-arXiv:2207-03381} Let $G$ be a bipartite $(p,q)$-graph with $V(G)=X\cup Y$ and $X\cap Y=\emptyset$, and let $k,d$ be non-negative integers. If $G$ admits a set-ordered $W$-constraint coloring $f$, that is the \emph{set-ordered constraint} $\max f(X)<\min f(Y)$, so we get a \emph{parameterized Topcode-matrix} defined by
\begin{equation}\label{eqa:definition-parameterized-topcode-matrix}
{
\begin{split}
P_{ara}(G,F|k,d)&=k\cdot I\,^0+d\cdot T_{code}(G,f)\\
&=\left(
\begin{array}{rrrrr}
f(u_1)d & f(u_2)d & \cdots & f(u_q)d\\
k+f(u_1v_1)d & k+f(u_2v_2)d & \cdots & k+f(u_qv_q)d\\
k+f(v_1)d & k+f(v_2)d & \cdots & k+f(v_q)d
\end{array}
\right)
\end{split}}
\end{equation}
where three Topcode-matrices $I\,^0$, $T_{code}(G,f)$ and $P_{ara}(G,F|k,d)$ have the same $3\times q$ order, and $F$ is a \emph{$W$-constraint parameterized coloring} of the bipartite $(p,q)$-graph $G$, as well as
\begin{equation}\label{eqa:set-type-topcode-matrix}
\centering
{
\begin{split}
T_{code}(G,f)= \left(
\begin{array}{ccccc}
f(u_1) & f(u_2) & \cdots & f(u_q)\\
f(u_1v_1) & f(u_2v_2) & \cdots & f(u_qv_q)\\
f(v_1) & f(v_2) & \cdots & f(v_q)
\end{array}
\right)
\end{split}}
\end{equation} hoding the $W$-constraint $f(u_kv_k)=W\langle f(u_k), f(v_k)\rangle $ for each edge $u_kv_k\in E(G)$ with $u_k\in X$ and $v_k\in Y$.\qqed
\end{defn}

\begin{defn}\label{defn:pan-Topcode-matrix}
\cite{Yao-Wang-Ma-Su-Wang-Sun-2020ITNEC} A \emph{pan-Topcode-matrix} is defined as $P_{code}=(X_{pan}, E_{pan}, Y_{pan}~)^{T}$ with three vectors
$$X_{pan}=(\alpha_1, \alpha_2, \dots , \alpha_q),~E_{pan}=(\gamma_1, \gamma_2, \dots , \gamma_q),~Y_{pan}=(\beta_1, \beta_2, \dots , \beta_q)$$ and $\alpha_j,\beta_j$ are the ends of $\gamma_j$. If there exits a constraint $W$ such that $\gamma_j=W\langle \alpha_j,\beta_j\rangle$ for each $j\in [1,q]$, then the pan-Topcode-matrix $P_{code}$ is \emph{$W$-constraint valued}.\qqed
\end{defn}

\begin{defn} \label{defn:evaluated-topcode-matrix}
\cite{Bing-et-al-arXiv-asymmetric-4520331} If $x_i:=\alpha_i$, $e_i:=\gamma_i$ and $y_i:=\beta_i$ in a Topcode-matrix $T_{code}$ defined in Definition \ref{defn:topcode-matrix-definition}, we get another Topcode-matrix $T^{evalu}_{code}=\big (X_{(:)},~E_{(:)},~Y_{(:)}\big )^{T}$ withe three vectors
$$X_{(:)}=(\alpha_1, \alpha_2, \dots, \alpha_q),~E_{(:)}=(\gamma_1, \gamma_2, \dots, \gamma_q)\textrm{ and }Y_{(:)}=(\beta_1, \beta_2, \dots, \beta_q)
$$ We call the matrix $T^{evalu}_{code}$ \emph{assignment Topcode-matrix} of $T_{code}$, and denote this face as $T_{code}:=T^{evalu}_{code}$. Moreover, there are Topcode-matrices $T^{evalu}_{code}(1)$, $T^{evalu}_{code}(2)$, $\dots$, $T^{evalu}_{code}(m)$ holding
\begin{equation}\label{eqa:555555}
T^{evalu}_{code}(k):=T^{evalu}_{code}(k+1)
\end{equation} for $k\in [1,m-1]$.\qqed
\end{defn}

\begin{rem}\label{rem:333333}
In Definition \ref{defn:pan-Topcode-matrix}, the elements $\alpha_i,\gamma_i,\beta_i$ of the pan-Topcode-matrix $P_{code}$ are graphs, matrices, vectors, strings, formulae, articles, any things if there are connections be tween them, then the pan-Topcode-matrices show these related things in topological structures.

The generalization $T^{gener}_{code}$ of a Topcode-matrix $T_{code}$ is that each of elements in the Topcode-matrix is a \emph{thing} in the world, such that the Topcode-matrix $T^{gener}_{code}$ brings these $3q$ things together topologically by a mathematical constraint, or a group of mathematical constraints for getting a complete ``mathematical story''.\qqed
\end{rem}

\begin{rem}\label{rem:333333}
\textbf{The assignment Topcode-matrices.} By Definition \ref{defn:evaluated-topcode-matrix}, we have a \emph{assignment Topcode-matrix} $T_{code}(G,f):=P_{ara}(G,F|k,d)$ defined in Eq.(\ref{eqa:definition-parameterized-topcode-matrix}), which converts a parameterized number-based string \begin{equation}\label{eqa:strings-plane-kd-coordinates}
s(k,d)=c_{1}(k,d)c_{2}(k,d)\cdots c_{3q}(k,d)
\end{equation} with longer bytes made by $P_{ara}(G,F|k,d)$ defined in Definition \ref{defn:bipartite-parameterized-topcode-matrix} to a string $s=c_{1}c_{2}\cdots c_{3q}$ with shorter bytes made by $T_{code}(G,f)$.

\textbf{The fractional strings.} The parameterized Topcode-matrix $P_{ara}(G,F|k,d)$ is useful in the discussion of fractional strings. The limitation $(k,d)\rightarrow (k_0,d_0)$ enables us to induce real-valued strings. For example, a \emph{fractional $(k_n,d_n)$-string} $s^*_n=c^*_{n,1}c^*_{n,2}\cdots c^*_{n,m}$ holds:

(i) there is at least one $c^*_{n,j}$ to be a positive fractional number;

(ii) $(k_n,d_n)\rightarrow (k_0,d_0)$ as $n\rightarrow \infty$;

(iii) there is a positive integer $M_n$ for each $n$ holding
$$
M_n[\bullet ]s^*_n=(M_n\cdot c^*_{n,1})(M_n\cdot c^*_{n,2})\cdots (M_n\cdot c^*_{n,m})=s_n
$$ such that $s_n$ is just a \emph{proper number-based string} with positive integer $M_n\cdot c^*_{n,j}$ for $j\in [1,m]$. In other words, the research of fractional strings can be translated into the investigation of proper number-based strings.\qqed
\end{rem}

\begin{defn} \label{defn:plane-coordinate-string-sequence}
\cite{Bing-et-al-arXiv-asymmetric-4520331} We call a parameterized number-based string $s(k,d)$ made by the parameterized Topcode-matrix $P_{ara}(G,F|k,d)$ defined in Eq.(\ref{eqa:definition-parameterized-topcode-matrix}) \emph{plane-curve-attached string} if $k=f(d)$, or $d=g(k)$ for $f$ and $g$ are funtions of onr variable. Let $pc(x,y)=0$ be a \emph{plane curve} defined on a domain $[\alpha,\beta]^{r}$ for $0\leq \alpha<\beta$. If there are positive integer points $(k_n,d_n)\in [\alpha,\beta]^{r}$ holding $pc(k_n,d_n)=0$ for integers $k_n,d_n\geq 0$ with $n\in [1,m]$, then we get a \emph{plane-curve-attached string sequence} $\{s(k_n,d_n)\}^m_{n=1}$ based on the plane curve $pc(x,y)=0$.\qqed
\end{defn}

\begin{thm}\label{thm:one-encryption-one-time}
$^*$ By Definition \ref{defn:plane-coordinate-string-sequence}, there are infinite plane-curve-attached string sequences based on a parameterized Topcode-matrix $P_{ara}(G,F|k,d)$ defined in Eq.(\ref{eqa:definition-parameterized-topcode-matrix}) and infinite plane curves, which provides the theoretical basis for the one-encryption one-time (also one-time pad) first invented by Major Joseph Mauborgne and Gilbert Vernam of AT\&T in 1917.
\end{thm}

\begin{rem}\label{rem:333333}
For a \emph{public-key graph} $G$ admitting a $W$-constraint parameterized coloring $F$, we use this \emph{parameter-colored graph} $G$ and a plane curve $pc(x,y)=0$ to form a \emph{private-key graph} in a \emph{topological signature authentication}, the private-key graph is denoted as $H=\langle G,F,pc(x,y)=0\rangle$. Since there are infinite real-valued functions and there are infinite integer points in a plane curve, so we can get infinite number-based strings to encrypt or to decrypt a file consisted of many segments in the method of \emph{asymmetric topology cryptography}, and these number-based strings are random since the plane curve are taken randomly in the private-key graphs like as $H=\langle G,F,pc(x,y)=0\rangle$.

If the plane curve $pc(x,y)=0$ is an elliptic curve: $y^2=x^3+ ax^2+bx+c$ defined on a finite field $[0,A_{prime}]$ with a prime number $A_{prime}$, then deciphering the plane-curve-attached string sequence $\{s(k_n,d_n)\}^m_{n=1}$ is even more difficult, even impossible.\qqed
\end{rem}

\subsubsection{Parameterized string-colorings and set-colorings}

\begin{defn} \label{defn:more-string-total-coloring}
$^*$ Let $G$ be a bipartite $(p,q)$-graph, and its vertex set $V(G)=X\cup Y$ with $X\cap Y=\emptyset$ such that each edge $uv\in E(G)$ holds $u\in X$ and $v\in Y$. There are a group of $W$-constraint $(k_s,d_s)$-colorings
\begin{equation}\label{eqa:555555}
f_s:X\rightarrow S_{m,0,0,d}=\big \{0,d,\dots ,md\big \},~f_s:Y\cup E(G)\rightarrow S_{n,k,0,d}=\big \{k,k+d,\dots ,k+nd\big \}
\end{equation} here it is allowed $f_s(u)=f_s(w)$ for some distinct vertices $u,w\in V(G)$) for $s\in [1,B]$ with integer $B\geq 2$, such that the $W$-constraint $(k_s,d_s)$-coloring $f_s$ is one of gracefully $(k_s,d_s)$-total coloring, odd-gracefully $(k_s,d_s)$-total coloring, edge anti-magic $(k_s,d_s)$-total coloring, harmonious $(k_s,d_s)$-total coloring, odd-elegant $(k_s,d_s)$-total coloring, edge-magic $(k_s,d_s)$-total coloring, edge-difference $(k_s,d_s)$-total coloring, felicitous-difference $(k_s,d_s)$-total coloring, graceful-difference $(k_s,d_s)$-total coloring, odd-edge edge-magic $(k_s,d_s)$-total coloring, odd-edge edge-difference $(k_s,d_s)$-total coloring, odd-edge felicitous-difference $(k_s,d_s)$-total coloring, odd-edge graceful-difference $(k_s,d_s)$-total coloring. and so on. We have:

(i) The bipartite $(p,q)$-graph $G$ admits a \emph{parameterized total string-coloring} $F$ holding
\begin{equation}\label{eqa:555555}
{
\begin{split}
F(u)&=f_{i_1}(u)f_{i_2}(u)\cdots f_{i_B}(u),\quad F(uv)=f_{j_1}(uv)f_{j_2}(uv)\cdots f_{j_B}(uv),\\
F(v)&=f_{s_1}(v)f_{s_2}(v)\cdots f_{s_B}(v)
\end{split}}
\end{equation} true for each edge $uv\in E(G)$, where $f_{i_1}(u)f_{i_2}(u)\cdots f_{i_B}(u)$ is a permutation of $f_1(u),f_2(u),\cdots $, $f_B(u)$, $f_{j_1}(uv)f_{j_2}(uv)\cdots f_{j_B}(uv)$ is a permutation of $f_1(uv),f_2(uv),\cdots ,f_B(uv)$ and $f_{s_1}(v)f_{s_2}(v)$ $\cdots f_{s_B}(v)$ is a permutation of $f_1(v),f_2(v),\cdots ,f_B(v)$.

Hence, there are $(B!)^3$ parameterized total string-colorings in total.

(ii) The bipartite $(p,q)$-graph $G$ admits a \emph{parameterized total set-coloring} $\theta$ holding
\begin{equation}\label{eqa:555555}
{
\begin{split}
\theta(u)&=\big \{f_1(u),f_2(u),\dots ,f_B(u)\big \},\quad \theta(uv)=\big \{f_1(uv),f_2(uv),\dots ,f_B(uv)\big \},\\
\theta(v)&=\big \{f_1(v),f_2(v),\dots ,f_B(v)\big \}
\end{split}}
\end{equation} true for each edge $uv\in E(G)$.

(iii) The bipartite $(p,q)$-graph $G$ admits a \emph{parameterized total vector-coloring} $\alpha$ holding
\begin{equation}\label{eqa:555555}
{
\begin{split}
\alpha(u)&=\big (f_1(u),f_2(u),\dots ,f_B(u)\big ),\quad \alpha(uv)=\big (f_1(uv),f_2(uv),\dots ,f_B(uv)\big ),\\
\alpha(v)&=\big (f_1(v),f_2(v),\dots ,f_B(v)\big )
\end{split}}
\end{equation} true for each edge $uv\in E(G)$

Similarly with (i), there are $(B!)^3$ parameterized total vector-colorings, in total.

(iv) $^*$ A bipartite $(p,q)$-graph $G$ admits a \emph{parameterized total hyperedge set-coloring} $\varphi: V(G)\cup E(G)\rightarrow \mathcal{E}\in \mathcal{E}(\Lambda^2)$, where $\Lambda=\{f_1,f_2,\dots ,f_B\}$. The, and the set-coloring $\varphi$ satisfies the following constraints:

(1) $|\varphi(u)|=|e_u|\geq \textrm{deg}_G(u)$ for each vertex $u\in V(G)$ and $e_u\in \mathcal{E}$;

(2) $|\varphi(uv)|=|e_{uv}|\geq 1$ for each edge $uv\in E(G)$ and $e_{uv}\in \mathcal{E}$;

(3) each $f_s\in \Lambda$ is in $\varphi(w)$ for some $w\in V(G)\cup E(G)$;

(4) each pair of adjacent edges $xy$ and $xz$ with $y,z\in N_{ei}(x)$ holds $\varphi(xy)\neq \varphi(xz)$;

(5) each $a_{uv}\in \varphi(uv)$ for each edge $uv\in E(G)$ corresponds $a_{u}\in \varphi(u)$ and $a_{v}\in \varphi(v)$, such that $W[f_s(a_{u}), f_s(a_{uv}), f_s(a_{v})]=0$ for some $f_s\in \Lambda$;

(6) each $b_{u}\in \varphi(u)$ (resp. $b_{v}\in \varphi(v)$) for each edge $uv\in E(G)$ corresponds $b_{uv}\in \varphi(uv)$ and $b_{v}\in \varphi(v)$ (resp. $b_{u}\in \varphi(u)$), such that $W[f_i(b_{u}), f_i(b_{uv}), f_i(b_{v})]=0$ for some $f_i\in \Lambda$;

(7) $\varphi(V(G)\cup E(G))=\Lambda$.\qqed
\end{defn}

\begin{defn} \label{defn:n-di-set-colorings-definition}
\cite{Bing-et-al-arXiv-asymmetric-4520331} \textbf{Homogeneous $(abc)$-magic set-colorings.} Let $\textbf{\textrm{S}}_{et}(\leq n)$ be the set of integer sets of form $\{\alpha_{1},\alpha_{2},\dots ,\alpha_{m}\}$ with each number $\alpha_{j}\in Z^0\setminus \{0\}$ for $j\in [1,m]$ and $m\leq n$. A $(p,q)$-graph $G$ admits a \emph{$\{W_i\}^A_{i=1}$-constraint total set-coloring} $\psi : V(G)\cup E(G)\rightarrow \textbf{\textrm{S}}_{et}(\leq n)$, such that each edge $u_kv_k\in E(G)=\{u_kv_k:k\in [1,q]\}$ holds
\begin{equation}\label{eqa:n-di-set-coloringss}
{
\begin{split}
&\psi(u_k)=\big \{a_{k,1},a_{k,2},\dots ,a_{k,n}\big \},~\psi(v_k)=\big \{b_{k,1},b_{k,2},\dots ,b_{k,n}\big \},\\
&\psi(u_kv_k)=\big \{c_{k,1},c_{k,2},\dots ,c_{k,n}\big \}
\end{split}}
\end{equation} subject to the $W_i$-constraint $W_i[\psi(u_k),\psi(u_kv_k),\psi(v_k)]=0$ for some $i\in [1,A]$. Let $\lambda$ and $\gamma$ be constants, there are the following $(abc)$-magic set-constraints:

\begin{asparaenum}[\textbf{\textrm{Set}}-1. ]
\item \label{4set-magic:uniform-edge-magic} Each $j\in [1,n]$ holds the \emph{edge-magic constraint} $a_{k,j}+b_{k,j}+c_{k,j}=\lambda$ true, denoted as
\begin{equation}\label{eqa:555555}
\psi(w_k)[+]\psi(w_kz_k)[+]\psi(z_k)=\lambda
\end{equation}
\item \label{4set-magic:uniform-edge-difference} Each $j\in [1,n]$ holds the \emph{edge-difference constraint} $c_{k,j}+|a_{k,j}-b_{k,j}|=\lambda$ true, denoted as
\begin{equation}\label{eqa:555555}
\psi(w_kz_k)[+]|\psi(w_k)[-]\psi(z_k)|=\lambda
\end{equation}
\item \label{4set-magic:uniform-graceful-difference} Each $j\in [1,n]$ holds the \emph{graceful-difference constraint} $\big ||a_{k,j}-b_{k,j}|-c_{k,j}\big |=\lambda$ true, denoted as
\begin{equation}\label{eqa:555555}
\big ||\psi(w_k)[-]\psi(z_k)|[-]\psi(w_kz_k)\big |=\lambda
\end{equation}
\item \label{4set-magic:uniform-felicitous-difference} Each $j\in [1,n]$ holds the \emph{felicitous-difference constraint} $|a_{k,j}+b_{k,j}-c_{k,j}|=\lambda$ true, denoted as
\begin{equation}\label{eqa:555555}
|\psi(w_k)[+]\psi(z_k)[-]\psi(w_kz_k)|=\lambda
\end{equation}
\item \label{4set-magic:weak-edge-magic} Some $r\in [1,n]$ holds the \emph{edge-magic constraint} $a_{k,r}+b_{k,r}+c_{k,r}=\gamma$ true, but not all, denoted as $\partial_r \langle \psi(w_k)[+]\psi(w_kz_k)[+]\psi(z_k)=\gamma \rangle$.
\item \label{4set-magic:weak-edge-difference} Some $s\in [1,n]$ holds the \emph{edge-difference constraint} $c_{k,s}+|a_{k,s}-b_{k,s}|=\gamma$ true, but not all, denoted as $\partial_s \langle \psi(w_kz_k)[+]|\psi(w_k)[-]\psi(z_k)|=\gamma \rangle$.
\item \label{4set-magic:weak-graceful-difference} Some $t\in [1,n]$ holds the \emph{graceful-difference constraint} $\big ||a_{k,t}-b_{k,t}|-c_{k,t}\big |=\gamma$ true, but not all, denoted as $\partial_t \langle \big ||\psi(w_k)[-]\psi(z_k)|[-]\psi(w_kz_k)\big |=\gamma \rangle$.
\item \label{4set-magic:weak-felicitous-difference} Some $d\in [1,n]$ holds the \emph{felicitous-difference constraint} $|a_{k,d}+b_{k,d}-c_{k,d}|=\gamma$ true, but not all, denoted as $\partial_d \langle |\psi(w_k)[+]\psi(z_k)[-]\psi(w_kz_k)|=\gamma \rangle$.
\end{asparaenum}
\textbf{We call the total set-coloring $\psi$ to be}
\begin{asparaenum}[\textbf{\textrm{Setabc}}-1. ]
\item a \emph{component edge-magic total set-coloring} if it holds Set-\ref{4set-magic:uniform-edge-magic} true.
\item a \emph{component edge-difference total set-coloring} if it holds Set-\ref{4set-magic:uniform-edge-difference} true.
\item a \emph{component graceful-difference total set-coloring} if it holds Set-\ref{4set-magic:uniform-graceful-difference} true.
\item a \emph{component felicitous-difference total set-coloring} if it holds Set-\ref{4set-magic:uniform-felicitous-difference} true.
\item a \emph{weak-component edge-magic total set-coloring} if it holds Set-\ref{4set-magic:weak-edge-magic} true.
\item a \emph{weak-component edge-difference total set-coloring} if it holds Set-\ref{4set-magic:weak-edge-difference} true.
\item a \emph{weak-component graceful-difference total set-coloring} if it holds Set-\ref{4set-magic:weak-graceful-difference} true.
\item a \emph{weak-component felicitous-difference total set-coloring} if it holds Set-\ref{4set-magic:weak-felicitous-difference} true.\qqed
\end{asparaenum}
\end{defn}

\begin{rem}\label{rem:333333}
The $W$-constraint $W_i\langle \psi(u_k),\psi(u_kv_k),\psi(v_k)\rangle=0$ in Definition \ref{defn:n-di-set-colorings-definition} is a group of constraints. Moreover, by the non-homogeneous idea, we can set the colors of vertices and edges as
\begin{equation}\label{eqa:different-dimendion-set-coloringss}
{
\begin{split}
&\psi(u_k)=\big \{a_{k,1},a_{k,2},\dots ,a_{k,n(k,r)}\big \},~\psi(v_k)=\big \{b_{k,1},b_{k,2},\dots ,b_{k,n(k,s)}\big \},\\
&\psi(u_kv_k)=\big \{c_{k,1},c_{k,2},\dots ,c_{k,n(k,t)}\big \}
\end{split}}
\end{equation} for each edge $u_kv_k\in E(G)=\{u_kv_k:k\in [1,q]\}$ under a \emph{$\{W_i\}^A_{i=1}$-constraint total set-coloring} $\psi : V(G)\cup E(G)\rightarrow \textbf{\textrm{S}}_{et}(\leq n)$ of a $(p,q)$-graph $G$. We modify the conditions of Definition \ref{defn:n-di-set-colorings-definition} slightly, and then get the same set-colorings defined in Definition \ref{defn:n-di-set-colorings-definition}.

For example, we set: If each number $c_{k,j}\in \psi(u_kv_k)$ corresponds to some $a_{k,r}\in \psi(u_k)$ and some $b_{k,s}\in \psi(v_k)$ holding the edge-magic constraint $a_{k,r}+b_{k,s}+c_{k,j}=\lambda$ true; each number $a_{k,r}\in \psi(u_k)$ corresponds to some $c_{k,j}\in \psi(u_kv_k)$ and some $b_{k,s}\in \psi(v_k)$ holding the edge-magic constraint $a_{k,r}+b_{k,s}+c_{k,j}=\lambda$ true; and each number $b_{k,s}\in \psi(v_k)$ corresponds to some $a_{k,r}\in \psi(u_k)$ and some $c_{k,j}\in \psi(u_kv_k)$ holding the edge-magic constraint $a_{k,r}+b_{k,s}+c_{k,j}=\lambda$ true. Then we call the total set-coloring $\psi$ \emph{component edge-magic total set-coloring}.

The set-colorings defined in Definition \ref{defn:n-di-set-colorings-definition} can be related with the vertex-intersected graphs of hypergraphs as the set $\textbf{\textrm{S}}_{et}(\leq n)$ appeared in Definition \ref{defn:n-di-set-colorings-definition} is a \emph{hyperedge set} $\mathcal{E}$ holding $\Lambda=\bigcup_{e\in \mathcal{E}}e$, where $\Lambda$ is a set of finite numbers.\qqed
\end{rem}

\begin{prop}\label{prop:99999}
\cite{Bing-et-al-arXiv-asymmetric-4520331} For a fixed set $U$, there are more groups of different sets $S_a,S_b,S_c$ holding:

(i) The \emph{set-edge-magic constraint} $S_a\cup S_b\cup S_c=U$.

(ii) The \emph{set-edge-difference constraint} $S_c\cup \big (S_a\setminus S_b\big )=U$.

(iii) The \emph{set-felicitous-difference constraint} $\big (S_a\cup S_b\big )\setminus S_c=U$.

(iv) The \emph{set-graceful-difference constraint} $\big (S_a\setminus S_b\big )\setminus S_c=U$.
\end{prop}

\begin{defn} \label{defn:111111}
\cite{Bing-et-al-arXiv-asymmetric-4520331} Let $S_{et}$ be a set of sets. Suppose that a graph $G$ admits a total set-coloring $F:V(G)\cup E(G)\rightarrow S_{et}$, such that $F(u)=S_{u}\in S_{et}$, $F(v)=S_{v}\in S_{et}$, and $F(uv)=S_{uv}\in S_{et}$ for each edge $uv\in E(G)$.

(i) If there is a fixed set $U$, such that each edge $uv\in E(G)$ holds the \emph{set-edge-magic constraint} $S_{u}\cup S_{v}\cup S_{uv}=U$ true, we say $F$ \emph{set-edge-magic total set-coloring}.

(ii) If there is a fixed set $U$, such that each edge $uv\in E(G)$ holds one of \emph{set-edge-difference constraints} $S_{uv}\cup (S_{u}\setminus S_{v})=U$ and $S_{uv}\cup (S_{v}\setminus S_{u})=U$ true, we say $F$ \emph{set-edge-difference total set-coloring}.

(iii) If there is a fixed set $U$, such that each edge $uv\in E(G)$ holds one of \emph{set-felicitous-difference constraints} $(S_{u}\cup S_{v})\setminus S_{uv}=U$ and $S_{uv}\setminus (S_{v}\cup S_{u})=U$ true, we say $F$ \emph{set-felicitous-difference total set-coloring}.

(iv) If there is a fixed set $U$, such that each edge $uv\in E(G)$ holds one of \emph{set-graceful-difference constraints} $(S_{u}\setminus S_{v})\setminus S_{uv}=U$, $(S_{v}\setminus S_{u})\setminus S_{uv}=U$, $S_{uv}\setminus (S_{u}\setminus S_{v})=U$ and $S_{uv}\setminus (S_{v}\setminus S_{u})=U$ true, we say $F$ \emph{set-graceful-difference total set-coloring}.\qqed
\end{defn}

\begin{thm}\label{thm:666666}
\cite{Bing-et-al-arXiv-asymmetric-4520331} Each connected $(p,q)$-graph $G$ admits a \emph{proper total string-coloring}
$$f:V(G)\cup E(G)\rightarrow \big \{a_{i}b_{i}:~a_{i},b_{i}\in [1,\Delta(G)-k]\big \}
$$ with $k\leq \Delta(G)-\sqrt{3[1+\Delta(G)]}$ if $\Delta(G)\geq 6$.
\end{thm}

\begin{conj}\label{conj:0000000000}
$^*$ Each connected $(p,q)$-graph $G$ admits a \emph{proper total string-coloring}
$$
f:V(G)\cup E(G)\rightarrow \big \{a_{i}b_{i}:~a_{i},b_{i}\in \big [1,\sqrt{\Delta(G)+2}\big ]\big \}
$$
\end{conj}

\begin{defn}\label{defn:graphs-vs-total-graphs}
\textbf{A transformation for colorings.} The \emph{total graph} $T(G)$ of a graph $G$ is a graph such that

(i) the vertex set $V(T(G))$ of $T(G)$ holds $V(T(G))=V(G)\cup E(G)$; and

(ii) two vertices are adjacent in $T(G)$ if and only if their corresponding elements are either adjacent or incident in the graph $G$.

Then, a total coloring $f_t$ of the graph $G$ becomes a (proper) vertex coloring $g_v$ of the total graph $T(G)$, that is, $f_t\sim g_v$. \qqed
\end{defn}

\begin{defn}\label{defn:graceful-intersection}
\cite{Yao-Zhang-Sun-Mu-Sun-Wang-Wang-Ma-Su-Yang-Yang-Zhang-2018arXiv} Suppose that a $(p,q)$-graph $G$ admits a set-labeling $F:V(G)\rightarrow [1,q]^2$~(resp. $[1,2q-1]^2)$, and induces an edge set-color $F(uv)=F(u)\cap F(v)$ for each edge $uv \in E(G)$. If we can select a \emph{representative} $a_{uv}\in F(uv)$ for each edge color set $F(uv)$ such that $\{a_{uv}:~uv\in E(G)\}=[1,q]$ (resp. $[1,2q-1]^o$), then $F$ is called \emph{graceful-intersection (resp. odd-graceful-intersection) total set-labeling} of the graph $G$.\qqed
\end{defn}

\begin{thm}\label{thm:graceful-total-set-labelings}
\cite{Yao-Zhang-Sun-Mu-Sun-Wang-Wang-Ma-Su-Yang-Yang-Zhang-2018arXiv} Each tree $T$ admits a \emph{graceful-intersection (resp. an odd-graceful-intersection) total set-labeling}.
\end{thm}

We define a \emph{regular rainbow set-sequence} $\big \{R_k\big \}^{q}_1$ as: $R_k=[1,k]$ with $k\in [1,q]$, where $[1,1]=\{1\}$.

\begin{thm}\label{thm:rainbow-total-set-labelings}
\cite{Yao-Zhang-Sun-Mu-Sun-Wang-Wang-Ma-Su-Yang-Yang-Zhang-2018arXiv} Each tree $T$ of $q$ edges admits a \emph{regular rainbow intersection total set-labeling} based on a \emph{regular rainbow set-sequence} $\big \{[1,k]\big \}^{q}_{k=1}$.
\end{thm}
\begin{proof} Suppose that a vertex $x$ is a leaf of a tree $T$ of $q$ edges, so the vertex-removed graph $T-x$ is just a tree of $(q-1)$ edges. Assume that $T-x$ admits a regular rainbow set-sequence $\big \{R_k\big \}^{q-1}_1$ total set-labeling $f$. Let $y$ be adjacent with $x$ in $T$. We define a labeling $g$ of the tree $T$ in this way: $g(w)=f(w)$ for $w\in V(T)\setminus \{y,x\}$, $g(y)=R_{q+1}=[1,q+1]$ and $g(x)=R_q=[1,q]$. Therefore, we have $g(u_iv_j)=g(u_i)\cap g(v_j)=[1,i]\cap [1,j]$ for $u_iv_j\in E(T)\setminus \{xy\}$, and $g(xy)=g(x)\cap g(y)=[1,q]$, and $g(s)\neq g(t)$ for any pair of vertices $s$ and $t$. We claim that $g$ is a regular rainbow intersection total set-labeling of $T$ by the hypothesis of induction.
\end{proof}

\begin{rem}\label{rem:333333}
Each tree admits a regular odd-rainbow intersection total set-labeling based on a \emph{regular odd-rainbow set-sequence} $\big \{R_k\big \}^{q}_1$ defined as: $R_k=[1,2k-1]$ with $k\in [1,q]$, where $[1,1]=\{1\}$. Moreover, we can define a \emph{regular Fibonacci-rainbow set-sequence} $\big \{R_k\big \}^{q}_1$ by $R_1=[1,1]$, $R_2=[1,1]$, and $R_{k+1}=R_{k-1}\cup R_{k}$ with $k\in [2,q]$; or a $\tau$-term Fibonacci-rainbow set-sequence $\big \{\tau,R_i\big \}^{q}_1$ holds: $R_i=[1,a_i]$ with $a_i>1$ and $i\in [1,q]$, and $R_k=\sum ^{k-1}_{i=k-\tau}R_i$ with $k>\tau$. It may be an interesting research on various rainbow set-sequences for non-tree graphs.\qqed
\end{rem}

\subsection{Number-based sequence colorings}

Sequence colorings are a class of specific set-colorings, strictly speaking.

\begin{defn}\label{defn:sequence-coloring}
\cite{Yao-Su-Wang-Hui-Sun-ITAIC2020} Let $G$ be a $(p, q)$-graph, and let a sequence $A_M=\{a_i\}_1^M$ hold $0\leq a_i< a_{i+1}$ for $i\in [1, M-1]$ and $p\leq M$, and let another sequence $B_q=\{b_j\}_1^q$ hold $0\leq b_j< b_{j+1}$ for $j\in [1, q-1]$, and let $k$ be a constant. The $(p, q)$-graph $G$ admits a mapping $f:S\rightarrow C$ with $f(S)=\{f(x):x\in S\}$ and there are the following constraints:
\begin{asparaenum}[Rec-1. ]
\item \label{rescondi:vertex-mapping} $S=V(G)$ and $C=A_M$;
\item \label{rescondi:total-mapping} $S=V(G)\cup E(G)$ and $C=A_M\cup B_q$;
\item \label{rescondi:adjacent-different} $f(u)\neq f(v)$ for any edge $uv\in E(G)$;
\item \label{rescondi:incident-different} $f(u)\neq f(uv)$ and $f(v)\neq f(uv)$ for each edge $uv\in E(G)$;
\item \label{rescondi:adjacent-edge-different} $f(uv)\neq f(uw)$ for distinct vertices $v, w\in N_{ei}(u)$;
\item \label{rescondi:not-full-1} $f(E(G))\subseteq B_q$;
\item \label{rescondi:all-full} $f(V(G))=A_M$;
\item \label{rescondi:edge-set-full} $f(E(G))=B_q$;
\item \label{rescondi:induced-edge-label} a function $O$ holding $f(uv)=O(f(u), f(v))$ for each edge $uv\in E(G)$;
\item \label{rescondi:magic-equation} a function $F$ holding $F(f(u), f(uv), f(v))=k$ for each edge $uv\in E(G)$;
\item \label{rescondi:set-ordered} $G$ is a bipartite graph with vertex bipartition $(X, Y)$ holding the set-ordered constraint $\min f(X)<\max f(Y)$.
\end{asparaenum}
\quad We refer to $f$ as:

\noindent ------ \emph{traditional-type}

\begin{asparaenum}[\textrm{Coloring}-1.]
\item an \emph{edge-induced sequence coloring} of the graph $G$ if Rec-\ref{rescondi:vertex-mapping} and Rec-\ref{rescondi:induced-edge-label} hold true;
\item a \emph{graceful edge-induced sequence coloring} of the graph $G$ if Rec-\ref{rescondi:vertex-mapping}, Rec-\ref{rescondi:edge-set-full} and Rec-\ref{rescondi:induced-edge-label} hold true, where $f(uv)=O(f(u), f(v))=|f(u)-f(v)|$;
\item a \emph{set-ordered edge-induced sequence coloring} of the graph $G$ if Rec-\ref{rescondi:vertex-mapping}, Rec-\ref{rescondi:induced-edge-label} and Rec-\ref{rescondi:set-ordered} hold true;
\item a \emph{set-ordered graceful edge-induced sequence coloring} of the graph $G$ if Rec-\ref{rescondi:vertex-mapping}, Rec-\ref{rescondi:edge-set-full}, Rec-\ref{rescondi:induced-edge-label} and Rec-\ref{rescondi:set-ordered} hold true, where $f(uv)=O(f(u), f(v))=|f(u)-f(v)|$;

\item a \emph{sequence graceful labelling} of the graph $G$ if Rec-\ref{rescondi:vertex-mapping}, Rec-\ref{rescondi:edge-set-full}, Rec-\ref{rescondi:all-full} and Rec-\ref{rescondi:induced-edge-label} hold true, where $f(uv)=O(f(u), f(v))=|f(u)-f(v)|$;

\item a \emph{set-ordered sequence graceful coloring} of the graph $G$ if Rec-\ref{rescondi:vertex-mapping}, Rec-\ref{rescondi:edge-set-full}, Rec-\ref{rescondi:set-ordered} and Rec-\ref{rescondi:induced-edge-label} hold true, where $f(uv)=O(f(u), f(v))=|f(u)-f(v)|$;
\item a \emph{set-ordered sequence graceful labelling} of the graph $G$ if Rec-\ref{rescondi:vertex-mapping}, Rec-\ref{rescondi:edge-set-full}, Rec-\ref{rescondi:all-full}, Rec-\ref{rescondi:set-ordered} and Rec-\ref{rescondi:induced-edge-label} hold true, where $f(uv)=O(f(u), f(v))=|f(u)-f(v)|$;

\item a \emph{sequence total coloring} of the graph $G$ if Rec-\ref{rescondi:total-mapping} and Rec-\ref{rescondi:adjacent-different} hold true;
\item a \emph{proper sequence total coloring} of the graph $G$ if Rec-\ref{rescondi:total-mapping}, Rec-\ref{rescondi:adjacent-different}, Rec-\ref{rescondi:incident-different} and Rec-\ref{rescondi:adjacent-edge-different} hold true;

------ \emph{graceful-type}

\item a \emph{graceful sequence total coloring} of the graph $G$ if Rec-\ref{rescondi:total-mapping}, Rec-\ref{rescondi:adjacent-different}, Rec-\ref{rescondi:edge-set-full} and Rec-\ref{rescondi:induced-edge-label} hold true, where $f(uv)=O(f(u), f(v))=|f(u)-f(v)|$;
\item a \emph{graceful sequence proper total coloring} of the graph $G$ if Rec-\ref{rescondi:total-mapping}, Rec-\ref{rescondi:adjacent-different}, Rec-\ref{rescondi:incident-different}, Rec-\ref{rescondi:adjacent-edge-different}, Rec-\ref{rescondi:edge-set-full} and Rec-\ref{rescondi:induced-edge-label} hold true, where $f(uv)=O(f(u), f(v))=|f(u)-f(v)|$;

\item a \emph{proper graceful-total sequence coloring} of the graph $G$ if Rec-\ref{rescondi:total-mapping}, Rec-\ref{rescondi:adjacent-different}, Rec-\ref{rescondi:incident-different}, Rec-\ref{rescondi:adjacent-edge-different}, Rec-\ref{rescondi:edge-set-full}, Rec-\ref{rescondi:all-full} and Rec-\ref{rescondi:induced-edge-label} hold true, where $f(uv)=O(f(u), f(v))=|f(u)-f(v)|$;

------ \emph{felicitous-type}

\item a \emph{felicitous sequence total coloring} of the graph $G$ if Rec-\ref{rescondi:total-mapping}, Rec-\ref{rescondi:adjacent-different}, Rec-\ref{rescondi:not-full-1} and Rec-\ref{rescondi:induced-edge-label} hold true, where $f(uv)=O(f(u), f(v))=f(u)+f(v)$~($\bmod^*~q^*$);
\item a \emph{felicitous sequence proper total coloring} of the graph $G$ if Rec-\ref{rescondi:total-mapping}, Rec-\ref{rescondi:adjacent-different}, Rec-\ref{rescondi:incident-different}, Rec-\ref{rescondi:adjacent-edge-different}, Rec-\ref{rescondi:not-full-1} and Rec-\ref{rescondi:induced-edge-label} hold true, where $f(uv)=O(f(u), f(v))=f(u)+f(v)$~($\bmod^*~q^*$);

\item a \emph{set-ordered felicitous sequence total coloring} of the graph $G$ if Rec-\ref{rescondi:total-mapping}, Rec-\ref{rescondi:adjacent-different}, Rec-\ref{rescondi:set-ordered}, Rec-\ref{rescondi:not-full-1} and Rec-\ref{rescondi:induced-edge-label} hold true, where $f(uv)=O(f(u), f(v))=f(u)+f(v)$~($\bmod^*~q^*$);
\item a \emph{set-ordered felicitous sequence proper total coloring} of the graph $G$ if Rec-\ref{rescondi:total-mapping}, Rec-\ref{rescondi:adjacent-different}, Rec-\ref{rescondi:incident-different}, Rec-\ref{rescondi:adjacent-edge-different}, Rec-\ref{rescondi:set-ordered}, Rec-\ref{rescondi:not-full-1} and Rec-\ref{rescondi:induced-edge-label} hold true, where $f(uv)=O(f(u), f(v))=f(u)+f(v)$~($\bmod^*~q^*$);

------ \emph{magic-constraints}

\item a \emph{sequence edge-magic total coloring} of the graph $G$ if Rec-\ref{rescondi:total-mapping}, Rec-\ref{rescondi:adjacent-different}, Rec-\ref{rescondi:adjacent-edge-different} and Rec-\ref{rescondi:magic-equation} hold true, where $F(f(u), f(uv), f(v))=f(u)+f(uv)+f(v)=k$;
\item a \emph{sequence proper edge-magic total coloring} of the graph $G$ if Rec-\ref{rescondi:total-mapping}, Rec-\ref{rescondi:adjacent-different}, Rec-\ref{rescondi:incident-different}, Rec-\ref{rescondi:adjacent-edge-different} and Rec-\ref{rescondi:magic-equation} hold true, where $F(f(u), f(uv), f(v))=f(u)+f(uv)+f(v)=k$;

\item a \emph{set-ordered sequence edge-magic total coloring} of the graph $G$ if Rec-\ref{rescondi:total-mapping}, Rec-\ref{rescondi:adjacent-different}, Rec-\ref{rescondi:adjacent-edge-different}, Rec-\ref{rescondi:set-ordered} and Rec-\ref{rescondi:magic-equation} hold true, where $F(f(u), f(uv), f(v))=f(u)+f(uv)+f(v)=k$;
\item a \emph{set-ordered sequence proper edge-magic total coloring} of the graph $G$ if Rec-\ref{rescondi:total-mapping}, Rec-\ref{rescondi:adjacent-different}, Rec-\ref{rescondi:incident-different}, Rec-\ref{rescondi:adjacent-edge-different}, Rec-\ref{rescondi:set-ordered} and Rec-\ref{rescondi:magic-equation} hold true, where $F(f(u), f(uv), f(v))=f(u)+f(uv)+f(v)=k$;

\item a \emph{sequence edge-difference total coloring} of the graph $G$ if Rec-\ref{rescondi:total-mapping}, Rec-\ref{rescondi:adjacent-different}, Rec-\ref{rescondi:adjacent-edge-different} and Rec-\ref{rescondi:magic-equation} hold true, where $F(f(u), f(uv), f(v))=f(uv)+|f(u)-f(v)|=k$;

\item a \emph{sequence proper edge-difference total coloring} of the graph $G$ if Rec-\ref{rescondi:total-mapping}, Rec-\ref{rescondi:adjacent-different}, Rec-\ref{rescondi:incident-different}, Rec-\ref{rescondi:adjacent-edge-different} and Rec-\ref{rescondi:magic-equation} hold true, where $F(f(u), f(uv), f(v))=f(uv)+|f(u)-f(v)|=k$;

\item a \emph{set-ordered sequence edge-difference total coloring} of the graph $G$ if Rec-\ref{rescondi:total-mapping}, Rec-\ref{rescondi:adjacent-different}, Rec-\ref{rescondi:adjacent-edge-different}, Rec-\ref{rescondi:set-ordered} and Rec-\ref{rescondi:magic-equation} hold true, where $F(f(u), f(uv), f(v))=f(uv)+|f(u)-f(v)|=k$;
\item a \emph{set-ordered sequence proper edge-difference total coloring} of the graph $G$ if Rec-\ref{rescondi:total-mapping}, Rec-\ref{rescondi:adjacent-different}, Rec-\ref{rescondi:incident-different}, Rec-\ref{rescondi:adjacent-edge-different}, Rec-\ref{rescondi:set-ordered} and Rec-\ref{rescondi:magic-equation} hold true, where $F(f(u), f(uv), f(v))=f(uv)+|f(u)-f(v)|=k$;

\item a \emph{sequence graceful-difference total coloring} of the graph $G$ if Rec-\ref{rescondi:total-mapping}, Rec-\ref{rescondi:adjacent-different}, Rec-\ref{rescondi:adjacent-edge-different} and Rec-\ref{rescondi:magic-equation} hold true, where $F(f(u), f(uv), f(v))=\big |f(uv)-|f(u)-f(v)|\big |=k$;

\item a \emph{sequence proper graceful-difference total coloring} of the graph $G$ if Rec-\ref{rescondi:total-mapping}, Rec-\ref{rescondi:adjacent-different}, Rec-\ref{rescondi:incident-different}, Rec-\ref{rescondi:adjacent-edge-different} and Rec-\ref{rescondi:magic-equation} hold true, where $F(f(u), f(uv), f(v))=\big |f(uv)-|f(u)-f(v)|\big |=k$;

\item a \emph{set-ordered sequence graceful-difference total coloring} of the graph $G$ if Rec-\ref{rescondi:total-mapping}, Rec-\ref{rescondi:adjacent-different}, Rec-\ref{rescondi:adjacent-edge-different}, Rec-\ref{rescondi:set-ordered} and Rec-\ref{rescondi:magic-equation} hold true, where $F(f(u), f(uv), f(v))=\big |f(uv)-|f(u)-f(v)|\big |=k$;
\item a \emph{set-ordered sequence proper graceful-difference total coloring} of the graph $G$ if Rec-\ref{rescondi:total-mapping}, Rec-\ref{rescondi:adjacent-different}, Rec-\ref{rescondi:incident-different}, Rec-\ref{rescondi:adjacent-edge-different}, Rec-\ref{rescondi:set-ordered} and Rec-\ref{rescondi:magic-equation} hold true, where $F(f(u), f(uv), f(v))=\big |f(uv)-|f(u)-f(v)|\big |=k$;

\item a \emph{sequence gracefully-total coloring} of the graph $G$ if Rec-\ref{rescondi:total-mapping}, Rec-\ref{rescondi:adjacent-different}, Rec-\ref{rescondi:adjacent-edge-different} and Rec-\ref{rescondi:magic-equation} hold true, where $F(f(u), f(uv), f(v))=|f(u)+f(v)-f(uv)|=k$;

\item a \emph{sequence proper gracefully-total coloring} of the graph $G$ if Rec-\ref{rescondi:total-mapping}, Rec-\ref{rescondi:adjacent-different}, Rec-\ref{rescondi:incident-different}, Rec-\ref{rescondi:adjacent-edge-different} and Rec-\ref{rescondi:magic-equation} hold true, where $F(f(u), f(uv), f(v))=|f(u)+f(v)-f(uv)|=k$;

\item a \emph{set-ordered sequence gracefully-total coloring} of the graph $G$ if Rec-\ref{rescondi:total-mapping}, Rec-\ref{rescondi:adjacent-different}, Rec-\ref{rescondi:adjacent-edge-different}, Rec-\ref{rescondi:set-ordered} and Rec-\ref{rescondi:magic-equation} hold true, where $F(f(u), f(uv), f(v))=|f(u)+f(v)-f(uv)|=k$;
\item a \emph{set-ordered sequence proper gracefully-total coloring} of the graph $G$ if Rec-\ref{rescondi:total-mapping}, Rec-\ref{rescondi:adjacent-different}, Rec-\ref{rescondi:incident-different}, Rec-\ref{rescondi:adjacent-edge-different}, Rec-\ref{rescondi:set-ordered} and Rec-\ref{rescondi:magic-equation} hold true, where $F(f(u), f(uv), f(v))=|f(u)+f(v)-f(uv)|=k$;

------ \emph{gcd-type}

\item a \emph{maxi-common-factor sequence total coloring} of the graph $G$ if Rec-\ref{rescondi:total-mapping}, Rec-\ref{rescondi:adjacent-different}, Rec-\ref{rescondi:adjacent-edge-different} and Rec-\ref{rescondi:induced-edge-label} hold true, where $f(uv)=O(f(u), f(v))=\textrm{gcd}(f(u), f(v))$;
\item a \emph{gracefully maxi-common-factor sequence total coloring} of the graph $G$ if Rec-\ref{rescondi:total-mapping}, Rec-\ref{rescondi:adjacent-different}, Rec-\ref{rescondi:adjacent-edge-different}, Rec-\ref{rescondi:edge-set-full} and Rec-\ref{rescondi:induced-edge-label} hold true, where $f(uv)=O(f(u), f(v))=\textrm{gcd}(f(u), f(v))$.\qqed
\end{asparaenum}
\end{defn}

\begin{thm}\label{thm:graceful-total-sequence-coloring}
\cite{Yao-Su-Wang-Hui-Sun-ITAIC2020} Every tree $T$ with diameter $D(T)\geq 3$ and $s+1=\left \lceil \frac{D(T)}{2}\right \rceil $ admits at least $2^{s}$ different \emph{gracefully total sequence colorings} if two sequences $A_M, B_q$ holding $0<b_j-a_i\in B_q$ for $a_i\in A_M$ and $b_j\in B_q$.
\end{thm}

\begin{lem}\label{thm:adding-leaves-keep-sequence-colorings}
\cite{Yao-Su-Wang-Hui-Sun-ITAIC2020} Suppose that a bipartite and connected graph $G$ admits a gracefully total sequence coloring based on two sequences $A_M, B_q$ holding $0<b_j-a_i\in B_q$ for $a_i\in A_M$ and $b_j\in B_q$, then a new bipartite and connected graph obtained by adding randomly leaves to $G$ admits a \emph{gracefully total sequence coloring} based on two sequences $A\,'_M, B\,'_q$ holding $0<b\,'_j-a\,'_i\in B\,'_q$ for $a\,'_i\in A\,'_M$ and $b\,'_j\in B\,'_q$.
\end{lem}

\begin{defn} \label{defn:111111}
\textbf{Colorings based on abstract sequences \cite{Yao-Su-Wang-Hui-Sun-ITAIC2020} .} Suppose that a $(p,q)$-graph $G$ admits a \emph{graceful coloring} $f:V(G)\cup E(G)\rightarrow [0,M]$ such that the edge color set
$$
f(E(G))=\{f(wz)=|f(w)-f(z)|:wz\in E(G)\}=[1,q]
$$ Let $C_M=\big \{c_i\big \}^M_{i=1}$ and $D_q=\big \{d_j\big \}^q_{j=1}$ be two \emph{abstract sequences}. We define a new coloring $f^*$ by $f^*(w)=c_i$ if $f(w)=i$ for vertex $w\in V(G)$, and $f^*(wz)=d_j$ if $f(wz)=j$ for edge $wz\in E(G)$. Then we call $f^*$ a \emph{graceful abstract-sequence coloring} of the graph $G$ if $f^*(E(G))=D_q$. Here an abstract sequence $C_M=\{c_i\}^M_{i=1}$ or $D_q=\{d_j\}^q_{j=1}$ is consisted of any things in the world. Thereby, $f^*$ is an \emph{abstract substitution} of $f$, conversely, $f$ is \emph{mapping homomorphism} to $f^*$. Let $E(G)=\{e_i=x_iy_i:~i\in [1,q]\}$. Then this $(p,q)$-graph $G$ has its own another Topcode-matrix $T^*_{code}(G)$ defined as
\begin{equation}\label{eqa:topcode-matrix-22}
{
\begin{split}
T^*_{code}(G)&= \left(
\begin{array}{cccccccccc}
f^*(x_1) & f^*(x_2) & \cdots & f^*(x_q)\\
f^*(e_1) & f^*(e_2) & \cdots & f^*(e_q)\\
f^*(y_1) & f^*(y_2) & \cdots & f^*(y_q)
\end{array}
\right)_{3\times q}\\
&=(f^*(x_i),f^*(e_i),f^*(y_i))^T_{e_i\in E(G)}
\end{split}}
\end{equation} In particular case of a Fibonacci-Lucas sequence, we have $f^*(w)=c_i=F[w_i, z_i]_n$ if $f(w)=i$ for vertex $w\in V(G)$ and $f^*(wz)=d_j=F[w_j, z_j]_n$ if $f(wz)=j$ for edge $wz\in E(G)$.\qqed
\end{defn}

Theorem \ref{thm:tree-graceful-total-coloringss} tells us that each tree $T$ admits a $(k,d)$-graceful total coloring, also, a set-ordered graceful total coloring as $(k, d)=(1, 1)$. Thereby, we have
\begin{thm}\label{thm:each-tree-graceful-abstract-coloring}
\cite{Yao-Su-Wang-Hui-Sun-ITAIC2020} Each tree admits a set-ordered graceful abstract-sequence coloring.
\end{thm}

\subsection{Distinguishing set-colorings}

Suppose that a graph $G$ admits a \textbf{\emph{coloring}} $\eta:S\rightarrow S_{et}$, where $S\subseteq V(G)\cup E(G)$ and $S_{et}$ is the set of sets $e_{1},e_{2},\dots ,e_{m}$. We will use the following neighbor color sets:

\begin{equation}\label{eqa:distinguishing-set-colorings}
{
\begin{split}
\mathcal{E}_v(x,\eta)=&\, \big \{\eta(y):y\in N_{ei}(x)\big \}\textrm{ (\emph{local v-set-color set})};\\
\mathcal{E}_v[x,\eta]=&\, \big \{\eta(x)\big \}\cup \mathcal{E}_v(x,\eta)\textrm{ (\emph{closed-local v-set-color set})};\\
\mathcal{E}_e(x,\eta)=&\, \big \{\eta(xy):y\in N_{ei}(x)\big \}\textrm{ (\emph{local e-set-color set})};\\
\mathcal{E}_e[x,\eta]=&\, \big \{\eta(x)\big \}\cup \mathcal{E}_e(x,\eta)\textrm{ (\emph{closed-local e-set-color set})};\\
\mathcal{E}_{ve}(x,\eta)=&\, \mathcal{E}_v(x,\eta)\cup \mathcal{E}_e(x,\eta)\textrm{ (\emph{local ve-set-color set})};\\
\mathcal{E}_{ve}[x,\eta]=&\, \big \{\eta(x)\big \}\cup \mathcal{E}_v(x,\eta)\cup \mathcal{E}_e(x,\eta)\textrm{ (\emph{closed-local ve-set-color set})};\\
\mathcal{E}_{ve}\big \{x,\eta\big \}=&\, \big \{\mathcal{E}_e(x,\eta), \mathcal{E}_e[x,\eta], \mathcal{E}_v[x,\eta],\mathcal{E}_{ve}[x,\eta]\big \}\textrm{ (\emph{closed-local $(4)$-set-color set})}.
\end{split}}
\end{equation}

Motivated from the distinguishing colorings introduced in \cite{Yao-Yang-Yao-2020-distinguishing}, we present the following distinguishing set-colorings:

\begin{defn}\label{defn:more-distinguishing-set-colorings}
$^*$ Suppose that a graph $G$ admits a coloring $\eta:S\rightarrow S_{et}=\{e_{1},e_{2},\dots ,e_{m}\}$ with each $e_{i}$ is a set, and the set $S\subseteq V(G)\cup E(G)$. By the color sets shown in Eq.(\ref{eqa:distinguishing-set-colorings}) and $x\in V(G)$, there are the following \textbf{constraints}:
\begin{asparaenum}[\textrm{\textbf{Co}}-1. ]
\item \label{asp:vertex-only} (Vertex-set) $S=V(G)$;
\item \label{asp:edge-only} (Edge-set) $S=E(G)$;
\item \label{asp:total} (Total-set) $S=V(G)\cup E(G)$;
\item \label{asp:adjacent-vertices} (Adjacent-vertices) $\eta(u)\neq \eta(v)$ for each edge $uv\in E(G)$;
\item \label{asp:adjacent-edges} (Adjacent-edges) $\eta(xy)\neq \eta(xw)$ for distinct vertices $y,w\in N_{ei}(x)$;
\item \label{asp:adjacent-vertices-edge} (Incident vertices and edges) $\eta(u)\neq \eta(uv)$ and $\eta(v)\neq \eta(uv)$ for each edge $uv\in E(G)$;
\item \label{asp:universal-vertices} (No-adjacent-vertices) $\eta(u)\neq \eta(x)$ for $ux\not\in E(G)$;
\item \label{asp:universal-edges} (No-adjacent-edges) $\eta(xy)\neq \eta(uv)$ for $xy,uv\in E(G)$ with $x\neq u$, $x\neq v$, $y\neq u$ and $y\neq v$;
\item \label{asp:neighbor-distinguishing} (Local vertex distinguishing) $\mathcal{E}_v(x,\eta)\neq \mathcal{E}_v(y,\eta)$ for each $y\in N_{ei}(x)$;
\item \label{asp:closed-neighbor-distinguishing} (Closed local vertex distinguishing) $\mathcal{E}_v[x,\eta]\neq \mathcal{E}_v[y,\eta]$ for each $y\in N_{ei}(x)$;
\item \label{asp:neighbor-e-distinguishing}(Local edge distinguishing) $\mathcal{E}_e(x,\eta)\neq \mathcal{E}_e(y,\eta)$ for each $y\in N_{ei}(x)$ ;
\item \label{asp:closed-neighbor-e-distinguishing} (Closed local ve-distinguishing) $\mathcal{E}_e[x,\eta]\neq \mathcal{E}_e[y,\eta]$ for each $y\in N_{ei}(x)$;
\item \label{asp:neighbor-ve-distinguishing} (Local ve-distinguishing) $\mathcal{E}_{ve}(x,\eta)\neq \mathcal{E}_{ve}(y,\eta)$ for each $y\in N_{ei}(x)$;
\item \label{asp:closed-neighbor-ve-distinguishing} (Closed local ve-distinguishing) $\mathcal{E}_{ve}[x,\eta]\neq \mathcal{E}_{ve}[y,\eta]$ for each $y\in N_{ei}(x)$;

\item \label{asp:universal-v-distinguishing} (Universal vertex distinguishing) $\mathcal{E}_v(x,\eta)\neq \mathcal{E}_v(w,\eta)$ for distinct vertices $x,w\in V(G)$;
\item \label{asp:closed-universal-v-distinguishing} (Closed universal vertex distinguishing) $\mathcal{E}_v[x,\eta]\neq \mathcal{E}_v[w,\eta]$ for distinct vertices $x,w\in V(G)$;
\item \label{asp:universal-e-distinguishing} (Universal edge distinguishing) $\mathcal{E}_e(x,\eta)\neq \mathcal{E}_e(w,\eta)$ for distinct vertices $x,w\in V(G)$;
\item \label{asp:closed-universal-e-distinguishing} (Universal edge distinguishing) $\mathcal{E}_e[x,\eta]\neq \mathcal{E}_e[w,\eta]$ for distinct vertices $x,w\in V(G)$;
\item \label{asp:universal-ve-distinguishing} (Universal ve-distinguishing) $\mathcal{E}_{ve}(x,\eta)\neq \mathcal{E}_{ve}(w,\eta)$ for distinct vertices $x,w\in V(G)$;
\item \label{asp:closed-universal-ve-distinguishing} (Closed universal ve-distinguishing) $\mathcal{E}_{ve}[x,\eta]\neq \mathcal{E}_{ve}[w,\eta]$ for distinct vertices $x,w\in V(G)$;
\item \label{asp:local-totally-ve-distinguishing} (Local (4)-totally ve-distinguishing) $\mathcal{E}_{ve}\{x,\eta\}\neq \mathcal{E}_{ve}\{y,\eta\}$ for each $y\in N_{ei}(x)$.

------ \emph{\textbf{distance}}

\item \label{asp:distance-v-distinguishing} ($\beta$-distance vertex distinguishing) $\mathcal{E}_v(u,\eta)\neq \mathcal{E}_v(v,\eta)$ for distinct vertices $u$ and $v$ with distance $d(u,v)\leq \beta$;
\item \label{asp:distance-closed-v-distinguishing} ($\beta$-distance closed vertex distinguishing) $\mathcal{E}_v[u,\eta]\neq \mathcal{E}_v[v,\eta]$ for distinct vertices $u$ and $v$ with
distance $d (u,v)\leq \beta$;
\item \label{asp:distance-e-distinguishing} ($\beta$-distance edge distinguishing) $\mathcal{E}_e(u,\eta)\neq \mathcal{E}_e(v,\eta)$ for distinct vertices $u$ and $v$ with distance $d(u,v)\leq \beta$;
\item \label{asp:distance-closed-e-distinguishing}($\beta$-distance closed edge distinguishing) $\mathcal{E}_e[u,\eta]\neq \mathcal{E}_e[v,\eta]$ for distinct vertices $u$ and $v$ with
distance $d (u,v)\leq \beta$ ;
\item \label{asp:distance-ve-distinguishing} ($\beta$-distance total distinguishing) $\mathcal{E}_{ve}(u,\eta)\neq \mathcal{E}_{ve}(v,\eta)$ for distinct vertices $u$ and $v$ with distance $d(u,v)\leq \beta$;
\item \label{asp:distance-closed-ve-distinguishing} ($\beta$-distance closed total distinguishing) $\mathcal{E}_{ve}[u,\eta]\neq \mathcal{E}_{ve}[v,\eta]$ for distinct vertices $u$ and $v$ with distance $d (u,v)\leq \beta$;

------ \emph{\textbf{equitable, acyclic}}

\item \label{asp:equitable-constraint} (Equitable sets) If $S\subseteq V(G)\cup E(G)$ holds $S=\bigcup ^k_{i=1}S_i$ such that no two elements of each subset $S_i$ with $i\in [1,k]$ are adjacent or
incident in $G$, also, subset $S_i$ is called an \emph{independent} (\emph{stable}) \emph{set}. $\big||S_i|-|S_j|\big |\leq 1$ with $i,j\in [1,k]$;
\item \label{asp:acyclic-constraint-3} (Acyclic property) By Co-\ref{asp:equitable-constraint}, the induced subgraph by $S_i\cup S_j$ with $i\neq j$ contains no cycle.
\end{asparaenum}

\noindent \textbf{Then, the coloring $\eta$ is}:

\noindent ------ \emph{proper}

\begin{asparaenum}[\textrm{\textbf{Setc}}-1. ]
\item a \emph{proper v-set-coloring} if the constraints Co-\ref{asp:vertex-only} and Co-\ref{asp:adjacent-vertices} hold true;
\item a \emph{proper e-set-coloring} if the constraints Co-\ref{asp:edge-only} and Co-\ref{asp:adjacent-edges} hold true;
\item a \emph{proper total set-coloring} if the constraints Co-\ref{asp:total}, Co-\ref{asp:adjacent-vertices}, Co-\ref{asp:adjacent-edges} and Co-\ref{asp:adjacent-vertices-edge} hold true;
\item a \emph{labeling} if the constraints Co-\ref{asp:vertex-only}, Co-\ref{asp:adjacent-vertices} and Co-\ref{asp:universal-vertices} hold true;

------ \emph{improper}

\item a \emph{v-set-coloring} if the constraint Co-\ref{asp:vertex-only} holds true;
\item an \emph{e-set-coloring} if the constraint Co-\ref{asp:edge-only} holds true;
\item a \emph{total set-coloring} if the constraint Co-\ref{asp:total} holds true;

------ \emph{improper distinguishing}

\item an \emph{adjacent-vertex distinguishing v-set-coloring} if the constraints Co-\ref{asp:vertex-only} and Co-\ref{asp:neighbor-distinguishing} hold true;
\item an \emph{adjacent-vertex distinguishing closed v-set-coloring} if the constraints Co-\ref{asp:vertex-only} and Co-\ref{asp:closed-neighbor-distinguishing} hold true;
\item an \emph{adjacent-vertex distinguishing e-set-coloring} if the constraints Co-\ref{asp:edge-only} and Co-\ref{asp:neighbor-e-distinguishing} hold true;
\item an \emph{adjacent-vertex distinguishing closed e-set-coloring} if the constraints Co-\ref{asp:edge-only} and Co-\ref{asp:closed-neighbor-e-distinguishing} hold true;
\item an \emph{adjacent-vertex distinguishing total set-coloring} if the constraints Co-\ref{asp:total} and Co-\ref{asp:neighbor-ve-distinguishing} hold true;
\item an \emph{adjacent-vertex distinguishing closed total set-coloring} if the constraints Co-\ref{asp:total} and Co-\ref{asp:closed-neighbor-ve-distinguishing} hold true;

------ \emph{local proper}

\item an \emph{adjacent-vertex distinguishing proper v-set-coloring} if the constraints Co-\ref{asp:vertex-only}, Co-\ref{asp:adjacent-vertices} and Co-\ref{asp:neighbor-distinguishing} hold true;
\item an \emph{adjacent-vertex distinguishing closed proper v-set-coloring} if the constraints Co-\ref{asp:vertex-only}, Co-\ref{asp:adjacent-vertices} and Co-\ref{asp:closed-neighbor-distinguishing} hold true;

\item an \emph{adjacent-vertex distinguishing proper e-set-coloring} if the constraints Co-\ref{asp:edge-only}, Co-\ref{asp:adjacent-edges} and Co-\ref{asp:neighbor-e-distinguishing} hold true;
\item an \emph{adjacent-vertex distinguishing closed proper e-set-coloring} if the constraints Co-\ref{asp:edge-only}, Co-\ref{asp:adjacent-edges} and Co-\ref{asp:closed-neighbor-e-distinguishing} hold true;
\item an \emph{adjacent-vertex distinguishing proper total set-coloring} if the constraints Co-\ref{asp:total}, Co-\ref{asp:adjacent-vertices}, Co-\ref{asp:adjacent-edges}, Co-\ref{asp:adjacent-vertices-edge} and Co-\ref{asp:neighbor-ve-distinguishing} hold true;
\item an \emph{adjacent-vertex distinguishing closed proper total set-coloring} if the constraints Co-\ref{asp:total}, Co-\ref{asp:adjacent-vertices}, Co-\ref{asp:adjacent-edges}, Co-\ref{asp:adjacent-vertices-edge} and Co-\ref{asp:closed-neighbor-ve-distinguishing} hold true;

------ \emph{distance}

\item a \emph{$\beta$-distance vertex distinguishing proper v-set-coloring} if the constraints Co-\ref{asp:vertex-only}, Co-\ref{asp:adjacent-vertices} and Co-\ref{asp:distance-v-distinguishing} hold true;
\item a \emph{$\beta$-distance vertex distinguishing closed proper v-set-coloring} if the constraints Co-\ref{asp:vertex-only}, Co-\ref{asp:adjacent-vertices} and Co-\ref{asp:distance-closed-v-distinguishing};

\item \label{Distance-distinguishing1} a \emph{$\beta$-distance vertex distinguishing proper e-set-coloring} if the constraints Co-\ref{asp:edge-only}, Co-\ref{asp:adjacent-edges} and Co-\ref{asp:distance-e-distinguishing} hold true;
\item a \emph{$\beta$-distance vertex distinguishing closed proper e-set-coloring} if the constraints Co-\ref{asp:edge-only}, Co-\ref{asp:adjacent-edges} and Co-\ref{asp:distance-closed-e-distinguishing};
\item a \emph{$\beta$-distance vertex distinguishing proper total set-coloring} if the constraints Co-\ref{asp:total}, Co-\ref{asp:adjacent-vertices}, Co-\ref{asp:adjacent-edges}, Co-\ref{asp:adjacent-vertices-edge} and Co-\ref{asp:distance-ve-distinguishing} hold true;

\item a \emph{$\beta$-distance vertex distinguishing closed proper total set-coloring} if the constraints Co-\ref{asp:total}, Co-\ref{asp:adjacent-vertices}, Co-\ref{asp:adjacent-edges}, Co-\ref{asp:adjacent-vertices-edge} and Co-\ref{asp:distance-closed-ve-distinguishing};

------ \emph{$(4)$-adjacent}

\item a \emph{$(4)$-adjacent-vertex distinguishing closed proper total set-coloring} if the constraints Co-\ref{asp:total}, Co-\ref{asp:adjacent-vertices}, Co-\ref{asp:adjacent-edges}, Co-\ref{asp:adjacent-vertices-edge} and Co-\ref{asp:local-totally-ve-distinguishing} hold true;

------ \emph{universal proper}

\item a \emph{vertex distinguishing proper v-set-coloring} if the constraints Co-\ref{asp:vertex-only}, Co-\ref{asp:adjacent-vertices} and Co-\ref{asp:universal-v-distinguishing} hold true;
\item a \emph{vertex distinguishing closed proper v-set-coloring} if the constraints Co-\ref{asp:vertex-only}, Co-\ref{asp:adjacent-vertices} and Co-\ref{asp:closed-universal-v-distinguishing} hold true;
\item a \emph{vertex distinguishing proper e-set-coloring} if the constraints Co-\ref{asp:edge-only}, Co-\ref{asp:adjacent-edges} and Co-\ref{asp:universal-e-distinguishing} hold true;
\item a \emph{vertex distinguishing closed proper e-set-coloring} if the constraints Co-\ref{asp:edge-only}, Co-\ref{asp:adjacent-edges} and Co-\ref{asp:closed-universal-e-distinguishing} hold true;

\item a \emph{vertex distinguishing proper total set-coloring} if the constraints Co-\ref{asp:total}, Co-\ref{asp:adjacent-vertices}, Co-\ref{asp:adjacent-edges}, Co-\ref{asp:adjacent-vertices-edge} and Co-\ref{asp:universal-ve-distinguishing} hold true;
\item a \emph{vertex distinguishing closed proper total set-coloring} if the constraints Co-\ref{asp:total}, Co-\ref{asp:adjacent-vertices}, Co-\ref{asp:adjacent-edges}, Co-\ref{asp:adjacent-vertices-edge} and Co-\ref{asp:closed-universal-ve-distinguishing} hold true;

------ \emph{equitable}

\item an \emph{equitably adjacent-vertex distinguishing proper e-set-coloring} if the constraints Co-\ref{asp:edge-only}, Co-\ref{asp:adjacent-edges}, Co-\ref{asp:neighbor-e-distinguishing} and Co-\ref{asp:equitable-constraint} hold true;
\item an \emph{equitably adjacent-vertex distinguishing closed proper e-set-coloring} if the constraints Co-\ref{asp:edge-only}, Co-\ref{asp:adjacent-edges}, Co-\ref{asp:closed-neighbor-e-distinguishing} and Co-\ref{asp:equitable-constraint} hold true;
\item an \emph{equitably adjacent-vertex distinguishing proper total set-coloring} if the constraints Co-\ref{asp:total}, Co-\ref{asp:adjacent-vertices}, Co-\ref{asp:adjacent-edges}, Co-\ref{asp:adjacent-vertices-edge}, Co-\ref{asp:neighbor-ve-distinguishing} and Co-\ref{asp:equitable-constraint} hold true;

\item an \emph{equitably vertex distinguishing proper e-set-coloring} if the constraints Co-\ref{asp:edge-only}, Co-\ref{asp:adjacent-edges}, Co-\ref{asp:universal-e-distinguishing} and Co-\ref{asp:equitable-constraint} hold true;
\item an \emph{equitably vertex distinguishing proper total set-coloring} if the constraints Co-\ref{asp:total}, Co-\ref{asp:adjacent-vertices}, Co-\ref{asp:adjacent-edges}, Co-\ref{asp:adjacent-vertices-edge}, Co-\ref{asp:closed-universal-ve-distinguishing} and Co-\ref{asp:equitable-constraint} hold true;

------ \emph{acyclic}

\item an \emph{acyclic adjacent-vertex distinguishing proper e-set-coloring} if the constraints Co-\ref{asp:edge-only}, Co-\ref{asp:adjacent-edges}, Co-\ref{asp:neighbor-e-distinguishing} and Co-\ref{asp:acyclic-constraint-3} hold true;
\item \label{color:acyclic-adjacent-vertex-distinguishing-total} an \emph{acyclic adjacent-vertex distinguishing proper total set-coloring} if the constraints Co-\ref{asp:total}, Co-\ref{asp:adjacent-vertices}, Co-\ref{asp:adjacent-edges}, Co-\ref{asp:adjacent-vertices-edge}, Co-\ref{asp:neighbor-ve-distinguishing} and Co-\ref{asp:acyclic-constraint-3} hold true;

\item an \emph{acyclic vertex distinguishing closed proper total set-coloring} if the constraints Co-\ref{asp:total}, Co-\ref{asp:adjacent-vertices}, Co-\ref{asp:adjacent-edges}, Co-\ref{asp:adjacent-vertices-edge}, Co-\ref{asp:closed-neighbor-ve-distinguishing} and Co-\ref{asp:acyclic-constraint-3} hold true.\qqed
\end{asparaenum}
\end{defn}

\begin{problem}\label{question:444444}
\textbf{Find} the smallest number $m$ of elements of the set-set $S_{et}=\{e_{1},e_{2},\dots ,e_{m}\}$, such that the Setc-$k$ set-coloring for $k\in [1,40]$ in Definition \ref{defn:more-distinguishing-set-colorings} holds true. However, \textbf{determining} the smallest number $m$ for all set-sets $S_{et}$ will be related with many \emph{graph coloring parameters} introduced in \cite{Bondy-2008}, \cite{Gallian2022} and \cite{Yao-Wang-2106-15254v1}.

In graph theory, there are chromatic numbers and chromatic indexes as follows: the \emph{total chromatic number} $\chi\,''(G)$, the \emph{adjacent-vertex distinguishing chromatic index} $\chi\,'_{as}(G)$, the \emph{adjacent-vertex distinguishing closed chromatic index} $\chi\,'_{cas}(G)$, the \emph{adjacent-vertex distinguishing total chromatic number} $\chi\,''_{as}(G)$, the \emph{adjacent-vertex distinguishing closed total chromatic number} $\chi\,''_{cas}(G)$, the $(4)$-\emph{adjacent-vertex distinguishing closed chromatic number} $\chi\,''_{(4)cas}(G)$, the \emph{strongly chromatic index} $\chi\,'_s(G)$, the \emph{vertex distinguishing total chromatic number} $\chi\,''_s(G)$, the \emph{adjacent-vertex distinguishing chromatic index} $\chi\,'_{eas}(G)$, the \emph{equitably adjacent-vertex distinguishing total chromatic number} $\chi\,''_{eas}(G)$, the \emph{equitably vertex distinguishing total chromatic number} $\chi\,''_{es}(G)$, the \emph{acyclic adjacent-vertex distinguishing chromatic index} $\chi\,'_{aas}(G)$, the \emph{acyclic adjacent-vertex distinguishing total chromatic number} $\chi\,''_{aas}(G)$, and the \emph{acyclic adjacent-vertex distinguishing closed total chromatic number} $\chi\,''_{caas}(G)$.

Since determining the chromatic number is NP-hard (Ref. \cite{Garey-M-R-Johnson-1979} and \cite{Garey-Johnson-Stockmeyer-1974}), we can get many sharp-P-complete and sharp-P-hard problems from graph chromatic numbers and graph chromatic indexes.

Various vertex distinguishing colorings of graph colorings of graph theory are not difficult to induce some set-colorings, or set-labelings defined here. However, it seems to be not easy to induce graph colorings and graph labelings by means of given set-colorings, or set-labelings, although set-colorings, or set-labelings are useful in designing complex number-based strings for asymmetric topology cryptography.\qqed
\end{problem}

The set-coloring has been introduced in \cite{Yao-Yang-Yao-2020-distinguishing}, and the hypergraph-coloring has been investigated in \cite{Bing-et-al-arXiv-asymmetric-4520331}.

\begin{defn} \label{defn:111111}
$^*$ By Definition \ref{defn:more-distinguishing-set-colorings}, Definition \ref{defn:hypergraph-basic-definition} and Remark \ref{rem:hypergraph-terminology-notations}, we have the following distinguishing-type set-colorings:

(i) If the color set $S_{et}=\mathcal{E}$ is a hyperedge set $\mathcal{E} \in \mathcal{E}(\Lambda^2)$ of a hypergraph $\mathcal{H}_{yper}=(\Lambda,\mathcal{E})$, immediately, we can obtain various \emph{distinguishing-type hypergraph-colorings} obtained by ``hypergraph-coloring'' to replace ``set-coloring'' defined in Definition \ref{defn:more-distinguishing-set-colorings}.

(ii) If the color set $S_{et}$ is an every-zero hypergraph group $\big \{G(\mathcal{E});[+][-]\big \}$ generated by a hyperedge set $\mathcal{E}\in \mathcal{E}(\Lambda^2)$ of a hypergraph $\mathcal{H}_{yper}=(\Lambda,\mathcal{E})$, then we can get various \emph{distinguishing-type graphic-group colorings} by ``graphic-group coloring'' to replace ``set-coloring'' defined in Definition \ref{defn:more-distinguishing-set-colorings}.

(iii) There are \emph{distinguishing-type sequence coloring}, \emph{distinguishing-type string-coloring}, \emph{etc}.\qqed
\end{defn}

\section{Set-Colorings And Hypergraphs}

\subsection{Concepts of hypergraphs}

\begin{defn}\label{defn:hypergraph-basic-definition}
\cite{Jianfang-Wang-Hypergraphs-2008} A \emph{hyperedge set} $\mathcal{E}$ is a family of distinct non-empty subsets $e_1,e_2$, $\dots $, $e_n$ of the power set $\Lambda^2$ based on a finite set $\Lambda=\{x_1,x_2,\dots ,x_n\}$, and satisfies:

(i) Each element $e\in \mathcal{E}$, called \emph{hyperedge}, holds $e\neq \emptyset $ true;

(ii) $\Lambda=\bigcup _{e\in \mathcal{E}}e$, where each element of $\Lambda$ is called a \emph{vertex}.

The symbol $\mathcal{H}_{yper}=(\Lambda,\mathcal{E})$ stands for a \emph{hypergraph} with its own \emph{hyperedge set} $\mathcal{E}$ defined on the \emph{vertex set} $\Lambda$, and the cardinality $|\mathcal{E}|$ is the \emph{size}, and the cardinality $|\Lambda|$ is the \emph{order} of the hypergraph $\mathcal{H}_{yper}$. \qqed
\end{defn}

\begin{example}\label{exa:8888888888}
Fig.\ref{fig:1-example-hypergraph} shows us a hypergraph $\mathcal{H}_{yper}=(\Lambda,\mathcal{E})$ with its own vertex set $\Lambda=[1,15]$ and its own hyperedge set $\mathcal{E}=\{e_1,e_2,e_3,e_4\}$.\qqed
\end{example}

\begin{figure}[h]
\centering
\includegraphics[width=15cm]{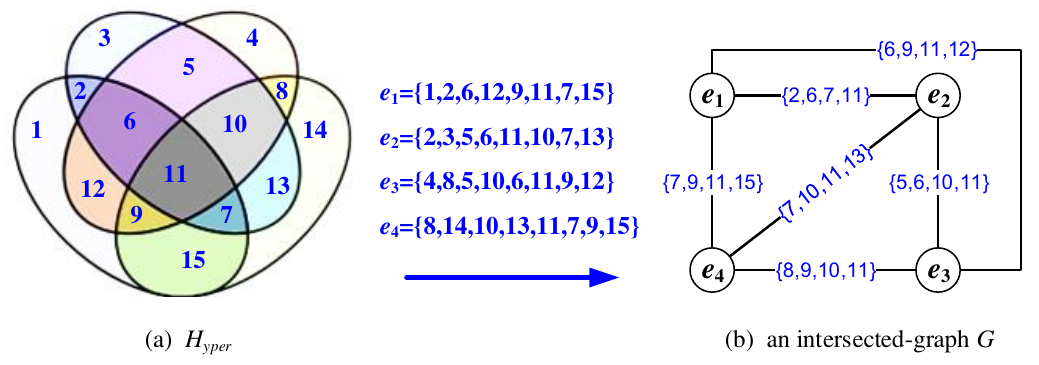}\\
\caption{\label{fig:1-example-hypergraph}{\small An example from an $8$-uniform hypergraph $H_{yper}$ to a vertex-intersected graph $G$ admitting a set-coloring subject to the constraint set $R_{set}(1)$, where (a) Venn's four-set diagram using four ellipses.}}
\end{figure}

\begin{rem}\label{rem:hypergraph-terminology-notations}
About Definition \ref{defn:hypergraph-basic-definition} and a finite set $\Lambda=\{x_1,x_2,\dots ,x_n\}$, there are the following terminology and notation of hypergraphs:
\begin{asparaenum}[$\bullet$~]
\item The set $\mathcal{E}\big (\Lambda^2\big )=\{\mathcal{E}_i:~i\in [1,n(\Lambda)]\}$ is called \emph{hypergraph set} based on the power set $\Lambda^2$, where each $\mathcal{E}_i$ is a hyperedge set, and $n(\Lambda)$ is the number of hyperedge sets based on the power set $\Lambda^2$, and the cardinality $|\Lambda^2|=2^n-1$. So, we have $n(\Lambda)$ hypergraphs $\mathcal{H}_{yper}=(\Lambda,\mathcal{E})$ for each hyperedge set $\mathcal{E}\in \mathcal{E}\big (\Lambda^2\big )$.
\item A hyperedge set $\mathcal{E}$ is \emph{proper} if any subset $e\in \mathcal{E}$ is not a subset of each $e\,'\in \mathcal{E}\setminus e$.
\item \cite{Jianfang-Wang-Hypergraphs-2008} An \emph{ear} $e\in \mathcal{E}$ holds:

\qquad (i) $e\cap e\,'=\emptyset$ for any hyperedge $e\,'\in \mathcal{E}\setminus \{e\}$; or

\qquad (ii) there exists another hyperedge $e^*\in \mathcal{E}$, such that each vertex of $e\setminus e^*$ is not in any element of $\mathcal{E}\setminus \{e\}$.

\item An \emph{isolated vertex} $x\in \Lambda$ belongs to a unique hyperedge $e_i\in \mathcal{E}$, such that $x\not\in e_j\in \mathcal{E}$ if $i\neq j$.
\item If each \emph{hyperedge} $e\in \mathcal{E}$ has its cardinality $|e|=r$, then we call $\mathcal{H}_{yper}=(\Lambda,\mathcal{E})$ \emph{$r$-uniform hypergraph}.
\item \cite{Jianfang-Wang-Hypergraphs-2008} A \emph{partial hypergraph} of a hypergraph $\mathcal{H}_{yper}=(\Lambda,\mathcal{E})$ has its own hyperedge set $\mathcal{E}^*\subset \mathcal{E}$.
\item \cite{Jianfang-Wang-Hypergraphs-2008} The \emph{Graham reduction} of a hyperedge set $\mathcal{E}$ is obtained by doing repeatedly

\qquad GR-1: delete a vertex $x$ if $x$ is an isolated vertex;

\qquad GR-2: delete $e_i$ if $e_i\subseteq e_j$ for $i\neq j$.
\item A hypergraph $\mathcal{H}_{yper}=(\Lambda,\mathcal{E})$ is called \emph{reduced hypergraph}, or \emph{simple hypergraph} if its hyperedge set $\mathcal{E}$ is the result of Graham reduction.
\item Suppose that $x_{i_1},x_{i_2},\dots ,x_{i_n}$ are a permutation of vertices of a hypergraph $\mathcal{H}_{yper}=(\Lambda,\mathcal{E})$, where $x_{i_j}\in e_{j-1}\cap e_j$, $x_{i_{n-1}}\in e_{n-1}\cap e_n$ and $x_{i_n}\in e_{n}\cap e_1$, then a \emph{hyperpath} is defined as $\mathcal{P}(e_1,e_n)=e_1e_2\cdots e_n$, and a \emph{hypercycle} is defined as $\mathcal{C}=e_1e_2\cdots e_ne_1$, and moreover $\mathcal{C}=e_1e_2\cdots e_ne_1$ is a \emph{Hamilton hypercycle} if $n=|\Lambda|$.
\item If each pair of subsets $e$ and $e\,'$ of the hyperedge set $\mathcal{E}$ corresponds a hyperpath $\mathcal{P}(e,e\,')$, then the hypergraph $\mathcal{H}_{yper}=(\Lambda,\mathcal{E})$ is \emph{connected}.
\item The \emph{hyperedge norm} $||\mathcal{E}||$ of a proper hyperedge set $\mathcal{E}\in \mathcal{E}\big (\Lambda^2\big )$ is determined as $||\mathcal{E}||=\sum_{e\in \mathcal{E}} |e|$, clearly,
\begin{equation}\label{eqa:555555}
n=\min \big \{||\mathcal{E}||: \mathcal{E}\in \mathcal{E}\big (\Lambda^2\big )\big \},\quad \frac{n(n-1)\cdots (n-k+1)}{k!}\leq \max \big \{||\mathcal{E}||: \mathcal{E}\in \mathcal{E}\big (\Lambda^2\big )\big \}
\end{equation}
\item For each vertex $x_j\in \Lambda$ with $j\in [1,n]$, the number of $x_j$ appeared in the subsets $e_{i,1}$, $e_{i,2}$, $\dots $, $e_{i,b_j}$ of a hyperedge set $\mathcal{E}_i$ is denoted as $b_j=\textrm{deg}_{\mathcal{E}_i}(x_i)$, called \emph{hypervertex degree}.
\item The hyperedge degree of a hyperedge $e\in \mathcal{E}$ is defined in Definition \ref{defn:more-terminology-group}.\qqed
\end{asparaenum}
\end{rem}

\begin{rem}\label{rem:important-terminology-hypergraphs}
About Definition \ref{defn:hypergraph-basic-definition}, we have: Since each $x_i\in \Lambda$ is a number in many books of hypergraphs in general. We can consider other cases: each $x_i\in \Lambda$ is a graph, or a matrix, or a string, or a vector, or a set, or a hypergraph, or any thing in the world.

In real application, a hypergraph $\mathcal{H}_{yper}=(\Lambda,\mathcal{E})$ is a network, and each subset $e_i\in \mathcal{E}$ can be seen as a community, or a local network \emph{etc}.\qqed
\end{rem}

\begin{problem}\label{question:444444}
We propose the following questions:
\begin{asparaenum}[\textbf{Hyper}-1]
\item For each integer $m$ subject to $n<m<n(n-1)$, \textbf{is} there a proper hyperedge set $\mathcal{E}^*\in \mathcal{E}(\Lambda^2)$ such that the hyperedge norm $||\mathcal{E}^*||=m$? \textbf{Find} connections between the elements of the hypergraph set $\mathcal{E}(\Lambda^2)$.
\item \textbf{What} connections are there in $\mathcal{H}_{yper}=(\Lambda,\mathcal{E})$ and $\mathcal{H}\,'_{yper}=(\Lambda,\mathcal{E}\,')$, as $\mathcal{E}\neq \mathcal{E}\,'$ based on the same set $\Lambda$.
\item \textbf{Is} there a connection between two hypergraphs $\mathcal{H}_{yper}=(\Lambda,\mathcal{E})$ and $\mathcal{H}^*_{yper}=(\Lambda^*,\mathcal{E}^*)$, as $\Lambda\neq \Lambda^*$, $\Lambda\not\subset \Lambda^*$ and $\Lambda^*\not\subset \Lambda$?
\item About the hypergraph set $\mathcal{E}\big (\Lambda^2\big )=\{\mathcal{E}_i:~i\in [1,n(\Lambda)]\}$, however, no report is for computing the number $n(\Lambda)$ of all hyperedge sets based on a finite set $\Lambda$. \textbf{Determine} the number $n(\Lambda)$ of all hyperedge sets of the finite set $\Lambda$.
\item Since $\overline{\mathcal{E}}=\{\overline{e}_i:i\in [1,b]\}=\{\Lambda \setminus \{e_i\}:i\in [1,b]\}$, \textbf{find} each matching $(\mathcal{E},\overline{\mathcal{E}})$ of hyperedge sets $\mathcal{E}$ (as a \emph{private-key}) and $\overline{\mathcal{E}}$ (as a \emph{public-key}) in $\mathcal{E}\big (\Lambda^2\big )$.
\end{asparaenum}
\end{problem}

\begin{problem}\label{question:444444}
\textbf{Extreme problem.} \textbf{Find} a finite set $\Lambda_{*}$ for which a graph $G$ admits a total set-coloring $F:V(G)\cup E(G)\rightarrow \mathcal{E}^{*}\in \mathcal{E}(\Lambda^2_{*})$ with $\Lambda_{*}=\bigcup_{e\in \mathcal{E}^{*}}e$, such that the graph $G$ admits a total set-coloring $f:V(G)\cup E(G)\rightarrow \mathcal{E}\in \mathcal{E}(\Lambda^2)$ with $\Lambda=\bigcup_{s\in \mathcal{E}}s$ holds $|\Lambda_{*}|\leq |\Lambda|$ and one or more of the following constraints:
\begin{asparaenum}[(i) ]
\item Two hyperedge sets $\mathcal{E}^{*}\in \mathcal{E}(\Lambda^2_{*})$ and $\mathcal{E}\in \mathcal{E}(\Lambda^2)$ are proper.
\item Each subset $e\in \mathcal{E}^{*}\in \mathcal{E}(\Lambda^2_{*})$ corresponds another subset $e\,'\in \mathcal{E}^{*}\in \mathcal{E}(\Lambda^2_{*})$ holding $e\cap e\,'\neq \emptyset$.
\item $F(uv)\supseteq F(u)\cap F(v)\neq \emptyset$ for each edge $uv\in E(G)$.
\item Each subset $s\in \mathcal{E}\in \mathcal{E}(\Lambda^2)$ corresponds another subset $s\,'\in \mathcal{E}\in \mathcal{E}(\Lambda^2)$ holding $s\cap s\,'\neq \emptyset$.
\item $f(xy)\supseteq f(x)\cap f(y)\neq \emptyset$ for each edge $xy\in E(G)$.
\end{asparaenum}
\end{problem}

\begin{thm}\label{thm:666666}
$^*$ If the 4-color conjecture of maximal planar graphs holds true, then each maximal planar graph $G$ admits a \emph{proper total set-coloring}
$$F:V(G)\cup E(G)\rightarrow \mathcal{E}=\big \{\{1,2\},\{1,3\},\{2,3\},\{1,2,3\}\big \}
$$ or, a \emph{proper total string-coloring} $F:V(G)\cup E(G)\rightarrow \big \{11,22,12,21\big \}$, such that $F(uv)=F(u)\cap F(v)$ for each edge $uv\in E(G)$.
\end{thm}

\begin{prop}\label{prop:99999}
$^*$ Let $\Lambda=\{x_1,x_2,\dots ,x_n\}$ be a finite set. We can observe the following results:

(i) Each subset $e\in \Lambda^2$ is in some hyperedge set $\mathcal{E}_i\in \mathcal{E}\big (\Lambda^2\big )$. Otherwise, we have a new hyperedge set $\mathcal{E}^*_i=\mathcal{E}_i\cup \{e\}$, such that $\mathcal{E}^*_i\not\in \mathcal{E}\big (\Lambda^2\big )$, a contradiction with the definition of $\mathcal{E}\big (\Lambda^2\big )$.

(ii) Any pair of two hyperedge sets $\mathcal{E}_i$ and $\mathcal{E}_j$ of $\mathcal{E}\big (\Lambda^2\big )$ holds the \emph{union operation} $\mathcal{E}_i\cup \mathcal{E}_j\in \mathcal{E}\big (\Lambda^2\big )$ true, where
\begin{equation}\label{eqa:union-operation-hypergraph}
\mathcal{E}_i\cup \mathcal{E}_j=\big [\mathcal{E}_i \setminus (\mathcal{E}_i\cap \mathcal{E}_j)\big ]\bigcup \big [\mathcal{E}_j \setminus (\mathcal{E}_i\cap \mathcal{E}_j)\big ]\bigcup (\mathcal{E}_i\cap \mathcal{E}_j)
\end{equation}

(iii) There are two particular hyperedge sets $\mathcal{E}_0=\Lambda=\{y_1,y_2, \dots , y_m\}$ and $\mathcal{E}\,'=\big \{\{y_1\}$, $\{y_2\}$, $\dots $, $\{y_m\}\big \}$. Any hyperedge set $\mathcal{E}_i\in \mathcal{E}\big (\Lambda^2\big )$ holds $2=|\mathcal{E}_0|+1\leq |\mathcal{E}_i|\leq 2^m-2-m$ true.
\end{prop}

\begin{thm}\label{thm:key-matching-equitable-uniform}
$^*$ The hypergraph set $\mathcal{E}\big (\Lambda^2\big )$ holds $|\mathcal{E}\big (\Lambda^2\big )\setminus \mathcal{E}_0|=$even with $\mathcal{E}_0=\Lambda$, such that the hypergraph set $\mathcal{E}\big (\Lambda^2\big ))\setminus \mathcal{E}_0=X\cup X_{comp}$ with $X\cap X_{comp}=\emptyset$, each hyperedge set $\mathcal{E}\in X$ corresponds to a hyperedge set $\overline{\mathcal{E}}\in X_{comp}$, such that each set $\overline{e}_i\in \overline{\mathcal{E}}$ equals to $\overline{e}_i=\Lambda\setminus e_i$ with each subset $e_i\in \mathcal{E}$ for $i\in [1,m]$, where $\mathcal{E}=\{e_i\}^m_{i=1}$ and $\overline{\mathcal{E}}=\{\overline{e}_i\}^m_{i=1}$, and $|X|=| X_{comp}|$. Moreover, there are:

(i) If a hyperedge set $\mathcal{E}\in X$ is \emph{$k$-uniform}, namely, $|e_i|=|e_j|=k$ for any pair of sets $e_i$ and $e_j$ of the hyperedge set $\mathcal{E}$ (as a \emph{private-key}), then $\overline{\mathcal{E}}$ (as a \emph{public-key}) is $\overline{k}$-uniform too, where $\overline{k}=|\Lambda|-k$.

(ii) If a hyperedge set $\mathcal{E}\in X$ is \emph{equitable}, namely, $\big ||e_i|-|e_j|\big |\leq 1$ for any pair of sets $e_i$ and $e_j$ of the hyperedge set $\mathcal{E}$ (as a \emph{private-key}), then $\overline{\mathcal{E}}$ (as a \emph{public-key}) is equitable too.

(iii) If $\mathcal{E}\in X$ (as a \emph{private-key}) is inequality from each other, namely, $|e_i|\neq |e_j|$ as $i\neq j$, then $\overline{\mathcal{E}}$ (as a \emph{public-key}) is inequality from each other too.
\end{thm}

\begin{cor}\label{cor:99999}
$^*$ If a graph $G$ admits a total set-coloring $F:V(G)\cup E(G)\rightarrow \mathcal{E}_0\in \mathcal{E}\big (\Lambda^2\big )$, then the graph $G$ admits total set-colorings $F_i:V(G)\cup E(G)\rightarrow \mathcal{E}_i\in\{G(\mathcal{E}_0);[+][-]\}$ for $i\in [1,m]$, where the hypergraph group coloring is defined in Definition \ref{defn:general-defi-hypergraph-groups}.
\end{cor}

\begin{problem}\label{qeu:444444}
Theorem \ref{thm:key-matching-equitable-uniform} enables us to build up a key-matching pair of sets $X$ (as a public-key set) and $X_{comp}$ (as a private-key set). However, it is important \textbf{to know} the cardinality $|X|$ and the topological structures of hyperedge sets of two sets $X$ and $X_{comp}$ for real applications.\qqed
\end{problem}

\begin{defn} \label{defn:set-hyperedge-sets}
\cite{Yao-Ma-arXiv-2201-13354v1} Let $H^{yper}_{color}(G)$ be the set of hyperedge sets of all set-colorings of a connected graph $G$, such that each hyperedge set $\mathcal{E}\in H^{yper}_{color}(G)$ defines a set-coloring $F:V(G)\rightarrow \mathcal{E}$, where $\bigcup_{e\in \mathcal{E}}e=\Lambda$, and the connected graph $G$ is a vertex-intersected graph of a hypergraph $\mathcal{H}_{yper}=(\Lambda,\mathcal{E})$. For another set-coloring
$$F\,':V(G)\rightarrow \mathcal{E}\,'\in H^{yper}_{color}(G)
$$ with $\bigcup_{e\in \mathcal{E}\,'}e=\Lambda\,'$, we defined the third set-coloring $F^*:V(G)\rightarrow \mathcal{E}^*$, such that $F^*(x)=F(x)\cup F\,'(x)$ for each vertex $x\in V(G)$ and $F^*(uv)=F(uv)\cup F\,'(uv)$ for each edge $uv\in E(G)$. Clearly, $\mathcal{E}^*=\mathcal{E}\cup \mathcal{E}\,'$, and

(i) $\Lambda^*=\Lambda=\Lambda\,'$; or

(ii) $\Lambda^*=\Lambda \cup \Lambda\,'$ if $\Lambda \neq \Lambda\,'$. \\
We call the set-coloring $F^*$ \emph{union set-coloring} of two set-colorings $F$ and $F\,'$, the hyperedge set $F^*$ \emph{hyperedge union set} of two hyperedge sets $\mathcal{E}$ and $\mathcal{E}\,'$. \qqed
\end{defn}

\begin{thm}\label{thm:666666}
$^*$ (i) Any connected graph $H$ admits a total set-coloring $f:V(H)\cup E(H)\rightarrow \mathcal{E}\in \mathcal{E}(\Lambda^2)$ with $\Lambda=\bigcup_{e\in \mathcal{E}}e$, such that $f(u)\cap f(v)\neq \emptyset$ for each edge $uv\in E(G)$.

(ii) A complete graph $K_n$ of $n$ vertices admits a total set-coloring $F:V(K_n)\cup E(K_n)\rightarrow \mathcal{E}\in \mathcal{E}([1,n]^2)$ with $[1,n]=\bigcup_{e\in \mathcal{E}}e$, such that $F(uv)\supseteq F(u)\cap F(v)\neq \emptyset$ for each edge $uv\in E(K_n)$, and the hyperedge set $\mathcal{E}$ is \emph{proper}.
\end{thm}

\begin{problem}\label{qeu:444444}
The set $H^{yper}_{color}(G)$ defined in Definition \ref{defn:set-hyperedge-sets} can be classified into two parts
\begin{equation}\label{eqa:555555}
H^{yper}_{color}(G)=I^{yper}_{color}(G)\bigcup N^{yper}_{color}(G)
\end{equation} such that each hyperedge set of $I^{yper}_{color}(G)$ (as a \emph{public-key set}) is not the union set of any two sets of the set $H^{yper}_{color}(G)$ (as a \emph{private-key set}), but each hyperedge set of $N^{yper}_{color}(G)$ is the hyperedge union set of some two hyperedge sets of $H^{yper}_{color}(G)$. \textbf{Determine} the hyperedge set $I^{yper}_{color}(G)$ for a connected graph $G$.
\end{problem}

Motivated from Definition \ref{defn:set-hyperedge-sets}, we have the following the union operation of hyperedge sets:

\begin{prop}\label{prop:hyperedge-sets-unions}
$^*$ \textbf{Union operation of hyperedge sets.} Let two finite sets $\Lambda_a$ (as a \emph{public-ky set}) and $\Lambda_b$ (as a \emph{private-ky set}) hold $\Lambda_a\cap \Lambda_b\neq \emptyset$, and $\mathcal{E}_a\in \mathcal{E}(\Lambda_a^2)$ and $\mathcal{E}_b\in \mathcal{E}(\Lambda_b^2)$. So, there is a new hyperedge set $\mathcal{E}^*=\mathcal{E}_a\cup \mathcal{E}_b$ (as a \emph{authentication set}) holding the new finite set $\Lambda=\bigcup_{e\in \mathcal{E}^*}e$ true, where $\Lambda=\Lambda_a\cup \Lambda_b$.
\end{prop}

\begin{rem}\label{rem:333333}
By Proposition \ref{prop:hyperedge-sets-unions}, there is some hyperedge set $\mathcal{E}_k\in \mathcal{E}\big (\Lambda^2\big )$ holding the hyperedge set $\mathcal{E}_k\setminus \mathcal{E}^*_k=\mathcal{E}_k\,''\in \mathcal{E}(\Lambda_k^2)$ with $k=a,b$, such that the hyperedge set $\mathcal{E}_k= \mathcal{E}^*_k\cup \mathcal{E}_k\,''$, however, not necessarily the set $ \mathcal{E}^*_k\in \mathcal{E}(\Lambda_k^2)$, even the set $\mathcal{E}^*_k$ is not a hyperedge set based on two finite sets $\Lambda_a$ and $\Lambda_b$.

(i) If the finite set $\Lambda_a$ is a proper subset of the finite set $\Lambda_b$, namely $\Lambda_a\subset \Lambda_b$, then the result of Proposition \ref{prop:hyperedge-sets-unions} is still valid.

(ii) If $\Lambda_a\cap \Lambda_b=\emptyset$ in Proposition \ref{prop:hyperedge-sets-unions}, then the \emph{union hyperedge set} $\mathcal{E}= \mathcal{E}\,'\cup \mathcal{E}\,''$ with $\mathcal{E}\,'\in \mathcal{E}(\Lambda_a^2)$ and $\mathcal{E}\,''\in \mathcal{E}(\Lambda_b^2)$ forms a \emph{hyperedge dis-connected hypergraph}.\qqed
\end{rem}

\subsection{Hypergraph homomorphism}

In \cite{Yao-Su-Ma-Wang-Yang-arXiv-2202-03993v1}, the authors have introduced (colored) graph homomorphism, $W$-constraint graph homomorphism, graph-operation graph homomorphism. We will study \emph{hypergraph homomorphisms} in this subsection.

\begin{prop}\label{prop:99999}
$^*$ For a vertex-intersected graph $H$ of a hypergraph $\mathcal{H}_{yper}=(\Lambda,\mathcal{E})$ and another vertex-intersected graph $G$ of the hypergraph $\mathcal{H}^*_{yper}=(\Lambda,\mathcal{E}^*)$, if there is a graph homomorphism $H\rightarrow G$, then we have a \emph{hypergraph homomorphism}
$$
\mathcal{H}_{yper}=(\Lambda,\mathcal{E})\rightarrow \mathcal{H}^*_{yper}=(\Lambda,\mathcal{E}^*)
$$ so we have projected $\mathcal{H}_{yper}=(\Lambda,\mathcal{E})$ onto $\mathcal{H}^*_{yper}=(\Lambda,\mathcal{E}^*)$.
\end{prop}

\begin{defn} \label{defn:hypergraph-operation-homomorphisms}
$^*$ \textbf{Hypergraph homomorphism.} For two hyperedge sets $\mathcal{E}_i$ and $\mathcal{E}_j$ of the hypergraph set $\mathcal{E}\big (\Lambda^2\big )$, we have:

(i) If the exists a coloring $\varphi:\mathcal{E}_i\rightarrow \mathcal{E}_j$ such that each hyperedge edge $e_{i,s}\in \mathcal{E}_i$ corresponds to its own image $\varphi(e_{i,s})\in \mathcal{E}_j$, and moreover two hyperedges $e_{i,s}$ and $e_{i,t}$ hold a property $P$ in $\mathcal{E}_i$ if and only if two hyperedges $\varphi(e_{i,s})$ and $\varphi(e_{i,t})$ of the hyperedge set $\mathcal{E}_j$ hold this property $P$ too, we say that $\mathcal{E}_i$ is \emph{hypergraph homomorphic} to $\mathcal{E}_j$ based on the property $P$, and this \emph{$P$-property hypergraph homomorphism} is denoted as $\mathcal{E}_i\rightarrow _P\mathcal{E}_j$.

(ii) For an operation ``$[\bullet]$'' on the hypergraph set $\mathcal{E}\big (\Lambda^2\big )$, a \emph{hypergraph $[\bullet]$-operation homomorphism} $(\mathcal{E}_i, \mathcal{E}_j)\rightarrow \mathcal{E}_{i[\bullet] j}$ if $\mathcal{E}_i[\bullet] \mathcal{E}_j=\mathcal{E}_{i[\bullet] j}\in \mathcal{E}\big (\Lambda^2\big )$. \qqed
\end{defn}

\begin{problem}\label{qeu:444444}
As the operation $[\bullet]=\bigcup$ in Definition \ref{defn:hypergraph-operation-homomorphisms}, we have a hypergraph $\bigcup$-operation homomorphism $(\mathcal{E}_i, \mathcal{E}_j)\rightarrow \mathcal{E}_{i\cup j}$ since $\mathcal{E}_{i\cup j}=\mathcal{E}_i\cup \mathcal{E}_j$ defined in Formula (\ref{eqa:union-operation-hypergraph}). However, it is useful to find more hypergraph $[\bullet]$-operation homomorphisms.
\end{problem}

\begin{defn} \label{defn:111111}
$^*$ By Definition \ref{defn:hypergraph-basic-definition} and Definition \ref{defn:hypergraph-operation-homomorphisms}, we have:

(i) A graph $G$ admits a \emph{total hypergraph-set coloring} $\beta:V(G)\cup E(G)\rightarrow \mathcal{E}\big (\Lambda^2\big )$, such that $\beta(uv)=\beta(u)\cup \beta(v)$ for each edge $uv\in E(G)$ holds the hypergraph $\bigcup$-operation homomorphism
\begin{equation}\label{eqa:555555}
(\beta(u), \beta(v))=(\mathcal{E}_i, \mathcal{E}_j)\rightarrow \mathcal{E}_{i\cup j}=\beta(uv)
\end{equation} where the hypergraph set $\mathcal{E}\big (\Lambda^2\big )$ is defined in Definition \ref{defn:hypergraph-basic-definition}.

(ii) A graph $G$ admits a \emph{hypergraph-set vertex coloring} $\theta:V(G)\rightarrow \mathcal{E}\big (\Lambda^2\big )$, such that each edge $uv\in E(G)$ holds a hypergraph $P$-property homomorphism $\theta(u)=\mathcal{E}_i\rightarrow _P\mathcal{E}_j=\theta(v)$ true.\qqed
\end{defn}

Motivated from Definition \ref{defn:definition-graph-homomorphism} and Definition \ref{defn:W-constraint-coloring-graph-homomorphism}, we have:

\begin{defn} \label{defn:111111}
$^*$ Let $G[\mathcal{E}]$ be a set-colored graph admitting a set-coloring $F:V(G)\rightarrow \mathcal{E}$ on a graph $G$ and a hypergraph $\mathcal{H}_{yper}=(\Lambda,\mathcal{E})$, and let $H[\mathcal{E}^*]$ be a set-colored graph admitting a set-coloring $F^*:V(H)\rightarrow \mathcal{E}^*$ on another graph $H$ and another hypergraph $\mathcal{H}_{yper}=(\Lambda_*,\mathcal{E}^*)$.

Since $F(uv)=F(u)[\bullet]F(v)$ if and only if $F^*(\varphi(u)\varphi(v))=F^*(\varphi(u))[\bullet]F^*(\varphi(v))$ under the coloring $\varphi: \Lambda \rightarrow \Lambda_*$ and an operation ``$[\bullet]$'', we get a \emph{set-colored graph homomorphism} $G[\mathcal{E}]\rightarrow H[\mathcal{E}^*]$, and a \emph{hypergraph homomorphism} $(\Lambda,\mathcal{E})\rightarrow (\Lambda_*,\mathcal{E}^*)$.\qqed
\end{defn}

\begin{example}\label{exa:8888888888}
According to Fig.\ref{fig:v-split-homomorphism} and Fig.\ref{fig:v-split-homomorphism-1}, we have four set-colored graph homomorphisms $S_k[\mathcal{E}_k]\rightarrow S_{k-1}[\mathcal{E}_{k-1}]$ for $k\in [1,4]$, where the set-colored graph $S_0=G$ shown in Fig.\ref{fig:1-example-hypergraph}(b) is a vertex-intersected graph of an $8$-uniform hypergraph shown in Fig.\ref{fig:1-example-hypergraph} (a). Conversely, each set-colored graph $S_i$ is a result of doing the vertex-splitting operation to $S_{i-1}$ with $i\in [1,4]$. More or less, we have shown the \emph{hyperedge-splitting operation} and the \emph{hyperedge-coinciding operation} of hypergraphs.\qqed
\end{example}

\begin{figure}[h]
\centering
\includegraphics[width=16.4cm]{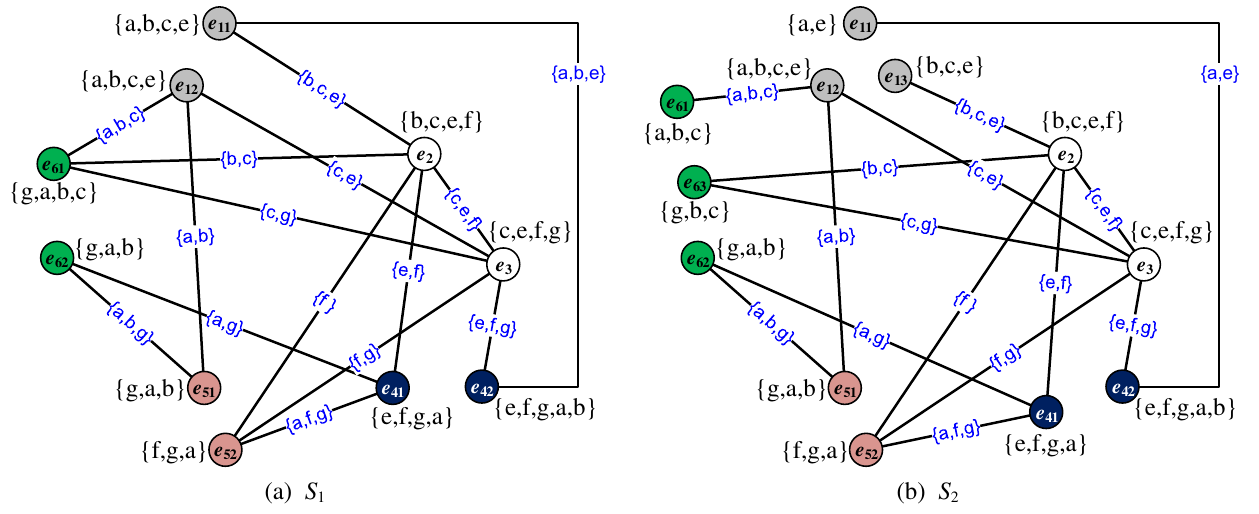}\\
\caption{\label{fig:v-split-homomorphism}{\small The first scheme for illustrating the hyperedge-splitting operation and the hyperedge-coinciding operation.}}
\end{figure}

\begin{defn} \label{defn:hyperedge-coinciding-splitting-homomorphism}
$^*$ \textbf{Hyperedge homomorphism.} Let $\mathcal{E}\big (\Lambda^2\big )$ be the hypergraph set defined on a finite set $\Lambda$. If there are two hyperedges $e_1,e_2\in \mathcal{E}^*\in \mathcal{E}\big (\Lambda^2\big )$ holding $|e_1\cup e_2|=|e_1|+|e_2|$, then we get a \emph{hyperedge homomorphism} $\mathcal{E}^*\rightarrow \mathcal{E}\in \mathcal{E}\big (\Lambda^2\big )$, where $e_1\cup e_2\in \mathcal{E}$, and $\mathcal{E}^*\setminus \{e_1,e_2\}=\mathcal{E}\setminus \{e_1\cup e_2\}$, this operation process is called \emph{hyperedge-coinciding operation} on hyperedge sets. Conversely, we split the hyperedge $e_1\cup e_2\in \mathcal{E}$ into two hyperedges $e_1$ and $e_2$ to obtain the hyperedge set $\mathcal{E}^*$, this operation process is called \emph{hyperedge-splitting operation} on hyperedge sets.\qqed
\end{defn}

\begin{rem}\label{rem:333333}
If $e_{\cap}=e_1\cap e_2\neq \emptyset $ in Definition \ref{defn:hyperedge-coinciding-splitting-homomorphism}, we hyperedge-split the hyperedge $e_1\cup e_2\in \mathcal{E}$ into two hyperedges $e_1=(e_1\cup e_2)\setminus \{e_2\}\cup e_{\cap}$ and $e_2=(e_1\cup e_2)\setminus \{e_1\}\cup e_{\cap}$. In other words, doing the hyperedge-splitting operation to a hypergraph set $\mathcal{E}$ can obtain two or more hypergraph sets $\mathcal{E}^*$ holding $\mathcal{E}^*\rightarrow \mathcal{E}$ true. For understanding the above hyperedge-splitting operation, refer to Fig.\ref{fig:v-split-homomorphism}, Fig.\ref{fig:v-split-homomorphism-1} and Fig.\ref{fig:v-split-homomorphism-2}, and we have the hyperedge homomorphisms:
$$
S_4\rightarrow _{hyperedge}S_3\rightarrow _{hyperedge}S_2\rightarrow _{hyperedge}S_1\rightarrow _{hyperedge}S
$$ by the hyperedge-splitting operation and the hyperedge-coinciding operation.\qqed
\end{rem}

\begin{figure}[h]
\centering
\includegraphics[width=16.4cm]{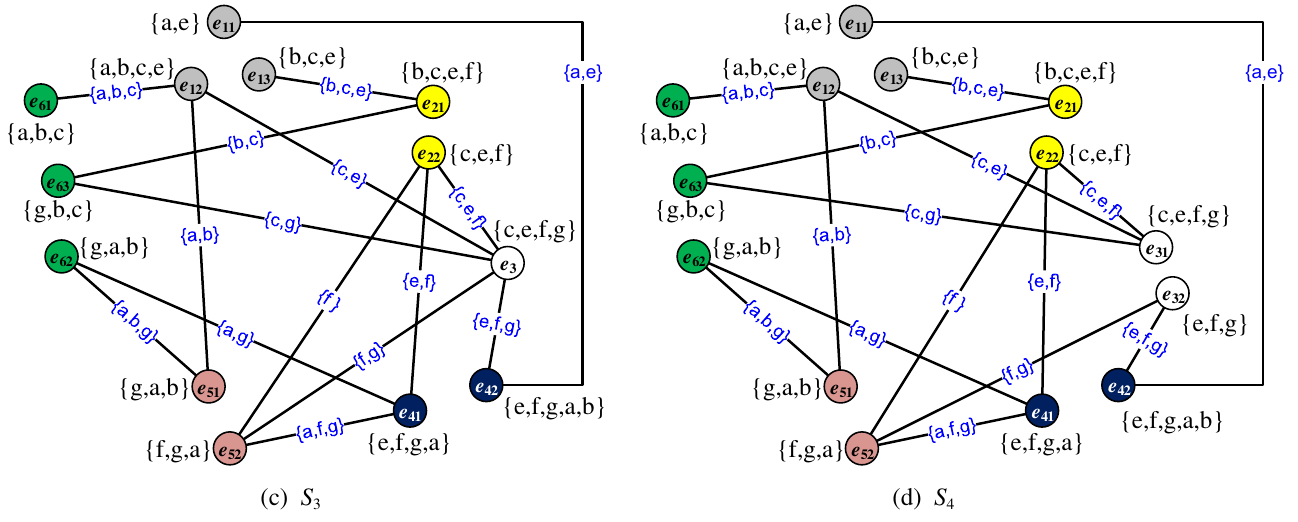}\\
\caption{\label{fig:v-split-homomorphism-1}{\small The second scheme for illustrating the hyperedge-splitting operation and the hyperedge-coinciding operation.}}
\end{figure}

\begin{thm}\label{thm:666666}
$^*$ Let $\mathcal{E}\big (\Lambda^2\big )$ be the hypergraph set defined on a finite set $\Lambda$. If each hyperedge set $\mathcal{E}\in \mathcal{E}\big (\Lambda^2\big )$ with some hyperedge $e$ holding $|e|\geq 2$ is hyperedge homomorphism to another hyperedge set $\mathcal{E}^*\in \mathcal{E}\big (\Lambda^2\big )$, then we have a hypergraph homomorphism $(\Lambda,\mathcal{E})\rightarrow (\Lambda,\mathcal{E}^*)$.
\end{thm}

\begin{figure}[h]
\centering
\includegraphics[width=11cm]{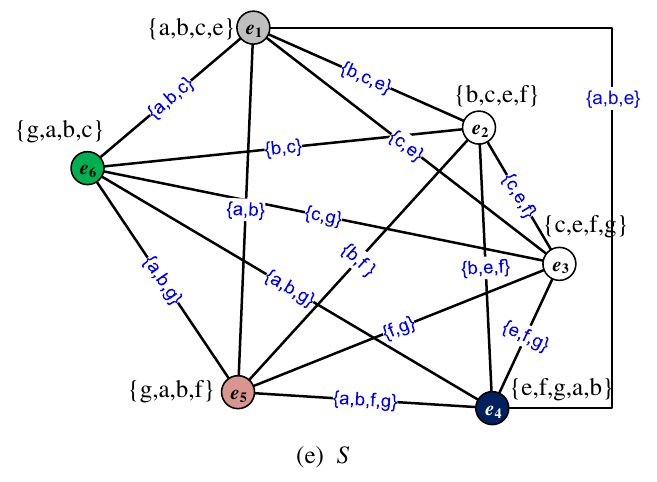}\\
\caption{\label{fig:v-split-homomorphism-2}{\small The third scheme for illustrating the hyperedge-splitting operation and the hyperedge-coinciding operation.}}
\end{figure}

\subsection{Strong hyperedge sets}

\begin{defn} \label{defn:strong-proper-hyperedge-sets}
$^*$ Let each set-set $\mathcal{E}_r(m,n_r)=\{e_{r,1},e_{r,2},\dots ,e_{r,n_r}\}$ with $r\in [1,A_m]$ be generated from a consecutive integer set $[1,m]$ such that each subset $e_{r,j}\in [1,m]^2$ for $j\in [1,n_r]$ and $r\in [1,A_m]$.

\begin{asparaenum}[\textbf{\textrm{Sthyset}}-1]
\item A \emph{strong hyperedge set} $\mathcal{E}_r(m,n_r)$ satisfies: Any pair of subsets $e_{r,i}$ and $e_{r,j}$ with $i\neq j$ holds

\qquad (1-i) $e_{r,i}\cap e_{r,j}\neq \emptyset$; and

\qquad (1-ii) $e_{r,i}\not \subset e_{r,j}$ and $e_{r,j}\not \subset e_{r,i}$.

\item A \emph{proper hyperedge set} $\mathcal{E}_r(m,n_r)$ with $r\in [1,A_m]$ satisfies:

\qquad (2-i) Any pair of subsets $e_{r,i}$ and $e_{r,j}$ holds $e_{r,i}\not \subset e_{r,j}$ and $e_{r,j}\not \subset e_{r,i}$ when $i\neq j$;

\qquad (2-ii) Each subset $e_{r,s}\in \mathcal{E}_r(m,n_r)$ corresponds another subset set $e_{r,t}\in \mathcal{E}_r(m,n_r)$ holding $e_{r,s}\cap e_{r,t}\neq\emptyset$.

\item A \emph{perfect hypermatching} of a hypergraph $\mathcal{H}_{yper}=(\Lambda,\mathcal{E})$ is a collection of hyperedges $M_1,M_2, \dots,M_m$ $\subseteq \mathcal{E}$, such that $M_i\cap M_j=\emptyset $ for $i\neq j$ and $\bigcup ^m_{i=1}M_i=\Lambda$.
\item \cite{Jianfang-Wang-Hypergraphs-2008} If a hyperedge set $\mathcal{E}=\bigcup^m_{j=1} \mathcal{E}_j$ holds $\mathcal{E}_i\cap \mathcal{E}_j=\emptyset$ for $i\neq j$, and each hyperedge $e\in \mathcal{E}$ belongs to one $\mathcal{E}_j$ and $e\not\in \bigcup^m_{k=1,k\neq j} \mathcal{E}_k$, then $\{\mathcal{E}_1,\mathcal{E}_2,\dots ,\mathcal{E}_m\}$ is called \emph{decomposition} of $\mathcal{E}$.
\item \cite{Jianfang-Wang-Hypergraphs-2008} A hyperedge set $\mathcal{E}$ is \emph{irreducible} if each hyperedge $e\in \mathcal{E}$ does not hold $e\subseteq e\,'$ for any hyperedge $e\,'\in \mathcal{E}$.\qqed
\end{asparaenum}
\end{defn}

\begin{problem}\label{qeu:444444}
\textbf{Compute} the exact value of each one of numbers $n_r$ and $A_m$ for proper hyperedge sets, or strong hyperedge sets introduced in Definition \ref{defn:strong-proper-hyperedge-sets}.
\end{problem}

\begin{example}\label{exa:8888888888}
\textbf{Build up} strong hyperedge sets $\mathcal{E}_r(m, n_r)=\{X_{r,1}, X_{r,2}, \dots , X_{r,n_r}\}$ from a consecutive integer set $[1, m]$, such that each set $X_{r,s}$ holds $|X_{r,s}|\geq r+1$ for $s\in [1, n_r]$ and $r\in [1,B_m]$, as well as each set-set $\mathcal{E}_r(m, n_r)$ holds the conditions of strong hyperedge set introduced in Definition \ref{defn:strong-proper-hyperedge-sets} true.

We construct particular sets $\mathcal{E}_t(m, t)=\{X_{t,1}, X_{t,2}, \dots , X_{t,t}\}$ with $t\in[2,m-1]$ and $|X_{t,s}|=t+1$ for $s\in [1, t]$ in the following procedure:

$X_{t,1}=\big \{[1,t]\cup \{a_i\}:~a_i\in [1,m]\setminus [1,t]\big \}\bigcup [2,m]$, and $|X_{r,1}|={m-t \choose 1}+1$;

$X_{t,2}=\big \{[2,t]\cup \{a_1,a_2\}:~a_i\in [1,m]\setminus [1,t]\big \}$, and $|X_{r,2}|={m-t \choose 1}+{m-t \choose 2}$;

$\cdots \cdots \cdots $

$X_{t,t}=\big \{\{t\}\cup \{a_1,a_2,\dots ,a_t\}:~a_i\in [1,m]\setminus [1,t]\big \}$, and $|X_{r,t}|={m-t \choose 1}+{m-t \choose 2}+\cdots +{m-t \choose t}$.

\vskip 0.4cm

\textbf{Case 1.} $m=4$. $S_1(4, 4)=\big \{\{$1, 2$\}$, $\{$1, 3$\}$, $\{$1, 4$\}$, $\{$2, 3, 4$\}\big \}$, since ${3 \choose 1}+1=4$.

$S_2(4, 3)=\big \{\{$1, 2, 3$\}$, $\{$1, 2, 4$\}$, $\{$2, 3, 4$\}\big \}$, since ${2 \choose 1}+{2 \choose 2}=3$.

\textbf{Case 2.} $m=5$. $S_1(5, 5)=\big \{\{$1, 2$\}$, $\{$1, 3$\}$, $\{$1, 4$\}$, $\{$1, 5$\}$, $\{$2, 3, 4, 5$\}\big \}$, since ${4 \choose 1}+1=5$.

$S_2(5, 7)=\big \{\{$1, 2, 3$\}$, $\{$1, 2, 4$\}$, $\{$1, 2, 5$\}$, $\{$2, 3, 4$\}$, $\{$2, 3, 5$\}$, $\{$2, 4, 5$\}$, $\{$3, 4, 5$\}\big \}$, since ${3 \choose 1}+{3 \choose 2}+1=7$.

$S_3(5, 3)=\big \{\{$1, 2, 3, 4$\}$, $\{$1, 2, 3, 5$\}$, $\{$2, 3, 4, 5$\}\big \}$, since ${2 \choose 1}+{2 \choose 2}=3$.

\textbf{Case 3.} $m=6$. $S_1(6, 6)=\big \{\{$1, 2$\}$, $\{$1, 3$\}$, $\{$1, 4$\}$, $\{$1, 5$\}$, $\{$1, 6$\}$, $\{$2, 3, 4, 5, 6$\}\big \}$, since ${5 \choose 1}+1=6$.

$S_2(6, 10)=\big \{\{$1, 2, 3$\}$, $\{$1, 2, 4$\}$, $\{$1, 2, 5$\}$, $\{$1, 2, 6$\}$, $\{$2, 3, 4$\}$, $\{$2, 3, 5$\}$, $\{$2, 3, 6$\}$, $\{$2, 4, 5$\}$, $\{$2, 4, 6$\}$, $\{$2, 5, 6$\}\big \}$, since ${4 \choose 1}+{4 \choose 2}=10$.

$S_3(6, 7)=\big \{\{$1, 2, 3, 4$\}$, $\{$1, 2, 3, 5$\}$, $\{$1, 2, 3, 6$\}$, $\{$2, 3, 4, 5$\}$, $\{$2, 3, 4, 6$\}$, $\{$2, 3, 5, 6$\}$, $\{$3, 4, 5, 6$\}\big \}$, since ${3 \choose 1}+{3 \choose 2}+{3 \choose 3}=7$.

\textbf{Case 4.} $m=7$. $S_1(7, 7)=\big \{\{$1, 2$\}$, $\{$1, 3$\}$, $\{$1, 4$\}$, $\{$1, 5$\}$, $\{$1, 6$\}$, $\{$1, 7$\}$, $\{$2, 3, 4, 5, 6, 7$\}\big \}$, since ${6 \choose 1}+1=7$.

$S_2(7, 15)=\big \{\{$1, 2, 3$\}$, $\{$1, 2, 4$\}$, $\{$1, 2, 5$\}$, $\{$1, 2, 6$\}$, $\{$1, 2, 7$\}$, $\{$2, 3, 4$\}$, $\{$2, 3, 5$\}$, $\{$2, 3, 6$\}$, $\{$2, 3, 7$\}$, $\{$2, 4, 5$\}$, $\{$2, 4, 6$\}$, $\{$2, 4, 7$\}$, $\{$2, 5, 6$\}$, $\{$2, 5, 7$\}$, $\{$2, 6, 7$\}\big \}$, since ${5 \choose 1}+{5 \choose 2}=15$.

$S_3(7, 14)=\big \{\{$1, 2, 3, 4$\}$, $\{$1, 2, 3, 5$\}$, $\{$1, 2, 3, 6$\}$, $\{$1, 2, 3, 7$\}$, $\{$2, 3, 4, 5$\}$, $\{$2, 3, 4, 6$\}$, $\{$2, 3, 4, 7$\}$, $\{$2, 3, 5, 6$\}$, $\{$2, 3, 5, 7$\}$, $\{$2, 3, 6, 7$\}$, $\{$3, 4, 5, 6$\}$, $\{$3, 4, 5, 7$\}$, $\{$3, 4, 6, 7$\}$, $\{$3, 5, 6, 7$\}\big \}$, since ${4 \choose 1}+{4 \choose 2}+{4 \choose 3}=14$.

$S_4(7, 7)=\big \{\{$1, 2, 3, 4, 5$\}$, $\{$1, 2, 3, 4, 6$\}$, $\{$1, 2, 3, 4, 7$\}$, $\{$2, 3, 4, 5, 6$\}$, $\{$2, 3, 4, 5, 7$\}$, $\{$2, 3, 4, 6, 7$\}$, $\{$3, 4, 5, 6, 7$\}\big \}$, since ${3 \choose 1}+{3 \choose 2}+{3 \choose 3}=7$.

\textbf{Case 5.} $m=8$. $S_1(8, 8)=\big \{\{$1, 2$\}$, $\{$1, 3$\}$, $\{$1, 4$\}$, $\{$1, 5$\}$, $\{$1, 6$\}$, $\{$1, 7$\}$, $\{$1, 8$\}$, $\{$2, 3, 4, 5, 6, 7, 8$\}\big \}$, since ${7 \choose 1}+1=8$.

$S_2(8, 21)=\big \{\{$1, 2, 3$\}$, $\{$1, 2, 4$\}$, $\{$1, 2, 5$\}$, $\{$1, 2, 6$\}$, $\{$1, 2, 7$\}$, $\{$1, 2, 8$\}$, $\{$2, 3, 4$\}$, $\{$2, 3, 5$\}$, $\{$2, 3, 6$\}$, $\{$2, 3, 7$\}$, $\{$2, 3, 8$\}$, $\{$2, 4, 5$\}$, $\{$2, 4, 6$\}$, $\{$2, 4, 7$\}$, $\{$2, 4, 8$\}$, $\{$2, 5, 6$\}$, $\{$2, 5, 7$\}$, $\{$2, 5, 8$\}$, $\{$2, 6, 7$\}$, $\{$2, 6, 8$\}$, $\{$2, 7, 8$\}\big \}$, since ${6 \choose 1}+{6 \choose 2}=21$.

$S_3(8, 25)=\big \{\{$1, 2, 3, 4$\}$, $\{$1, 2, 3, 5$\}$, $\{$1, 2, 3, 6$\}$, $\{$1, 2, 3, 7$\}$, $\{$1, 2, 3, 8$\}$, $\{$2, 3, 4, 5$\}$, $\{$2, 3, 4, 6$\}$, $\{$2, 3, 4, 7$\}$, $\{$2, 3, 4, 8$\}$, $\{$2, 3, 5, 6$\}$, $\{$2, 3, 5, 7$\}$, $\{$2, 3, 5, 8$\}$, $\{$2, 3, 6, 7$\}$, $\{$2, 3, 6, 8$\}$, $\{$2, 3, 7, 8$\}$, $\{$3, 4, 5, 6$\}$, $\{$3, 4, 5, 7$\}$, $\{$3, 4, 5, 8$\}$, $\{$3, 4, 6, 7$\}$, $\{$3, 4, 6, 8$\}$, $\{$3, 4, 7, 8$\}$, $\{$3, 5, 6, 7$\}$, $\{$3, 5, 6, 8$\}$, $\{$3, 5, 7, 8$\}$, $\{$3, 6, 7, 8$\}\big \}$, since ${5 \choose 1}+{5 \choose 2}+{5 \choose 3}=25$.

$S_4(8, 15)=\big \{\{$1, 2, 3, 4, 5$\}$, $\{$1, 2, 3, 4, 6$\}$, $\{$1, 2, 3, 4, 7$\}$, $\{$1, 2, 3, 4, 8$\}$, $\{$2, 3, 4, 5, 6$\}$, $\{$2, 3, 4, 5, 7$\}$, $\{$2, 3, 4, 5, 8$\}$, $\{$2, 3, 4, 6, 7$\}$, $\{$2, 3, 4, 6, 8 $\}$, $\{$2, 3, 4, 7, 8$\}$, $\{$3, 4, 5, 6, 7$\}$, $\{$3, 4, 5, 6, 8$\}$, $\{$3, 4, 5, 7, 8$\}$, $\{$3, 4, 6, 7, 8$\}$, $\{$4, 5, 6, 7, 8$\}\big \}$, since ${4 \choose 1}+{4 \choose 2}+{4 \choose 3}+{4 \choose 4}=15$.

$S_5(8, 7)=\big \{\{$1, 2, 3, 4, 5, 6$\}$, $\{$1, 2, 3, 4, 5, 7$\}$, $\{$1, 2, 3, 4, 5, 8$\}$, $\{$2, 3, 4, 5, 6, 7$\}$, $\{$2, 3, 4, 5, 6, 8$\}$, $\{$2, 3, 4, 5, 7, 8$\}$, $\{$3, 4, 5, 6, 7, 8$\}\big \}$, since ${3 \choose 1}+{3 \choose 2}+{3 \choose 3}=7$.

\textbf{Case 6.} $m=9$. $S_1(9, 9)=\big \{\{$1, 2$\}$, $\{$1, 3$\}$, $\{$1, 4$\}$, $\{$1, 5$\}$, $\{$1, 6$\}$, $\{$1, 7$\}$, $\{$1, 8$\}$, $\{$1, 9$\}$, $\{$2, 3, 4, 5, 6, 7, 8, 9$\}\big \}$, since ${8 \choose 1}+1=9$.

$S_2(9, 28)=\big \{\{$1, 2, 3$\}$, $\{$1, 2, 4$\}$, $\{$1, 2, 5$\}$, $\{$1, 2, 6$\}$, $\{$1, 2, 7$\}$, $\{$1, 2, 8$\}$, $\{$1, 2, 9$\}$, $\{$2, 3, 4$\}$, $\{$2, 3, 5$\}$, $\{$2, 3, 6$\}$, $\{$2, 3, 7$\}$, $\{$2, 3, 8$\}$, $\{$2, 3, 9$\}$, $\{$2, 4, 5$\}$, $\{$2, 4, 6$\}$, $\{$2, 4, 7$\}$, $\{$2, 4, 8$\}$, $\{$2, 4, 9$\}$, $\{$2, 5, 6$\}$, $\{$2, 5, 7$\}$, $\{$2, 5, 8$\}$, $\{$2, 5, 9$\}$, $\{$2, 6, 7$\}$, $\{$2, 6, 8$\}$, $\{$2, 6, 9$\}$, $\{$2, 7, 8$\}$, $\{$2, 7, 9$\}$, $\{$2, 8, 9$\}\big \}$, since ${7 \choose 1}+{7 \choose 2}=28$.

$S_3(9, 41)=\big \{\{$1, 2, 3, 4$\}$, $\{$1, 2, 3, 5$\}$, $\{$1, 2, 3, 6$\}$, $\{$1, 2, 3, 7$\}$, $\{$1, 2, 3, 8$\}$, $\{$1, 2, 3, 9$\}$, $\{$2, 3, 4, 5$\}$, $\{$2, 3, 4, 6$\}$, $\{$2, 3, 4, 7$\}$, $\{$2, 3, 4, 8$\}$, $\{$2, 3, 4, 9$\}$, $\{$2, 3, 5, 6$\}$, $\{$2, 3, 5, 7$\}$, $\{$2, 3, 5, 8$\}$, $\{$2, 3, 5, 9$\}$, $\{$2, 3, 6, 7$\}$, $\{$2, 3, 6, 8$\}$, $\{$2, 3, 6, 9$\}$, $\{$2, 3, 7, 8$\}$, $\{$2, 3, 7, 9$\}$, $\{$2, 3, 8, 9$\}$, $\{$3, 4, 5, 6$\}$, $\{$3, 4, 5, 7$\}$, $\{$3, 4, 5, 8$\}$, $\{$3, 4, 5, 9$\}$, $\{$3, 4, 6, 7$\}$, $\{$3, 4, 6, 8$\}$, $\{$3, 4, 6, 9$\}$, $\{$3, 4, 7, 8$\}$, $\{$3, 4, 7, 9$\}$, $\{$3, 4, 8, 9$\}$, $\{$3, 5, 6, 7$\}$, $\{$3, 5, 6, 8$\}$, $\{$3, 5, 6, 9$\}$, $\{$3, 5, 7, 8$\}$, $\{$3, 5, 7, 9$\}$, $\{$3, 5, 8, 9$\}$, $\{$3, 6, 7, 8$\}$, $\{$3, 6, 7, 9$\}$, $\{$3, 6, 8, 9$\}$, $\{$3, 7, 8, 9$\}\big \}$, since ${6 \choose 1}+{6 \choose 2}+{6 \choose 3}=41$.

$S_4(9, 30)=\big \{\{$1, 2, 3, 4, 5$\}$, $\{$1, 2, 3, 4, 6$\}$, $\{$1, 2, 3, 4, 7$\}$, $\{$1, 2, 3, 4, 8$\}$, $\{$1, 2, 3, 4, 9$\}$, $\{$2, 3, 4, 5, 6$\}$, $\{$2, 3, 4, 5, 7$\}$, $\{$2, 3, 4, 5, 8$\}$, $\{$2, 3, 4, 5, 9$\}$, $\{$2, 3, 4, 6, 7$\}$, $\{$2, 3, 4, 6, 8$\}$, $\{$2, 3, 4, 6, 9$\}$, $\{$2, 3, 4, 7, 8$\}$, $\{$2, 3, 4, 7, 9$\}$, $\{$2, 3, 4, 8, 9$\}$, $\{$3, 4, 5, 6, 7$\}$, $\{$3, 4, 5, 6, 8$\}$, $\{$3, 4, 5, 6, 9$\}$, $\{$3, 4, 5, 7, 8$\}$, $\{$3, 4, 5, 7, 9$\}$, $\{$3, 4, 5, 8, 9$\}$, $\{$3, 4, 6, 7, 8$\}$, $\{$3, 4, 6, 7, 9$\}$, $\{$3, 4, 6, 8, 9$\}$, $\{$3, 4, 7, 8, 9$\}$, $\{$4, 5, 6, 7, 8$\}$, $\{$4, 5, 6, 7, 9$\}$, $\{$4, 5, 6, 8, 9$\}$, $\{$4, 5, 7, 8, 9$\}$, $\{$4, 6, 7, 8, 9$\}\big \}$, since ${5 \choose 1}+{5 \choose 2}+{5 \choose 3}+{5 \choose 4}=30$.

$S_5(9, 15)=\big \{\{$1, 2, 3, 4, 5, 6$\}$, $\{$1, 2, 3, 4, 5, 7$\}$, $\{$1, 2, 3, 4, 5, 8$\}$, $\{$1, 2, 3, 4, 5, 9$\}$, $\{$2, 3, 4, 5, 6, 7$\}$, $\{$2, 3, 4, 5, 6, 8$\}$, $\{$2, 3, 4, 5, 6, 9$\}$, $\{$2, 3, 4, 5, 7, 8$\}$, $\{$2, 3, 4, 5, 7, 9$\}$, $\{$2, 3, 4, 5, 8, 9$\}$, $\{$3, 4, 5, 6, 7, 8$\}$, $\{$3, 4, 5, 6, 7, 9$\}$, $\{$3, 4, 5, 6, 8, 9$\}$, $\{$3, 4, 5, 7, 8, 9$\}$, $\{$4, 5, 6, 7, 8, 9$\}\big \}$, since ${4 \choose 1}+{4 \choose 2}+{4 \choose 3}+{4 \choose 4}=15$.

$S_6(9, 7)=\big \{\{$1, 2, 3, 4, 5, 6, 7$\}$, $\{$1, 2, 3, 4, 5, 6, 8$\}$, $\{$1, 2, 3, 4, 5, 6, 9$\}$, $\{$2, 3, 4, 5, 6, 7, 8$\}$, $\{$2, 3, 4, 5, 6, 7, 9$\}$, $\{$2, 3, 4, 5, 6, 8, 9$\}$, $\{$3, 4, 5, 6, 7, 8, 9$\}\big \}$, since ${3 \choose 1}+{3 \choose 2}+{3 \choose 3}=7$.

$S_7(9, 3)=\big \{\{$1, 2, 3, 4, 5, 6, 7, 8$\}$, $\{$1, 2, 3, 4, 5, 6, 7, 9$\}$, $\{$2, 3, 4, 5, 6, 7, 8, 9$\}\big \}$, since ${2 \choose 1}+{2 \choose 2}=3$.

In the above sets $X_{r,s}\in \mathcal{E}_r(m, n_r)$ with $r\in [2, 7]$ and $m\in [4, 9]$, we can see $|X_{r,s}|=r+1$ for $s\in [1, n_r]$ and $r\in [2, 7]$, except $S_1(k, k)$ for $k\in [4, 9]$.\qqed
\end{example}

\begin{problem}\label{qeu:strong-hyperedge-set-condition-all}
Notice that each subset $X\subseteq \mathcal{E}_r(m,n_r)$ satisfies the strong hyperedge set condition defined in Definition \ref{defn:strong-proper-hyperedge-sets}. \textbf{Find} all hyperedge sets of the power set $[1,m]^2$, such that each hyperedge set $\mathcal{E}_k$ has at least two subsets of the power set $[1,m]^2$ and $\bigcup_{e\in \mathcal{E}_k}e=[1,m]$, and holds the strong hyperedge set condition: For any pair of two subsets $e_{i}$ and $e_{j}$, we have $e_{i}\cap e_{j}\neq \emptyset$, $e_{i}\not \subset e_{j}$ and $e_{j}\not \subset e_{i}$ when $i\neq j$.
\end{problem}

\begin{example}\label{exa:8888888888}
For a given set $[1, 7]$, as an example for Problem \ref{qeu:strong-hyperedge-set-condition-all}, we take a proper subset $S^*_3=\{$1, 4, 7$\}\subset [1, 7]$, so the remainder set is $\{$2, 3, 5, 6$\}=[1, 7]\setminus S^*_3$. Then we have a hyperedge set $S^*_3(7, 14)=\big \{\{$1, 4, 7, 2$\}$, $\{$1, 4, 7, 3$\}$, $\{$1, 4, 7, 5$\}$, $\{$1, 4, 7, 6$\}$; $\{$4, 7, 2, 3$\}$, $\{$4, 7, 2, 5$\}$, $\{$4, 7, 2, 6$\}$, $\{$4, 7, 3, 5$\}$, $\{$4, 7, 3, 6$\}$, $\{$4, 7, 5, 6$\}$; $\{$7, 2, 3, 5$\}$, $\{$7, 2, 3, 6$\}$, $\{$7, 2, 5, 6$\}$, $\{$7, 3, 5, 6$\}\big \}$, since ${4 \choose 1}+{4 \choose 2}+{4 \choose 3}=14$. Clearly, the hyperedge set $S^*_3(7, 14)$ holds the condition of strong hyperedge set defined in Definition \ref{defn:strong-proper-hyperedge-sets}.\qqed
\end{example}

\begin{conj}\label{conj:strong-hyperedge-set-proper-edge-set-coloring}
$^*$ For each connected graph $G$, there is a hyperedge set $\mathcal{E}=\{e_1,e_2,\dots ,e_m\}$ satisfying $\Lambda =\bigcup_{e\in \mathcal{E}}e$ with $\frac{\Delta(G)+1}{2}\leq |\Lambda|\leq\Delta(G)+1$ and $e_1\cap e_2\cap \cdots \cap e_m=\emptyset$, such that $G$ admits a strong hyperedge-set proper edge set-coloring $F: E(G)\rightarrow \mathcal{E}$ with $|F(uv)|\geq 2$ for each edge $uv\in E(G)$ (Ref. Definition \ref{defn:tradition-vs-set-colorings} and Definition \ref{defn:strong-proper-hyperedge-sets}).
\end{conj}

\begin{conj}\label{conj:strong-hyperedge-set-proper-total-set-coloring}
$^*$ For each connected graph $H$, there is a hyperedge set $\mathcal{E}_T=\{e_1,e_2,\dots ,e_m\}$ holding $\Lambda_T =\bigcup_{e\in \mathcal{E}_T}e$ with the cardinality $\frac{\Delta(H)+1}{2}\leq |\Lambda_T|\leq \Delta(H)+2$ and $e_1\cap e_2\cap \cdots \cap e_m=\emptyset$, such that $H$ admits a strong hyperedge-set proper total set-coloring $F: V(H)\cup E(H)\rightarrow \mathcal{E}$ with $|F(w)|\geq 2$ for each element $w\in V(H)\cup E(H)$ (Ref. Definition \ref{defn:tradition-vs-set-colorings} and Definition \ref{defn:strong-proper-hyperedge-sets}).
\end{conj}

Motivated from the harmonious coloring, we present a new coloring as follows:

\begin{defn} \label{defn:local-proper-harmonious-coloring}
$^*$ A graph $G$ admits a \emph{proper local-harmonious coloring} $f:V(G)\rightarrow [1,k]$ if

(i) $f(x)\neq f(y)$ for each edge $xy\in E(G)$;

(ii) each edge $xy$ is colored by an induced color set $f(xy)=\{f(x),f(y)\}$; and

(iii) $f(xy)\neq f(uv)$ for any pair of adjacent edges $uv$ and $uw$ with $v,w\in N_{ei}(u)$.\\
We call the extremal number
\begin{equation}\label{eqa:555555}
\chi_{lh}(G)=\min_f \big \{k: ~f:V(G)\rightarrow [1,k]\textrm{ is a proper local-harmonious coloring of }G \big \}
\end{equation} the \emph{local-harmonious chromatic number} of the graph $G$.\qqed
\end{defn}

\begin{thm}\label{thm:local-harmonious-coloring-hyperedge}
$^*$ Let each $\mathcal{E}^{|2|}_i$ with $i\in [1,m]$ be a hyperedge set holding $|e|=2$ for each hyperedge $e\in \mathcal{E}^{|2|}_i$, and $e\neq e\,'$ for $e,e\,'\in \mathcal{E}^{|2|}_i$, and $\Lambda=\bigcup _{e\in \mathcal{E}^{|2|}_i}e$. If a graph $G$ admits a proper edge set-coloring $F:E(G)\rightarrow \mathcal{E}^{|2|}_i$ for each hyperedge set $\mathcal{E}^{|2|}_i$
 with $i\in [1,m]$, then the graph $G$ admits a \emph{proper local-harmonious coloring} $f:V(G)\rightarrow [1,k]$ defined in Definition \ref{defn:local-proper-harmonious-coloring} when as $\Lambda=[1,k]$.
\end{thm}

\subsection{Every-zero hypergraph groups based on consecutive integer sets}

\begin{defn}\label{defn:hypergraph-group-definition}
$^*$ \textbf{Every-zero hypergraph group.} Suppose that the vertex set $\Lambda_{[1,N]}=[1,N]$ is a consecutive integer set. For getting a structural representation of the hypergraph set $\mathcal{E}\big (\Lambda_{[1,N]}^2\big )$ of all hyperedge sets defined on a consecutive integer set $\Lambda_{[1,N]}=[1,N]$, we take a particular hyperedge set $\mathcal{E}_1=\{e_{1,1},e_{1,2},\dots $, $e_{1,a_1}\}\in \mathcal{E}\big (\Lambda_{[1,N]}^2\big )$, where $a_1\geq 2$ and $e_{1,s}=\{x_{1,s,1},x_{1,s,2},\dots ,x_{1,s,m_s}\}$ holding $x_{1,s,t}\in [1,N]$ with $t\in [1,m_s]$ and $s\in [1,a_1]$.

We use this especial hyperedge set $\mathcal{E}_1$ to make hyperedge sets $\mathcal{E}_i=\{e_{i,1},e_{i,2},\dots ,e_{i,a_i}\}$ with $i\in [1,N]$ and $e_{i,s}=\{x_{i,s,1},x_{i,s,2},\dots ,x_{i,s,m_s}\}$ holding $x_{i,s,t}=x_{1,s,t}+(i-1)~(\bmod~N)\in [1,N]$ for $t\in [1,m_s]$ and $s\in [1,a_1]$. The set of hyperedge sets $\mathcal{E}_i$ generated by $\mathcal{E}_1$ is denoted as $G(\mathcal{E}_1)$.

We get an \emph{every-zero hypergraph group} $\big \{G(\mathcal{E}_1);[+][-]\big \}$ since any pair of hyperedge sets $\mathcal{E}_i$ and $\mathcal{E}_j$ of the set $G(\mathcal{E}_1)$ holds the finite module Abelian additive operation $\mathcal{E}_i[+_k]\mathcal{E}_j$ true, where the finite module Abelian additive operation
\begin{equation}\label{eqa:every-zero-hypergraph-group11}
\mathcal{E}_i[+_k]\mathcal{E}_j:=\mathcal{E}_i[+]\mathcal{E}_j[-]\mathcal{E}_k=\mathcal{E}_\lambda
\end{equation} for any preappointed \emph{zero} $\mathcal{E}_k\in G(\mathcal{E}_1)$ is defined as follows:
\begin{equation}\label{eqa:hypergraph-group-operation}
{
\begin{split}
\big (x_{i,s,t}+x_{j,s,t}\big )-x_{k,s,t}&=x_{1,s,t}+(i-1)+x_{1,s,t}+(j-1)-\big [x_{1,s,t}+(k-1)\big ]~(\bmod~N)\\
&=x_{1,s,t}+(i+j-k-1)~(\bmod~N)\\
&=x_{\lambda,s,t}\in [1,N]
\end{split}}
\end{equation} where $\lambda=i+j-k-1~(\bmod~N)$, $x_{i,s,t}\in e_{i,s}\in \mathcal{E}_i$, $x_{j,s,t}\in e_{j,s}\in \mathcal{E}_j$ and $x_{k,s,t}\in e_{k,s}\in \mathcal{E}_k$ with $t\in [1,m_s]$ and $s\in [1,a_1]$.\qqed
\end{defn}

The every-zero hypergraph group $\big \{G(\mathcal{E}_1);[+][-]\big \}$ defined in Definition \ref{defn:hypergraph-group-definition} has the following properties:

(1) \textbf{Zero}. Each hyperedge set $\mathcal{E}_k\in G(\mathcal{E}_1)$ can be as \emph{zero}.

(2) \textbf{Inverse}. Since $\mathcal{E}_i[+]\mathcal{E}_{2k-i}[-]\mathcal{E}_k=\mathcal{E}_k$, then each hyperedge set $\mathcal{E}_i\in G(\mathcal{E}_1)$ has its own \emph{inverse} $\mathcal{E}_{2k-i}\in G(\mathcal{E}_1)$.

(3) \textbf{Uniqueness and Closureness}. If $\mathcal{E}_i[+]\mathcal{E}_j[-]\mathcal{E}_k=\mathcal{E}_\lambda$ and $\mathcal{E}_i[+]\mathcal{E}_j[-]\mathcal{E}_k=\mathcal{E}_\mu$, we have the indices $\lambda=i+j-k~(\bmod~N)$ and $\mu=i+j-k~(\bmod~N)$, immediately, $\lambda=\mu$. Closureness follows Eq.(\ref{eqa:hypergraph-group-operation}).

(4) \textbf{Associative law}. $\mathcal{E}_i[+_k]\mathcal{E}_j=\mathcal{E}_j[+_k]\mathcal{E}_i$.

(5) \textbf{Commutative law}. $\big (\mathcal{E}_i[+_k]\mathcal{E}_j\big )[+_k]\mathcal{E}_l=\mathcal{E}_i[+_k]\big (\mathcal{E}_j[+_k]\mathcal{E}_l\big )$.

\vskip 0.4cm

In general, by the finite module Abelian additive operation, each hyperedge set $\mathcal{E}\in \mathcal{E}\big (\Lambda_{[1,N]}^2\big )$ forms an every-zero hypergraph group $\big \{G(\mathcal{E});[+][-]\big \}$, then we get a structural representation of the hypergraph set $\mathcal{E}\big (\Lambda_{[1,N]}^2\big )$ as follows:

\begin{thm}\label{thm:kinds-every-zero-number-based-string-groups}
$^*$ For a consecutive integer set $\Lambda_{[1,N]}=[1,N]$, the \emph{hypergraph set} $\mathcal{E}\big (\Lambda_{[1,N]}^2\big )$ contains all hyperedge sets defined on the vertex set $\Lambda_{[1,N]}=[1,N]$ (Ref. Definition \ref{defn:hypergraph-basic-definition}).

(i) The hypergraph set $\mathcal{E}\big (\Lambda_{[1,N]}^2\big )$ can be classified into several kinds of every-zero number-based set-groups, such that each hyperedge set $\mathcal{E}\in \mathcal{E}\big (\Lambda_{[1,N]}^2\big )$ is in an every-zero number-based set-group of the hypergraph set $\mathcal{E}\big (\Lambda_{[1,N]}^2\big )$.

(ii) Suppose that a hyperedge set $\mathcal{E}\in \mathcal{E}\big (\Lambda_{[1,N]}^2\big )$ holds that there are hyperedges $e_i,e_j\in \mathcal{E}$, such that $e_i\cap e_j\neq \emptyset$, then there are at least $N-1$ hyperedge sets of the hypergraph set $\mathcal{E}\big (\Lambda_{[1,N]}^2\big )$ satisfy the property $P$ if the hyperedge set $\mathcal{E}$ has a property $P$.
\end{thm}

We, by Theorem \ref{thm:kinds-every-zero-number-based-string-groups}, have the following result for real applications:

\begin{cor}\label{cor:things-hyperedge-set group}
$^*$ \textbf{Every-zero hyperedge-set group.} For a finite set $\Lambda=\{a_1,a_2,\dots ,a_n\}$ with each $a_i$ is a general data, we define the finite module Abelian additive operation on the set $\Lambda$ as
\begin{equation}\label{eqa:Abelian-additive-operation-general-data}
a_i[+_k]a_j:=a_i[+]a_j[-]a_k=a_{\lambda}\in \Lambda
\end{equation} for a preappointed \emph{zero} $a_k$, if $\lambda=i+j-k~(\bmod~n)$. In other words, the finite module Abelian additive operation is only for the index set $[1,n]=\{1,2,\dots ,n\}$ of elements of the set $\Lambda$. Naturally, we obtain an \emph{every-zero hyperedge-set group} $\big \{G(\mathcal{E});[+][-]\big \}$ for a hyperedge set $\mathcal{E}\in \mathcal{E}\big (\Lambda^2\big )$, where the hyperedge set $G(\mathcal{E})=\{\mathcal{E}_1,\mathcal{E}_2,\dots ,\mathcal{E}_n\}$ holding Eq.(\ref{eqa:Abelian-additive-operation-general-data}) on the elements of subset $\{a_{i,j,1},a_{i,j,2},\dots ,a_{i,j,b(i,j)}\}= e_{i,j}\in \mathcal{E}_i$ with $j\in [1,b]$ and $i\in [1,n]$, where $a_{i,j,s}\in \Lambda$ with $s\in [1,b(i,j)]$, and each cardinality $|\mathcal{E}_i|=b$ for $i\in [1,n]$.
\end{cor}

\subsubsection{Graph colorings based on hypergraph sets and hypergraph groups}

\begin{defn} \label{defn:general-defi-hypergraph-groups}
$^*$ Let $G$ be a graph, and let $\mathcal{E}\big (\Lambda^2\big )$ be a hypergraph set defined in Definition \ref{defn:hypergraph-basic-definition}, and let $\big \{G(\mathcal{E}_1);[+][-]\big \}$ be an every-zero hypergraph group defined in Definition \ref{defn:hypergraph-group-definition}. We define:

(i) \textbf{The \textbf{$[\bullet]$-graphs of hypergraphs}.} The graph $G$ admits a \emph{hyperedge total coloring} $f:V(G)\cup E(G)\rightarrow \mathcal{E}\in \mathcal{E}\big (\Lambda^2\big )$, such that each edge $uv\in E(G)$ satisfies that $f(u)=e_i$, $f(v)=e_j$ and $f(uv)\supseteq e_i[\bullet]e_j=f(u)[\bullet]f(v)$ under an operation ``$[\bullet]$'' based on hyperedge sets. So, we call the colored graph $G$ as the \emph{$[\bullet]$-graph} of a hypergraph $\mathcal{H}_{yper}=(\Lambda,\mathcal{E})$ (Ref. the vertex-intersected graphs of hypergraphs).

(ii) \textbf{\textbf{Hypergraph-group coloring}.} The graph $G$ admits a \emph{total hypergraph-group coloring}
$$\eta:V(G)\cup E(G)\rightarrow \big \{G(\mathcal{E});[+][-]\big \}
$$ with the hyperedge set $G(\mathcal{E})=\{\mathcal{E}_1,\mathcal{E}_2,\dots ,\mathcal{E}_n\}$, such that each edge $uv\in E(G)$ satisfies that hypergraph $\eta(u)=\mathcal{E}_i$, $\eta(v)=\mathcal{E}_j$ and
\begin{equation}\label{eqa:555555}
\eta(uv)=\mathcal{E}_\lambda=\mathcal{E}_i[+]\mathcal{E}_j[-]\mathcal{E}_k=\eta(u)[+]\eta(v)[-]\mathcal{E}_k
\end{equation} with $\lambda=i+j-k~(\bmod~N)$ for any preappointed \emph{zero} $\mathcal{E}_k\in \big \{G(\mathcal{E});[+][-]\big \}$.\qqed
\end{defn}

\subsubsection{Networks overall encrypted by hypergraph groups}

\textbf{Network overall topological encryption algorithm (NOTE-algorithm).}

\textbf{Input:} A dynamic network $N(t)$ with $t\in [\alpha, \beta]$, and a thing set $S_{thing}$ of things $s_{1},s_{2},\dots ,s_{N}$ with $N\geq 2$.

\textbf{Output:} A dynamic network $N_{thing}(t)$ with $t\in [\alpha, \beta]$ encrypted by the thing set $S_{thing}$.

\textbf{Initialization.} A set $S_{thing}$ of things $s_{1},s_{2},\dots ,s_{N}$ with $N\geq 2$ is substituted by a consecutive integer set $\Lambda_{[1,N]}=[1,N]$ for the overall topological encryption of using hypergraphs.

\textbf{Hypergraph-group encryption.} We, for the dynamic network $N(t)$ at time $t\in [\alpha, \beta]$, define a \emph{total hypergraph-group coloring} $$F_t:V(N(t))\cup E(N(t))\rightarrow \{G(\mathcal{E},t);[+][-]\}
$$ for a hyperedge set $\mathcal{E}$ of the hypergraph set $\mathcal{E}\big (\Lambda_{[1,N]}^2\big )$ based on the vertex set $\Lambda_{[1,N]}=[1,N]$, where $\{G(\mathcal{E},t);[+][-]\}=\{\mathcal{E}_1(t),\mathcal{E}_2(t),\dots ,\mathcal{E}_N(t)\}$ with $\mathcal{E}_1(t)=\mathcal{E}(t)$, such that each edge $uv\in E(N(t))$ holds $F_t(u)\neq F_t(v)$ and
$$F_t(uv)=F_t(u)[+]F_t(v)[-]\mathcal{E}_{k}(t)=\mathcal{E}_{i}(t)[+]\mathcal{E}_{j}(t)[-]\mathcal{E}_{k}(t)=\mathcal{E}_{\lambda}(t)
$$ with $\lambda=i+j-k~(\bmod~N)$ for any preappointed \emph{zero} $\mathcal{E}_{k}(t)\in \{G(\mathcal{E},t);[+][-]\}$ at time $t\in [\alpha, \beta]$. The set-colored dynamic network is denoted as $N_{color}(t)$.

\textbf{Thing encryption.} The substitution of elements of the set-color set $Q=F_t(V(N_{color}(t))\cup E(N_{color}(t)))$: We by each thing $s_{i}\in S_{thing}=\{s_{1},s_{2},\dots ,s_{N}\}$ replace a hyperedge set $\mathcal{E}_{i(t)}\in Q$, and then we get a dynamic network colored by the thing set $S_{thing}$ with $N\geq 2$, we write this dynamic network overall encrypted by the thing set $S_{thing}$ as $N_{thing}(t)$ (Ref. Corollary \ref{cor:things-hyperedge-set group}).

\begin{rem}\label{rem:333333}
The NOTE-algorithm has the following theocratical guarantee and computational security:
\begin{asparaenum}[(i) ]
\item Since two dynamic networks $N(t_i)\neq N(t_j)$ if $t_i\neq t_j$, thus, two every-zero hypergraph groups $\{G(\mathcal{E},t_i);[+][-]\}\neq \{G(\mathcal{E},t_j);[+][-]\}$.
\item The preappointed \emph{zero} $\mathcal{E}_{k}(t)\in \{G(\mathcal{E},t);[+][-]\}$ changes randomly over time $t\in [\alpha, \beta]$, in other words, two zeros $\mathcal{E}_{k}(t_i)\neq \mathcal{E}_{k}(t_j)$ for $t_i\neq t_j$.
\item By graph theory, we can get $|F_t(V(N(t)))|\leq \Delta(N(t))\leq N$ for holding $F_t(u)\neq F_t(v)$ for each edge $uv\in E(N(t))$.
\item Theorem \ref{thm:kinds-every-zero-number-based-string-groups} tells us: The hypergraph set $\mathcal{E}\big (\Lambda_{[1,N]}^2\big )$ can be classified into several kinds of every-zero number-based set-groups, such that each hyperedge set $\mathcal{E}\in \mathcal{E}(\Lambda_{[1,N]}^2)$ is in an every-zero number-based set-group of the set $\mathcal{E}\big (\Lambda_{[1,N]}^2\big )$. Thereby, the every-zero hypergraph group $\{G(\mathcal{E},t);[+][-]\}$ can run over on the hypergraph set $\mathcal{E}\big (\Lambda_{[1,N]}^2\big )$.
\item Each thing $s_{i}\in S_{thing}=\{s_{1},s_{2},\dots ,s_{N}\}$ can be a colored graph, or a vector, or a matrix, or a Topcode-matrix, or a hypergraph, or a number-based string, even a novel, or a story, or a poem, or an essay, \emph{etc.} (Ref. the every-zero hyperedge-set group above). Our goal is to increase the time cost for decipherers, and protects passwords created by using topology technology in the era of quantum computers.\qqed
\end{asparaenum}
\end{rem}

\subsection{The vertex/edge-intersected graphs of hypergraphs}

The vertex/edge-intersected graphs are some visualization tools of hypergraphs, which can help us understand, study, and apply hypergraphs.

\begin{defn}\label{defn:vertex-intersected-graph-hypergraph}
\cite{Yao-Ma-arXiv-2201-13354v1} Let $\mathcal{E}$ be a \emph{hyperedge set} defined on a finite set $\Lambda=\{x_1,x_2,\dots ,x_m\}$ (Ref. Definition \ref{defn:hypergraph-basic-definition}). Suppose that a $(p,q)$-graph $H$ admits a proper total set-labeling $F: V(H)\cup E(H)\rightarrow \mathcal{E}$ with $F(x)\neq F(y)$ for each edge $xy\in E(H)$, and $R_{est}(c_0,c_1,c_2,\dots ,c_m)$ with $m\geq 0$ is a constraint set, such that each edge $uv$ of $E(H)$ is colored with an edge color set $F(uv)$ and

(i) the first constraint $c_0$: $F(uv)\supseteq F(u)\cap F(v)\neq \emptyset$.

(ii) the $k$th constraint $c_k$: There is a function $\varphi_k$ for some $k\in [1,m]$, we have three numbers $c_{uv}\in F(uv)$, $a_u\in F(u)$ and $b_v\in F(v)$ holding the $k$th constraint $c_k:\varphi_k[a_u,c_{uv},b_v]=0$ true.

If a pair of hyperedges $e,e\,'\in \mathcal{E}$ with $|e\cap e\,'|\geq 1$ corresponds an edge $xy\in E(H)$, such that $F(x)=e$, $F(xy)\supseteq e\cap e\,'$ and $F(y)=e\,'$, then we call $H$ \emph{vertex-intersected graph} of the hypergraph $\mathcal{H}_{yper}=(\Lambda,\mathcal{E})$ subject to the constraint set $R_{est}(c_0,c_1,c_2,\dots ,c_m)$.\qqed
\end{defn}

\begin{rem}\label{rem:333333}
About Definition \ref{defn:vertex-intersected-graph-hypergraph}, we notice that:

(i) The total color set $F(V(H)\cup E(H))=\mathcal{E}$.

(ii) Each vertex $u$ of a vertex-intersected graph $H$ of the hypergraph $\mathcal{H}_{yper}=(\Lambda,\mathcal{E})$ corresponds to a hyperedge $e\in \mathcal{E}$, and each edge $uv$ of a vertex-intersected graph $H$ corresponds to $e\cap e\,'$ as $u$ corresponds to $e\in \mathcal{E}$ and $v$ corresponds to $e\,'\in \mathcal{E}$.

(iii) Each one of path, cycle and Hamilton cycle in a vertex-intersected graph $H$ of a hypergraph $\mathcal{H}_{yper}=(\Lambda,\mathcal{E})$ differs from that in the hypergraph $\mathcal{H}_{yper}=(\Lambda,\mathcal{E})$.

(iv) A vertex-intersected graph $H$ holds $F(uv)\not \subseteq F(uw)$ and $F(uw)\not \subseteq F(uc)$ for two adjacent edges $uv,uw\in E(H)$ with $v,w\in N_{ei}(u)$, then we say $H$ to be \emph{strong}.\qqed
\end{rem}

\begin{figure}[h]
\centering
\includegraphics[width=14.4cm]{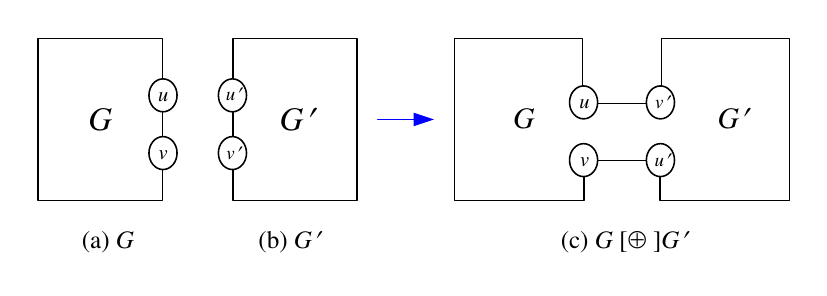}\\
\caption{\label{fig:inginite-v-intersedts}{\small A diagram for illustrating Theorem \ref{thm:infinite-v-intersected-graphs}.}}
\end{figure}

In Fig.\ref{fig:inginite-v-intersedts}, a graph $G$ is a vertex-intersected graph of a hypergraph, and the graph $G\,'$ is a copy of the vertex-intersected graph $G$, so the graph $G[\oplus]G\,'$ is a vertex-intersected graph of the hypergraph after removing two edges $uv$ and $u\,'v\,'$ from two graphs $G$ and $G\,'$, and joining the vertex $u$ with the vertex $v\,'$ by a new edge $uv\,'$, and joining the vertex $v$ with the vertex $u\,'$ by a new edge $vu\,'$.

\begin{thm}\label{thm:infinite-v-intersected-graphs}
$^*$ A hypergraph $\mathcal{H}_{yper}=(\Lambda,\mathcal{E})$ has infinite vertex-intersected graphs.
\end{thm}

\begin{problem}\label{qeu:444444}
Based on Definition \ref{defn:vertex-intersected-graph-hypergraph}, we propose the following questions:

(i) \textbf{Find} large integer $m\geq 0$ for the constraint set $R_{est}(c_0,c_1,c_2,\dots ,c_m)$ appeared in Definition \ref{defn:vertex-intersected-graph-hypergraph}.

(ii) \textbf{How} many hyperedge sets based on a finite set $\Lambda$ are there?

(iii) \textbf{Characterize} a vertex-intersected graph $H$ of a hypergraph $\mathcal{H}_{yper}=(\Lambda,\mathcal{E})$ defined in Definition \ref{defn:vertex-intersected-graph-hypergraph}, such that

(iii-a) $F(uv)\supseteq F(u)\cap F(v)\neq \emptyset$ and $|F(u)\cap F(v)|\geq k\geq 2$ for each edge $uv\in E(H)$.

(iii-b) $F(uv)\supseteq F(u)\cap F(v)\neq \emptyset$ and $\big ||F(u)|-|F(v)|\big |=1$ for each edge $uv\in E(H)$.

(iv) Since, there are many vertex-intersected graphs of a hypergraph $\mathcal{H}_{yper}=(\Lambda,\mathcal{E})$ defined in Definition \ref{defn:vertex-intersected-graph-hypergraph}, find a vertex-intersected graph $H^*$ of the hypergraph $\mathcal{H}_{yper}=(\Lambda,\mathcal{E})$, such that any proper subgraph $T\subset H^*$ is not a vertex-intersected graph of the hypergraph $\mathcal{H}_{yper}=(\Lambda,\mathcal{E})$.
\end{problem}

\begin{defn} \label{defn:edge-intersected-graph-hypergraphs}
$^*$ An \emph{edge-intersected graph} $L$ of a hypergraph $\mathcal{H}_{yper}=(\Lambda,\mathcal{E})$ is a colored graph admitting an edge set-coloring $\psi:E(L)\rightarrow \mathcal{E}$, such that each hyperedge $e\in \mathcal{E}$ corresponds to an edge $uv\in E(L)$ holding $e=\psi(uv)$ true, and each vertex $x\in V(L)$ admits an induced vertex set-color $\psi(x)$ defined by
\begin{equation}\label{eqa:555555}
\psi(x)=\psi(xy_1)\cap \psi(xy_2)\cap \cdots \cap \psi(xy_d)\neq \emptyset
\end{equation} for $y_i\in N_{ei}(x)=\{y_i:i\in [1,d]\}$ with vertex degree $d=\textrm{deg}_{L}(x)$, and $|\psi(E(L))|=|\mathcal{E}|$.\qqed
\end{defn}

\begin{rem}\label{rem:333333}
In Definition \ref{defn:edge-intersected-graph-hypergraphs}, the condition $\psi(x)\neq \emptyset $ is \emph{stronger}, which makes that a hypergraph $\mathcal{H}_{yper}=(\Lambda,\mathcal{E})$ may correspond to more edge-intersected graphs.\qqed
\end{rem}

\begin{problem}\label{question:444444}
We consider some particular vertex/edge-intersected graphs as follows:
\begin{asparaenum}[\textbf{Extre}-1. ]
\item If the hyperedge set $\mathcal{E}$ is a strong hyperedge set defined in Definition \ref{defn:strong-proper-hyperedge-sets}, then a vertex-intersected graph $H$ of the hypergraph $\mathcal{H}_{yper}=(\Lambda,\mathcal{E})$ is a complete graph.
\item If the hyperedge set $\mathcal{E}$ is a proper hyperedge set defined in Definition \ref{defn:strong-proper-hyperedge-sets}, then a vertex-intersected graph $H$ of the hypergraph $\mathcal{H}_{yper}=(\Lambda,\mathcal{E})$ is \emph{connected}.
\item A vertex-intersected graph of a hypergraph $\mathcal{H}_{yper}=(\Lambda,\mathcal{E})$ in Definition \ref{defn:vertex-intersected-graph-hypergraph} is one of planar graph, $H$-graph, bipartite graph, Euler graph, or some particular graphs.
\end{asparaenum}
\end{problem}

\begin{figure}[h]
\centering
\includegraphics[width=16.5cm]{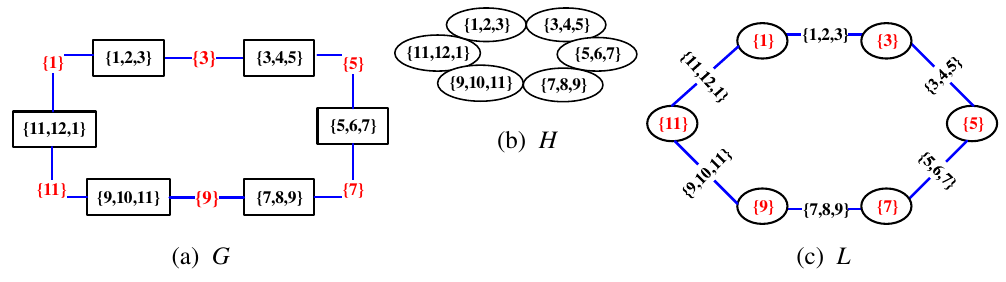}\\
\caption{\label{fig:inserted-graph-hypercycle}{\small (a) A vertex-intersected graph $G$ defined in Definition \ref{defn:vertex-intersected-graph-hypergraph} having a Hamilton cycle; (c) an edge-intersected graph $L$ defined in Definition \ref{defn:edge-intersected-graph-hypergraphs}.}}
\end{figure}

\begin{problem}\label{question:444444}
In Fig.\ref{fig:inserted-graph-hypercycle}, there are a hypergraph $\mathcal{H}_{yper}=(\Lambda,\mathcal{E})$ with $\Lambda=[1,12]$ and a hyperedge set
$$
\mathcal{E}=\big \{\{1,2,3\},\{3,4,5\},\{5,6,7\},\{7,8,9\},\{9,10,11\},\{11,12,1\}\big \}
$$ such that a vertex-intersected graph $G$ and the edge-intersected graph $L$ of a hypergraph $\mathcal{H}_{yper}=(\Lambda,\mathcal{E})$ have the same topological structure, namely, $G\cong L$. \textbf{Find} some conditions for a vertex-intersected graph $G_{vin}$ (as a \emph{private-key}) and the edge-intersected graph $G_{ein}$ (as a \emph{public-key}) of a hypergraph $\mathcal{H}_{yper}=(\Lambda,\mathcal{E})$, such that $G_{vin}\cong G_{ein}$.
\end{problem}
\begin{defn} \label{defn:more-terminology-group}
\cite{Yao-Ma-arXiv-2201-13354v1} About a vertex-intersected graph $H$ of a hypergraph $\mathcal{H}_{yper}=(\Lambda,\mathcal{E})$ defined in Definition \ref{defn:vertex-intersected-graph-hypergraph}, we, by means of terminology of graphs, define:
\begin{asparaenum}[\textbf{\textrm{Ter}}-1. ]
\item Each hyperedge $e\in \mathcal{E}$ has its own \emph{hyperedge degree} $\textrm{deg}_{\mathcal{E}}(e)=\textrm{deg}_H(x)$ if $F(x)=e$ for $x\in V(H)$.
\item The \emph{hyperedge degree sequence} $\{\textrm{deg}_{\mathcal{E}}(e_1),\textrm{deg}_{\mathcal{E}}(e_2),\dots, \textrm{deg}_{\mathcal{E}}(e_n)\}$ with $e_i\in \mathcal{E}$ satisfies Erd\"{o}s-Galia Theorem. In fact, each hyperedge degree
\begin{equation}\label{eqa:555555}
\textrm{deg}_{\mathcal{E}}(e_i)=\Big |\big \{e_j:e_i\cap e_j\neq \emptyset, e_j\in \mathcal{E}\setminus \{e_i\}\big \}\Big |
\end{equation}
\item If each hyperedge $e\in \mathcal{E}$ has its own hyperedge degree to be even, then $\mathcal{E}$ is called an \emph{Euler's hyperedge set}.
\item A \emph{hyperedge path} $\mathcal{P}$ in a hypergraph $\mathcal{H}_{yper}=(\Lambda,\mathcal{E})$ is
\begin{equation}\label{eqa:555555}
\mathcal{P}(e_1,e_m)=e_1e_2\cdots e_m=e_1(e_1\cap e_2)e_2(e_2\cap e_3)\cdots (e_{m-1}\cap e_m)e_m
\end{equation} with hyperedge intersections $e_{i}\cap e_{i+1}\neq \emptyset$ for $i\in [1,m-1]$, and each hyperedge $e_i$ is not an ear for $i\in [2,m-1]$. Moreover, the hyperedge path $\mathcal{P}$ is \emph{pure} if $e_1$ and $e_m$ are not ears of $\mathcal{E}$. If hyperedge intersections $|e_{i}\cap e_{i+1}|\geq r$ for $i\in [1,m-1]$, we call $\mathcal{P}$ $r$-\emph{uniform hyperedge path}.
\item If each pair of hyperedges $e,e\,'$ of $\mathcal{E}$ corresponds a hyperedge path $\mathcal{P}(e,e\,')$, then $\mathcal{H}_{yper}=(\Lambda,\mathcal{E})$ is a \emph{hyperedge connected hypergraph}, correspondingly, and its vertex-intersected graph is \emph{hyperedge connected}.
\item A \emph{hyperedge cycle} of a hypergraph $\mathcal{H}_{yper}=(\Lambda,\mathcal{E})$ is
\begin{equation}\label{eqa:555555}
\mathcal{C}_m=e_1e_2\cdots e_me_1=e_1(e_1\cap e_2)e_2(e_2\cap e_3)\cdots (e_{m-1}\cap e_m)e_m(e_m\cap e_1)e_1
\end{equation} with hyperedge intersections $e_{i}\cap e_{i+1}\neq \emptyset$ for $i\in [1,m-1]$ and $e_m\cap e_1\neq \emptyset$, as well as each $e_j$ with $j\in [1,m]$ is not an ear of $\mathcal{E}$. Furthermore, $\mathcal{C}_m$ is called \emph{hyperedge Hamilton cycle} if $m=|\mathcal{E}|$, and we call $\mathcal{C}_m$ $r$-\emph{uniform hyperedge cycle} if $|e_{i}\cap e_{i+1}|\geq r$ for $i\in [1,m-1]$ and $|e_m\cap e_1|\geq r$. By the way, if $\Lambda=\{x_1,x_2,\dots,x_m\}$, and $x_i\in e_{i}\cap e_{i+1}$ with $i\in [1,m-1]$, $x_m\in e_{m}\cap e_{1}$, we call $\mathcal{C}_m$ \emph{hypervertex Hamilton cycle}.

\item If a vertex-intersected graph $H$ is bipartite, then we have the hyperedge set $\mathcal{E}=X_{\mathcal{E}}\cup Y_{\mathcal{E}}$ with $X_{\mathcal{E}}\cap Y_{\mathcal{E}}=\emptyset$, such that any two hyperedges $e,e\,'\in X_{\mathcal{E}}$ (resp. $e,e\,'\in Y_{\mathcal{E}}$) satisfies $e\cap e\,'=\emptyset$.
\item A \emph{spanning hypertree} $\mathcal{T}$ of a vertex-intersected graph $H$ holds that each vertex color set $F(x)$ is not an ear of $\mathcal{E}$ if $x\not\in L(\mathcal{T})$, where $L(\mathcal{T})$ is the set of all leaves of $\mathcal{T}$, and $\mathcal{T}$ contains no hyperedge cycle.
\item If a vertex-intersected graph $H$ admits a proper vertex coloring $\theta:V(H)\rightarrow [1,\chi(H)]$, then the hyperedge set $\mathcal{E}$ admits a proper hyperedge coloring $\theta:\mathcal{E}\rightarrow [1,\chi(H)]$ such that $\theta(e)$ differs from $\theta(e\,')$ if $e\,'\cap e\neq \emptyset $ for any pair of subsets $e,e\,'\in \mathcal{E}$.
\item If a vertex-intersected graph $H$ is connected and the hyperedge set $\mathcal{E}$ contains no ear, so the \emph{diameter} $D(H)$ of a vertex-intersected graph $H$ is defined by
$$\max \{d(x,y): d(x,y)\textrm{ is the length of a shortest path between two vertices $x$ and $y$ in } H\}
$$ then the \emph{hyperdiameter} $D(\mathcal{E})$ of the hyperedge set $\mathcal{E}$ is defined by $D(\mathcal{E})=D(H)$.

\item A \emph{dominating hyperedge set} $\mathcal{E}_{domi}$ is a proper subset of the hyperedge set $\mathcal{E}$ and holds: Each hyperedge $e\in \mathcal{E}\setminus \mathcal{E}_{domi}$ corresponds some hyperedge $e^*\in \mathcal{E}_{domi}$ such that $e\cap e^*\neq \emptyset$.
\item The \emph{dual} $\mathcal{H}_{dual}$ of a hypergraph $\mathcal{H}_{yper}=(\Lambda,\mathcal{E})$ is also a hypergraph having its own vertex set $\Lambda_{dual}=\mathcal{E}=\{e_1,e_2,\dots ,e_n\}$ and its own hyperedge set $\mathcal{E}_{dual}=\{X_j\}^n_{j=1}$ with $X_j=\{e_i:x_j\in e_i\}$ and $n=|\mathcal{E}|$. Clearly, the dual of the hypergraph $\mathcal{H}_{dual}$ is just the original hypergraph $\mathcal{H}_{yper}=(\Lambda,\mathcal{E})$.\qqed
\end{asparaenum}
\end{defn}

\begin{rem}\label{rem:333333}
The problem of determining whether there exists a spanning tree in a given connected hypergraph is NP-complete, even when restricted to 3-regular linear hypergraphs or 4-uniform hypergraphs.\qqed
\end{rem}

\begin{thm} \label{them:hyperedge-Hamilton-cycles}
Let $G$ be a vertex-intersected graph of a hypergraph $\mathcal{H}_{yper}=(\Lambda,\mathcal{E})$, so a vertex-intersected graph $G$ admits a proper total set-labeling $F: V(G)\cup E(G)\rightarrow \mathcal{E}$ with $F(x)\neq F(y)$ for each edge $xy\in E(G)$ defined in Definition \ref{defn:vertex-intersected-graph-hypergraph}. Then we have:
\begin{asparaenum}[(1) ]
\item \cite{Yao-Ma-arXiv-2201-13354v1} If a vertex-intersected graph $G$ contains a Hamilton cycle, then the hypergraph $\mathcal{H}_{yper}$ contains a \emph{hyperedge Hamilton cycle}.
\item $^{*}$ If a vertex-intersected graph $G$ is a tree, then the hypergraph $\mathcal{H}_{yper}$ is acyclic by the Graham reduction defined in \cite{Jianfang-Wang-Hypergraphs-2008}.
\item $^{*}$ If a vertex-intersected graph $G$ holds $|F(E(G))|=|E(G)|$ and $F(uv)\cap F(xy)=\emptyset$ for any pair of edges $uv$ and $xy$ of $E(G)$, and a vertex-intersected graph $G$ is not a tree, then the hypergraph $\mathcal{H}_{yper}$ contains a \emph{hyperedge cycle}.
\end{asparaenum}
\end{thm}

\begin{example}\label{exa:8888888888}
In Fig.\ref{fig:1-example-hypergraph}, we can observe:

(a) An $8$-uniform hypergraph $\mathcal{H}_{yper}=(\Lambda,\mathcal{E})$ with its vertex set $\Lambda=[1,15]$ and hyperedge set $\mathcal{E}=\{e_1,e_2,e_3,e_4\}$;

(b) a vertex-intersected graph $G$ of the $8$-uniform hypergraph $\mathcal{H}_{yper}$ in (a), where $G$ admits a set-coloring $F:V(G)\rightarrow \mathcal{E}$ such that each edge $uv$ is colored with $F(uv)$ holding $F(uv)\supseteq F(u)\cap F(v)\neq \emptyset$, and $|F(uv)|=4$ for each edge $uv\in E(G)$, as well as $[1,15]=\Lambda=\bigcup _{e_i\in \mathcal{E}}e_i$.

We get a \emph{$7$-uniform adjacent hypergraph} $\overline{\mathcal{H}}_{yper}=(\Lambda,\overline{\mathcal{E}})$ defined in \cite{Jianfang-Wang-Hypergraphs-2008}, where $\overline{\mathcal{E}}=\{\overline{e}_1,\overline{e}_2,\overline{e}_3,\overline{e}_4\}$ with $\overline{e}_1=\{$3, 4, 5, 8, 10, 13, 14$\}$, $\overline{e}_2=\{$1, 4, 8, 9, 12, 14, 15$\}$, $\overline{e}_3=\{$1, 2, 3, 7, 13, 14, 15$\}$, $\overline{e}_4=\{$1, 2, 3, ,4, 5, 6, 12$\}$. Each edge of a vertex-intersected graph $\overline{H}$ of the adjacent hypergraph $\overline{\mathcal{H}}_{yper}$ is colored with a set having cardinality 3. Thereby, we call $\mathcal{H}_{yper}=(\Lambda,\mathcal{E})$ and $\overline{\mathcal{H}}_{yper}=(\Lambda,\overline{\mathcal{E}})$ \emph{ve-double uniform hypergraphs}.\qqed
\end{example}

\begin{example}\label{exa:8888888888}
We have the \emph{dual hypergraph} $\mathcal{H}_{dual}=(\Lambda_{dual},\mathcal{E}_{dual})$ of an $8$-uniform hypergraph $\mathcal{H}_{yper}=(\Lambda,\mathcal{E})$ shown in Fig.\ref{fig:1-example-hypergraph} (a) with its vertex set $\Lambda_{dual}=\mathcal{E}=\{e_1,e_2,e_3,e_4\}$, and its hyperedge set
$$
\mathcal{E}_{dual}=\big \{X_j\big \}^{15}_{j=1}=\big \{\{e_i:x_j\in e_i\in \mathcal{E}\}\big \}^{15}_{j=1}
$$ where $X_1=\{e_1\}$, $X_2=\{e_1,e_2\}$, $X_3=\{e_2\}$, $X_4=\{e_3\}$, $X_5=\{e_2,e_3\}$, $X_6=\{e_1,e_2,e_3\}$, $X_7=\{e_1,e_2,e_4\}$, $X_8=\{e_3,e_4\}$, $X_9=\{e_1,e_3,e_4\}$, $X_{10}=\{e_2,e_3,e_4\}$, $X_{11}=\{e_1,e_2,e_3,e_4\}$, $X_{12}=\{e_1,e_3\}$, $X_{13}=\{e_2,e_4\}$, $X_{14}=\{e_4\}$ and $X_{15}=\{e_1,e_4\}$.

In the above dual hypergraph $\mathcal{H}_{dual}$, we have the hyperedge degree $\textrm{deg}(e_i)=8$ with $i\in [1,4]$, so the hyperedge degree sequence $\{\textrm{deg}(e_1),\textrm{deg}(e_2),\textrm{deg}(e_3),\textrm{deg}(e_4)\}$ satisfies the Erd\"{o}s-Galia Theorem. The vertex-intersected graph $G_{dual}$ of the dual hypergraph $\mathcal{H}_{dual}$ admits a set-coloring $F_{dual}:V(G_{dual})\rightarrow \mathcal{E}_{dual}$, such that each induced edge color
$$
F_{dual}(u_iv_j)\supseteq F_{dual}(u_i)\cap F_{dual}(v_j)=X_i\cap X_j\neq \emptyset
$$ holds true.\qqed
\end{example}

\begin{example}\label{exa:8888888888}
Fig.\ref{fig:coloring-hypergraph} (b) shows us a vertex-intersected graph $G_{yper}$ of a hypergraph $\mathcal{H}_{yper}=(\Lambda,\mathcal{E})$ subject to the constraint set $R_{est}(c_1)$, and $\mathcal{H}_{yper}$ has its own vertex set $\Lambda=[1,12]$ and one hyperedge set
\begin{equation}\label{eqa:example-hypergraph}
{
\begin{split}
\mathcal{E}=&\big \{\{2,12\},\{1,11\},\{1,10\},\{6,10,11,12\},\{4,5,6\},\{5,7\}, \{7,8,9\},\\
& \{1,8\},\{4,9\},\{3,4\},\{2,8\},\{1,2,3\},\{1,7\}\big \}
\end{split}}
\end{equation} since $[1,12]=\bigcup _{e\in \mathcal{E}}e$.

Moreover, a vertex-intersected graph $G_{yper}$ of the hypergraph $\mathcal{H}_{yper}=([1,12]$, $\mathcal{E})$ contains a clique $\big \{\{1,2,3\},\{1,11\}$, $\{1,10\}$, $\{1,8\}$, $\{1,7\}\big \}$, and four hyperedge cycles $\big \{\{1,2,3\},\{2,8\},\{2,12\}\big \}$,
$\big \{ \{2,8\},\{1,8\},\{7,8,9\}\big \}$, $\big \{\{4,5,6\},\{4,9\},\{3,4\}\big \}$ and $\big \{\{1,7\},\{7,8,9\},\{5,7\}\big \}$.

Furthermore, the intersected-hypergraph $G_{yper}$ has a \emph{Hamilton hyperedge cycle}
$${
\begin{split}
C_{yper}=&\{2,12\}\{\textbf{\textcolor[rgb]{0.00,0.00,1.00}{12}}\}\{6,10,11,12\}\{\textbf{\textcolor[rgb]{0.00,0.00,1.00}{6}}\}\{4,5,6\}\{\textbf{\textcolor[rgb]{0.00,0.00,1.00}{5}}\}\{5,7\}
\{\textbf{\textcolor[rgb]{0.00,0.00,1.00}{7}}\}\{7,8,9\}\{\textbf{\textcolor[rgb]{0.00,0.00,1.00}{9}}\}\{4,9\}\{\textbf{\textcolor[rgb]{0.00,0.00,1.00}{4}}\}\{3,4\}\\
&\{\textbf{\textcolor[rgb]{0.00,0.00,1.00}{3}}\}\{1,2,3\}\{\textbf{\textcolor[rgb]{0.00,0.00,1.00}{1}}\}\{1,7\}\{\textbf{\textcolor[rgb]{0.00,0.00,1.00}{1}}\}
\{1,10\}\{\textbf{\textcolor[rgb]{0.00,0.00,1.00}{1}}\}\{1,11\}\{\textbf{\textcolor[rgb]{0.00,0.00,1.00}{1}}\}\{1,8\}\{\textbf{\textcolor[rgb]{0.00,0.00,1.00}{8}}\}\{2,8\}\{\textbf{\textcolor[rgb]{0.00,0.00,1.00}{2}}\}\{2,12\}
\end{split}}
$$ because of a vertex-intersected graph $G_{yper}$ has no \emph{ear}.

\begin{figure}[h]
\centering
\includegraphics[width=16.4cm]{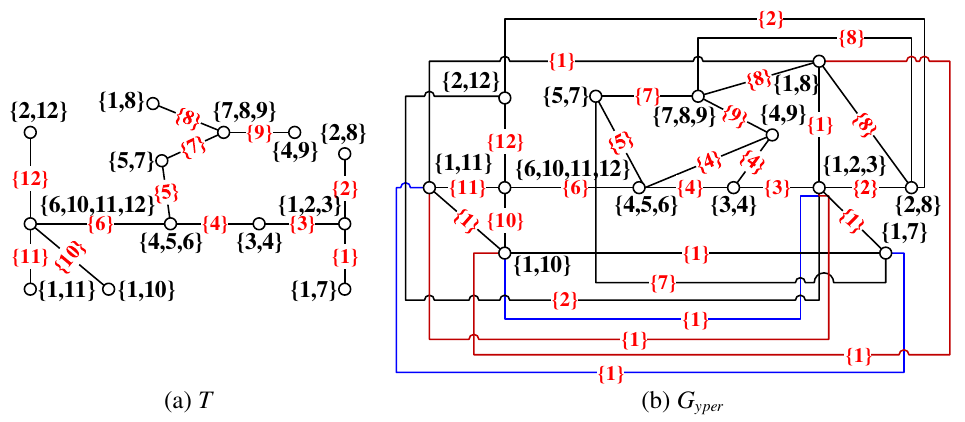}\\
\caption{\label{fig:coloring-hypergraph}{\small (a) A tree $T$ admits a graceful-intersection total set-labeling $F: V(T)\cup E(T)\rightarrow \mathcal{E}$, where $\mathcal{E}$ is defined in Eq.(\ref{eqa:example-hypergraph}); (b) a vertex-intersected graph $G_{yper}$ of a hypergraph $\mathcal{H}_{yper}=(\Lambda,\mathcal{E})$ subject to the constraint set $R_{est}(c_1)$, and $T$ is a spanning tree of a vertex-intersected graph $G_{yper}$.}}
\end{figure}

It is noticeable, there are two or more graphs admitting graceful-intersection total set-labelings defined on a unique hyperedge set $\mathcal{E}$, see examples $T,T_1,T_2$ and $T_3$ shown in Fig.\ref{fig:coloring-hypergraph} (a) and Fig.\ref{fig:more-trees-one-hypergraph} (a), Fig.\ref{fig:more-trees-one-hypergraph} (b) and Fig.\ref{fig:more-trees-one-hypergraph} (c). Clearly, a vertex-intersected graph $G_{yper}$ contains each of $T,T_1,T_2,T_3$ as its set-colored subgraphs. Notice that each set-colored tree $T_i$ has no \emph{ear}.\qqed
\end{example}

\begin{figure}[h]
\centering
\includegraphics[width=16.4cm]{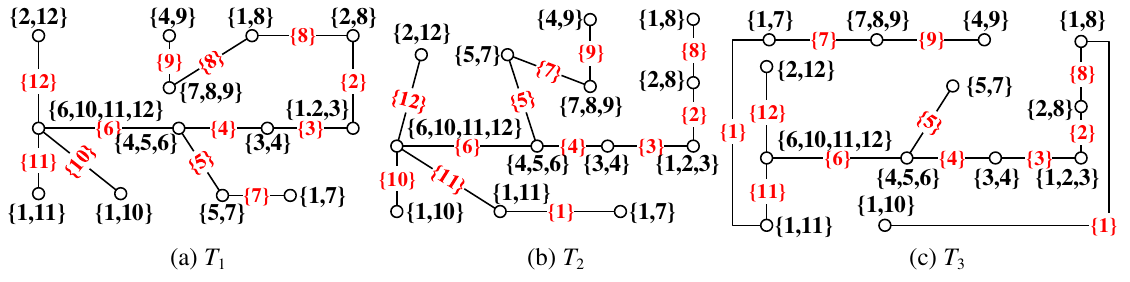}\\
\caption{\label{fig:more-trees-one-hypergraph}{\small Three trees $T_1,T_2,T_3$ admit three graceful-intersection total set-labelings defined on a unique hyperedge set $\mathcal{E}$ defined in Eq.(\ref{eqa:example-hypergraph}), however, $T_i\not\cong T_j$ for $i\neq j$.}}
\end{figure}

\begin{example}\label{exa:8888888888}
In Fig.\ref{fig:K4-join-K6}, we show a set-colored graph $K^*=K_4 \bowtie K_6$, where the hyperedge set $\mathcal{E}=\{e_1$, $e_2,\dots ,e_{10}\}$ defined on a finite set $\Lambda=\{a,b,c,e,f,g,h,i,j,k\}$, each vertex color set of $K_4$ is an ear, however each vertex color set of $K_6$ is not an ear, and moreover $K^*$ contains non-ear-hyperedge cycles $C_6$ and ear-hyperedge cycles $C_4$. The operation ``$\bowtie$'' is called \emph{intersection operation}.\qqed
\end{example}

\begin{figure}[h]
\centering
\includegraphics[width=16.4cm]{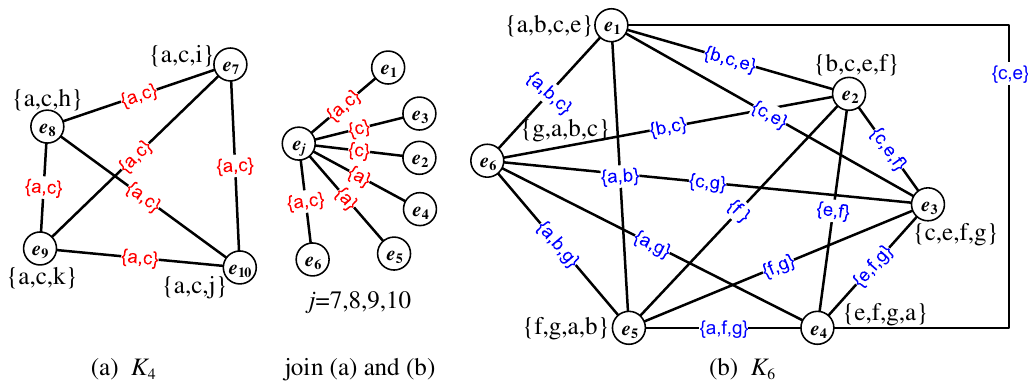}\\
\caption{\label{fig:K4-join-K6}{\small A set-colored graph $K^*=K_4 \bowtie K_6$, where (a) a complete graph $K_4$; (b) a complete graph $K_6$.}}
\end{figure}

\begin{rem}\label{rem:333333}
The concepts of hyperedge path, hyperedge cycle and Hamilton hyperedge cycle are defined here for distinguishing popular path, cycle and Hamilton cycle of graphs. A graph $G$ has cycles or paths, however, if $G$ is a set-colored graph defined by a set-coloring $F:V(G)\rightarrow \mathcal{E}$, where $\mathcal{E}$ is a hyperedge set defined on a finite set $\Lambda$, it may happen that $G$ has no hyperedge cycle or hyperedge path. Others, such as hyperedge degree, hyperedge coloring, hyperdiameter \emph{etc.}, are defined here for providing methods of measuring hypergraphs.

In graphs, a vertex $x$ has its own degree $\textrm{deg}(x)=|\{y: y\in N_{ei}(x)\}|=|N_{ei}(x)|$, and an edge $uv$ has its own degree $\textrm{deg}(uv)=|N_{ei}(u)|+|N_{ei}(v)|-2$. In hypergraphs, a hypervertex $x\in \Lambda$ has its own \emph{hypervertex degree} $\textrm{deg}_{\mathcal{E}}(x)=\big |\big \{e_i: x\in e_i\in \mathcal{E}\big \}\big |$, and a hyperedge $e_i\in \mathcal{E}$ has its own \emph{hyperedge degree} $\textrm{deg}(e_i)=\big |\big \{e_j: e_i\cap e_j\neq \emptyset,e_j\in \mathcal{E}\setminus \{e_i\}\big \}\big |$. Thereby, the hypervertex degree is a concept only related with a unique hyperedge set, not for other hyperedge sets.\qqed
\end{rem}

\begin{prop}\label{qeu:99999}
$^*$ If a hypergraph $\mathcal{H}_{yper}=(\Lambda,\mathcal{E})$ has its own vertex-intersected graphs to be connected, then each bipartition $(\mathcal{E}_1,\mathcal{E}_2)$ of $\mathcal{E}$ holds that there exists a pair of hyperedges $e\in \mathcal{E}_1$ and $e\,'\in \mathcal{E}_2$ satisfying $e\cap e\,'\neq \emptyset$.
\end{prop}

\begin{problem}\label{qeu:zzzzzzzzzzz}
\textbf{Characterize} vertex-intersected graphs having at least one of the following properties: (i) double-uniform; (ii) non-ear; and (iii) each vertex color set is an ear.
\end{problem}

\begin{figure}[h]
\centering
\includegraphics[width=15.4cm]{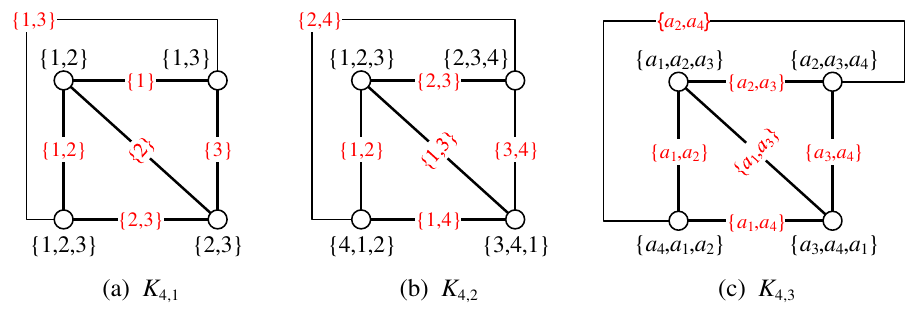}\\
\caption{\label{fig:optimal-set-matrix} {\small Three different set-colorings of the complete graph $K_4$.}}
\end{figure}

\begin{thm}\label{thm:theorem3535}
\cite{Yao-Ma-arXiv-2201-13354v1} For any connected $(p,q)$-graph $G$, there is a set-coloring $F:V(G)\rightarrow \mathcal{E}$, where the hyperedge set $\mathcal{E}\subset \Lambda^2$ and $\bigcup_{e\in \mathcal{E}}e=\Lambda$, such that $F(x)\neq F(y)$ for distinct vertices $x,y\in V(G)$ and induced edge colors $F(uv)=F(u)\cap F(v)$ for each edge $uv\in E(G)$ holding $F(xy)\neq F(wz)$ for distinct edges $xy,wz\in E(G)$, as well as $|\Lambda|$ is the smallest one for any set-coloring $F^*:V(G)\rightarrow \mathcal{E}^*$ with $\mathcal{E}^*\subset \mathcal{E}(\Lambda^2_{*})$ and $\bigcup_{e\in \mathcal{E}^*}e=\Lambda_{*}$, and moreover $F(x)\not\subset F(y)$ for distinct vertices $x,y\in V(G)$ and $F(xy)\not\subset F(wz)$ for distinct edges $xy,wz\in E(G)$.
\end{thm}

For understanding Theorem \ref{thm:theorem3535}, see two examples $K_{4,1}$ and $K_{4,2}$ shown in Fig.\ref{fig:optimal-set-matrix}.

\begin{problem}\label{qeu:444444}
Let $\mathcal{E}\big (\Lambda^2\big )=\big \{\mathcal{E}_1,\mathcal{E}_2,\dots ,\mathcal{E}_{n(G)}\big \}$ be the set of hyperedge sets, such that each hyperedge set $\mathcal{E}_i$ holds $\Lambda=\bigcup _{e\in \mathcal{E}_i}e$ and forms a hypergraph $\mathcal{H}^i_{yper}=(\Lambda, \mathcal{E}_i)$, as well as $G_i$ is a vertex-intersected graph of the hypergraph $\mathcal{H}^i_{yper}=(\Lambda, \mathcal{E}_i)$. \textbf{Find} some connections between two hypergraphs $\mathcal{H}^i_{yper}=(\Lambda, \mathcal{E}_i)$ and $\mathcal{H}^j_{yper}=(\Lambda, \mathcal{E}_j)$ if $i\neq j$, and \textbf{estimate} the value of the number $n(G)$.
\end{problem}

\begin{rem}\label{rem:333333}
Notice that a vertex-intersected graph of a hypergraph $\mathcal{H}_{yper}=(\Lambda,\mathcal{E})$ subject to the constraint set $R_{est}(c_1)$ appeared in \cite{Jianfang-Wang-Hypergraphs-2008}.

(i) One set-type Topcode-matrix $T^{set}_{code}=(X^{set},~E^{set},~Y^{set})^{T}$ defined in Definition \ref{defn:set-type-topcode-matrix-definition} may correspond two or more vertex-intersected graphs. There are four set-colored graphs $T_1,T_2,T_3,T_4$ shown in Fig.\ref{fig:vertex-split-k-4}, they correspond the set-type Topcode-matrix $H^{yper}_{code}(K_{4,2})$ shown in Eq.(\ref{eqa:hyper-graph-set-topcode-matrix}). In the topological structure of view, $K_{4,2}$ shown in Fig.\ref{fig:optimal-set-matrix}(b), $T_1,T_2,T_3$ and $T_4$ are not isomorphic from each other. In fact, these four set-colored graphs $T_1,T_2,T_3,T_4$ are obtained from the set-colored graph $K_{4,2}$ and the vertex-splitting operation, conversely, each set-colored graph $T_i$ with $i\in[1,4]$ is \emph{graph homomorphism} to $K_{4,2}$, that is, $T_i\rightarrow K_{4,2}$.

(ii) By the set-type Topcode-matrix $T^{set}_{code}(H)$ defined in Eq.(\ref{eqa:set-type-topcode-matrix}), we have a \emph{string-type Topcode-matrix} $T^{string}_{code}(H)=(X^*,~E^*,~Y^*)^{T}$ with \emph{v-vector} $X^*=(x_1, x_2, \dots, x_q)$, \emph{e-vector} $E^*=(e_1$, $e_2 $, $ \dots $, $e_q)$ and \emph{v-vector} $Y^*=(y_1, y_2, \dots, y_q)$ consist of number-based strings $x_i$, $e_i$ and $y_i$ for $i\in [1,q]$, where each number-based string $x_i$ with $i\in [1,q]$ is a permutation of $a_{i,1},a_{i,2},\dots ,a_{i,A_i}$ with $a_{i,j}\in F(x_i)$ and cardinality $A_i=|F(x_i)|$, each number-based string $e_i$ is a permutation of $c_{i,1},c_{i,2},\dots ,c_{i,C_i}$ with $c_{i,j}\in F(e_i)$ and cardinality $C_i=|F(e_i)|$, and each number-based string $y_i$ is a permutation of $b_{i,1},b_{i,2},\dots ,b_{i,B_i}$ with $b_{i,j}\in F(y_i)$ and cardinality $B_i=|F(y_i)|$.

Thereby, one set-type Topcode-matrix $T^{set}_{code}(H)$ enables us to obtain $n(ABC)$ numbered-based string Topcode-matrices like $T^{string}_{code}(H)$, and we get $n(ABC)\cdot (3q)!$ number-based strings in total, where the number
\begin{equation}\label{eqa:555555}
n(ABC)=\prod^q_{i=1}(A_i)!\prod^q_{i=1}(C_i)!\prod^q_{i=1}(B_i)!
\end{equation}

It is meaningful to explore more applications of vertex-intersected graphs of hypergraphs defined in Definition \ref{defn:vertex-intersected-graph-hypergraph}.\qqed
\end{rem}

Observe a set-colored graph $K_{4,2}$ shown in Fig.\ref{fig:optimal-set-matrix}, and this graph $K_{4,2}$ corresponds its own set-type Topcode-matrix $T^{set}_{code}(K_{4,2})$ as follows:

\begin{equation}\label{eqa:hyper-graph-set-topcode-matrix}
\centering
{
\begin{split}
T^{set}_{code}(K_{4,2})= \left(
\begin{array}{cccccc}
\{1,2,3\} & \{1,2,3\} & \{1,2,3\} & \{4,1,2\} & \{2,3,4\} & \{2,3,4\}\\
\{1,2\} & \{1,3\} & \{2,3\} & \{1,4\} & \{2,4\} & \{3,4\}\\
\{4,1,2\} & \{3,4,1\} & \{2,3,4\} & \{3,4,1\} & \{4,1,2\} & \{3,4,1\}
\end{array}
\right)
\end{split}}
\end{equation}

\begin{figure}[h]
\centering
\includegraphics[width=16.4cm]{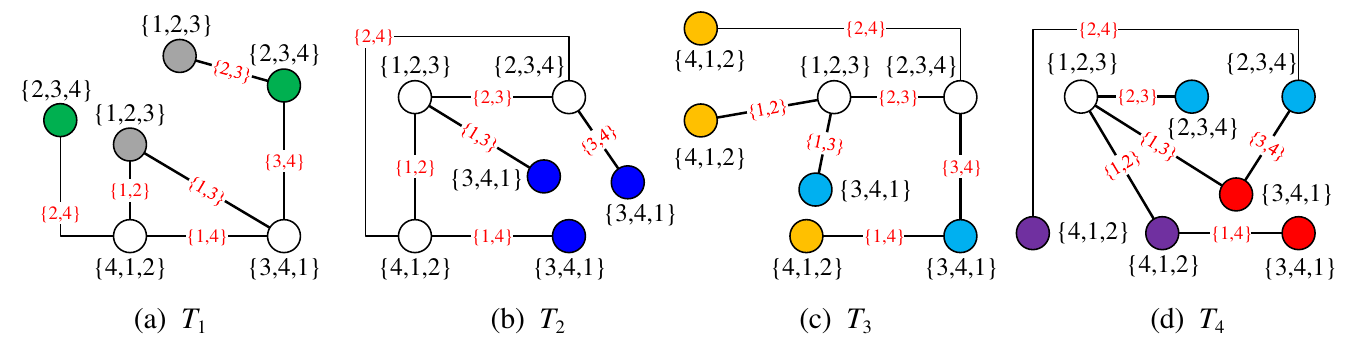}\\
\caption{\label{fig:vertex-split-k-4} {\small Four set-colored graphs correspond a unique set-type Topcode-matrix $H^{yper}_{code}(K_{4,2})$ shown in Eq.(\ref{eqa:hyper-graph-set-topcode-matrix}).}}
\end{figure}

\begin{thm}\label{thm:666666}
\cite{Yao-Ma-arXiv-2201-13354v1} If a $(p,q)$-graph $G$ admits a graceful-intersection total set-labeling defined on a hyperedge set $\mathcal{E}\subseteq [0,q]^2$ with $\bigcup_{e_i\in \mathcal{E}} e_i=[0,q]$, then a vertex-intersected graph of the hypergraph $\mathcal{H}_{yper}=([0,q],\mathcal{E})$ contains each of $(p,q)$-graphs admitting graceful-intersection total set-labelings defined on the hyperedge set $\mathcal{E}$.
\end{thm}

\begin{problem}\label{qeu:444444}
Suppose that a connected graph $G$ admits a graceful-intersection set-coloring $F$ defined on a hyperedge set $\mathcal{E}$ based on a set $\Lambda$, and $G$ is a vertex-intersected graph of the hypergraph $\mathcal{H}_{yper}=(\Lambda,\mathcal{E})$. \textbf{Does} the hyperedge set $\mathcal{E}=F(V(G))$ contain a perfect hypermatching $\{M_1,M_2, \dots,M_s\}\subset \mathcal{E}$ such that $M_i\cap M_j=\emptyset $ for $i\neq j$ and $\bigcup ^s_{i=1}M_i=\Lambda$?
\end{problem}

\begin{problem}\label{question:444444}
$^*$ If a hyperedge set $\mathcal{E}$ defined on a finite set $\Lambda$ holds $\Lambda =\bigcup_{e,e\,'\in\mathcal{E}}(e\cap e\,')$ true, then $\mathcal{E}^*=\{e\cap e\,':e,e\,'\in\mathcal{E}\}$ is a hyperedge set, \textbf{characterize} a vertex-intersected graph of a hypergraph $\mathcal{H}_{yper}=(\Lambda,\mathcal{E}^*)$.
\end{problem}

\textbf{The hypergraph vertex-intersected graph base} is $\textbf{\textrm{G}}_{yper}(t)=\{G_1(t), G_2(t),\dots ,G_{n}(t)\}$ with $t\in [a,b]$, where each graph $G_i(t)$ is a vertex-intersected graph of a hypergraph $\mathcal{H}_{yper}(t)=(\Lambda(t),\mathcal{E}_i(t))$, and each hyperedge set $\mathcal{E}_i(t)\in \mathcal{E}(\Lambda^2(t))=\{\mathcal{E}_1(t), \mathcal{E}_2(t),\dots ,\mathcal{E}_{n}(t)\}$ based on the vertex set $\Lambda(t)=\{s_1(t)$, $ s_2(t)$, $\dots$ ,$s_{m}(t)\}$, which is a thing set. Since each vertex-intersected graph $G_i(t)$ admits a total set-coloring $F_i:V(G_i(t))\cup E(G_i(t))\rightarrow \mathcal{E}_{i}(t)$, such that each edge $uv\in E(G_i(t))$ holds $F_i(uv)\supseteq F_i(u)\cap F_i(v)$ with $F_i(u)\cap F_i(v)\neq \emptyset$.

Do the non-multi-edge vertex-coinciding operation to a vertex-intersected graph $G_i(t)$ and another vertex-intersected graph $G_j(t)$, we get a vertex-coincided graph $L_{i,j}(t)=G_i(t)[\bullet_{k(i,j)}]G_j(t)$, and it admits a total set-coloring $F_{i,j}$ defined by

(i) $F_{i,j}(w_{i,j})=F_i(u_i)\cup F_i(v_i)$, where coincided vertices $w_{i,j}=u_i\bullet v_i$ for $u_i\in E(G_i(t))$ and $v_i\in E(G_j(t))$ with $i\in [1,k(i,j)]$);

(ii) $F_{i,j}(x)=F_i(x)$ for $x\in (V(G_i(t)\cup E(G_i(t))\setminus \{u_i:i\in [1,k(i,j)]\}$; and

(iii) $F_{i,j}(y)=F_j(y)$ for $y\in (V(G_j(t)\cup E(G_j(t))\setminus \{v_i:i\in [1,k(i,j)]\}$.

We call the following set
\begin{equation}\label{eqa:pan-lattices55}
{
\begin{split}
\textbf{\textrm{L}}_{yper}(Z^0[\textbf{\textrm{O}}]\textbf{\textrm{G}}(t))=\Big \{\big [\bullet_{\varphi}\big ]^n_{k=1}e_kG_k(t):~b_k\in Z^0,G_k(t)\in \textbf{\textrm{G}}_{yper}(t)\Big \},~\sum^n_{k=1}e_k\geq 1
\end{split}}
\end{equation} \emph{hypergraph vertex-intersected graph graphic lattice}, such that each graph $L\in \textbf{\textrm{L}}_{yper}(Z^0[\textbf{\textrm{O}}]\textbf{\textrm{G}}(t))$ has its own edge number $|E(L)|=\sum ^n_{k=1}e_k|E(G_k(t))|$, and is a vertex-intersected graph of a hypergraph $\mathcal{H}_{yper}(t)=(\Lambda(t),\mathcal{E}_k(t))$ based on some hyperedge set $\mathcal{E}_k(t)\in \mathcal{E}(\Lambda^2(t))$.

\begin{thm}\label{thm:666666}
$^*$ The hypergraph vertex-intersected graph graphic lattice $\textbf{\textrm{L}}_{yper}(Z^0[\textbf{\textrm{O}}]\textbf{\textrm{G}}(t))$ in Eg.(\ref{eqa:pan-lattices55}) has shown: There are infinite vertex-intersected graphs for a hypergraph $\mathcal{H}_{yper}=(\Lambda, \mathcal{E})$.
\end{thm}

\subsection{Parameterized hypergraphs}

We will use the following terminology and natation:

\begin{asparaenum}[$\bullet$~]
\item For a \emph{parameterized set} $\Lambda_{(m,b,n,k,a,d)}=S_{m,0,b,d}\cup S_{n,k,a,d}$ with
\begin{equation}\label{eqa:555555}
S_{m,0,b,d}=\big \{bd,(b+1)d\dots, md\big \},~ S_{n,k,a,d}=\big \{k+ad,k+(a+1)d,\dots ,k+(a+n)d\big \}
\end{equation} the power set $\Lambda^2_{(m,b,n,k,a,d)}$ collects all subsets of the parameterized set $\Lambda_{(m,b,n,k,a,d)}$.

\item A \emph{parameterized hyperedge set} $\mathcal{E}^P\subset \Lambda^2_{(m,b,n,k,a,d)}$ holds $\bigcup _{e\in \mathcal{E}^P}e=\Lambda_{(m,b,n,k,a,d)}$ true.
\end{asparaenum}

Motivated from the hypergraph definition, we present the parameterized hypergraph as follows:

\begin{defn}\label{defn:parameterized-hypergraph-basic-definition}
\cite{Bing-Yao-arXiv:2207-03381} A \emph{parameterized hypergraph} $\mathcal{P}_{hyper}=(\Lambda_{(m,b,n,k,a,d)},\mathcal{E}^P)$ defined on a parameterized hypervertex set $\Lambda_{(m,b,n,k,a,d)}$ holds:

(i) Each element of $\mathcal{E}^P$ is not empty and called a \emph{parameterized hyperedge};

(ii) $\Lambda_{(m,b,n,k,a,d)}=\bigcup _{e\in \mathcal{E}^P}e$, where each element of $\Lambda_{(m,b,n,k,a,d)}$ is called a \emph{vertex}, and the cardinality $|\Lambda_{(m,b,n,k,a,d)}|$ is the \emph{order} of $\mathcal{P}_{hyper}$;

(iii) $\mathcal{E}^P$ is called \emph{parameterized hyperedge set}, and the cardinality $|\mathcal{E}^P|$ is the \emph{size} of $\mathcal{P}_{hyper}$. \qqed
\end{defn}

\begin{problem}\label{problem:99999}
By Definition \ref{defn:parameterized-hypergraph-basic-definition}, let $S(\Lambda^P)=\big \{\mathcal{E}^P_1,\mathcal{E}^P_2,\dots ,\mathcal{E}^P_{M}\big \}$ be the set of parameterized hyperedge sets defined on a parameterized set $\Lambda_{(m,b,n,k,a,d)}$, so that each parameterized hyperedge set $\mathcal{E}^P_i$ with $i\in [1,M]$ holds $\Lambda_{(m,b,n,k,a,d)}=\bigcup _{e\in \mathcal{E}^P_i}e$ true. \textbf{Estimate} the number $M$ of parameterized hypergraphs defined on the parameterized set $\Lambda_{(m,b,n,k,a,d)}$.
\end{problem}

\begin{defn} \label{defn:operation-graphs-pa-hypergraph}
\cite{Bing-Yao-arXiv:2207-03381} Let $\textbf{\textrm{O}}=(O_1,O_2,\dots ,O_m)$ be an operation set with $m\geq 1$, and if an element $c$ is obtained by implementing an operation $O_i\in \textbf{\textrm{O}}$ to other two elements $a,b$, we write this fact as $c=a[O_i]b$. Suppose that a $(p,q)$-graph $G$ admits a proper $(k,d)$-total set-coloring $F: V(G)\cup E(G)\rightarrow \mathcal{E}^P$, where $\mathcal{E}^P$ is the parameterized hyperedge set of a parameterized hypergraph $\mathcal{P}_{hyper}=(\Lambda_{(m,b,n,k,a,d)},\mathcal{E}^P)$ defined in Definition \ref{defn:parameterized-hypergraph-basic-definition}. There are the following constraints:
\begin{asparaenum}[\textrm{\textbf{Pahy}}-1.]
\item \label{pahyper:parameterized-graph-3} Only one operation $O_k\in \textbf{\textrm{O}}$ holds $c_{uv}=a_u[O_k]b_v$ for each edge $uv\in E(G)$, where $a_u\in F(u)$, $b_v\in F(v)$ and $c_{uv}\in F(uv)$.
\item \label{pahyper:parameterized-graph-1} For each operation $O_i\in \textbf{\textrm{O}}$, each edge $uv\in E(G)$ holds $c_{uv}=a_u[O_i]b_v$ for $a_u\in F(u)$, $b_v\in F(v)$ and $c_{uv}\in F(uv)$.
\item \label{pahyper:parameterized-graph-2} Each $z\in F(uv)$ for each edge $uv\in E(G)$ corresponds to an operation $O_j\in \textbf{\textrm{O}}$, such that $z=x[O_j]y$ for some $x\in F(u)$ and $y\in F(v)$.
\item \label{pahyper:parameterized-graph-v} Each $a_x\in F(x)$ for any vertex $x\in V(G)$ corresponds to an operation $O_s\in \textbf{\textrm{O}}$ and an adjacent vertex $y\in N_{ei}(x)$, such that $z_{xy}=a_x[O_s]b_y$ for some $z_{xy}\in F(xy)$ and $b_y\in F(y)$.
\item \label{pahyper:parameterized-graph-4} Each operation $O_t\in \textbf{\textrm{O}}$ corresponds to some edge $xy\in E(G)$ holding $c_{xy}=a_x[O_t]b_y$ for $a_x\in F(x)$, $b_y\in F(y)$ and $c_{xy}\in F(xy)$.
\item \label{pahyper:parameterized-graph-must} If there are three different sets $e_i,e_j,e_k\in F(V(G))$ and an operation $O_r\in \textbf{\textrm{O}}$ holding $e_i[O_r]e_j\subseteq e_k$, then there exists an edge $xy\in E(G)$, such that $F(x)=e_i$, $F(y)=e_j$ and $F(xy)=e_k$.
\end{asparaenum}
\noindent \textbf{Then $G$ is}:
\begin{asparaenum}[\textrm{\textbf{Ograph}}-1.]
\item a \emph{$(k,d)$-$\textbf{\textrm{O}}$-vertex operation graph} of $\mathcal{P}_{hyper}$ if Pahy-\ref{pahyper:parameterized-graph-1} and Pahy-\ref{pahyper:parameterized-graph-must} hold true.
\item a \emph{$(k,d)$-$\textbf{\textrm{O}}$-edge operation graph} of $\mathcal{P}_{hyper}$ if Pahy-\ref{pahyper:parameterized-graph-1}, Pahy-\ref{pahyper:parameterized-graph-2} and Pahy-\ref{pahyper:parameterized-graph-must} hold true.
\item a \emph{$(k,d)$-$\textbf{\textrm{O}}$-total operation graph} of $\mathcal{P}_{hyper}$ if Pahy-\ref{pahyper:parameterized-graph-1}, Pahy-\ref{pahyper:parameterized-graph-2}, Pahy-\ref{pahyper:parameterized-graph-v} and Pahy-\ref{pahyper:parameterized-graph-must} hold true.
\item a \emph{$(k,d)$-edge operation graph} of $\mathcal{P}_{hyper}$ if Pahy-\ref{pahyper:parameterized-graph-2} and Pahy-\ref{pahyper:parameterized-graph-must} hold true.
\item a \emph{$(k,d)$-vertex operation graph} of $\mathcal{P}_{hyper}$ if Pahy-\ref{pahyper:parameterized-graph-v} and Pahy-\ref{pahyper:parameterized-graph-must} hold true.
\item a \emph{$(k,d)$-total operation graph} of $\mathcal{P}_{hyper}$ if Pahy-\ref{pahyper:parameterized-graph-2}, Pahy-\ref{pahyper:parameterized-graph-v} and Pahy-\ref{pahyper:parameterized-graph-must} hold true.
\item a \emph{$(k,d)$-non-uniform operation graph} of $\mathcal{P}_{hyper}$ if Pahy-\ref{pahyper:parameterized-graph-4} and Pahy-\ref{pahyper:parameterized-graph-must} hold true.
\item a \emph{$(k,d)$-one operation graph} of $\mathcal{P}_{hyper}$ if Pahy-\ref{pahyper:parameterized-graph-3} and Pahy-\ref{pahyper:parameterized-graph-must} hold true.\qqed
\end{asparaenum}
\end{defn}

\begin{example}\label{exa:8888888888}
If the operation set $\textbf{\textrm{O}}$ contains only one operation ``$\cap$'', which is the \emph{intersection operation} on sets, and the $(p,q)$-graph $G$ satisfies Pahy-\ref{pahyper:parameterized-graph-1} and Pahy-\ref{pahyper:parameterized-graph-must} in Definition \ref{defn:operation-graphs-pa-hypergraph}, then $G$ is a $(k,d)$-one operation graph of the parameterized hypergraph $\mathcal{P}_{hyper}$, also, $G$ is called \emph{$(k,d)$-vertex-intersected graph} of the parameterized hypergraph $\mathcal{P}_{hyper}$. As $(k,d)=(1,1)$, the graph $G$ is just the \emph{vertex-intersected graph} of a hypergraph $H_{hyper}$ defined in \cite{Jianfang-Wang-Hypergraphs-2008}.\qqed
\end{example}

\begin{thm}\label{thm:graph-admits-6-set-colorings}
\cite{Bing-Yao-arXiv:2207-03381} Each connected graph $G$ admits each one of the following $W$-constraint $(k,d)$-total set-colorings for $W$-constraint $\in \{$graceful, harmonious, edge-difference, edge-magic, felicitous-difference, graceful-difference$\}$.
\end{thm}
\begin{proof} Since a connected graph $G$ can be vertex-split into a tree $T$, such that we have a graph homomorphism $T\rightarrow G$ under a mapping $\theta:V(T)\rightarrow V(G)$, and each tree admits a graceful $(k,d)$-total coloring $g_1$, a harmonious $(k,d)$-total coloring $g_2$, an edge-difference $(k,d)$-total coloring $g_3$, an edge-magic $(k,d)$-total coloring $g_4$, a felicitous-difference $(k,d)$-total coloring $g_5$ and a graceful-difference $(k,d)$-total coloring $g_6$. Then the connected graph $G$ admits a $(k,d)$-total set-coloring $F$ defined by $F(u)=\{g_i(u^*):i\in [1,6],u^*\in V(T)\}$ for each vertex $u\in V(G)$ with $u=\theta(u^*)$ for $u^*\in V(T)$, and $F(uv)=\{g_i(u^*v^*):i\in [1,6],u^*v^*\in E(T)\}$ for each edge $uv\in E(G)$ with $uv=\theta(u^*)\theta(v^*)$ for each edge $u^*v^*\in E(T)$.

The proof of the theorem is complete.
\end{proof}

\begin{example}\label{exa:complete-graph-K-4}
Theorem \ref{thm:graph-admits-6-set-colorings} tells us: Each connected graph $G$ admits each one of the following $W$-constraint $(k,d)$-total set-colorings for $W$-constraint $\in \{$graceful, harmonious, edge-difference, edge-magic, felicitous-difference, graceful-difference$\}$. In Fig.\ref{fig:k-d-operation-graph}, doing the vertex-split to a complete graph $K_4$ produces a tree $T$, such that $E(K_4)=E(T)$, and ``$T\rightarrow K_4$'' is a graph homomorphism from the tree $T$ into the complete graph $K_4$. And moreover, there are six colored trees $T_i$ for $i\in [1,6]$, where

$T_1$ admits a graceful $(k,d)$-total labeling $f_1$;

$T_2$ admits a harmonious $(k,d)$-total labeling $f_2$;

$T_3$ admits a felicitous-difference $(k,d)$-total labeling $f_3$;

$T_4$ admits an edge-magic $(k,d)$-total labeling $f_4$;

$T_5$ admits an edge-difference $(k,d)$-total labeling $f_5$; and

$T_6$ admits a graceful-difference $(k,d)$-total labeling $f_6$.

\begin{figure}[h]
\centering
\includegraphics[width=16.4cm]{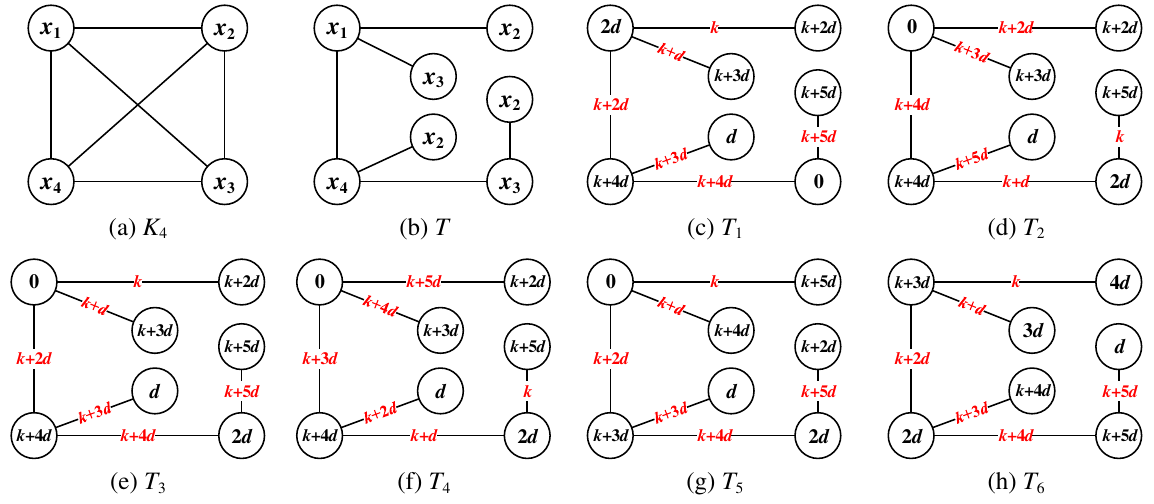}\\
\caption{\label{fig:k-d-operation-graph}{\small (a) A complete graph $K_4$; (b)-(h) six colored trees $T_i$ for $i\in [1,6]$.}}
\end{figure}

By Fig.\ref{fig:k-d-operation-graph}, we get a parameterized hypergraph $\mathcal{P}^*_{hyper}=(\Lambda_{(4,0,5,k,0,d)},\mathcal{E}^P)$ with the parameterized hypervertex set
$$\Lambda_{(4,0,5,k,0,d)}=S_{4,0,0,d}\cup S_{5,k,0,d}=\big \{0,d,2d,3d,4d\}\cup \{k,k+d,k+2d,k+3d,k+4d,k+5d\big \}$$
and the parameterized hyperedge set $\mathcal{E}^P=\{e_i:~i\in [1,10]\}$ containing $e_1=\{0,2d,k+3d\}$,

$e_2=\big \{0,d,3d,4d,k+2d,k+4d,k+5d\big \}$, $e_3=\big \{0,2d,3d,k+3d,k+4d,k+5d\big \}$,

$e_4=\big \{d,2d,k+3d,k+4d\big \}$, $e_5=\big \{0,k,k+2d,k+5d\big \}$,

$e_6=\big \{k+d,k+3d,k+4d\big \}$, $e_7=\big \{k+2d,k+3d,k+4d\big \}$,

$e_8=\big \{k,k+5d\big \}$, $e_9=\big \{d,k+2d,k+3d,k+5d\big \}$ and $e_{10}=\big \{k+d,k+4d\big \}$.

\vskip 0.2cm

Clearly, $\Lambda_{(m,b,n,k,a,d)}=\bigcup _{e_i\in \mathcal{E}^P}e_i$. The complete graph $K_4$ admits a proper $(k,d)$-total set-coloring $F: V(K_4)\cup E(K_4)\rightarrow \mathcal{E}^P$, where $F(x_1)=e_1$, $F(x_2)=e_2$, $F(x_3)=e_3$, $F(x_4)=e_4$, $F(x_1x_2)=e_5$, $F(x_1x_3)=e_6$, $F(x_1x_4)=e_7$, $F(x_2x_3)=e_8$, $F(x_2x_4)=e_9$ and $F(x_3x_4)=e_{10}$.

We have an operation set $\textbf{\textrm{O}}=(O_1,O_2,\dots ,O_7)$, where
\begin{asparaenum}[(1) ]
\item The operation $O_1$ is the graceful $(k,d)$-total labeling $f_1$, such that the constraint $f_1(uv)=|f_1(u)-f_1(v)|$ for each edge $uv \in E(T_1)$ and $f_1(E(T_1))=S_{5,k,0,d}$.
\item The operation $O_2$ is the harmonious $(k,d)$-total labeling $f_2$, such that the constraint $f_2(uv)=f_2(u)+f_2(v)~(\bmod~6d)$ for each edge $uv \in E(T_2)$ and $f_2(E(T_2))=S_{5,k,0,d}$.
\item The operation $O_3$ is the felicitous-difference $(k,d)$-total labeling $f_3$, such that the felicitous-difference constraint $|f_3(u)+f_2(v)-f_3(uv)|=2d$ for each edge $uv \in E(T_3)$ and $f_3(E(T_3))=S_{5,k,0,d}$.
\item The operation $O_4$ is the edge-magic $(k,d)$-total labeling $f_4$, such that the edge-magic constraint $f_4(u)+f_4(uv)+f_4(v)=2k+7d$ for each edge $uv \in E(T_4)$ and $f_4(E(T_4))=S_{5,k,0,d}$.
\item The operation $O_5$ is the edge-difference $(k,d)$-total labeling $f_5$, such that the edge-difference constraint $f_5(uv)+|f_5(u)-f_5(v)|=2k+5d$ for each edge $uv \in E(T_5)$ and $f_5(E(T_5))=S_{5,k,0,d}$.
\item The operation $O_6$ is the graceful-difference $(k,d)$-total labeling $f_6$, such that the graceful-difference constraint $\big ||f_6(u)-f_6(v)|-f_6(uv)\big |=d$ for each edge $uv \in E(T_6)$ and $f_6(E(T_6))=S_{5,k,0,d}$.
\item The operation $O_7$ is the intersection operation ``$\bigcap $'', such that $F(x_ix_j)\bigcap [F(x_i)\cap F(x_j)]\neq \emptyset $ for each edge $x_ix_j \in E(K_4)$.
\end{asparaenum}

Thereby, we claim that the complete graph $K_4$ is every one of the operation graphs Ograph-1, Ograph-2, Ograph-3, Ograph-4, Ograph-5 and Ograph-6 of the parameterized hypergraph $\mathcal{P}^*_{hyper}=(\Lambda_{(4,0,5,k,0,d)},\mathcal{E}^P)$ according to Definition \ref{defn:operation-graphs-pa-hypergraph}.\qqed
\end{example}

\begin{figure}[h]
\centering
\includegraphics[width=16.4cm]{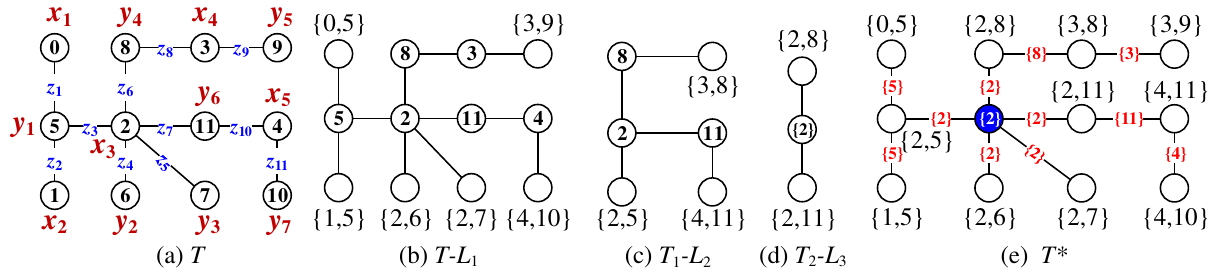}\\
\caption{\label{fig:VSET-coloring-algorithm-1}{\small An example for understanding the proof of Theorem \ref{thm:build-hyperedge-set}, cited from \cite{Yao-Ma-arXiv-2201-13354v1}.}}
\end{figure}

\begin{example}\label{exa:8888888888}
In Fig.\ref{fig:VSET-coloring-algorithm-1} (a), the tree $T$ has its own Topcode-matrix as follows
{\small
\begin{equation}\label{eqa:VSET-coloring-algorithm-matrix-1}
\centering
{
\begin{split}
T_{code}(T)&= \left(
\begin{array}{cccccccccccc}
x_1 & x_2 & x_3 & x_3 & x_3 & x_3 & x_3 & x_4 & x_4 & x_5 & x_5\\
x_1y_1 & x_2y_1 & x_3y_1 & x_3y_2 & x_3y_3 & x_3y_4 & x_3y_6 & x_4y_4 & x_4y_5 & x_5y_6 & x_5y_7\\
y_1 & y_1 & y_1 & y_2 & y_3 & y_4 & y_6 & y_4 & y_5 & y_6 &y_7
\end{array}
\right)_{3\times 11}\\
&=(X_T,E_T,Y_T)^{T}_{3\times 11}
\end{split}}
\end{equation}
}with the vertex-vector $X_T=(x_1, x_2, x_3, x_3, x_3, x_3, x_3, x_4, x_4, x_5, x_5)$, the edge-vector
$$E_T=(x_1y_1, x_2y_1, x_3y_1, x_3y_2, x_3y_3, x_3y_4, x_3y_6, x_4y_4, x_4y_5, x_5y_6, x_5y_7)=(z_1,z_2,\cdots ,z_{11})$$
and the vertex-vector $Y_T=(y_1, y_1, y_1, y_2, y_3, y_4, y_6, y_4, y_5,y_6,y_7)$, where $V(T)=X_T\cup Y_T$ and $E(T)=E_T$.

\vskip 0.4cm

Notice that the tree $T$ admits a vertex labeling $f$ holding $f(u)\neq f(v)$ for any pair of distinct vertices $u,v\in V(T)$ shown in Fig.\ref{fig:VSET-coloring-algorithm-1} (a), then we get the colored Topcode-matrix
{\footnotesize
\begin{equation}\label{eqa:VSET-coloring-algorithm-matrix-22}
\centering
{
\begin{split}
T_{code}(T,f)= \left(
\begin{array}{cccccccccccc}
0 & 1 & 2 & 2 & 2 & 2 & 2 & 3 & 3 & 4 & 4\\
f(z_1) & f(z_2) & f(z_3) & f(z_4) & f(z_5) & f(z_6) & f(z_7) & f(z_8) & f(z_9) & f(z_{10}) & f(z_{11})\\
5 & 5 & 5 & 6 & 7 & 8 & 11 & 8 & 9 & 11 & 10
\end{array}
\right)
\end{split}}
\end{equation}
}with $f(z_i)=z_i$ for $i\in [1,11]$.

For \emph{bipartite graphs}, especially, we define the \emph{unite Topcode-matrix} as follows
\begin{equation}\label{eqa:unit-Topcode-matrix}
{
\begin{split}
I\,^0=\left(
\begin{array}{ccccc}
0 & 0 & \cdots & 0\\
1 & 1 & \cdots & 1\\
1 & 1 & \cdots & 1
\end{array}
\right)_{3\times q}=(X\,^0,~E\,^0,~Y\,^0)^T
\end{split}}
\end{equation} with two vertex-vectors $X\,^0=(0, 0, \dots ,0)_{1\times q}$ and $Y\,^0=(1, 1, \dots ,1)_{1\times q}$, and the edge-vector $E\,^0=(1, 1, \dots ,1)_{1\times q}$.

By Eq.(\ref{eqa:unit-Topcode-matrix}) and Eq.(\ref{eqa:VSET-coloring-algorithm-matrix-22}), the tree $T$ has its own parameterized Topcode-matrix $P_{ara}(T,F|k,d)$ defined as
{\footnotesize
\begin{equation}\label{eqa:VSET-coloring-algorithm-matrix-33}
\centering
{
\begin{split}
&\quad P_{ara}(T,F)=k\cdot I\,^0+d\cdot T_{code}(T,f)\\
&= \left(
\begin{array}{cccccccccccc}
0 & d & 2d & 2d & 2d & 2d & 2d & 3d & 3d & 4d & 4d\\
F(z_1) & F(z_2) & F(z_3) & F(z_4) & F(z_5) & F(z_6) & F(z_7) & F(z_8) & F(z_9) & F(z_{10}) & F(z_{11})\\
k+5d & k+5d & k+5d & k+6d & k+7d & k+8d & k+11d & k+8d & k+9d & k+11d & k+10d
\end{array}
\right)
\end{split}}
\end{equation}
}with $F(z_i)=k+f(z_i)\cdot d$ for $i\in [1,11]$, where we can define $f(z_i)$ to be a number obtained by some $W$-constraint coloring of graph theory.

Fig.\ref{fig:VSET-coloring-algorithm-1} (e) shows us a set-coloring $g$ of the tree $T^*$ as follows:

\textbf{(a-1)} $g(x_1)=\{0,5\}$, $g(x_2)=\{1,5\}$, $g(x_3)=\{2\}$, $g(x_4)=\{3,8\}$ and $g(x_5)=\{4,11\}$;

\textbf{(a-2)} $g(y_1)=\{2,5\}$, $g(y_2)=\{2,6\}$, $g(y_3)=\{2,7\}$, $g(y_4)=\{2,8\}$, $g(y_5)=\{3.9\}$, $g(y_6)=\{2,11\}$ and $g(y_7)=\{4,10\}$; and

\textbf{(a-3)} $g(z_1)=\{5\}$, $g(z_2)=\{5\}$, $g(z_3)=\{2\}$, $g(z_4)=\{2\}$, $g(z_5)=\{2\}$, $g(z_6)=\{2\}$, $g(z_7)=\{2\}$, $g(z_8)=\{8\}$, $g(z_9)=\{3\}$, $g(z_{10})=\{11\}$, $g(z_{11})=\{4\}$.

\vskip 0.4cm

We have a parameterized hypervertex set
$${
\begin{split}
\Lambda_{(4,0,11,k,0,d)}&=S_{4,0,0,d}\cup S_{11,k,5,d}\\
&=\big \{0,d,2d,3d,4d\big \}\cup \big \{k+5d,k+6d,k+7d,k+8d,k+9d,k+10d,k+11d\big \}
\end{split}}
$$

Thereby, the tree $T$ admits a $(k,d)$-total set-coloring $F_{k,d}:V(T)\cup E(T)\rightarrow \Lambda^2_{(4,0,11,k,0,d)}$ defined as:

\textbf{(b-1)} The vertex $(k,d)$-colors are $F_{k,d}(x_1)=\{0,k+5d\}$, $F_{k,d}(x_2)=\{d,k+5d\}$, $F_{k,d}(x_3)=\{2d\}$, $F_{k,d}(x_4)=\{3d,k+8d\}$ and $F_{k,d}(x_5)=\{4d,k+11d\}$;

\textbf{(b-2)} The vertex $(k,d)$-colors are $F_{k,d}(y_1)=\{2d,k+5d\}$, $F_{k,d}(y_2)=\{2d,k+6d\}$, $F_{k,d}(y_3)=\{2d,k+7d\}$, $F_{k,d}(y_4)=\{2d,k+8d\}$, $F_{k,d}(y_5)=\{3d,k+9d\}$, $F_{k,d}(y_6)=\{2d,k+11d\}$ and $F_{k,d}(y_7)=\{4d,k+10d\}$; and

\textbf{(b-3)} The edge $(k,d)$-colors are $F_{k,d}(z_1)=\{k+5d\}$, $F_{k,d}(z_2)=\{k+5d\}$, $F_{k,d}(z_3)=\{2d\}$, $F_{k,d}(z_4)=\{2d\}$, $F_{k,d}(z_5)=\{2d\}$, $F_{k,d}(z_6)=\{2d\}$, $F_{k,d}(z_7)=\{2d\}$, $F_{k,d}(z_8)=\{k+8d\}$, $F_{k,d}(z_9)=\{3d\}$, $F_{k,d}(z_{10})=\{k+11d\}$, $F_{k,d}(z_{11})=\{4d\}$.

\vskip 0.4cm

We get a \emph{parameterized hypergraph} $\mathcal{P}_{hyper}=(\Lambda_{(4,0,11,k,0,d)},\mathcal{E}^P)$, where the parameterized hyperedge set $\mathcal{E}^P=\{e_j:j\in [1,12]\}$ with elements $e_1=\{0,k+5d\}$, $e_2=\{d,k+5d\}$, $e_3=\{2d\}$, $e_4=\{3d,k+8d\}$ and $e_5=\{4d,k+11d\}$, $e_6=\{2d,k+5d\}$, $e_7=\{2d,k+6d\}$, $e_8=\{2d,k+7d\}$, $e_9=\{2d,k+8d\}$, $e_{10}=\{3d,k+9d\}$, $e_{11}=\{2d,k+11d\}$ and $e_{12}=\{4d,k+10d\}$.

Since $\Lambda_{(4,0,11,k,0,d)}=\bigcup _{e_j\in \mathcal{E}^P}e_j$, then the colored tree $T$ admitting the $(k,d)$-total set-coloring $F_{k,d}$ is a subgraph of some vertex-intersected graph of the parameterized hypergraph $\mathcal{P}_{hyper}$.\qqed
\end{example}

\subsubsection{PWCSC-algorithms on colored trees}

The sentence ``Producing $W$-constraint set-coloring algorithm'' is abbreviated as ``PWCSC-algorithm'' in the following discussion. The content of this subsection is cited from \cite{Bing-Yao-arXiv:2207-03381}.

\vskip 0.4cm

\textbf{PWCSC-algorithm-A for ordered-path.}

\textbf{Initialization-A.} Suppose that $T$ is a tree admitting a $W$-constraint labeling $f$ holding $f(uv)\neq f(xy)$ for any pair of edges $uv$ and $xy$ of the tree $T$, and each edge $uv\in E(T)$ holds the $W$-constraint $f(uv)=W\langle f(u),f(v)\rangle $, as well as $|f(V(T))|=|V(T)|$.

\vskip 0.2cm

\textbf{Step A-1.} Do the VSET-coloring algorithm introduced in the proof of Theorem \ref{thm:build-hyperedge-set} to $T$ first. We select a longest path
$$P_1=w^1_1w^1_2w^1_3\cdots w^1_{m_1-2}w^1_{m_1-1}w^1_{m_1}
$$ of the tree $T$, then we have the neighbor set $NN_{ei}(w^1_2)=L(w^1_2)\cup \big \{w^1_3\big \}$, where the leaf set $L(w^1_2)=\big \{w^1_1, v^1_{2,1}, v^1_{2,2}$, $ \dots $, $v^1_{2,d_2}\big \}$ with $d_2=\textrm{deg}_T(w^1_2)-2$, and another adjacent neighbor set $NN_{ei}(w^1_{m_1-1})=\big \{w^1_{m_1-2}\big \}\cup L(w^1_{m_1-1})$ with the leaf set $L(w^1_{m_1-1})=\big \{w^1_{m_1}, u^1_{m_1,1}, u^1_{m_1,2},\dots , u^1_{m_1,d_{m_1}}\big \}$, where $d_{m_1}=\textrm{deg}_T(w^1_{m_1-1})$ $-2$. We define a total set-coloring $F$ for the three $T$ as: The vertex set-colors are
\begin{equation}\label{eqa:step-a-11}
F_{path}(x)=\big \{f(x),f(w^1_2)\big \}_1,~x\in L(w^1_2);\quad F_{path}(y)=\big \{f(y),f(w^1_{m_1-1})\big \}_1,~y\in L(w^1_{m_1-1})
\end{equation}

\textbf{Step A-2.} We get a tree $T_1=T-\big [L(w^1_2)\cup L(w^1_{m_1-1})\big ]$ by removing all leaves of two vertices $w^1_2$ and $w^1_{m_1-1}$ of the tree $T$, and then do the VSET-coloring algorithm to $T_1$. Notice that the tree $T_1$ admits the set-coloring $F$, so we select a longest path
$$P_2=w^2_1w^2_2w^2_3\cdots w^2_{m_2-2}w^2_{m_2-1}w^2_{m_2}
$$ of $T_1$, then we have the neighbor set $NN_{ei}(w^2_2)=L(w^2_2)\cup \big \{w^2_3\big \}$, where the leaf set $L(w^2_2)=\big \{w^2_1, v^2_{2,1}$, $ v^2_{2,2}$, $\dots $, $v^2_{2,d_2}\big \}$ with $d_2=\textrm{deg}_T(w^2_2)-2$, and another neighbor set $NN_{ei}(w^2_{m_2-1})=\big \{w^2_{m_2-2}\big \}\cup L(w^2_{m_2-1})$ with the leaf set $L(w^2_{m_2-1})=\big \{w^2_{m_2}, u^2_{m_2,1}, u^2_{m_2,2},\dots , u^2_{m_2,d_{m_2}}\big \}$, where $d_{m_2}=\textrm{deg}_T(w^2_{m_2-1})-2$. We get the following vertex set-colors
\begin{equation}\label{eqa:step-a-22}
F_{path}(x)=\big \{f(x),f(w^2_2)\big \}_2,~x\in L(w^2_2);\quad F_{path}(y)=\big \{f(y),f(w^2_{m_2-1})\big \}_2,~y\in L(w^2_{m_2-1})
\end{equation}

\textbf{Step A-3.} If the tree $T_2=T_1-\big [L(w^2_2)\cup L(w^2_{m_2-1})\big ]$ has its own diameter $D(T_2)\geq 3$, then we goto Step A-2.

\vskip 0.2cm

\textbf{Step A-4.} After $k-1$ times, we get the tree $T_k=T_{k-1}-\big [L(w^k_2)\cup L(w^k_{m_k-1})\big ]$ to be a star $K_{1,n}$ with its own diameter $D(K_{1,n})=2$, then we have the vertex set $V(K_{1,n})=\big \{x_0, y_i:i\in [1,n]\big \}$ and the edge set $E(K_{1,n})=\big \{x_0y_i:i\in [1,n]\big \}$. We have the following vertex set-colors
\begin{equation}\label{eqa:step-a-kk}
F_{path}(x_0)=\big \{f(x_0)\big \}_{k+1};\quad F_{path}(y_i)=\big \{f(y_i),f(x_0)\big \}_{k+1},~y_i\in L(x_0)
\end{equation} Notice that $|F_{path}(u)|=2$ for $u\in V(T)\setminus \big \{x_0\big \}$, and $|F_{path}(x_0)|=1$.

\textbf{Step A-5.} By the $W$-constraint we recolor the edges of the tree $T$ as follows:
\begin{equation}\label{eqa:step-a-ee}
F_{path}(uv)=[F_{path}(u)\cap F_{path}(v)]\cup \big \{W\langle a,b\rangle:~a\in F_{path}(u),~b\in F_{path}(v)\big \},~uv\in E(T)
\end{equation} since $F_{path}(u)\cap F_{path}(v)\neq \emptyset$.

\textbf{Step A-6.} Return the $W$-constraint proper total set-coloring $F$ of the tree $T$, since $F_{path}(s)\neq F_{path}(t)$ for any pair of adjacent, or incident elements $s,t\in V(T)\cup E(T)$.

\vskip 0.2cm

By the PWCSC-algorithm-A for ordered-path, we present a result as follows:

\begin{thm}\label{thm:set-ordered-graceful-PWCSC-algorithm-A for ordered-path}
\cite{Bing-Yao-arXiv:2207-03381} If a tree $T$ admits a set-ordered $W$-constraint labeling, then $T$ admits a $W$-constraint proper total set-coloring $F$ obtained by the PWCSC-algorithm-A for ordered-path, such that $|F(u)\cap F(v)|=1$ and $|F(uv)|\geq 2$ for each edge $uv\in E(T)$, and $|F(x)|=2$ for $x\in V(T)\setminus \{x_0\}$, and $|F(x_0)|=1$.
\end{thm}

\begin{example}\label{exa:PWCSC-algorithm-A-ordered-path-11}
An example for understanding the above PWCSC-algorithm-A for ordered-path is shown in Fig.\ref{fig:VSET-coloring-algorithm-2}. A tree $T$ shown in Fig.\ref{fig:VSET-coloring-algorithm-2} (a) admits a set-ordered graceful labeling $f$, such that the set-ordered constraint $\max f(X)=f(x_6)<f(y_1)=\min f(Y)$ holds true, where $X=\{x_i:i\in [1,6]\}$ and $Y=\{y_j:j\in [1,8]\}$. The last tree $T_4$ is a star $K_{1,4}$ shown in Fig.\ref{fig:VSET-coloring-algorithm-2} (e), $T_4$ admits a $W$-constraint proper total set-coloring. The tree $T_6$ admits a graceful proper total set-coloring $F$ satisfied Theorem \ref{thm:set-ordered-graceful-PWCSC-algorithm-A for ordered-path}. We have the following facts:

\textbf{Fact-1. }The tree $T$ has its own Topcode-matrix $T_{code}(T,f)$ as
{\small
\begin{equation}\label{eqa:PWCSC-algorithm-A for ordered-path-11}
\centering
{
\begin{split}
T_{code}(T,f)=\left(
\begin{array}{cccccccccccccc}
5 & 5 & 5 & 4 & 4 & 4 & 4 & 3 & 2 & 1 & 0 & 0 & 0\\
1 & 2 & 3 & 4 & 5 & 6 & 7 & 8 & 9 & 10 & 11 & 12 & 13\\
6 & 7 & 8 & 8 & 9 & 10 & 11 & 11 & 11 & 11 & 11 & 12 & 13
\end{array}
\right)_{3\times 13}=(X_T,E_T,Y_T)^T
\end{split}}
\end{equation}
}with
$${
\begin{split}
X_T=&(5, 5, 5, 4, 4, 4, 4, 3, 2, 1, 0, 0, 0)\\
=&\big (f(x_6),f(x_6),f(x_6),f(x_5),f(x_5),f(x_5),f(x_5),f(x_4),f(x_3),f(x_2),f(x_1),f(x_1),f(x_1)\big )\\
E_T=&(1, 2, 3, 4, 5, 6, 7, 8, 9, 10, 11, 12, 13)\\
=&\big (f(x_6y_1),f(x_6y_2),f(x_6y_3),f(x_5y_3),f(x_5y_4),f(x_5y_5),f(x_5y_6),f(x_4y_6),f(x_3y_6),f(x_2y_6),\\
&f(x_1y_6),f(x_1y_7),f(x_1y_8)\big )\\
Y_T=&(6, 7, 8, 8, 9, 10, 11, 11, 11, 11, 11, 12, 13)\\
=&\big (f(y_1),f(y_2),f(y_3),f(y_3),f(y_4),f(y_5),f(y_6),f(y_6),f(y_6),f(y_6),f(y_6),f(y_7),f(y_8)\big )
\end{split}}
$$Clearly, $\max X_T=5<6=\min Y_T$, and $|E_T|=[1,13]$.

\vskip 0.2cm

\textbf{Fact-2. }The tree $T$ has its own parameterized Topcode-matrix $P_{ara}(T,\theta|k,d)$ defined as
\begin{equation}\label{eqa:VSET-coloring-algorithm-matrix-33}
\centering
{
\begin{split}
P_{ara}(T,\theta|k,d)=k\cdot I\,^0+d\cdot T_{code}(T,f)=(X_P,E_P,Y_P)^T
\end{split}}
\end{equation} with the vertex $(k,d)$-color vectors $X_P,Y_P$ and the edge $(k,d)$-color vector $E_P$, where
$${
\begin{split}
X_P=&(5d,~5d,~5d,~4d,~4d,~4d,~4d,~3d,~2d,~d,~0,~0,~0)\\
E_P=&\big (k+d,~ k+2d,~ k+3d,~ k+4d,~ k+5d,~ k+6d,~ k+7d,~ k+8d,~ k+9d,~ k+10d,~ \\
&k+11d,~k+12d,~ k+13d\big )\\
Y_P=&\big (k+6d,~ k+7d,~ k+8d,~ k+8d,~ k+9d,~ k+10d,~ k+11d,~ k+11d,~ k+11d,~ k+11d,\\
&k+11d,~ k+12d,~ k+13d\big )
\end{split}}$$

\textbf{Fact-3. }By the graceful proper total set-coloring $F$, the tree $T_6$ has its own \emph{set-type Topcode-matrix} $S_{et}(T_6,F)=(X_{et},E_{et},Y_{et})^T$ with
\begin{equation}\label{eqa:set-coloring}
{
\begin{split}
X_{et}=&\big (F(x_6),F(x_6),F(x_6),F(x_5),F(x_5),F(x_5),F(x_5),F(x_4),F(x_3),F(x_2),\\
&F(x_1),F(x_1),F(x_1)\big )\\
E_{et}=&\big (F(x_6y_1),F(x_6y_2),F(x_6y_3),F(x_5y_3),F(x_5y_4),F(x_5y_5),F(x_5y_6),F(x_4y_6),\\
&F(x_3y_6),F(x_2y_6),F(x_1y_6),F(x_1y_7),F(x_1y_8)\big )\\
Y_{et}=&\big (F(y_1),F(y_2),F(y_3),F(y_3),F(y_4),F(y_5),F(y_6),F(y_6),F(y_6),F(y_6),F(y_6),\\
&F(y_7),F(y_8)\big )
\end{split}}
\end{equation} where

(i) $F(x_1)=\{11,0\}_2$, $F(x_2)=\{11,1\}_2$, $F(x_3)=\{11,2\}_2$, $F(x_4)=\{11,3\}_2$, $F(x_5)=\{4\}$, $F(x_6)=\{8,5\}_2$;

(ii) $F(y_1)=\{5,6\}_1$, $F(y_2)=\{5,7\}_1$, $F(y_3)=\{4,8\}_3$, $F(y_4)=\{4,9\}_3$, $F(y_5)=\{4,10\}_3$, $F(y_6)=\{4,11\}_3$, $F(y_7)=\{0,12\}_1$, $F(y_8)=\{0,13\}_1$;

(iii) By the $W$-constraint Eq.(\ref{eqa:step-a-ee}), the edge colors are

$F(x_6y_1)=\{5\}\cup \{0,1,2,3\}$, $F(x_6y_2)=\{5\}\cup \{0,1,3\}$, $F(x_6y_3)=\{8\}\cup \{0,1,4,3\}$,

$F(x_5y_3)=\{4\}\cup \{0\}$, $F(x_5y_4)=\{4\}\cup \{0,5\}$, $F(x_5y_5)=\{4\}\cup \{0,6\}$, $F(x_5y_6)=\{4\}\cup \{0,7\}$,

$F(x_4y_6)=\{11\}\cup \{0,1,7,8\}$, $F(x_3y_6)=\{11\}\cup \{0,2,7,9\}$, $F(x_2y_6)=\{11\}\cup \{0,3,7,10\}$,

$F(x_1y_6)=\{11\}\cup \{0,4,7\}$, $F(x_1y_7)=\{0\}\cup \{1,11,12\}$, $F(x_1y_8)=\{0\}\cup \{2,11,13\}$.

\textbf{Fact-4. }By Eq.(\ref{eqa:VSET-coloring-algorithm-matrix-33}), we get a $(k,d)$-total set-coloring $F_{k,d}$ of the tree $T_6$ and the \emph{$(k,d)$-set-type Topcode-matrix} $S_{et}(T_6,F_{k,d})=(X^{et}_{k,d},E^{et}_{k,d},Y^{et}_{k,d})^T$ with
\begin{equation}\label{eqa:k-dset-coloring}
{
\begin{split}
X^{et}_{k,d}=&\big (F_{k,d}(x_6),F_{k,d}(x_6),F_{k,d}(x_6),F_{k,d}(x_5),F_{k,d}(x_5),F_{k,d}(x_5),F_{k,d}(x_5),F_{k,d}(x_4),F_{k,d}(x_3),\\
&F_{k,d}(x_2),F_{k,d}(x_1),F_{k,d}(x_1),F_{k,d}(x_1)\big )\\
E^{et}_{k,d}=&\big (F_{k,d}(x_6y_1),F_{k,d}(x_6y_2),F_{k,d}(x_6y_3),F_{k,d}(x_5y_3),F_{k,d}(x_5y_4),F_{k,d}(x_5y_5),F_{k,d}(x_5y_6),\\
&F_{k,d}(x_4y_6),F_{k,d}(x_3y_6),F_{k,d}(x_2y_6),F_{k,d}(x_1y_6),F_{k,d}(x_1y_7),F_{k,d}(x_1y_8)\big )\\
Y^{et}_{k,d}=&\big (F_{k,d}(y_1),F_{k,d}(y_2),F_{k,d}(y_3),F_{k,d}(y_3),F_{k,d}(y_4),F_{k,d}(y_5),F_{k,d}(y_6),F_{k,d}(y_6),F_{k,d}(y_6),\\
&F_{k,d}(y_6),F_{k,d}(y_6),F_{k,d}(y_7),F_{k,d}(y_8)\big )
\end{split}}
\end{equation} where the vertex set-$(k,d)$-colors and the edge set-$(k,d)$-colors are

(1) $F_{k,d}(x_1)=\{k+11d,0\}_2$, $F_{k,d}(x_2)=\{k+11d,d\}_2$, $F_{k,d}(x_3)=\{k+11d,2d\}_2$, $F_{k,d}(x_4)=\{k+11d,3d\}_2$, $F_{k,d}(x_5)=\{4d\}$, $F_{k,d}(x_6)=\{k+8d,5d\}_2$;

(2) $F_{k,d}(y_1)=\{5d,k+6d\}_1$, $F_{k,d}(y_2)=\{5d,k+7d\}_1$, $F_{k,d}(y_3)=\{4d,k+8d\}_3$, $F_{k,d}(y_4)=\{4d,k+9d\}_3$, $F_{k,d}(y_5)=\{4d,k+10d\}_3$, $F_{k,d}(y_6)=\{4d,k+11d\}_3$, $F_{k,d}(y_7)=\{0,k+12d\}_1$, $F_{k,d}(y_8)=\{0,k+13d\}_1$;

(3) By the $W$-constraint Eq.(\ref{eqa:step-a-ee}), the edge $(k,d)$-colors are

$F_{k,d}(x_6y_1)=\{5d\}\cup \{0,k+d,2d,k+3d\}$, $F_{k,d}(x_6y_2)=\{5d\}\cup \{0,k+2d,k+d,k+3d\}$,

$F_{k,d}(x_6y_3)=\{k+8d\}\cup \{0,d,k+3d,k+4d\}$, $F_{k,d}(x_5y_3)=\{4d\}\cup \{0\}$,

$F_{k,d}(x_5y_4)=\{4d\}\cup \{0,k+5d\}$, $F_{k,d}(x_5y_5)=\{4d\}\cup \{0,k+6d\}$,

$F_{k,d}(x_5y_6)=\{4d\}\cup \{0,k+7d\}$, $F_{k,d}(x_4y_6)=\{k+11d\}\cup \{0,d,k+7d,k+8d\}$,

$F_{k,d}(x_3y_6)=\{k+11d\}\cup \{0,2d,k+7d,k+9d\}$, $F_{k,d}(x_2y_6)=\{k+11d\}\cup \{0,3d,k+7d,k+10d\}$,

$F_{k,d}(x_1y_6)=\{k+11d\}\cup \{0,4d,k+7d\}$, $F_{k,d}(x_1y_7)=\{0\}\cup \{k+d,k+11d,k+12d\}$,

$F_{k,d}(x_1y_8)=\{0\}\cup \{2d,k+11d,k+13d\}$.

\textbf{Fact-5. }As the tree $T$ shown in Fig.\ref{fig:VSET-coloring-algorithm-2} (a) is selected as a \emph{topological public-key}, then the tree $T_6$ shown in Fig.\ref{fig:VSET-coloring-algorithm-2} (g) is just a \emph{topological private-key}. Thereby, the bytes of a number-based string $D_T$ induced from the Topcode-matrix $T_{code}(T,f)$ is shorter than that of a number-based string $D_{T_6}$ from the set-type Topcode-matrix $S_{et}(T_6,F)$ defined in Fact-3, or the $(k,d)$-set-type Topcode-matrix $S_{et}(T_6,F_{k,d})$ defined in Fact-4, since they are related with the ordered paths of the trees $T$ and $T_6$ according to the PWCSC-algorithm-A for ordered-path. \qqed
\end{example}

\begin{figure}[h]
\centering
\includegraphics[width=16.4cm]{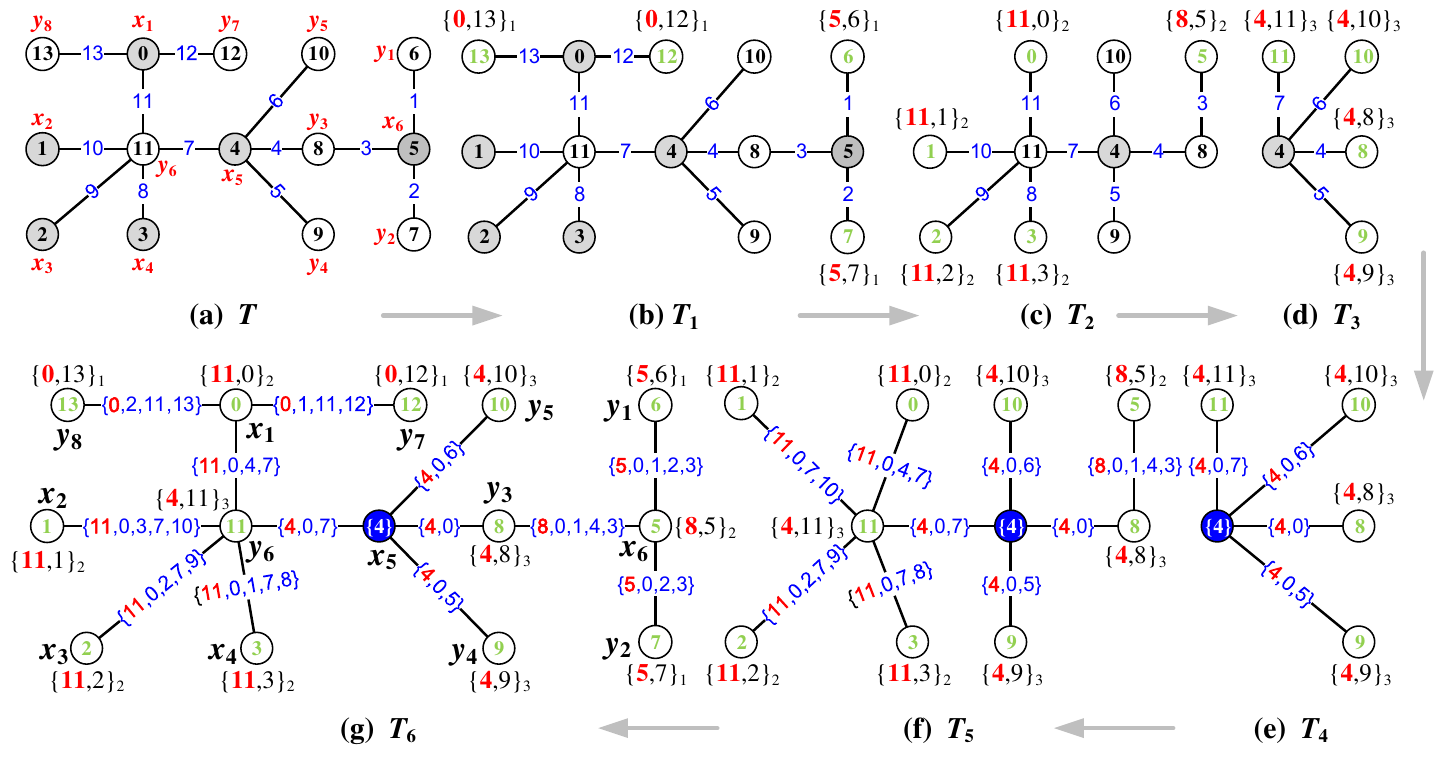}\\
\caption{\label{fig:VSET-coloring-algorithm-2}{\small An example for understanding the PWCSC-algorithm-A for ordered-path.}}
\end{figure}

\begin{thm}\label{thm:n-times-PWCSC-algorithm-A-ordered-path}
\cite{Bing-Yao-arXiv:2207-03381} After $n$ times of doing the PWCSC-algorithm-A for ordered-path to a tree $T$ admitting a set-ordered $W$-constraint labeling, we get a $W$-constraint set-coloring $F_n$ of the tree $T$ and $|F_n(u)\cap F_n(v)|\geq n=\lfloor \frac{m}{2}\rfloor$ for each edge $uv\in E(T)$ and $F_n(x)\neq F_n(y)$ for distinct vertices $x,y\in V(T)$, where $D(T)=m$ is the diameter of the tree $T$.
\end{thm}

See an example shown in Fig.\ref{fig:VSET-coloring-algorithm-3} for understanding Theorem \ref{thm:n-times-PWCSC-algorithm-A-ordered-path}. By Example \ref{exa:PWCSC-algorithm-A-ordered-path-11}, we present a theorem as follows:

\begin{figure}[h]
\centering
\includegraphics[width=16.4cm]{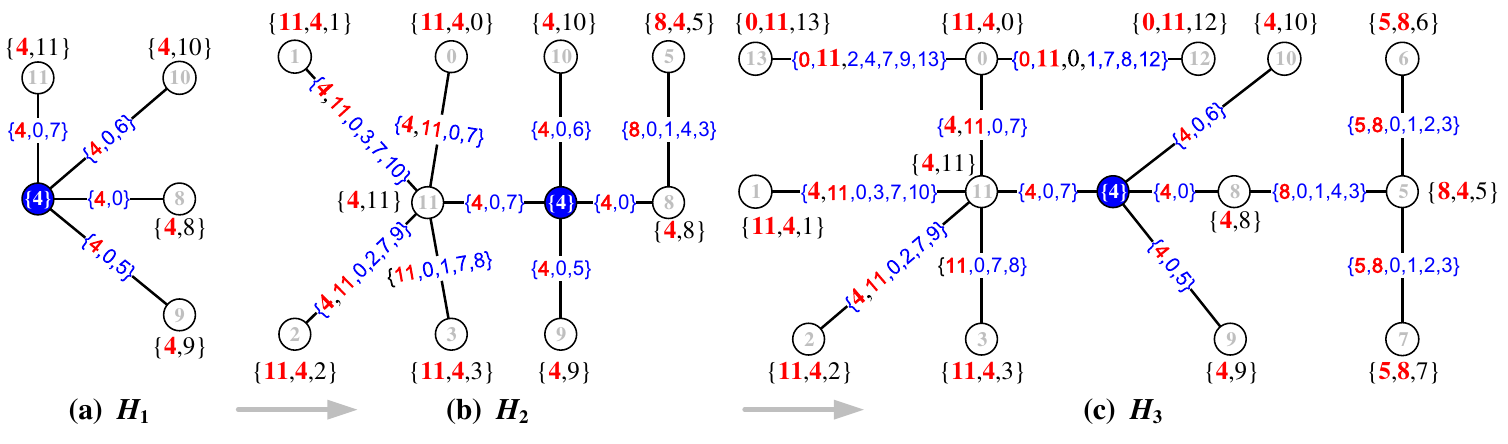}\\
\caption{\label{fig:VSET-coloring-algorithm-3}{\small The set-colored tree $H_3$ obtained by doing the PWCSC-algorithm-A for ordered-path to $D_6$ shown in Fig.\ref{fig:VSET-coloring-algorithm-2} (g).}}
\end{figure}

\begin{thm}\label{thm:set-ordered-w-cond-set-co}
\cite{Bing-Yao-arXiv:2207-03381} If a tree admits a set-ordered $W$-constraint labeling, then it admits:

(i) a $W$-constraint proper total set-coloring;

(ii) a $W$-constraint proper $(k,d)$-total coloring; and

(ii) a $W$-constraint proper $(k,d)$-total set-coloring;
\end{thm}

\vskip 0.4cm

\textbf{PWCSC-algorithm-B for level-leaf.}

\textbf{Initialization-B.} Suppose that $T$ is a tree admitting a $W$-constraint labeling $f$ holding $|f(V(T))|=|V(T)|$, and the induced edge color $f(uv)$ for each edge $uv\in E(T)$ holds the $W$-constraint $f(uv)=W\langle f(u),f(v)\rangle $.

\textbf{Step B-1.} Let $L(T)$ be the set of leaves of the tree $T$. Define a set-coloring $F_{le}$ for the tree $T$ as: $F_{le}(x)=\{f(x),f(u)\}$ for $x\in L(T)$ and $xu\in E(T)$.

\textbf{Step B-2.} Let $L(T_1)$ be the set of leaves of the tree $T_1$, where the tree $T_1=T-L(T)$. Color each leaf $y\in L(T_1)$ with $F_{le}(y)=\{f(y),f(v)\}$ if the edge $yv\in E(T)$.

\textbf{Step B-3.} After doing $k-1$ times Step B-2, if tree $T_k=T_{k-1}-L(T_{k-1})$ has its own diameter $D(T_k)\geq 3$ goto Step B-2. Otherwise, $T_k$ is a star $K_{1,n}$ with the vertex set $V(K_{1,n})=\{x_0, y_i:i\in [1,n]\}$ and the edge set $E(K_{1,n})=\{x_0y_i:i\in [1,n]\}$. Color each leaf $y_i\in L(T_k)$ with $F_{le}(y_i)=\{f(y_i),f(x_0)\}$ and $F_{le}(x_0)=\{f(x_0)\}$.

\textbf{Step B-4.} By the $W$-constraint we color each edge $uv\in E(T)$ with
\begin{equation}\label{eqa:w-condition-step-b-ee}
F_{le}(uv)=\big [F_{le}(u)\cap F_{le}(v)\big ]\cup \big \{W\langle a,b\rangle:~a\in F_{le}(u),~b\in F_{le}(v)\big \}
\end{equation}since $F_{le}(u)\cap F_{le}(v)\neq \emptyset$.

\textbf{Step B-5.} Return the $W$-constraint proper total set-coloring $F_{le}$ of the tree $T$ by the conditions in Initialization.

\vskip 0.2cm

\textbf{PWCSC-algorithm-C for neighbor-vertex.}

\textbf{Initialization-C.} Suppose that a tree $T$ admits a $W$-constraint labeling $f$ holding $|f(V(T))|=|V(T)|$, and the induced edge color $f(uv)$ for each edge $uv\in E(T)$ holds the $W$-constraint $f(uv)=W\langle f(u),f(v)\rangle $.

\textbf{Step C-1.} Define a total set-coloring $F_{nv}$ as: $F_{nv}(x)=\{f(y):y\in N_{ei}(x)\}$ for each $x\in V(T)$. Clearly, $F_{nv}(x)\neq F_{nv}(u)$ for distinct $x,u\in V(T)$, since $|f(V(T))|=|V(T)|$.

\textbf{Step C-2.} By the $W$-constraint we color the edges of the tree $T$ as
\begin{equation}\label{eqa:w-condition-step-c2-ee}
F_{nv}(uv)=\big [F_{nv}(u)\cap F_{nv}(v)\big ]\cup \big \{W\langle a,b\rangle:~a\in F_{nv}(u),~b\in F_{nv}(v)\big \},~uv\in E(T)
\end{equation}since $F_{nv}(u)\cap F_{nv}(v)\neq \emptyset$.

\textbf{Step C-3.} Return the $W$-constraint proper total set-coloring $F_{nv}$ of the tree $T$.

\vskip 0.2cm

\textbf{PWCSC-algorithm-D for neighbor-edge.}

\textbf{Initialization-D.} Suppose that $T$ is a tree admitting a $W$-constraint total labeling $f$ holding $f(uv)\neq f(xy)$ for any pair of edges $uv$ and $xy$ of the tree $T$, and each edge $uv\in E(T)$ holds the $W$-constraint $f(uv)=W\langle f(u),f(v)\rangle $.

\textbf{Step D-1.} Define a total set-coloring $F_{ne}$ by setting $F_{ne}(x)=\{f(xz):z\in N_{ei}(x)\}$ for each $x\in V(T)$, clearly, $F_{ne}(x)\neq F_{ne}(u)$ for distinct vertices $x,u\in V(T)$, since $|f(E(T))|=|E(T)|$.

\textbf{Step D-2.} By the $W$-constraint we color the edges of the tree $T$ as
\begin{equation}\label{eqa:w-condition-step-c2-ee}
F_{ne}(uv)=\big [F_{ne}(u)\cap F_{ne}(v)\big ]\cup \big \{W\langle a,b\rangle:~a\in F_{ne}(u),~b\in F_{ne}(v)\big \},~uv\in E(T)
\end{equation}since $F_{ne}(u)\cap F_{ne}(v)\neq \emptyset$.

\textbf{Step D-3.} Return the $W$-constraint proper total set-coloring $F_{ne}$ of the tree $T$.

\vskip 0.2cm

\textbf{PWCSC-algorithm-E for neighbor-edge-vertex.}

\textbf{Initialization-E.} Suppose that $T$ is a tree admitting a $W$-constraint total labeling $f$ holding $f(uv)\neq f(xy)$ for any pair of distinct edges $uv,xy\in E(T)$, and each edge $uv\in E(T)$ holds the $W$-constraint $f(uv)=W\langle f(u),f(v)\rangle $.

\textbf{Step E-1.} Define a total set-coloring $F_{nve}$ by setting
$$
F_{nve}(x)=\{f(y):y\in N_{ei}(x)\}\cup \{f(xz):z\in N_{ei}(x)\},~x\in V(T)
$$ Clearly, $F_{nve}(x)\neq F_{nve}(u)$ for distinct vertices $x,u\in V(T)$, since $|f(E(T))|=|E(T)|$.

\textbf{Step E-2.} By the $W$-constraint we color the edges of the tree $T$ as
\begin{equation}\label{eqa:w-condition-step-c2-ee}
F_{nve}(uv)=\big [F_{nve}(u)\cap F_{nve}(v)\big ]\cup \big \{W\langle a,b\rangle:~a\in F_{nve}(u),~b\in F_{nve}(v)\big \},~uv\in E(T)
\end{equation}since $F_{nve}(uv) \subseteq F_{nve}(u)\cap F_{nve}(v)$.

\textbf{Step E-3.} Return the $W$-constraint proper total set-coloring $F_{nve}$ of the tree $T$.

\subsubsection{Graph homomorphisms with parameterized set-colorings}

Since a connected non-tree $(p,q)$-graph $G$ can be vertex-split into a tree $T$ of $q+1$ vertices by doing the vertex-splitting tree-operation once time, so we have a set $T_{ree}(G)$ of trees obtained from vertex-splitting $G$, such that each tree $T\in T_{ree}(G)$ is graph homomorphism into $G$, that is $T\rightarrow G$. So, we can use a connected non-tree $(p,q)$-graph $G$ and its tree set $T_{ree}(G)$ in asymmetric cryptosystem, and make topological number-based strings generated from the graph $G$ and the tree set $T_{ree}(G)$.

\vskip 0.4cm

\textbf{Situation-A.} Suppose that a connected non-tree $(p,q)$-graph $G$ is not colored by any coloring. Since each tree $T\in T_{ree}(G)$ admits each one of the colorings: graceful $(k,d)$-total coloring, harmonious $(k,d)$-total coloring, (odd-edge) edge-magic $(k,d)$-total coloring, (odd-edge) graceful-difference $(k,d)$-total coloring, (odd-edge) edge-difference $(k,d)$-total coloring, (odd-edge) felicitous-difference $(k,d)$-total coloring and edge-antimagic $(k,d)$-total coloring introduced in Definition \ref{defn:kd-w-type-colorings} and \cite{Yao-Su-Ma-Wang-Yang-arXiv-2202-03993v1}.

Suppose that a tree $T\in T_{ree}(G)$ admits a $W$-constraint coloring/labeling $f$. By the PWCSC-algorithms introduced above, the tree $T$ admits a $W$-constraint set-coloring $F_f$ induced from the $W$-constraint coloring/labeling $f$.

Using the graph homomorphism $T\rightarrow G$ defined on a vertex coloring $\varphi:V(T)\rightarrow V(G)$, such that $\varphi(u)\varphi(v)\in E(G)$ for each edge $uv\in E(T)$, thereby, this graph $G$ admits a $W$-constraint set-coloring $F^*_f$ defined as:

(i) $F^*_f(w)=\{F_f(x): \varphi(x)=w,~x\in V(T)\}$ for $w\in V(G)$, and

(ii) $F^*_f(\varphi(u)\varphi(v))=F_f(uv)$ for $uv \in E(T)$.

\vskip 0.2cm

\textbf{Analysis of complexity of Situation-A}:
\begin{asparaenum}[\textbf{\textrm{Complexity-A.}}1.]
\item \textbf{Determining} the tree set $T_{ree}(G)$ obtained by vertex-splitting the connected non-tree $(p,q)$-graph $G$ into trees will meet the Subgraph Isomorphic Problem, although each tree $T\in T_{ree}(G)$ has exactly $q$ edges, since there are 279,793,450 trees of $26$ vertices and 5,759,636,510 rooted trees of $26$ vertices.
\item There is no way to know \textbf{how many} $W$-constraint colorings/labelings admitted by each tree $T\in T_{ree}(G)$, and moreover \textbf{no algorithm} can realize all colorings/labelings holding a fixed $W$-constraint for each tree $T\in T_{ree}(G)$, since it is sharp-P-hard.
\end{asparaenum}

\vskip 0.2cm

\textbf{Situation-B.} Suppose that a connected non-tree $(p,q)$-graph $G$ admits a $W$-constraint coloring/labeling $g$, so each tree $H\in T_{ree}(G)$ admits a $W$-constraint coloring/labeling $g^*$ induced by the $W$-constraint coloring/labeling $g$.

Since there is a vertex coloring $\theta:V(H)\rightarrow V(G)$ with $\theta(u)\theta(v)\in E(G)$ for each edge $uv\in E(H)$, the \emph{colored graph homomorphism} $H\rightarrow G$ means a \emph{Topcode-matrix homomorphism}
\begin{equation}\label{eqa:topcode-matrix-homomorphism}
T_{code}(H,g^*)\rightarrow T_{code}(G,g)
\end{equation}Thereby, we have two graph sets $S_{graph}[T_{code}(H,g^*)]$ and $S_{graph}[T_{code}(G,g)]$, such that

(a) Each graph $J\in S_{graph}[T_{code}(H,g^*)]$ admits a $W$-constraint coloring/labeling $f_J$ holding $T_{code}(J,f_J)=T_{code}(H,g^*)$; and

(b) each graph $I\in S_{graph}[T_{code}(G,g)]$ holds $T_{code}(I,f_I)=T_{code}(G,g)$ for a $W$-constraint coloring/labeling $f_I$ admitted by $I$.

Hence, we get a \emph{graph-set homomorphism}
\begin{equation}\label{eqa:graph-set-homomorphism}
S_{graph}[T_{code}(H,g^*)]\rightarrow S_{graph}[T_{code}(G,g)]
\end{equation}

\vskip 0.2cm

\textbf{Analysis of complexity of Situation-B}:

\begin{asparaenum}[\textbf{\textrm{Complexity-B.}}1.]
\item \textbf{Determining} two graph sets $S_{graph}[T_{code}(H,g^*)]$ and $S_{graph}[T_{code}(G,g)]$ will meet the Subgraph Isomorphic Problem, a NP-hard problem.
\item If a graph $I\in S_{graph}[T_{code}(G,g)]$ is a \emph{public-key}, \textbf{no algorithm} is for finding a \emph{private-key} $J\in S_{graph}[T_{code}(H,g^*)]$, such that $J\rightarrow I$ is just a colored graph homomorphism.
\end{asparaenum}

\begin{figure}[h]
\centering
\includegraphics[width=16.4cm]{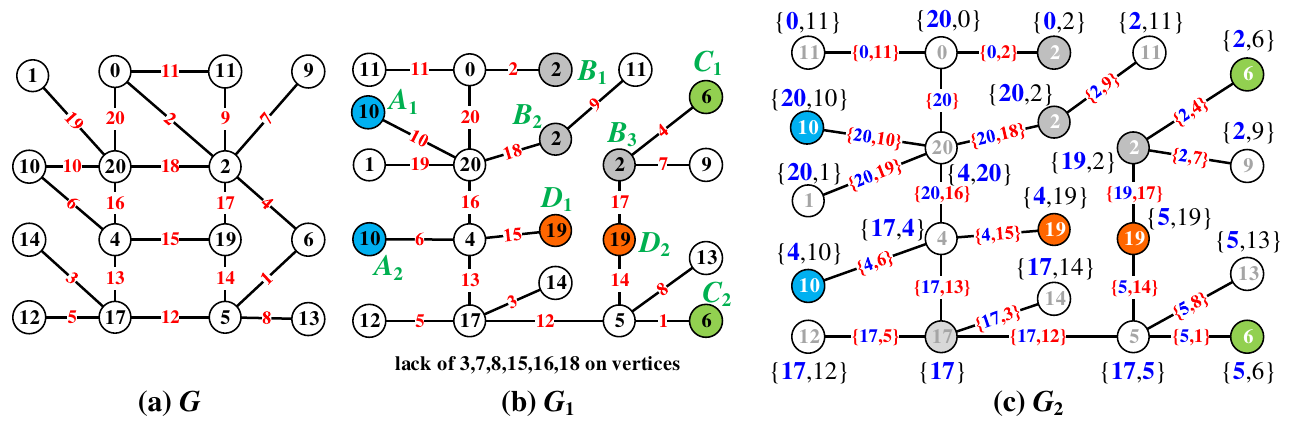}\\
\caption{\label{fig:PWCSC-algorithm-cycle-11}{\small An example for understanding the Situation-B, cited from \cite{Yao-Ma-arXiv-2201-13354v1}.}}
\end{figure}

\begin{figure}[h]
\centering
\includegraphics[width=13.4cm]{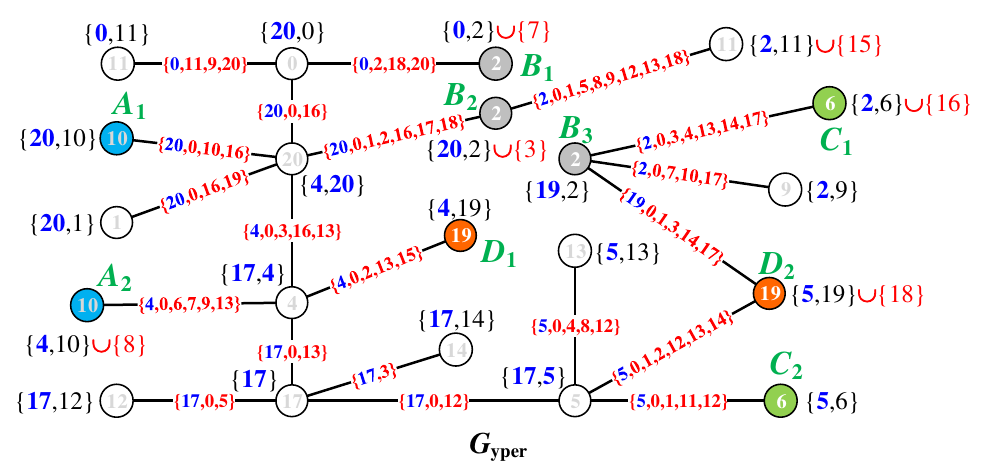}\\
\caption{\label{fig:PWCSC-algorithm-cycle-22}{\small Another example for understanding the Situation-B, cited from \cite{Yao-Ma-arXiv-2201-13354v1}.}}
\end{figure}

\subsubsection{Normal set-colorings based on hyperedge sets}

\begin{defn} \label{defn:tradition-vs-set-colorings}
\cite{Bing-Yao-arXiv:2207-03381} Let $G$ be a $(p,q)$-graph, and let $\Lambda$ be a finite set of numbers. There is a \emph{set-coloring} $F: S\rightarrow \mathcal{E}$ for a \emph{hyperedge set} $\mathcal{E}\subseteq \Lambda^2$, where $S\subseteq V(G)\cup E(G)$. There are the following constraints conditions:
\begin{asparaenum}[\textbf{\textrm{Nset}}-1.]
\item \label{old:vertex} $S=V(G)$.
\item \label{old:edge} $S=E(G)$.
\item \label{old:total} $V(G)\cup E(G)$.
\item \label{old:adjacent-vertex} $F(u)\neq F(v)$ for each edge $uv\in E(G)$.
\item \label{old:adjacent-edge} $F(uv)\neq F(uw)$ for adjacent edges $uv,uw\in E(G)$, where $w\in N_{ei}(u)$ and $u\in V(G)$.
\item \label{old:incident-edge-vertex} $F(u)\neq F(uv)$ and $F(v)\neq F(uv)$ for each edge $uv\in E(G)$.
\item \label{old:hyperedge-set} $\Lambda=\bigcup _{e\in \mathcal{E}}e$.
\item \label{old:join-oper-verticeice} $F(u)\cap F(v)\neq \emptyset$ for each edge $uv\in E(G)$.
\item \label{old:join-oper-dajacent-edges} $F(uv)\cap F(uw)\neq \emptyset$ for adjacent edges $uv,uw\in E(G)$.
\item \label{old:join-oper-vertex-edge} $F(uv)\cap F(u)\neq \emptyset$ and $F(uv)\cap F(v)\neq \emptyset$ for each edge $uv\in E(G)$.
\end{asparaenum}
\textbf{Then $F$ is}

\noindent --------- \emph{traditional set-colorings}

\begin{asparaenum}[\textbf{\textrm{Setc}}-1.]
\item a \emph{proper set-coloring} if Setcond-\ref{old:vertex} and Setcond-\ref{old:adjacent-vertex} hold true.
\item a \emph{proper edge set-coloring} if Setcond-\ref{old:edge} and Setcond-\ref{old:adjacent-edge} hold true.
\item a \emph{proper total set-coloring} if Setcond-\ref{old:total}, Setcond-\ref{old:adjacent-vertex}, Setcond-\ref{old:adjacent-edge} and Setcond-\ref{old:incident-edge-vertex} hold true.

\noindent --------- \emph{hyperedge set-colorings}

\item a \emph{proper hyperedge set-coloring} if Setcond-\ref{old:vertex}, Setcond-\ref{old:adjacent-vertex} and Setcond-\ref{old:hyperedge-set} hold true.
\item a \emph{proper edge hyperedge set-coloring} if Setcond-\ref{old:edge}, Setcond-\ref{old:adjacent-edge} and Setcond-\ref{old:hyperedge-set} hold true.
\item a \emph{proper total hyperedge set-coloring} if Setcond-\ref{old:total}, Setcond-\ref{old:adjacent-vertex}, Setcond-\ref{old:adjacent-edge}, Setcond-\ref{old:incident-edge-vertex} and Setcond-\ref{old:hyperedge-set} hold true.

\noindent --------- \emph{intersected set-colorings}

\item a \emph{proper intersected set-coloring} if Setcond-\ref{old:vertex}, Setcond-\ref{old:adjacent-vertex} and Setcond-\ref{old:join-oper-verticeice} hold true.
\item a \emph{proper edge intersected set-coloring} if Setcond-\ref{old:edge}, Setcond-\ref{old:adjacent-edge} and Setcond-\ref{old:join-oper-dajacent-edges} hold true.
\item a \emph{proper totally intersected set-coloring} if Setcond-\ref{old:total}, Setcond-\ref{old:adjacent-vertex}, Setcond-\ref{old:adjacent-edge}, Setcond-\ref{old:incident-edge-vertex}, Setcond-\ref{old:join-oper-verticeice}, Setcond-\ref{old:join-oper-dajacent-edges} and Setcond-\ref{old:join-oper-vertex-edge} hold true.

\noindent --------- \emph{intersected hyperedge set-colorings}

\item a \emph{proper intersected hyperedge set-coloring} if Setcond-\ref{old:vertex}, Setcond-\ref{old:adjacent-vertex}, Setcond-\ref{old:hyperedge-set} and Setcond-\ref{old:join-oper-verticeice} hold true.
\item a \emph{proper edge intersected hyperedge set-coloring} if Setcond-\ref{old:edge}, Setcond-\ref{old:adjacent-edge}, Setcond-\ref{old:hyperedge-set} and Setcond-\ref{old:join-oper-dajacent-edges} hold true.
\item a \emph{proper all-intersected hyperedge set-coloring} if Setcond-\ref{old:total}, Setcond-\ref{old:adjacent-vertex}, Setcond-\ref{old:adjacent-edge}, Setcond-\ref{old:incident-edge-vertex}, Setcond-\ref{old:hyperedge-set}, Setcond-\ref{old:join-oper-verticeice}, Setcond-\ref{old:join-oper-dajacent-edges} and Setcond-\ref{old:join-oper-vertex-edge} hold true.\qqed
\end{asparaenum}
\end{defn}

\begin{problem}\label{qeu:444444}
For a $(p,q)$-graph $G$ appeared in Definition \ref{defn:tradition-vs-set-colorings}, \textbf{consider} the following extremum questions:
\begin{asparaenum}[\textbf{\textrm{Eque}}-1.]
\item \textbf{Find} the \emph{extremum set-number} $\Lambda_{ex}(G)=\min_F \{|\Lambda|\}$ over all proper set-colorings.
\item \textbf{Find} the \emph{extremum set-index} $\Lambda\,'_{ex}(G)=\min_F \{|\Lambda|\}$ over all proper edge set-colorings.
\item \textbf{Find} the \emph{extremum total set-number} $\Lambda\,''_{ex}(G)=\min_F \{|\Lambda|\}$ over all proper total set-colorings.
\end{asparaenum}
\end{problem}

\begin{problem}\label{qeu:444444}
\textbf{Consider} the set-colorings defined in Definition \ref{defn:tradition-vs-set-colorings} based on the following cases:
\begin{asparaenum}[\textbf{\textrm{Case}}-1.]
\item Each hyperedge $e\in \mathcal{E}$ has its own cardinality $|e|\geq 2$.
\item Any two hyperedges $e,e\,'\in \mathcal{E}$ hold $e\not \subset e\,'$ and $e\,'\not \subset e$.
\item Each hyperedge $e\in \mathcal{E}$ holds $|e|=k\geq 2$, so the corresponding to each set-coloring is $k$-uniformed.
\item The sequence $\{|e_i|\}$ for $\mathcal{E}=\{e_i:i\in [1,n]\}$ forms a series, such as arithmetic progression, geometric series, Fibonacci series, \emph{etc}.
\end{asparaenum}
\end{problem}

\begin{defn} \label{defn:normai-set-coloring-intersected-graph}
\cite{Bing-Yao-arXiv:2207-03381} (A-1) If a graph $G$ admits a \emph{proper intersected hyperedge set-coloring} defined in Definition \ref{defn:tradition-vs-set-colorings}, and each pair of hyperedges $e,e\,' \in \mathcal{E}$ holding $e\cap e\,'\neq \emptyset$ corresponds to an edge $xy\in E(G)$ with $F(x)=e$ and $F(y)=e\,'$, then we call $G$ a \emph{vertex-intersected graph} of the hypergraph $\mathcal{H}_{yper}=(\Lambda,\mathcal{E})$ defined in \cite{Jianfang-Wang-Hypergraphs-2008}.

(A-2) Suppose that $F$ is the proper intersected hyperedge set-coloring of a vertex-intersected graph $G$ of some hypergraph $\mathcal{H}_{yper}=(\Lambda,\mathcal{E})$, then a vertex-intersected graph $G$ admits an \emph{edge-intersected total set-coloring} $F^*$ defined by $F^*(u)=F(u)$ for each vertex $u\in V(G)$, and $F^*(xy)=F(x)\cap F(y)$ for each edge $xy\in E(G)$.

(B-1) If a graph $G$ admits a \emph{proper edge intersected hyperedge set-coloring} defined in Definition \ref{defn:tradition-vs-set-colorings}, such that any two hyperedges $e,e\,' \in \mathcal{E}$ holding $e\cap e\,'\neq \emptyset$ correspond to two edges $xy,xw\in E(G)$ with $F(xy)=e$ and $F(xw)=e\,'$, then we call $G$ an \emph{edge-intersected graph} of the hypergraph $\mathcal{H}_{yper}=(\Lambda,\mathcal{E})$.

(B-2) Suppose that $\varphi$ is the proper edge intersected hyperedge set-coloring of the edge-intersected graph $G$ of some hypergraph $\mathcal{H}_{yper}=(\Lambda,\mathcal{E})$, then the edge-intersected graph $G$ admits a \emph{vertex-intersected total set-coloring} $\varphi^*$ defined as: $\varphi^*(u)=\{F(uv)\cap F(uw):~v,w\in N_{ei}(u)\}$ for each vertex $u\in V(G)$, and each edge $xy\in E(G)$ is colored with $\varphi^*(xy)=F(xy)$.\qqed
\end{defn}

For a set $S\,^i_X=\{a\,^i_1,a\,^i_2,\dots , a\,^i_{s(i)}\}$ with integer $s(i)\geq 1$ and $i\in [1,n]$ and integers $k,d\geq 0$, we define the following two operations:
\begin{equation}\label{eqa:555555}
{
\begin{split}
d\cdot \langle S\,^i_X\rangle =&\big \{d\cdot a\,^i_1,d\cdot a\,^i_2,\dots , d\cdot a\,^i_{s(i)}\big \}\\
k[+]d \langle S\,^i_X\rangle =&\big \{k+d\cdot a\,^i_1,k+d\cdot a\,^i_2,\dots , k+d\cdot a\,^i_{s(i)}\big \}
\end{split}}
\end{equation}
And moreover, for a set $S^*=\{S\,^1_X,S\,^2_X,\dots ,S\,^n_X\}$, we have two operations:

(i) $d \langle S^*\rangle =\{d\langle S\,^1_X\rangle, d\langle S\,^2_X\rangle,\dots , d\langle S\,^n_X\rangle\}$;

(ii) $k[+]d\langle S^*\rangle =\big \{k[+]d\langle S\,^1_X\rangle, k[+]d\langle S\,^2_X\rangle,\dots , k[+]d\langle S\,^n_X\rangle\big \}$.

\begin{defn} \label{defn:normai-kd-type-set-coloring}
\cite{Bing-Yao-arXiv:2207-03381} A connected bipartite $(p,q)$-graph $G$ has its own vertex set bipartition $V(G)=X\cup Y$ with $X\cap Y=\emptyset$ and admits a $W$-constraint set-coloring $F$ defined in Definition \ref{defn:tradition-vs-set-colorings}, then $G$ has its own set-type Topcode-matrix $T_{code}(G,F)=(X_S,E_S,Y_S)^T$, where each element in three \emph{set vectors} $X_S,E_S$ and $Y_S$ is a \emph{set}. So we have a set-type parameterized Topcode-matrix
\begin{equation}\label{eqa:66666666666}
\centering
{
\begin{split}
P^{set}_{ara}(G,\Phi|k,d)=k\cdot I\,^0[+]d\cdot T_{code}(G,F)=(d\cdot X_S,~k[+]d\cdot E_S,~k[+]d\cdot Y_S)^T
\end{split}}
\end{equation} which defines a $(k,d)$-type set-coloring
\begin{equation}\label{eqa:555555}
\Phi:S\rightarrow \mathcal{E}^P,~S\subseteq V(G)\cup E(G),~\mathcal{E}^P\subseteq \Lambda^2_{(m,b,n,k,a,d)}
\end{equation} for the connected bipartite $(p,q)$-graph $G$.\qqed
\end{defn}

\subsection{Pan-operation graphs of hypergraphs}

We will extend the vertex-intersected graphs of hypergraphs defined in Definition \ref{defn:vertex-intersected-graph-hypergraph} to general situations in this subsection.

\begin{defn} \label{defn:pan-operation-graph-a-hypergraph}
$^*$ \textbf{Pan-operation graphs of hypergraphs.} Let $[\bullet]$ be an operation on sets, and let $\mathcal{E}$ be a hyperedge set defined on a finite set $\Lambda$ such that each hyperedge $e\in \mathcal{E}$ corresponds another hyperedge $e\,'\in \mathcal{E}$ holding $e[\bullet] e\,'$ to be one subset of the power set $\Lambda^2$, and $\Lambda=\bigcup _{e\in \mathcal{E}}e$, as well as $R_{est}(c_0,c_1,c_2,\dots ,c_m)$ be a constraint set with $m\geq 0$, in which the first constraint $c_0:= e[\bullet] e\,'\in \Lambda^2$. If a graph $G$ admits a total set-coloring $\theta: V(H)\cup E(H)\rightarrow \mathcal{E}$ holding $\theta(x)\neq \theta(y)$ for each edge $uv\in E(G)$, and

(i) the first constraint $c_0: ~\theta(uv)\supseteq \theta(u)[\bullet]\theta(v)\neq \emptyset$.

(ii) For some $k\in [1,m]$, there is a function $\pi_k$, such that the $k$th constraint $c_k:=\pi_k (b_{u},b_{uv},b_{v})=0$ with $k\in [1,m]$ for $b_{uv}\in \theta(uv)$, $b_{u}\in \theta(u)$ and $b_{v}\in \theta(v)$.

(iii) If a pair of hyperedges $e_i,e_j\in \mathcal{E}$ holding $e_i[\bullet] e_j\in \Lambda^2$, then there is an edge $x_iy_j$ of the graph $G$, such that $\theta(x_i)=e_i$ and $\theta(y_j)=e_j$.

We call $G$ \emph{pan-operation graph} of the hypergraph $\mathcal{H}_{yper}=(\Lambda,\mathcal{E})$. \qqed
\end{defn}

\begin{example}\label{exa:8888888888}
Let $\mathcal{E}$ be a hyperedge set defined on a consecutive integer set $\Lambda=[1,q]$. For a fixed integer $k>0$, there exists a number $c_{i,j}\in e_{i,j}\in \mathcal{E}$, such that the operation $e_i[\bullet]e_j:=\varphi(a_i,c_{i,j},b_j)=k$ for some numbers $a_i\in e_i$ and $b_j\in e_j$, where $\varphi(a_i,c_{i,j},b_j)=k$ is one of $a_i+c_{i,j}+b_j=k$, $c_{i,j}+|a_i-b_j|=k$, $|a_i+b_j-c_{i,j}|=k$ and $\big ||a_i-b_j|-c_{i,j}\big |=k$, refer to Definition \ref{defn:kd-w-type-colorings}.

For a fixed integer $k>0$, another hyperedge set $\mathcal{E}^*$ defined on a consecutive integer set $\Lambda=[1,q]$ holds $|e|\geq 3$ for each hyperedge $e\in \mathcal{E}^*$. We have sets $\varphi(e,k)=\{a_i,b_i,c_i\in e:\varphi(a_i,c_i,b_i)=k\}$ for each hyperedge $e\in \mathcal{E}^*$. Each hyperedge $e\in \mathcal{E}^*$ corresponds another hyperedge $e\,'\in \mathcal{E}^*$, such that the operation $e_i[\bullet]e_j$ to hold $\varphi(e,k)\cap \varphi(e\,',k)\neq \emptyset$.\qqed
\end{example}

\begin{example}\label{exa:8888888888}
Let $G$ be a bipartite $(p,q)$-graph with vertex set bipartition $(X,Y)$, and let $\mathcal{E}$ be a hyperedge set defined on a consecutive integer set $\Lambda=[1,q]$. The graph $G$ admits a set-coloring $F:V(G)\rightarrow \mathcal{E}$, such that
\begin{equation}\label{eqa:555555}
\max\{\max |F(x)|:~x\in X\}<\min\{\max |F(y)|:~y\in Y\}
\end{equation}
And the operation $e_i[\bullet]e_j$ means that each edge $xy\in E(G)$ corresponds some hyperedge $e_{x,y}\in \mathcal{E}$ containing a number $c_{x,y}=a_y-b_x>0$ for some numbers $a_y\in F(y)$ and $b_x\in F(x)$. Moreover, if the set $\{c_{x,y}\in e_{x,y}:xy\in E(G),e_{x,y}\in \mathcal{E}\}=[1,q]$, we call $F$ a \emph{set-ordered e-graceful set-coloring} based on the hyperedge set $\mathcal{E}$.\qqed
\end{example}

\begin{defn}\label{defn:every-zero-total-graphic-group}
\cite{Yao-Wang-2106-15254v1} Let $G_1,G_2,\dots, G_{p+q}$ be the copies of a $(p,q)$-graph $G$ admitting a magic total labeling $h:V(G)\cup E(G)\rightarrow [1,p+q]$ such that $h(u)+h(v)=k+h(uv)$ for each edge $uv\in E(G)$. Each graph $G_i$ with $i\in [1,p+q]$ admits a magic total labeling $h_i:V(G_i)\cup E(G_i)\rightarrow [1,p+q]$ holding
\begin{equation}\label{eqa:555555}
h_i(x)+h_i(y)=k_i+h_i(xy)~(\textrm{mod}~p+q)
\end{equation} true for each edge $xy\in E(G_i)$, and the index set $\{k_1,k_2,\dots ,k_{p+q}\}=[1,p+q]$, and we write $F_{p+q}(G,h)=\{G_1,G_2$, $\dots$, $G_{p+q}\}$. For any preappointed \emph{zero} $G_k\in F_{p+q}(G,h)$, we have
\begin{equation}\label{eqa:graphic-group-1111}
h_i(w)+h_i(w)-h_k(w)=h_{\lambda}(w)
\end{equation}
with $\lambda=i+j-k~(\bmod~p+q)$ for each $w\in V(G)\cup E(G)$, and we call the set $F_{p+q}(G,h)$ \emph{every-zero total graphic group}, and rewrite it as $\{F_{p+q}(G,h);[+][-]\}$.\qqed
\end{defn}

\begin{defn}\label{defn:graphic-group-definition}
\cite{Yao-Zhang-Sun-Mu-Sun-Wang-Wang-Ma-Su-Yang-Yang-Zhang-2018arXiv} An \emph{every-zero graphic group} $\{F_{n}(H);[+][-]\}$ made by a colored graph $H$ admitting a $\varepsilon$-labeling $h$ contains its own elements $H_i\in F_{n}(H)=\{H_i:i\in [1,n]\}$ admitting a $\varepsilon$-labeling $h_i$ induced by $h$ with $i\in [1,n]$ for $n\geq q$ ($n\geq 2q-1$), and hold the finite module Abelian additive operation
\begin{equation}\label{eqa:graphic-group-Abelian-additive-operation}
H_i[+_k]H_j:=H_i[+]H_j[-]H_k
\end{equation}
defined as
\begin{equation}\label{eqa:graphic-group-definition}
h_i(x)+h_j(x)-h_k(x)=h_{\lambda}(x),\quad x\in V(H)
\end{equation}
with $\lambda=i+j-k\,(\bmod\,n)$ for any preappointed \emph{zero} $H_k\in F_{n}(H)$.\qqed
\end{defn}

\begin{example}\label{exa:graphic-group-example}
Refer to Definition \ref{defn:every-zero-total-graphic-group} and Definition \ref{defn:graphic-group-definition}. Let $\{F(G);[+][-]\}$ be an \emph{every-zero graphic group} based on a graph set $F(G)=\{G_1,G_2,\dots ,G_q\}$ obtained by a $(p,q)$-graph $G_1$ admitting a graceful labeling $f_1$, where each graph $G_i$ holds $G_i\cong G_1$ and admits a labeling $f_i$ defined as $f_i(x)=f_1(x)+i-1~(\bmod~q)$ for $x\in V(G_i)=V(G_1)$ and $f_i(xy)=|f_i(x)-f_i(y)|$ for $xy\in E(G_i)=E(G_1)$ with $i\in [1,q]$. Selecting arbitrarily any graph $G_k\in F(G)$ as the preappointed \emph{zero}, the finite module Abelian additive operation
\begin{equation}\label{eqa:Abelian-additive-operation-graphic group}
G_i[+_k]G_j:=G_i[+]G_j[-]G_k
\end{equation}
is defined by
\begin{equation}\label{eqa:555555}
f_i(x)+f_j(x)-f_k(x)=f_{\beta}(x),\quad x\in V(G_i)=V(G_j)=V(G_k)=V(G_1)
\end{equation}
with $\beta=i+j-k~(\bmod~q)$ for any preappointed \emph{zero} $G_k\in F(G)$. It is not hard to show that $\{F(G);[+][-]\}$ holds Zero, Inverse, Uniqueness, Closure and Associative law.\qqed
\end{example}

\begin{example}\label{exa:8888888888}
Let $\mathcal{E}$ be a hyperedge set based on a finite set $\Lambda=F(G)$ shown in Example \ref{exa:graphic-group-example} and the finite module Abelian additive operation Eg.(\ref{eqa:Abelian-additive-operation-graphic group}).

\textbf{Case 1.} Suppose that a $(p\,',q\,')$-graph $H$ admits an \emph{e-index-graceful set-labeling} $F:V(H)\rightarrow \mathcal{E}$, and the induced edge color $F(uv)=\{G_\beta\}$ obtained by $G_i[+_k]G_j=G_\beta$ with $G_i\in F(u)$ and $G_j\in F(v)$ and $\beta=i+j-k~(\bmod~q)$, such that the index set $\{r:\{G_r\}=F(uv), uv\in E(H)\}=[1,q\,']$. It is noticeable, the operation $G_i[+_k]G_j$ means the operation $e_i[\bullet]e_j$ defined in Definition \ref{defn:pan-operation-graph-a-hypergraph}, since $e_i=F(u)$ and $e_j=F(v)$.

\textbf{Case 2.} Suppose that a $(p\,',q\,')$-graph $H$ admits a total set-labeling $F^*:V(H)\cup E(H)\rightarrow \mathcal{E}$, for each edge $uv\in E(H)$, there are $G_i\in F^*(u)$, $G_j\in F^*(v)$ and $G_\beta\in F^*(uv)$ holding $G_i[+_k]G_j=G_\beta$ with $\beta=i+j-k~(\bmod~q)$, Here, the operation $e_i[\bullet]e_j$ defined in Definition \ref{defn:pan-operation-graph-a-hypergraph} is that $G_i[+_k]G_j=G_\beta\in F^*(uv)$ based on $e_i=F(u)$ and $e_j=F(v)$. However, $F(uv)\neq F^*(uv)$ for each edge $uv\in E(H)$ in general, where the labeling $F$ is defined in the above \textbf{Case 1}.\qqed
\end{example}

\begin{problem}\label{question:444444}
A connected $(p,q)$-graph $G$ can be vertex-split into trees $T_1,T_2,\dots, T_m$ of $q+1$ vertices. Suppose that each tree $T_i$ admits different proper total colorings $f_{i,1},f_{i,2},\dots, f_{i,a_i}$ with $i\in [1,m]$ and $a_i\geq 1$, and let $T_{i,j}$ be the tree obtained by coloring totally the vertices and edges of the each tree $T_i$ with the proper total coloring $f_{i,j}$ with $j\in [1,a_i]$ and $i\in [1,m]$. So, vertex-coinciding each colored tree $T_{i,j}$ produces the original connected $(p,q)$-graph $G$ admitting a total coloring $F_{i,j}$, there are the following cases:

(i) $F_{i,j}(uv)$ for each edge $uv\in E(G)$ is a number, however it may happen $F_{i,j}(uv)=F_{i,j}(uw)$ for $v,w\in N_{ei}(u)$ for some vertex $u\in V(G)$;

(ii) $F_{i,j}(x)$ for each vertex $x\in V(G)$ is a set, since it may happen that $x$ is the result of vertex-coinciding two or more vertices of the tree $T_{i,j}$, that is, $x=x_{i,1}\bullet x_{i,2}\bullet \cdots \bullet x_{i,b}$ for $x_{i,s}\in V(T_{i,j})$ (Ref. Definition \ref{defn:vertex-split-coinciding-operations}).

(iii) \textbf{Good-property}: $F_{i,j}(V(G)\cup E(G))=[1,\chi\,''(G)]$, where $\chi\,''(G)$ is the total chromatic number of the connected $(p,q)$-graph $G$; $|F_{i,j}(y)|=1$ for each vertex $y\in V(G)$; and $F_{i,j}(xy)\neq F_{i,j}(xz)$ for $y,z\in N_{ei}(x)$ for each vertex $x\in V(G)$.

Let $V_{\textrm{sp-co}}(G)=\{T_{i,j}:j\in [1,a_i],i\in [1,m]\}$. Clearly, each totally colored tree $T_{i,j}\in V_{\textrm{sp-co}}(G)$ is \emph{colored graph homomorphism} to $G$, also, $T_{i,j}\rightarrow G$. \textbf{Find} totally colored trees $T_{i,j}$ from the set $V_{\textrm{sp-co}}(G)$, such that $T_{i,j}$ is colored graph homomorphism to $G$ admitting a total coloring $F_{i,j}$, which has the Good-property.

We set a \emph{base} $T_{splt}(G)=(T_1,T_2,\dots, T_m)$ with $T_k\rightarrow G$, each graph in the \emph{tree-graph lattice}
\begin{equation}\label{eqa:tree-graph-lattice5}
\textrm{\textbf{L}}(Z^0[\bullet]T_{splt}(G))=\big \{[\bullet]^m_{k=1}a_kT_k: a_k\in Z^0,T_k\in T_{splt}(G)\big \}
\end{equation} with $\sum^m_{k=1}a_k\geq 1$ can be vertex-split into trees $a_1T_1,a_2T_2,\dots, a_mT_m$, where the vertex-coinciding operation $[\bullet]$ is defined in Definition \ref{defn:vertex-split-coinciding-operations} and Remark \ref{rem:vertex-coinciding-operations}.
\end{problem}

\begin{defn} \label{defn:pan-hypergraphs}
$^*$ \textbf{Pan-hypergraphs.} Let $S_{thing}$ be a set of things, and let $[\bullet]$ be an operation on sets. Then each hyperedge set $\mathcal{E}\in \mathcal{E}\big (S^2_{thing} \big )$ defines a \emph{pan-hypergraph} $\mathcal{H}_{yper}=(S_{thing},\mathcal{E})$ subject a constraint set $R_{est}(c_1,c_2,\dots ,c_m)$ with $c_1:=[\bullet]$ and $m\geq 1$ if each hyperedge $e\in \mathcal{E}$ corresponds another hyperedge $e\,'\in \mathcal{E}$ holding $e[\bullet] e\,'\in S^2_{thing}$ by Definition \ref{defn:pan-operation-graph-a-hypergraph}.\qqed
\end{defn}

\section{Subgraphs Of Vertex-Intersected Graphs}

\begin{problem}\label{question:subgraphs-vertex-intersected}
For a given connected graph $G$ admitting a total set-coloring $F:V(G)\cup E(G)\rightarrow \mathcal{E}$, how to judge the given graph $G$ to be a vertex-intersected graph of a hypergraph $\mathcal{H}_{yper}=(\Lambda,\mathcal{E})$?
\end{problem}

We can answer partly Problem \ref{question:subgraphs-vertex-intersected} as follows:

(i) If $G$ is a complete graph, then $G$ is a vertex-intersected graph.

(ii) If there are no vertices $x,y$ of the graph $G$ holding $xy\not\in E(G)$ and $F(x)\cap F(y)\neq \emptyset$, then $G$ is a vertex-intersected graph.

\subsection{VSETC-algorithm}

\begin{thm}\label{thm:build-hyperedge-set}
\cite{Yao-Ma-arXiv-2201-13354v1} If a tree $T$ admits a coloring $f:V(T)\rightarrow [1,p-1]$ with $p=|V(T)|$ and $f(z)\neq f(y)$ for any pair of distinct vertices $x,y$, then $T$ admits a set-coloring defined on a hyperedge set $\mathcal{E}$ such that $T$ is a subgraph of a vertex-intersected graph of the hypergraph $\mathcal{H}_{yper}=([0,p-1],\mathcal{E})$.
\end{thm}
\begin{proof} According the hypothesis of the theorem, we define a set-coloring $F$ for the tree $T$ by means of the following so-called \textbf{VSETC-algorithm}:

\textbf{Step 1.} Each leaf $w_j$ of the tree $T$ is set-colored with $F(w_j)=\{f(w_j),f(v)\}$, where the edge $w_jv\in E(T)$.

\textbf{Step 2.} Each leaf $w^1_j$ of the tree $T_1=T-L(T)$, where $L(T)$ is the set of leaves of the tree $T$, is colored by $F(w^1_j)=\{f(w^1_j),f(z)\}$, where the edge $w^1_jz\in E(T_1)$.

\textbf{Step 3.} Each leaf $w^r_j$ of the tree $T_r=T-L(T_{r-1})$, where $L(T_{r-1})$ is the set of leaves of $T_{r-1}$, is colored by $F(w^r_j)=\{f(w^r_j),f(u)\}$, where the edge $w^r_ju\in E(T_r)$.

\textbf{Step 4.} Suppose $T_k$ is a star $K_{1,m}$ with vertex set $V(K_{1,m})=\big \{w^k_j,u_0:j\in [1,m]\big \}$ and edge set $E(K_{1,m})=\big \{u_0w^k_1,u_0w^k_2,\dots ,u_0w^k_m\big \}$, we color each leaf $w^k_j$ with $F(w^k_j)=\big \{f(w^k_j),f(u_0)\big \}$, and $F(u_0)=\{f(u_0)\}$.

\textbf{Step 5.} The above steps show $F(uv)=F(u)\cap F(v)\neq \emptyset$ for each edge $uv\in E(T)$.

Thereby, the tree $T$ admits the set-coloring $F$ subject to the constraint set $R_{est}(c_0)$, such that $c_0$ holds $F(uv)=F(u)\cap F(v)\neq \emptyset$ for each edge $uv\in E(T)$. So, $T$ is a subgraph of a vertex-intersected graph of a hypergraph $\mathcal{H}_{yper}=(\Lambda,\mathcal{E})$ with its hyperedge set $\mathcal{E}=F(V(T))$, and $\Lambda=[1,p-1]=\bigcup _{e\in \mathcal{E}}e$, we are done.
\end{proof}

\begin{figure}[h]
\centering
\includegraphics[width=16.4cm]{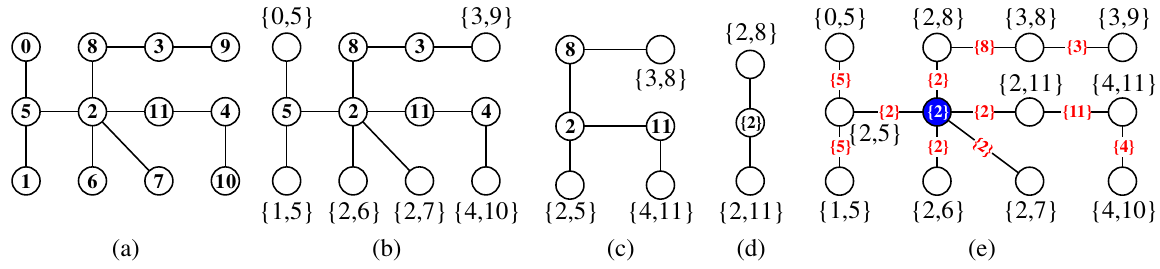}\\
\caption{\label{fig:get-set-coloring-00}{\small An example for understanding the proof of Theorem \ref{thm:build-hyperedge-set}.}}
\end{figure}

\begin{cor}\label{thm:99999}
\cite{Yao-Ma-arXiv-2201-13354v1} If a tree $T$ admits a graceful labeling, then $T$ admits a graceful set-coloring $F$ subject to the constraint set $R_{est}(c_0,c_1)$, such that the constraint $c_0$ holds $F(uv)\supseteq F(u)\cap F(v)\neq \emptyset$ for each edge $uv\in E(T)$, and the constraint $c_1$ holds $a_{uv}=|a_{u}-a_{v}|$ for some $a_{uv}\in F(uv)$, $a_u\in F(u)$ and $a_v\in F(v)$, which derives a consecutive integer set
$$
\{a_{uv}=|a_{u}-a_{v}|:uv\in E(T)\}=[1,|E(T)|\,]
$$ and moreover $T$ is a subgraph of a vertex-intersected graph of the hypergraph $\mathcal{H}_{yper}=(\Lambda,\mathcal{E})$ with $\mathcal{E}=F(V(T))$ and $\Lambda=[0,|E(T)|\,]=\bigcup_{e\in \mathcal{E}}e$.
\end{cor}

\begin{example}\label{exa:multiple-topological-authentication}
A \emph{multiple topological authentication} is shown in Fig.\ref{fig:22-multiple-authen}, we observe that a lobster $T$ is as a \emph{topological public-key}, and other lobsters $T_1,T_2,T_3,T_4,T_5,T_6$ form a group of \emph{topological private-keys}, where

\begin{asparaenum}[(i) ]
\item $T_2$ admits a \emph{pan-edge-magic total labeling} $f_2$ holding the \emph{edge-magic constraint} $f_2(x_i)+f_2(e_i)+f_2(y_i)=16$ for each edge $e_i=x_iy_i\in E(T_2)$;

\item $T_3$ admits a \emph{pan-edge-magic total labeling} $f_3$ holding the \emph{edge-magic constraint} $f_3(x_i)+f_3(e_i)+f_3(y_i)=27$ for each edge $e_i=x_iy_i\in E(T_3)$;

\item $T_4$ admits a \emph{felicitous labeling} $f_4$ holding the \emph{harmonious constraint} $f_4(e_i)=f_4(x_i)+f_4(y_i)~(\bmod~11)$ for each edge $e_i=x_iy_i\in E(T_4)$;

\item $T_5$ admits an \emph{edge-magic graceful labeling} $f_5$ holding the \emph{felicitous-difference constraint} $\big |f_5(x_i)+f_5(y_i)-f_5(e_i)\big |=4$ for each edge $e_i=x_iy_i\in E(T_5)$;

\item $T_6$ admits an \emph{edge-odd-graceful labeling} $f_6$ holding $\{f_6(x_i)+f_6(y_i)+f_6(e_i):e_i=x_iy_i\in E(T_6)\}=[16,26]$.\qqed
\end{asparaenum}
\end{example}

\begin{defn} \label{defn:22-one-v-multiple-e-labeling}
\cite{Yao-Ma-arXiv-2201-13354v1} Suppose that a graph $G$ admits a vertex labeling $f:V(G)\rightarrow [a,b]$ with $f(x)\neq f(y)$ for distinct vertices $x,y\in V(G)$. If there is a group of edge labelings $f_1,f_2,\dots,f_m$ of the graph $G$ induced by $f$ such that each edge labeling $f_k$ holds a $W_k$-constraint $W_k[f(u),f_k(uv),f(v)]=0$ for each edge $uv\in E(G)$, then we call $f$ \emph{one-v multiple-e labeling} of the graph $G$.\qqed
\end{defn}

\begin{figure}[h]
\centering
\includegraphics[width=16cm]{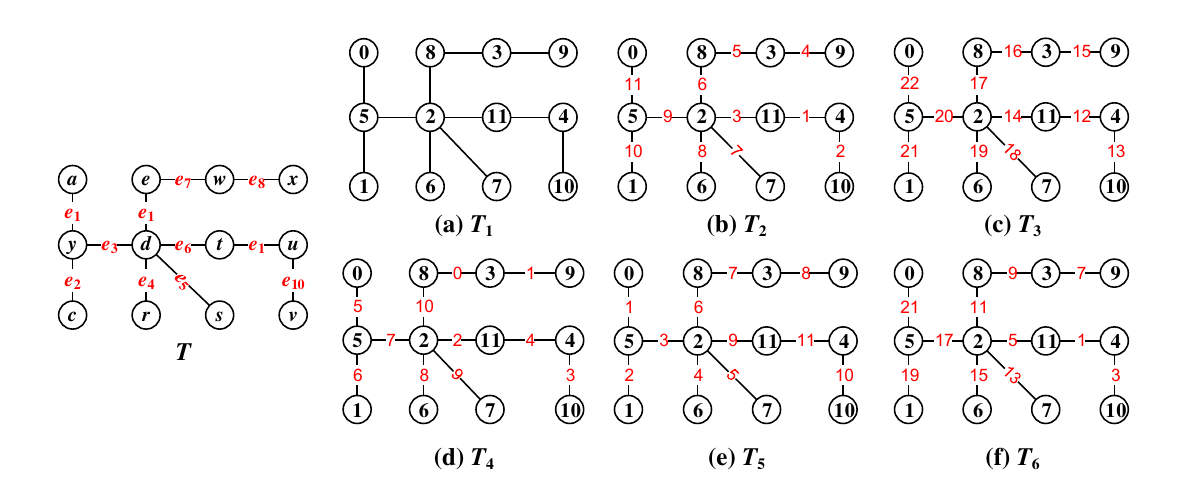}\\
\caption{\label{fig:22-multiple-authen}{\small A lobster $T_1$ admits a vertex labeling $f_1$ shown in (a), and this labeling $f_1$ induces other labelings $f_2,f_3,f_4,f_5$ and $f_6$, cited from \cite{Yao-Wang-2106-15254v1}.}}
\end{figure}

\begin{example}\label{exa:8888888888}
In Example \ref{exa:multiple-topological-authentication}, the lobster $T$ shown in Fig.\ref{fig:22-multiple-authen} admits a one-v multiple-e labeling defined in Definition \ref{defn:22-one-v-multiple-e-labeling} consisted of $f_1:V(T)\rightarrow [0,11]$ and $f_k$ induced by $f_1$ as follows:
\begin{asparaenum}[$\textbf{c}_1$:]
\item a \emph{pan-edge-magic total labeling} $f_2$ holds the \emph{edge-magic constraint} $f_2(x_i)+f_2(x_iy_i)+f_2(y_i)=16$ for each edge $x_iy_i\in E(T)$;

\item a \emph{pan-edge-magic total labeling} $f_3$ holds the \emph{edge-magic constraint} $f_3(x_i)+f_3(x_iy_i)+f_3(y_i)=27$ for each edge $x_iy_i\in E(T)$;

\item a \emph{felicitous labeling} $f_4$ holds the \emph{harmonious constraint} $f_4(x_iy_i)=f_4(x_i)+f_4(y_i)~(\bmod~11)$ for each edge $x_iy_i\in E(T)$;

\item an \emph{edge-magic graceful labeling} $f_5$ holds the \emph{felicitous-difference constraint} $\big |f_5(x_i)+f_5(y_i)-f_5(x_iy_i)\big |=4$ for each edge $x_iy_i\in E(T)$;

\item an \emph{edge-odd-graceful labeling} $f_6$ holds $\{f_6(x_i)+f_6(y_i)+f_6(x_iy_i):x_iy_i\in E(T)\}=[16,26]$.
\end{asparaenum}

Then we get a set-coloring $F$ of the lobster $T$ defined by $F(w)=\{f_1(w)\}$ for $w\in V(T)$, and
$$F(x_iy_i)=\{f_2(x_iy_i),f_3(x_iy_i),f_4(x_iy_i),f_5(x_iy_i),f_6(x_iy_i)\},~x_iy_i\in E(T)$$
\end{example}

The set-colored graph $L$ shown in Fig.\ref{fig:build-intersected-graphs} (a) is the above lobster $T$ colored with a set-coloring. Moreover, we, by using the VSETC-algorithm introduced in the proof of Theorem \ref{thm:build-hyperedge-set}, get a graph $L_{yper}$ shown in Fig.\ref{fig:build-intersected-graphs} (b) admitting a set-coloring $\varphi$, and $L_{yper}$ is a subgraph of a vertex-intersected graph of the hypergraph $\mathcal{H}^*_{yper}=([0,11],\mathcal{E}^*)$ subject to the constraint set $R_{est}(c_0,c_1$, $c_2$, $ c_3$, $c_4$, $c_5)$, where the hyperedge set
\begin{equation}\label{eqa:555555}
{
\begin{split}
\mathcal{E}^*=\big \{\{0,5\},\{1,5\},\{2,5\},\{2\},\{2,6\},\{2,7\},\{2,8\},\{3,8\},\{3,9\},\{2,11\},\{4,11\},\{4,10\}\big \}
\end{split}}
\end{equation}
holding $[0,11]=\bigcup_{e\in \mathcal{E}^*}e$, such that the vertex color set $\varphi(V(L_{yper}))=\mathcal{E}^*$ and the edge color set
\begin{equation}\label{eqa:555555}
{
\begin{split}
\varphi(E(L_{yper}))=&\big \{\{5,1,11,22,21\}, \{5,2,3,7,9,17,20\}, \{5,2,6,10,19,21\}, \{2,4,8,15,19\}, \\
&\{2,5,7,9,13,18\}, \{2,5,3,9,14\}, \{2,6,10,11,17\}, \{8,0,5,7,9,16\},\\
& \{3,1,4,7,8,15\},\{11,1,4,12\}, \{4,2,3,10,13\} \big \}
\end{split}}
\end{equation}
The Graham reduction of the hyperedge set $\mathcal{E}^*$ is an empty set, and $\mathcal{E}^*$ contains an ear set $\big \{\{0,5\},\{1,5\},\{2,6\},\{2,7\},\{3,9\},\{4,10\}\big \}$. The set-coloring $F$ of the set-colored graph $L$ and the set-coloring $\varphi$ of the set-colored graph $L_{yper}$ hold
$$F(V(L)\cup E(L))=\varphi(V(L_{yper})\cup E(L_{yper}))\setminus \mathcal{X}
$$ with $\mathcal{X}=\big \{\{2\},\{3\},\{4\},\{5\},\{8\},\{11\} \big \}$.\qqed

\begin{figure}[h]
\centering
\includegraphics[width=16.4cm]{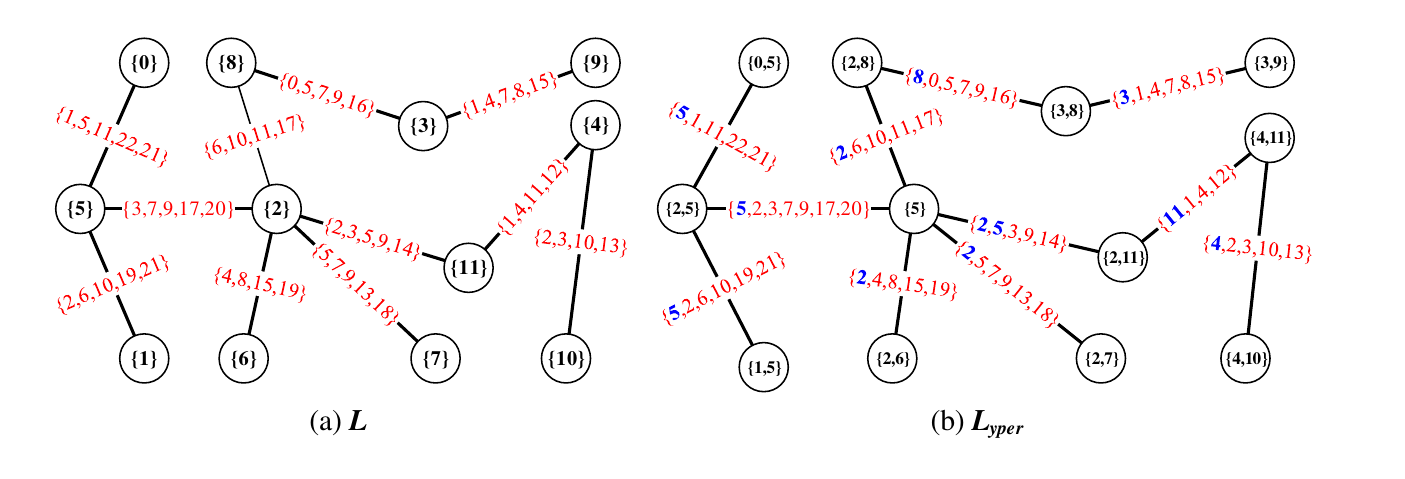}\\
\caption{\label{fig:build-intersected-graphs}{\small (a) A graph $L$ admits a set-coloring, but the graph $L$ is not a subgraph of any vertex-intersected graph; (b) $L_{yper}$ is a subgraphs of some vertex-intersected graph.}}
\end{figure}

\begin{thm}\label{thm:666666}
\cite{Yao-Ma-arXiv-2201-13354v1} If a graph $G$ admits a one-v multiple-e labeling, which induces a set-coloring $\varphi$, then there is another graph $G^*$ admitting this set-coloring $\varphi$ such that $G^*\cong G$, and the graph $G^*$ is a subgraph of a vertex-intersected graph of a hypergraph $\mathcal{H}_{yper}=(\Lambda,\mathcal{E})$.
\end{thm}

\begin{thm}\label{thm:producing-set-colorings-from-labelings}
\cite{Yao-Ma-arXiv-2201-13354v1} Each tree admits a graceful-intersection total set-labeling and an odd-graceful-intersection total set-labeling.
\end{thm}

\begin{problem}\label{qeu:444444}
\textbf{Find} graceful labelings and odd-graceful labelings of trees corresponding to some graceful-intersection total set-labelings and some odd-graceful-intersection total set-labelings of the trees.
\end{problem}

\begin{thm}\label{thm:each-graph-admits-set-colorings}
\cite{Yao-Ma-arXiv-2201-13354v1} Each connected $(p,q)$-graph $G$ admits a set-coloring $F:V(G)\rightarrow \mathcal{E}$ with a hyperedge set $\mathcal{E}$ defined on a consecutive integer set $\Lambda=[1,p]$, such that $F(x)\neq F(y)$ for distinct vertices $x,y\in V(G)$, and each edge $uv\in E(G)$ is colored by an induced hyperedge color set $F(uv)$ holding $F(uv)\supseteq F(u)\cap F(v)\neq \emptyset$.
\end{thm}

\begin{thm}\label{thm:graceful-inter-total-set-colorings-number}
\cite{Yao-Ma-arXiv-2201-13354v1} Each connected $(p,q)$-graph $G$ admits at least $q!$ graceful-intersection total set-colorings defined on hyperedge sets $\mathcal{E}$ holding $[1,q]=\bigcup_{e\in \mathcal{E}}e$ true.
\end{thm}

\begin{thm}\label{thm:find-perfect-hypermatchings}
\cite{Yao-Ma-arXiv-2201-13354v1} If a tree $T$ of $q$ edges admits a graceful-intersection total set-labeling $F:V(T)\rightarrow \mathcal{E}\subseteq [0,q]^2$ with $\bigcup_{e\in \mathcal{E}} e=[1,q]$, then the tree $T$ contains an independent vertex set $X$ so that the vertex color set $\{e=F(x):x\in X\}$ is just a perfect hypermatching $\mathcal{M}\subset \mathcal{E}$, and $\bigcup_{e\in \mathcal{M}} e=[1,q]$.
\end{thm}

\begin{thm}\label{thm:set-ordered-perfect-hypermatchings}
\cite{Yao-Ma-arXiv-2201-13354v1} If a tree $T$ of $q$ edges admits a set-ordered graceful labeling $f$ holding the set-ordered constraint $\max f(X)<\min f(Y)$ for the bipartition $(X,Y)$ of $V(T)$ (Ref. Definition \ref{defn:basic-W-type-labelings}), then this set-ordered graceful labeling $f$ induces a graceful-intersection total set-labeling $F:V(T)\cup E(T)\rightarrow \mathcal{E}^*\subseteq [0,q]^2$ with $\bigcup_{e\in \mathcal{E}^*} e=[1,q]$, such that two vertex color sets $M_X=\{e=F(x):x\in X\}$ and $M_Y=\{e\,'=F(y):y\in Y\}$ are two perfect hypermatchings of the set-colored tree $T$ holding $M_X\cup M_Y=\mathcal{E}^*$ and $M_X\cap M_Y=\emptyset$ true.
\end{thm}

\begin{problem}\label{qeu:444444}
\textbf{Characterize} trees admitting set-ordered graceful labelings by means of perfect hypermatchings generated from graceful-intersection total set-labelings.
\end{problem}

\subsection{PSCS-algorithms}

For producing set-colorings (PSCS) of graphs, we introduce the following methods in terms of algorithmic representations.

\textbf{PSCS-algorithm-1.} Suppose that $T$ is a tree admitting a $W$-constraint coloring or a $W$-constraint labeling $f$ holding $f(uv)\neq f(xy)$ for any pair of edges $uv$ and $xy$ of the tree $T$, and $|f(V(T))|=|V(T)|$, then we have:

\textbf{Step 1.1.} Apply the VSETC-algorithm introduced in the proof of Theorem \ref{thm:build-hyperedge-set} to the tree $T$ first, such that $T$ admits a $W$-constraint set-coloring $F_1$ induced by the labeling $f$, and each leaf $w$ of the tree $T$ is colored with $F_1(w)$ holding $|F_1(w)\cap F_1(z)|=1$, where the vertex $z$ is adjacent with $w$ in $T$, and $F_1(x)\neq F_1(u)$ for two distinct vertices $x,u\in V(T)$ since $f(uv)\neq f(xy)$ for any pair of edges $uv$ and $xy$ of the tree $T$.

\textbf{Step 1.2.} Do the VSETC-algorithm introduced in the proof of Theorem \ref{thm:build-hyperedge-set} to the tree $T$ once time, here, the tree $T$ admits a $W$-constraint set-coloring $F_2$ induced by $F_1$ such that each each leaf $w$ of the tree $T$ and its adjacent vertex $z$ are colored with $F_2(w)$ and $F_2(z)$ that hold $|F_2(w)\cap F_2(z)|\geq 2$, clearly, $F_2(x)\neq F_2(u)$ for two distinct vertices $x,u\in V(T)$ since $f(uv)\neq f(xy)$ for any pair of edges $uv$ and $xy$ of the tree $T$.

\textbf{Step 1.$k$.} After $k$ times of doing the VSETC-algorithm to the tree $T$, we get a $W$-constraint set-coloring $F_k$ of the tree $T$ and $|F_k(w)\cap F_k(z)|\geq k$ for each leaf $w$ and its adjacent vertex $z$. Suppose that the diameter $D(T)=m$, then $|F_k(w)\cap F_k(z)|\geq k=\lfloor \frac{m}{2}\rfloor$ for each leaf $w$ and its adjacent vertex $z$, as well as $F_k(x)\neq F_k(u)$ for two distinct vertices $x,u\in V(T)$.

\vskip 0.4cm

See trees $D_1,D_2,D_3,D_4$ shown in Fig.\ref{fig:make-set-coloring-11} and Fig.\ref{fig:make-set-coloring-22} for understanding the PSCS-algorithm-1. Moreover, $F_k(V(T))=\mathcal{E}_k$ is a \emph{hyperedge set} defined on the vertex color set $\Lambda_k=f(V(T))$.

\begin{figure}[h]
\centering
\includegraphics[width=16.4cm]{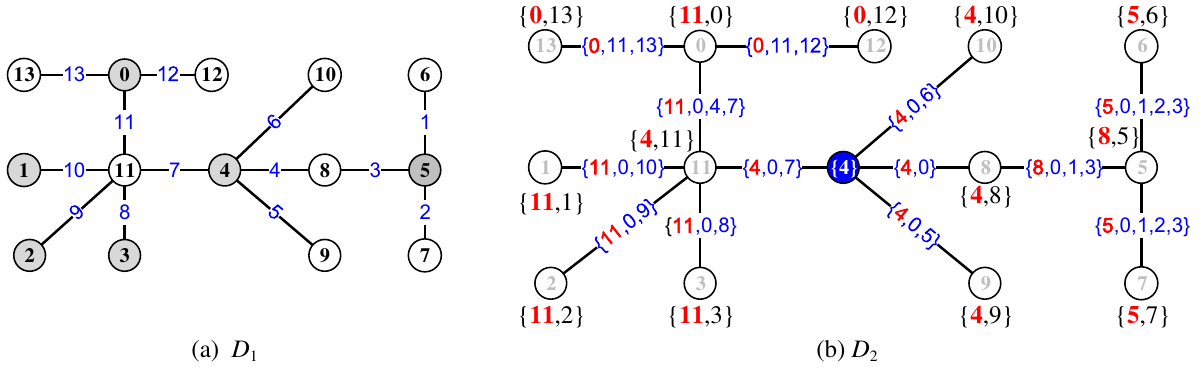}\\
\caption{\label{fig:make-set-coloring-11}{\small (a) A tree $D_1$ admits a graceful labeling $f$; (b) the tree $D_2$ admits a graceful set-coloring $F_1$ generated by $f$ and the VSETC-algorithm introduced in the proof of Theorem \ref{thm:build-hyperedge-set}.}}
\end{figure}

\begin{figure}[h]
\centering
\includegraphics[width=16.4cm]{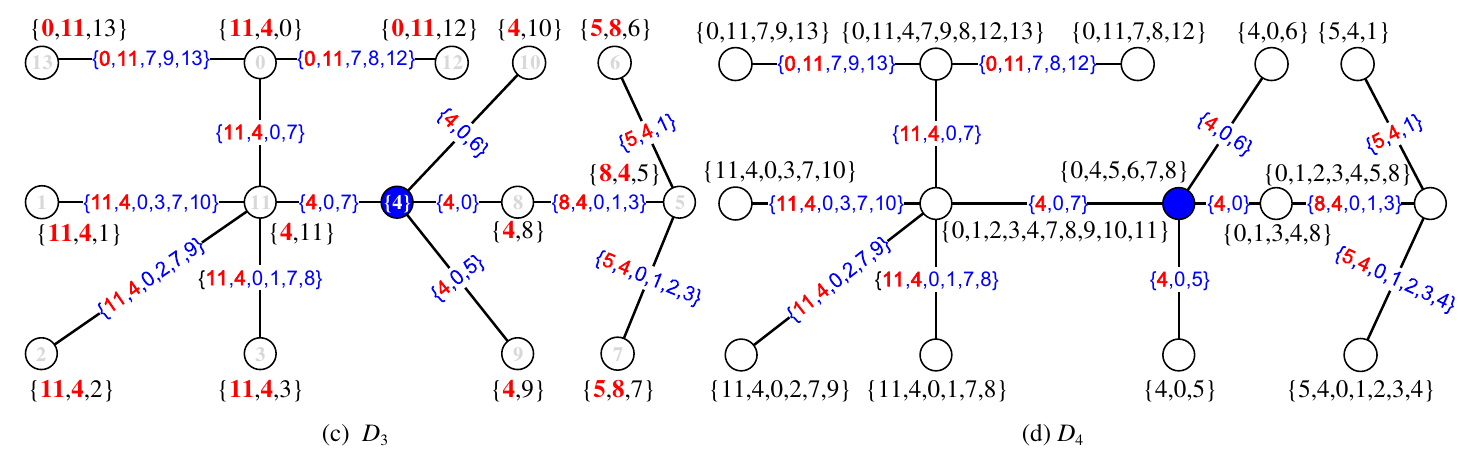}\\
\caption{\label{fig:make-set-coloring-22}{\small (c) The tree $D_3$ admits a graceful set-coloring $F_2$ generated by the VSETC-algorithm and the graceful set-coloring $F_1$ of the tree $D_2$ shown in Fig.\ref{fig:make-set-coloring-11}; (d) the tree $D_4$ admits a graceful set-coloring $F_3$ generated by the VSETC-algorithm and the graceful set-coloring $F_2$ of the tree $D_3$.}}
\end{figure}

\vskip 0.4cm

\textbf{PSCS-algorithm-2.} Suppose that a tree $H$ admits a $W$-constraint coloring or a $W$-constraint labeling $g$ with $g(uv)\neq g(xy)$ for any pair of edges $uv$ and $xy$ of the tree $H$, and $|g(V(H))|=|V(H)|=\Lambda$.

Firstly, by the method introduced in the proof of Theorem \ref{thm:producing-set-colorings-from-labelings}, we define a $W$-constraint total set-labeling $F$ of the tree $H$ in the following way: $F(x)=\{g(xy):y\in N_{ei}(x)\}$ for $x\in V(H)$, where $N_{ei}(x)$ is the set of neighbors of the vertex $x$, and $F(uv)=\{f(uv)\}$ for each edge $uv\in E(H)$, so
$$F(x)=\{g(xy):y\in N_{ei}(x)\}\neq \{g(wz):z\in N_{ei}(w)\}=F(w)
$$ for any two vertices $x,w\in V(H)$ based on the edge color set $g(E(T))$, then $F(V(H))=\mathcal{E}$ is a hyperedge set define on the set $\Lambda=g(V(H))$.

Secondly, we do the VSETC-algorithm introduced in the proof of Theorem \ref{thm:build-hyperedge-set} to the tree $H$ admitting the $W$-constraint total set-labeling $F$ for more complex set-colorings.

\vskip 0.4cm

\textbf{PSCS-algorithm-3.} Suppose that a connected $(p,q)$-graph $G$ contains at least a cycle and admits a $W$-constraint labeling or a $W$-constraint coloring $h$ satisfying $h(uv)\neq h(xy)$ for any pair of distinct edges $uv$ and $xy$ of the graph $G$.

\textbf{Step 3.1.} Do the vertex-splitting operation to the connected $(p,q)$-graph $G$ for getting a tree $T_G$ of $(q+1)$ vertices, first, we select randomly an edge $xy$ in a cycle $C$ of the graph $G$, and vertex-split $x$ into two vertices $x\,'$ and $x\,''$, such that the adjacent neighbor set $N_{ei}(x)=N_{ei}(x\,')\cup N_{ei}(x\,'')$ with $N_{ei}(x\,')\cap N_{ei}(x\,'')=\emptyset$ and cardinality $|N_{ei}(x\,')|=1$, so the resultant graph $G_1$ is connected, and define a total coloring $h_1$ for $G_1$ by setting $h_1(w)=h(w)$ for $w\in V(G-\{x\,',x\,''\})\cup E(G-\{x\,',x\,''\})$, and $h_1(x\,')=h(x)$, $h_1(x\,'')=h(x)$, as well as $h_1(x\,'y)=h(xy)$ for $y\in N_{ei}(x\,')$ and $h_1(ux\,'')=h(ux)$ for $u\in N_{ei}(x\,'')$. Go on in this way, after $q-p+1$ times, we get the desired tree $T_G$ of $(q+1)$ vertices admitting a $W$-constraint coloring $h_{q-p+1}$ satisfying $h_{q-p+1}(uv)\neq h_{q-p+1}(xy)$ for any pair of distinct edges $uv,xy\in E(T_G)$. See for examples shown in Fig.\ref{fig:graceful-coloring-hyper-11} (a) and (b).

\begin{figure}[h]
\centering
\includegraphics[width=16.4cm]{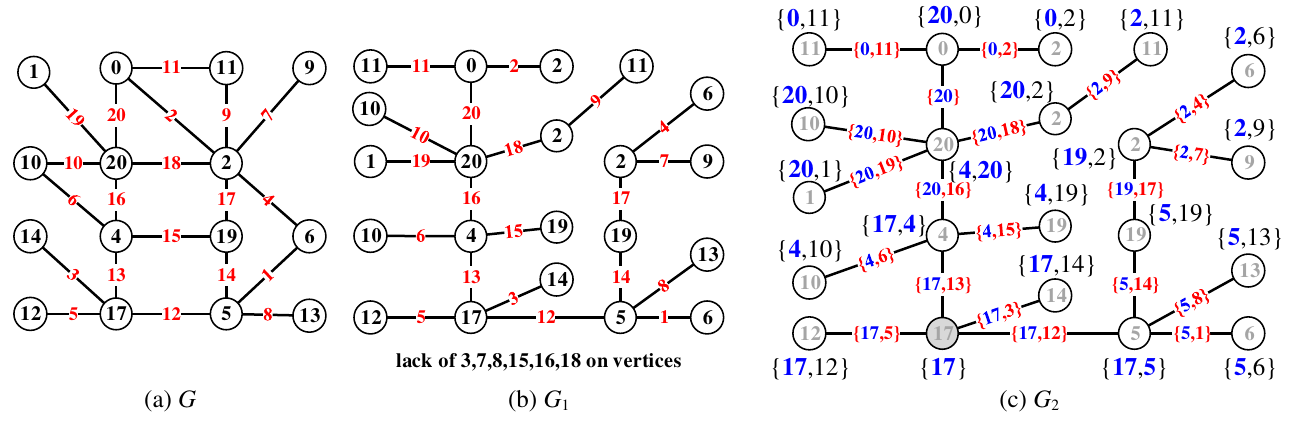}\\
\caption{\label{fig:graceful-coloring-hyper-11}{\small (a) A non-tree graph $G$ admitting a graceful labeling; (b) a tree $G_1$ obtained from $G$ by dong the vertex-splitting operation; (c) the tree $G_2$ admitting a graceful-type set-coloring subject to the constraint set $R_{est}(c_1,c_2)$, where $G_2\cong G_1$.}}
\end{figure}

\textbf{Step 3.2.} Define a set-coloring $F$ for the tree $T_G$ subject to the constraint set $R_{est}(c_0,c_1,c_2$, $\dots $, $c_m)$ by doing the VSETC-algorithm to $T_G$ several times, such that each edge is colored with $F(uv)$ holding the first constraint $c_0:F(uv)\supseteq F(u)\cap F(v)\neq \emptyset$, and one $c_j$ holds $a^j_{uv}=|a^j_{u}-a^j_{v}|$ for some $a^j_{u}\in F(u)$ and $a^j_{v}\in F(v)$, that is,
$$
F(uv)=\big [F(u)\cap F(v)\big ]\cup \big \{a^j_{uv}=|a^j_{u}-a^j_{v}|:j\in [1,m]\big \}
$$ subject to the constraint set $R_{est}(c_0,c_1,c_2,\dots ,c_m)$. See a tree $G_3$ shown in Fig.\ref{fig:graceful-coloring-hyper-11} (c) for understanding this procedure.

\textbf{Step 3.3.} Let $S^*=\{j_1,j_2,\dots,j_s\}=\Lambda \setminus \{F(x):x\in V(T_G)\}$, where $\Lambda=[a,b]$ is a consecutive integer set with integers $a,b$ subject to $0\leq a\leq b$. Making a new set-coloring $F\,'$ by setting $F\,'(uv)=F(uv)$ for $uv\in E(T_G)$, and each vertex $x\in V(T_G)$ is colored with $F\,'(x)=F(x)$ or $F\,'(x)=F(x)\cup S^*_x$ with $S^*_x=\{j_a,j_b,\dots ,j_l\}\subseteq S^*$, such that $F\,'(V(T_G))=\Lambda$ and $F\,'(x)\neq F\,'(w)$ for distinct vertices $x,w\in V(T_G)$. See the vertex colors of a graph $G_{yper}$ shown in Fig.\ref{fig:graceful-coloring-hyper-22}.

\textbf{Step 3.4.} Define a set-coloring $F^*$ subject to the constraint set $R^*_{est}(c^*_0,c^*_1,\dots ,c^*_{m})$ in the following way:
$F^*(y)=F\,'(y)$ for $y\in V(T_G)$,
$$
F^*(uv)=\big [F^*(u)\cap F^*(v)\big ]\cup \{a^i_{uv}=|a^i_{u}-a^i_{v}|:a^i_{u}\in F^*(u), a^i_{v}\in F^*(v),i\in [1,m]\}
$$ where $R^*_{est}(c^*_0,c^*_1,\dots ,c^*_{m})$ consisted of $c^*_0:F^*(uv)\supseteq F^*(u)\cap F^*(v)\neq \emptyset$, and $c^*_i:a^i_{uv}=|a^i_{u}-a^i_{v}|$ for some $a^i_{u}\in F^*(u)$ and $a^i_{v}\in F^*(v)$ with $i\in [1,m]$. The graph $G_{yper}$ shown in Fig.\ref{fig:graceful-coloring-hyper-22} admits such a set-coloring $F^*$ subject to the constraint set $R^*_{est}(c^*_0,c^*_1,\dots ,c^*_{m})$. Thereby, $F^*(V(T_G))=\mathcal{E}^*$ is a hyperedge set defined on a consecutive integer set $\Lambda=[a,b]$.

After vertex-coinciding $u=A_1\bullet A_2$, $v=B_1\bullet B_2\bullet B_2$, $w=C_1\bullet C_2$ and $z=D_1\bullet D_2$ of $G_{yper}$ shown in Fig.\ref{fig:graceful-coloring-hyper-22}, we obtain the original connected graph $G$ shown in Fig.\ref{fig:graceful-coloring-hyper-11}(a) admitting a set-coloring $\varphi^*$ induced by the set-coloring $F^*$ of the tree $G_{yper}$, where
$$\varphi^*(u)=F^*(A_1)\cup F^*(A_2),\varphi^*(v)=F^*(B_1)\cup F^*(B_2)\cup F^*(B_3),\varphi^*(w)=F^*(C_1)\cup F^*(C_2)
$$ and $\varphi^*(z)=F^*(D_1)\cup F^*(D_2)$, and each $w$ of other edges and vertices of the graph $G$ is colored with $\varphi^*(w)=F^*(w)$.

\begin{figure}[h]
\centering
\includegraphics[width=14.4cm]{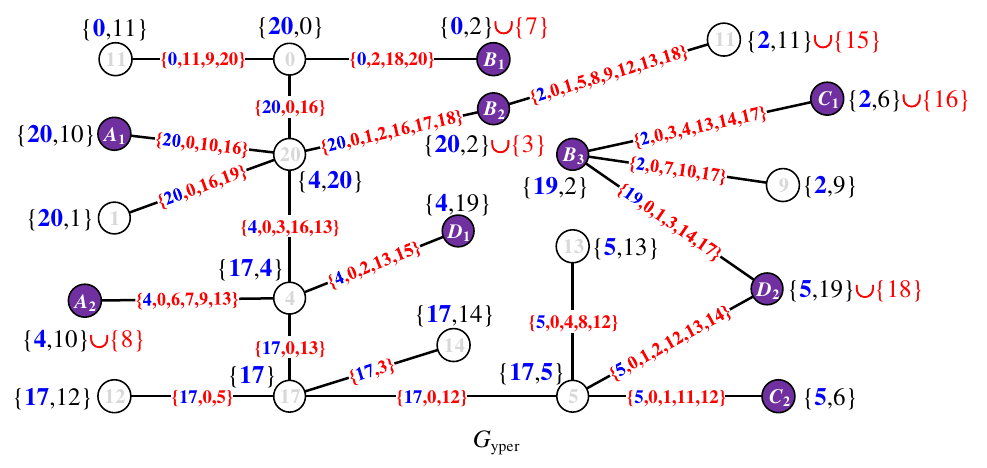}\\
\caption{\label{fig:graceful-coloring-hyper-22}{\small $G_{\textrm{yper}}$ is a subgraph of some vertex-intersected graph admitting the set-coloring $F^*$.}}
\end{figure}

\vskip 0.4cm

\textbf{PSCS-algorithm-4.} By Definition \ref{defn:e-odd-graceful-v-matching-labeling}, an edge-odd-graceful graph base $\textbf{B}=(G_1,G_2,\dots $, $G_m)$ with $G_i\not \subset G_j$ and $G_i\not \cong G_j$ for $i\neq j$ admits an edge-odd-graceful vertex-matching labeling $F=\uplus ^m_{i=1}f_i$, where each connected $(p_i,q_i)$-graph $G_i$ admits a labeling $f_i$ holding
$$f_i(E(G_i))=\{f_i(u_jv_j)=|f_i(u_j)-f_i(v_j)|:~u_jv_j\in E(G_i)\}=[1,2q_i-1]^o
$$ true, and $\bigcup ^m_{i=1}f_i(V(G_i))=[0,M]$ with $m\geq 2$. Then we have a graph $G^*$, denoted as $[\bullet\cup]^m_{i=1}G_i$, obtained by vertex-coinciding those vertices of $G_1,G_2,\dots , G_m$ colored with the same color into one vertex, and union all graphs that have no vertices colored with the same color together. Now, we define a set-coloring for $G^*$ in the following procedure:

\textbf{Step 4.1.} If each graph $G_i\in \textbf{B}$ is a tree, then we get a set-coloring $F_i$ of $G_i$ by the PSCS-algorithm-1, such that $F_i(V(G_i))=\mathcal{E}_i$ defined on a set $\Lambda_i$.

\textbf{Step 4.2.} If each graph $G_j\in \textbf{B}$ is not a tree, then there is a set-coloring $F_j$ of $G_j$ obtained by the PSCS-algorithm-3, like the set-coloring $\varphi^*$ of the graph $G$, such that $F_j(V(G_j)=\mathcal{E}_j$ defined on a set $\Lambda_j$.

\textbf{Step 4.3.} We define a set-coloring $\theta$ of the graph $G^*=[\bullet\cup]^m_{i=1}G_i$ in the following way:

(i) $\theta(w)=\bigcup ^p_{i=1}F_{k_i}(w_{k_i})$ if a vertex $w$ of the graph $G^*$ is $w=[\bullet ]^p_{i=1} w_{k_i}$ obtained by the vertex $w_{k_i}$ of graphs $G_{k_i}$ of the base $\textbf{B}$ with $i\in [1,p]$.

(ii) If an edge $wz$ off the graph $G^*$ is $wz=[\ominus] ^s_{r=1} w_{j_r}z_{j_r}$ generated by some edges $w_{j_r}z_{j_r}$ of graphs $G_{j_r}$ of the base $\textbf{B}$ for $r\in [1,s]$, then we color it as $\theta(wz)=\bigcup ^s_{r=1}F_{j_r}(w_{j_r}z_{j_r})$.

(iii) For other $G^*$'s vertices $x$ and edges $xy$ that do not participate into vertex-coincided vertices $w=[\bullet ]^p_{i=1} w_{k_i}$ and edge-coincided edges $wz=[\ominus] ^s_{r=1} w_{j_r}z_{j_r}$, we color $\theta(x)=F_i(x)$ if $x\in V(G_i)$ and $\theta(xy)=F_j(xy)$ if $xy\in E(G_j)$.

Thereby, $\theta(V(G^*))=\mathcal{E}^*=\bigcup ^m_{j=1}\mathcal{E}_j$ and $\Lambda ^*=\bigcup_{e\in \mathcal{E}^*}e$ for $\Lambda^* =\bigcup ^m_{j=1}\Lambda_j$, such that $G^*$ is a subgraph of a vertex-intersected graph of the hypergraph $\mathcal{H}_{yper}=(\Lambda^*,\mathcal{E}^*)$.

\subsection{Set-colorings with multiple intersections}

\begin{defn} \label{defn:W-join-type-se-total-coloring}
\cite{Yao-Ma-arXiv-2201-13354v1} Let $\mathcal{E}$ be a set of subsets of a finite set $\Lambda$ such that each hyperedge $e\in \mathcal{E}$ satisfies $e\neq \emptyset$ and corresponds another hyperedge $e\,'\in \mathcal{E}$ holding $e\cap e\,'\neq \emptyset$, as well as $\Lambda=\bigcup _{e\in \mathcal{E}}e$. Suppose that a connected graph $H$ admits a \emph{total set-labeling} $\pi: V(H)\cup E(H)\rightarrow \mathcal{E}$ with $\pi(w)\neq \pi(z)$ for distinct vertices $w,z\in V(H)$, and the edge color set $\pi (uv)$ for each edge $uv\in E(G)$ holds $\pi (uv)\neq \pi (uw)$ for each neighbor $w\in N_{ei}(u)$ and each vertex $u\in V(G)$. There are the following intersected-type constraints:
\begin{asparaenum}[\textbf{\textrm{Chyper}}-1.]
\item \label{join:join-edge} $\pi (u)\cap \pi (v)\subseteq \pi (uv)$ and $\pi (u)\cap \pi (v)\neq \emptyset $ for each edge $uv\in E(G)$.
\item \label{join:join-edge-r-rank} $\pi (u)\cap \pi (v)\subseteq \pi (uv)$ and $|\pi (u)\cap \pi (v)|\geq r\geq 2$ for each edge $uv\in E(G)$.
\item \label{join:join-edge-vertex} $\pi (uv)\cap \pi (u)\neq \emptyset$ and $\pi (uv)\cap \pi (v)\neq \emptyset$ for each edge $uv\in E(G)$.
\item \label{join:join-dajacent-edges} $\pi (uv)\cap \pi (uw)\neq \emptyset$ for each neighbor $w\in N_{ei}(u)$ and each vertex $u\in V(G)$.
\item \label{join:dajacent-edges-no-join} $\pi (uv)\cap \pi (uw)=\emptyset$ for each neighbor $w\in N_{ei}(u)$ and each vertex $u\in V(G)$.
\end{asparaenum}
\textbf{Then we have:}
\begin{asparaenum}[\textbf{\textrm{Cgraph}}-1.]
\item If Chyper-\ref{join:join-edge} holds true, then $G$ is called a \emph{subvertex-intersected graph} and $\pi$ is called \emph{subintersected total set-labeling} of the graph $G$.
\item If Chyper-\ref{join:join-edge-r-rank} holds true, then $G$ is called a \emph{$r$-rank subvertex-intersected graph} and $\pi$ is called a \emph{$r$-rank subintersected total set-labeling} of the graph $G$.
\item If Chyper-\ref{join:join-edge} and Chyper-\ref{join:join-edge-vertex} hold true, then $G$ is called an \emph{intersected-edge-intersected graph} and $\pi$ is called an \emph{intersected-edge-intersected total set-labeling} of the graph $G$.
\item If Chyper-\ref{join:join-edge-r-rank} and Chyper-\ref{join:join-edge-vertex} hold true, $G$ is called a \emph{$r$-rank intersected-edge-intersected graph} and $\pi$ is called a \emph{$r$-rank intersected-edge-intersected total set-labeling} of the graph $G$.
\item $G$ is called an \emph{edge-intersected graph} if Chyper-\ref{join:join-edge-vertex} holds true, and $\pi$ is called an \emph{edge-intersected total set-labeling} of the graph $G$.
\item If Chyper-\ref{join:join-edge-vertex} and Chyper-\ref{join:join-dajacent-edges} hold true, then $G$ is called an \emph{adjacent edge-intersected graph}, and $\pi$ is called an \emph{adjacent edge-intersected total set-labeling} of the graph $G$.
\item If Chyper-\ref{join:join-edge-vertex} and Chyper-\ref{join:dajacent-edges-no-join} hold true, then $G$ is called an \emph{individual edge-intersected graph}, and $\pi$ is called an \emph{individual edge-intersected total set-labeling} of the graph $G$.\qqed
\end{asparaenum}
\end{defn}

\begin{problem}\label{qeu:444444}
Since $\Lambda=\bigcup _{e\in \mathcal{E}}e$ for a hyperedge set $\mathcal{E}$ based on a finite set $\Lambda$, and a connected graph $G$ admits a \emph{$W$-constraint-intersected total set-labeling} $F$ defined in Definition \ref{defn:W-join-type-se-total-coloring}, \textbf{characterize} hyperedge sets $\mathcal{E}$ and estimate the extremum number $\min \big \{\max \Lambda :\Lambda=\bigcup _{e\in \mathcal{E}}e \big \}$ when each $\Lambda$ is a consecutive nonnegative integer set.
\end{problem}

\begin{thm}\label{thm:five-edge-join-total-set-labelings}
Each tree admits one of subintersected total set-labeling, intersected-edge-intersected total set-labeling, edge-intersected total set-labeling, adjacent edge-intersected total set-labeling and individual edge-intersected total set-labeling defined in Definition \ref{defn:W-join-type-se-total-coloring}.
\end{thm}
\begin{proof} By Theorem \ref{thm:tree-graceful-total-coloringss}, a tree $T$ admits a gracefully total coloring $f$ with its own vertex color set $f(V(T))\subseteq [0,q]$ and its own edge color set $f(E(T))=[1,q]$, where $|E(T)|=q$. Let $\Lambda=[1,q]$ in the following deduction.

\textbf{F-1.} We, for the tree $T$, define a total set-labeling $F_1$ as: $F_1(x)=\{f(xv):v\in N_{ei}(x)\}$ for each vertex $x\in V(T)$, and each edge $uv$ of $E(T)$ is colored with $F_1(uv)=\{f(uv)\}$ and for each edge $uv\in E(T)$. Clearly, $F_1(x)\neq F_1(y)$ for distinct vertices $x,y\in V(T)$ and $[1,q]=F_1(V(T))$, since $f(E(T))=\Lambda$. So, $F_1(V(T))$ is a \emph{hyperedge set} $\mathcal{E}_1$ defined on a finite set $\Lambda$, and moreover $\{f(uv)\}= F_1(u)\cap F_1(v)\subseteq F_1(uv)$. Thereby, we claim that $F_1$ is a \emph{subintersected total set-labeling} of the tree $T$ according to Definition \ref{defn:W-join-type-se-total-coloring}, and the tree $T$ is a subvertex-intersected graph, as a subgraph of a vertex-intersected graph of a hypergraph $\mathcal{H}_{yper}=(\Lambda,\mathcal{E}_1)$.

\textbf{F-2.} Notice that $F_1(uv)\cap F_1(u)=\{f(uv)\}$ and $F_1(uv)\cap F_1(v)=\{f(uv)\}$, so $F_1$ is an \emph{intersected-edge-intersected total set-labeling} defined in Definition \ref{defn:W-join-type-se-total-coloring}, and the tree $T$ is an intersected-edge-intersected graph.

\textbf{F-3.} We have $F_1$ to be an \emph{edge-intersected total set-labeling} and the tree $T$ to be an edge-intersected graph because $F_1(uv)\cap F_1(u)=\{f(uv)\}$ and $F_1(uv)\cap F_1(v)=\{f(uv)\}$ from Definition \ref{defn:W-join-type-se-total-coloring}.

\textbf{F-4.} Notice that

(i) $F_1(uv)\cap F_1(uw)=\emptyset$ for $w\in N_{ei}(u)$ and each vertex $u\in V(T)$;

(ii) and $F_1(uv)\cap F_1(u)=\{f(uv)\}$ and $F_1(uv)\cap F_1(v)=\{f(uv)\}$.

We claim that the tree $T$ is an individual edge-intersected graph admitting an \emph{individual edge-intersected total set-labeling} $F_1$.

\textbf{F-5.} Let $L(T_i)$ be the leaf-set of leaves of the tree $T_i$, where a leaf is a vertex having degree one. So, we have trees $T_{i+1}=T_{i}-L(T_{i})$ for $i\in [1,m]$ with $m=\big \lfloor \frac{D(T)+1}{2}\big \rfloor -1$, where $T_1=T$ and $D(T)$ is the diameter of the tree $T$. And, for $i\in [1,m]$, each leaf-set $L(T_i)$ can be cut into subsets $L_r(T_i)=\{w_{i,r,1},w_{i,r,2},\dots, w_{i,r,d(v_{i,r})}\}$ for $r\in [1,A_i]$, $w_{i,r,j}\in N_{ei}(v_{i,r})$ for $j\in [1,d(v_{i,r})]$, where $v_{i,r}$ is not a leaf of $T_i$ and has its own degree $d(v_{i,r})=\textrm{deg}_T(v_{i,r})\geq 2$.

We define a new total set-labelling $F_2$ for the tree $T$ as follows: $F_2(x)=F_1(x)$ for each vertex $x\in V(T)$. For coloring edges of the tree $T$, there are the following cases:

\textbf{Case 1.} If the tree $T$ is a star $K_{1,n}$ with its center vertex $x_0$ and leaf set $L(K_{1,n})=\{x_1,x_2,\dots ,x_n\}$, we set $F_2(x_0x_i)=\bigcup^n_{j=1}F_1(x_0x_j)$ for each $i\in [1,n]$. So, $F_2$ is an adjacent edge-intersected total set-labeling of $K_{1,n}=T$.

\textbf{Case 2.} If the diameter $D(T)\geq 3$. We present the following algorithm:

\textbf{Step C2.1.} For the leaf set $L(T_1)$ and $r\in [1,A_1]$, we color each edge $w_{1,r,j}v_{1,r}$ with
\begin{equation}\label{eqa:555555}
F_{2,1}(w_{1,r,j}v_{1,r})=F_1(v_{1,r}z_{1,r})\bigcup \left [\bigcup ^{d(v_{1,r})}_{j=1}F_1(w_{1,r,j}v_{1,r})\right ]
\end{equation} for $j\in [1,d(v_{i,r})]$, and $F_{2,1}(w_{1,r,j}v_{1,r})=F_{2,1}(v_{1,r}z_{1,r})$, where $z_{1,r}\not\in L(T_1)$, however, $v_{1,r}\in L(T_2)$ and $z_{1,r}\in L(T_3)$.

\textbf{Step C2.2.} For leaf set $L(T_2)$ and $r\in [1,A_2]$, notice that $v_{1,r}\in L(T_2)$ for $r\in [1,A_1]$, so $v_{1,r}\in L_r(T_2)=\{w_{2,r,1},w_{2,r,2},\dots, w_{2,r,d(v_{2,r})}\}$, and $v_{1,r}=w_{2,r,s}$ for some $s$ and $z_{1,r}=v_{2,r}$, then each edge $w_{2,r,j}v_{2,r}$ is colored with
\begin{equation}\label{eqa:555555}
F_{2,2}(w_{2,r,j}v_{2,r})=F_{2,1}(v_{1,r}z_{1,r})\bigcup F_1(v_{2,r}z_{2,r})\bigcup \left [\bigcup ^{d(v_{2,r})}_{j=1,j\neq s}F_1(w_{2,r,j}v_{2,r})\right ]
\end{equation} for $j\in [1,d(v_{2,r})]$, and $F_{2,2}(w_{2,r,j}v_{2,r})=F_{2,2}(v_{2,r}z_{2,r})$, where $z_{2,r}\not\in L(T_2)$, however, $v_{2,r}\in L(T_3)$ and $z_{2,r}\in L(T_4)$.

\textbf{Step C2.k.} For leaf set $L(T_k)$ and $r\in [1,A_k]$, we have leaves $v_{k-1,r}\in L(T_k)$ for $r\in [1,A_{k-1}]$, and $v_{k-1,r}\in L_r(T_k)=\{w_{k,r,1},w_{k,r,2},\dots, w_{k,r,d(v_{k,r})}\}$, as well as $v_{k-1,r}=w_{k,r,s}$ for some $s$ and $z_{k-1,r}=v_{k,r}$, and we color each edge $w_{k,r,j}v_{k,r}$ with
\begin{equation}\label{eqa:555555}
F_{2,k}(w_{k,r,j}v_{k,r})=F_{2,k-1}(v_{k-1,r}z_{k-1,r})\bigcup F_1(v_{k,r}z_{k,r})\bigcup \left [\bigcup ^{d(v_{k,r})}_{j=1,j\neq s}F_1(w_{k,r,j}v_{k,r})\right ]
\end{equation} for $j\in [1,d(v_{k,r})]$, and $F_{2,k}(w_{k,r,j}v_{k,r})=F_{2,k}(v_{k,r}z_{k,r})$, where $z_{k,r}\not\in L(T_k)$, however, $v_{k,r}\in L(T_{k+1})$ and $z_{2,r}\in L(T_{k+2})$.

Go on in the above procedure, we meet:

(a) The last tree $T_{m}=T_{m-1}-L(T_{m-1})$ is a star $K_{1,1}$ with $V(K_{1,1})=\{w_{m,1,1}, v_{m,1}\}$ and $E(K_{1,1})=\{w_{m,1,1}v_{m,1}\}$, we set
$$F_{2,m}(w_{m,1,1}v_{m,1})=F_1(w_{m,1,1}v_{m,1})\cup F_{2,m-1}(w_{m,1,1}w_{m-1,i})\cup F_{2,m-1}(v_{m,1}v_{m-1,j})
$$ where $w_{m-1,i},v_{m-1,j}\in L(T_{m-1})$.

(b) The last tree $T_{m}=T_{m-1}-L(T_{m-1})$ is a star $K_{1,p}$ with $V(K_{1,p})=\{v_{m,1},w_{m,1,q}:q\in [1,p]\}$ and $E(K_{1,p})=\{v_{m,1}w_{m,1,q}:q\in [1,p]\}$, so $w_{m,1,q}w_{m-1,j,q}\in E(t)$ and $w_{m-1,j,q}\in L(T_{m-1})$, we color each edge $v_{m,1}w_{m,1,q}$ as
\begin{equation}\label{eqa:555555}
F_{2,m}(v_{m,1}w_{m,1,q})=\left [\bigcup ^{p}_{q=1}F_{2,m-1}(w_{m,1,q}w_{m-1,j,q}) \right]\bigcup \left [\bigcup ^{p}_{q=1}F_1(v_{m,1}w_{m,1,q})\right ]
\end{equation} for $q\in [1,p]$

Thereby, $F_2(xy)$ for each edge $xy\in E(T)$ is defined as $F_2(xy)=F_{2,k}(xy)$ if $x\in L(T_k)$ and $y\in L(T_{k-1})$ for $k\in [1,m]$. It is not difficult to see $F_2(uv)\cap F_2(uw)\neq \emptyset$ for each neighbor $w\in N_{ei}(u)$ and each vertex $u\in V(T)$, so we claim that $F_2$ is an adjacent edge-intersected total set-labeling of the tree $T$, and $F_2(V(T))$ is a hyperedge set of the hypergraph $\mathcal{H}_{yper}=(\Lambda,F_2(V(T)))$. See Fig.\ref{fig:adjacent-edge-join-set-11} for obtaining an adjacent edge-intersected total set-labeling of a tree.

The proof of the theorem is complete.
\end{proof}

\begin{figure}[h]
\centering
\includegraphics[width=15.6cm]{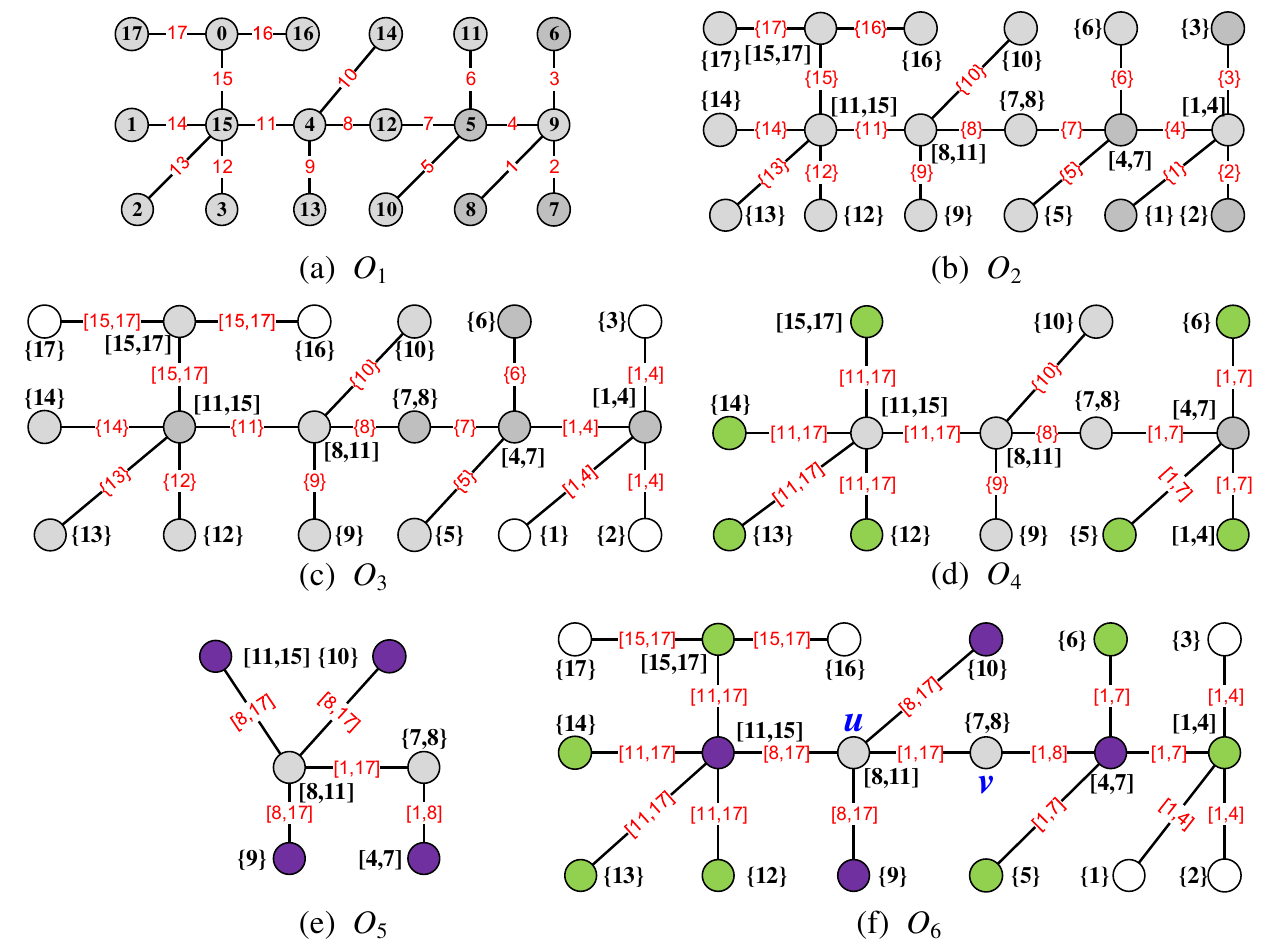}\\
\caption{\label{fig:adjacent-edge-join-set-11}{\small An algorithm for obtaining the adjacent edge-intersected total set-labeling $F$ of the tree $O_1$ shown in (a), where $F(uv)=[1,17]$ in $O_6$.}}
\end{figure}

\subsection{An algorithm for adjacent edge-intersected total set-labelings}

Another algorithm is shown in Fig.\ref{fig:adjacent-edge-join-set-22}. We use a longest path to make an adjacent edge-intersected total set-labeling of a tree at each time. Let $T_1$ be a tree with diameter $D(T_1)\geq 3$, and $T_1$ admit a graceful coloring $g$ with $g(V(T_1))\subseteq [0,q]$ and $g(E(T_1))=[1,q]$, where $|E(T_1)|=q$. Let $\Lambda=[1,q]$, $m=\big \lfloor \frac{D(T)+1}{2}\big \rfloor -1$.

\textbf{Step 1.} Define a total set-labeling $F^*$ as: $F^*(x)=\{g(xy):y\in N_{ei}(x)\}$ for each vertex $x\in V(T)$, and $F^*(uv)=\{g(uv)\}$ for each edge $uv\in E(T_1)$. Clearly, $F^*(x)\neq F^*(y)$ for distinct vertices $x,y\in V(T_1)$ and $F^*(V(T))=\Lambda$, since $g(E(T_1))=\Lambda$. So, $F^*(V(T_1))=\mathcal{E}_1$ is a \emph{hyperedge set} defined on a finite set $\Lambda$, and moreover $\{g(uv)\}= F^*(u)\cap F^*(v)\subseteq F^*(uv)$.

\textbf{Step 2.} We define a new total set-labelling $F^*_1$ for $T_1$ in the following steps:

\textbf{Step 2.1.} $F^*_1(x)=F^*(x)$ for each vertex $x\in V(T_1)$.

\textbf{Step 2.2.} Suppose that $P_1=x_{1,1}x_{1,2}\dots x_{1,n_1}$ is a longest path of the tree $T_1$ with $m\geq 3$, $x_{1,2}$ and $x_{1,n_1-1}$ are not leaves of $T_1$. The adjacent neighbor set $N_{ei}(x_{1,2})$ of the vertex $x_{1,2}$ is of form as $N_{ei}(x_{1,2})=\{x_{1,3}\}\cup L_{eaf}(x_{1,2})$ for leaf set
$$
L_{eaf}(x_{1,2})=\big \{x_{1,1}\big \}\cup \big \{y_{1,j}:j\in [1,\textrm{deg}(x_{1,2})-1],y_{1,j}\in N_{ei}(x_{1,2})\big \}\subset L(T_1)
$$ and the adjacent neighbor set $N_{ei}(x_{1,n_1-1})=\{x_{1,n_1-2}\}\cup L_{eaf}(x_{1,n_1-1})$ for leaf set
$$L_{eaf}(x_{1,n_1-1})=\big \{x_{1,n_1}\big \}\cup \big \{u_{1,i}:i\in [1,\textrm{deg}(x_{1,n_1-1})-1],u_{1,i}\in N_{ei}(x_{1,n_1-1})\big \}\subset L(T_1)
$$ Now, we color these leaf-edges $x_{1,2}y$ and $x_{1,n_1-1}u$ as
\begin{equation}\label{eqa:555555}
{
\begin{split}
&F^*_1(x_{1,2}y)=F^*(x_{1,2}y)\cup F^*(x_{1,2}x_{1,3}),~ y\in L_{eaf}(x_{1,2})\\
&F^*_1(x_{1,n_1-1}u)=F^*(x_{1,n_1-1}u)\cup F^*(x_{1,n_1-1}x_{1,n_1-2}),~ u\in L_{eaf}(x_{1,n_1-1})
\end{split}}
\end{equation}

\textbf{Step 2.3.} Let $T_2=T_1-L_{eaf}(x_{1,2})-L_{eaf}(x_{1,n_1-1})$ with $D(T_2)\geq 3$, and take a longest path $P_2=x_{2,1}x_{2,2}\dots x_{2,n_2}$ of the tree $T_2$. We have two leaf sets $L_{eaf}(x_{2,2})$ and $L_{eaf}(x_{2,n_2-1})$, and two vertices $x_{2,3},x_{2,n_2-2}$ are not in $L(T_2)$. We color those leaf-edges $x_{2,2}y$ and $x_{2,n_2-1}u$ in the following
\begin{equation}\label{eqa:555555}
{
\begin{split}
&F^*_1(x_{2,2}y)=F^*(x_{2,2}y)\cup F^*(x_{2,2}x_{2,3}),~ y\in L_{eaf}(x_{2,2})\\
&F^*_1(x_{2,n_2-1}u)=F^*(x_{2,n_2-1}u)\cup F^*(x_{2,n_2-1}x_{2,n_2-2}),~ u\in L_{eaf}(x_{2,n_2-1})
\end{split}}
\end{equation}

\textbf{Step 2.$k+2$.} We have a tree $T_{k+1}=T_{k}-L_{eaf}(x_{k,2})-L_{eaf}(x_{k,n_k-1})$ with $D(T_{k+1})\geq 3$ and $k\in [1,m-1]$, and take a longest path $P_{k+1}=x_{k+1,1}x_{k+1,2}\dots x_{k+1,n_{k+1}}$ of the tree $T_{k+1}$. There are two leaf sets $L_{eaf}(x_{k+1,2})$ and $L_{eaf}(x_{k+1,n_{k+1}-1})$, and two vertices $x_{k+1,3},x_{k+1,n_{k+1}-2}$ are not in $L(T_{k+1})$. We color those leaf-edges $x_{k+1,2}y$ and $x_{k+1,n_{k+1}-1}u$ as follows
{\small
\begin{equation}\label{eqa:555555}
{
\begin{split}
&F^*_1(x_{k+1,2}y)=F^*(x_{k+1,2}y)\cup F^*(x_{k+1,2}x_{k+1,3}),y\in L_{eaf}(x_{k+1,2})\\
&F^*_1(x_{k+1,n_{k+1}-1}u)=F^*(x_{k+1,n_{k+1}-1}u)\cup F^*(x_{k+1,n_{k+1}-1}x_{k+1,n_{k+1}-2}), u\in L_{eaf}(x_{k+1,n_{k+1}-1})
\end{split}}
\end{equation}
}

Thereby, we have recolored the edges of the tree $T_1$ well, and it is not hard to see $F^*_1(uv)\cap F^*_1(uw)\neq \emptyset$ for $w\in N_{ei}(u)$ and each vertex $u\in V(T_1)$, see examples shown in Fig.\ref{fig:adjacent-edge-join-set-22}. We claim that $F^*_1$ is an adjacent edge-intersected total set-labeling of $T_1$.

\begin{figure}[h]
\centering
\includegraphics[width=16.4cm]{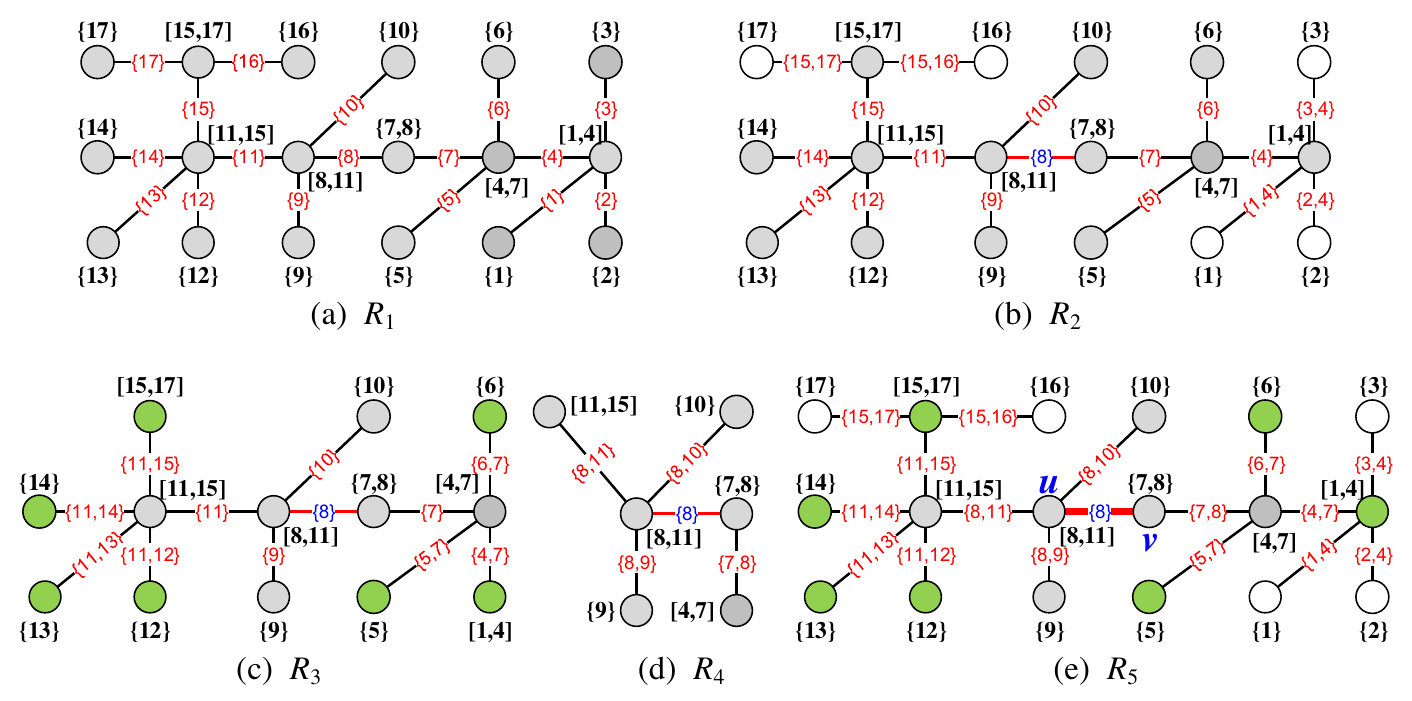}\\
\caption{\label{fig:adjacent-edge-join-set-22}{\small Another algorithm for obtaining the adjacent edge-intersected total set-labeling $F^*$ of the tree $O_1$ shown in Fig.\ref{fig:adjacent-edge-join-set-11} (a), where $F^*(uv)=\{8\}$ in the set-colored tree $R_5$, and there are two perfect hypermatchings in a vertex-intersected graph $R_1$.}}
\end{figure}

\begin{thm}\label{thm:graph-five-edge-join-total-set-labelings}
Each connected graph admits each one of subintersected total set-labeling, intersected-edge-intersected total set-labeling, edge-intersected total set-labeling, adjacent edge-intersected total set-labeling and individual edge-intersected total set-labeling defined in Definition \ref{defn:W-join-type-se-total-coloring}.
\end{thm}
\begin{proof} By the vertex-splitting operation, we vertex-split a connected $(p,q)$-graph $G$ into a tree $T_G$ of $(q+1)$ vertices. Since $T_G$ admits one of subintersected total set-labeling, intersected-edge-intersected total set-labeling, edge-intersected total set-labeling, adjacent edge-intersected total set-labeling and individual edge-intersected total set-labeling according to Theorem \ref{thm:five-edge-join-total-set-labelings}.

Doing the vertex-coinciding operation to $T_G$ produces the original connected $(p,q)$-graph $G$, we define a set-coloring $F^*$ for $G$ in the following: Since $T_G$ admits a total set-coloring $F:V(T_G)\cup E(T_G)\rightarrow \mathcal{E}$ to be one intersected-type total set-labeling defined in Definition \ref{defn:W-join-type-se-total-coloring}, we define the desired set-coloring $F^*$ as: $F^*(z)=F(x)\cup F(y)$ if $z=x\bullet y$ for $z\in V(G)$ and $x,y\in V(T_G)$, $F^*(uv)=F(uv)$ for each edge $uv\in E(G)$ and $uv\in E(T_G)$, then we claim that the connected $(p,q)$-graph $G$ admits each one of the intersected-type total set-labelings defined in Definition \ref{defn:W-join-type-se-total-coloring}.

Notice that there are trees $T^1_{G}, T^2_{G}, \dots,T^M_{G}$ of $(q+1)$ vertices obtained from the connected $(p,q)$-graph $G$ by means of the vertex-splitting operation, so the connected $(p,q)$-graph $G$ admits more set-coloring defined in Definition \ref{defn:W-join-type-se-total-coloring}.

The proof of the theorem is completed.
\end{proof}

\section{Graph Operations Of Vertex-Intersected Graphs}

\subsection{Set-increasing and set-decreasing operations}

Let $\mathcal{E}=\{e_1$, $e_2$, $\dots$, $e_n\}=\big \{e_i\big \}^n_{i=1}$ and $\mathcal{X}=\{e^*_1,e^*_2,\dots,e^*_n\}=\big \{e^*_i\big \}^n_{i=1}$ be two hyperedge sets based on a finite set $\Lambda$. We do a set subtraction operation to $\mathcal{E}$ and $\mathcal{X}$, such that each set $e\,'_i$ of the resultant subset set $\mathcal{E}\,'=\big \{e\,'_i\big \}^n_{i=1}$ holds $e\,'_i=e_i\setminus e^*_{i_j}$ or $e\,'_i=e_i$ with $e_i\in \mathcal{E}$ and $e^*_{i_j}\in \mathcal{X}$ and $i\in [1,n]$, and there is at least a set $e\,'_j$ holding $e\,'_j\neq e_j$ for some $j$. We write $\mathcal{E}\,'=\mathcal{E}[\setminus ]\mathcal{X}$ and $\mathcal{E}=\mathcal{E}\,'[\cup ]\mathcal{X}$, respectively. And we call $\mathcal{H}\,'_{yper}=(\Lambda,\mathcal{E}\,')$ \emph{set-decreased hypergraph} of the hypergraph $\mathcal{H}_{yper}=(\Lambda,\mathcal{E})$ according to $\mathcal{E}\,'=\mathcal{E}[\setminus ]\mathcal{X}$; conversely, $\mathcal{H}_{yper}$ is a \emph{set-increased hypergraph} of the hypergraph $\mathcal{H}\,'_{yper}$ because of $\mathcal{E}=\mathcal{E}\,'[\cup ]\mathcal{X}$.

\begin{defn} \label{defn:111111}
\cite{Yao-Ma-arXiv-2201-13354v1} For the hyperedge set $\mathcal{E}$ of a hypergraph $\mathcal{H}_{yper}=(\Lambda,\mathcal{E})$ if there is no hyperedge set $\mathcal{X}\in \mathcal{E}(\Lambda^2)$ such that $\mathcal{E}[\setminus ]\mathcal{X}$ is the hyperedge set $\mathcal{E}^*$ of some hypergraph $\mathcal{H}^*_{yper}=(\Lambda,\mathcal{E}^*)$, then we say that the hyperedge set $\mathcal{E}$ defined on a finite set $\Lambda$ is not \emph{decreasing}.\qqed
\end{defn}

\begin{example}\label{exa:8888888888}
A hyperedge set $\mathcal{E}$ defined in Eq.(\ref{eqa:example-hypergraph}) and another hyperedge set $\mathcal{X}=\{\{1\}$, $\{2\}$, $\{7\}$, $\{8\}\}$ produce the following hyperedge set
\begin{equation}\label{eqa:reducing-hypergraph-two-operations}
{
\begin{split}
\mathcal{E}\,'=\mathcal{E}[\setminus ]\mathcal{X}=&\big \{\{12\},\{11\},\{10\},\{6,10,11,12\},\{4,5,6\},\{5,7\},\\
& \{7,8,9\},\{8\},\{4,9\},\{3,4\},\{2\},\{1,2,3\},\{1\}\big \}
\end{split}}
\end{equation} and $\bigcup _{e_i\in \mathcal{E}\,'}e_i=[1,12]$, so we get a hypergraph $\mathcal{H}\,'_{yper}=([1,12],\mathcal{E}\,')$ defined by the hyperedge set $\mathcal{E}\,'$ shown in Eq.(\ref{eqa:reducing-hypergraph-two-operations}) and its vertex-intersected graph $L$ shown in Fig.\ref{fig:reduce-splitting-coinciding} (a). The hypergraph $\mathcal{H}\,'_{yper}=([1,12],\mathcal{E}\,')$ has a \emph{perfect hypermatching} $\big \{ \{12\},\{11\},\{10\},\{4,5,6\},\{7,8,9\},\{1,2,3\}\big \}$, which is one of a vertex-intersected graph $L$ too. Clearly, the Graham reduction of the hyperedge set $\mathcal{E}\,'$ is an empty set, so the hypergraph $\mathcal{H}\,'_{yper}$ is acyclic. On the other hands, a vertex-intersected graph $L$ is a tree, which implies that the hypergraph $\mathcal{H}\,'_{yper}$ is acyclic.

Since a vertex-intersected graph $L$ is a tree, so the corresponding hypergraph $\mathcal{H}\,'_{yper}$ is \emph{acyclic}, and $L$ is the unique vertex-intersected graph of the hypergraph $\mathcal{H}\,'_{yper}$. Moreover, the hypergraph $\mathcal{H}\,'_{yper}$ is a set-decreased hypergraph of the hypergraph $\mathcal{H}_{yper}$ defined by the hyperedge set $\mathcal{E}$ shown in Eq.(\ref{eqa:example-hypergraph}), and a vertex-intersected graph $L$ is a subgraph of a vertex-intersected graph $G_{yper}$ shown in Fig.\ref{fig:coloring-hypergraph} (b).\qqed
\end{example}

Notice that the hyperedge set $\mathcal{E}\,'$ defined in Eq.(\ref{eqa:reducing-hypergraph-two-operations}) is not decreasing, we have the following result:

\begin{thm}\label{thm:666666}
If the hyperedge set $\mathcal{E}$ of a hypergraph $\mathcal{H}_{yper}=(\Lambda,\mathcal{E})$ is not decreasing, then this hypergraph $\mathcal{H}_{yper}$ is acyclic, also, its vertex-intersected graph is acyclic too.
\end{thm}

\begin{rem}\label{rem:333333}
A set-colored graph $G$ may be not a vertex-intersected graph of any hypergraph, however, some proper subgraph of the graph $G$ is really a vertex-intersected graph of some hypergraph, see a set-colored graph $L_2$ shown in Fig.\ref{fig:reduce-splitting-coinciding}.\qqed
\end{rem}

\begin{figure}[h]
\centering
\includegraphics[width=16.4cm]{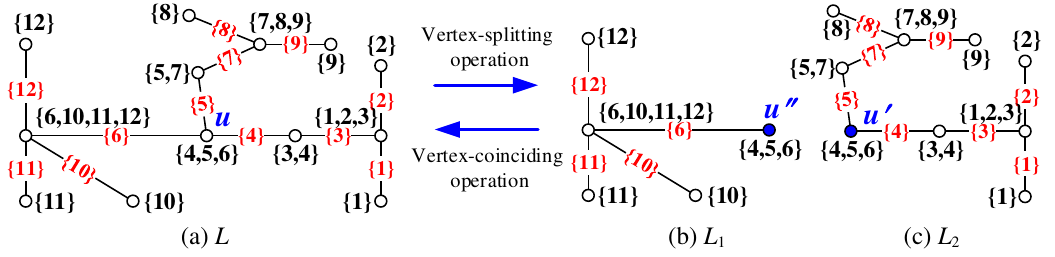}\\
\caption{\label{fig:reduce-splitting-coinciding}{\small (a) A tree $L$ admitting a graceful-intersection total set-labeling defined on the hyperedge set $\mathcal{E}\,'$ shown in Eq.(\ref{eqa:reducing-hypergraph-two-operations}) is a vertex-intersected graph; (b) and (c) are the resultant graphs of doing the vertex-splitting operation to $L$, in which $L_2$ is a vertex-intersected graph based on the hyperedge set $\big\{ \{1,2,3\}$, $\{1\}$, $ \{2\}$, $\{3,4\}$, $\{4,5,6\}$, $\{5,7\}$, $\{7,8,9\}$, $\{8\}$, $\{9\}\big\}\subset [1,9]^2$.}}
\end{figure}

\subsection{Splitting-type and coinciding-type operations}

Suppose that a vertex-intersected graph $H$ of a hypergraph $\mathcal{H}_{yper}=(\Lambda,\mathcal{E})$ subject to the constraint set $R_{est}(c_0,c_1,c_2,\dots ,c_m)$ admits a set-coloring $F:V(H)\rightarrow \mathcal{E}$, such that each edge $uv$ of $E(H)$ is colored with an induced edge color $F(uv)$ defined in Definition \ref{defn:vertex-intersected-graph-hypergraph}. We introduce the following basic splitting and coinciding operations on the vertices and edges of vertex-intersected graphs or graphs.

\subsubsection{Edge-splitting and edge-coinciding operations}

Let $G$ be an intersected-$(p,q)$-graph admitting a $W$-constraint set-coloring $F$ defined on a hypergraph $\mathcal{H}_{yper}=(\Lambda,\mathcal{E})$. We edge-split an edge $uv\in E(G)$ into two edges $u\,'v\,'$ and $u\,''v\,''$, such that the adjacent neighbor sets $N_{ei}(u)=N_{ei}(u\,')\cup N_{ei}(u\,'')$ and $N_{ei}(v)=N_{ei}(v\,')\cup N_{ei}(v\,'')$, the resultant graph is denoted as $G\wedge uv$, which has $|V(G\wedge uv)|=|V(G)|+2=p+2$ vertices and $|E(G\wedge uv)|=|E(G)|+1=q+1$ edges. We define a set-coloring $F^*$ of the edge-split graph $G\wedge uv$ as:
\begin{asparaenum}[\textbf{\textrm{Esc}}-1]
\item $F^*(z)=F(z)$ for $z\in V(G-\{u,uv,v\})\cup E(G-\{u,uv,v\})$.
\item $F^*(xu\,')=F(xu)$ for $x\in N_{ei}(u\,')\subset N_{ei}(u)$, $F^*(yu\,'')=F(yu)$ for $y\in N_{ei}(u\,'')\subset N_{ei}(u)$.
\item $F^*(wv\,')=F(wv)$ for $w\in N_{ei}(v\,')\subset N_{ei}(v)$, $F^*(zv\,'')=F(zv)$ for $z\in N_{ei}(v\,'')\subset N_{ei}(v)$.
\item $F^*(u\,')=F(u)$ and $F^*(u\,'')=F(u)$, $F^*(v\,')=F(v)$ and $F^*(v\,'')=F(v)$.
\item $F^*(u\,'v\,')=F(uv)$ and $F^*(u\,''v\,'')=F(uv)$.
\end{asparaenum}

Thereby, the edge-split $(p+2,q+1)$-graph $G\wedge uv$ admitting a $W$-constraint set-coloring $F^*$ is a vertex-intersected $(p+2,q+1)$-graph of a hypergraph $\mathcal{H}_{yper}=(\Lambda^*,\mathcal{E}^*)$, where $\Lambda^*=\Lambda$ and $\mathcal{E}^*=\mathcal{E}$.

Conversely, we have the original vertex-intersected graph $G=H[u\,'v\,'\ominus u\,''v\,'']$ by the edge-coinciding operation, where $H=G\wedge uv$.

\begin{problem}\label{qeu:444444}
If two hypergraphs $\mathcal{H}_{yper}=(\Lambda,\mathcal{E})$ and $\mathcal{H}_{yper}=(\Lambda^*,\mathcal{E}^*)$ hold $\Lambda^*=\Lambda$ and $\mathcal{E}^*=\mathcal{E}$ true, \textbf{characterize} their vertex-intersected graphs.
\end{problem}

Suppose that $H_1$ and $H_2$ are two vertex-disjoint graphs and each $H_i$ is a vertex-intersected graph of the hypergraph $\mathcal{H}^i_{yper}=(\Lambda_i,\mathcal{E}_i)$ for $i=1,2$. We take edges $u_{i,j}v_{i,j}\in E(H_i)$ for $i=1,2$ and $j\in [1,s]$, and do the edge-coinciding operation to edges $u_{1,j}v_{1,j}$ and $u_{2,j}v_{2,j}$ in order to obtain edge-coincided edges $u_{j}v_{j}=u_{1,j}v_{1,j}\ominus u_{2,j}v_{2,j}$ with vertex-coincided vertices $u_{j}=u_{1,j}\bullet u_{2,j}$, $v_{j}=v_{1,j}\bullet v_{2,j}$ for $j\in [1,s]$, the resultant graph is denoted as $H_1[\bullet ]H_2$.

Suppose that each vertex-intersected graph $H_i$ admits a set-coloring $F_i$ defined on the hyperedge set $\mathcal{E}_i$ based on a finite set $\Lambda_i$ with $i=1,2$, and then we define a set-coloring $F$ for $H_1[\bullet ]H_2$ by setting $F(u_{j})=F_1(u_{1,j})\cup F_2(u_{2,j})$, $F(v_{j})=F_1(v_{1,j})\cup F_2(v_{2,j})$ and $F(u_{j}v_{j})=F_1(u_{1,j}v_{1,j})\cup F_2(u_{2,j}v_{2,j})$ with $j\in [1,s]$, and
$$
F(w)=F_i(w),~w\in \big [V(H_i)\cup E(H_i)\big ]\setminus \big [\{u_{i,j},v_{i,j}:j\in [1,s]\}\cup \{u_{i,j}v_{i,j}:j\in [1,s]\}\big ]
$$ with $i=1,2$. So, the graph $H_1[\bullet ]H_2$ is a vertex-intersected graph of a hypergraph $\mathcal{H}^{\ominus}_{yper}=\big (\Lambda_1\cup \Lambda_2,\mathcal{E}_1\cup \mathcal{E}_2 \big )$, where
\begin{equation}\label{eqa:hypergraphs-hyperedge-coinciding}
\mathcal{H}^{\ominus}_{yper}=\ominus\big \langle \mathcal{H}^1_{yper}, \mathcal{H}^2_{yper}\big \rangle=\big (\Lambda_1\cup \Lambda_2,\mathcal{E}_1\cup \mathcal{E}_2 \big )
\end{equation}

In general, let $\textbf{\textrm{H}}=(H_1,H_2,\dots ,H_m)$ with $H_i\not \subset H_j$ and $H_i\not \cong H_j$ if $i\neq j$ be a \emph{vertex-intersected graph base}, where each vertex-intersected graph $H_i$ admits a set-coloring $F_i$ defined on the hyperedge set $\mathcal{E}_i$ of the hypergraph $\mathcal{H}^i_{yper}=(\Lambda_i,\mathcal{E}_i)$ with $i\in [1,m]$. Let $G_1,G_2,\dots ,G_A$ be a permutation of vertex-intersected graphs $a_1H_1,a_2H_2,\dots ,a_mH_m$, where $A=\sum^m_{k=1} a_k\geq 1$. Now, we do the edge-coinciding operation to edges $x_{k,j}y_{k,j}\in E(G_k)$ and $x_{k+1,j}y_{k+1,j}\in E(G_{k+1})$ for $k\in [1,A-1]$ to obtain a graph
\begin{equation}\label{eqa:555555}
L=\ominus\langle G_1, G_2,\dots ,G_A\rangle=[\ominus] ^m_{k=1} a_kH^k
\end{equation}
to be a subgraph of a vertex-intersected graph of a hypergraph
\begin{equation}\label{eqa:hypergraphs-hyperedge-coinciding11}
\mathcal{H}^{\ominus}_{yper}(L)=\big (\Lambda(L),\mathcal{E}(L) \big )=[\ominus ]^m_{k=1} a_k\mathcal{H}^k_{yper}
 \end{equation}where $\Lambda(L)=\bigcup ^m_{k=1}\Lambda_k$ and $\mathcal{E}(L)=\bigcup ^m_{k=1}\mathcal{E}_k$.

We get an \emph{edge-coincided vertex-intersected graph lattice} as
\begin{equation}\label{eqa:hypergraph-vertex-intersected graph-lattice-edge}
\textbf{\textrm{L}}\big (Z^0[\ominus]\textbf{\textrm{H}}\big )=\big \{[\ominus] ^m_{k=1} a_kH^k:a_k\in Z^0,H_k\in \textbf{\textrm{H}}\big \}
\end{equation} with $\sum^m_{k=1} a_k\geq 1$. Clearly, each set-colored graph $L\in \textbf{\textrm{L}}\big (Z^0[\ominus]\textbf{\textrm{H}}\big )$ can be decomposed into vertex disjoint vertex-intersected graphs $a_1H_1,a_2H_2,\dots ,a_mH_m$ by the vertex-splitting operation.

By Eq.(\ref{eqa:hypergraphs-hyperedge-coinciding11}) and Eq.(\ref{eqa:hypergraph-vertex-intersected graph-lattice-edge}), we have a \emph{hyperedge-coincided hypergraph lattice} as follows
\begin{equation}\label{eqa:hypergraph-lattice-by-hyperedge-coinciding}
\textbf{\textrm{L}}\big (Z^0[\ominus]\textbf{\textrm{H}}_{yper}\big )=\big \{[\ominus ]^m_{k=1} a_k\mathcal{H}^k_{yper}:a_k\in Z^0,\mathcal{H}^k_{yper}\in \textbf{\textrm{H}}_{yper}\big \}
\end{equation} where $\textbf{\textrm{H}}_{yper}=\left (\mathcal{H}^1_{yper},\mathcal{H}^2_{yper},\dots ,\mathcal{H}^m_{yper}\right )$ is \emph{hypergraph lattice base}.

\begin{thm}\label{thm:666666}
\cite{Yao-Ma-arXiv-2201-13354v1} In a vertex-intersected graph base $\textbf{\textrm{H}}$ of the hyperedge-coincided hypergraph lattice $\textbf{\textrm{L}}\big (Z^0[\ominus]\textbf{\textrm{H}}\big )$ defined in Eq.(\ref{eqa:hypergraph-vertex-intersected graph-lattice-edge}), if each vertex-intersected graph $H_i$ is a tree, then $L=[\ominus ]^m_{k=1}a_kH_k $ is a tree too, and moreover if $\Lambda=\Lambda_i$ for $i\in [1,m]$, $\Lambda(L)=\Lambda$, then $L$ is a vertex-intersected graph of the hypergraph $\mathcal{H}^{\ominus}_{yper}=\big (\Lambda,\bigcup ^m_{k=1}\mathcal{E}_k\big )$.
\end{thm}

\subsubsection{Vertex-splitting and vertex-coinciding operations}

We introduce the concept of vertex-split graphs in the following way: Select randomly a subset $X\subset V(H)$, where the graph $H$ admits a set-coloring $F:V(H)\rightarrow \mathcal{E}$ defined on a hypergraph $\mathcal{H}_{yper}=(\Lambda,\mathcal{E})$, and do the vertex-splitting operation to each vertex $w_i\in X$, such that $w_i$ is split into $w\,'_i$ and $w\,''_i$, and the adjacent neighbor set $N_{ei}(w_i)=N_{ei}(w\,'_i)\cup N_{ei}(w\,''_i)$ with $N_{ei}(w\,'_i)\cap N_{ei}(w\,''_i)=\emptyset$. The resultant graph is denoted as $H\wedge X$, called \emph{vertex-split graph}. Now, we define a set-coloring $F^*$ by $F$ for the vertex-split graph $H\wedge X$ as follows:
\begin{asparaenum}[\textbf{\textrm{Vsc}}-1.]
\item $F^*(z)=F(z)$ for $z\in V(H-X)\cup E(H-X)$.
\item $F^*(xw\,'_i)=F(xw\,'_i)$ for $x\in N_{ei}(w\,'_i)$.
\item $F^*(xw\,''_i)=F(xw\,''_i)$ for $x\in N_{ei}(w\,''_i)$.
\item $F^*(w\,'_i)=F(w_i)$ and $F^*(w\,''_i)=F(w_i)$.

\item $F(w_i)=F^*(w\,'_i)\cup F^*(w\,''_i)$ and $F^*(w\,'_i)\cap F(w\,''_i)=\emptyset$.

\item $F^*(w\,'_i)\subset F(w_i)$ and $F^*(w\,''_i)\subset F(w_i)$, and $F^*(w\,'_i)\cap F^*(w\,''_i)\neq \emptyset$.
\end{asparaenum}

In Fig.\ref{fig:reduce-splitting-coinciding}, we have a vertex-split graph $L\wedge u$, which has two components $L_1$ and $L_2$ shown in Fig.\ref{fig:reduce-splitting-coinciding} (b) and (c), we call $L_1$ and $L_2$ two \emph{partial hypergraphs} of the hypergraph having its vertex-intersected graph $L=L_1[\bullet]L_2$.

If the vertex-split graph $H\wedge X$ is disconnected, that is, the vertex-split graph $H\wedge X$ consists of vertex-disjoint components $H_1,H_2,\dots, H_s$ with $s\geq 2$, and the vertex-split graph $H\wedge X$ admits a set-coloring $F^*$ holding the above \textbf{Vsc}-1, \textbf{Vsc}-2, \textbf{Vsc}-3 and \textbf{Vsc}-4, we call each $H_i$ \emph{partial hypergraph} of the hypergraph
\begin{equation}\label{eqa:555555}
\mathcal{H}_{yper}=(\Lambda,\mathcal{E})=\left (\bigcup ^m_{i=1}\Lambda_i,\bigcup ^s_{i=1}\mathcal{E}_i\right )
\end{equation} where each hypergraph $\mathcal{H}^i_{yper}=(\Lambda_i,\mathcal{E}_i)$ for $i\in [1,s]$.

We do the vertex-coinciding operation to the vertex-split graph $H\wedge X$ admitting the set-coloring $F^*$ by vertex-coinciding each pair of vertices $w\,'_i$ and $w\,''_i$ into one vertex $w_i=w\,'_i\bullet w\,''_i$, and make $F(w_i)=F^*(w\,'_i)\cup F^*(w\,''_i)$, such that the resultant graph is just the original graph $H$, we write
$$H=H_1[\bullet]H_2[\bullet]\cdots [\bullet]H_s=[\bullet]^s_{i=1} H_i,~ \textrm{or}~ H=\odot[H\wedge X]$$

\vskip 0.4cm

\textbf{DC-algorithm for the decomposition and composition of hypergraphs.} We do the vertex-coinciding operation to vertex-disjoint vertex-intersected graphs $G_1,G_2,\dots, G_n$ for getting a new graph
$$
G^*=\odot\big \langle G_1,G_2,\dots,G_n\big \rangle =[\bullet]^n_{k=1} G_k
$$ Suppose that each vertex-intersected graph $G_j$ admits a set-coloring $F_j$ defined on a hyperedge set $\mathcal{E}_j$, we define a set-coloring $F^*$ of $G^*$ as:
\begin{asparaenum}[\textbf{\textrm{Op}}-1]
\item If $z_{i,j}=z_i\bullet z_j$ for $z_i\in V(G_i)$ and $z_j\in V(G_j)$ with $i\neq j$, set $F^*(z_{i,j})=F_i(z_i)\cup F_j(z_j)$.
\item If $z_{i,j}=z_i\bullet z_j$ and $w_{i,j}=w_i\bullet w_j$ for $z_i,w_i\in V(G_i)$ and $z_j,w_j\in V(G_j)$ with $i\neq j$, and $z_iw_i$ is an edge of $G_i$ and $z_jw_j$ is an edge of $G_j$, we obtain a coincided edge $z_{i,j}w_{i,j}=z_iw_i\bullet z_jw_j$ by doing the edge-coinciding operation on these two edges $z_iw_i$ and $z_jw_j$, and moreover we set
 $$F^*(z_{i,j})=F_i(z_i)\cup F_j(z_j),F^*(w_{i,j})=F_i(w_i)\cup F_j(w_j),F^*(z_{i,j}w_{i,j})=F_i(z_iw_i)\cup F_j(z_jw_j)$$
\end{asparaenum}

Since each vertex-intersected graph $G_j$ holds $F_j(V(G_j))=\mathcal{E}_j$ in the hypergraph $\mathcal{H}^j_{yper}=(\Lambda_j,\mathcal{E}_j)$, then $F^*(V(G^*))=\mathcal{E}^*=\bigcup ^n_{j=1}\mathcal{E}_j$. We have the following particular hypergraphs generated from the graph $G^*$:
\begin{asparaenum}[\textbf{\textrm{Hyper}}-1]
\item If $\Lambda=\Lambda_j$ for $j\in [1,n]$, then $\mathcal{E}^*$ is a hyperedge set based on a finite set $\Lambda$, immediately, we get a hypergraph $\mathcal{H}_{yper}=(\Lambda,\mathcal{E}^*)$.

\item If $\Lambda^*=\bigcup _{e\in \mathcal{E}^*}e$, then $G^*$ is a vertex-intersected graph of the hypergraph $\mathcal{H}_{yper}=(\Lambda^*,\mathcal{E}^*)$.
\end{asparaenum}

We have an example $T[\bullet]T_1[\bullet]T_2[\bullet]T_3 \subseteq \mathcal{H}_{yper}$, where the graphs $T,T_1,T_2,T_3$ and a vertex-intersected graph $\mathcal{H}_{yper}$ shown in Fig.\ref{fig:coloring-hypergraph} and Fig.\ref{fig:more-trees-one-hypergraph}.

\vskip 0.4cm

\begin{figure}[h]
\centering
\includegraphics[width=16.4cm]{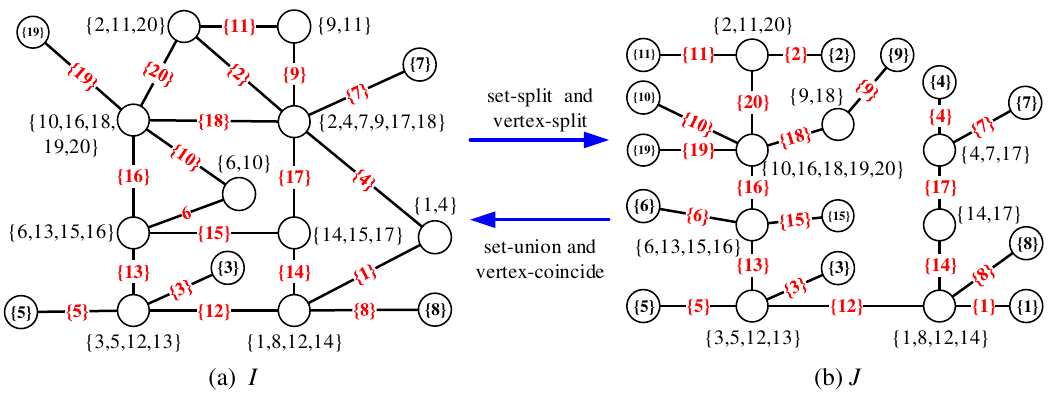}\\
\caption{\label{fig:set-split-vertex-split}{\small An example for the \emph{hyperedgeset-splitting} and \emph{hyperedgeset-coinciding} operations.}}
\end{figure}

\subsection{Vertex-intersected graph lattices}

We consider to build up \emph{vertex-coincided vertex-intersected graph lattices} in this subsection.

Let $\textbf{\textrm{T}}=(T_1,T_2$, $\dots $, $T_n)$ be a \emph{vertex-intersected graph base}, where each vertex-intersected graph $T_j$ admits a set-coloring $\varphi_j$ defined on the hyperedge set $\mathcal{E}_j$ of each hypergraph $\mathcal{T}^j_{yper}=(\Lambda_j,\mathcal{E}_j)$ with $j\in [1,n]$, and let $J_1,J_2,\dots ,J_B$ be a permutation of vertex-intersected graphs $b_1T_1,b_2T_2$, $\dots $, $b_nT_n$, where $B=\sum^n_{k=1} b_k\geq 1$. Now, we do the vertex-coinciding operation to these graphs $J_1,J_2,\dots ,J_B$, such that the resultant graph
\begin{equation}\label{eqa:vertex-coincided-group}
I=[\bullet]\big \langle J_1,J_2,\dots ,J_B\big \rangle=[\bullet]^n_{k=1}b_kT_k
\end{equation}
is a subgraph of a vertex-intersected graph of the hypergraph $\mathcal{H}^{\bullet}_{yper}(I)=(\Lambda(I),\mathcal{E}(I))$, where $\Lambda(I)=\bigcup ^m_{k=1}\Lambda_k$ and $\mathcal{E}(I)=\bigcup ^m_{k=1}\mathcal{E}_k$. We get a \emph{vertex-coincided vertex-intersected graph lattice}
\begin{equation}\label{eqa:hypergraph-vertex-intersected graph-lattice-vertex}
\textbf{\textrm{L}}\big (Z^0\bullet\textbf{\textrm{T}}\big )=\left \{[\bullet]^n_{k=1}b_kT_k :b_k\in Z^0,T_k\in \textbf{\textrm{T}}\right \},~\sum^n_{k=1} b_k\geq 1
\end{equation} such that each set-colored graph $G\in \textbf{\textrm{L}}\big (Z^0\bullet\textbf{\textrm{T}}\big )$ can be decomposed into edge-disjoint vertex-intersected graphs $b_1T_1,b_2T_2$, $\dots $, $b_nT_n$.

However, it is not easy to realize the decomposition of a graph to some particular graphs, in general, since the subgraph isomorphism is NP-complete.

Correspondingly, by Eq.(\ref{eqa:vertex-coincided-group}) and Eq.(\ref{eqa:hypergraph-vertex-intersected graph-lattice-vertex}), we have a \emph{hyperedge-coincided hypergraph lattice}
\begin{equation}\label{eqa:hypergraph-lattice-by-hyperedge-coinciding}
\textbf{\textrm{L}}\big (Z^0\bullet \textbf{\textrm{T}}_{yper}\big )=\big \{[\bullet ]^n_{k=1}a_k\mathcal{T}^k_{yper}:a_k\in Z^0,\mathcal{T}^k_{yper}\in \textbf{\textrm{T}}_{yper}\big \},~\sum ^n_{k=1}a_k\geq 1
\end{equation} where $\textbf{\textrm{T}}_{yper}=(\mathcal{T}^1_{yper},\mathcal{T}^2_{yper},\dots ,\mathcal{T}^n_{yper})$ is a \emph{hyperedge-coincided hypergraph lattice base}.

Moreover, we do mixed operations of the vertex-coinciding operation and the edge-coinciding operation to a permutation $I_1,I_2,\dots ,I_B$ of vertex-intersected graphs $c_1T_1,c_2T_2$, $\dots $, $c_nT_n$, such that the resultant graph is written as
\begin{equation}\label{eqa:555555}
[\bullet\ominus]\big \langle I_1,I_2,\dots ,I_B \big \rangle =[\bullet\ominus]^n_{k=1}c_kT_k
\end{equation} with $\sum^n_{k=1} c_k\geq 1$. We have a \emph{mixed vertex-intersected graph lattice}
\begin{equation}\label{eqa:hypergraph-vertex-intersected graph-lattice-mixed}
\textbf{\textrm{L}}\big (Z^0[\bullet\ominus]\textbf{\textrm{T}}\big )=\big \{[\bullet\ominus]^n_{k=1}c_kT_k :c_k\in Z^0,T_k\in \textbf{\textrm{T}}\big \},~\sum^n_{k=1} c_k\geq 1
\end{equation} and a \emph{mixed-coincided hypergraph lattice}
\begin{equation}\label{eqa:hypergraph-lattice-mixed-hyperedge-coinciding}
\textbf{\textrm{L}}\big (Z^0[\bullet\ominus] \textbf{\textrm{T}}_{yper}\big )=\big \{[\bullet\ominus] ^n_{k=1}a_k\mathcal{T}^k_{yper} :a_k\in Z^0,\mathcal{T}^k_{yper}\in \textbf{\textrm{T}}_{yper}\big \},~\sum^n_{k=1} a_k\geq 1
\end{equation} and the hypergraph lattice base $\textbf{\textrm{T}}_{yper}=\big (\mathcal{T}^1_{yper},\mathcal{T}^2_{yper},\dots ,\mathcal{T}^n_{yper}\big )$ with $\mathcal{T}^i_{yper}\not \subset \mathcal{T}^j_{yper}$ and $\mathcal{T}^i_{yper}\not \cong \mathcal{T}^j_{yper}$ if $i\neq j$.

\vskip 0.4cm

\begin{thm}\label{thm:particular-base-vertex-intersected-graph-lattice}
\cite{Yao-Ma-arXiv-2201-13354v1} In a hyperedge-coincided hypergraph lattice $\textbf{\textrm{L}}\big (Z^0\bullet \textbf{\textrm{T}}_{yper}\big )$ defined in Eq.(\ref{eqa:hypergraph-lattice-by-hyperedge-coinciding}), if $\Lambda=\Lambda_1=\Lambda_2=\cdots =\Lambda_m$, and any pair of two sets $e_i\in \mathcal{E}_i$ and $e_j\in \mathcal{E}_j$ hold $e_i\cap e_j=\emptyset$ as $i\neq j$, each vertex-intersected graph $H_i$ is connected. Then each graph $G=[\ominus ]^m_{k=1}a_kH_k $ with $a_k=1$ or $0$ defined in Eq.(\ref{eqa:hypergraph-vertex-intersected graph-lattice-edge}) is a vertex-intersected graph of the hypergraph $\mathcal{H}_{yper}=\big (\Lambda,\bigcup ^m_{i=1}\mathcal{E}_i \big )$, so is the graph $I=[\bullet ]^n_{k=1}b_kT_k$ with $b_k=1$ or $0$ defined in Eq.(\ref{eqa:hypergraph-vertex-intersected graph-lattice-vertex}), and so is the graph $[\bullet\ominus]^n_{k=1}\langle c_kT_k\rangle $ with $c_k=1$ or $0$ defined in Eq.(\ref{eqa:hypergraph-vertex-intersected graph-lattice-mixed}).
\end{thm}

\begin{rem}\label{rem:333333}
The number of vertex-intersected graphs $G=[\ominus ]^m_{k=1}a_kH_k $ with $a_k=1$ or $0$ is $2^m$ in total, so are the numbers of vertex-intersected graphs $I=[\bullet ]^n_{k=1}b_kT_k$ with $b_k=1$ or $0$ and $[\bullet\ominus]^n_{k=1}c_kT_k$ with $c_k=1$ or $0$ in Theorem \ref{thm:particular-base-vertex-intersected-graph-lattice}. The above operations lead to the \emph{hyperedgeset-splitting} operation and the \emph{hyperedgeset-coinciding} operation of hypergraphs.

It is noticeable, there are many operations on hypergraphs for building up various hypergraph lattices. Planting some results and problems of graphic lattices and traditional lattices into hypergraph lattices is important and meaningful in asymmetric cryptography.\qqed
\end{rem}

\subsection{Isomorphisms of vertex-intersected graphs and hypergraphs }

Recall Kelly-Ulam's Reconstruction Conjecture (1942):
\begin{conj}\label{conj:0000000000}
\cite{Bondy-2008} Let both $G$ and $H$ be graphs with $n$ vertices. If there is a bijection $f: V(G)\rightarrow V(H)$ such that $G-u \cong H-f(u)$ for each vertex $u\in V(G)$, then $G \cong H$.
\end{conj}

\begin{thm}\label{thm:2-vertex-split-graphs-isomorphic}
\cite{Yao-Wang-2106-15254v1} Suppose that two connected graphs $G$ and $H$ admit a coloring $f:V(G)\rightarrow V(H)$. In general, a vertex-split graph $G\wedge u$ with $\textrm{deg}_G(u)\geq 2$ is not unique, so we have a vertex-split graph set $S_G(u)=\{G\wedge u\}$, similarly, we have another vertex-split graph set $S_H(f(u))=\{H\wedge f(u)\}$. If each vertex-split graph $L\in S_G(u)$ corresponds another vertex-split graph $T\in S_H(f(u))$ holding $L\cong T$ true, and vice versa, we write this fact as
\begin{equation}\label{eqa:555555}
G\wedge u\cong H\wedge f(u)
\end{equation} then we claim that $G$ is \emph{isomorphic} to $H$, namely, $G\cong H$.
\end{thm}

\begin{rem}\label{rem:333333}
The vertex-splitting graph $G\wedge u$ with degree $\textrm{deg}_G(u)=m\geq 2$ forms a vertex-split graph set $S_G(u)=\{G\wedge u\}$ in Theorem \ref{thm:2-vertex-split-graphs-isomorphic}. However, determining this graph set $S_G(u)=\{G\wedge u\}$ will meet \emph{Integer Partition Problem}.

For the computational complexity of the Integer Partition Problem, the authors, in \cite{Bing-et-al-arXiv-asymmetric-4520331}, have partitioned a positive integer $m\geq 2$ into a group of $a_i$ parts $m_{i,1},m_{i,2},\dots ,m_{i,a_i}$ holding
$$m=m_{i,1}+m_{i,2}+\cdots +m_{i,a_i}
$$ with each $m_{i,j}>0$ and $a_i\geq 2$. Correspondingly, the vertex $u$ of the graph $G$ is vertex-split into vertices $u_{i,1},u_{i,2},\dots ,u_{i,a_i}$, such that the adjacent neighbor set $N_{ei}(u)=\bigcup^{a_i}_{j=1}N_{ei}(u_{i,j})$ in the vertex-splitting graph $G\wedge u$, where two adjacent neighbor sets $N_{ei}(u_{i,j})\cap N_{ei}(u_{i,k})=\emptyset$ for $j\neq k$. Suppose there are $P_{art}(m)$ groups of such $a_i$ parts. Computing the number $P_{art}(m)$ can be transformed into finding the number $A(m,a_i)$ of solutions of \emph{Diophantine equation} $m=\sum ^k_{i=1}ix_i$. There is a recursive formula
\begin{equation}\label{eqa:c3xxxxx}
A(m,a_i)=A(m,a_i-1)+A(m-a_i,a_i)
\end{equation}
with $0 \leq a_i\leq m$. It is not easy to compute the exact value of $A(m,a_i)$, for example, the authors in \cite{Shuhong-Wu-Accurate-2007} and \cite{WU-Qi-qi-2001} computed exactly
$${
\begin{split}
A(m,6)=&\biggr\lfloor \frac{1}{1036800}(12m^5 +270m^4+1520m^3-1350m^2-19190m-9081)+\\
&\frac{(-1)^m(m^2+9m+7)}{768}+\frac{1}{81}\left[(m+5)\cos \frac{2m\pi}{3}\right ]\biggr\rfloor
\end{split}}
$$

On the other hands, for any odd integer $m\geq 7$, it was conjectured $m=p_1+p_2+p_3$ with three primes $p_1,p_2,p_3$ from the famous Goldbach's conjecture:

\begin{quote}
\emph{Every even integer, greater than 2, can be expressed as the sum of two primes}.
\end{quote}
In other words, determining $A(m,3)$ is difficult, also, it is difficult to express an odd integer $m=\sum^{3n}_{k=1} p\,'_k$ with each $p\,'_k$ is a prime integer.\qqed
\end{rem}

\begin{thm}\label{thm:hypergraphs-vetex-split-isomorphic}
\cite{Yao-Ma-arXiv-2201-13354v1} \textbf{Hypergraph isomorphism.} Suppose that two connected graphs $G$ and $H$ admit a coloring $f:V(G)\rightarrow V(H)$, where $G$ is a vertex-intersected graph of a hypergraph $\mathcal{H}_{yper}=(\Lambda,\mathcal{E})$, and $H$ is a vertex-intersected graph of another hypergraph $\mathcal{H}^*_{yper}=(\Lambda^*,\mathcal{E}^*)$. Vertex-splitting a vertex $u$ of a vertex-intersected graph $G$ with $\textrm{deg}_G(u)\geq 2$ produces a vertex-split graph set $I_G(u)=\{G\wedge u\}$ (resp. a hypergraph set $\{\mathcal{H}_{yper}\wedge e\}$ with $e=F(u)$, since $F:V(G)\rightarrow \mathcal{E}$). Similarly, another vertex-split graph set $I_H(f(u))=\{H\wedge f(u)\}$ is obtained by vertex-splitting a vertex $f(u)$ of a vertex-intersected graph $H$ with $\textrm{deg}_H(f(u))\geq 2$, so we have a hypergraph set $\{\mathcal{H}^*_{yper}\wedge e\,'\}$ with $e\,'=F\,'(f(u))$, since $F\,':V(H)\rightarrow \mathcal{E}^*$. If each vertex-split graph $L\in I_G(u)$ (resp. $\mathcal{L}\in \{\mathcal{H}_{yper}\wedge e\}$) corresponds another vertex-split graph $T\in I_H(f(u))$ (resp. $\mathcal{T}\in \{\mathcal{H}^*_{yper}\wedge e\,'\}$) such that $L\cong T$ (resp. $\mathcal{L}\cong \mathcal{T}$), and vice versa, we write this fact as
\begin{equation}\label{eqa:555555}
G\wedge u\cong H\wedge f(u),\quad (\textrm{resp.} ~\mathcal{H}_{yper}\wedge e\cong \mathcal{H}^*_{yper}\wedge e\,')
\end{equation} then we claim that $G\cong H$ (resp. $\mathcal{H}_{yper}\cong \mathcal{H}^*_{yper}$).
\end{thm}

We present the following isomorphism conjectures:

\begin{conj}\label{conj:0000000000}
\textbf{Graph isomorphism conjectures.}

(i) $^*$ Assume that there are edge subsets $E_G\subset E(G)$ and $E_H\subset E(H)$ with $|E_G|=|E_H|$ such that two \emph{edge-removed graphs}
$$G-E_G\cong H-E_H
$$ for two connected $(p,q)$-graphs $G$ and $H$ admitting the isomorphic subgraph similarity. Then $G\cong H$ by Kelly-Ulam's Reconstruction Conjecture.

(ii) $^*$ If each spanning tree $T_a$ of a connected $(p,q)$-graph $G_a$ corresponds a spanning tree $T_b$ of another connected graph $G_b$ such that $T_a\cong T_b$, and vice versa, then $G_a\cong G_b$.

(iii) \cite{Yao-Ma-arXiv-2201-13354v1} Let both $G$ and $H$ be graphs with $n$ vertices. If each proper subset $S_G\in V(G)\cup E(G)$ corresponds another proper subset $S_H\in V(H)\cup E(H)$ such that two \emph{subset-removed graphs}
$$
G-S_G \cong H-S_H
$$ then $G \cong H$.
\end{conj}

\begin{conj}\label{conj:0000000000}
\cite{Yao-Ma-arXiv-2201-13354v1} \textbf{Hypergraph isomorphism.} Let $\mathcal{H}_{yper}=(\Lambda,\mathcal{E})$ and $\mathcal{H}^*_{yper}=(\Lambda^*,\mathcal{E}^*)$ both be hypergraphs with two cardinalities $|\mathcal{E}|=|\mathcal{E}^*|$. If there is a bijection $\theta: \mathcal{E}\rightarrow \mathcal{E}^*$ such that each hyperedge $e\in \mathcal{E}$ and $\theta(e)\in \mathcal{E}^*$ hold two isomorphic hypergraphs
\begin{equation}\label{eqa:555555}
(\Lambda\setminus\{e_{\cap}\},\mathcal{E}\setminus e)\cong (\Lambda^*\setminus\{\theta(e_{\cap})\},\mathcal{E}^*\setminus \theta(e))
\end{equation} where $e_{\cap}\subset e\in \mathcal{E}$ and $e_{\cap}\cap e\,'=\emptyset$ for any $e\,'\in \mathcal{E}$, then $\mathcal{H}_{yper} \cong \mathcal{H}^*_{yper}$.
\end{conj}

\begin{thm}\label{thm:666666}
By Theorem \ref{thm:graph-set-graph-isomorphism}, two totally colored and connected graphs $G$ and $H$ correspond two totally colored graph sets $G_{raph}(G)$ and $G_{raph}(H)$, such that each totally colored graph $L\in G_{raph}(G)$ holds $L\rightarrow _{color}G$, and each totally colored graph $T\in G_{raph}(H)$ holds $T\rightarrow _{color}H$. Suppose that each totally colored graph $L\,'\in G_{raph}(G)$ corresponds a totally colored graph $H\,'\in G_{raph}(H)$ holding $L\,'\cong H\,'$ true, and vice versa, then we claim that $G\cong H$.
\end{thm}

\section{Properties Of Hypergraphs}

\subsection{Connectivity of hypergraphs}

By the vertex-splitting operation introduced in Definition \ref{defn:vertex-split-coinciding-operations}, we vertex-split each vertex $w$ in a non-empty vertex subset $S$ of a \emph{hyperedge connected vertex-intersected graph} $G$ into two vertices $w\,'$ and $w\,''$, such that the adjacent neighbor set $N_{ei}(w)=N_{ei}(w\,')\cup N_{ei}(w\,'')$ with $N_{ei}(w\,')\cap N_{ei}(w\,'')=\emptyset$, $|N_{ei}(w\,')|\geq 1$ and $|N_{ei}(w\,'')|\geq 1$, the resultant graph is denoted as $G\wedge S$, and let $S\,'=\{w\,':w\in S\}$ and $S\,''=\{w\,'':w\in S\}$, so $V(G\wedge S)=V(G-S)\cup S\,'\cup S\,''$. Since $G$ is the hyperedge connected vertex-intersected graph of a hyperedge connected hypergraph $\mathcal{H}_{yper}=(\Lambda,\mathcal{E})$ (Ref. Definition \ref{defn:more-terminology-group}), so it admits a total set-coloring $F:V(G)\cup E(G)\rightarrow \mathcal{E}$. We define a total set-coloring $F^*$ of the vertex-split graph $G\wedge S$ as:

(A-1) $F^*(x)=F(x)$ for $x\not\in S\,'\cup S\,''$ and $F^*(w\,')=F(w)$ and $F^*(w\,'')=F(w)$ for $w\in S$ and $w\,',w\,''\in S\,'\cup S\,''$; and

(A-2) $F^*(uv)=F(uv)$ for $uv\in E(G-S)$, $F^*(uw\,')=F(uw)$ and $F^*(uw\,'')=F(uw)$ for edges $uw\,'$, $uw\,''\in E(G\wedge S)$.

Suppose that the vertex-split graph $G\wedge S$ has subgraphs $G_1,G_2,\dots ,G_m$, such that

(i) $V(G\wedge S)=\bigcup ^m_{i=1}V(G_i)$ with $V(G_i)\cap V(G_j)=\emptyset$ if $i\neq j$, and

(ii) $E(G\wedge S)=\bigcup ^m_{i=1}E(G_i)$ with $E(G_i)\cap E(G_j)=\emptyset$ if $i\neq j$,\\
then $G\wedge S$ is not hyperedge connected, also, $G\wedge S$ is disconnected in general, and we call $S$ a \emph{vertex-split set-cut-set}. For each subgraph $G_i$, the set-coloring $F$ induces a total set-coloring $F_i:V(G_i)\cup E(G_i)\rightarrow \mathcal{E}_i$ with $\mathcal{E}_i\subset \mathcal{E}$, and call each subgraph $G_i$ \emph{partial hypergraph}.

Conversely, we do the vertex-coinciding operation defined in Definition \ref{defn:vertex-split-coinciding-operations} to the subgraphs $G_1,G_2,\dots ,G_m$ of the vertex-split graph $G\wedge S$ by vertex-coinciding two vertices $w\,'$ and $w\,''$ into one vertex $w=w\,'\bullet w\,''$, and get the original hyperedge connected vertex-intersected graph $G$, we write this case as $G=[\bullet]^m_{k=1} G_k$.

\begin{defn} \label{defn:hypergraph-connectivity}
\cite{Yao-Ma-arXiv-2201-13354v1} \textbf{Hypergraph connectivity.} Suppose that a vertex-split graph $G\wedge S^*$ of the hyperedge connected vertex-intersected graph $G$ of a hypergraph $\mathcal{H}_{yper}=(\Lambda,\mathcal{E})$ holds $|S^*|\leq |S|$ for any vertex-split graph $G\wedge S$, where $G\wedge S$ is not hyperedge connected, then the number $|S^*|$ is called the \emph{hyperedge split-connected number}, written as $n_{vsplit}(G)=n_{vsplit}(\mathcal{H}_{yper})$. Since $G-S^*$ is a disconnected graph having components $G_1-S^*,G_2-S^*,\dots ,G_m-S^*$, that is, the hyperedge connected vertex-intersected graph $G$ is \emph{vertex $|S^*|$-connectivity}, and the hypergraph $\mathcal{H}_{yper}$ is \emph{hyperedge $|\mathcal{E}^*|$-connectivity}, where $\mathcal{E}^*=\{e:F(w)=e\in \mathcal{E}, w\in S^*\}$ makes the hyperedge set $\mathcal{E}\setminus \mathcal{E}^*$ to be \emph{hyperedge disconnected}, we call the hyperedge set $\mathcal{E}^*$ a \emph{hyperedge set-cut-set} of the hypergraph $\mathcal{H}_{yper}$.\qqed
\end{defn}

\begin{rem}\label{rem:333333}
In the view of decomposition, the hyperedge connected vertex-intersected graph $G$ can be decomposed into vertex-disjoint partial hypergraphs $G_1,G_2,\dots ,G_m$.

For $w\in S$ and $w\,',w\,''\in S\,'\cup S\,''$ in (A-1) above, we redefine $F^*(w\,')=F(uw)$ and $F^*(w\,'')=F(w)\setminus F^*(w\,')$, since $F(uw)\subseteq F(w)\cap F(u)$. Thereby, we get more families of subgraphs like $G_1,G_2,\dots ,G_m$, in other words, a hyperedge connected hypergraph can be decomposed into many groups of hyperedge disjoint partial hypergraphs. \qqed
\end{rem}

As known, the vertex-splitting connectivity of a connected graph is equivalent to its own vertex connectivity proven in \cite{Wang-Su-Yao-2021-computer-science}, so we have the following result:

\begin{thm}\label{thm:666666}
The hyperedge split-connected number of a hyperedge connected hypergraph is equal to its own hyperedge connectivity.
\end{thm}

\subsection{Colorings of hypergraphs}

\begin{defn} \label{defn:join-type-edge-set-coloring}
\cite{Yao-Ma-arXiv-2201-13354v1} Let $\mathcal{E}$ be a set of subsets of a finite set $\Lambda$ such that each hyperedge $e\in \mathcal{E}$ satisfies $e\neq \emptyset$ and corresponds some hyperedge $e\,'\in \mathcal{E}$ holding $e\cap e\,'\neq \emptyset$, as well as $\Lambda=\bigcup _{e\in \mathcal{E}}e$. Suppose that a connected graph $H$ admits an \emph{edge set-labeling} $F\,': E(H)\rightarrow \mathcal{E}$ holding $F\,'(uv)\neq F\,'(uw)$ for any two adjacent edges $uv,uw\in E(H)$, and the vertex color set $F\,'(w)$ for each vertex $w\in V(G)$ is induced by one of the following cases:
\begin{asparaenum}[\textbf{\textrm{Edgeinduce}}-1.]
\item $F\,'(w)=\{F\,'(wz):z\in N_{ei}(w)\}\subseteq \Lambda^2$.
\item $F\,'(w)=\bigcup _{z\in N_{ei}(w)}F\,'(wz)\subseteq \Lambda$.
\end{asparaenum}
We call $H$ an \emph{edge-set-colored graph}. The \emph{edge-induced graph} $L_H$ of the edge-set-colored graph $H$ has its own vertex set $V(L_H)= E(H)$, and admits a vertex set-coloring $F\,': V(L_H)\rightarrow \mathcal{E}$, such that two vertices $w_{uv}(:=uv)$ and $w_{xy}(:=xy)$ of $V(L_H)$ are the ends of an edge of $L_H$ if and only if $F\,'(w_{uv})\cap F\,'(w_{xy})\neq \emptyset$ (i.e. $F\,'(uv)\neq F\,'(xy)$ in $H$). So, this edge-induced graph $L_H$ is a \emph{vertex-intersected graph} of the hypergraph $\mathcal{H}_{yper}=(\Lambda,\mathcal{E})$.\qqed
\end{defn}

\begin{thm}\label{thm:hyperedge-vs-vertex-coloring}
A proper hyperedge coloring of a hypergraph $\mathcal{H}_{yper}=(\Lambda,\mathcal{E})$ is equivalent to a proper vertex-coloring of a vertex-intersected graph $G$ of the hypergraph $\mathcal{H}_{yper}$, and vice versa. Thereby, we have $\chi(G)=\chi\,'(\mathcal{H}_{yper})$, where $\chi(G)$ is the \emph{chromatic number} of the graph $G$, and $\chi\,'(\mathcal{E})$ is the \emph{hyperedge chromatic index} of the hypergraph $\mathcal{H}_{yper}$.
\end{thm}

\begin{rem}\label{rem:333333}
In Definition \ref{defn:join-type-edge-set-coloring}, an edge-set-colored graph $H$ may be a subgraph of a vertex-intersected graph of a hypergraph $\mathcal{H}_{yper}=(\Lambda,\mathcal{E}^*)$ with $\mathcal{E}^*\neq \mathcal{E}$ because of \textbf{Edgeinduce}-1 and \textbf{Edgeinduce}-2 defined in Definition \ref{defn:join-type-edge-set-coloring}. The \emph{edge-induced graph} $L_H$ is not the \emph{line graph} of the edge-set-colored graph $H$.

Theorem \ref{thm:hyperedge-vs-vertex-coloring} tells us: The proper hyperedge coloring problem of a hypergraph is a NP-type problem, since there is a well-known conjecture in the proper vertex-coloring of graphs, that is, Bruce Reed in 1998 conjectured: The \emph{chromatic number} $\chi (G)\leq \left \lceil \frac{\Delta(G)+1+K(G)}{2}\right \rceil $, where $\Delta(G)$ is the \emph{maximum degree} of the graph $G$ and $K(G)$ is the \emph{maximum clique number} of the graph $G$.\qqed
\end{rem}

\begin{problem}\label{qeu:444444}
\textbf{How} to color the vertices of the hypergraph $\mathcal{H}_{yper}=(\Lambda,\mathcal{E})$ such that each edge $e$ of $\mathcal{E}$ contains vertices colored with differen colors from each other.
\end{problem}

\begin{rem}\label{rem:333333}
For a given hyperedge set $\mathcal{E}$, we color the vertices of the vertex set $\Lambda$ with $k$ colors such that each set $e\in \mathcal{E}$ contains two vertices colored with different colors if the cardinality $|e|\geq 2$. Clearly, different hyperedge sets correspond different vertex-colorings of the vertex set $\Lambda$. The number
\begin{equation}\label{eqa:555555}
\chi(\Lambda,\mathcal{E})=\min \{k:\textrm{ each $k$-coloring of $\Lambda$ based on a hyperedge set $\mathcal{E}$}\}
\end{equation} is called the \emph{hypervertex chromatic number} of the hypergraph $\mathcal{H}_{yper}=(\Lambda,\mathcal{E})$.\qqed
\end{rem}

\begin{defn} \label{defn:hyperedge-hyper-total-coloring}
\cite{Yao-Ma-arXiv-2201-13354v1} A \emph{hyper-total coloring} $\theta$ of a hypergraph $\mathcal{H}_{yper}=(\Lambda,\mathcal{E})$ is defined by

(i) $\theta:\mathcal{E}\rightarrow [1,b]$, and $\theta(e_i)\neq \theta(e_j)$ if $e_i\cap e_j\neq \emptyset$;

(ii) $\theta(x_{i,j})\in [a,b]$, and $\theta(x_{i,j})\neq \theta(x_{i,k})$ for some distinct $x_{i,j},x_{i,k}\in e_i$ if $|e_i|\geq 2$.\\
And $\chi\,''(\Lambda,\mathcal{E})$ is the \emph{smallest number} of $b$ for which $\mathcal{H}_{yper}$ admits a hyper-total coloring.

A \emph{hyperedge coloring} $\varphi: \mathcal{E}\rightarrow [1,M]$, such that the elements of hyperedge set
$$
N_{ei}(e_i)=\big \{e_j: e_i\cap e_j\neq \emptyset,e_j\in \mathcal{E}\setminus \{e_i\} \big \}
$$ are colored different colors from each other, and the largest number $\Delta(\mathcal{E}_{\cap})=\max\{|N_{ei}(e_i)|:e_i\in \mathcal{E}\}$ holds the following inequalities
\begin{equation}\label{eqa:555555}
\Delta(\mathcal{E}_{\cap})\leq M\leq \Delta(\mathcal{E}_{\cap})+1
\end{equation} true by the famous Vizing's theorem on the edge coloring of a vertex-intersected graph of the hypergraph $\mathcal{H}_{yper}=(\Lambda,\mathcal{E})$. \qqed
\end{defn}

\subsection{Hyperedge-set colorings}

\begin{defn} \label{defn:distinguishing-hyperedge-set-colorings}
$^*$ Let $G$ be a $(p,q)$-graph, and let $\Lambda$ be a finite set of numbers. There is a \emph{hyperedge-set coloring} $F: S\rightarrow \mathcal{E}$, where $\mathcal{E}\in \mathcal{E}(\Lambda^2)$ holds $\Lambda=\bigcup _{e\in \mathcal{E}}e$, and $S\subseteq V(G)\cup E(G)$. There are the following constraints:
\begin{asparaenum}[\textbf{\textrm{Hyset}}-1.]
\item \label{hyper:vertex} $S=V(G)$.
\item \label{hyper:edge} $S=E(G)$.
\item \label{hyper:total} $V(G)\cup E(G)$.

--- \emph{local distinguishing}

\item \label{hyper:adjacent-vertex} $F(u)\neq F(v)$ for each edge $uv\in E(G)$.
\item \label{hyper:adjacent-edge} $F(uv)\neq F(uw)$ for adjacent edges $uv,uw\in E(G)$ and $u\in V(G)$.
\item \label{hyper:incident-edge-vertex} $F(u)\neq F(uv)$ and $F(v)\neq F(uv)$ for each edge $uv\in E(G)$.

--- \emph{local intersected}

\item \label{hyper:join-oper-verticeice} $F(u)\cap F(v)\neq \emptyset$ for each edge $uv\in E(G)$.
\item \label{hyper:join-oper-adjacent-edges} $F(uv)\cap F(uw)\neq \emptyset$ for adjacent edges $uv,uw\in E(G)$.
\item \label{hyper:join-edge} $F(u)\cap F(v)\subseteq F(uv)$ and $F(u)\cap F(v)\neq \emptyset $ for each edge $uv\in E(H)$.
\item \label{hyper:join-oper-vertex-edge} $F(uv)\cap F(u)\neq \emptyset$ and $F(uv)\cap F(v)\neq \emptyset$ for each edge $uv\in E(G)$.

--- \emph{v-adjacent distinguishing}

\item \label{hyper:adjacent-vertex-union-dis} $\bigcup _{v\in N_{ei}(u)}F(v)\neq \bigcup _{z\in N_{ei}(w)}F(z)$ for each edge $uw\in V(H)$.
\item \label{hyper:adjacent-vertex-intersect-dis} $\bigcap _{v\in N_{ei}(u)}F(v)\neq \bigcap _{z\in N_{ei}(w)}F(z)$ for each edge $uw\in V(H)$.
\item \label{hyper:adjacent-all-vertex-union-dis} $F(u)\cup \big [\bigcup _{v\in N_{ei}(u)}F(v)\big ]\neq F(w)\cup \big [\bigcup _{z\in N_{ei}(w)}F(z)\big ]$ for each edge $uw\in V(H)$.
\item \label{hyper:adjacent-all-vertex-intersect-dis} $F(u)\cap \big [\bigcap _{v\in N_{ei}(u)}F(v)\big ]\neq F(w)\cap \big [\bigcap _{z\in N_{ei}(w)}F(z)\big ]$ for each edge $uw\in V(H)$.

--- \emph{e-adjacent distinguishing}

\item \label{hyper:adjacent-edges-union-dis} $\bigcup _{v\in N_{ei}(u)}F(uv)\neq \bigcup _{z\in N_{ei}(w)}F(wz)$ for each edge $uw\in V(H)$.
\item \label{hyper:adjacent-edges-intersect-dis} $\bigcap _{v\in N_{ei}(u)}F(uv)\neq \bigcap _{z\in N_{ei}(w)}F(wz)$ for each edge $uw\in V(H)$.

--- \emph{ve-adjacent distinguishing}

\item \label{hyper:adjacent-vertex-edges-union-dis} $F(u)\cup \big [\bigcup _{v\in N_{ei}(u)}F(uv)\big ]\neq F(w)\cup \big [\bigcup _{z\in N_{ei}(w)}F(wz)\big ]$ for each edge $uw\in V(H)$.
\item \label{hyper:adjacent-vertex-edges-intersect-dis} $F(u)\cap \big [\bigcap _{v\in N_{ei}(u)}F(uv)\big ]\neq F(w)\cap \big [\bigcap _{z\in N_{ei}(w)}F(wz)\big ]$ for each edge $uw\in V(H)$.
\end{asparaenum}
\textbf{Then, we have:}

\noindent --- \emph{hyperedge-set colorings}

\begin{asparaenum}[\textbf{\textrm{Scolo}}-1.]
\item $F$ is called \emph{proper hyperedge-set coloring} if the constraints Hyset-\ref{hyper:vertex}, and Hyset-\ref{hyper:adjacent-vertex} hold true.
\item $F$ is called \emph{proper edge hyperedge-set coloring} if the constraints Hyset-\ref{hyper:edge}, and Hyset-\ref{hyper:adjacent-edge} hold true.
\item $F$ is called \emph{proper total hyperedge-set coloring} if the constraints Hyset-\ref{hyper:total}, Hyset-\ref{hyper:adjacent-vertex}, Hyset-\ref{hyper:adjacent-edge}, and Hyset-\ref{hyper:incident-edge-vertex} hold true.

--- \emph{intersected hyperedge-set colorings}

\item $F$ is called \emph{$v$-intersected proper hyperedge-set coloring} if the constraints Hyset-\ref{hyper:vertex}, Hyset-\ref{hyper:adjacent-vertex} and Hyset-\ref{hyper:join-oper-verticeice} hold true.
\item $F$ is called \emph{$e$-intersected proper edge hyperedge-set coloring} if the constraints Hyset-\ref{hyper:edge}, Hyset-\ref{hyper:adjacent-edge} and Hyset-\ref{hyper:join-oper-adjacent-edges} hold true.
\item $F$ is called \emph{$ee$-intersected proper total hyperedge-set coloring} if the constraints Hyset-\ref{hyper:total}, Hyset-\ref{hyper:adjacent-vertex}, Hyset-\ref{hyper:adjacent-edge}, Hyset-\ref{hyper:incident-edge-vertex}, Hyset-\ref{hyper:join-oper-verticeice}, Hyset-\ref{hyper:join-oper-adjacent-edges} and Hyset-\ref{hyper:join-oper-vertex-edge} hold true.

--- \emph{vertex-distinguishing}

\item $F$ is called \emph{$v$-union adjacent-$v$ distinguishing proper hyperedge-set coloring} if the constraints Hyset-\ref{hyper:vertex}, Hyset-\ref{hyper:adjacent-vertex} and Hyset-\ref{hyper:adjacent-vertex-union-dis} hold true.
\item $F$ is called \emph{$v$-intersected adjacent-$v$ distinguishing proper hyperedge-set coloring} if the constraints Hyset-\ref{hyper:vertex}, Hyset-\ref{hyper:adjacent-vertex} and Hyset-\ref{hyper:adjacent-vertex-intersect-dis} hold true.
\item $F$ is called \emph{$[v]$-union adjacent-$v$ distinguishing proper hyperedge-set coloring} if the constraints Hyset-\ref{hyper:vertex}, Hyset-\ref{hyper:adjacent-vertex} and Hyset-\ref{hyper:adjacent-all-vertex-union-dis} hold true.
\item $F$ is called \emph{$[v]$-intersected adjacent-$v$ distinguishing proper hyperedge-set coloring} if the constraints Hyset-\ref{hyper:vertex}, Hyset-\ref{hyper:adjacent-vertex} and Hyset-\ref{hyper:adjacent-all-vertex-intersect-dis} hold true.

--- \emph{edge-distinguishing}

\item $F$ is called \emph{$v$-union adjacent-$v$ distinguishing proper edge hyperedge-set coloring} if the constraints Hyset-\ref{hyper:edge}, Hyset-\ref{hyper:adjacent-edge} and Hyset-\ref{hyper:adjacent-edges-union-dis} hold true.
\item $F$ is called \emph{$v$-intersected adjacent-$v$ distinguishing proper edge hyperedge-set coloring} if the constraints Hyset-\ref{hyper:edge}, Hyset-\ref{hyper:adjacent-edge} and Hyset-\ref{hyper:adjacent-edges-intersect-dis} hold true.
\item $F$ is called \emph{$(e,v)$-intersected adjacent-$v$ distinguishing proper edge hyperedge-set coloring} if the constraints Hyset-\ref{hyper:edge}, Hyset-\ref{hyper:adjacent-edge}, Hyset-\ref{hyper:join-oper-adjacent-edges} and Hyset-\ref{hyper:adjacent-edges-intersect-dis} hold true.

--- \emph{total-distinguishing}

\item $F$ is called \emph{$v$-union adjacent-$v$ distinguishing proper total hyperedge-set coloring} if the constraints Hyset-\ref{hyper:total}, Hyset-\ref{hyper:adjacent-vertex}, Hyset-\ref{hyper:adjacent-edge}, and Hyset-\ref{hyper:incident-edge-vertex} and Hyset-\ref{hyper:adjacent-vertex-union-dis} hold true.
\item $F$ is called \emph{$v$-intersected adjacent-$v$ distinguishing proper total hyperedge-set coloring} if the constraints Hyset-\ref{hyper:total}, Hyset-\ref{hyper:adjacent-vertex}, Hyset-\ref{hyper:adjacent-edge}, and Hyset-\ref{hyper:incident-edge-vertex} and Hyset-\ref{hyper:adjacent-vertex-intersect-dis} hold true.

\item $F$ is called \emph{$[v]$-union adjacent-$v$ distinguishing proper total hyperedge-set coloring} if the constraints Hyset-\ref{hyper:total}, Hyset-\ref{hyper:adjacent-vertex}, Hyset-\ref{hyper:adjacent-edge}, and Hyset-\ref{hyper:incident-edge-vertex} and Hyset-\ref{hyper:adjacent-all-vertex-union-dis} hold true.
\item $F$ is called \emph{$[v]$-intersected adjacent-$v$ distinguishing proper total hyperedge-set coloring} if the constraints Hyset-\ref{hyper:total}, Hyset-\ref{hyper:adjacent-vertex}, Hyset-\ref{hyper:adjacent-edge}, and Hyset-\ref{hyper:incident-edge-vertex} and Hyset-\ref{hyper:adjacent-all-vertex-intersect-dis} hold true.
\item $F$ is called \emph{$e$-union adjacent-$v$ distinguishing proper total hyperedge-set coloring} if the constraints Hyset-\ref{hyper:total}, Hyset-\ref{hyper:adjacent-vertex}, Hyset-\ref{hyper:adjacent-edge}, and Hyset-\ref{hyper:incident-edge-vertex} and Hyset-\ref{hyper:adjacent-edges-union-dis} hold true.
\item $F$ is called \emph{$e$-intersected adjacent-$v$ distinguishing proper total hyperedge-set coloring} if the constraints Hyset-\ref{hyper:total}, Hyset-\ref{hyper:adjacent-vertex}, Hyset-\ref{hyper:adjacent-edge}, and Hyset-\ref{hyper:incident-edge-vertex} and Hyset-\ref{hyper:adjacent-edges-intersect-dis} hold true.

\item $F$ is called \emph{$[ve]$-union adjacent-$v$ distinguishing proper total hyperedge-set coloring} if the constraints Hyset-\ref{hyper:total}, Hyset-\ref{hyper:adjacent-vertex}, Hyset-\ref{hyper:adjacent-edge}, and Hyset-\ref{hyper:incident-edge-vertex} and Hyset-\ref{hyper:adjacent-vertex-edges-union-dis} hold true.
\item $F$ is called \emph{$[ve]$-intersected adjacent-$v$ distinguishing proper total hyperedge-set coloring} if the constraints Hyset-\ref{hyper:total}, Hyset-\ref{hyper:adjacent-vertex}, Hyset-\ref{hyper:adjacent-edge}, and Hyset-\ref{hyper:incident-edge-vertex} and Hyset-\ref{hyper:adjacent-vertex-edges-intersect-dis} hold true.

\item $F$ is called \emph{$ee$-intersected $v$-union adjacent-$v$ distinguishing proper total hyperedge-set coloring} if the constraints Hyset-\ref{hyper:total}, Hyset-\ref{hyper:adjacent-vertex}, Hyset-\ref{hyper:adjacent-edge}, Hyset-\ref{hyper:incident-edge-vertex}, Hyset-\ref{hyper:join-oper-verticeice}, Hyset-\ref{hyper:join-oper-adjacent-edges}, Hyset-\ref{hyper:join-oper-vertex-edge} and Hyset-\ref{hyper:adjacent-vertex-union-dis} hold true.
\item $F$ is called \emph{$ee$-$v$-intersected adjacent-$v$ distinguishing proper total hyperedge-set coloring} if the constraints Hyset-\ref{hyper:total}, Hyset-\ref{hyper:adjacent-vertex}, Hyset-\ref{hyper:adjacent-edge}, Hyset-\ref{hyper:incident-edge-vertex}, Hyset-\ref{hyper:join-oper-verticeice}, Hyset-\ref{hyper:join-oper-adjacent-edges}, Hyset-\ref{hyper:join-oper-vertex-edge} and Hyset-\ref{hyper:adjacent-vertex-intersect-dis} hold true.
\item $F$ is called \emph{$ee$-intersected $[v]$-union adjacent-$v$ distinguishing proper total hyperedge-set coloring} if the constraints Hyset-\ref{hyper:total}, Hyset-\ref{hyper:adjacent-vertex}, Hyset-\ref{hyper:adjacent-edge}, Hyset-\ref{hyper:incident-edge-vertex}, Hyset-\ref{hyper:join-oper-verticeice}, Hyset-\ref{hyper:join-oper-adjacent-edges}, Hyset-\ref{hyper:join-oper-vertex-edge} and Hyset-\ref{hyper:adjacent-all-vertex-union-dis} hold true.
\item $F$ is called \emph{$ee$-$[v]$-intersected adjacent-$v$ distinguishing proper total hyperedge-set coloring} if the constraints Hyset-\ref{hyper:total}, Hyset-\ref{hyper:adjacent-vertex}, Hyset-\ref{hyper:adjacent-edge}, Hyset-\ref{hyper:incident-edge-vertex}, Hyset-\ref{hyper:join-oper-verticeice}, Hyset-\ref{hyper:join-oper-adjacent-edges}, Hyset-\ref{hyper:join-oper-vertex-edge} and Hyset-\ref{hyper:adjacent-all-vertex-intersect-dis} hold true.
\item $F$ is called \emph{$ee$-intersected $e$-union adjacent-$v$ distinguishing proper total hyperedge-set coloring} if the constraints Hyset-\ref{hyper:total}, Hyset-\ref{hyper:adjacent-vertex}, Hyset-\ref{hyper:adjacent-edge}, Hyset-\ref{hyper:incident-edge-vertex}, Hyset-\ref{hyper:join-oper-verticeice}, Hyset-\ref{hyper:join-oper-adjacent-edges}, Hyset-\ref{hyper:join-oper-vertex-edge} and Hyset-\ref{hyper:adjacent-edges-union-dis} hold true.
\item $F$ is called \emph{$ee$-$e$-intersected adjacent-$v$ distinguishing proper total hyperedge-set coloring} if the constraints Hyset-\ref{hyper:total}, Hyset-\ref{hyper:adjacent-vertex}, Hyset-\ref{hyper:adjacent-edge}, Hyset-\ref{hyper:incident-edge-vertex}, Hyset-\ref{hyper:join-oper-verticeice}, Hyset-\ref{hyper:join-oper-adjacent-edges}, Hyset-\ref{hyper:join-oper-vertex-edge} and Hyset-\ref{hyper:adjacent-edges-intersect-dis} hold true.
\item $F$ is called \emph{$ee$-intersected $[ve]$-union adjacent-$v$ distinguishing proper total hyperedge-set coloring} if the constraints Hyset-\ref{hyper:total}, Hyset-\ref{hyper:adjacent-vertex}, Hyset-\ref{hyper:adjacent-edge}, Hyset-\ref{hyper:incident-edge-vertex}, Hyset-\ref{hyper:join-oper-verticeice}, Hyset-\ref{hyper:join-oper-adjacent-edges}, Hyset-\ref{hyper:join-oper-vertex-edge} and Hyset-\ref{hyper:adjacent-vertex-edges-union-dis} hold true.
\item $F$ is called \emph{$ee$-$[ve]$-intersected adjacent-$v$ distinguishing proper total hyperedge-set coloring} if the constraints Hyset-\ref{hyper:total}, Hyset-\ref{hyper:adjacent-vertex}, Hyset-\ref{hyper:adjacent-edge}, Hyset-\ref{hyper:incident-edge-vertex}, Hyset-\ref{hyper:join-oper-verticeice}, Hyset-\ref{hyper:join-oper-adjacent-edges}, Hyset-\ref{hyper:join-oper-vertex-edge} and Hyset-\ref{hyper:adjacent-vertex-edges-intersect-dis} hold true.\qqed
\end{asparaenum}
\end{defn}

\subsection{Compound hypergraphs}

From studying relationship between communities in networks, which is the topological structure between hypergraphs, we propose the following two concepts of set-set-coloring and compound hypergraphs:

\begin{defn} \label{defn:set-set-coloring}
\cite{Yao-Ma-arXiv-2201-13354v1} A graph $G$ admits a proper \emph{set-set-coloring} $\theta:V(G)\rightarrow \{\mathcal{S}_i\}^n_{i=1}$ with $\theta(x)\neq \theta(y)$ for each edge $xy\in E(G)$, where each $\mathcal{S}_i$ is a set of subsets of the power set $\Lambda^2$ based on a finite set $\Lambda$, such that each induced edge color set is defined as $\theta(u_iv_j)=\theta(u_i)[\bullet ]\theta(v_j)=\mathcal{S}_i[\bullet ] \mathcal{S}_j$ subject to a constraint set $R_{est}(c_1,c_2,\dots, c_m)$ based on an abstract operation ``$[\bullet ]$''. In particularly, each edge is colored with an induced set $\theta(u_iv_j)$ holding
$$
\theta(u_iv_j)\supseteq \theta(u_i)\cap \theta(v_j)=\mathcal{S}_i\cap \mathcal{S}_j\neq\emptyset
$$ when the operation ``$[\bullet]$'' $=$ ``$\bigcap $'' is the intersection operation on sets.\qqed
\end{defn}

\begin{defn} \label{defn:compound-set-coloring-hypergraph}
\cite{Yao-Ma-arXiv-2201-13354v1} Suppose that a graph $G$ admits a proper \emph{compound set-coloring} $\Gamma:V(G)\rightarrow \{\mathcal{E}_i\}^n_{i=1}$ with $\Gamma(x)\neq \Gamma(y)$ for each edge $xy\in E(G)$, where each $\mathcal{E}_i$ is a set of subsets of a finite set $\Lambda$, and $\Lambda=\bigcup^n_{i=1}\Lambda_i$ with $\Lambda_i=\bigcup_{e_{i,j}\in \mathcal{E}_i}e_{i,j}$, and each $(\Lambda_i, \mathcal{E}_i)$ is a hypergraph. If there are a function $\psi$ and a constraint set $R_{est}(c_1,c_2,\dots, c_m)$ such that each edge $u_iv_j$ is colored with an induced edge color set
\begin{equation}\label{eqa:compound-set-coloring-hypergraph-11}
\Gamma(u_iv_j)=\psi(\Gamma(u_i), \Gamma(v_j))\supseteq \Gamma(u_i)\cap \Gamma(v_j)=\mathcal{E}_{u_i}\cap \mathcal{E}_{u_j}\neq\emptyset
\end{equation} or there is an abstract operation ``$[\bullet ]$'' such that each induced edge color set
\begin{equation}\label{eqa:compound-set-coloring-hypergraph-22}
\Gamma(u_iv_j)\supseteq \Gamma(u_i)[\bullet ]\Gamma(v_j)=\mathcal{E}_{u_i}[\bullet ] \mathcal{E}_{u_j}\neq \emptyset
\end{equation} subject to the constraint set $R_{est}(c_1,c_2,\dots, c_m)$, then we call $\mathcal{H}_{comp}=\big (\Lambda, \bigcup^n_{i=1}\mathcal{E}_i\big )$ \emph{compound hypergraph}, and the graph $G$ \emph{compound vertex-intersected graph} of the compound hypergraph $\mathcal{H}_{comp}$ if each $\mathcal{E}_{u_i}\cap \mathcal{E}_{u_j}\neq\emptyset$ corresponds an edge $u_iv_j$ of the graph $G$. \qqed
\end{defn}

\begin{example}\label{exa:8888888888}
Let $K_{2n}$ be a complete graph of $2n$ vertices. Then $K_{2n}$ has perfect matching groups $\textbf{\textrm{M}}_i=\{M_{i,1},M_{i,2}$, $\dots $, $M_{i,2n-1}\}$ for $i\in [1,m]$. There is a graph $G$ admitting a set-coloring $F:V(G)\rightarrow \mathcal{E}_i$ with each hyperedge set $\mathcal{E}_i=\textbf{\textrm{M}}_i$ such that $F(x)\neq F(y)$ for distinct vertices $x,y\in V(G)$, and each edge $uv\in E(G)$ is colored with an induced set $F(uv)=F(u)\cup F(v)$ if $E(C_k)=F(u)\cup F(v)$, where $C_k$ is a Hamilton cycle of $K_{2n}$. If $G$ is a complete graph, then it shows the \textbf{Perfect 1-Factorization Conjecture} (Anton Kotzig, 1964):

``\emph{For integer $n \geq 2$, $K_{2n}$ can be decomposed into $2n-1$ perfect matchings such that the union of any two matchings forms a hamiltonian cycle of $K_{2n}$}''.

Here, the union operation ``$\bigcup $'' is the abstract operation ``$[\bullet ]$'' appeared in Definition \ref{defn:compound-set-coloring-hypergraph}.\qqed
\end{example}

\begin{defn} \label{defn:hypergraph-type-topcode-matrix}
\cite{Yao-Ma-arXiv-2201-13354v1} According to Definition \ref{defn:compound-set-coloring-hypergraph}, the compound vertex-intersected graph $G$ admits a compound set-coloring $\Gamma:V(G)\rightarrow \{\mathcal{E}_i\}^n_{i=1}$ holding Eq.(\ref{eqa:compound-set-coloring-hypergraph-11}), so the compound vertex-intersected graph $G$ corresponds its own \emph{hypergraph Topcode-matrix}
\begin{equation}\label{eqa:hypergraph-type-Topcode-matrix}
{
\begin{split}
H^{comp}_{yper}(G)&=\left(
\begin{array}{ccccc}
\Gamma(x_1) & \Gamma(x_2) & \cdots & \Gamma(x_q)\\
\Gamma(x_1y_1) & \Gamma(x_2y_2) & \cdots & \Gamma(x_qy_q)\\
\Gamma(y_1) & \Gamma(y_2) & \cdots & \Gamma(y_q)
\end{array}
\right)=\left(
\begin{array}{ccccc}
\mathcal{E}_{x_1} & \mathcal{E}_{x_2} & \cdots & \mathcal{E}_{x_q}\\
\mathcal{E}_{x_1y_1} & \mathcal{E}_{x_2y_2} & \cdots & \mathcal{E}_{x_qy_q}\\
\mathcal{E}_{y_1} & \mathcal{E}_{y_2} & \cdots & \mathcal{E}_{y_q}
\end{array}
\right)\\
&=(X(\mathcal{E}),~E(\mathcal{E}),~Y(\mathcal{E}))^{T}
\end{split}}
\end{equation}
where $E(G)=\{x_iy_i:i\in [1,q]\}$, $X(\mathcal{E})=(\mathcal{E}_{x_1},\mathcal{E}_{x_2}, \cdots ,\mathcal{E}_{x_q})$ and $Y(\mathcal{E})=(\mathcal{E}_{y_1},\mathcal{E}_{y_2}, \cdots ,\mathcal{E}_{y_q})$ are called \emph{v-hypergraph vectors}, and $E(\mathcal{E})=(\mathcal{E}_{x_1y_1}, \mathcal{E}_{x_2y_2}, \cdots ,\mathcal{E}_{x_qy_q})$ is called \emph{e-hypergraph intersection vector}, as well as $\mathcal{E}_{x_iy_i}=\psi(\mathcal{E}_{x_i},\mathcal{E}_{y_i})$ for each edge $x_iy_i\in E(G)$.\qqed
\end{defn}

\begin{rem}\label{rem:333333}
The matrix $H^{comp}_{yper}(G)$ defined in Eq.(\ref{eqa:hypergraph-type-Topcode-matrix}) of Definition \ref{defn:hypergraph-type-topcode-matrix} is a three dimensional matrix, and the compound vertex-intersected graph $G$ has its own vertices as hypergraphs, its own edges as intersections of hypergraphs.\qqed
\end{rem}

\subsection{Graphic groups based on hypergraphs}

\begin{defn} \label{defn:set-colored-graphic-group}
\cite{Yao-Ma-arXiv-2201-13354v1} If a set-colored graph set $F_{\mathcal{E}}(G)=\{G_1,G_2,\dots ,G_n\}$ holds:

(i) Each graph $G_i$ is isomorphic to $G_1$, i.e. $G_i\cong G_1$.

(ii) Each set-colored graph $G_i$ admits a total hyperedge set-coloring $F_i:V(G_i)\cup E(G_i)\rightarrow \mathcal{E}_i$, where $\mathcal{E}_i$ is a hyperedge set belonging to the hypegraph set $\mathcal{E}(\Lambda^2_i)$ defined on a consecutive integer set $\Lambda_i$ for $i\in [1,n]$.

(iii) There is a positive integer $M=\max |\Lambda_1|$, such that the color $F_i(w_s)=\{b_{i,s,1},b_{i,s,2},\dots $, $b_{i,s,c(i,s)}\}$ of each element $w_s$ of $V(G_i)\cup E(G_i)$ is defined as $b_{i,s,r}=b_{1,s,r}+i-1~(\bmod ~M)$ with $r\in [1,c(i,s)]$ and $s\in [1,n]$.

(iv) The finite module Abelian additive operation
$$G_i [+_k] G_j:=G_i [+] G_j[-]G_k
$$ is defined by
\begin{equation}\label{eqa:set-colored-graphic-groups}
b_{i,s,r}+b_{j,s,r}-b_{k,s,r}=b_{\lambda,s,r}
\end{equation} with $\lambda=i+j-k~(\bmod ~M)$ for some $b_{i,s,r}\in F_i(w_s)$, $b_{j,s,r}\in F_j(w_s)$ and $b_{\lambda,s,r}\in F_{\lambda}(w_s)$, as well as $b_{k,s,r}\in F_k(w_s)$, where $G_k$ is a preappointed \emph{zero}.

Thereby, we call the set-colored graph set $F_{\mathcal{E}}(G)$ \emph{every-zero set-colored graphic group}, rewrite it as $\{F_{\mathcal{E}}(G)$; $[+][-]\}$.\qqed
\end{defn}

An every-zero set-colored graphic group is shown in Fig.\ref{fig:set-coloring-group}.

We show the following proofs for the every-zero set-colored graphic group $\{F_{\mathcal{E}}(G);[+][-]\}$ defined in Definition \ref{defn:set-colored-graphic-group}:

\textbf{Zero.} Any set-colored graph $G_k\in \{F_{\mathcal{E}}(G);[+][-]\}$ can be as the \emph{zero}, so that $G_i [+_k] G_k=G_i\in \{F_{\mathcal{E}}(G);[+][-]\}$ according to Eq.(\ref{eqa:set-colored-graphic-groups}).

\textbf{Uniqueness and closureness.} If two set-colored graphs $G_i,G_j\in \{F_{\mathcal{E}}(G);[+][-]\}$ hold $G_i [+_k] G_j=G_s$ and $G_i [+_k] G_j=G_r$, then $G_s=G_r\in \{F_{\mathcal{E}}(G);[+][-]\}$ by Eq.(\ref{eqa:set-colored-graphic-groups}). Closureness stands up by Eq.(\ref{eqa:set-colored-graphic-groups}).

\textbf{Inverse.} The inverse $G_{i^{-1}}$ of each set-colored graph $G_i$ holds $k=i+i^{-1}-k~(\bmod ~M)$ for one of $i^{-1}=2k-i$ and $i^{-1}=M+2k-i$.

\textbf{Associative law.} Three set-colored graphs of $\{F_{\mathcal{E}}(G);[+][-]\}$ satisfy
$$
\big (G_i [+_k] G_j\big )[+_k] G_r=G_i [+_k]\big (G_j[+_k] G_r\big )
$$

\textbf{Commutative law.} Any pair of set-colored graphs $G_i$ and $G_j$ of $\{F_{\mathcal{E}}(G);[+][-]\}$ holds
$$G_i [+_k] G_j=G_j[+_k] G_i
$$

\begin{figure}[h]
\centering
\includegraphics[width=16.4cm]{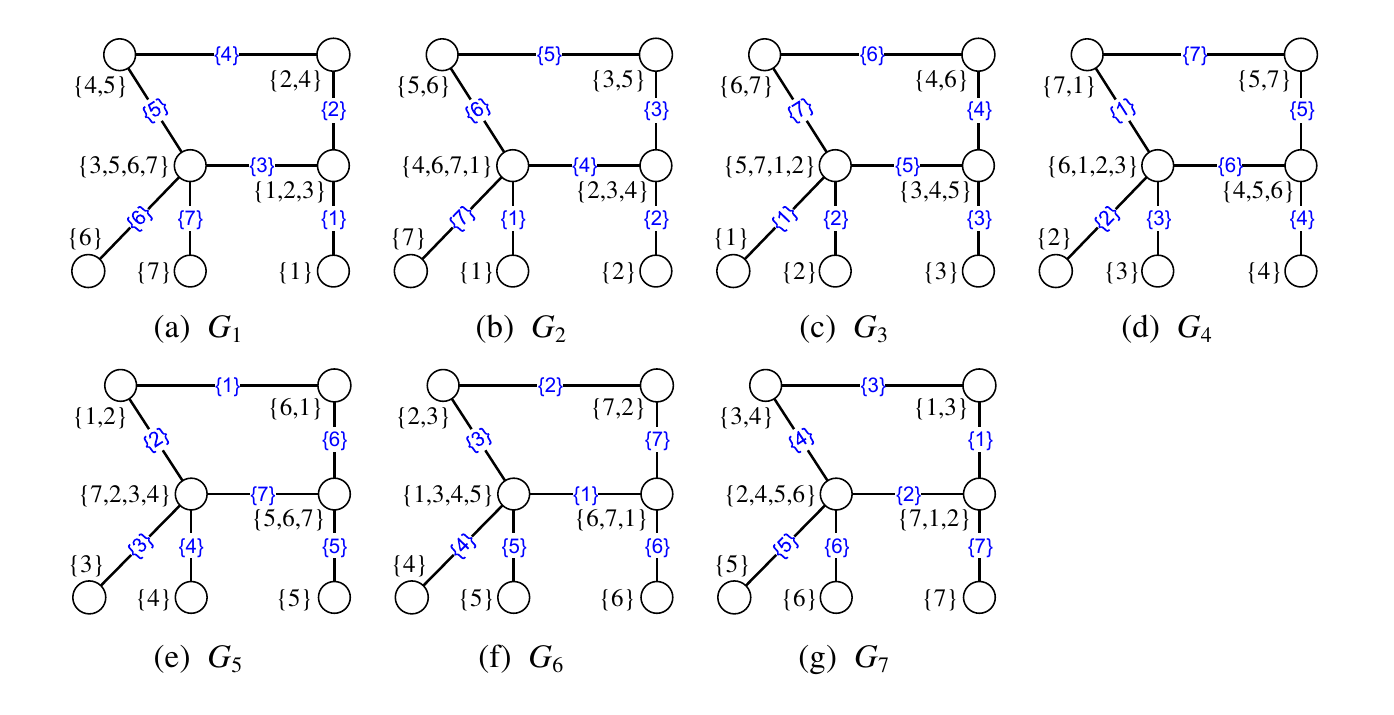}\\
\caption{\label{fig:set-coloring-group}{\small An every-zero set-colored graphic group for illustrating Definition \ref{defn:set-colored-graphic-group}.}}
\end{figure}

\begin{thm}\label{thm:every-zero-set-colored-graphic-group-preserve}
\cite{Yao-Ma-arXiv-2201-13354v1} Suppose that $\{F_{\mathcal{E}}(G);[+][-]\}$ is an every-zero set-colored graphic group defined on a set-colored graph $G$ admitting a total hyperedge set-coloring $F:V(G)\cup E(G)\rightarrow \mathcal{E}$, where $\mathcal{E}$ is a hyperedge set defined on a consecutive integer set $\Lambda$, then

(i) Each set-colored graph $G_i\in \{F_{\mathcal{E}}(G);[+][-]\}$ admitting a total hyperedge set-coloring $F_i:V(G)\cup E(G)\rightarrow \mathcal{E}_i$ is a vertex-intersected graph of some hypergraph if $G$ is a vertex-intersected graph of the hypergraph $\mathcal{H}_{yper}=(\Lambda,\mathcal{E})$, also, $F_i$ is a total intersected-hyperedge set-coloring defined in Definition \ref{defn:general-definition-set-colorings}.

(ii) Each hyperedge set $\mathcal{E}_i$ contains a perfect hypermatching if the hyperedge set $\mathcal{E}$ contains a perfect hypermatching.

(iii) Each set-colored graph $G_i\in \{F_{\mathcal{E}}(G);[+][-]\}$ contains a hyperedge path (res. hyperedge cycle) if the set-colored graph $G$ contains a hyperedge path (res. hyperedge cycle).

(iv) Each set-colored graph $G_i\in \{F_{\mathcal{E}}(G);[+][-]\}$ is set-colored graph homomorphism to a set-colored graph $H_i$ if $G$ is set-colored graph homomorphism to $H$, so that $H_i\cong H$.
\end{thm}

\begin{defn} \label{defn:matrix-graph-type-topcode-matrices}
\cite{Yao-Ma-arXiv-2201-13354v1} Suppose that a $(p,q)$-graph $J$ admits a total graphic group coloring $\phi:V(J)\cup E(J)\rightarrow \{F_{\mathcal{E}}(G);[+][-]\}$, where the very-zero set-colored graphic group $\{F_{\mathcal{E}}(G);[+][-]\}$ is defined in Definition \ref{defn:set-colored-graphic-group}, such that each edge $xy$ holds $\phi(x)\neq \phi(y)$ and $\phi(xy)=\phi(x)[+_k] \phi(y)$ under a preappointed \emph{zero} $G_k\in \{F_{\mathcal{E}}(G);[+][-]\}$, and we get a \emph{graph-type Topcode-matrix} of the $(p,q)$-graph $J$ as follows:
\begin{equation}\label{eqa:graph-type-topcode-matrix-1}
{
\begin{split}
T^{graph}_{code}(J)&=\left(
\begin{array}{ccccc}
\phi(u_1) & \phi(u_2) & \cdots & \phi(u_q)\\
\phi(u_1v_1) & \phi(u_2v_2) & \cdots & \phi(u_qv_q)\\
\phi(v_1) & \phi(v_2) & \cdots & \phi(v_q)
\end{array}
\right)=\left(
\begin{array}{ccccc}
G_{u_1} & G_{u_2} & \cdots & G_{u_q}\\
G_{u_1v_1} & G_{u_2v_2} & \cdots & G_{u_qv_q}\\
G_{v_1} & G_{v_2} & \cdots & G_{v_q}
\end{array}
\right)\\
&=(X(F_{\mathcal{E}}),~E(F_{\mathcal{E}}),~Y(F_{\mathcal{E}}))^{T}
\end{split}}
\end{equation}
where $u_iv_i\in E(J)=\{u_iv_i:i\in [1,q]\}$. We call two vectors $X(F_{\mathcal{E}})=(G_{u_1}, G_{u_2}, \cdots , G_{u_q})$ and $Y(F_{\mathcal{E}})=(G_{v_1}, G_{v_2}, \cdots , G_{v_q})$ \emph{v-graph vectors}, and $E(F_{\mathcal{E}})=(G_{u_1v_1}, G_{u_2v_2}, \cdots , G_{u_qv_q})$ \emph{e-graph vector}. Correspondingly, the above graph-type Topcode-matrix $T^{graph}_{code}(J)$ of the $(p,q)$-graph $J$ induces a \emph{matrix-type Topcode-matrix} defined as
\begin{equation}\label{eqa:graph-type-topcode-matrix-2}
{
\begin{split}
T^{matrix}_{code}(J)&=\left(
\begin{array}{ccccc}
T_{code}(G_{u_1}) & T_{code}(G_{u_2}) & \cdots & T_{code}(G_{u_q})\\
T_{code}(G_{u_1v_1}) & T_{code}(G_{u_2v_2}) & \cdots & T_{code}(G_{u_qv_q})\\
T_{code}(G_{v_1}) & T_{code}(G_{v_2}) & \cdots & T_{code}(G_{v_q})
\end{array}
\right)_{3\times q}\\
&=(X(T_{code}),~E(T_{code}),~Y(T_{code}))^{T}
\end{split}}
\end{equation} with an \emph{e-Topcode-matrix vector} $E(T_{code})=(T_{code}(G_{u_1v_1})$, $T_{code}(G_{u_2v_2})$, $ \cdots $, $T_{code}(G_{u_qv_q}))$, and two \emph{v-Topcode-matrix vectors} $X(T_{code})=(T_{code}(G_{u_1})$, $ T_{code}(G_{u_2})$, $ \cdots $, $T_{code}(G_{u_q}))$ and $Y(T_{code})=(T_{code}(G_{v_1})$, $ T_{code}(G_{v_2}), \cdots , T_{code}(G_{v_q}))$.\qqed
\end{defn}

\begin{rem}\label{rem:333333}
The matrix-type Topcode-matrix $T^{matrix}_{code}(J)$ defined in Eq.(\ref{eqa:graph-type-topcode-matrix-2}) in Definition \ref{defn:matrix-graph-type-topcode-matrices} is a \emph{three dimensional matrix} having the elements as Topcode-matrices $T_{code}(G_{u_i})$, $T_{code}(G_{v_i})$ and $T_{code}(G_{u_iv_i})$ of $3\times q$ rank, more or less, like some things in \emph{Tensor}. \qqed
\end{rem}

By a vertex-intersected graphs defined in Definition \ref{defn:vertex-intersected-graph-hypergraph}, we have:
\begin{defn}\label{defn:set-type-topcode-matrix-definition}
\cite{Yao-Ma-arXiv-2201-13354v1} Since a vertex-intersected $(p,q)$-graph $H$ of a hypergraph $\mathcal{H}_{yper}=(\Lambda,\mathcal{E})$ subject to a constraint set $R_{est}(c_0,c_1,c_2,\dots ,c_m)$ admits a $W$-constraint set-coloring $F:V(H)\rightarrow \mathcal{E}$, so the vertex-intersected $(p,q)$-graph $H$ has its own \emph{set-type Topcode-matrix}
\begin{equation}\label{eqa:set-type-topcode-matrix}
\centering
{
\begin{split}
T^{set}_{code}(H)= \left(
\begin{array}{ccccc}
F(x_1) & F(x_2) & \cdots & F(x_q)\\
F(e_1) & F(e_2) & \cdots & F(e_q)\\
F(y_1) & F(y_2) & \cdots & F(y_q)
\end{array}
\right)_{3\times q}=
\left(\begin{array}{c}
X^{set}\\
E^{set}\\
Y^{set}
\end{array} \right)=(X^{set},~E^{set},~Y^{set})^{T}_{3\times q}
\end{split}}
\end{equation} with \emph{v-set-vector} $X^{set}=(F(x_1), F(x_2), \dots, F(x_q))$, \emph{e-set-vector} $E^{set}=(F(e_1)$, $F(e_2)$, $ \dots $, $F(e_q))$ and \emph{v-set-vector} $Y^{set}=(F(y_1), F(y_2), \dots, F(y_q))$ such that all sets $F(x_i), F(y_k)\in \mathcal{E}$, and there is a function $\phi_s$ for some $s$th constraint $c_s\in R_{est}(c_0,c_1,c_2,\dots ,c_m)$ holding $F(e_j)=\phi_s(F(x_j),F(y_j))$, as well as $F(e_j)\supseteq F(x_j)\cap F(y_j)\neq \emptyset$ for each $j\in [1,q]$.\qqed
\end{defn}

\subsection{Constructing hypergraphs}

\subsubsection{$G$-hypergraphs}

Suppose that a connected $(p,q)$-graph $G$ admits a proper vertex coloring $f:V(G)\rightarrow [1,p]$ such that $f(V(G))=[1,p]$. We vertex-split the connected $(p,q)$-graph $G$ into connected graphs $G_i$ with $i\in [1,n_{vs}(G)]$, where $n_{vs}(G)$ is the number of connected graphs, such that $G_i\not\cong G_j$ and $E(G_i)\cap E(G_j)=\emptyset$ if $i\neq j$, and $E(G)=\bigcup ^{n_{vs}(G)}_{i=1}E(G_i)$, as well as each graph $G_i$ admits a proper vertex coloring $f_i:V(G_i)\rightarrow [1,p]$, such that $f_i$ is induced by the proper vertex coloring $f$.

Moreover, we get integer sets $e_i=\{f(w):~w\in V(G_i)\}$ with $i\in [1,n_{vs}(G)]$, clearly, each set $e_i$ is a subset of the power set $[1,p]^2$. For a hyperedge set $\mathcal{E}\in \mathcal{E}\big ([1,p]^2\big )$, we get a \emph{$G$-hypergraph} $H_{yper}=([1,p],\mathcal{E})$ since $[1,p]=\bigcup_{e\in \mathcal{E}}e$. Especially, a hyperedge set $\mathcal{E}=\{e_s\}$ containing one element only corresponds to a graph $G_s$ of the connected $(p,q)$-graph $G$, since $e_s=V(G_s)$.

For a hyperedge set $\mathcal{E}=\{e_{i_1},e_{i_2},\dots ,e_{i_n}\}\in \mathcal{E}\big ([1,p]^2\big )$, correspondingly, we have the graphs $G_{i_1},G_{i_2},\dots ,G_{i_n}$ by doing vertex-splitting operation to the connected $(p,q)$-graph $G$, and we do the vertex-coincide operation to the graphs $G_{i_1},G_{i_2},\dots ,G_{i_n}$, and get a vertex-coincided graph $[\bullet]^{i_n}_{j=1}G_{i_j}$ of the connected $(p,q)$-graph $G$, where the vertex-coincide operation is to vertex-coincide a vertex $u$ of $G_{i_j}$ with a vertex $v$ of $G_{i_t}$ with $i_j\neq i_t$ into one vertex $u\bullet v$ if $f_{i_j}(u)=f_{i_t}(v)$.

See the subsection ``Assembling graphs with hypergraphs'' for other $G$-hypergraphs.

\subsubsection{$C_{olor}$-hypergraphs}

\begin{defn} \label{defn:111111}
Suppose that a graph $G$ admits colorings $f_1,f_2,\dots ,f_n$, and each coloring $f_i$ corresponds to another coloring $f_j$ with $i\neq j$ such that there is a transformation $\theta_{i,j}$ holding $f_j=\theta_{i,j}(f_i)$. Let $\Lambda_{C}=\{f_1,f_2,\dots ,f_n\}$, each $\mathcal{E}_{C}\in \mathcal{E}(\Lambda^2_{C})$ holding $\Lambda_{C}=\bigcup_{e\in \mathcal{E}_{C}}e$, we get a \emph{$C_{olor}$-hypergraph} $H_{yper}=(\Lambda_{C},\mathcal{E}_{C})$ if $e\in \mathcal{E}_{C}$ corresponds to another subset $e\,'\in \mathcal{E}_{C}$ such that $f\in e$ and $f\,'\in e\,'$ hold $f\,'=\theta(f\,')$.

The Topcode-matrix set $\Lambda_{matrix}=\{T_{code}(G,f_i):i\in [1,n]\}$ forms a Topcode-matrix hypergraph $H^{matrix}_{yper}=(\Lambda_{matrix},\mathcal{E})$ for $\mathcal{E}\in \mathcal{E}(\Lambda^2_{matrix})$.

The Topcode-matrix graph set $\Lambda_{graph}=\{G_{raph}(T_{code}(G,f_i)):i\in [1,n]\}$ forms a graph-set hypergraph $H^{graph}_{yper}=(\Lambda_{graph},\mathcal{E})$ for $\mathcal{E}\in \mathcal{E}(\Lambda^2_{graph})$.\qqed
\end{defn}

\begin{thm}\label{thm:666666}
$^*$ Suppose that $\Lambda=\{g_i:i\in [1,m]\}$ is a coloring set, and there is transformation, such that $theta_{i,j}$ $g_j=\theta_{i,j}(g_i)$ for any pair of two colorings $g_i,g_j\in \Lambda$. Then each connected graph $G$ admits a set-coloring $F:V(G)\cup E(G)\rightarrow \mathcal{E}$, where $\mathcal{E}\in \mathcal{E}(\Lambda^2)$.
\end{thm}

Theorem \ref{thm:graceful-total-sequence-coloring} tells us: Every tree $T$ with diameter $D(T)\geq 3$ and $s+1=\left \lceil \frac{D(T)}{2}\right \rceil $ admits at least $2^{s}$ different \emph{gracefully total sequence colorings} if two sequences $A_M, B_q$ holding $0<b_j-a_i\in B_q$ for $a_i\in A_M$ and $b_j\in B_q$. So, we get $2^{s}$ different $C_{olor}$-hypergraph $H^i_{yper}=(A_M\cup B_q,\mathcal{E}^i_{C})$.

In Problem \ref{question:assembl-graphs-hypergraphs}, there are some hyperedge sets $\mathcal{E}^a\in \mathcal{E}\big (\Lambda^2(G)\big )$ with $a\in [1,m]$, such that each hypergraph $\mathcal{H}^a_{yper}=(\Lambda(G),\mathcal{E}^a)$ produces a proper total coloring of the graph $G$, where each hyperedge set $\mathcal{E}^a=\bigcup^k_{i=1} S^a_i=\bigcup^k_{i=1} (V^a_i\cup E^a_i)$ with $a\in [1,m]$, and moreover subsets $V^a_i\subset V(G)$ and $E^a_i\subset E(G)$ are independent sets of the graph $G$, and each vertex of $V^a_i$ is colored with the $i$th color, and each edge of $E^a_i$ is colored with the $i$th color.

\begin{problem}\label{qeu:444444}
A $k$-level $T$-tree $H$ is a tree, where the tree $T$ is the root tree, such that each tree $H_{i-1}=H_{i}-L(H_{i})$ for $I\in [1,k]$ with $k\geq 1$, and $H_{0}=T$, \textbf{characterize} $k$-level $T$-trees with $k\geq 1$.
\end{problem}

Suppose that a connected graph $G$ admits a proper total coloring $f:V(G)\cup E(G)\rightarrow [a,b]$, and holds $f(V(G)\cup E(G))=[a,b]$. Then we have subsets $e_u=\{f(u),f(ux_i):x_i\in N_{ei}(u)\}$, $e_{uv}=\{f(uv)\}$ and $e_v=\{f(v),f(vy_j):y_j\in N_{ei}(v)\}$ for each edge $uv\in E(G)$. There are hyeredge sets $\mathcal{E}\in \mathcal{E}([a,b]^2)$ with $\bigcup_{e\in \mathcal{E}}e=[a,b]$, and we have some total set-coloring $F:V(G)\cup E(G)\rightarrow \mathcal{E}$ holding $f(u)\in F(u)$, $f(uv)\in F(uv)$ and $f(v)\in F(v)$ for each edge $uv\in E(G)$ and the $C_{olor}$-hypergraph $H_{yper}=([a,b],\mathcal{E})$. In other words, the hypergraph set $\mathcal{E}([a,b]^2)$ contains all proper total colorings of the connected graph $G$.

\subsubsection{$K_{tree}$-hypergraphs, $G_{tree}$-hypergraphs}

\begin{defn} \label{defn:complete-graph-spanning-trees}
$^*$ Let $S_{pan}(K_n)$ be the set of spanning trees of a complete graph $K_n$ of $n$ vertices admitting a vertex coloring $f:V(K_n)\rightarrow [1,n]$ holding $f(V(K_n))=[1,n]$ true. We have:

(i) The cardinality $|S_{pan}(K_n)|=n^{n-2}$ by the famous Cayley's formula $\tau (K_{n})=n^{n-2}$.

(ii) Each spanning tree $T_i\in S_{pan}(K_n)$ admits a vertex coloring $f_i:V(T_i)\rightarrow [1,n]$ holding $f_i(V(T_i))=[1,n]$.

(iii) $S_{pan}(K_n)=\bigcup ^{A_n}_{k=1}N^k_{pan}(K_n)$, where $A_n$ is the number of non-isomorphic spanning tree classes in the $n^{n-2}$ spanning trees, such that

\quad (3-1) any pair of spanning trees $T_{k,i},T_{k,j}\in N^k_{pan}(K_n)$ holds $T_{k,i}\not\cong T_{k,j}$ if $i\neq j$;

\quad (3-2) any spanning tree $T_{k,i}\in N^k_{pan}(K_n)$ is isomorphic to some spanning tree $T_{l,t}\in N^l_{pan}(K_n)$ for each $l\in [1,A_n]\setminus \{k\}$.

(iv) $S_{pan}(K_n)=\bigcup ^{B_n}_{k=1}I^k_{pan}(K_n)$, where $B_n$ is the number of isomorphic spanning tree classes in the $n^{n-2}$ spanning trees, such that

\quad (4-1) any pair of spanning trees $T_{i,j},T_{i,t}\in I^i_{pan}(K_n)$ holds $T_{i,j}\cong T_{i,t}$;

\quad (4-2) any spanning tree $T_{k,i}\in I^k_{pan}(K_n)$ is not isomorphic to any spanning tree $T_{l,s}\in I^p_{pan}(K_n)$ for each $p\in [1,B_n]\setminus \{k\}$.\qqed
\end{defn}

\begin{prop}\label{prop:99999}
Each spanning tree set $I^k_{pan}(K_n)$ in (iv) of Definition \ref{defn:complete-graph-spanning-trees} is the union of several graphic groups, that is, $I^k_{pan}(K_n)=\bigcup ^{m_k}_{i=1}\{F(G_i);[+][-]\}$, where each spanning tree set $F(G_i)=\{T_{i,1},T_{i,2},\dots ,T_{i,n}\}$ with $T_{i,j}\cong T_{i,y}$ and each $T_{i,s}$ admits a vertex coloring $f_{i,s}$ defined in (ii) of Definition \ref{defn:complete-graph-spanning-trees}, such that $f_{i,j}(x)=f_{i,1}(x)+j-1~(\bmod ~n)$ for $x\in V(T_{i,1})=V(T_{i,j})$ with $j\in [1,n]$.
\end{prop}

\begin{example}\label{exa:8888888888}
By the notation of Definition \ref{defn:complete-graph-spanning-trees}, the complete graph $K_{26}$ has
$$26^{24}=9,106,685,769,537,220,000,000,000,000,000,000$$ colored spanning trees in $S_{pan}(K_{26})$. By Table-1 in Appendix A, there are $t_{26}=279,793,450$ non-isomorphic spanning trees of 26 vertices, and there are
\begin{equation}\label{eqa:555555}
26^{24}\div t_{26}=32,547,887,627,595,300,000,000,000
\end{equation} trees being isomorphic to others.

For (iii) of Definition \ref{defn:complete-graph-spanning-trees}, we have $A_{26}~(\geq 26^{24}\div t_{26})$ classes $N^k_{pan}(K_{26})$ of non-isomorphic spanning trees, and each cardinality $|N^k_{pan}(K_{26})|\leq t_{26}$ for $k\in [1,A_{26}]$.

Clearly, $B_{26}=t_{26}$ in (iv) of Definition \ref{defn:complete-graph-spanning-trees}.\qqed
\end{example}

\begin{defn} \label{defn:vertex-coinciding-two-spanning-trees}
$^*$ We define an operation $[\bullet^{coin}_{colo}]$ between two spanning trees $T_j,T_k\in S_{pan}(K_n)$ defined in Definition \ref{defn:complete-graph-spanning-trees} by vertex-coinciding a vertex $u_i$ of $T_j$ with a vertex $x_i$ of $T_k$ into one vertex $u_i\bullet x_i$ if $f_j(u_i)=f_k(x_i)=i$ for $i\in [1,n]$, and the resultant graph after removing multiple-edge is denoted as $T_j[\bullet^{coin}_{colo}]T_k$.\qqed
\end{defn}

\begin{thm}\label{thm:666666}
$^*$ By Definition \ref{defn:complete-graph-spanning-trees} and Definition \ref{defn:vertex-coinciding-two-spanning-trees}, any spanning tree $T_i\in S_{pan}(K_n)$ corresponds two spanning trees $T_j,T_k\in S_{pan}(K_n)$, such that the spanning tree $T_i$ is a subgraph of the graph $T_j[\bullet^{coin}_{colo}]T_k$, here, $T_i\not \cong T_j$, $T_i\not \cong T_k$ and $T_j\not \cong T_k$.
\end{thm}
\begin{proof}Obviously, the graph $T_j[\bullet^{coin}_{colo}]T_k$ is connected and has at least a cycle $C=x_1x_2\cdots x_mx_1$ with $m\geq 3$, because of $T_j\not \cong T_k$. Without loss of generality, $x_mx_1\in E(T_j)$ and $x_1x_2\in E(T_k)$, we remove an edge $x_2x_3$ from the graph $T_j[\bullet^{coin}_{colo}]T_k$, the resultant graph is denoted as $G_1=\big (T_j[\bullet^{coin}_{colo}]T_k\big )-x_2x_3$, clearly, $G_1$ is connected and $x_1x_2,x_mx_1\in E(G_1)$.

If $G_1$ has a cycle $C_1$, we remove an edge $u_1v_1\in E(C_1)\setminus \{x_1x_2,x_mx_1\}$ from $G_1$, and get a connected graph $G_2=G_1-u_1v_1$, go on in this way, we obtain a spanning tree $T_i\in S_{pan}(K_n)$ holding $x_1x_2,x_mx_1\in E(T_i)$ true. Notice that $f(V(K_n))=[1,n]=f(V(T_i))=f(V(T_j))=f(V(T_j))$, so $T_i\not \cong T_j$, $T_i\not \cong T_k$ and $T_j\not \cong T_k$.

The proof of the theorem is complete.
\end{proof}

\begin{problem}\label{qeu:444444}
Determine the number $A_n$ defined in Definition \ref{defn:complete-graph-spanning-trees}.
Since the number $\tau (K_{m,n})$ of all spanning trees of a bipartite complete graph $K_{m,n}$ is $\tau (K_{m,n})=m^{n-1}n^{m-1}$, do researching works like that of the complete graph $K_{n}$.
\end{problem}

\begin{defn} \label{defn:spanning-tree-hyperedge-set}
$^*$ A hyperedge set $\mathcal{E}_i=\{e_{i,1},e_{i,2},\dots ,e_{i,i_a}\}$ is a set of subsets of $S_{pan}(K_n)$ holds $S_{pan}(K_n)=\bigcup_{e_{i,s}\in \mathcal{E}_i}e_{i,s}$, and has one of the following properties:
\begin{asparaenum}[\textbf{\textrm{Prop}}-1.]
\item \label{Prop:11} Each subset $e_{i,s}\in \mathcal{E}_i$ corresponds some subset $e_{i,t}\in \mathcal{E}_i$, such that $e_{i,s}\cap e_{i,t}\neq \emptyset$.
\item \label{Prop:22} A spanning tree $T_{i,k}\in e_{i,k}$ is a proper subgraph of the graph $T_{i,s}[\bullet^{coin}_{colo}]T_{i,t}$ for some spanning trees $T_{i,s}\in e_{i,s}$ and $T_{i,t}\in e_{i,t}$.
\item \label{Prop:33} Two spanning trees $T_{i,s}\in e_{i,s}$ and $T_{i,t}\in e_{i,t}$ hold
$$
T_{i,t}=T_{i,s}+x_iy_i-u_iv_i\quad (T_{i,t}-x_iy_i\cong T_{i,s}-u_iv_i)
$$ for $x_iy_i\not\in E(T_{i,s})$ and $u_iv_i\in E(T_{i,s})$, we write this fact by $T_{i,t}=\pm_e[T_{i,t}]$.\qqed
\end{asparaenum}
\end{defn}

\begin{thm}\label{thm:adding-edge-removing-sets}
\cite{Yao-Su-Ma-Wang-Yang-arXiv-2202-03993v1} Let $T_{\pm e}(\leq n)$ be the set of trees of $p$ vertices with $p\leq n$. Then each tree $H\in T_{\pm e}(\leq n)$ is a star $K_{1,p-1}$, or corresponds another tree $T\in T_{\pm e}(\leq n)$ holding $H-uv\cong T-xy$ for $xy\in E(T)$ and $uv \in E(H)$.
\end{thm}

\begin{thm}\label{thm:adding-edge-removing-w-type-sets}
\cite{Yao-Su-Ma-Wang-Yang-arXiv-2202-03993v1} Let $G_{tree}(\leq n)$ be the set of trees of $p$ vertices with $p\leq n$. Then each tree $H\in G_{tree}(\leq n)$ admits a $W$-type coloring and corresponds another tree $T\in G_{tree}(\leq n)$ admitting a $W$-type coloring holding $H-uv\cong T-xy$ (or $H+xy\cong T+uv$) for some edges $xy\in E(T)$ and $uv \in E(H)$.
\end{thm}

\begin{defn} \label{defn:three-vertex-coincided-intersected-graphs}
$^*$ Suppose that a graph $H$ admits a total coloring $F:V(H)\cup E(H)\rightarrow \mathcal{E}_i$, where $\mathcal{E}_i\in \mathcal{E}(S^2_{pan}(K_n))$ (Ref. Definition \ref{defn:complete-graph-spanning-trees}).

(i) If each edge $uv\in E(H)$ holds $F(uv)\supseteq F(u)\cap F(v)\neq \emptyset$, also, holds \textbf{\textrm{Prop}}-\ref{Prop:11} in Definition \ref{defn:spanning-tree-hyperedge-set}, then $H$ is called \emph{$K_{tree}$-spanning vertex-intersected graph} of the $K_{tree}$-hypergraph $\mathcal{H}_{yper}=(S_{pan}(K_n),\mathcal{E}_i)$.

(ii) If each edge $uv\in E(H)$ holds $T_{i,k}\subset T_{i,s}[\bullet^{coin}_{colo}]T_{i,t}$ for some spanning trees $T_{i,k}\in F(uv)$, $T_{i,s}\in F(u)$ and $T_{i,t}\in F(v)$ (also, \textbf{\textrm{Prop}}-\ref{Prop:22} in Definition \ref{defn:spanning-tree-hyperedge-set}), then $H$ is called \emph{$K_{tree}$-spanning vertex-coincided graph} of the $K_{tree}$-hypergraph $\mathcal{H}_{yper}=(S_{pan}(K_n),\mathcal{E}_i)$.

(iii) If each edge $uv\in E(H)$ holds $T_{i,k}=\pm_e[T_{i,s}]$ and $T_{i,t}=\pm_e[T_{i,k}]$ (also, \textbf{\textrm{Prop}}-\ref{Prop:33} in Definition \ref{defn:spanning-tree-hyperedge-set}) for some spanning trees $T_{i,k}\in F(uv)$, $T_{i,s}\in F(u)$ and $T_{i,t}\in F(v)$, then $H$ is called \emph{$K_{tree}$-spanning added-edge-removed graph} of the $K_{tree}$-hypergraph $\mathcal{H}_{yper}=(S_{pan}(K_n),\mathcal{E}_i)$.\qqed
\end{defn}

\begin{rem}\label{rem:333333}
$^*$ We can generalize Definition \ref{defn:complete-graph-spanning-trees}, Definition \ref{defn:vertex-coinciding-two-spanning-trees}, Definition \ref{defn:spanning-tree-hyperedge-set} and Definition \ref{defn:three-vertex-coincided-intersected-graphs}, to connected graphs. Let $S_{pan}(G)$ be the set of all spanning trees of a connected graph $G$. There are the $G_{tree}$-spanning vertex-intersected graph, $G_{tree}$-spanning vertex-coincided graph and the $G_{tree}$-spanning added-edge-removed graph of the $G_{tree}$-hypergraph $\mathcal{H}_{yper}=(S_{pan}(G),\mathcal{E})$.

However, vertex-splitting a graph $G$ into edge-disjoint spanning trees is a NP-complete problem, since ``Counting trees in a graph is \#P-complete'' \cite{Jerrum-Mark-1994-information}.\qqed
\end{rem}

\textbf{$K_{tree}$-spanning lattice.} We select randomly spanning trees $T^c_1,T^c_2,\dots ,T^c_m$ from $S_{pan}(K_n)$ to form a spanning tree base $\textbf{\textrm{T}}^c=(T^c_1,T^c_2,\dots $, $T^c_m)$ with $T^c_i\not \subset T^c_j$ and $T^c_i\not \cong T^c_j$ if $i\neq j$ and have a permutation $J_1,J_2,\dots ,J_A$ of edge-disjoint spanning trees $a_1T^c_1,a_2T^c_2,\dots ,a_mT^c_m$, where $A=\sum ^m_{k=1}a_k\geq 1$. By Definition \ref{defn:vertex-split-coinciding-operations}, we do the vertex-coinciding operation ``$[\bullet ]$'' to two spanning trees $J_1$ and $J_2$ by vertex-coinciding a vertex $u$ of the spanning tree $J_1$ with a vertex $v$ of the spanning tree $J_2$ into one vertex $u\bullet v$ if these two vertices are colored the same color, and then get a connected graph $H_1=J_1[\bullet]J_2$, next we get another connected graph $H_2=H_1[\bullet]J_3$ by the same action for obtaining the connected graph $H_1=J_1[\bullet]J_2$; go on in this way, we get connected graphs $H_k=H_{k-1}[\bullet]J_{k+1}$ for $k\in [1,A]$ with $H_0=J_1$. We write $H_A=[\bullet]^m_{k=1}a_kT^c_k$, and call the following set
\begin{equation}\label{eqa:K-tree-spanning-lattice}
\textbf{\textrm{L}}(Z^0[\bullet]\textbf{\textrm{T}}^c)=\big \{[\bullet]^m_{k=1}a_kT^c_k:~a_k\in Z^0, T^c_k\in \textbf{\textrm{T}}^c \big \}
\end{equation} \emph{$K_{tree}$-spanning lattice} based on the spanning tree base $\textbf{\textrm{T}}^c \subseteq S_{pan}(K_n)$ which is the set of spanning trees of a complete graph $K_n$.

Thereby, each connected graph $G\in \textbf{\textrm{L}}(Z^0[\bullet]\textbf{\textrm{T}}^c)$ can be vertex-split into the edge-disjoint spanning trees $a_1T^c_1$, $a_2T^c_2$, $\dots $, $a_mT^c_m$ with $\sum ^m_{k=1}a_k\geq 1$, and forms a $G_{tree}$-hypergraph $\mathcal{H}_{yper}=(S_{pan}(G),\mathcal{E})$ for each hyperedge set $\mathcal{E}\in \mathcal{E}(S^2_{pan}(G))$ (Ref. Remark \ref{rem:hypergraph-terminology-notations} and Definition \ref{defn:hypergraph-basic-definition}).

\begin{problem}\label{problem:99999}
Let $H=\{G_1,G_2,\dots ,G_m\}$ be a set of connected graphs, where $G_1=K_{1,q}$, and each connected graph $G_i$ for $i\in [2,m]$ is not a tree and has $q=|E(G_i)|$ edges, and moreover $S_{plit}(G_i)$ for $i\in [2,m]$ is a tree set of $q$ edges obtained by vertex-splitting $G_i$ into trees of $q$ edges. We have a tree set $T_{ree}(q)=\bigcup^m_{i=1}S_{plit}(G_i)$, where $S_{plit}(G_1)=\{G_1=K_{1,q}\}$, such that $T_{ree}(q)$ contains all trees of $q$ edges. Then, we want to

(i) \textbf{How many} groups like the connected graph set $H$ are there?

(ii) \textbf{Determine} a smallest integer $m\geq 2$ about the connected graph set $H$.

(iii) By the \emph{$q$-tree base} $\textbf{\textrm{S}}_{plit}=(S_{plit}(G_1),S_{plit}(G_2),\dots ,S_{plit}(G_m))$, we, by the vertex-coinciding operation, get a \emph{vertex-coinciding tree-lattice} as follows
\begin{equation}\label{eqa:123}
\textbf{\textrm{L}}(Z^0[\bullet]\textbf{\textrm{S}}_{plit})=\Big \{[\bullet]^m_{k=1}a_kS_{plit}(G_k):~a_k\in Z^0, S_{plit}(G_k)\in \textbf{\textrm{S}}_{plit}\Big \},~\sum ^m_{k=1}a_k\geq 1
\end{equation} and an \emph{edge-joining tree-lattice}
\begin{equation}\label{eqa:44455}
\textbf{\textrm{L}}(Z^0[\ominus]\textbf{\textrm{S}}_{plit})=\Big \{[\ominus]^m_{k=1}b_kS_{plit}(G_k):~b_k\in Z^0, S_{plit}(G_k)\in \textbf{\textrm{S}}_{plit}\Big \},~\sum ^m_{k=1}b_k\geq 1
\end{equation}
\end{problem}

\subsubsection{$F_{orest}$-hypergraphs}

\begin{defn} \label{defn:produce-forests-spanning-trees}
$^*$ Removing the edges of one edge set $E_i=\{e_{i,1},e_{i,2}, \dots,e_{i,a_i}\}\subset E(T_i)$ from a spanning tree $T_i$ of a complete graph $K_n$, such that the resultant graph $F_{i}=T_i-E_i$ is just a forest having component trees $T_{i,1},T_{i,2}, \dots,T_{i,b_i}$ with $|V(T_{i,j})|\geq 2$ for $j\in [1,b_i]$. Notice that
$$
V(K_n)=V(T_i)=V(F_i)=\bigcup^{b_i}_{j=1}V(T_{i,j})
$$

Let $F_{orest}(T_i)=\{F_{j}:j\in[1,n_{orest}(T_i)]\}$ be the set of distinct forests produced from a spanning tree $T_i$ of $S_{pan}(K_n)$, where $n_{orest}(T_i)$ is the number of distinct forests produced by the spanning tree $T_i$. We, next, have the forest set $F_{orest}(K_n)=\{F_{orest}(T_i):T_i\in S_{pan}(K_n)\}$ of distinct forests produced from the spanning trees of $S_{pan}(K_n)$.\qqed
\end{defn}

\begin{defn} \label{defn:two-operations-forests}
$^*$ \textbf{Adding and removing edge set operation.} Since the spanning tree $T_i=F_{i}+E_i$ obtained by adding the edges of an edge set $E_i\subset E(K_n)$ in Definition \ref{defn:produce-forests-spanning-trees}, then we can add the edges of other edge set $E_j\subset E(K_n)$ to the forest $F_{i}$, such that the resultant graph $F_{i}+E_j$ is just a spanning tree $T_j\in S_{pan}(K_n)$, that is, $T_j=F_{i}+E_j$. We get graphs $T_i-E_i=F_{i}=T_j-E_j$, thus
\begin{equation}\label{eqa:add-eset-move-2-spanning-trees}
T_j=T_i-E_i+E_j,~E_i\subset E(T_i),~E_j\cap E(T_i)=\emptyset
\end{equation} denoted as $T_j=\pm_E[T_i]$.

\textbf{Forest vertex-coinciding operation.} We define an operation $[\bullet^{coin}_{forest}]$ between two forests $F_j,F_k\in F_{orest}(K_n)$ defined in Definition \ref{defn:produce-forests-spanning-trees} by vertex-coinciding a vertex $x_i$ of $F_j$ with a vertex $y_i$ of $F_k$ into one vertex $x_i\bullet y_i$ if $f_j(x_i)=f_k(y_i)=i$ for $i\in [1,n]$ (it is guaranteed by (ii) in Definition \ref{defn:complete-graph-spanning-trees}), and the resultant graph after removing multiple-edge is denoted as $F_j[\bullet^{coin}_{forest}]F_k$.\qqed
\end{defn}

\begin{thm}\label{thm:666666}
$^*$ A forest of a connected graph $G$ can be produced by many spanning trees of the spanning tree set $S_{pan}(G)$ by the removing edge operation defined in Definition \ref{defn:two-operations-forests}.
\end{thm}

\begin{thm}\label{thm:666666}
$^*$ Two forests $F_j,F_k\in F_{orest}(K_n)$ produce a forest $F_i\subset F_j[\bullet^{coin}_{forest}]F_k$ under the operation $[\bullet^{coin}_{forest}]$ on the forests of $F_{orest}(K_n)$ defined in Definition \ref{defn:produce-forests-spanning-trees} and Definition \ref{defn:two-operations-forests}.
\end{thm}

\begin{defn} \label{defn:coincided-intersected-graphs-by-hyperedge-sets}
$^*$ A hyperedge set $\mathcal{E}_i=\{e_{i,1},e_{i,2},\dots ,e_{i,i_a}\}$ is a set of subsets of the forest set $F_{orest}(K_n)$ defined in Definition \ref{defn:produce-forests-spanning-trees} holds $F_{orest}(K_n)=\bigcup_{e_{i,s}\in \mathcal{E}_i}e_{i,s}$, and has one of the following properties:
\begin{asparaenum}[\textbf{\textrm{Forest}}-1.]
\item \label{Spantree:11} Each subset $e_{i,s}\in \mathcal{E}_i$ corresponds some subset $e_{i,t}\in \mathcal{E}_i$, such that $e_{i,s}\cap e_{i,t}\neq \emptyset$.
\item \label{Spantree:22} A forest $F_{i,k}\in e_{i,k}$ is a proper subgraph of the graph $F_{i,s}[\bullet^{coin}_{forest}]F_{i,t}$ defined in Definition \ref{defn:two-operations-forests} for some forests $F_{i,s}\in e_{i,s}$ and $F_{i,t}\in e_{i,t}$.
\item \label{Spantree:33} Two spanning trees $T_{i,s}\in e_{i,s}$ and $T_{i,t}\in e_{i,t}$ hold $T_{i,t}=T_{i,s}+E^*_i-E_i$ for two edge sets $E^*_i\cap E(T_{i,s})=\emptyset$ and $E_i\subset E(T_{i,s})$, that is, $T_{i,t}=\pm_E[T_{i,t}]$ according to Definition \ref{defn:two-operations-forests}.\qqed
\end{asparaenum}
\end{defn}

Similarly with the graphs defined in Definition \ref{defn:three-vertex-coincided-intersected-graphs}, then Definition \ref{defn:coincided-intersected-graphs-by-hyperedge-sets} enables us to have the following graphs:

\begin{defn} \label{defn:111111}
$^*$ A graph $H$ admits a total coloring $g:V(H)\cup E(H)\rightarrow \mathcal{E}_i$, where the hyperedge set $\mathcal{E}_i\in \mathcal{E}\big (F^2_{orest}(K_n)\big )$.

(i) If each edge $uv\in E(H)$ holds $g(uv)\supseteq g(u)\cap g(v)\neq \emptyset$, also, holds \textbf{\textrm{Spantree}}-\ref{Spantree:11} in Definition \ref{defn:coincided-intersected-graphs-by-hyperedge-sets}, then $H$ is called \emph{$K$-forest vertex-intersected graph} of the $K_{tree}$-hypergraph $\mathcal{H}_{yper}=(F_{orest}(K_n),\mathcal{E}_i)$.

(ii) If each edge $uv\in E(H)$ holds $F_{i,s}[\bullet^{coin}_{forest}]F_{i,t}$ for some forests $F_{i,k}\in g(uv)$, $F_{i,s}\in g(u)$ and $F_{i,t}\in g(v)$ (also, \textbf{\textrm{Spantree}}-\ref{Spantree:22} in Definition \ref{defn:coincided-intersected-graphs-by-hyperedge-sets}), then $H$ is called \emph{$K$-forest vertex-coincided graph} of the $K_{tree}$-hypergraph $\mathcal{H}_{yper}=(F_{orest}(K_n),\mathcal{E}_i)$.

(iii) If each edge $uv\in E(H)$ holds $T_{i,k}=\pm_e[T_{i,s}]$ and $T_{i,t}=\pm_e[T_{i,k}]$ (also, \textbf{\textrm{Spantree}}-\ref{Spantree:33} in Definition \ref{defn:coincided-intersected-graphs-by-hyperedge-sets}) for some spanning trees $F_{i,k}\in g(uv)$, $F_{i,s}\in g(u)$ and $F_{i,t}\in g(v)$, then $H$ is called \emph{$K$-forest added-edgeset-removed graph} of the $K_{tree}$-hypergraph $\mathcal{H}_{yper}=(F_{orest}(K_n),\mathcal{E}_i)$.\qqed
\end{defn}

\begin{rem}\label{rem:333333}
A forest $S(k)=\{T_1,T_2,\dots, T_k\}$ with $k\in [1,n]$ has its vertex number $n=\sum^{k} _{i=1}|V(T_i)|$. Takacs, in \cite{L-Takacs-574-781-1990}, showed the number $F_{orest}(n)$ of distinct forests $S(k)$ with $k\in [1,n]$ to be
\begin{equation}\label{eqa:Forest-number-Takacs}
F_{orest}(n)=\frac{n!}{n+1}\sum^{\lfloor\frac{n}{2}\rfloor}_{k=0}(-1)^k\frac{(2k+1)(n+1)^{n-2k}}{2^k\cdot k!\cdot (n-2k)!}
\end{equation} and
\begin{equation}\label{eqa:555555}
F_{orest}(n)=H_n(n+1)-nH_{n-1}(n+1)
\end{equation} where $H_n(x)$ is the $n$-th Hermite polynomial such that
\begin{equation}\label{eqa:555555}
H_n(x)=n!\sum^{\lfloor\frac{n}{2}\rfloor}_{k=0}\frac{(-1)^k x^{n-2k}}{2^k\cdot k!\cdot (n-2k)!}=\frac{1}{\sqrt{2\pi}}\int ^{+\infty}_{-\infty}e^{-\frac{u^2}{2}}(x-iu)^n\textrm{d}u
\end{equation}

As $k=0$ in the formula (\ref{eqa:Forest-number-Takacs}), we get the famous Caley's formula $F_{orest}(n)=\tau(K_{n+1})=(n+1)^{n-1}$. As $k=1$ in the formula (\ref{eqa:Forest-number-Takacs}), we have \begin{equation}\label{eqa:555555}
F_{orest}(n)=(n+1)^{n-1}+\frac{n!}{n+1}\frac{(-1)3(n+1)^{n-2}}{2(n-2)!}=(n+1)^{n-1}-\frac{3n(n-1)(n+1)^{n-3}}{2}
\end{equation}

However, partitioning an positive integer $n$ as a sum $n=\sum^{k} _{i=1}|V(T_i)|$ with $|V(T_i)|\geq 2$ for $i\in [1,k]$ is the Integer Partition Problem which is a difficult problem since no polynomial algorithm is for solving it. \qqed
\end{rem}

\subsubsection{$E_{hami}$-hypergraphs}

Let $G$ be an edge-hamiltonian $(p,q)$-graph (Ref. Problem \ref{qeu:characterize-edge-hamiltonian}), and $f$ be a vertex coloring from $V(G)$ to $[1,p]$, such that $f(V(G))=[1,p]$. We have a Hamilton-cycle set $E_{hami}(G)$, which contains all Hamilton-cycles of the edge-hamiltonian $(p,q)$-graph $G$. Each hyperedge set $\mathcal{E}\in \mathcal{E}(E^2_{hami}(G))$ with $\bigcup_{e\in \mathcal{E}}e=E_{hami}(G)$ forms a hypergraph $\mathcal{H}_{yper}=(E_{hami}(G),\mathcal{E})$. We define the following operations:
\begin{asparaenum}[\textbf{Ehami}-1]
\item Each subset $e_i\in \mathcal{E}$ corresponds another subset $e\,'_i\in \mathcal{E}$, such that $e_i\cap e\,'_i=\{H_{i,1},H_{i,2}$, $\dots $, $H_{i,a_i}\}$ with $a_i\geq 1$, where each $H_{i,j}$ for $j\in [1,a_i]$ is a Hamilton-cycle of the edge-hamiltonian $(p,q)$-graph $G$.
\item A Hamilton-cycle $H_{i}\in e_i$ corresponds another Hamilton-cycle $H_{j}\in e_j$ holding $H_{k}\subset H_{i}\cup H_{j}$ for some Hamilton-cycle $H_{k}\in e_k\in \mathcal{E}$.
\item A Hamilton-cycle $H_{i}\in e_i$ corresponds another Hamilton-cycle $H_{j}\in e_j$ holding
$$H_{i}-\{x_1y_1,x_2y_2\}\cong H_{j}-\{u_1v_1,u_2v_2\}
$$ for $x_1y_1,x_2y_2\in E(H_{i})$, and $u_1v_1,u_2v_2\in E(H_{j})$.
\end{asparaenum}

Thereby, the hypergraph $\mathcal{H}_{yper}=(E_{hami}(G),\mathcal{E})$ has its own $\Gamma$-operation graph $H$ admitting a total set-coloring $F:V(H)\cup E(H)\rightarrow \mathcal{E}\in \mathcal{E}(E^2_{hami}(G))$, such that edge color $F(\alpha\beta)$ for each edge $\alpha\beta\in E(H)$, vertex color $F(\alpha)$ and vertex color $F(\beta)$ hold one operation \textbf{Ehami}-$t$ with $t\in [1,3]$.

\subsubsection{$W$-constraint hyperedge sets}

For constructing $W$-constraint hyperedge sets, we show an example first as follows:

\begin{example}\label{exa:8888888888}
Suppose that a bipartite $(p,q)$-graph $H$ admits a total coloring $f:V(H)\cup E(H)\rightarrow [0,q]$ subject to a constraint set $R_{est}(c_1,c_2,c_3,c_4)$. Since $V(H)=X\cup Y$ and $X\cap Y=\emptyset $, the coloring $f$ holds the constraints of the constraint set $R_{est}(c_1,c_2,c_3,c_4)$ as follows:

$c_1:$ The \emph{labeling constraint} $f(u)\neq f(w)$ for distinct vertices $u,w\in V(H)$, and $|f(V(H))|=p$;

$c_2:$ the \emph{set-ordered constraint} $\max f(X)<\min f(X)$;

$c_3:$ the \emph{$W$-constraint} $W[f(x),f(xy),f(y)]=\big [f(y)-f(x)\big ]-f(xy)=0$ for each edge $xy\in E(H)$ with $x\in X$ and $y\in Y$; and

$c_4:$ the edge color set holds the graceful constraint
$$
f(E(H))=[1,q]=\{f(xy)=f(y)-f(x):x\in X,y\in Y,xy\in E(H)\}
$$ true.

Also, the coloring $f$ is called \emph{set-ordered graceful labeling}. \qqed
\end{example}

The following Definition \ref{defn:odd-even-separable-6C-labeling} shows us a constraint set $R_{est}(c_1,c_2,\dots, c_m)$ with $m\geq 7$ as follows:

\begin{defn}\label{defn:odd-even-separable-6C-labeling}
\cite{Yao-Sun-Zhang-Mu-Sun-Wang-Su-Zhang-Yang-Yang-2018arXiv} A total labeling $f:V(G)\cup E(G)\rightarrow [1,p+q]$ for a bipartite $(p,q)$-graph $G$ is a bijection and holds the following constraints:

(i) (e-magic) $f(uv)+|f(u)-f(v)|=k$;

(ii) (ee-difference) each edge $uv$ matches with another edge $xy$ holding one of $f(uv)=|f(x)-f(y)|$ and $f(uv)=2(p+q)-|f(x)-f(y)|$ true;

(iii) (ee-balanced) let $s(uv)=|f(u)-f(v)|-f(uv)$ for $uv\in E(G)$, then there exists a constant $k\,'$ such that each edge $uv$ matches with another edge $u\,'v\,'$ holding one of $s(uv)+s(u\,'v\,')=k\,'$ and $2(p+q)+s(uv)+s(u\,'v\,')=k\,'$ true;

(iv) (EV-ordered) $\min f(V(G))>\max f(E(G))$ (resp. $\max f(V(G))<\min f(E(G))$, or $f(V(G))$ $\subseteq f(E(G))$, or $f(E(G))$ $\subseteq f(V(G))$, or $f(V(G))$ is an odd-set and $f(E(G))$ is an even-set);

(v) (ve-matching) there exists a constant $k\,''$ such that each edge $uv$ matches with one vertex $w$ such that $f(uv)+f(w)=k\,''$, and each vertex $z$ matches with one edge $xy$ such that $f(z)+f(xy)=k\,''$, except the \emph{singularity} $f(x_0)=\lfloor \frac{p+q+1}{2}\rfloor $;

(vi) (set-ordered constraint) $\max f(X)<\min f(Y)$ (resp. $\min f(X)>\max f(Y)$) for the bipartition $(X,Y)$ of $V(G)$.

(vii) (odd-even separable) $f(V(G))$ is an odd-set containing only odd numbers, as well as $f(E(G))$ is an even-set containing only even numbers.

We call $f$ \emph{odd-even separable 6C-labeling}.\qqed
\end{defn}

Suppose that a bipartite $(p,q)$-graph $H$ admits a $W$-constraint total coloring $f:V(H)\cup E(H)\rightarrow [a,b]$, such that the total color set $f(V(H)\cup E(H))=[a,b]$. Since $V(H)=X\cup Y$, then we have a hyperedge set $\mathcal{E}^*=\{f(X), f(E(H)),f(Y)\}$ with
\begin{equation}\label{eqa:makk-hyperedge-sets}
[a,b]=\bigcup_{e\in \mathcal{E}^*}e=f(X)\bigcup f(E(H))\bigcup f(Y)
\end{equation} such that each $f(x_i)\in f(X)$ corresponds to some $f(e_i)\in f(E(H))$ and $f(y_i)\in f(Y)$ holding a $W$-constraint $W[f(x_i),f(e_i),f(y_i)]=0$ and some other constraints of a constraint set $R_{est}(c_1,c_2,\dots, c_m)$.

\begin{defn} \label{defn:set-orderedw-hyperedge-set}
$^*$ A \emph{set-ordered $W$-constraint hyperedge set}:
\begin{equation}\label{eqa:555555}
\mathcal{E}_i=\big \{e^x_{i,j}:j\in [1,a_x]\big\}\bigcup \big \{e^E_{i,j}:j\in [1,b_E]\big\}\bigcup \big \{e^y_{i,j}:j\in [1,c_y]\big\}
\end{equation} holds:
\begin{asparaenum}[\textbf{\textrm{Sochs}}-1.]
\item \textbf{(Hyperedge set)} $\mathcal{E}_i\in \mathcal{E}\big([a,b]^2\big)$ and $[a,b]=\bigcup_{e\in \mathcal{E}_i}e$;
\item \textbf{(Set-ordered constraint)} $\max \big \{\max e^x_{i,j}:j\in [1,a_x]\big \}<\min \big \{\min e^y_{i,j}:j\in [1,c_y]\big \}$;
\item \textbf{($W$-constraint)} each $\gamma \in e^E_{i,k}$ with $k\in [1,b_E]$ corresponds $\alpha\in e^x_{i,s}$ for some $s\in [1,a_x]$ and $\beta\in e^y_{i,t}$ for some $t\in [1,c_y]$ holding the $W$-constraint $W[\alpha,\gamma,\beta]=0$.
\end{asparaenum}

And moreover, we say a set-ordered $W$-constraint hyperedge set $\mathcal{E}_i$ to be \emph{\textbf{full}} if

(i) each $\alpha\in e^x_{i,s}$ for $s\in [1,a_x]$ corresponds $\gamma \in e^E_{i,k}$ for some $k\in [1,b_E]$ and $\beta\in e^y_{i,t}$ for some $t\in [1,c_y]$ holding the $W$-constraint $W[\alpha,\gamma,\beta]=0$; and

(ii) each $\beta\in e^y_{i,t}$ with $t\in [1,c_y]$ corresponds $\alpha\in e^x_{i,s}$ for some $s\in [1,a_x]$ and $\gamma \in e^E_{i,k}$ for some $k\in [1,b_E]$ holding the $W$-constraint $W[\alpha,\gamma,\beta]=0$.\qqed
\end{defn}

\begin{example}\label{exa:8888888888}
In the hyperedge set $\mathcal{E}^*=\{f(X), f(E(H)),f(Y)\}$, three color sets $f(X)$, $f(E(H))$ and $f(Y)$ defined in Eq.(\ref{eqa:makk-hyperedge-sets}) can form small color subsets, so there are many hyperedge sets $\mathcal{E}\in \mathcal{E}([a,b]^2)$, like the above hyperedge set $\mathcal{E}^*=\{f(X), f(E(H))$, $f(Y)\}$.

In Fig.\ref{fig:W-constraint-hypergraphs}, the bipartite Hanzi-graph $H_1$ admits a total set-ordered graceful labeling $h:V(H_1)\cup E(H_1)\rightarrow [0,9]$, such that the edge color set $h(V(H_1)\cup E(H_1))=[0,9]$.

(i) Notice that $V(H_1)=X\cup Y$, we have a hyperedge set $\mathcal{E}_1=\{e_{1,1},e_{1,2},e_{1,3}\}$, where $h(X)=e_{1,1}=\{0,2,3,4\}$, $h(Y)=e_{1,2}=\{5,7,8,9\}$ and $h(E(H_1))=e_{1,3}=[1,9]$, such that each $\alpha\in e_{1,1}$ corresponds $\beta\in e_{1,2}$ and $\gamma\in e_{1,3}$ holding the set-ordered graceful-constraint
\begin{equation}\label{eqa:set-ordered-graceful-constraints}
\{\gamma=\beta-\alpha:\alpha\in e_{1,1},\beta\in e_{1,2},\gamma\in e_{1,3}\}=[1,9]
\end{equation} and vice versa. So, $\mathcal{E}_1$ is full, and moreover, $[0,9]=\bigcup_{e_{1,i}\in \mathcal{E}_1}e_{1,i}=\bigcup ^3_{i=1}e_{1,i}$.

(ii) The second hyperedge set is $\mathcal{E}_2=\{e_{2,1},e_{2,2},e_{2,3}, e_{2,4},e_{2,5}\}$, where $e_{2,1}=\{0,2\}$, $e_{2,2}=\{3,4\}$, $e_{2,3}=\{5,7\}$, $e_{2,4}=\{8,9\}$ and $e_{2,5}=[1,9]$. Clearly, the hyperedge set $\mathcal{E}_2$ is full and holds the set-ordered graceful-constraint like that shown in Eq.(\ref{eqa:set-ordered-graceful-constraints}), as well as $[0,9]=\bigcup_{e_{2,i}\in \mathcal{E}_2}e_{2,i}=\bigcup ^5_{i=1}e_{2,i}$.

(iii) The third hyperedge set $\mathcal{E}_3=\{e_{3,1},e_{3,2},e_{3,3}, e_{3,4},e_{3,5}\}$ (see $H_2$ shown in Fig.\ref{fig:W-constraint-hypergraphs}), where $e_{3,1}=\{0,2\}$, $e_{3,2}=\{0,3\}$, $e_{3,3}=\{0,4\}$, $e_{3,4}=\{5,7\}$, $e_{3,5}=\{5,8\}$, $e_{3,6}=\{5,9\}$, $e_{3,7}=[1,5]$ and $e_{3,8}=[6,9]$. It is not hard to verify that the hyperedge set $\mathcal{E}_3$ is full and holds the set-ordered graceful-constraint like that shown in Eq.(\ref{eqa:set-ordered-graceful-constraints}), as well as $[0,9]=\bigcup_{e_{3,i}\in \mathcal{E}_3}e_{3,i}=\bigcup ^8_{i=1}e_{3,i}$.

The above three hyperedge sets $\mathcal{E}_1,\mathcal{E}_2,\mathcal{E}_3$ are the full set-ordered $W$-constraint hyperedge sets according to Definition \ref{defn:set-orderedw-hyperedge-set}. \qqed
\end{example}

\begin{figure}[h]
\centering
\includegraphics[width=16.4cm]{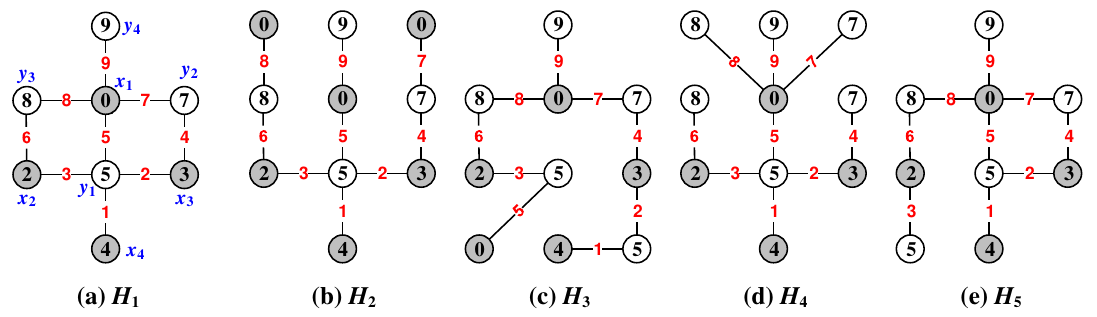}\\
\caption{\label{fig:W-constraint-hypergraphs}{\small A scheme for illustrating the $W$-constraint hyperedge sets.}}
\end{figure}

\begin{defn} \label{defn:hyperedge-set-generator}
$^*$ Suppose that a bipartite graph $G$ admits a set-ordered $W$-constraint total coloring $f:V(G)\cup E(G)\rightarrow [a,b]$, and $V(G)=X\cup Y$ holding the set-ordered constraint $\max f(X)<\min f(Y)$. If a set-ordered $W$-constraint hyperedge set $\mathcal{E}_i$ defined in Definition \ref{defn:set-orderedw-hyperedge-set} holds $f(X)=\big \{e^x_{i,j}:j\in [1,a_x]\big \}$, $f(E)=\big \{e^E_{i,j}:j\in [1,b_E]\big \}$ and $f(Y)=\big \{e^y_{i,j}:j\in [1,c_y]\big \}$, then the bipartite graph $G$ is called \emph{topological generator} of the set-ordered $W$-constraint hyperedge set $\mathcal{E}_i$, we use a symbol $G_{ener}(G,\mathcal{E})$ to denote the set containing all set-ordered $W$-constraint hyperedge sets generated topologically by the bipartite graph $G$.\qqed
\end{defn}

\begin{thm}\label{thm:666666}
$^*$ If a bipartite graph $G$ admits a set-ordered $W$-constraint total coloring $f:V(G)\cup E(G)\rightarrow [a,b]$, or a $(k,d)$-$W$-constraint total coloring defined in Definition \ref{defn:kd-w-type-colorings}, then there exists a topological generator $G_{ener}(G,\mathcal{E})$ defined in Definition \ref{defn:hyperedge-set-generator}.
\end{thm}

\begin{defn} \label{defn:w-constraint-hypergraphs}
$^*$ Let $[a,b]$ be a consecutive integer set. For a hyperedge set $\mathcal{E}\in \mathcal{E}([a,b]^2)$, if each subset $e\in \mathcal{E}$ corresponds to two subsets $e\,'\in \mathcal{E}$ and $e\,''\in \mathcal{E}$, such that there are numbers $\alpha\in e$, $\beta\in e\,'$ and $\gamma\in e\,''$ holding the $W$-constraint $W[\alpha,\beta,\gamma]=0$ and some other constraints of a constraint set $R_{est}(c_1,c_2,\dots, c_m)$. We call the hypergraph $\mathcal{H}_{yper}=([a,b],\mathcal{E})$ \emph{$W$-constraint hypergraph}, and the hyperedge set $\mathcal{E}$ to be \emph{$W$-constraint hyperedge set}.

Moreover, if each number $\alpha\in e\in \mathcal{E}$ corresponds to two numbers $\beta\in e\,'\in \mathcal{E}$ and $\gamma\in e\,''\in \mathcal{E}$ holding the $W$-constraint $W[\alpha,\beta,\gamma]=0$ and some other constraints of the constraint set $R_{est}(c_1,c_2,\dots, c_m)$, then we say the $W$-constraint hyperedge set $\mathcal{E}$ to be \emph{full}, and the hypergraph $\mathcal{H}_{yper}=([a,b],\mathcal{E})$ to be \emph{full $W$-constraint hypergraph}.\qqed
\end{defn}

\begin{problem}\label{question:444444}
$^*$ \textbf{(A)} By Definition \ref{defn:set-orderedw-hyperedge-set} and Definition \ref{defn:w-constraint-hypergraphs}, \textbf{make} $W$-constraint hyperedge sets based on a consecutive integer set $[a,b]$. For example, three numbers $\alpha\in e\in \mathcal{E}$, $\beta\in e\,'\in \mathcal{E}$ and $\gamma\in e\,''\in \mathcal{E}$ hold the following \emph{magic-constraints}:

(A-i) The \emph{edge-magic constraint} $\alpha+\gamma+\beta=k$ for a non-negative constant $k$.

(A-ii) The \emph{edge-difference constraint} $\gamma+|\alpha-\beta|=k$ for a non-negative constant $k$.

(A-iii) The \emph{graceful-difference constraint} $\big||\alpha-\beta|-\gamma \big|=k$ for a non-negative constant $k$.

(A-iv) The \emph{felicitous-difference constraint} $\big|\alpha+\beta-\gamma \big|=k$ for a non-negative constant $k$.

\textbf{(B)} \textbf{Determine} the topological generator $G_{ener}(G,\mathcal{E})$ for a bipartite graph $G$ shown in Definition \ref{defn:hyperedge-set-generator} and one of the edge-magic constraint, the edge-difference constraint, the graceful-difference constraint and the felicitous-difference constraint.
\end{problem}

\begin{defn} \label{defn:111111}
$^*$ \textbf{Edge-$W$-constraint graphs of $W$-constraint hypergraphs.} For a $W$-constraint hypergraph $\mathcal{H}_{yper}=([a,b],\mathcal{E})$, a graph $G$ admits a \emph{total hyperedge-set coloring} $F:V(G)\cup E(G)\rightarrow \mathcal{E}$, such that there are $c\in F(uv)$, $a\in F(u)$ and $b\in F(v)$ holding the $W$-constraint $W[a,c,b]=0$ and some other constraints of a constraint set $R_{est}(c_1,c_2,\dots, c_m)$.

Conversely, three numbers $\alpha\in e\in \mathcal{E}$, $\beta\in e\,'\in \mathcal{E}$ and $\gamma\in e\,''\in \mathcal{E}$ holding the $W$-constraint $W[\alpha,\beta,\gamma]=0$ and some other constraints of the constraint set $R_{est}(c_1,c_2,\dots, c_m)$ correspond always an edge $xy\in E(G)$ such that $\gamma\in F(xy)$, $\alpha\in F(x)$ and $\beta\in F(y)$, then we call the graph $G$ to be \emph{edge-$W$-constraint graph} of the $W$-constraint hypergraph $\mathcal{H}_{yper}=([a,b],\mathcal{E})$.\qqed
\end{defn}

Notice that some hyperedge-set colorings have been defined in Definition \ref{defn:distinguishing-hyperedge-set-colorings}.

\subsubsection{$D_{nei}$-hypergraphs}

\begin{defn} \label{defn:dist-nei-hypergraphs}
$^*$ Suppose that a $(p,q)$-graph $G$ is $G$ admits a total coloring $f:V(G)\cup E(G)\rightarrow S_{thing}$ holding $f(V(G)\cup E(G))=S_{thing}$, such that $f(u)$, $f(uv)$ and $f(v)$ for each edge $uv\in E(G)$ satisfy the constraint set $R_{est}(c_1,c_2,\dots ,c_m)$. There are four neighbor sets of a vertex $u$ of the graph $G$ as follows:
\begin{equation}\label{eqa:555555}
{
\begin{split}
C_v(u)&=\{f(v):v\in N_{ei}(u)\},~C_v[u]=C_v(u)\cup \{f(u)\}\\
C_e(u)&=\{f(uv):v\in N_{ei}(u)\},~C_e[u]=C_e(u)\cup \{f(u)\}
\end{split}}
\end{equation}
We have a particular hyperedge set $\mathcal{E}_{ei}\in \mathcal{E}\big (S^2_{thing}\big )$ as follows
\begin{equation}\label{eqa:555555}
\mathcal{E}_{ei}=\{C_v(u),C_v[u],C_e(u),C_e[u]:u\in V(G)\}
\end{equation} Clearly, $f(V(G)\cup E(G))=S_{thing}=\bigcup_{e\in \mathcal{E}_{ei}}e$. Then the graph $G$ admits a total hyperedge-set coloring $F:V(G)\rightarrow \mathcal{E}_{ei}$ holding one of the following neighbor cases
\begin{asparaenum}[\textbf{\textrm{Nei}}-1.]
\item (v-neighbor) $F(u)=C_v[u]$ for each vertex $u\in V(G)$
\item (e-neighbor) $F(u)=C_e(u)$ for each vertex $u\in V(G)$
\item (ve-neighbor) $F(u)=C_e(u)\cup C_v(u)$ for each vertex $u\in V(G)$
\item (all-neighbor) $F(u)=C_e(u)\cup C_v(u)\cup \{f(u)\}$ for each vertex $u\in V(G)$
\end{asparaenum}
and obey the constraint set $R^*_{est}(a_1,a_2,\dots ,a_n)$, then we say that $G$ is the \emph{neighbor-set graph} of the \emph{$D_{nei}$-hypergraph} $H_{yper}=(S_{thing},\mathcal{E}_{ei})$.\qqed
\end{defn}

\begin{thm}\label{thm:666666}
$^*$ Each $W$-constraint coloring $f$ of a graph $G$ corresponds a $D_{nei}$-hypergraph $H_{yper}=(\bigwedge_f,\mathcal{E}_{f})$ with its own vertex set $\bigwedge_f=f(S)$ for $S\subseteq V(G)\cup E(G)$ and hyperedge set $\mathcal{E}_{f}\in \mathcal{E}\big (\bigwedge^2_f\big )$, such that each subset $e\in \mathcal{E}_{f}$ corresponds another subset $e\,'\in \mathcal{E}_{f}$ holding $e\cap e\,'\neq\emptyset$.
\end{thm}

\begin{example}\label{exa:8888888888}
Suppose that the $(p,q)$-graph $G$ is a graph appeared in Definition \ref{defn:dist-nei-hypergraphs}, we present the following examples for constructing $D_{nei}$-hypergraphs:
\begin{asparaenum}[\textbf{\textrm{Dnei}}-1.]
\item $S_{thing}=[1,M]$ is a consecutive integer set. About the constraint set $R_{est}(c_1,c_2$, $\dots $, $c_m)$, there are:

\qquad $c_1:f(u)\neq f(v)$ for each edge $uv\in E(G)$;

\qquad $c_2:f(uv)\neq f(uw)$ for $v,w\in N_{ei}(u)$ and $u\in V(G)$.

\qquad About the constraint set $R^*_{est}(a_1,a_2,\dots ,a_n)$, there are:

\qquad $a_1:F(x)\supset C_v(x)\neq C_v(y)\subset F(y)$ for each edge $xy\in E(G)$ with degrees $\textrm{deg}_G(u)\geq 2$ and $\textrm{deg}_G(v)\geq 2$;

\qquad $a_2:F(x)\supset C_e(x)\neq C_e(y)\subset F(y)$ for each edge $xy\in E(G)$ with degrees $\textrm{deg}_G(u)\geq 2$ and $\textrm{deg}_G(v)\geq 2$.

\qquad Hence, the graph $G$ is called \emph{adjacent ve-neighbor distinguishing graph} of the \emph{$D_{nei}$-hypergraph} $H_{yper}=([1,M]$, $\mathcal{E}_{ei})$ based on the hyperedge-set coloring $F$ defined in Definition \ref{defn:dist-nei-hypergraphs}. As $M=\chi\,''(G)$, it is not easy to determine the total coloring $f$ in Definition \ref{defn:dist-nei-hypergraphs}.

\item $S_{thing}=[1,p+q]$ is a consecutive integer set, and $G$ is a $(p,q)$-graph. About the constraint set $R_{est}(c_1,c_2,\dots ,c_m)$, there are:

\qquad $c_1:f(u)\neq f(x)$ for distinct vertices $u,x\in V(G)$;

\qquad $c_1:f(uv)\neq f(xy)$ for distinct edges $uv,xy\in E(G)$;

\qquad $c_3:f(u)+f(uv)+f(v)=k$ for each edge $uv\in E(G)$.

\qquad About the constraint set $R^*_{est}(a_1,a_2,\dots ,a_n)$, there are:

\qquad $a_1:F(u)\cap F(v)\neq \emptyset$ for each edge $uv\in E(G)$.

\qquad $a_2:C_e(u)=F(u)\neq F(v)=C_e(v)$ for each edge $uv\in E(G)$ with degrees $\textrm{deg}_G(u)\geq 2$ and $\textrm{deg}_G(v)\geq 2$.

\qquad So, the $(p,q)$-graph $G$ is called the \emph{adjacent e-neighbor distinguishing edge-magic graph} of the \emph{$D_{nei}$-hypergraph} $H_{yper}=([1,p+q]$, $\mathcal{E}_{ei})$ based on the hyperedge-set coloring $F$ defined in Definition \ref{defn:dist-nei-hypergraphs}.

\item $S_{thing}=[0,q]$ is a consecutive integer set, and $G$ is a $(p,q)$-graph. About the constraint set $R_{est}(c_1,c_2,\dots ,c_m)$, there are:

\qquad $c_1:f(u)\neq f(x)$ for distinct vertices $u,x\in V(G)$;

\qquad $c_2:f(uv)\neq f(xy)$ for distinct edges $uv,xy\in E(G)$;

\qquad $c_3:f(uv)=|f(u)-f(v)|$ for each edge $uv\in E(G)$;

\qquad $c_4:f(E(G))=\{f(uv):uv\in E(G)\}=[1,q]$.

\qquad About the constraint set $R^*_{est}(a_1,a_2,\dots ,a_n)$, there are:

\qquad $a_1:F(u)\cap F(v)\neq \emptyset$ for each edge $uv\in E(G)$.

\qquad $a_2:C_e(u)\cup C_v(u)\cup \{f(u)\}=F(u)\neq F(v)=C_e(v)\cup C_v(v)\cup \{f(v)\}$ for each edge $uv\in E(G)$ with degrees $\textrm{deg}_G(u)\geq 2$ and $\textrm{deg}_G(v)\geq 2$.

\qquad Thereby, we say the $(p,q)$-graph $G$ to be the \emph{adjacent all-neighbor distinguishing graceful graph} of the \emph{$D_{nei}$-hypergraph} $H_{yper}=([0,q]$, $\mathcal{E}_{ei})$ based on the hyperedge-set coloring $F$ defined in Definition \ref{defn:dist-nei-hypergraphs}.

\item $S_{thing}=[0,q]$ is a consecutive integer set, and $G$ is a bipartite $(p,q)$-graph with $V(G)=X\cup Y$ and $X\cup Y=\emptyset$. About the constraint set $R_{est}(c_1,c_2,\dots ,c_m)$, there are:

\qquad $c_1:u\in X$ and $v\in Y$ for each edge $uv\in E(G)$;

\qquad $c_2:f(u)\neq f(x)$ for distinct vertices $u,x\in V(G)$;

\qquad $c_3:f(uv)\neq f(xy)$ for distinct edges $uv,xy\in E(G)$;

\qquad $c_4:f(uv)=|f(u)-f(v)|$ for each edge $uv\in E(G)$;

\qquad $c_5:f(E(G))=\{f(uv):uv\in E(G)\}=[1,q]$;

\qquad $c_6:$ the set-ordered constraint $\max f(X)<\min f(Y)$ holds true.

\qquad About the constraint set $R^*_{est}(a_1,a_2,\dots ,a_n)$, there are:

\qquad $a_1:F(u)\cap F(v)\neq \emptyset$ for each edge $uv\in E(G)$.

\qquad $a_2:C_e(u)\cup C_v(u)\cup \{f(u)\}=F(u)\neq F(v)=C_e(v)\cup C_v(v)\cup \{f(v)\}$ for each edge $uv\in E(G)$ with degrees $\textrm{deg}_G(u)\geq 2$ and $\textrm{deg}_G(v)\geq 2$.

\qquad Thereby, the bipartite $(p,q)$-graph $G$ is called the \emph{adjacent all-neighbor distinguishing set-ordered graceful graph} of the \emph{$D_{nei}$-hypergraph} $H_{yper}=([0,q]$, $\mathcal{E}_{ei})$ based on the hyperedge-set coloring $F$ defined in Definition \ref{defn:dist-nei-hypergraphs}.\qqed
\end{asparaenum}
\end{example}

\begin{cor}\label{cor:666666}
$^*$ Since each tree $T$ of $q$ edges admits a set-ordered gracefully total coloring $f$ by Theorem \ref{thm:tree-graceful-total-coloringss}, then the tree $T$ induces a $D_{nei}$-hypergraph $H_{yper}=([0,q],\mathcal{E}_{ei})$ by Definition \ref{defn:dist-nei-hypergraphs}.
\end{cor}

Since every tree $T$ with diameter $D(T)\geq 3$ and $s+1=\left \lceil \frac{D(T)}{2}\right \rceil $ admits at least $2^{s}$ different \emph{gracefully total sequence colorings} if two sequences $A_M, B_q$ holding $0<b_j-a_i\in B_q$ for $a_i\in A_M$ and $b_j\in B_q$ according to Theorem \ref{thm:graceful-total-sequence-coloring}, by Definition \ref{defn:dist-nei-hypergraphs}, we have:

\begin{cor}\label{cor:666666}
$^*$ Each tree $T$ with diameter $D(T)\geq 3$ and $s+1=\left \lceil \frac{D(T)}{2}\right \rceil $ induces at least $2^{s}$ different $D_{nei}$-hypergraph $H^i_{yper}=(A_M\cup B_q,\mathcal{E}^i_{ei})$ for $i\in [1,2^{s}]$.
\end{cor}

\begin{quote}
\textbf{Hypergraph-string problem}: For a given $[0,9]$-string $s=c_1c_2\cdots c_n$ with $c_i\in [0,9]=\{0,1,2,\dots ,9\}$, \textbf{find} a $(p,q)$-graph $G$ admitting a $W$-constraint set-coloring $f$, such that the $W$-constraint set-coloring $f$ induces a $D_{nei}$-hypergraph $H_{yper}=(\bigwedge_f,\mathcal{E}_{ei})$, and the graph $G$ admits a set-coloring $F:S\rightarrow \mathcal{E}_{ei}$ with $S\subseteq V(G)\cup E(G)$ and $f(S)=\bigwedge_f$, then the Topcode-matrix $T_{code}(G,F)$ produces just the given number-based string $s$, we call the number-based string $s$ \emph{hypergraph-string}.
\end{quote}

About the complexity of the Hypergraph-string problem, we point:

\begin{asparaenum}[\textbf{\textrm{NPs}}-1.]
\item \textbf{Finding} the $(p,q)$-graph $G$ will meet Subgraph Isomorphic NP-complete problem.
\item \textbf{Finding} the $W$-constraint set-coloring $f$ for the $(p,q)$-graph $G$ having its own Topcode-matrix $T_{code}(G,F)$ based on the hyperedge-set coloring $F$ defined in Definition \ref{defn:dist-nei-hypergraphs} is sharp-P-hard.
\item \textbf{Partitioning} the given number-based string $s$ to $(3q)!$ segments, which can be produced from the Topcode-matrix $T_{code}(G,F)$, is a work having no polynomial, even a NP-type problem.
\end{asparaenum}

Our goal is: The above Hypergraph-string problem can resist attacks equipped AI technology and quantum computing.

\subsubsection{Hyper-hypergraphs}

\begin{defn} \label{defn:hyper-hypergraphs}
$^*$ By Definition \ref{defn:hypergraph-basic-definition} and a finite set $\Lambda=\{x_1,x_2,\dots ,x_n\}$, we defined a \emph{hyper-hypergraph} as follows: Let $\Phi=\mathcal{E}\big (\Lambda^2\big )=\{\mathcal{E}_i:~i\in [1,n(\Lambda)]\}$ with $n(\Lambda)=|\mathcal{E}\big (\Lambda^2\big )|$, each hyper-hyperedge set $P_i\in \mathcal{E}(\Phi^2)$ is a hyperedge-set set $P_i=\{S_{i,j}:j\in [1,a_i]\}$ with
$$
S_{i,j}=\big \{\mathcal{E}_{i,j,s}\in \mathcal{E}\big (\Lambda^2\big ):s\in [1,b_{i,j}],~j\in [1,a_i]\big \}
$$ and $\Phi=\bigcup _{S_{i,j}\in P_i}S_{i,j}$. We get hypergraphs $H^i_{yper}=(\Phi,P_i)$ with $i\in [1,n(\Phi)]$ and $n(\Phi)=|\mathcal{E}\big (\Phi^2\big )|$, called \emph{hyper-hypergraph}, or written as \emph{hyper$(2)$hypergraph}.\qqed
\end{defn}

By Definition \ref{defn:hyper-hypergraphs}, we have hyper$(n)$hypergraphs for $n\geq 1$.

\section{Overall Topological Encryption Of Networks}

\subsection{Topological groups}

In \cite{Yao-Zhao-Mu-Sun-Zhang-Zhang-Yang-IAEAC2019, Yao-Mu-Sun-Zhang-Wang-Su-Ma-IAEAC-2018, Yao-Sun-Zhao-Li-Yan-2017, Sun-Zhang-Zhao-Yao-2017}, the authors have investigated new-type groups (Abelian additive finite group), called \emph{every-zero graphic groups}. For network overall topological encryption, we will define new topological groups of topology code theory, such as every-zero graphic group, every-zero Topcode-matrix group, every-zero parameterized Topcode-matrix group, every-zero adjacent-matrix group, every-zero topological string group, every-zero topological encoding graph set group, every-zero mixed-graphic group, every-zero graphic groups based on hypergraph, every-zero hypergraph group and pan-group \emph{etc.}

\subsubsection{Graphic groups and their homomorphisms}

\begin{defn} \label{defn:topo-code-graph-groups}
\cite{Yao-Su-Ma-Wang-Yang-arXiv-2202-03993v1, Wang-Su-Yao-submitted-ITOEC2020} \textbf{Graphic group.} Suppose that a $(p,q)$-graph $G$ admits a total $W$-constraint coloring $f: V(G)\cup E(G)\rightarrow [a, b]$, and $M=\max\{f(w): w\in V(G)\cup E(G)\}$. Then we have a graph set $F_f(G)=\{G_1, G_2, \dots , G_M\}$, each graph $G_i$ admits a total $W$-constraint coloring $f_i$ defined by $f_i(w)=f(w)+i-1 ~(\bmod ~M)$ for $w\in V(G)\cup E(G)$, and $G_k \cong G$ and $f_0=f$. The finite module Abelian additive operation
\begin{equation}\label{eqa:555555}
G_i[+_k]G_j:=G_i[+]G_j[-]G_k=G_{\lambda}
\end{equation} is defined by
\begin{equation}\label{eqa:abelian-additive-operation-group}
h_i(w)+h_j(w)-h_k(w)=h_{\lambda}(w),~w\in V(G)\cup E(G)
\end{equation} with $\lambda=i+j-k~(\bmod ~M)$, and $G_k$ is an arbitrarily preappointed \emph{zero}, such that $G_{\lambda}\in F_f(G)$. It is not hard to verify the following facts:
\begin{asparaenum}[(i) ]
\item \textbf{Zero.} Each graph $G_{\lambda}\in F_f(G)$ can be selected as zero in the finite module Abelian additive operation, such that $G_i[+_k]G_j:=G_{\lambda}\in F_f(G)$.

\item \textbf{Inverse.} For $i+j=2k~(\bmod ~M)$, then $G_i[+_k]G_j:=G_k$.
\item \textbf{Uniqueness.} If two graphs $G_i,G_j\in F_f(G)$ hold $G_i [+_k] G_j=G_s$ and $G_i [+_k] G_j=G_r$, then $G_s=G_r\in F_f(G)$ by Eq.(\ref{eqa:abelian-additive-operation-group}).
\item \textbf{Closureness.} For $\lambda=i+j-k~(\bmod ~M)$, then $G_i[+_k]G_j:=G_{\lambda}\in F_f(G)$.

\item \textbf{Associative law.} $G_i[+_k] \big (G_j[+_k]G_r\big )=\big (G_i[+_k]G_j \big )G_i[+_k]G_r$.

\item \textbf{Commutative law.} $G_i[+_k]G_j=G_j[+_k]G_i$.
\end{asparaenum} Then the graph set $F_f(G)$ is called an \emph{every-zero graphic group}, denoted as $\{F_f(G);[+][-]\}$. \qqed
\end{defn}

\begin{problem}\label{qeu:444444}
Suppose that a $(p,q)$-graph $G$ admits another total $W$-constraint coloring $h: V(G)\cup E(G)\rightarrow [c, d]$, $M^*=\max\{h(w): w\in V(G)\cup E(G)\}$. So the graph set $F_h(G)=\{G_1, G_2, \dots , G_M\}$ with each graph $G_i$ admitting a total $W$-constraint coloring $h_i$ defined by
$$h_i(w)=h(w)+i-1 ~(\bmod ~M),~w\in V(G)\cup E(G)
$$ forms another every-zero graphic group $\{F_h(G);[+][-]\}$. \textbf{Consider} relationship between two every-zero graphic groups $\{F_f(G);[+][-]\}$ and $\{F_h(G);[+][-]\}$.
\end{problem}

Let $\{F_\theta(H);[+][-]\}$ be an every-zero graphic group based on a graph set $F_\theta(H)=\{H_1, H_2, \dots $, $H_M\}$. If there are graph homomorphisms $H_i\rightarrow G_i$ for each $i\in [1,M]$, we have defined a \emph{graphic group homomorphism}
\begin{equation}\label{eqa:555555}
\{F_\theta(H);[+][-]\}\rightarrow \{F_f(G);[+][-]\}
\end{equation} from a graph set $F_\theta(H)$ to another graph set $F_f(G)$, also, the \emph{graph set homomorphism} $F_\theta(H)\rightarrow F_f(G)$, introduced in \cite{Bing-Yao-2020arXiv}.

Since the Topcode-matrix $T_{code}(G,f)$ corresponds a graph set $G_{raph}(G)$, such that each graph $H\in G_{raph}(G)$ has its own Topcode-matrix $T_{code}(H,f)=T_{code}(G,f)$ (Ref. Remark \ref{rem:topcode-matrix-homomorphisms-isomor} and Definition \ref{defn:total-coloring-Topcode-matrixs}), then we have the following result:

\begin{thm}\label{thm:666666}
$^*$ \textbf{Graphic group homomorphism.} For a $(p,q)$-graph $G$ admitting a total $W$-constraint coloring $f: V(G)\cup E(G)\rightarrow [a, b]$, there are two or more colored graphs $H$ admitting total $W$-constraint coloring $h$ to form a \emph{graphic group homomorphism}
$$
\{F_h(H);[+][-]\}\rightarrow \{F_f(G);[+][-]\}
$$
\end{thm}

\begin{defn} \label{defn:graphic-group-coincide-general}
$^*$ \textbf{Vertex-coincided graphic group lattice.} For an every-zero graphic group $\{F_f(G);[+][-]\}$ defined in Definition \ref{defn:topo-code-graph-groups}, we vertex-coincide a vertex $u_{i,j}$ of $G_i$ and a vertex $u_{i+1,s}$ of $G_{i+1}$ if $h_i(u_{i,j})=h_{i+1}(u_{i+1,s})$ into one vertex $w_{i,i+1}=u_{i,j} \bullet u_{i+1,s}$ such that the vertex-coincided graph $G_i[\bullet]G_{i+1}$ holds $|E(G_i[\bullet]G_{i+1})|=|E(G_i)|+|E(G_{i+1})|$, so we have a vertex-coincided graph
\begin{equation}\label{eqa:graph-group-vertex-coincide}
H=G_1[\bullet]G_{2}[\bullet]G_{3}\cdots [\bullet]G_{M}=[\bullet ^{vertex}_{color}]^M_{k=1}G_{k}
\end{equation} with $|E(H)|=\sum ^M_{k=1}|E(G_{k})|$. For graphs $a_1G_1,a_2G_2,\dots ,a_{M}G_{M}$ with $G_{i}\in F_f(G)$, $a_i\in Z^0$ and $\sum a_i=A\geq 1$, we can defined a vertex-coincided graph like that defined in Eq.(\ref{eqa:graph-group-vertex-coincide}) as follows
\begin{equation}\label{eqa:graph-group-vertex-coincide-general}
T=T_1[\bullet]T_2[\bullet]T_3\cdots [\bullet]T_{A}=[\bullet ^{vertex}_{color}]^A_{k=1}T_{k}=[\bullet ^{vertex}_{color}]^M_{k=1}a_kG_{k}
\end{equation} with the edge numbers
$$|E(T)|=\sum ^A_{k=1}|E(T_{k})|=\sum ^M_{k=1}a_k|E(G_{k})|
$$ where $T_1,T_2,\dots,T_A$ is a permutation of the graphs $a_1G_1$, $a_2G_2$, $\dots $, $a_{M}G_{M}$. Therefore, we get a \emph{vertex-coincided graphic group lattice}
\begin{equation}\label{eqa:graphic-group-lattices}
\textbf{\textrm{L}}\big (Z^0[\bullet]F_f(G)\big )=\big \{[\bullet ^{vertex}_{color}]^M_{k=1}a_kG_{k}:a_k\in Z^0,G_{i}F_f(G)\big \}
\end{equation} under the vertex-coinciding operation, where $F_f(G)=\{G_1, G_2, \dots , G_M\}$ in Definition \ref{defn:topo-code-graph-groups} is the \emph{vertex-coincided graphic group lattice base}.\qqed
\end{defn}

\begin{rem}\label{rem:v-coin-graph-group-lattice}
About Definition \ref{defn:graphic-group-coincide-general}, vertex-splitting a vertex-coincided graph $T=[\bullet]^M_{k=1}a_kG_{k}$ defined in Eq.(\ref{eqa:graph-group-vertex-coincide-general}) into the graphs $a_1G_1$, $a_2G_2$, $\dots $, $a_{M}G_{M}$ is not easy, since it will meet the Subgraph isomorphic NP-complete problem.

For the generalization of Definition \ref{defn:graphic-group-coincide-general}, suppose that a connected graph $G$ admits a $W$-constraint coloring $f$, each connected graph $H_i$ admits a $W_i$-constraint coloring $g_i$ for $i\in [1,m]$, and $f(V(G))\cap g_i(V(H_i))\neq \emptyset$. Then we vertex-coincide some vertices of the connected graph $H_i$ and some vertices of the connected graph $G$ together as if these vertices are colored with the same colors, the resultant graph denoted as $B_G=G[\bullet ^{vertex}_{color}]^m_{k=1}H_k$ holds the edge number formula
$$|E(B_G)|=|E(G)|+\sum ^m_{k=1}|E(H_k)|
$$ and is connected too. Clearly, $B_G$ is like a ``book'', the connected graph $G$ is the \emph{spine} of the book, and each connected graph $H_i$ is a \emph{page} of the book.

In a \emph{graph network} proposed by 27 scientists from DeepMind, GoogleBrain, MIT and University of Edinburgh \cite{Peter-Battaglia-Jessica-et-al-arXiv-2018}, the above connected graph $G$ is the \emph{graph network framework}, and each connected graph $H_i$ is a \emph{graph block}.\qqed
\end{rem}

\subsubsection{Topcode-matrix groups and topological string groups}

Notice that Proposition \ref{prop:translated-number-based-string} shows us:

(i) Each simple graph can be translated into a number-based string.

(ii) A number-based string can be generated by the Topcode-matrices of two colored graphs $G$ and $H$, such that $G\not\cong H$.

Similarly with Definition \ref{defn:every-zero-total-graphic-group} and Definition \ref{defn:graphic-group-definition}, we redefine the every-zero graphic group and the every-zero Topcode-matrix group as follows:

\begin{defn} \label{defn:topocode-matrices-groups}
$^*$ \textbf{Topcode-matrix group.} An every-zero graphic group $\{F_f(G);[+][-]\}$ (Ref. Definition \ref{defn:topo-code-graph-groups}) is based on a graph set $F_f(G)=\{G_1,G_2,\dots ,G_m\}$ with $G_i\cong G_1$ for $i\in [1,m]$, where each colored graph $G_i$ admits a coloring $f_i$ and has its own Topcode-matrix $T_{code}(G_i,f_i)$, these Topcode-matrices form a Topcode-matrix set
$$F(T_{code}(G))=\big \{T_{code}(G_1,f_1),T_{code}(G_2,f_2),\dots ,T_{code}(G_m,f_m)\big \}
$$ and holding the finite module Abelian additive operation
\begin{equation}\label{eqa:topcode-matrix-abelian-additive-operation}
T_{code}(G_i,f_i)[+]T_{code}(G_j,f_j)[-]T_{code}(G_k,f_k)=T_{code}(G_\lambda,f_\lambda)
\end{equation} with $\lambda=i+j-k~(\bmod~m)$ for any preappointed \emph{zero} $T_{code}(G_k,f_k)\in F(T_{code}(G))$, where Eq.(\ref{eqa:topcode-matrix-abelian-additive-operation}) is defined by
\begin{equation}\label{eqa:topcode-matrix-abelian-additive-operation11}
f_i(w)[+]f_j(w)[-]f_k(w)=f_\lambda(w),~w\in V(G_1)\cup E(G_1)
\end{equation} with $\lambda=i+j-k~(\bmod~m)$, so we get a \emph{Topcode-matrix group} denoted as $\{F(T_{code}(G));[+][-]\}$.\qqed
\end{defn}

\begin{thm}\label{thm:666666}
\textbf{Topological group homomorphism.} By Definition \ref{defn:topocode-matrices-groups}, we have:

(i) \emph{Topcode-matrix group homomorphism}
\begin{equation}\label{eqa:555555}
\{F(T_{code}(H));[+][-]\}\rightarrow \{F(T_{code}(G));[+][-]\}
\end{equation}

(ii) \emph{Topological string group homomorphism}.

(iii) \emph{Topological encoding graph set group homomorphism}.

(iv) \emph{Mixed-graphic group homomorphism} by Definition \ref{defn:mixed-graphic-groups-11}.
\end{thm}

\begin{defn} \label{defn:topocode-matrices-groups11}
$^*$ Suppose that a $(p,q)$-graph $G$ admits a $W$-constraint total coloring $f: V(G)\cup E(G)\rightarrow [a,b]$, and $M=\max \{f(w): w\in V(G)\cup E(G)\}$. By the every-zero graphic group $\{F_f(G);[+][-]\}$ defined in Definition \ref{defn:topo-code-graph-groups}, we present the following topological groups:

\begin{asparaenum}[\textrm{\textbf{Topogroup}}-1. ]
\item \textbf{Adjacent-matrix group.} Under the finite module Abelian additive operation, the total coloring adjacent-matrix set $T_{adj}(G, F_{adj})=\{T_{adj}(G_i, h_i): i\in [1, M]\}$ forms an every-zero adjacent-matrix group $\{T_{adj}(G, F_{adj}); [+][-]\}$.
\item \textbf{Topcode-matrix group.} The Topcode-matrix group $\{F(T_{code}(G));[+][-]\}$ defined in Definition \ref{defn:topocode-matrices-groups}.
\item \textbf{Topological string group.} For $i\in [1,m]$, each Topcode-matrix $T_{code}(G_i,f_i)$ of the Topcode-matrix group $\{F(T_{code}(G))$; $[+][-]\}$ defined in Definition \ref{defn:topocode-matrices-groups} can induces $(3q)!$ number-based strings. We use a fixed algorithm $\pi$ to take a number-based string $s_i$ from $T_{code}(G_i,f_i)$ with $i\in [1,m]$, so the sequence $\{s_i\}^m_{i=1}$ forms a \emph{topological string group} under the finite module Abelian additive operation $s_i[+]s_j[-]s_k=s_\lambda$ with $\lambda=i+j-k~(\bmod~m)$. Thereby, the Topcode-matrix group $\{F(T_{code}(G));[+][-]\}$ produces $(3q)!$ topological string groups.
\item \textbf{Topological encoding graph set group.} Since each Topcode-matrix $T^i_{code}\in \{F(T_{code}(G));[+][-]\}$ corresponds a topological encoding graph set $G_{raph}(T^i_{code})$, then the topological encoding graph set $G_{raph}=\{G_{raph}(T^i_{code}): i\in [1, M]\}$ forms an \emph{every-zero topological encoding graph set group} $\{F(G_{raph});[+][-]\}$ based on the finite module Abelian additive operation.

\item \textbf{Parameterized Topcode-matrix group.} A parameterized Topcode-matrix set
\begin{equation}\label{eqa:555555}
P_{code}(G, L|k,d)=\big \{P_{code}(G_1, L_1 | k, d), P_{code}(G_2, L_2 | k, d), \dots, P_{code}(G_M, L_M | k, d)\big \}
\end{equation} each parameterized Topcode-matrix $P_{code}(G_i, L_i|k, d)=k\cdot I^0+d\cdot T_{code}(G_i, h_i)$ with $i\in [1, M]$, $k\geq 1$ and $d\geq 1$. Under the finite module Abelian additive operation
\begin{equation}\label{eqa:555555}
P_{code}(G_i, h_i)[+]P_{code}(G_j, h_j)[-]P_{code}(G_k, h_k)=P_{code}(G_{\lambda}, h_{\lambda})
\end{equation} the parameterized Topcode-matrix set $P_{code}(G, L|k, d)$ forms an \emph{every-zero parameterized Topcode-matrix group} $\{P_{code}(G, L|k, d);[+][-]\}$. \qqed
\end{asparaenum}
\end{defn}

\begin{thm}\label{thm:666666}
$^*$ Each $(p,q)$-graph admitting $m$ total colorings forms:

(i) $m$ every-zero graphic groups;

(ii) $m$ every-zero Topcode-matrix groups;

(iii) $m\cdot (3q)!$ every-zero number-based string groups.
\end{thm}

Definition \ref{defn:matrix-graph-type-topcode-matrices} and Definition \ref{defn:set-type-topcode-matrix-definition} have defined: graph-type Topcode-matrix, matrix-type Topcode-matrix, set-type Topcode-matrix. By Definition \ref{defn:topo-code-graph-groups} and Definition \ref{defn:topocode-matrices-groups11}, we present the following $W$-group colorings:

\begin{defn} \label{defn:total-W-group-coloring-5}
$^*$ Suppose that a graph $H$ admits a total $W$-group coloring $F_X: V(H)\cup E(H)\rightarrow W\textrm{-group}$, such that each edge $uv\in E(H)$ holds $F_X(u)\neq F_X(v)$, and holds the finite module Abelian additive operation
\begin{equation}\label{eqa:555555}
F^i_X[+]F^j_X[-]F^k_X= F^{\lambda}_X\in W\textrm{-group}
\end{equation}
with $\lambda=i+j-k~(\bmod~M^*)$, where $F^k_X(w)$ is a preappointed \emph{zero}, $F_X$ is one of colorings $F_{adj}, F_{gra}, F_{mat}, F_{string}, F_{param}, F_{set}$ and $F_{thing}$. We define the $W$-group colorings as follows:
\begin{asparaenum}[\textrm{\textbf{Groupc}}-1. ]
\item The graph $H$ admits a \emph{total-colored adjacent matrix group coloring} $F_{adj}: V(H)\cup E(H)\rightarrow \{T_{adj}(G, F_{adj}); [+][-]\}$, where $T_{adj}(G, F_{adj})$ is a total-colored adjacent matrix, such that each element of Topcode-matrix $T_{code}(H$, $F_{adj})$ is a \emph{total-colored adjacent matrix} of the total-colored adjacent matrix set $T_{adj}(G$, $F_{adj})$, like \emph{Tensor matrix}.

\item The graph $H$ admits a \emph{total-colored graph matrix group coloring} $F_{gra}: V(H)\cup E(H)\rightarrow \{F_f(G); [+][-]\}$, such that each element of the Topcode-matrix $T_{code}(H, F_{gra})$ is a colored graph $G_i$, so $T_{code}(H, F_{gra})$ is a \emph{colored graph Topcode-matrix} of the total-colored graph matrix set $F_f(G)$.

\item The graph $H$ admits a \emph{total-colored Topcode-matrix group coloring}
$$F_{mat}: V(H)\cup E(H)\rightarrow \{T_{code}(G, F_{mat}); [+][-]\}
$$ such that each element of the total-colored Topcode-matrix $T_{code}(H, F_{mat})$ is a total-colored Topcode-matrix, so $T_{code}(H, F_{mat})$ is a \emph{Tensor Topcode-matrix} of the total-colored Topcode-matrix set $T_{code}(G, F_{mat})$.
\item Since each total-colored Topcode-matrix $T_{code}(G_i,f_i)$ of the Topcode-matrix group $\{F(T_{code}(G))$; $[+][-]\}$ defined in Definition \ref{defn:topocode-matrices-groups} induces $(3q)!$ number-based strings. We use a fixed algorithm $\pi$ to take a number-based string $s_i$ from each total-colored Topcode-matrix $T_{code}(G_i,f_i)$ with $i\in [1,m]$, so the sequence $\{s_i\}^m_{i=1}$ forms a \emph{topological string group}. The graph $H$ admits a \emph{topological string group coloring} $F_{string}: V(H)\cup E(H)\rightarrow \{s_i\}^m_{i=1}$, such that each element of the Topcode-matrix $T_{code}(H, F_{string})$ is a topological string. Since $G$ is a total-colored $(p,q)$-graph, so there are $(3q)!$ topological string group colorings for the graph $H$.
\item The graph $H$ admits a \emph{total-colored parameterized matrix group coloring}
$$F_{param}: V(H)\cup E(H)\rightarrow \{P_{code}(G,L|k,d); [+][-]\}
$$ such that each element of the Topcode-matrix $T_{code}(H$, $F_{param})$ is a total-colored parameterized Topcode-matrix of the total-colored parameterized matrix set $P_{code}(G,L|k,d)$.
\item The graph $H$ admits a \emph{total-colored graph-set group coloring}
$$F_{set}: V(H)\cup E(H)\rightarrow \{F(G_{raph}); [+][-]\}
$$ such that each element of the Topcode-matrix $T_{code}(H, F_{set})$ is a total-colored graph set of the total-colored graph-set set $F(G_{raph})$.
\item The graph $H$ admits a \emph{thing group coloring}
$$F_{thing}: V(H)\cup E(H)\rightarrow \{F(S_{thing});[+][-]\}
$$ such that each element of the Topcode-matrix $T_{code}(H, F_{thing})$ is a thing set of the thing set $F(S_{thing})$.\qqed
\end{asparaenum}
\end{defn}

\begin{thm}\label{thm:total-W-group-coloring-6}
$^*$ Suppose that a graph $H$ admits a total graphic group coloring $F: V(H)\cup E(H)\rightarrow \{F_f(G); [+][-]\}$, by Definition \ref{defn:total-W-group-coloring-5}, then the graph $H$ admits:
\begin{asparaenum}[(i)]
\item a total-colored adjacent matrix group coloring
$$F_{adj}: V(H)\cup E(H)\rightarrow \{T_{adj}(G, F_{adj}); [+][-]\}$$

\item a total-colored graph matrix group coloring $F_{gra}: V(H)\cup E(H)\rightarrow \{F_f(G); [+][-]\}$;

\item a total-colored Topcode-matrix group coloring
$$F_{mat}: V(H)\cup E(H)\rightarrow \{T_{code}(G, F_{mat});[+][-]\}$$
\item a topological string group coloring $F_{string}: V(H)\cup E(H)\rightarrow \{s_i\}^m_{i=1}$;
\item a total-colored parameterized matrix group coloring

$$F_{param}: V(H)\cup E(H)\rightarrow \{P_{code}(G,L|k,d); [+][-]\}$$

\item a total-colored graph set group coloring
$$F_{param}: V(H)\cup E(H)\rightarrow \{F(G_{raph}); [+][-]\}$$

\item a thing group coloring $F_{thing}: V(H)\cup E(H)\rightarrow \{F(S_{thing}); [+][-]\}$.
\end{asparaenum}
\end{thm}

\begin{cor}\label{cor:total-W-group-coloring-7}
$^*$ If a graph admits a graphic group coloring, then there exists a graph set, such that each graph in this graph set produces a graphic group coloring admitted by the graph.
\end{cor}

\subsubsection{Mixed-graphic groups}

The mixed-graphic group has been defined and investigated in \cite{Wang-Su-Yao-submitted-ITOEC2020}, we will do more researching works on it in this subsection.

\begin{defn} \label{defn:mixed-graphic-groups-11}
\cite{Wang-Su-Yao-submitted-ITOEC2020} Suppose that a $(p,q)$-graph $G$ admits a total $W$-constraint coloring $f: V(G)\cup E(G)\rightarrow [1, R]$, such that the vertex color set $f(V(G))=\{f(x): x\in V(G)\}$ and the edge color set $f(E(G))=\{f(uv): uv\in E(G)\}$ hold the $W$-constraint. The graph set $R_f(G)=\{G_{s,k}: s\in [1,p], k\in [1,q]\}$ is called \emph{mixed colored graph set}, each colored graph $G_{s, k}$ is isomorphic to the graph $G$, and admits a total $W$-constraint coloring $h_{s, k}$, such that $h_{s, k}(x)=f(x)+s~(\bmod ~p)$ for $x\in V(G_{s, k})$, and $h_{s, k}(uv)=f(uv)+k~ (\bmod ~q)$ for $uv\in E(G_{s, k})$.
Selecting arbitrarily a zero $G_{a,b}\in R_f(G)$, we define the finite module Abelian additive operation
\begin{equation}\label{eqa:mixed-abelian-additive-operation}
G_{s,k}[+] G_{i,j}[-]G_{a,b}=G_{\lambda, \mu}
\end{equation} for the graph set $R_f(G)$ as follows:

(i) $h_{s,k}(x)+h_{i,j}(x)-h_{a,b}(x)=h_{\lambda, \mu}(x)$ with $\lambda=s+i-a~(\bmod ~p)$ for each $x\in V(G)=V(G_{s,k})=V(G_{i,j})=V(G_{a,b})$;

(ii) $h_{s,k}(e)+h_{i,j}(e)-h_{a,b}(e)=h_{\lambda, \mu}(e)$ with $\mu=k+j-b~(\bmod ~q)$ for each edge $e\in E(G)=E(G_{s,k})=E(G_{i,j})=E(G_{a,b})$;

(iii) $G_{\lambda, \mu}\in R_f(G)$.

Under the finite module Abelian additive operation Eq.(\ref{eqa:mixed-abelian-additive-operation}), the mixed colored graph set $R_f(G)$ holds: Each $G_{s,k}\in R_f(G)$ can be appointed as zero; Each $G_{s,k}\in R_f(G)$ has its own inverse; Uniqueness; Closureness; Associative law and Commutative law. So we call $R_f(G)$ \emph{every-zero mixed-graphic group}, denoted as $\{R_f(G);[+][-]\}$.\qqed
\end{defn}

\begin{defn} \label{defn:mixed-graphic-groups-22}
$^*$ Suppose that a $(p,q)$-graph $G$ admits a total $W$-constraint coloring $f: V(G)\cup E(G)\rightarrow [1, B]$. By the every-zero mixed-graphic group $\{R_f(G);[+][-]\}$ defined in Definition \ref{defn:mixed-graphic-groups-11}, the $(p,q)$-graph $G$ admits a total graphic group coloring $F:V(G)\cup E(G)\rightarrow \{R_f(G);[+][-]\}$, so we have a Topcode-matrix $R_{at}(w)=T_{code}(G_{s,k}, h_{s,k})$ with $F(w)=G_{s,k}$ for each $w\in V(G)\cup E(G)$. We get an \emph{every-zero mixed Topcode-matrix group} $\{R_{at}(G);[+][-]\}$, where the graph set $R_{at}(G)=\{R_{at}(w)=T_{code}(G_{s,k}, h_{s,k}):~G_{s,k}\in R_f(G)\}$.\qqed
\end{defn}

\begin{defn} \label{defn:mixed-graphic-groups-33}
$^*$ Suppose that a $(p,q)$-graph $G$ admits a total $W$-constraint coloring $f: V(G)\cup E(G)\rightarrow [1, B]$, by Definition \ref{defn:mixed-graphic-groups-11} and the finite module Abelian additive operation Eq.(\ref{eqa:mixed-abelian-additive-operation}), then we have the every-zero mixed Topcode-matrix group $\{R_{at}(G);[+][-]\}$.

Since each Topcode-matrix $T_{code}(G_{s,k}, h_{s,k})$ is $3\times q$ rank, and produces $(3q)!$ strings, so there are $(3q)!$ algorithms, such that each algorithm $C_{\theta}$ for $\theta\in [1, (3q)!]$ induces a string $s(C_{\theta}, T_{code}(G_{s,k}, h_{s,k})$ from the Topcode-matrix $T_{code}(G_{s,k}, h_{s,k})$.
For each integer $\theta\in [1, (3q)!]$, the string set
\begin{equation}\label{eqa:mixed-string-groups-00}
S_{tring}(\theta)=\{s(C_{\theta}, T_{code}(G_{s,k}, h_{s,k}): G_{s,k}\in R_f(G)\}
\end{equation} forms an \emph{every-zero mixed string group}, denoted as $\{S_{tring}(\theta);[+][-]\}$ for $\theta\in [1, (3q)!]$.\qqed
\end{defn}

\begin{thm}\label{thm:666666}
$^*$ Suppose that a $(p,q)$-graph $G$ admits a total $W$-constraint coloring $f$ based on the mixed colored graph set $R_f(G)=\{G_{s,k}: s\in [1,p], k\in [1,q]\}$ (Ref. Definition \ref{defn:mixed-graphic-groups-11}) and the finite module Abelian additive operation Eq.(\ref{eqa:mixed-abelian-additive-operation}), then there are the every-zero mixed-graphic group $\{R_f(G);[+][-]\}$, the every-zero mixed Topcode-matrix group $\{R_{at}(G);[+][-]\}$, and $(3q)!$ every-zero mixed string groups $\{S_{tring}(\theta);[+][-]\}$ with $\theta\in [1, (3q)!]$ (Ref. Definition \ref{defn:mixed-graphic-groups-33}).
\end{thm}

\subsubsection{Infinite mixed-graphic groups}

\begin{defn} \label{defn:mixed-graphic-groups-44}
$^*$ Suppose that a $(p,q)$-graph $G$ admits a total $W$-constraint coloring $h$. By Definition \ref{defn:mixed-graphic-groups-11}, we get an infinite mixed graph set $I^{+\infty}_{-\infty}(G, h; [+][-])=\{G_{s,k}:~-\infty <s, k<+\infty \}$, where each colored graph $G_{s,k}\cong G$ and $G_{0,0}\cong G$, and each colored graph $G_{s,k}$ admits a total $W$-constraint coloring $h_{s,k}$ defined by $h_{s,k}(x)=h(x)+s$ for $ x\in V(G_{s,k})$, $h_{s,k}(uv)=h(uv)+k$ for $uv\in E(G_{s,k})$, such that the finite module Abelian additive operation Eq.(\ref{eqa:mixed-abelian-additive-operation}) holds true for any preappointed zero $G_{a,b}\in I^{+\infty}_{-\infty}(G, h; [+][-])$, call the infinite mixed graph set $I^{+\infty}_{-\infty}(G, h; [+][-])$ \emph{every-zero infinite mixed-graphic group}.\qqed
\end{defn}

By Definition \ref{defn:mixed-graphic-groups-11}, Definition \ref{defn:mixed-graphic-groups-22}, Definition \ref{defn:mixed-graphic-groups-33} and Definition \ref{defn:mixed-graphic-groups-44},we get the following topological groups:

(i) \textbf{Every-zero infinite mixed-graphic group.} The every-zero mixed-graphic group $\{R_f(G)$; $[+][-]\}$ is a proper subset of the every-zero infinite mixed-graphic group $I^{+\infty}_{-\infty}(G, h; [+][-])$.

(ii) \textbf{Every-zero infinite mixed Topcode-matrix group.} The every-zero mixed Topcode-matrix group $\{R_{at}(G);[+][-]\}$ is a proper subset of the every-zero infinite Topcode-matrix group
\begin{equation}\label{eqa:555555}
M^{\infty}_{at}(G)=\big \{T_{code}(G_{s,k}, h_{s,k}): G_{s,k}\in I^{+\infty}_{-\infty}(G, h; [+][-])\big \}
\end{equation}

(iii) \textbf{Every-zero infinite mixed string group.} For $-\infty<\theta<+\infty$, we have the every-zero infinite mixed string group
\begin{equation}\label{eqa:mixed-string-groups-11}
S^{\infty}_{tring}(\theta)=\big \{s(C_{\theta}, T_{code}(G_{s,k}, h_{s,k}): G_{s,k}\in I^{+\infty}_{-\infty}(G, h; [+][-])\big \}
\end{equation}

(iv) Each integer point $(s,k)$ in xOy-plane corresponds a colored graph $H_{s,k}$, and the Topcode-matrix $T_{code}(G_{s,k}, h_{s,k})$, as well as $(3q)!$ strings.

Fig.\ref{fig:infinite-mixed-graphic-group} shows us an example of the infinite mixed-graphic group.

\begin{figure}[h]
\centering
\includegraphics[width=16.4cm]{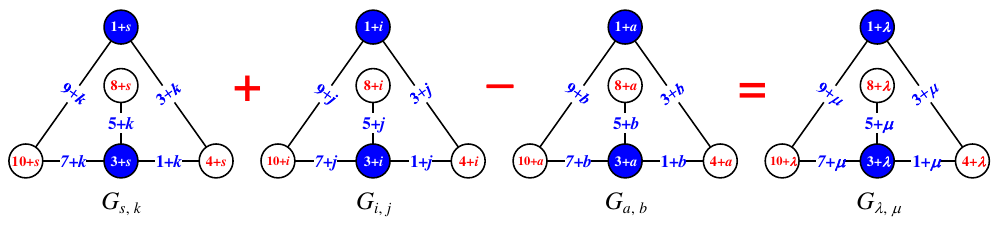}\\
\caption{\label{fig:infinite-mixed-graphic-group}{\small An infinite mixed-graphic group.}}
\end{figure}

\begin{rem}\label{rem:333333}
Since the every-zero mixed-graphic group $\{R_f(G)$; $[+][-]\}$ is a proper subset of the every-zero infinite mixed-graphic group $I^{+\infty}_{-\infty}(G, h; [+][-])$, we regard it \emph{every-zero mixed-graphic root-group}. Definition \ref{defn:mixed-graphic-groups-11} and Definition \ref{defn:mixed-graphic-groups-44} can help us to realize large scale of network overall topological encryption.\qqed
\end{rem}

\begin{defn} \label{defn:infinite-mixed-graphic-groups}
$^*$ \textbf{Every-zero infinite 3-dimension mixed Topcode-matrix group.} Let a bipartite graph $T$ have its own vertex set $V(T)=(X, Y)$, and $T$ admit a set-ordered total $W$-constraint coloring $F$, such that
\begin{equation}\label{eqa:555555}
\max\{F(x): x\in X\}=\max F(X)<\min F(Y)=\min\{F(y): y\in Y\}
\end{equation} Suppose that each colored graph $T_{i,j,k}$ of the set $I^{\infty}_{3-dime}(T, F; [+][-])=\{T_{i,j,k}:-\infty < i, j, k <+\infty\}$ holds $T_{i,j,k}\cong T$, and admits a total coloring $f_{i,j,k}$ defined by $f_{i,j,k}(x)= F(x)+i$ for $x\in X$, $f_{i,j,k}(y)= F(y)+j$ for $y\in Y$, and $f_{i,j,k}(uv)= F(uv)+k$ for $uv\in E(T)$, such that $f_{i,j,k}(x)<f_{i,j,k}(y)$, the Abelian additive operation
\begin{equation}\label{eqa:555555}
T_{i,j,k}[+]T_{r,s,t}[-]T_{a,b,c}=T_{\lambda,\mu,\gamma}
\end{equation} is defined by

(i) $f_{i,j,k}(x)+f_{r,s,t}(x)-f_{a,b,c}(x)=f_{\lambda,\mu,\gamma}(x), ~x\in X$;

(ii) $f_{i,j,k}(y)+f_{r,s,t}(y)-f_{a,b,c}(y)=f_{\lambda,\mu,\gamma}(y), ~y\in Y$;

(iii) $f_{i,j,k}(uv)+f_{r,s,t}(uv)-f_{a,b,c}(uv)=f_{\lambda,\mu,\gamma}(uv), ~uv\in E(T)$,\\
where $T_{a,b,c}$ is a preappointed zero, and $T_{0,0,0}\cong T$. We call the set $I^{\infty}_{3-dime}(T, F; [+][-])$ \emph{every-zero infinite 3-dimension mixed-graphic group}. \qqed
\end{defn}

\begin{rem}\label{rem:333333}
Similarly with Definition \ref{defn:infinite-mixed-graphic-groups}, we can define

(i) every-zero infinite 3-dimension mixed Topcode-matrix group; and

(ii) every-zero infinite 3-dimension mixed string group.

And moreover Definition \ref{defn:set-colored-graphic-group} is for graphic groups based on hypergraphs, as well as Hypergraph groups are defined in Definition \ref{defn:hypergraph-group-definition} and Definition \ref{defn:general-defi-hypergraph-groups}.\qqed
\end{rem}

\subsubsection{Pan-groups}

Since a thing set $S_{thing}=\{t_a,t_{a+1},\dots, t_b\}$ corresponds the integer set $[a,b]$, we can build up an \emph{every-zero thing group} $\{S_{thing};[+][-]\}$ by the every-zero graphic group $\{F(G);[+][-]\}$ defined in Definition \ref{defn:topo-code-graph-groups}. Here, each thing $t_i\in S_{thing}$ is a vector, or a set, or a number-based string, or a Topcode-matrix, or a hyperedge set, or a hypergraph, or a set-set, or a graph set, or a matrix set, or a $W$-constraint group, \emph{etc}. The colored graph $G$ in the every-zero graphic group $\{F(G);[+][-]\}$ shows a topological relationship of things of the thing set $S_{thing}$.

\begin{defn} \label{defn:pan-groups-pan-element-sets}
$^*$ An \emph{every-zero pan-group} $\{P_{an}(G);[+][-]\}$ has its own \emph{pan-element set} $P_{an}(G)=\{G_1,G_2,\dots, G_m\}$ with each pan-element $G_i$ admitting a total pan-coloring $h_i$, and $G_i\cong G_1$. The finite module Abelian additive operation
$$G_i[+]G_j[-]G_k=G_\lambda
$$ based on the pan-element set $P_{an}(G)$ is defined by
\begin{equation}\label{eqa:3333333}
h_i(w)[+]h_j(w)[-]h_k(w)=h_\lambda(w),~w\in V(G_1)\cup E(G_1)
\end{equation} with $\lambda=i+j-k~(\bmod~M)$ for any preappointed \emph{zero} $G_k\in \{P_{an}(G);[+][-]\}$.\qqed
\end{defn}

\subsection{The strength index of asymmetric topological keys}

The technical indicators of strength of asymmetric topology keys should include the following basic requirements:
\begin{asparaenum}[\textbf{\textrm{Sindex}}-1.]
\item Asymmetric topology keys consist of \emph{topological structures} and \emph{mathematical constraints}.
\item \label{Sindex:Byte-lengths} The lengths of asymmetric topological keys are:

\qquad (\ref{Sindex:Byte-lengths}-1) The byte lengths of asymmetric topological keys satisfy practical application requirements.

\qquad (\ref{Sindex:Byte-lengths}-2) The byte lengths of asymmetric topological keys can withstand AI attacks with quantum computing capabilities, such that deciphering them exceeds beyond existing computer computing power. For example, some byte lengths are not less than $1024^n$MB with $n\geq 4$.

\item \label{Sindex:Multiple-constraints} Asymmetric topology keys have:

\qquad (\ref{Sindex:Multiple-constraints}-1) \textbf{Multiple constraints.} Mathematical constraint set $R_{est}(c_1, c_2, \dots , c_m)$ with $m\geq 2$ on topological keys are numerous and complex.

\qquad (\ref{Sindex:Multiple-constraints}-2) \textbf{Multiple properties.}Topology keys not only have mathematical constraints, but also have randomness, one-to-many, and many-to-many properties.

\qquad (\ref{Sindex:Multiple-constraints}-3) \textbf{More matchings.} Each topological key has at least one of its own matching topological keys which are not easy to find out (see Fig.\ref{fig:one-vs-more}).
\item \label{Sindex:structures-spaces} \textbf{Complex topological structures and huge topological space.} The topological structures for topological keys are:

\qquad (\ref{Sindex:structures-spaces}-1) Considering the resistance to quantum computing, the cardinality of the topological structure space for constructing topological keys is at least $2e=2^{200}$; and

\qquad (\ref{Sindex:structures-spaces}-2) Since two numbers $n_{23}$ and $n_{24}$ of different topological structures of graphs on 23 vertices and 24 vertices hold $n_{23}\approx 2^{179}$ and $n_{24}\approx 2^{197}$, then the vertex number of a graph for making some asymmetric topological keys is not less than 50.

\item \label{Sindex:increased-randoms}\textbf{Randomness.} Asymmetric topology keys can be increased randomly (see Fig.\ref{fig:increase-random}) based on the following techniques:

\qquad (\ref{Sindex:increased-randoms}-1) Topological structures can be increased randomly (see Fig.\ref{fig:increase-random}).

\qquad (\ref{Sindex:increased-randoms}-2) Parameterized asymmetric topology keys can be made by Definition \ref{defn:kd-w-type-colorings} and a parameterized Topcode-matrix $P_{ara}(G,F|k,d)$ defined in (\ref{eqa:definition-parameterized-topcode-matrix}).

\qquad (\ref{Sindex:increased-randoms}-3) Parameterized number-based $(k,d)$-strings induced from parameterized Topcode-matrix $P_{ara}(G,F|k,d)$ can rely on an arbitrary function $d=f(k)$ having infinite integer points $(k,d)$ in the xOy-plane.

\item \label{Sindex:More-NP-type problem} \textbf{NP-type problems.} Each topological key has a property $P$ which is related with NP-type problem, and there is no polynomial method to realize the property $P$.

\qquad (\ref{Sindex:More-NP-type problem}-1) Many mathematical conjectures are NP-hard, or NP-complete. For example, computing the chromatic number is NP-hard (Ref. \cite{Garey-M-R-Johnson-1979} and \cite{Garey-Johnson-Stockmeyer-1974}). For edge coloring, the proof of Vizing's result gives an algorithm that uses at most $\Delta(G)+1$ colors. However, deciding between the two candidate values $\Delta(G)$ and $\Delta(G)+1$ for the edge chromatic number is NP-complete (Ref. \cite{Holyer-I-1981}). The harmonious coloring problem is NP-hard (Ref. \cite{M-Kubale-harmonious-2000}). Determining the chromatic index $\chi \,'(G)\leq 3$ is NP-complete (Ref. \cite{Holyer-I-1981}, \cite{Hanna-Andrzej-Marek-2016}).

\qquad (\ref{Sindex:More-NP-type problem}-2) \textbf{Sharp-P-hard (\#P-hard) and sharp-P-complete (\#P-complete)} problems. For example, computing the coefficients of the chromatic polynomial is Sharp-P-hard, and Problem \ref{question:sharp-NP-complete-total-coloring} is just \#P-hard, because of determining the total chromatic number is NP-hard. Some Sharp-P-complete problems are:

\qquad \textbf{How many} different variable assignments will satisfy a given 2-satisfiability problem?

\qquad \textbf{How many} perfect matchings are there for a given bipartite graph?

\qquad \textbf{How many} graph colorings using $k$ colors are there for a particular graph?

\qquad \textbf{How many} edge colorings using $\chi \,'(G)$ colors are there for a connected graph $G$?

\qquad \textbf{How many} total colorings using $\chi \,''(G)$ colors are there for a connected graph $G$?

\qquad (\ref{Sindex:More-NP-type problem}-3) \textbf{Finding} colored graphs is the Subgraph isomorphism NP-complete problem.

\item \textbf{Infinite keys.} Asymmetric topology keys have the one-time pad system (key's infiniteness), since the Shannon theory has already proven that one-time pad system is a perfect secret system, see techniques from parameterized Topcode-matrices, Parameterized number-based $(k,d)$-strings, parameterized hypergraphs, \emph{etc.}
\item \textbf{Chinese character graphs.} The topological structure of topological keys are Hanzi-graphs (Chinese character graphs). It is difficult to vertex-split a (total colored) graph to obtain a group of (total colored) Hanzi-graphs, see Fig.\ref{fig:G-v-splt-Hanzi-graphs} and Fig.\ref{fig:G-v-splt-Hanzi-graphs-11}.
\item The above techniques for asymmetric topology keys can handle this worst-case scenario: if the deciphers are very familiar with the creation of asymmetric topology keys and have obtained the number-based strings of encryption.
\end{asparaenum}

\begin{problem}\label{question:sharp-NP-complete-total-coloring}
If a connected graph $G$ admits a proper total coloring $f:V(G)\cup E(G)\rightarrow [1,\chi\,''(G)]$, \textbf{how many} such proper total colorings does $G$ admit?
\end{problem}

However, we have proved: It is impossible to find out a total colored graph or a Topcode-matrix from a number-based string.

\begin{figure}[h]
\centering
\includegraphics[width=16.4cm]{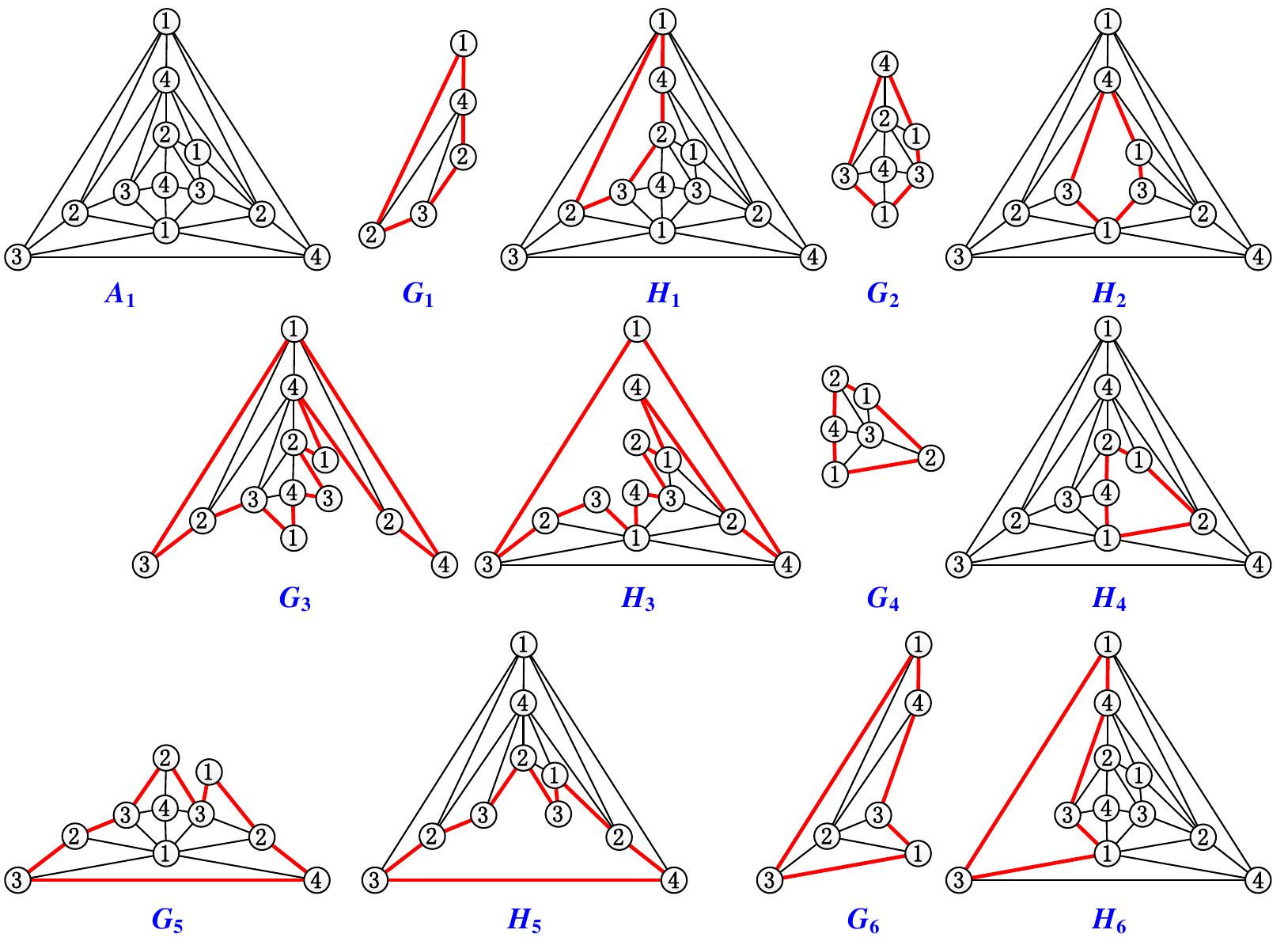}\\
\caption{\label{fig:one-vs-more}{\small One maximal planar graph $A_1$ produces more matchings of public-keys $G_1,G_2,\dots ,G_6$ and private-keys $H_1,H_2,\dots ,H_6$, such that $A_1=G_i\big [\ominus^{cyc}_k\big ]H_i$ for each $i\in [1,6]$ and cycle length $k=5,8,12$.}}
\end{figure}

\begin{thm}\label{thm:666666}
A maximal planar graph of $p$ vertices with $p\geq 10$ produces two or more matchings of public-keys and private-keys under the cycle-coinciding operation, like that shown in Fig.\ref{fig:one-vs-more}.
\end{thm}

\begin{figure}[h]
\centering
\includegraphics[width=13cm]{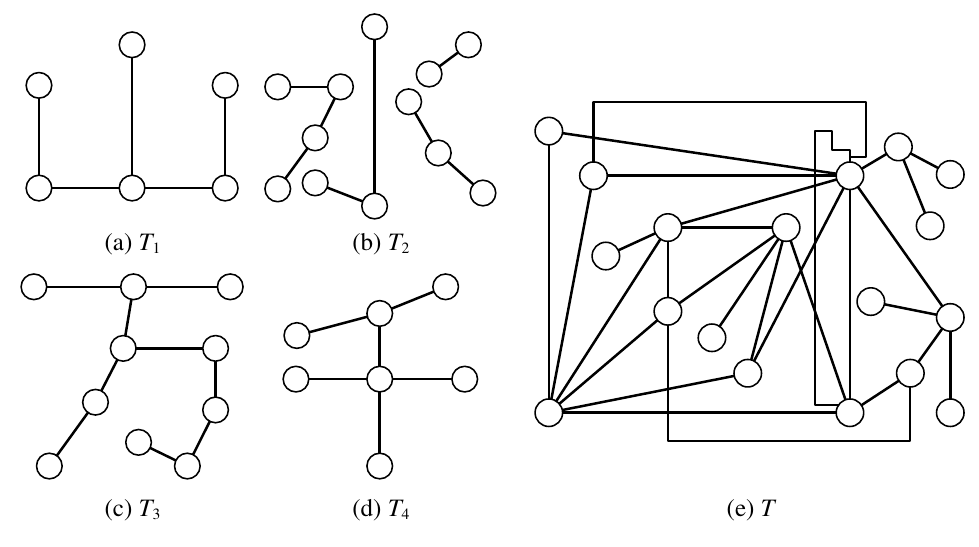}\\
\caption{\label{fig:G-v-splt-Hanzi-graphs}{\small A group of Hanzi-graphs $T_1,T_2,T_3,T_4$ obtained by vertex-splitting the total colored connected graph $T$, where $T=[\bullet]^4_{k=1}T_k$.}}
\end{figure}

\begin{problem}\label{question:444444}
\textbf{How many} groups of Hanzi-graphs \textbf{can} a connected graph be vertex-split apart?
\end{problem}

\begin{figure}[h]
\centering
\includegraphics[width=16.4cm]{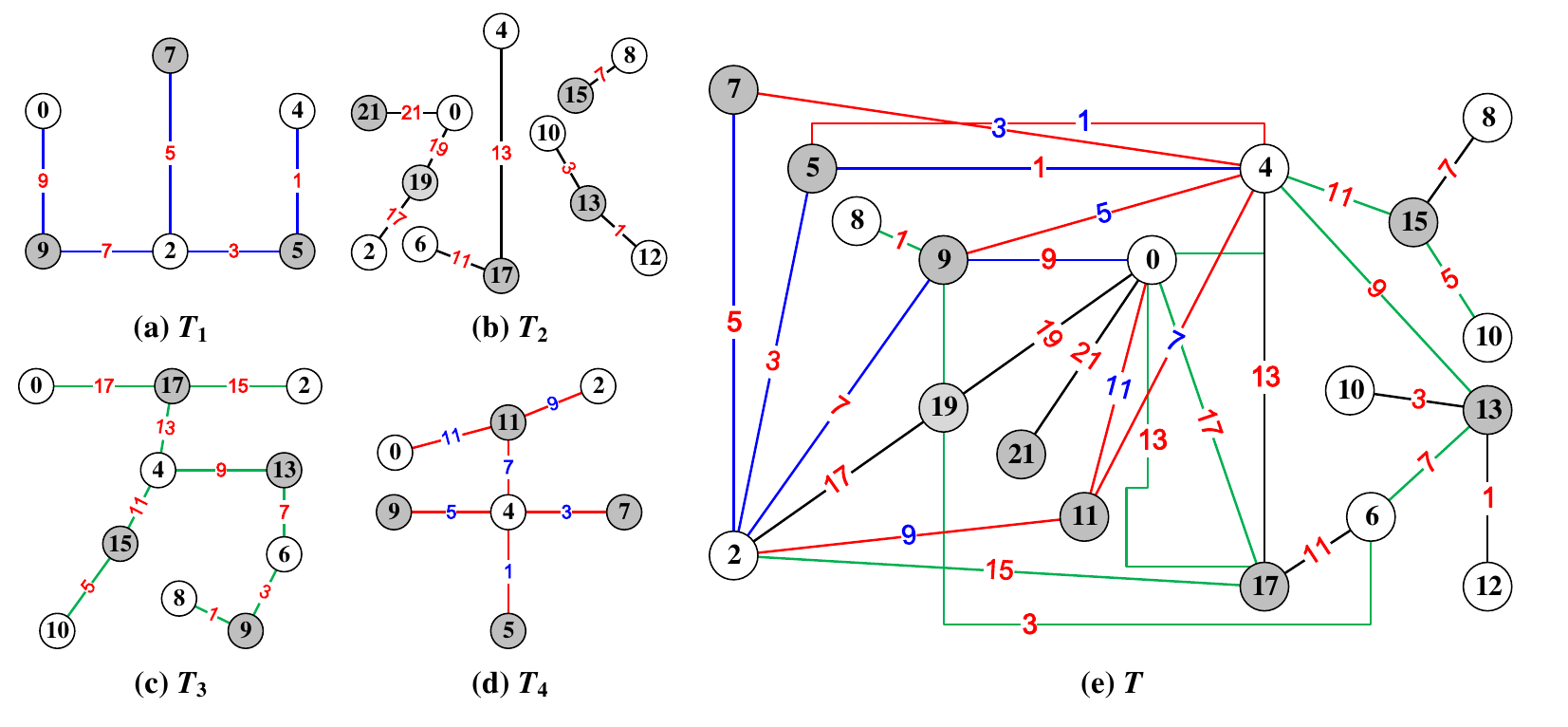}\\
\caption{\label{fig:G-v-splt-Hanzi-graphs-11}{\small Five total colored graphs $T_1,T_2,T_3,T_4$ and $T$, refer to Fig.\ref{fig:G-v-splt-Hanzi-graphs}, where $T=[\bullet]^4_{k=1}T_k$.}}
\end{figure}

\subsection{Topological authentication}

\begin{defn} \label{defn:topo-authentication-number-string}
\cite{Yao-Su-Ma-Wang-Yang-arXiv-2202-03993v1} A \emph{topological authentication $T_{auth}$ of number-based strings} is defined as
\begin{equation}\label{eqa:topo-authentication-00}
{
\begin{split}
T_{auth}=&F\big [ G\leftarrow \lambda(H),\quad g\leftarrow \varphi(f),\quad T_{code}(G,g)\leftarrow \varphi (T_{code}(H,f),\\
&S_{tring}(m)=\theta(T_{code}(G,g))\leftarrow S_{tring}(n)=\theta(T_{code}(H,f))\big ]
\end{split}}
\end{equation} based on a \emph{topological private-key} $P_{ri}[G,g,T_{code}(G,g), S_{tring}(m)] $ and a \emph{topological public-key} $P_{ub}[H$, $f$, $T_{code}(H,f)$, $S_{tring}(n)] $, also
\begin{equation}\label{eqa:555555}
P_{ri}[G,g,T_{code}(G,g), S_{tring}(m)]\leftarrow P_{ub}[H,f,T_{code}(H,f),S_{tring}(n)]
\end{equation} where $\lambda$ is a \emph{graph operation}, and $\varphi$ is a \emph{transformation function} on two colorings or two Topcode-matrices, and $\theta$ is a \emph{$W$-constraint function} for producing number-based strings, as well as the \emph{public-key graph} $H$ admits a $W_i$-constraint coloring $f$ which induces a Topcode-matrix $T_{code}(H,f)$, and the \emph{private-key graph} $G$ admits a $W_j$-constraint coloring $g$ which induces a Topcode-matrix $T_{code}(G,g)$ (Ref. Fig.\ref{fig:authentication-system}).\qqed
\end{defn}

\begin{defn} \label{defn:topo-authentication-multiple-variables}
\cite{Yao-Su-Ma-Wang-Yang-arXiv-2202-03993v1} A \emph{topological authentication} $\textbf{T}_{\textbf{a}}\big [\textbf{X},\textbf{Y}\big ]$ \emph{of multiple variables} is defined as follows
\begin{equation}\label{eqa:topo-authentication-11}
\textbf{T}_{\textbf{a}}\big [\textbf{X},\textbf{Y}\big ] =P_{ub}(\textbf{X}) \rightarrow _{\textbf{F}} P_{ri}(\textbf{Y})
\end{equation} where $P_{ub}(\textbf{X})=(\alpha_1,\alpha_2,\dots ,\alpha_m)$ and $P_{ri}(\textbf{Y})=(\beta_1,\beta_2,\dots ,\beta_m)$ both are \emph{variable vectors}, in which both $\alpha_1$ and $\beta_1$ are two graphs or sets of graphs (resp. colored graphs, uncolored graphs), and $\textbf{F}=(\theta_1,\theta_2,\dots $, $\theta_m)$ is an \emph{operation vector}, $P_{ub}(\textbf{X})$ is a \emph{topological public-key vector} and $P_{ri}(\textbf{Y})$ is a \emph{topological private-key vector}, such that $\theta_k(\alpha_k)\rightarrow \beta_k$ for each $k\in [1,m]$ with $m\geq 1$.\qqed
\end{defn}

\begin{figure}[h]
\centering
\includegraphics[width=16cm]{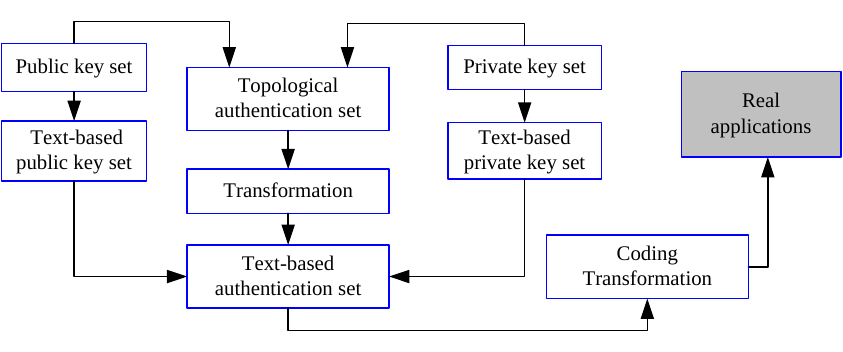}
\caption{\label{fig:authentication-system}{\small A topological authentication system, cited from \cite{Yao-Su-Ma-Wang-Yang-arXiv-2202-03993v1}.}}
\end{figure}

\begin{defn} \label{defn:4-coloring-characterized-graph}
\cite{Jin-Xu-55-56-configurations-arXiv-2107-05454v1} Let $C_4(G)=\{f_1, f_2,\dots ,f_n\}$ be the set of all different proper vertex $4$-colorings of a planar graph $G$. A \emph{$4$-coloring characterized graph} $C_c(G)$ has its own vertex set $V(C_c(G))=C_4(G)$, and $C_c(G)$ has an edge $f_if_j$ with two ends $f_i,f_j\in V(C_c(G))$ if $f_j$ is obtained by doing the Kempe transformation to $f_i$, or exchange the colors of some vertices of the planar graph $G$ under $f_i$ to make the proper vertex $4$-coloring $f_j$ of the planar graph $G$.\qqed
\end{defn}

\begin{example}\label{exa:8888888888}
In Fig.\ref{fig:44-more-public-keys}, we have four topological authentications $A_i=P_i\big [\ominus ^{cyc}_{m_i}\big ]T_i$ for $i\in [1,4]$, where each public-key $P_i\in G_{pub}$ and each private-key $T_i\in G_{pri}$, and the operation ``$\big [\ominus ^{cyc}_{m_i}\big ]$'' is defined in Definition \ref{defn:W-splitting-coinciding-operation} and Remark \ref{rem:edge-split-coinciding-operation}. The 4-coloring $f_i$ of each \emph{topological authentication} $A_i=P_i\big [\ominus ^{cyc}_{m_i}\big ]T_i$ for $i\in [1,4]$ is determined by the 4-coloring $g_i$ of each \emph{public-key} $P_i$ and the 4-coloring $h_i$ of each \emph{private-key} $T_i$. In Fig.\ref{fig:55-more-public-keys}, each \emph{topological authentication} $A_i=P_i\big [\ominus ^{cyc}_{m_i}\big ]T_i$ for $i\in [1,4]$.

In a topological authentication $\textbf{T}_{\textbf{a}}\langle\textbf{X},\textbf{Y}\rangle $ defined in Definition \ref{defn:topo-authentication-multiple-variables}, we have two variable vectors
$$
P_{ub}(\textbf{X})=(G_{pub}, g_1,g_2, g_3,g_4),\quad P_{ri}(\textbf{Y})=(G_{pri}, h_1,h_2,h_3 ,h_4)
$$ and an operation vector $\textbf{F}=\big (\big [\ominus^{cyc}_{m_i}\big ],f_1,f_2,f_3 ,f_4 \big )$, such that $G_{pub}\rightarrow G_{pri}$ made by $A_i=P_i\big [\ominus^{cyc}_{m_i}\big ]T_i$ with 4-coloring $f_i=g_i\uplus h_i$ for $m_i\geq 3$ and $i\in [1,4]$.
\end{example}

\begin{figure}[h]
\centering
\includegraphics[width=15.4cm]{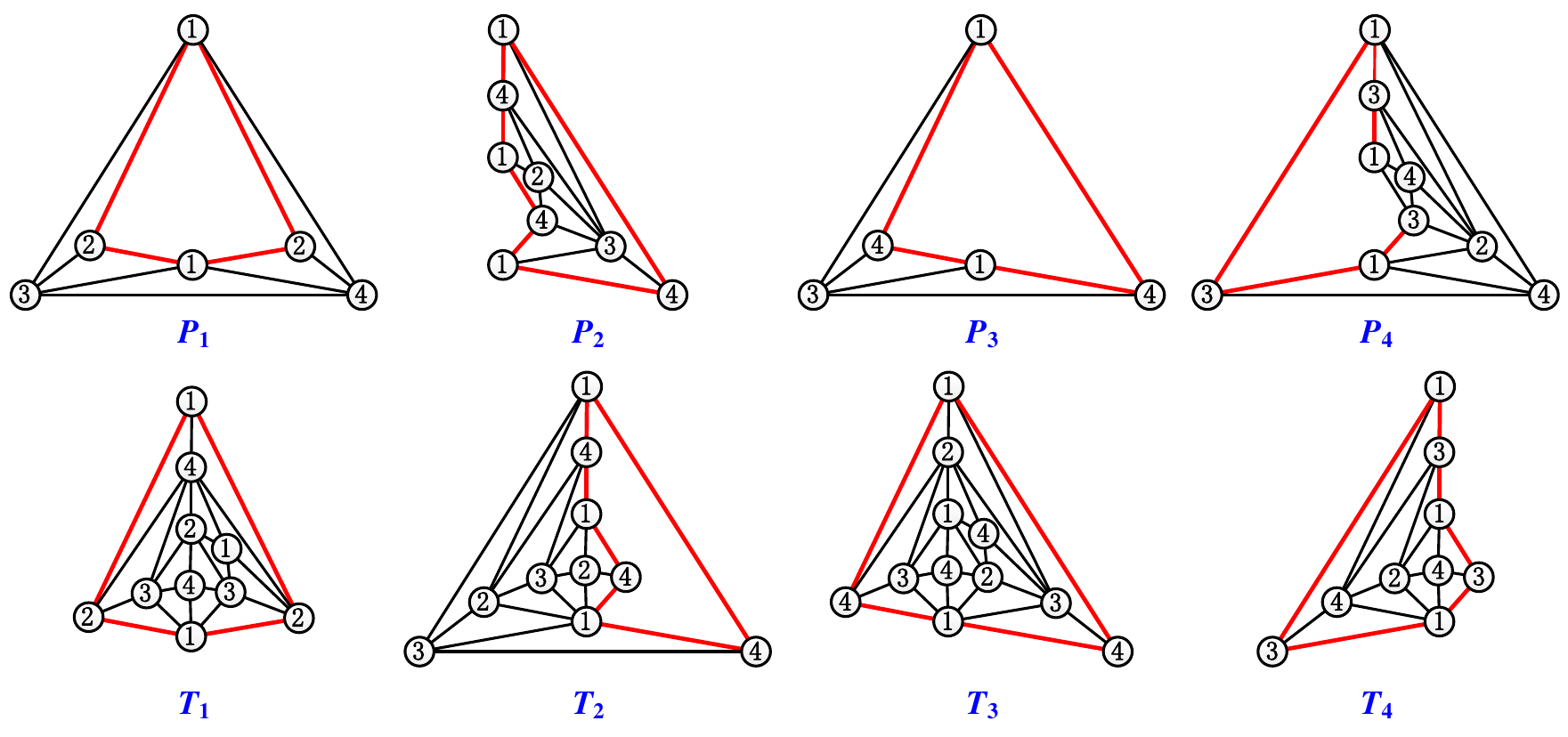}\\
\caption{\label{fig:44-more-public-keys} {\small A public-key group $G_{pub}=(P_1,P_2,P_3 ,P_4)$ and a private-key group $G_{pri}=(T_1,T_2,T_3 ,T_4)$, cited from \cite{Yao-Su-Ma-Wang-Yang-arXiv-2202-03993v1}.}}
\end{figure}

\begin{figure}[h]
\centering
\includegraphics[width=15.4cm]{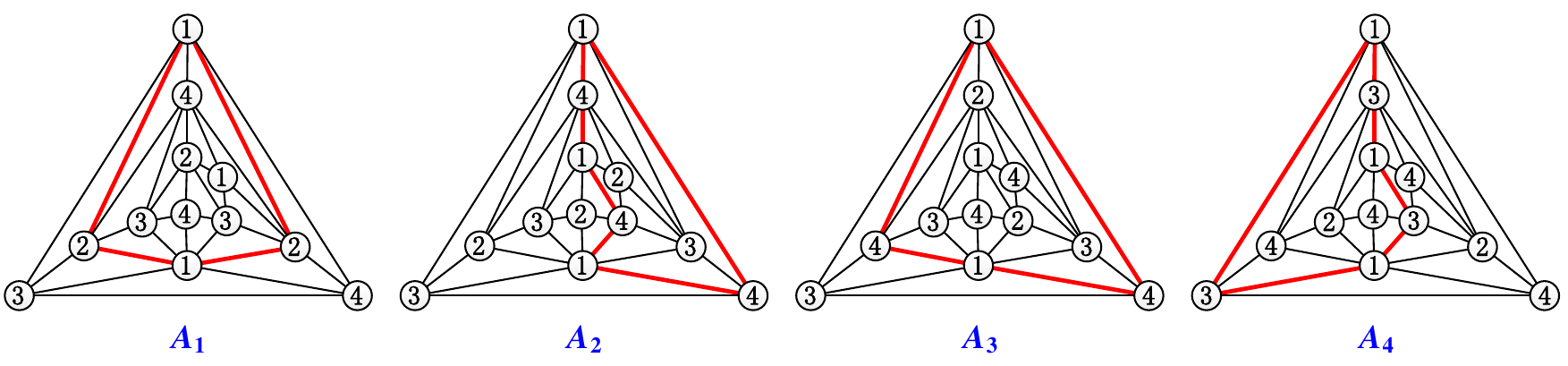}\\
\caption{\label{fig:55-more-public-keys} {\small A topological authentication group $A_{uth}=(A_1,A_2,A_3 ,A_4)$, cited from \cite{Yao-Su-Ma-Wang-Yang-arXiv-2202-03993v1}.}}
\end{figure}

\begin{example}\label{exa:semi-maximal-planar-graphs-4-coloring}
Using the notation and the terminology appeared in Definition \ref{defn:topo-authentication-multiple-variables}. A maximal planar graph $G=G^C_{out}\big [\ominus^C_k\big ]G^C_{in}$ admits a proper vertex $4$-coloring $f$, so the semi-maximal planar graph $G^C_{out}$ admits a proper vertex $4$-coloring $f_{out}$ induced by $f$, and the semi-maximal planar graph $G^C_{in}$ admits a proper vertex $4$-coloring $f_{in}$ induced by $f$. So, we have a topological public-key
\begin{equation}\label{eqa:semi-maximal-planar-graph-au-11}
P_{ub}(\textbf{X})=\big (G^C_{out}, f_{out}, T_{code}(G^C_{out}),A(G^C_{out})\big )
\end{equation}
and a topological private-key
\begin{equation}\label{eqa:semi-maximal-planar-graph-au-22}
P_{ri}(\textbf{Y})=\big (G^C_{in},f_{in}, T_{code}(G^C_{in}),A(G^C_{in})\big )
\end{equation}
and an operation vector
\begin{equation}\label{eqa:semi-maximal-planar-graph-au-33}
\textbf{F}=\left (\theta_1,\theta_2,\theta_3,\theta_4\right)=\Big (\big [\ominus^C_k\big ],f_{out}\cup f_{in},T_{code}(G^C_{out})\uplus T_{code}(G^C_{in}), A(G)\Big )
\end{equation}
where $A(G^C_{out})$, $A(G^C_{in})$ and $A(G)$ are the adjacent matrices of three graphs $G$, $G^C_{out}$ and $G^C_{in}$. Notice that two adjacent matrices $A(T)$ and $A(L)$ are not similar from each other if two non-isomorphic graphs $T$ and $L$, there is no matrix $B$ holding $B^{-1}A(T)B=A(L)$ true.

By Eq.(\ref{eqa:semi-maximal-planar-graph-au-11}), Eq.(\ref{eqa:semi-maximal-planar-graph-au-22}) and Eq.(\ref{eqa:semi-maximal-planar-graph-au-33}), we get a topological authentication
\begin{equation}\label{eqa:555555}
P_{ub}(\textbf{X}) \rightarrow _{\textbf{F}} P_{ri}(\textbf{Y})
\end{equation} made by

$G=\theta_1(G^C_{out},G^C_{in})=G^C_{out}\big [\ominus^C_k\big ]G^C_{in}$, $\theta_2(f_{out}, f_{in})=f_{out}\cup f_{in}$,

$\theta_3(T_{code}(G^C_{out}), T_{code}(G^C_{in}))=T_{code}(G^C_{out})\uplus T_{code}(G^C_{in})$, and

$A(G)=\theta_4(A(G^C_{out},A(G^C_{in}))=A(G^C_{out}\cup A(G^C_{in})$.

\vskip 0.4cm

Moreover, we may meet a set of topological public-keys consisted of maximal planar graphs $G_i=G^C_{out}\big [\ominus^C_k\big ]G^C_{in}(i)$ for $i\in [1,m]$ with $m\geq 2$, and each semi-maximal planar graph $G^C_{in}(i)$ admits a proper vertex $4$-coloring $f^i_{in}$, so we get a group of topological private-key vectors
\begin{equation}\label{eqa:555555}
P_{ri}(\textbf{Y}_i)=\Big (G^C_{in}(i),f^i_{in}, T_{code}(G^C_{in}(i)),A(G^C_{in}(i))\Big )
\end{equation}
and a group of operation vectors
\begin{equation}\label{eqa:555555}
\textbf{F}_i=\left (\theta^i_1,\theta^i_2,\theta^i_3,\theta^i_4\right)=\Big (\big [\ominus^C_k\big ],f_{out}\cup f^i_{in},T_{code}(G^C_{out})\uplus T_{code}(G^C_{in}(i)), A(G_i)\Big )
\end{equation} for $i\in [1,m]$. Thereby, we get a group of topological authentications
\begin{equation}\label{eqa:topo-authentication-group}
\textbf{T}_{\textbf{a}}\langle\textbf{X},\textbf{Y}_i\rangle =P_{ub}(\textbf{X}) \rightarrow _{\textbf{F}_i} P_{ri}(\textbf{Y}_i),~i\in [1,m]
\end{equation}
\end{example}

\begin{example}\label{exa:characterized-graph-vs-mpgs}
Using the concepts and notations in Definition \ref{defn:topo-authentication-multiple-variables} and Definition \ref{defn:4-coloring-characterized-graph}. Let $C_c(G)$ be a proper vertex $4$-coloring characterized graph of a planar graph $G$ with its coloring set $C_4(G)=\{f_1, f_2,\dots ,f_n\}$ of all different proper vertex $4$-colorings. We select randomly a graph $J$ admitting a proper vertex coloring $F$, and make a topological public-key
\begin{equation}\label{eqa:555555}
P_{ub}(\textbf{X}_{\textrm{characg}})=(J, F)
\end{equation}
where ``characg = characterized graph'', and \textbf{find} out a maximal planar graph $H$ with its 4-coloring set $C_4(H)$ as a topological private-key below
\begin{equation}\label{eqa:555555}
P_{ri}(\textbf{Y}_{\textrm{mpg}})=(H,C_4(H))
\end{equation} where ``mpg = maximal planar graph'', and \textbf{determine} an operation vector $\textbf{F}=(\theta_1,\theta_2)$, where $\theta_1(J)\rightarrow H$, $\theta_2(F)\rightarrow C_4(H)$, that is, $J=C_c(H)$ and $F(J)=C_4(H)$. Thereby, we obtain a topological authentication
\begin{equation}\label{eqa:characterized-graph-vs-mpgs}
\textbf{T}_{\textbf{a}}\langle\textbf{X}_{\textrm{characg}},\textbf{Y}_{\textrm{mpg}}\rangle =P_{ub}(\textbf{X}_{\textrm{characg}}) \rightarrow _{\textbf{F}} P_{ri}(\textbf{Y}_{\textrm{mpg}})
\end{equation} However, realizing the topological authentication $\textbf{T}_{\textbf{a}}\langle\textbf{X}_{\textrm{characg}},\textbf{Y}_{\textrm{mpg}}\rangle$ defined in Eq.(\ref{eqa:characterized-graph-vs-mpgs}) is not easy.
\end{example}

\begin{example}\label{exa:characterized-graph-vs-general-graph}
Let $C^W_{\textrm{harac}}$ be a $W$-type coloring characterized graph admitting a coloring $F$ defined on a $W$-type coloring set $C_{olor}=\{f_1, f_2,\dots ,f_m\}$. We have a topological public-key vector $P_{ub}(\textbf{X}_W)=(C^W_{\textrm{harac}}, F)$, and we will \textbf{determine} a topological private-key graph $H$ admitting $W$-type colorings to form a topological private-key vector $
P_{ri}(\textbf{Y}_W)=(H,C_W(H))$, where $C_W(H)$ is the set of all different $W$-type colorings of $H$; and \textbf{find} an operation vector $\textbf{F}_W=(\theta_1,\theta_2)$ with
$$
\theta_1 \big (C^W_{\textrm{harac}}\big )\rightarrow H,~\theta_2(F)\rightarrow C_W(H)
$$ that is, $C^W_{\textrm{harac}}=C^W_{\textrm{harac}}(H)$ and $F\big (C^W_{\textrm{harac}}\big )=C_{olor}=C_W(H)$, we get a topological authentication
\begin{equation}\label{eqa:characterized-graph-vs-general-graph}
\textbf{T}_{\textbf{a}}\langle\textbf{X}_W,\textbf{Y}_W\rangle =P_{ub}(\textbf{X}_W) \rightarrow _{\textbf{F}_W} P_{ri}(\textbf{Y}_W)
\end{equation}
\end{example}

\begin{defn}\label{defn:99-4-color-star-graphic-lattice}
\cite{Bing-Yao-2020arXiv} \textbf{The 4-color ice-flower system.} Each star $K_{1,d}$ with $d\in [2,M]$ admits a proper vertex-coloring $f_i$ with $i\in [1,4]$ defined as $f_i(x_0)=i$, $f_i(x_j)\in [1,4]\setminus\{i\}$, and $f_i(x_s)\neq f_i(x_t)$ for some $s,t\in [1,d]$, where $V(K_{1,d})=\{x_0,x_j:j\in [1,d]\}$. For each pair of $d$ and $i$, $K_{1,d}$ admits $n(d,i)$ proper vertex-colorings like $f_i$ defined above. Such colored star $K_{1,d}$ is denoted as $P_dS_{i,k}$, we have a set $(P_{d}S_{i,k})^{n(a,i)}_{k=1}$ with $i\in [1,4]$ and $d\in [2,M]$, and moreover we obtain a \emph{4-color ice-flower system} $\textbf{\textrm{I}}_{ce}(PS,M)=I_{ce}(P_{d}S_{i,k})^{n(a,i)}_{k=1})^{4}_{i=1})^{M}_{d=2}$, which induces a \emph{$4$-color star-graphic lattice}
\begin{equation}\label{eqa:4-color-star-system-lattices}
\textbf{\textrm{L}}(\Delta\overline{\ominus} \textbf{\textrm{I}}_{ce}(PS,M))=\big \{[\Delta\overline{\ominus}]^{A}_{(d,i,k)} a_{d,i,k}P_{d}S_{i,j}:~a_{d,i,k}\in Z^0,~P_{d}S_{i,j}\in \textbf{\textrm{I}}_{ce}(PS,M)\big \}
\end{equation} for $\sum ^{A}_{(d,i,k)} a_{d,i,k}\geq 3$ with $A=|\textbf{\textrm{I}}_{ce}(PS$, $M)|$, and the operation ``$[\Delta\overline{\ominus}]$'' is doing a series of leaf-coinciding operations to colored stars $a_{d,i,k}P_{d}S_{i,j}$ such that the resultant graph to be a planar graph with each inner face being a triangle.\qqed
\end{defn}

See two examples shown in Fig.\ref{fig:process-leaf-splitting} and Fig.\ref{fig:4-color-star-system-example} about the 4-color ice-flower systems defined in Definition \ref{defn:99-4-color-star-graphic-lattice}.

\begin{figure}[h]
\centering
\includegraphics[width=16cm]{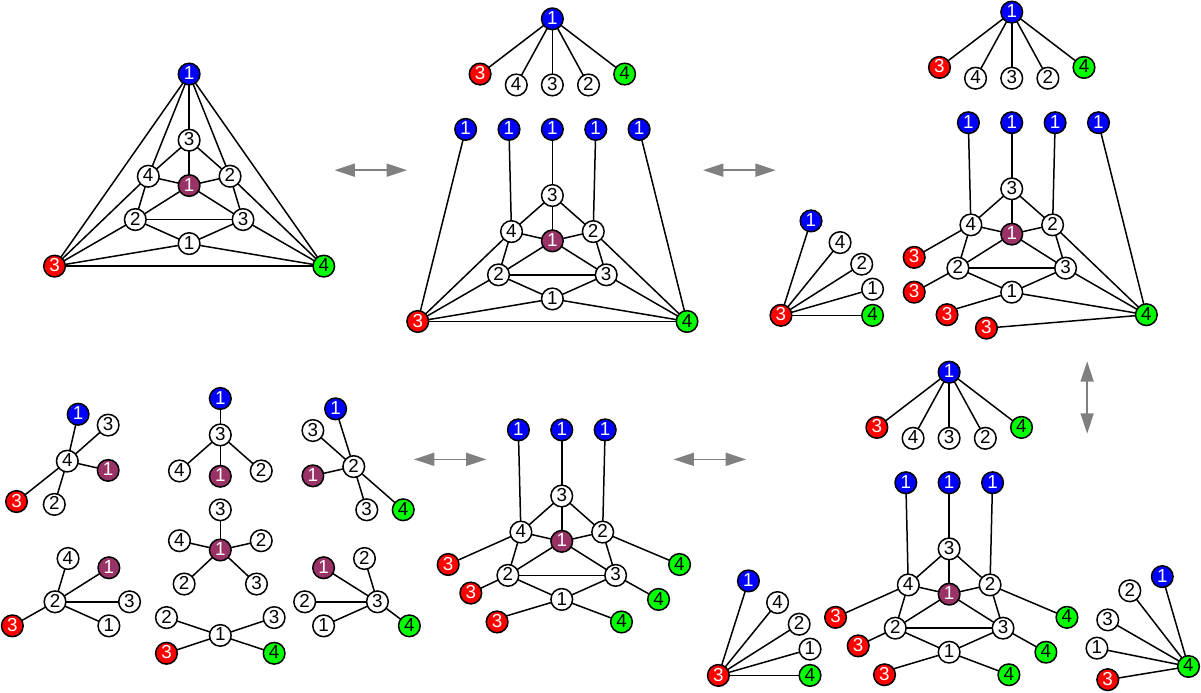}
\caption{\label{fig:process-leaf-splitting}{\small A process of doing leaf-splitting operations, cited from \cite{Bing-Yao-2020arXiv}.}}
\end{figure}

\begin{figure}[h]
\centering
\includegraphics[width=16cm]{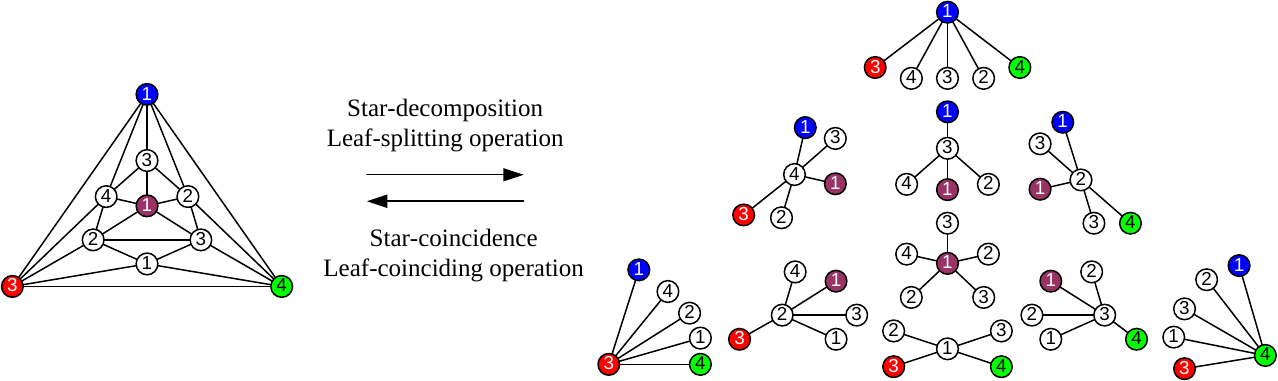}
\caption{\label{fig:4-color-star-system-example}{\small An example for understanding the $4$-coloring ice-flower system, cited from \cite{Bing-Yao-2020arXiv}.}}
\end{figure}

\begin{problem}\label{question:444444}
\textbf{Does} the $4$-color star-graphic lattice $\textbf{\textrm{L}}(\Delta\overline{\ominus} \textbf{\textrm{I}}_{ce}(PS,M))$ defined in Eq.(\ref{eqa:4-color-star-system-lattices}) of Definition \ref{defn:99-4-color-star-graphic-lattice} contain each maximal planar graph?
\end{problem}

\begin{conj}\label{conj:0000000000}
\cite{Yao-Su-Ma-Wang-Yang-arXiv-2202-03993v1} The planar dual graph $G^*$ of a maximal planar graph $G$ can be decomposed into a spanning tree $T$ and a perfect matching $M$ such that $E(G^*)=E(T)\cup M$.
\end{conj}

\begin{problem}\label{qeu:444444}
Adding the edges of an edge set $E^*\not \subset E(G)$ to a maximal planar graph $G$ admitting 4-colorings produces an edge-added graph $G+E^*$ admitting a 4-coloring. \textbf{Determine} the edge-added graph set $\{G+E^*\}$ over all possible edge sets $E^*$.
\end{problem}

\subsection{TKPDRA-center}

\subsubsection{Functions of TKPDRA-center}

For the goal to topologically encrypt overall networks and manage keys, we set up TKPDRA-center (the center of topology key-pairs deduced from
asymmetric topology code theory). With the help of graphic groups, Topcode-matrix groups, topological string groups and hypergraph groups in topology code theory, TKPDRA-center conducts overall network topological encryption, key generation, key distribution, and key management to networks, and maintain the security of communities and local networks all the time, in which the key management contains: Key registration, key production, key distribution, key usage, key authentication, key revocation and key update \emph{etc}.

TKPDRA-center algorithm is abbreviated as ``TKPDRAC-algorithm'' hereafter.

\vskip 0.4cm

\textbf{TKPDRA-center has the following main functions:}
\begin{asparaenum}[\textbf{\textrm{Func}}-1.]
\item TKPDRA-center uses Topocode-groups to create keys and distribute key-packages to users in networks, and can randomly adjust the \emph{zeros} of Topocode-groups.
\item TKPDRA-center is responsible for verifying the topology signature or number-based string authentication between each pair of users in networks.
\item TKPDRA-center help users to avoid creating keys by high-tech and to memorize various keys, even keys with ultra long bytes. Chinese users can request Hanzi-graphs (Chinese character graphs) from TKPDRA-center and randomly select a group of Hanzi-graphs or topological string groups according to the TKPDRA-center's instructions. In this way, they can obtain their favorite and easy to remember Hanzi-graphs topology signature key-pairs and string key-pairs.
\item TKPDRA-center can provide users with a protection mechanism for multiple topology authentications.
\item TKPDRA-center, for resisting attacks out of networks, can randomly replace the \emph{zeros} of graphic groups, topological string groups and hypergraph groups for users in networks, or TKPDRA-center can replace running Topocode-groups by other Topocode-groups.
\item TKPDRA-center can timely change the security system of the entire networks, resist damage and attacks out of networks.
\item TKPDRA-center can customize the private authentications for particular users for distinguishing them from other users in networks.
\item TKPDRA-center is equipped with anti-key-destruction software, just like ordinary antivirus software, which constantly monitors TKPDRA-center and its service objects, protecting the keys of network users.
\end{asparaenum}

\subsubsection{TKPDRAC-algorithms}

We have designed the following algorithms of TKPDRA-center for using asymmetric topological encryption.

\textbf{TKPDRAC-algorithm-I:} The files of users in the network are transmitted through TKPDRA-center.

\textbf{Step-I-1.} Alice encrypts a plaintext $F$ by her private topological signature $G_{\textrm{Apri}}$ and private number-based string $s_{\textrm{Apri}}$, the resultant ciphertext is denoted as $F_2$ which has a 2-level protection. And Alice sends the ciphertext $F_2$ to TKPDRA-center, and requests TKPDRA-center to perform technical processing on the ciphertext $F_2$ and send it to Bob.

\textbf{Step-I-2.} TKPDRA-center makes a package $P(F_2)$ containing the ciphertext $F_2$, Alice's public topological signature $G_{\textrm{Apub}}$ and Alice's public number-based string $s_{\textrm{Apub}}$, and then uses Bob's public topological signature $G_{\textrm{Bpub}}$ and Bob's public number-based string $s_{\textrm{Bpub}}$ to encode this package $P(F_2)$ to produce the resultant ciphertext $F_4$, which has a 4-level protection, and sends it to Bob.

\textbf{Step-I-3.} After receiving the ciphertext $F_4$ from TKPDRA-center, Bob uses his private topological signature $G_{\textrm{Bpri}}$ (recognizing that this ciphertext was sent to himself) and his private number-based string $s_{\textrm{Bpri}}$ to decrypt the ciphertext $F_4$, such that he obtains the ciphertext $F_2$, Alice's public topological signature $G_{\textrm{Apub}}$ and Alice's public number-based string $s_{\textrm{Apub}}$. Notice that Bob has the authentications of Bob's topological signature and Bob's number-based string.

\textbf{Step-I-4.} Bob uses Alice's public topological signature $G_{\textrm{Apub}}$ to the ciphertext $F_2$, so he knows this ciphertext $F_1$ came from Alice, and he uses Alice's public number-based string $s_{\textrm{Apub}}$ to decode $F_1$, finally, Bob can read the plaintext $F$. Notice that two authentications of Alice's topological signature and Alice's number-based string must pass through the authentications of TKPDRA-center.

The advantages of TKPDRAC-algorithm-I are as follows:

(i) \textbf{High security.} Four keys including Alice's public topological signature $G_{\textrm{Apub}}$, Alice's public number-based string $s_{\textrm{Apub}}$, Bob's public topological signature $G_{\textrm{Bpub}}$ and Bob's public number-based string $s_{\textrm{Bpub}}$ are in TKPDRA-center, not publicly disclosed.

(ii) \textbf{Simplicity and convenience.} TKPDRA-center help network users to make keys and encrypt files, simplify the encryption process of files for network, ensure the authenticity and completeness of the ciphertexts.

(iii) \textbf{Multiple users.} TKPDRA-center can help Alice to send the files to Alice's other Communication users Bob-1, Bob-2, $\dots$, Bob-$m$, such that Alice doesn't need to have access to the public topological signature and public number-based string of each of her users. This reduces the technical requirements for network users and achieves secure and efficient communication.

(iv) \textbf{Multiple communication methods.} Since TKPDRA-center help that Alice and Bob both have mastered the non-public topological signatures and non-public number-based strings, such that Bob and Alice, in future communication, are possible to communicate without relying on TKPDRA-center.

\vskip 0.4cm

There are three schemes Fig.\ref{fig:Topological-signature-22}, Fig.\ref{fig:string-graph-group-center-2} and Fig.\ref{fig:string-graph-group-center-11} for illustrating the functions of TKPDRA-center.

\vskip 0.4cm

\textbf{TKPDRAC-algorithm-II:} Users transferring files and various topology authentications in the network must pass TKPDRA-center (see Fig.\ref{fig:TKPDRAC-algorithm-2}).

\begin{figure}[h]
\centering
\includegraphics[width=16.4cm]{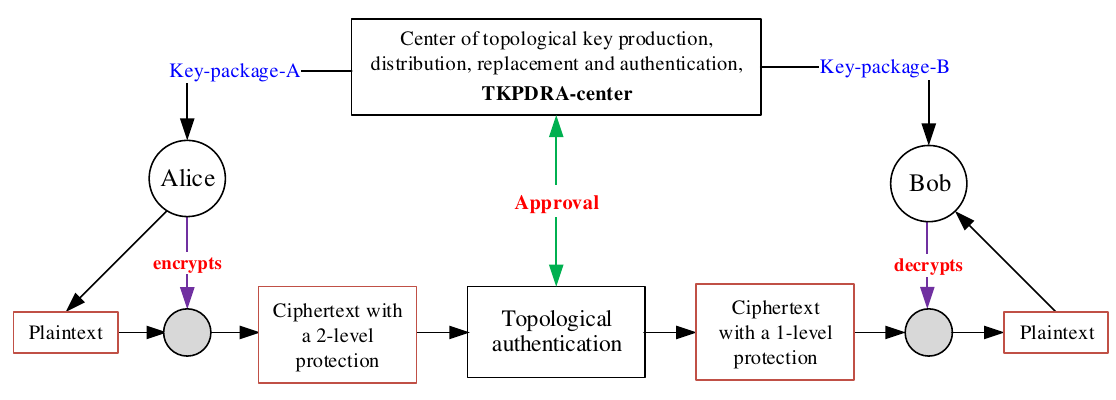}\\
\caption{\label{fig:TKPDRAC-algorithm-2}{\small A scheme for illustrating TKPDRAC-algorithm-II.}}
\end{figure}

\textbf{Step-II-1.} Alice use her private topological signature $G_{\textrm{Apri}}$ and her private number-based string $s_{\textrm{Apri}}$ to encode a plaintext $F$, and obtains a ciphertext $F_2$ having a 2-level protection. Alice sends this ciphertext $F_2$ to TKPDRA-center, ask TKPDRA-center for dealing with it and sends to Bob.

\textbf{Step-II-2.} TKPDRA-center makes a package containing the ciphertext $F_2$, Alice's public topological signature $G_{\textrm{Apub}}$ and Alice's public number-based string $s_{\textrm{Apub}}$, and uses Bob's public topological signature $G_{\textrm{Bpub}}$ and Bob's public number-based string $s_{\textrm{Bpub}}$ to encrypt the package, finally, produces a ciphertext $F_4$ with a 4-level protection, and sends to Bob.

After receiving the ciphertext $F_4$ from TKPDRA-center, Bob uses his private topological signature $G_{\textrm{Bpri}}$ (recognizing that this ciphertext was sent to himself) and his private number-based string $s_{\textrm{Bpri}}$ to decrypt the ciphertext $F_4$, such that he obtains the ciphertext $F_2$, Alice's public topological signature $G_{\textrm{Apub}}$ and Alice's public number-based string $s_{\textrm{Apub}}$. Notice that Bob has the authentications of Bob's topological signature and Bob's number-based string.

\textbf{Step-II-3.} Bob encrypts the ciphertext $F_4$ by his own private topological signature $G_{\textrm{Bpri}}$ for recognizing that this ciphertext was sent to himself, and uses his private number-based string $s_{\textrm{Bpri}}$ to decrypt continually the ciphertext. After the authentications from TKPDRA-center to Bob's topological signature and Bob's number-based string, then Bob obtains the ciphertext $F_2$, Alice's public topological signature $G_{\textrm{Apub}}$ and Alice's public number-based string $s_{\textrm{Apub}}$.

\textbf{Step-II-4.} Bob uses Alice's public topological signature $G_{\textrm{Apub}}$ to the ciphertext $F_2$, and know that the ciphertext $F_2$ was really sent by Alice, so Bob uses Alice's public number-based string $s_{\textrm{Apub}}$ to decrypt continually the ciphertext. After the authentications from TKPDRA-center to Alice's topological signature and Alice's number-based string, Bob can see the plaintext $F$.

\subsubsection{TKPDRA-center group-algorithms}

Due to the thousands of nodes in a network, it is necessary to solve the following problems: Key's topological structure spaces, key's quantity, key's kinds, key's matching, key's infiniteness. TKPDRA-center uses various topological groups to encrypt topologically overall networks.

By Definition \ref{defn:total-W-group-coloring-5}, Theorem \ref{thm:total-W-group-coloring-6} and Corollary \ref{cor:total-W-group-coloring-7}, as well as mixed-graphic groups and infinite mixed-graphic groups. we have the following TKPDRA-center group-algorithms:

\vskip 0.2cm

\begin{asparaenum}[\textrm{\textbf{Group-algo}}-I.]
\item The overall topological encryption algorithm for local area networks based on topological code matrix groups:

\qquad \textbf{\textrm{I-1}}. \textbf{Group-algo}-I based on user private topology signatures and the overall zero.

\qquad \textbf{\textrm{I-2}}. \textbf{Group-algo}-I based on user private topology signatures and the community private zero.

\qquad \textbf{\textrm{I-3}}. \textbf{Group-algo}-I based on user private topology signatures and the user private zero.

\vskip 0.2cm

\item The overall topological encryption algorithm based on the total-colored adjacent matrix group:

\qquad \textbf{\textrm{II-1}}. \textbf{Group-algo}-II based on user private topology signatures and the overall zero.

\qquad \textbf{\textrm{II-2}}. \textbf{Group-algo}-II based on user private topology signatures and the community private zero.

\qquad \textbf{\textrm{II-3}}. \textbf{Group-algo}-II based on user private topology signatures and the user private zero.

\item In the overall topological encryption of a network, different communities use different topological groups, and the communication between communities relies on the coloring transformation of the topological groups.
\end{asparaenum}

\begin{figure}[h]
\centering
\includegraphics[width=16.4cm]{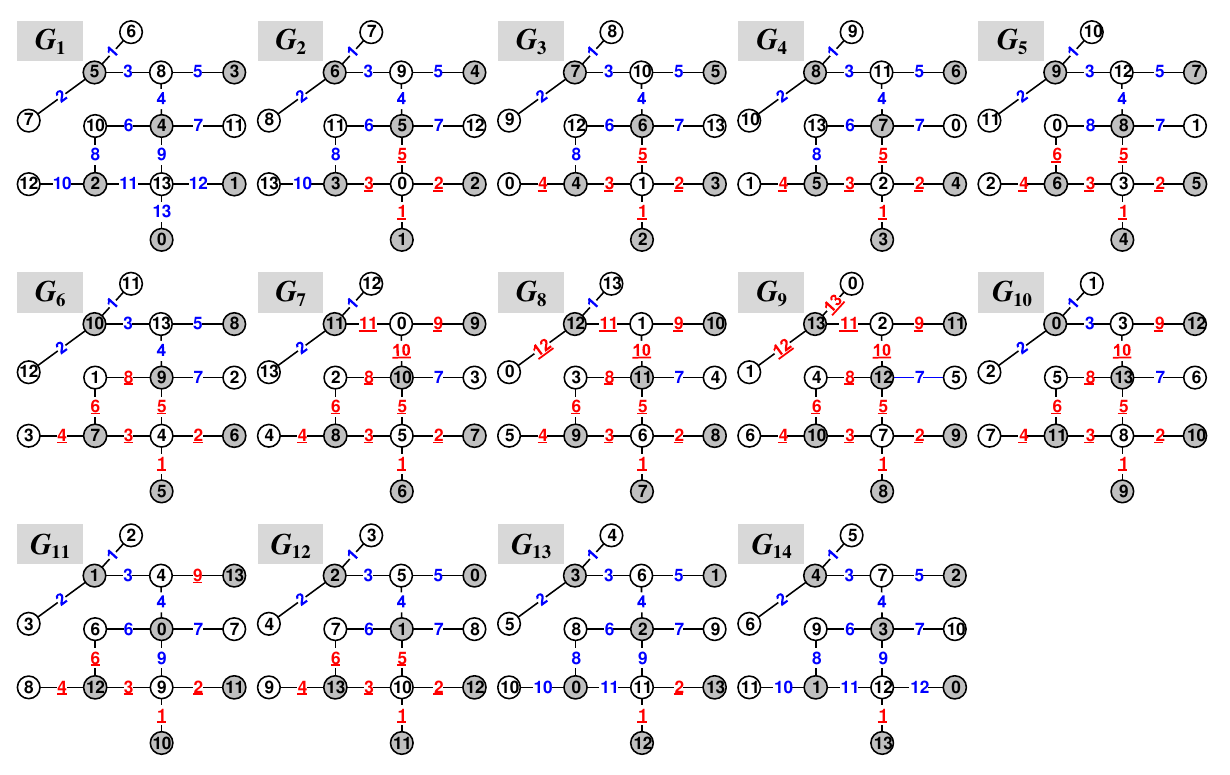}\\
\caption{\label{fig:Hanzi-nian-groups}{\small An every-zero Hanzi-graph group.}}
\end{figure}

\begin{example}\label{exa:8888888888}
Fig.\ref{fig:network-key-11} and Fig.\ref{fig:network-key-22} show us eight networks $B_1,B_2,\dots, B_8$, and they admit graphic group colorings, and $B_i\cong B_j$, as well as $V(B_i)=V(B_j)$.

In Fig.\ref{fig:network-key-11}, each $B_i$ of four networks $B_1,B_2,B_3,B_4$ is encrypted integrally by the every-zero Hanzi-graph group shown in Fig.\ref{fig:Hanzi-nian-groups}, and $B_i$ admits a total graphic-group coloring $F_i$ for $i\in [1,4]$ holding $F_1(V(B_1))=F_2(V(B_2))=F_3(V(B_3))=F_4(V(B_4))$.

In Fig.\ref{fig:network-key-22}, each $B_i$ of four networks $B_1,B_2,B_3,B_4$ is encrypted integrally by the every-zero Hanzi-graph group shown in Fig.\ref{fig:Hanzi-nian-groups}, and $B_i$ admits a total graphic-group coloring $F_i$ for $i\in [1,4]$ holding $F_i(V(B_i))\neq F_j(V(B_j))$ for $i\neq j$.\qqed
\end{example}

\begin{figure}[h]
\centering
\includegraphics[width=16.4cm]{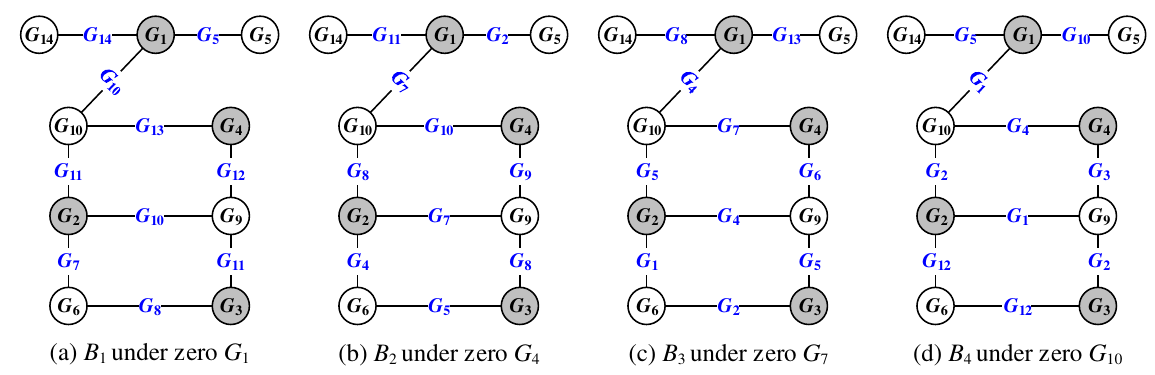}\\
\caption{\label{fig:network-key-11}{\small The first scheme for illustrating TKPDRA-center group-algorithms.}}
\end{figure}

\begin{figure}[h]
\centering
\includegraphics[width=16.4cm]{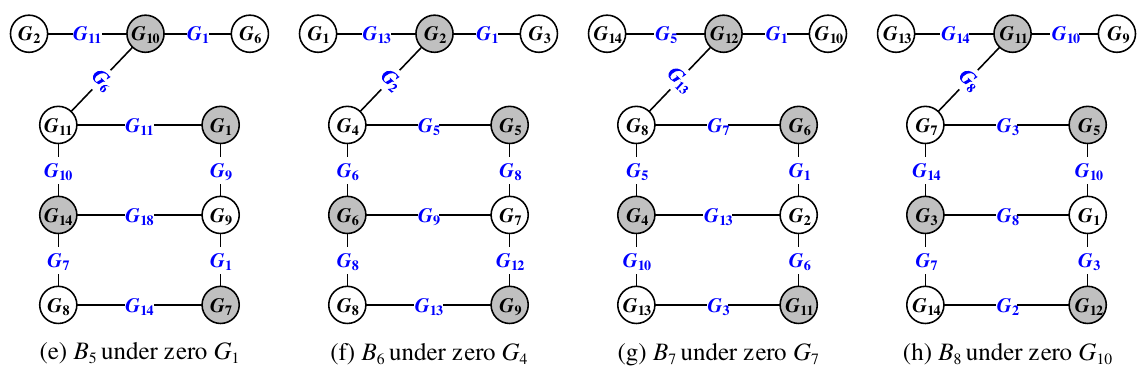}\\
\caption{\label{fig:network-key-22}{\small The second scheme for illustrating TKPDRA-center group-algorithms: Each one of four networks is encrypted integrally by the every-zero Hanzi-graph group shown in Fig.\ref{fig:Hanzi-nian-groups}.}}
\end{figure}

\section{Miscellaneous Topics}

\subsection{Operation-colorings of graphs}

\begin{defn} \label{defn:pan-vertex-coloring-thing-set}
$^*$ A graph $\xi$ admits a \emph{vertex operation-coloring} $F:V(\xi)\rightarrow S_{thing}$ if $F(x)=[\bullet^v_{W}]F(y)$ for each edge $xy \in E(\xi)$, where $[\bullet^v_{W}]$ is a $W$-constraint operation based on the thing set $S_{thing}$.\qqed
\end{defn}

\begin{example}\label{exa:pan-coloring-thing-set-examples}
For illustrating Definition \ref{defn:pan-vertex-coloring-thing-set}, we present the following examples:
\begin{asparaenum}[\textbf{\textrm{PanVc}}-1.]
\item \textbf{Vertex-splitting operation.} The thing set $S_{thing}$ is a graph set $S_{plit}(G)$ obtained by vertex-splitting a connected $(p,q)$-graph into graphs of $q$ edges, and the $W$-constraint operation $[\bullet^v_{W}]$ is the \emph{graph vertex-coinciding operation}. We get the graph $\xi$ with its own vertex set $V(\xi)=S_{plit}(G)$, such that two vertices $u$ and $v$ of the graph $\xi$ is adjacent if and only if the vertex color $F(u)=T_i\in S_{plit}(G)$ can be vertex-split into the vertex color $F(v)=T_j\in S_{plit}(G)$, and $\big | |V(T_i)|-|V(T_i)|\big |=1$.
\item \textbf{Intersection operation.} The thing set $S_{thing}$ is a hyperedge set $\mathcal{E}$ of the hypergraph set $\mathcal{E}(\Lambda^2)$ of a finite set $\Lambda$ (Ref. Remark \ref{rem:hypergraph-terminology-notations}), and the $W$-constraint operation $[\bullet^v_{W}]$ is the \emph{intersection operation}. The graph $\xi$ has its own vertex set $V(\xi)=\mathcal{E}$, such that two vertices $u$ and $v$ of the graph $\xi$ is adjacent if and only if two vertex colors hold $F(u)\cap F(v)=e_i\cap e_j\neq \emptyset$ (or $m\leq |F(u)\cap F(v)|$ with $m\geq 1$).
\item \textbf{Graph homomorphism operation.} The thing set $S_{thing}$ is a graph set $G_{raph}(T_{code})$, and the $W$-constraint operation $[\bullet^v_{W}]$ is the (colored) \emph{graph homomorphism}. Since a graphable Topcode-matrix $T_{code}$ corresponds many graphs, we collect these graphs into a graph set $G_{raph}(T_{code})$. Hence, there is a graph $\xi$ with its own vertex set $V(\xi)=G_{raph}(T_{code})$, such that two vertices $u$ and $v$ of the graph $\xi$ is adjacent if and only if two vertex colors $F(u)=G_i\in G_{raph}(T_{code})$ and $F(v)=G_j\in G_{raph}(T_{code})$ hold the (colored) graph homomorphism $F(u)=G_i\rightarrow G_j=F(v)$, and $\big | |V(G_i)|-|V(G_i)|\big |=1$.
\item \label{eqa:tres-trees}\textbf{Graph add-edge-remove operation.} The thing set $S_{thing}$ is the set $S_{pan}(K_n)$ of spanning trees of a complete graph $K_n$ admitting a labeling defined on $[1,n]$, so the cardinality $|S_{pan}(K_n)|=n^{n-2}$ by the famous Cayley's formula. And the $W$-constraint operation $[\bullet^v_{W}]$ is the \emph{add-edge-remove operation} ``$[\pm_e]$'' on $S_{pan}(K_n)$. The graph $\xi$ with its own vertex set $V(\xi)=S_{pan}(K_n)$, such that two vertices $u$ and $v$ of the graph $\xi$ is adjacent if and only if two vertex colors $F(u)=T_i$ and $F(v)=T_j$ hold the add-edge-remove operation ``$[\pm_e]$'' defined by $T_j=T_i+x_iy_i-u_iv_i$ for $x_iy_i\not\in E(T_i)$ and $u_iv_i\in E(T_i)$.
\item The thing set $S_{thing}=S_{plit}(G)$ is the set of trees of $q$ edges obtained by doing the vertex-splitting operation to a connected $(p,q)$-graph $G$; and the $W$-constraint operation $[\bullet^v_{W}]$ is the \emph{add-edge-remove operation} $[\pm_e]$ on $S_{plit}(G)$. The remainder is as the same as \textbf{\textrm{PanVc}}-\ref{eqa:tres-trees}.
\item The thing set $S_{thing}=S_{pan}(G)$ is the set of all spanning trees of a connected graph $G$; and the $W$-constraint operation $[\bullet^v_{W}]$ is the \emph{add-edge-remove operation} $[\pm_e]$ on $S_{pan}(G)$. The remainder is as the same as \textbf{\textrm{PanVc}}-\ref{eqa:tres-trees}.
\item The thing set $S_{thing}=S_{pan}(G)$ is the set of all spanning trees of a connected graph $G$; and the $W$-constraint operation $[\bullet^v_{W}]$ is the operation $[\bullet^{coin}_{colo}]$ between two spanning trees $T_j,T_k\in S_{pan}(G)$ defined in Definition \ref{defn:vertex-coinciding-two-spanning-trees}. The remainder is as the same as \textbf{\textrm{PanVc}}-\ref{eqa:tres-trees}.

\item The thing set $S_{thing}=S_{uniq}(C)$ is the set of unique cycle graphs of $p$ vertices; and the $W$-constraint operation $[\bullet^v_{W}]$ is the \emph{add-edge-remove operation} ``$[\pm_e]$'' on $S_{uniq}(C)$. The remainder is as the same as \textbf{\textrm{PanVc}}-\ref{eqa:tres-trees}.
\item \textbf{Finite module Abelian additive operation on graphs.} The thing set $S_{thing}=\{G(\mathcal{E})$; $[+][-]\}$ is an every-zero hgypergraph group; and the $W$-constraint operation $[\bullet^v_{W}]$ is the finite module Abelian additive operation
\begin{equation}\label{eqa:555555}
\mathcal{E}_i[+_k]\mathcal{E}_j:=\mathcal{E}_i[+]\mathcal{E}_j[-]\mathcal{E}_k
\end{equation} based on an every-zero hypergraph group $\big \{G(\mathcal{E});[+][-]\big \}$ defined in Definition \ref{defn:general-defi-hypergraph-groups}. The graph $\xi$ admits a \emph{total hypergraph-group coloring} $F:V(\xi)\cup E(\xi)\rightarrow \big \{G(\mathcal{E});[+][-]\big \}$, such that each edge $uv\in E(\xi)$ satisfies that hypergraph $F(u)=\mathcal{E}_i$, $F(v)=\mathcal{E}_j$ and holds the finite module Abelian additive operation
\begin{equation}\label{eqa:555555}
F(uv)=\mathcal{E}_\lambda=\mathcal{E}_i[+]\mathcal{E}_j[-]\mathcal{E}_k=F(u)[+]F(v)[-]\mathcal{E}_k
\end{equation} with $\lambda=i+j-k~(\bmod~N)$ for any preappointed \emph{zero} $\mathcal{E}_k\in \big \{G(\mathcal{E});[+][-]\big \}$.
\item \textbf{Finite module Abelian additive operation on matrices.} The thing set $S_{thing}=\{F(T_{code});[+][-]\}$, and the $W$-constraint operation $[\bullet^v_{W}]$ is the finite module Abelian additive operation. The graph $\xi$ admits a total matrix-group coloring $F:V(\xi)\cup E(\xi)\rightarrow \{F(T_{code})$; $[+][-]\}$, such that each edge $uv\in E(\xi)$ is colored with an induced edge color $F(uv)=T_{code}(G_\lambda,f_\lambda)$ defined by Eq.(\ref{eqa:topcode-matrix-abelian-additive-operation}) and Eq.(\ref{eqa:topcode-matrix-abelian-additive-operation11}) in Definition \ref{defn:topocode-matrices-groups}.\qqed
\end{asparaenum}
\end{example}

\begin{problem}\label{question:444444}
\textbf{Determine} the various graph parameters and topological structures for the graph $\xi$ appeared in Definition \ref{defn:pan-vertex-coloring-thing-set} and Example \ref{exa:pan-coloring-thing-set-examples}.
\end{problem}

\begin{defn} \label{defn:pan-total-coloring-thing-set}
$^*$ Let $S_{thing}$ be a thing set. A graph $\phi$ admits a \emph{total operation-coloring} $F:V(\phi)\cup E(\phi)\rightarrow S_{thing}$ if $F(uv)=F(u)[\bullet^T_{W}]F(v)$ for each edge $uv \in E(\phi)$, where $[\bullet^T_{W}]$ is a $W$-constraint operation on $S_{thing}$, and the graph $\phi$ is called \emph{topen-graph}.\qqed
\end{defn}

\begin{rem}\label{rem:333333}
About Definition \ref{defn:pan-total-coloring-thing-set}, there are many $W$-constraint operations $[\bullet^T_{W}]$ in graph theory and other mathematical regions.

The topological structure of a topen-graph is a tool of \emph{semi-structured data}, and a \emph{soft mathematical expression} in the field of graph theory. Notice that a graph is a natural representation of the encoding relationship structure, and calculations defined through graph structured data are widely used in various fields.

The thing set $S_{thing}$ in Definition \ref{defn:pan-vertex-coloring-thing-set} and Definition \ref{defn:pan-total-coloring-thing-set} may be a number set, or a coloring set, or a set-set, or a graph set, or a matrix set, or a hyperedge set, or a vector set, or a string set, or a function set, or a set of elements of a topological group, or a set of any things, \emph{etc}.\qqed
\end{rem}

\begin{example}\label{exa:8888888888}
For understanding Definition \ref{defn:pan-total-coloring-thing-set}, we have the following examples:
\begin{asparaenum}[\textbf{\textrm{PanTc}}-1.]
\item \textbf{Magic-constraint operation.} The thing set $S_{thing}=\textbf{\textrm{S}}_{et}(\leq n)$ is the set of integer sets of form $\{\alpha_{1},\alpha_{2},\dots ,\alpha_{m}\}$ with each element $\alpha_{j}\in Z^0$ for $j\in [1,m]$ and $m\leq n$. The $W$-constraint operation $[\bullet^T_{W}]$ is the homogeneous $(abc)$-magic operation. The graph $\phi$ admits a \emph{$\{W_i\}^A_{i=1}$-constraint total set-coloring} based on the set $\textbf{\textrm{S}}_{et}(\leq n)$ (Ref. Definition \ref{defn:n-di-set-colorings-definition}).
\item \textbf{Parameterized operation.} Definition \ref{defn:more-string-total-coloring} shows us the parameterized total string-colorings, parameterized total set-colorings and a parameterized total vector-colorings for the operation-colorings of graphs.

\qquad The graph $\phi$ admits each one of $W$-constraint $(k_s,d_s)$-colorings including gracefully $(k_s,d_s)$-total coloring, odd-gracefully $(k_s,d_s)$-total coloring, edge anti-magic $(k_s,d_s)$-total coloring, harmonious $(k_s,d_s)$-total coloring, odd-elegant $(k_s,d_s)$-total coloring, edge-magic $(k_s,d_s)$-total coloring, edge-difference $(k_s,d_s)$-total coloring, felicitous-difference $(k_s,d_s)$-total coloring, graceful-difference $(k_s,d_s)$-total coloring, odd-edge edge-magic $(k_s,d_s)$-total coloring, odd-edge edge-difference $(k_s,d_s)$-total coloring, odd-edge felicitous-difference $(k_s,d_s)$-total coloring, odd-edge graceful-difference $(k_s,d_s)$-total coloring in Definition \ref{defn:more-string-total-coloring}.

\qquad The $W$-constraint operation $[\bullet^T_{W}]$ is one of $F(uv)=|F(u)-F(v)|$, $F(uv)=F(u)+F(v)~(\bmod~M)$, $F(uv)+|F(u)-F(v)|=k$, $\big ||F(u)-F(v)|-F(uv)\big |=k$, $F(u)+F(uv)+F(v)=k$ and $\big |F(u)+F(v)-F(uv)\big |=k$ for each edge $uv \in E(\phi)$.

\item \textbf{Topological group operation.} Definition \ref{defn:total-W-group-coloring-5} shows us:

\qquad (i) The total-colored adjacent matrix group coloring $F_{\textrm{adj}}$;

\qquad (ii) the total-colored graph matrix group coloring $F_{gra}$;

\qquad (iii) the Topcode-matrix group coloring $F_{mat}$;

\qquad (iv) the topological string group coloring $F_{string}$;

\qquad (v) the total-colored parameterized matrix group coloring $F_{param}$;

\qquad (vi) the total-colored graph-set group coloring $F_{set}$; and

\qquad (vii) the thing group coloring $F_{\textrm{thing}}$.\qqed
\end{asparaenum}
\end{example}

\subsection{PCTSMGHS-string problem}

\begin{defn} \label{defn:generalization-pan-coloring-thing-set}
$^*$ \textbf{A generalization for Definition \ref{defn:pan-total-coloring-thing-set}.} Each graph $\phi_i$ with $i\in [1,n(S_{thing})]$ admits a \emph{total operation-coloring} $F_i:V(\phi_i)\cup E(\phi_i)\rightarrow S_{thing}$ based on the $W_i$-constraint operation defined on a thing set $S_{thing}$, such that each edge $uv \in E(\phi)$ is colored with $F(uv)=F(u)[\bullet^T_{W_i}]F(v)$, and moreover if
\begin{equation}\label{eqa:555555}
S_{thing}=\bigcup^{n(S_{thing})}_{i=1} F_i(V(\phi_i)\cup E(\phi_i))
\end{equation} then the coloring family of $F_1,F_2,\dots ,F_{n(S_{thing})}$ is \emph{full}.\qqed
\end{defn}

\begin{defn} \label{defn:compound-pan-coloring-thingsets}
$^*$ \textbf{Compound operation-coloring.} Since $S_{thing}=\bigcup^{n(S_{thing})}_{i=1} F_i(V(\phi_i)\cup E(\phi_i))$ in Definition \ref{defn:generalization-pan-coloring-thing-set}, let $e^{F}_i=F_i(V(\phi_i)\cup E(\phi_i))$ be a subset of the power set $S^2_{thing}$ with $i\in [1,n(S_{thing})]$, we get a hyperedge set $\{e^{F}_i:~i\in [1,n(S_{thing})]\}=\mathcal{E}^{F}\in \mathcal{E}\big (S_{thing}^2\big )$ (Ref. Remark \ref{rem:hypergraph-terminology-notations}). A $(p,q)$-graph $L$ admits a \emph{total compound-set-coloring} $\theta:V(L)\cup E(L)\rightarrow \mathcal{E}^{F}$ holding $\theta(xy)\supseteq \theta(x)\cap \theta(y)\neq \emptyset$ for each edge $xy \in E(L)$. And moreover, we say that the total set-coloring $\theta$ to be \emph{full} if the total color set $\theta(V(L)\cup E(L))=\mathcal{E}^{F}$.\qqed
\end{defn}

\begin{quote}
\textbf{PCTSMGHS-string problem.} Suppose that each thing $t_i\in S_{thing}=\{t_1,t_2,\dots, t_n\}$ in Definition \ref{defn:pan-total-coloring-thing-set} corresponds a number-based string $s_{nbs}(t_i)$ defined on $[0,9]$ with $i\in [1,n]$.

By Definition \ref{defn:generalization-pan-coloring-thing-set}, each subset $e^{F}_i=F_i(V(\phi_i)\cup E(\phi_i))=\{t_{i,1},t_{i,2},\dots ,t_{i,c_i}\}$ with $t_{i,j}$ for $j\in[1,c_i]$ $i\in [n(S_{thing})]$, each $t_{i,j}$ corresponds a number-based string $s_{nbs}(t_{i,j})$, so each subset $e^{F}_i\in \mathcal{E}^{F}$ corresponds a number-based string $s_{nbs}\big (e^{F}_i\big )$ which is a permutation of number-based strings $s_{nbs}(t_{i,1}), s_{nbs}(t_{i,2})$, $\dots $, $s_{nbs}(t_{i,c_i})$, there are $(c_i)!$ number-based strings like the number-based string $s_{nbs}\big (e^{F}_i\big )$.

In Definition \ref{defn:compound-pan-coloring-thingsets}, the $(p,q)$-graph $L$ has its own Topcode-matrix $T_{code}(L,\theta)_{3\times q}=(X,E$, $Y)^T$, where $X=(\theta(x_1)$, $\theta(x_2)$, $\dots $, $\theta(x_q))$, $E=(\theta(x_1y_1),\theta(x_2y_2),\dots ,\theta(x_qy_q))$, and $y=(\theta(Y_1)$, $\theta(Y_2)$, $\dots $, $\theta(Y_q))$. We have:

$\theta(x_i)=e^{F}_{\alpha_i}$ corresponds a number-based string $s_{nbs}\big (e^{F}_{\alpha_i}\big )$ for some $\alpha_i\in [1,n(S_{thing})]$

$\theta(x_j y_j)=e^{F}_{\beta_j}$ corresponds a number-based string $s_{nbs}\big (e^{F}_{\beta_j}\big )$ for some $\beta_j\in [1,n(S_{thing})]$,

$\theta(y_k)=e^{F}_{\gamma_k}$ corresponds a number-based string $s_{nbs}\big (e^{F}_{\gamma_k}\big )$ for some $\gamma_k\in [1,n(S_{thing})]$

Thereby, the Topcode-matrix $T_{code}(L,\theta)_{3\times q}=(X,E,Y)^T$ induces $(3q)!$ number-based strings, in which each number-based string $s=c_1c_2\cdots c_m$ with $c_i\in [0,9]$ is a permutation of number-based strings $s_{nbs}\big (e^{F}_{\alpha_1}\big )$, $s_{nbs}\big (e^{F}_{\alpha_2}\big )$, $\dots $, $s_{nbs}\big (e^{F}_{\alpha_q}\big )$, $s_{nbs}\big (e^{F}_{\beta_1}\big )$, $s_{nbs}\big (e^{F}_{\beta_2}\big )$, $\dots $, $s_{nbs}\big (e^{F}_{\beta_q}\big )$, $s_{nbs}\big (e^{F}_{\gamma_1})$, $s_{nbs}(e^{F}_{\gamma_2}\big )$, $\dots $, $s_{nbs}\big (e^{F}_{\gamma_q}\big )$. We say the number-based string $s$ to be \emph{PCTSMGHS-string}, since $s$ was generated from operation-colorings based on thing set, two or more graphs, hyperedge sets.
\end{quote}

If someone want to decode a PCTSMGHS-string $s=c_1c_2\cdots c_m$ with $c_i\in [0,9]$ introduced in the PCTSMGHS-string problem, then he will have to

\begin{asparaenum}[(1)]
\item \textbf{Find} out a $(p,q)$-graph $L$ (as a \emph{topological signature}), has its own authentication matching $L^*$ admitting a coloring $\theta$ (Ref. Definition \ref{defn:compound-pan-coloring-thingsets}), and use the Topcode-matrix $T_{code}(L,\theta)_{3\times q}$ to induce the number-based string $s$.
\item \textbf{Find} out $\theta(x_i)=e^{F}_{\alpha_i}$, $\theta(x_j y_j)=e^{F}_{\beta_j}$ and $\theta(y_k)=e^{F}_{\gamma_k}$ with $i,j,k\in [1,q]$ from the number-based string $s$ in order to determine the hyperedge set $\mathcal{E}^{F}$ for the total compound-set-coloring $\theta:V(L)\cup E(L)\rightarrow \mathcal{E}^{F}$.
\item \textbf{Find} out each graph $\phi_i$ (as a \emph{topological signature}, has its own authentication matching $\phi_i^*$) with $i\in [1,n(S_{thing})]$ admitting a total operation-coloring $F_i:V(\phi_i)\cup E(\phi_i)\rightarrow S_{thing}$ based on the $W_i$-constraint operation (Ref. Definition \ref{defn:generalization-pan-coloring-thing-set})
\item \textbf{Find} out the thing set $S_{thing}$ (as a \emph{public-key}), and use it to make an authentication with a pregiven thing set $S^*_{thing}$ (as a \emph{private-key}).
\item \textbf{Finding} the graphs $L$ and $\phi_i$ with $i\in [1,n(S_{thing})]$ is NP-complete, since the subgraph isomorphic problem is NP-complete.
\item \textbf{Seeking} the total compound-set-coloring $\theta$ and total operation-colorings $F_i$ with $i\in [1$, $n(S_{thing})]$ also is not only NP-hard, bur also \#P-hard.
\item \textbf{Determining} the elements of the thing set $S_{thing}$ is simply impossible, since each element of the thing set $S_{thing}$ has been replaced by a number-based string.
\end{asparaenum}
\vskip 0.2cm

\textbf{We claim}: PCTSMGHS-strings cannot be deciphered, even with quantum computation when as the graphs $L$ and $\phi_i$ with $i\in [1,n(S_{thing})]$ have huge numbers of vertices and edges.

\subsection{Assembling graphs with hypergraphs}

Hypergraphs are in many where of graph theory, an example is as follows:

\begin{defn} \label{defn:equitable-total-coloring-def}
An $i$-\emph{color ve-set} $S_i=V_i\cup E_i$ consists of $V_i=\{u: f(u)=i,u\in V(G)\}$ and $E_i=\{xy: f(xy)=i,xy\in E(G)\}$ for a graph $G$ admitting a proper total coloring $f:V(G)\cup E(G)\rightarrow [1,k]$. If $\big ||S_i|-|S_j|\big |\leq 1$ for any pair of color ve-sets $S_i$ and $S_j$ when $i\neq j$, we say $f$ to be $k$-\emph{equitable total coloring} of the graph $G$.
Straightly, the number $$\chi\,''_{e}(G)=\min\{k: \text{ over all
$k$-equitable total colorings of }G \}$$ is called \emph{equitable total chromatic number} of the graph $G$.\qqed
\end{defn}

\begin{conj} \label{conj:c4-Weifan-Wang-2002}
(Weifan Wang, 2002) For every graph $G$, then $\chi\,''_e(G)\leq \Delta(G)+2$.
\end{conj}

\begin{rem}\label{rem:333333}
Since the $i$-color ve-set $S_i=V_i\cup E_i$ with $i\in [1,k]$, we get a hypergraph $\mathcal{H}_{yper}=(\Lambda,\mathcal{E})$ with its own vertex set $\Lambda(G)=V(G)\cup E(G)$ and the hyperedge set $\mathcal{E}=\bigcup^k_{i=1}S_i$ holding $\Lambda(G)=\bigcup_{S_i\in \mathcal{E}}S_i$.

However, assembling the original graph $G$ admitting a proper total coloring $f$ by the hyperedge set $\mathcal{E}=\bigcup^k_{i=1}S_i$ is not easy, and produces other graphs differing from the original graph $G$.\qqed
\end{rem}

\begin{problem}\label{question:assembl-graphs-hypergraphs}
A total coloring is a partitioning of the vertices and edges of the graph into \emph{total independent sets}. For a graph $G$, we

(i) let $V_i=\{u_{i,1},u_{i,2},\dots ,u_{i,a_i}\}$ be an independent vertex set of $V(G)$, and each vertex $u_{i,j}\in V_i$ with $j\in [1,a_i]$ is colored with color $i$ (Ref. Eq.(\ref{eqa:edges-not-colored})).

(ii) let $E_i=\{e_{i,1},e_{i,2},\dots ,e_{i,b_i}\}$ be an independent edge set of $E(G)$, and each edge $e_{i,j}\in E_i$ with $j\in [1,b_i]$ is colored with $i$th color, but two ends of the edge $e_{i,j}$ are not colored (Ref. Eq.(\ref{eqa:vertices-not-colored})).

For the pregiven $i$-color ve-set $S_i=V_i\cup E_i$ with $i\in [1,k]$ defined above, we get a hypergraph $\mathcal{H}_{yper}=(\Lambda(G),\mathcal{E})$ with its own vertex set $\Lambda(G)=V(G)\cup E(G)$ and its own hyperedge set $\mathcal{E}=\bigcup^k_{i=1}S_i$ holding $\Lambda(G)=\bigcup_{S_i\in \mathcal{E}}S_i$. \textbf{Dose} $G$ admit a proper total coloring $f:V(G)\cup E(G)\rightarrow [1,k]$, such that $\{w:~w\in V(G)\cup E(G)\textrm{ and }f(w)=i\}=S_i$ for each $i\in [1,k]$?

Without doubt, there are some hyperedge sets $\mathcal{E}^a\in \mathcal{E}\big (\Lambda^2(G)\big )$ with $a\in [1,m]$, such that each hypergraph $\mathcal{H}^a_{yper}=(\Lambda(G),\mathcal{E}^a)$ produces a proper total coloring of the graph $G$, where each hyperedge set $\mathcal{E}^a=\bigcup^k_{i=1} S^a_i=\bigcup^k_{i=1} (V^a_i\cup E^a_i)$ with $a\in [1,m]$, and moreover subsets $V^a_i\subset V(G)$ and $E^a_i\subset E(G)$ are independent sets of the graph $G$.
\end{problem}

\begin{example}\label{exa:8888888888}
(i) A colored graph $H$ can be vertex-split into edge-disjoint colored graphs $H_1,H_2$, $\dots$, $H_m$. We can use these colored graphs $H_1,H_2,\dots,H_m$ to encrypt digital files, and use the colored graph $H$ to decrypt the encrypted files, in other words, we can assemble the colored graphs $H_1,H_2,\dots,H_m$ into the original colored graph $H$. In general, it is quite difficult to assemble the colored graphs $H_1,H_2,\dots,H_m$ into the original colored graph $H$, since no polynomial algorithm is for assembling action.

(ii) By a Topcode-matrix $T_{code}(G,f)$ defined in Definition \ref{defn:total-coloring-Topcode-matrixs}, we can make the following two incomplete Topcode-matrices

\begin{equation}\label{eqa:edges-not-colored}
{
\begin{split}
T_{code}(G,f_v)= \left(
\begin{array}{cccccccccc}
f(x_1) & f(x_2) & \cdots & f(x_q)\\
ue_1 & ue_2 & \cdots & ue_q\\
f(y_1) & f(y_2) & \cdots & f(y_q)
\end{array}
\right)_{3\times q}=(f(X),un(E),f(Y))^T_{3\times q}
\end{split}}
\end{equation}
and
\begin{equation}\label{eqa:vertices-not-colored}
{
\begin{split}
T_{code}(G,f_e)= \left(
\begin{array}{cccccccccc}
uw_1 & uw_2 & \cdots & uw_q\\
f(e_1) & f(e_2) & \cdots & f(e_q)\\
uv_1 & uv_2 & \cdots & uv_q
\end{array}
\right)_{3\times q}=(un(X),f(E),un(Y))^T_{3\times q}
\end{split}}
\end{equation} where $un(X),un(Y)$ and $un(E)$ are unknown vectors.

Since two incomplete Topcode-matrices $T_{code}(G,f_e)$ and $T_{code}(G,f_v)$ can be assembled into $T_{code}(G,f)$, we can consider $T_{code}(G,f_e)$ as a \emph{private-key}, and $T_{code}(G,f_v)$ as a \emph{public-key}, then the Topcode-matrix $T_{code}(G,f)$ is just the topological authentication of the private-key $T_{code}(G,f_e)$ and the public-key $T_{code}(G,f_v)$.\qqed
\end{example}

\subsection{Applying topological lattices}

We, in the previous sections, have introduced the following topological lattices:

\begin{asparaenum}[\textbf{\textrm{Latt}}-1.]
\item The tree-graph lattice $\textbf{\textrm{L}}(Z^0[\bullet]T_{splt}(G))$ defined in Eq.(\ref{eqa:tree-graph-lattice5});
\item the edge-coincided vertex-intersected graph lattice $\textbf{\textrm{L}}\big (Z^0[\ominus]\textbf{\textrm{H}}\big )$ defined in Eq.(\ref{eqa:hypergraph-vertex-intersected graph-lattice-edge});
\item the hyperedge-coincided hypergraph lattice $\textbf{\textrm{L}}\big (Z^0[\ominus]\textbf{\textrm{H}}_{yper}\big )$ defined in Eq.(\ref{eqa:hypergraph-lattice-by-hyperedge-coinciding});
\item the vertex-coincided vertex-intersected graph lattice $\textbf{\textrm{L}}\big (Z^0[\bullet]\textbf{\textrm{T}}\big )$ defined in Eq.(\ref{eqa:hypergraph-vertex-intersected graph-lattice-vertex});
\item the mixed vertex-intersected graph lattice $\textbf{\textrm{L}}\big (Z^0[\bullet\ominus]\textbf{\textrm{T}}\big )$ defined in Eq.(\ref{eqa:hypergraph-vertex-intersected graph-lattice-mixed});
\item the $K_{tree}$-spanning lattice $\textbf{\textrm{L}}(Z^0[\bullet]\textbf{\textrm{T}}^c)$ defined in Eq.(\ref{eqa:K-tree-spanning-lattice});
\item the vertex-coincided graphic group lattice $\textbf{\textrm{L}}\big (Z^0[\bullet]F_f(G)\big )$ defined in Eq.(\ref{eqa:graphic-group-lattices}) of Definition \ref{defn:graphic-group-coincide-general};
\item the $4$-color star-graphic lattice $\textbf{\textrm{L}}(\Delta\overline{\ominus} \textbf{\textrm{I}}_{ce}(PS,M))$ defined in Eq.(\ref{eqa:4-color-star-system-lattices}).
\item the operation vertex-intersected network lattice $\textbf{\textrm{L}}\big (\textbf{\textrm{F}}(t)\lhd Z^0\textbf{\textrm{O}}_i\big )$ defined in Eq.(\ref{eqa:operation-base-graphic-latticess});
\item the scale-free operation vertex-intersected network meta-lattice $\textbf{\textrm{L}}\big (\textbf{\textrm{F}}_{scale}(t)\lhd Z^0\textbf{\textrm{O}}_i\big )$ defined in Eq.(\ref{eqa:scale-free-network-lattice11});
\item the hypernetwork meta-lattices $\textbf{\textrm{L}}^{sf}\big (\Lambda^{p(t)}_{\cup}(t),\mathcal{E}_{\cup}(t)\mid \textbf{\textrm{B}}(H(t),p(t))\big )$ defined in Eq.(\ref{eqa:scale-free-meta-lattices-11}), and $\overline{\textbf{\textrm{L}}}^{sf}\big (\Lambda^{p(t)}_{\cup}(t),\overline{\mathcal{E}}_{\cup}(t)\mid \textbf{\textrm{B}}(H(t),p(t))\big )$ defined in Eq.(\ref{eqa:scale-free-meta-lattices-22});
\end{asparaenum}

We try to apply topological lattices in some problems of graph theory in this subsection.

\subsubsection{Hamiltonian graph lattices}

\textbf{Cycle-joining operation.} In \cite{Yao-Wang-Ma-Wang-Degree-sequences-2021}, the \emph{cycle-joining operation} $[\overline{\circ \circ}]$ on hamiltonian graphs is defined as: Let $xy$ be an edge of a Hamilton cycle of a Hamilton graph $G$, $uv$ be an edge of a Hamilton cycle of a Hamilton graph $H$, and $12$ be an edge of a Hamilton cycle of a Hamilton graph $L$ show in Fig.\ref{fig:cycle-hamilton}. Remove the edge $xy$ from $G$ and remove the edge $uv$ from $H$, and join the vertex $x$ with the vertex $u$ by a new edge $xu$, next join the vertex $y$ with the vertex $v$ by a new edge $yv$, the resultant graph is denoted as $G[\overline{\circ \circ}]H$, so there is a Hamilton cycle in the \emph{Hamilton cycle-joining graph} $G[\overline{\circ \circ}]H$ (see Fig.\ref{fig:cycle-hamilton}). Moreover, we can do the cycle-joining operation on three Hamilton graphs $L$, $F$ and $G[\overline{\circ \circ}]H$, and get a Hamilton cycle-joining graph $\{[G[\overline{\circ \circ}]H][\overline{\circ \circ}]L\}[\overline{\circ \circ}]F$.

\begin{figure}[h]
\centering
\includegraphics[width=16cm]{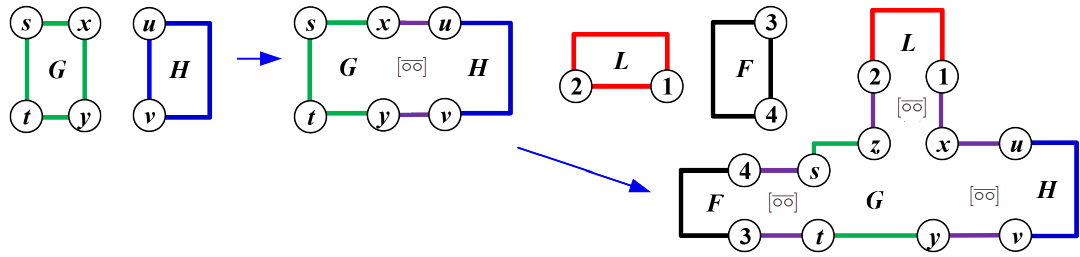}\\
\caption{\label{fig:cycle-hamilton}{\small A scheme for illustrating the cycle-joining operation.}}
\end{figure}

Suppose that $G_k$ is a hamiltonian graph satisfied a hamiltonian condition $k$ with $k\in [1, m]$, the \emph{hamiltonian graph base} $\textbf{\textrm{H}}=(G_k)^m_{k=1}$ is linear independent from each other, that is, each hamiltonian graph $G_k$ is not the result of some $G_{i_1},G_{i_2},\dots ,G_{i_j}$ under the cycle-joining operation ``$[\overline{\circ \circ}]$''. Each graph of the following set
\begin{equation}\label{eqa:555555}
\textbf{\textrm{L}}(Z^0[\overline{\circ \circ}]\textbf{\textrm{H}})=\big \{[\overline{\circ \circ}]^m_{k=1}a_kG_k:a_k\in Z^0,G_k\in \textbf{\textrm{H}}\big \}
\end{equation} is a hamiltonian graph, so we call $\textbf{\textrm{L}}(Z^0[\overline{\circ \circ}]\textbf{\textrm{H}})$ a \emph{hamiltonian graph lattice}. Obviously, each graph $H\in \textbf{\textrm{L}}(Z^0[\overline{\circ \circ}]\textbf{\textrm{H}})$ does not obey any one of these $m$ hamiltonian conditions when $A=\sum ^m_{k=1}a_k\geq 2$. Thereby, we have at least $B=2^m-1=\sum ^m_{k=1}{A \choose k}$ different hamiltonian graphs $H_i\in \textbf{\textrm{L}}(Z^0[\overline{\circ \circ}]\textbf{\textrm{H}})$ with $i\in [1,B]$, such that $(H_k)^{B}_{k=1}$ is a new hamiltonian graph base.

The above argue means that no necessary and sufficient condition exists for judging whether graphs are hamiltonian.

\begin{rem}\label{rem:ABC-conjecture}
Let $C$ be a Hamilton cycle of a graph $T$ in a hamiltonian graph lattice $\textbf{\textrm{L}}(Z^0[\overline{\circ \circ}]\textbf{\textrm{H}})$. We add some edges $x_iy_i$ to join some pairs of non-adjacent vertices $x_i$ and $y_i$ of the Hamilton cycle $C$ for $i\in [1,s]$, the resultant graph $T+E_{s}$ has its own Hamilton cycle $C$, so that $G_k\subseteq T+E_{s}$ if $a_k\neq 0$.\qqed
\end{rem}

\begin{thm}\label{thm:hamiltonian-iff}
\cite{Yao-Wang-2106-15254v1} A connected $(p,q)$-graph $G$ with $p\geq 4$ and $q\geq p+1$ is hamiltonian if and only if there exists an edge $x_ix_j$ of the connected graph $G$, such that doing an edge-splitting operation to the edge $x_ix_j$ produces a graph $G\wedge x_ix_j$ having two vertex-disjoint Hamilton graphs $G_1,G_2$, and an edge set $E(V(G_1),V(G_2))$ between $G_1$ and $G_2$.
\end{thm}

\begin{rem}\label{rem:hamiltonian-iff}
\cite{Yao-Wang-2106-15254v1} The result in Theorem \ref{thm:hamiltonian-iff} is constructional, since the graph $G$ is not a cycle. Some one of $G_1$ and $G_2$ maybe a cycle, so we can use Theorem \ref{thm:hamiltonian-iff} to another one, until, each Hamilton graph is a cycle only, we stop to apply Theorem \ref{thm:hamiltonian-iff} to graphs.\qqed
\end{rem}

\subsubsection{Edge-hamiltonian lattices}

\begin{problem}\label{qeu:characterize-edge-hamiltonian}
Suppose that a graph $G$ contains a particular subgraph $L$ with $|E(L)|\geq 1$, we say $G$ to be \emph{edge-graph-$L$}, if each edge $uv\in E(G)$ is in some subgraph $L$ of the graph $G$. Obviously, any connected graph is \emph{edge-spanning-tree}. If each edge $uv\in E(G)$ is in a Hamilton cycle of a connected graph $G$, we say $G$ to be \emph{edge-hamiltonian}. \textbf{Characterize} edge-hamiltonian graphs.
\end{problem}

\begin{example}\label{exa:8888888888}
In Fig.\ref{fig:edge-hamiltonian}, the complete graph $K_4$ is edge-hamiltonian defined in Problem \ref{qeu:characterize-edge-hamiltonian}, so are $G_1$ and $G_2$. However, the graph $G_3$ shown in Fig.\ref{fig:edge-hamiltonian}(d) is not edge-hamiltonian, since two edges $uw$ and $xy$ of $G_3$ are not in any Hamilton cycle of $G_3$.\qqed
\end{example}

\begin{figure}[h]
\centering
\includegraphics[width=15cm]{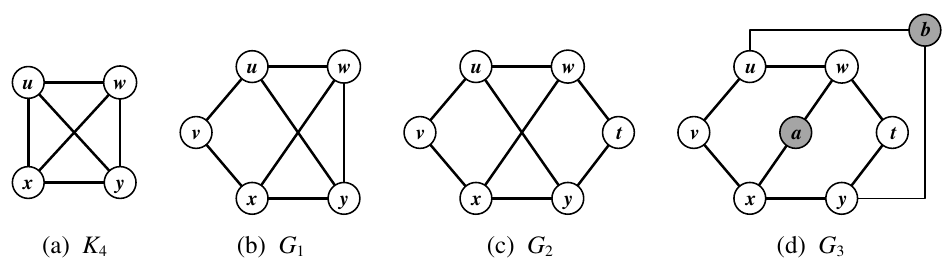}\\
\caption{\label{fig:edge-hamiltonian}{\small (a), (b) and (c) are edge-hamiltonian, (d) the graph $G_3$ is not edge-hamiltonian since two edges $uw$ and $xy$ of the graph $G_3$ are not in any hamiltonian cycle of the graph $G_3$.}}
\end{figure}

\begin{example}\label{exa:8888888888}
There are three Hamilton cycles $C^*_1=uwyxu$, $C^*_2=uywxu$ and $C^*_3=uwxyu$ in the complete graph $K_4$ shown in Fig.\ref{fig:edge-hamiltonian}(a), and we let $e_i=E(C^*_i)$ with $i\in [1,3]$, thus, we get the hypergraph $\mathcal{H}^{K_4}_{yper}=(\Lambda,\mathcal{E}^*)$, where $\mathcal{E}^*=\{e_1,e_2,e_3\}$. Clearly, a vertex-intersected graph $G(K_4)$ of the hypergraph $\mathcal{H}^{K_4}_{yper}$ is a complete graph $K_3$, since $e_1\cap e_2=\{ux,wy\}$, $e_1\cap e_3=\{xy,uw\}$ and $e_2\cap e_3=\{uy,wx\}$.\qqed
\end{example}

\begin{conj}\label{conj:0000000000}
$^*$ Let $H$ be an edge-hamiltonian graph, and let $E_{dha}(H)$ be the set of Hamilton cycles of the graph $H$. For each hyperedge set $\mathcal{E}\in \mathcal{E}(E^2_{dha}(H))$ with $\bigcup _{e\in \mathcal{E}}e=E_{dha}(H)$, we have a hypergraph $\mathcal{H}_{yper}=(E_{dha}(H),\mathcal{E})$. We say: The number of Hamilton cycles of the set $E_{dha}(H)$ holds $|E_{dha}(H)|\leq \big \lceil \frac{\Delta(H)}{2}\big \rceil$.
\end{conj}

For considering Problem \ref{qeu:characterize-edge-hamiltonian}, we show the concept of \emph{edge-hamiltonian graphic lattice}.

\begin{defn} \label{defn:definition-edge-hamiltonian-lattices}
$^*$ Suppose that $\textbf{\textrm{H}}=(H_1,H_2,\dots ,H_m)$ with $H_i\not \subset H_j$ and $H_i\not \cong H_j$ if $i\neq j$ is an \emph{edge-hamiltonian graph base}, where each graph $H_i$ with $i\in [1,m]$ is an edge-hamiltonian graph, and there is an \emph{operation base} $\textbf{\textrm{O}}=(O_1,O_2,\dots ,O_n)$ with $O_i\neq O_j$ if $i\neq j$ such that doing each operation $O_k$ to two edge-hamiltonian graphs $H_i$ and $H_j$ produces a new edge-hamiltonian graph, denoted as $O_k\langle H_i,H_j \rangle$.
We take a permutation $T_1,T_2,\dots ,T_A$ of edge-hamiltonian graphs $a_1H_1,a_2H_2,\dots ,a_mH_m$, where $A=\sum ^m_{i=1}a_i\geq 1$ for $a_i\in Z^0$ and $H_i\in \textbf{\textrm{H}}$, and do an operation $O_{i_1}$ selected randomly from $\textbf{\textrm{O}}$ to $T_1,T_2$, the resultant graph is a new edge-hamiltonian graph $L_1=O_{i_1}\langle T_1,T_2\rangle$ with selecting randomly $O_{i_1}\in \textbf{\textrm{O}}$. Again we get a new edge-hamiltonian graph $L_2=O_{i_2}\langle L_1,T_3\rangle$ with selecting randomly $O_{i_2}\in \textbf{\textrm{O}}$, in general, we have edge-hamiltonian graphs $L_k=O_{i_k}\langle L_{k-1},T_{k+1}\rangle$ with $O_{i_k}\in \textbf{\textrm{O}}$ and $k\in [1,A-1]$. For simplicity, we write the last edge-hamiltonian graph $L_{A-1}$ as follows
\begin{equation}\label{eqa:555555}
L_{A-1}=O_{i_{A-1}}\langle L_{A-2},T_{A}\rangle =\big [\textbf{\textrm{O}}\big ]^A_{j=1}T_j=\big [\textbf{\textrm{O}}\big ]^M_{i=1}a_iH_i,
\end{equation} so we get an \emph{edge-hamiltonian graphic lattice}
\begin{equation}\label{eqa:edge-hamiltonian-graphic-lattice}
\textbf{\textrm{L}}(Z^0[\textbf{\textrm{O}}]\textbf{\textrm{H}})=\left \{\big [\textbf{\textrm{O}}\big ]^M_{i=1}a_iH_i:a_i\in Z^0,H_i\in \textbf{\textrm{H}}\right \}
\end{equation} with $\sum ^m_{i=1}a_i\geq 1$. \qqed
\end{defn}

\begin{example}\label{exa:operation-base-not-empty}
Let $\textbf{\textrm{H}}=(H_1,H_2,\dots ,H_m)$ be an edge-hamiltonian graph base. Since each Hamilton cycle of an edge-hamiltonian graph $H_i$ runs over each vertex $u_i\in V(H_i)$, we vertex-split the vertex $u_i$ into $a_i$ and $b_i$ such that $u_i=a_i\bullet b_i$ (see Fig.\ref{fig:edge-hamiltonian-graphs}(a)). Similarly, we vertex-split a vertex $v_j$ of another edge-hamiltonian graph $H_j$ with $i\neq j$ into $s_j$ and $t_j$ such that $v_j=s_j\bullet t_j$ (see Fig.\ref{fig:edge-hamiltonian-graphs}(b)). We show the following operations for constructing edge-hamiltonian graphs as follows:

$O_1$: We join vertex $a_i$ with vertex $s_j$ together by an edge $a_is_j$, and join $b_i$ with $t_j$ together by an edge $b_it_j$, the resultant graph is denoted as $\ominus_2\langle H_i\wedge u_i, H_j\wedge v_j\rangle$, which is an edge-hamiltonian graph too (see Fig.\ref{fig:edge-hamiltonian-graphs}(c)).

$O_2$: We vertex-coincide two vertices $a_i$ and $s_j$ into one vertex $w=a_i\bullet s_j$, and vertex-coincide $b_i$ and $t_j$ into one vertex $z=b_i\bullet t_j$, the resultant graph is denoted as $[\bullet_2 ]\langle H_i\wedge u_i,H_j\wedge v_j\rangle $, which is clearly an edge-hamiltonian graph (see Fig.\ref{fig:edge-hamiltonian-graphs}(d)).

$O_3$: We vertex-coincide two vertices $a_i$ and $s_j$ into one vertex $w=a_i\bullet s_j$, and join $b_i$ with $t_j$ together by an edge $b_it_j$, the resultant graph is denoted as $[\bullet\ominus]\langle H_i\wedge u_i, H_j\wedge v_j\rangle $, which is just an edge-hamiltonian graph (see Fig.\ref{fig:edge-hamiltonian-graphs}(e)).\qqed
\end{example}

\begin{figure}[h]
\centering
\includegraphics[width=13cm]{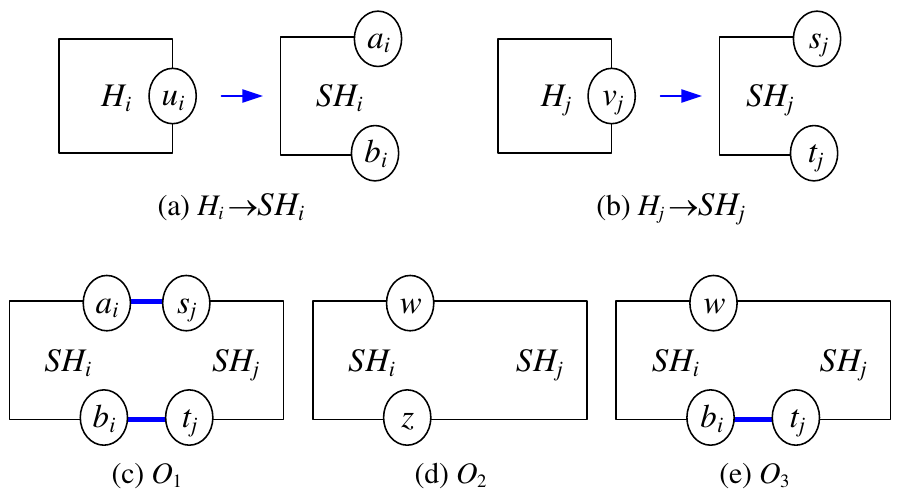}\\
\caption{\label{fig:edge-hamiltonian-graphs}{\small A scheme for illustrating three operations $O_1,O_2,O_3$ shown in Example \ref{exa:operation-base-not-empty}.}}
\end{figure}

\begin{example}\label{exa:8888888888}
There are many ways for constructing edge-hamiltonian graphs: We vertex-split a vertex $u$ of the graph $G$ into two vertices $u\,'$ and $u\,''$, and then

(i) add a new edge $u\,'u\,''$ to the vertex-split graph $G\wedge u$, the resultant graph is denoted as $H=G\wedge u+u\,'u\,''$;

(ii) add a path $P(u\,', u\,'')$ to the vertex-split graph $G\wedge u$ to join two vertices $u\,'$ and $u\,''$, the resultant graph is denoted as $H\,'=G\wedge u+P(u\,', u\,'')$; and

(iii) add a complete graph $K_m$ to vertex-coincide a vertex $x\in V(K_m)$ with the vertex $u\,'$ into one, and to vertex-coincide a vertex $y\in V(K_m)$ with the vertex $u\,''$ into one, the resultant graph is denoted as $H^*=G\wedge u+u\,'[\bullet]K_m+u\,''[\bullet]K_m$.\qqed
\end{example}

We have the following results:

\begin{prop}\label{qeu:99999}
\cite{Yao-Ma-arXiv-2201-13354v1} A connected graph $G$ is edge-hamiltonian if and only if each of $H=G\wedge u+u\,'u\,''$, $H\,'=G\wedge u+P(u\,', u\,'')$ and $H^*=G\wedge u+u\,'[\bullet]K_m+u\,''[\bullet]K_m$ is edge-hamiltonian.
\end{prop}

\begin{rem}\label{rem:333333}
Notice that the operation base $\textbf{\textrm{O}}$ appeared in Definition \ref{defn:definition-edge-hamiltonian-lattices} holds $|\textbf{\textrm{O}}|=n\geq 1$ according to Example \ref{exa:operation-base-not-empty}. It is not hard to see that each graph $G\in \textbf{\textrm{L}}(Z^0[\textbf{\textrm{O}}]\textbf{\textrm{H}})$ is just an edge-hamiltonian graph, so the edge-hamiltonian graphic lattice $\textbf{\textrm{L}}(Z^0[\textbf{\textrm{O}}]\textbf{\textrm{H}})$ is a generator of edge-hamiltonian graphs. In other word, each graph $G\in \textbf{\textrm{L}}(Z^0[\textbf{\textrm{O}}]\textbf{\textrm{H}})$ with $a_i\geq 1$ for $i\in [1,M]$ does not obey any judging condition of any one of $H_1,H_2,\dots ,H_m$ to be edge-hamiltonian graph, which means that there is no necessary and sufficient condition for judging edge-hamiltonian graphs, also, for judging a hypergraph having \emph{hyperedge Hamilton cycle}.\qqed
\end{rem}

\begin{problem}\label{qeu:444444}
By Problem \ref{qeu:characterize-edge-hamiltonian} and the edge-hamiltonian graphic lattice defined in Definition \ref{defn:definition-edge-hamiltonian-lattices}, so there is no necessary and sufficient condition for judging edge-hamiltonian graphs, since each graph may not satisfy any one of the judging conditions of $G=\big [\textbf{\textrm{O}}\big ]^M_{i=1}a_iH_i$ defined in Eq.(\ref{eqa:edge-hamiltonian-graphic-lattice}) to be edge-hamiltonian graphs, and moreover new operations for constructing edge-hamiltonian graphs and new conditions of judging edge-hamiltonian graphs will probably arise everyday.

In general, if some sparse graphs $G^1_{sparse},G^2_{sparse},\dots ,G^A_{sparse}$ and dense graphs $G^1_{dense}$, $G^2_{dense}$, $\dots $, $G^B_{dense}$ of $p$ vertices hold a \emph{topological structure property} $P$ consisted of topological structures and mathematical constraints, and each graph $H$ defined as
\begin{equation}\label{eqa:no-iff-conditions}
H=\big [\textbf{\textrm{O}}\big ]^A_{i=1}a_iG^i_{sparse}+\big [\textbf{\textrm{O}}\big ]^B_{j=1}b_iG^j_{dense},\quad \sum ^A_{i=1}a_i\geq 1, ~\sum ^B_{j=1}b_j\geq 1
\end{equation} has the topological structure property $P$, so we conjecture: If this graphic property $P$ has been proved to be sharp-P-hard, then we conjecture: There is \textbf{\emph{no necessary and sufficient condition}} for determining graphs shown in Eq.(\ref{eqa:no-iff-conditions}) with the graphic property $P$, where $G^1_{sparse}$ is a cycle, and $G^B_{dense}$ is a complete graph.
\end{problem}

\begin{problem}\label{qeu:444444}
(i) \cite{Yao-Ma-arXiv-2201-13354v1} Suppose that $G$ is a connected graph containing some Hamilton cycles. If each $k$-path of $k$ edges with $k\geq 1$ is in a Hamilton cycle of the graph $G$, we say that $G$ is \emph{$k$-path-hamiltonian}. \textbf{Characterize} $k$-path-hamiltonian graphs.

(ii) $^*$ \textbf{Determine} a positive real number $\mu$, such that each connected graph $G$ is edge-hamiltonian if minimal degree $\delta(G)\geq \mu \cdot |V(G)|$.

(iii) $^*$ Let $T$ be a preappointed connected graph. If each connected subgraph $T$ of a connected graph $G$ corresponds a cycle $C$ of the connected graph $G$ holding $|V(T)\cap V(C)|\geq 1$, we say the connected graph $G$ to be $T$-cycle. For example, $T$ is a $k$-path with $k\geq 1$. \textbf{Characterize} a $T$-cycle graph for a preappointed connected graph $T$.
\end{problem}

\subsubsection{Maximal planar graphic lattices}

\begin{defn} \label{defn:111111}
$^*$ Suppose that $\textbf{\textrm{T}}=(T_1,T_2,\dots ,T_m)$ is an \emph{unavoidable and reducible configuration base}, where each graph $T_i$ with $i\in [1,m]$ is a maximal planar graph containing just one unavoidable and reducible configuration, and no two $T_i$ and $T_j$ with $i\neq j$ contain the same unavoidable and reducible configuration, and there is an \emph{operation base} $\textbf{\textrm{O}}=(O_1,O_2,\dots ,O_n)$ such that doing each operation $O_k$ to two maximal planar graphs $T_i$ and $T_j$ produces a maximal planar graph, denoted as $O_k\langle T_i,T_j \rangle$; also, one operation $O_k$ can be applied to a maximal planar graphs $T_i$, for example, the 4-colored rhombus operation introduced in \cite{Hongyu-Wang-2018-Doctor-thesis}. More operations of maximal planar graphs can be found in \cite{Jin-Xu-Maximal-Science-Press-2019}, \cite{Jin-Xu-55-56-configurations-arXiv-2107-05454v1} and \cite{Hongyu-Wang-2018-Doctor-thesis}.

About the operation base $\textbf{\textrm{O}}$, refer to the subsection ``$W$-coinciding and $W$-splitting operations''.

For a permutation $G_1,G_2,\dots ,G_B$ of maximal planar graphs $b_1T_1,b_2T_2,\dots ,b_mT_m$, where $B=\sum ^m_{i=1}b_i$ for $a_i\in Z^0$ and $T_i\in \textbf{\textrm{T}}$, we do an operation $O_{i_1}$ selected randomly from $\textbf{\textrm{O}}$ to $G_1$ and $G_2$, we get a maximal planar graph $H_1=O_{i_1}\langle G_1,G_2\rangle$, in the same way, we get the second maximal planar graph $H_2=O_{i_2}\langle H_1,G_3\rangle$, and maximal planar graphs $H_{k}=O_{i_k}\langle H_{k-1},G_{k+1}\rangle$ with $k\in [2,m-1]$. In general, we write
$$H_{m-1}=O_{i_{m-1}}\langle H_{m-2},G_{m}\rangle=\big [\textbf{\textrm{O}}\big ]^m_{i=1}b_iT_i
$$ and get a \emph{maximal planar graphic lattice}
\begin{equation}\label{eqa:maximal-planar-graph-configuration}
\textbf{\textrm{L}}\big (Z^0[\textbf{\textrm{O}}]\textbf{\textrm{T}}\big )=\Big \{\big [\textbf{\textrm{O}}\big ]^m_{i=1}b_iT_i:b_i\in Z^0,T_i\in \textbf{\textrm{T}}\Big \}
\end{equation} with $\sum ^m_{i=1}b_i\geq 1$.\qqed
\end{defn}

\begin{problem}\label{question:444444}
\textbf{Do} some maximal planar graphs of the maximal planar graphic lattice $\textbf{\textrm{L}}\big (Z^0[\textbf{\textrm{O}}]\textbf{\textrm{T}}\big )$ contain new unavoidable and reducible configurations, which differ completely from the known $1936$ unavoidable and reducible configurations proposed in \cite{Kenneth-Appel-Wolfgang-Haken1976}, and differ completely from the known $633$ unavoidable and reducible configurations proposed in \cite{Robertson-Sanders-Daniel-Seymour-Thomas1996}?
\end{problem}

\begin{rem}\label{rem:333333}
In \cite{Robertson-Sanders-Daniel-Seymour-Thomas1996}, Robertson, Daniel, Seymour, and Thomas created a quadratic-time algorithm, improving on a quartic-time algorithm based on Appel and Haken's proof. This new proof is similar to Appel and Haken's but more efficient because it reduces the complexity of the problem and requires checking only 633 unavoidable and reducible configurations. Appel and Haken's method needed to check 1936 unavoidable and reducible configurations. However, both the unavoidability and reducibility parts of this new proof must be executed by computer and are impractical to check by hand \cite{RThomas-Robin-1998}.

If nothing new unavoidable and reducible configurations is discovered, then the graphic lattice $\textbf{\textrm{L}}\big (Z^0[\textbf{\textrm{O}}]\textbf{\textrm{T}}\big )$ shown in Eq.(\ref{eqa:maximal-planar-graph-configuration}) can be as a mathematical proof of the 4-color conjecture; otherwise, the computer proofs of the 4-color conjecture given in \cite{Kenneth-Appel-Wolfgang-Haken1976} and \cite{Robertson-Sanders-Daniel-Seymour-Thomas1996} can only be considered as some experiments.\qqed
\end{rem}

Jin Xu has done many meaningful works on the 4-color conjecture of maximal planar graphs (Ref. \cite{Jin-Xu-Maximal-Science-Press-2019}), and he has shown ``55-configurations and 56-configurations are reducible'' in \cite{Jin-Xu-55-56-configurations-arXiv-2107-05454v1} for achieving the mathematical proof of the 4-color conjecture of maximal planar graphs.

However, if every planar graph can be colored with four colors, but it is NP-complete in complexity to decide whether an arbitrary planar graph can be colored with just three colors (Ref. \cite{Dailey-D-P-1980}), and if a planar graph $G$ can be colored with just three colors, \textbf{how many} 3-colorings does $G$ admit? In other words, it is a sharp-P-complete problem.

\vskip 0.2cm

\begin{conj} \label{conj:c2-Albertson-Chappell}
(1) (Albertson-Berman \cite{Albertson-Berman-large-forest-1979}) Every planar graph $G$ has an induced subgraph with at least half of the vertices that is a forest that is, the acyclic number $a(G)\geq \frac{1}{2}|G|$.

(2) (M. Albertson and R. Haas, 1998) Every bipartite planar graph $G$ has an induced subgraph with at least $\frac{5}{8}$ of the vertices that is a forest, namely, the acyclic number $a(G)\geq \frac{5}{8}|G|$.

(3) (Chappell) Every planar graph $G$ has an induced subgraph with more than $\frac{4}{9}$ of the vertices that is a linear forest, that is, the acyclic number $a(G)\geq
\frac{4}{9}|G|$.
\end{conj}

\begin{conj}\label{conj:2-color-tree-forest}
\cite{Yao-Su-Ma-Wang-Yang-arXiv-2202-03993v1} There is a proper vertex $4$-coloring $f$ of a maximal planar graph $G$ of $p$ vertices such that a 2-color tree (or forest) $T$ holds $|V(T)|\geq \frac{p}{2}$ true. It is related with the Big Forest Conjecture of planar graphs \cite{Albertson-Berman-large-forest-1979}.
\end{conj}

\begin{rem}\label{rem:333333}
About Conjecture \ref{conj:2-color-tree-forest}, we recall the Big Forest Conjecture of planar graphs is proposed in \cite{Albertson-Berman-large-forest-1979}: ``Every planar graph of order $n$ contains an induced forest of order at least $\frac{n}{2}$.'' In 1986, Erd\"{o}s \emph{et al.}, in \cite{Erdos-Saks-Sos-1986}, proved that the sum of the dwindling number of a graph $G$ and the order of the maximum induced forest of the graph $G$ is exactly equal to the order of $G$. If the Big Forest Conjecture of planar graphs is true, then any planar graph of order $n$ contains an independent set having at least $\frac{n}{4}$ vertices, which can be obtained directly from the Four Color Problem on planar graphs.
\end{rem}

\section{Exploring Hypernetworks}

Tom Siegfried said \cite{Tom-Siegfried-Wolfram-hypergraph-2020}: ``\emph{Wolfram's hypergraphs reproduce many of the consequences of various physical theories, such as Einstein's special theory of relativity. Traveling rapidly slows down time (as special relativity says) because \emph{hypergraph structures} corresponding to moving objects make an angle through the hypergraph that extends the distance between updates (or time steps). The speed of light is a maximum velocity, as relativity states, because it represents the maximum rate that information can spread through the hypergraph as it updates. And gravity described by Einstein's general theory of relativity emerges in the relationship between features in the hypergraph that can be interpreted as matter particles. Particles would be small sets of linked points that persist as the hypergraph updates, something like ``little lumps of space'' with special properties}.''

Three physical scientists Newman, Barab\'{a}si and Watts, in \cite{Newman-Barabasi-Watts2006}, pointed:``\emph{Pure graph theory is elegant and deep, but it is not especially relevant to networks arising in the real world. Applied graph theory, as its name suggests, is more concerned with real-world network problems, but its approach is oriented toward design and engineering.}''

\subsection{Concepts of hypernetworks}

Hypernetworks have been mentioned or studied by scholars for a long time. However, there is no recognized definition of the hypernetwork in our memory, and moreover some articles consider hypernetworks as hypergraphs.

We try do some researching works on \emph{hypernetworks}, using dynamic vertex-intersected networks to observe indirectly hypernetworks, and try to plant some researching results of dynamic networks into hypernetworks.

\begin{defn} \label{defn:dynamic-hypernetworks}
\cite{Yao-Ma-arXiv-2201-13354v1} \textbf{Dynamic hypernetwork.} At each time step $t\in [a,b]$ for two integers $a$ and $b$ subject $0\leq a<b$, a \emph{dynamic hypernetwork} $\mathcal{N}_{yper}(t)=(\Lambda(t),\mathcal{E}(t))$ holds:

\textbf{Hynet-1.} Each hyperedge $e(t)\in \mathcal{E}(t)$ is not an empty set and corresponds another hyperedge $e\,'(t)\in \mathcal{E}(t)\setminus\{e(t)\}$ to hold $e(t)\cap e\,'(t)\neq \emptyset$ true;

\textbf{Hynet-2.} $\Lambda(t)=\bigcup _{e(t)\in \mathcal{E}(t)}e(t)$ is a finite set.\qqed
\end{defn}

\begin{rem}\label{rem:333333}
There are many dynamic hypernetworks of form $\mathcal{N}_{yper}(t)=(\Lambda(t),\mathcal{E}(t))$ for each time step $t\in [a,b]$, since there are many hyperedge sets $\mathcal{E}\in \mathcal{E}\big (\Lambda^2(t)\big )$ holding \textbf{Hynet}-1 and \textbf{Hynet}-2 of Definition \ref{defn:dynamic-hypernetworks}.

Hereafter, we will omit ``dynamic'' from ``dynamic hypernetwork'' for the simplicity of statement. In Definition \ref{defn:dynamic-hypernetworks}, each hyperedge $e(t)\in \mathcal{E}(t)$ describes a local network (community), blockchain, logistics network, cybergroup, \emph{etc}. Hypernetworks describe dynamic changes among groups and communities in networks, and compound hypernetworks depict correlations between two or more hypernetworks.\qqed
\end{rem}

Motivated from Proposition \ref{prop:hyperedge-sets-unions}, we define a proper increasing hypernetwork as follows:

\begin{defn} \label{defn:hypernetworks-00}
$^*$ A hypernetwork $\mathcal{N}_{yper}(t)=(\Lambda(t),\mathcal{E}(t))$ is proper increasing if $\Lambda(t_i)\subset \Lambda(t_{i+1})$ with $t_{i}<t_{i+1}$, and the vertex set $\Lambda(t)=\Lambda(t_i)\cup \Lambda(t_{i+1})$, as well as hyperedge set $\mathcal{E}(t)=\mathcal{E}(t_i)\cup \mathcal{E}(t_{i+1})$ with $\mathcal{E}(t_i)\in \mathcal{E}(\Lambda^2(t_i))$ and $\mathcal{E}(t_{i+1})\in \mathcal{E}(\Lambda^2(t_{i+1}))$ at each time step $t_{i+1}\in [a,b]$ for two integers $a$ and $b$ subject $0\leq a<b$.\qqed
\end{defn}

\begin{defn} \label{defn:hypernetwork-vertex-intersected graphs}
$^*$ We say that $N_{int}(t)$ is a \emph{vertex-intersected network} of a hypernetwork $\mathcal{N}_{yper}(t)=(\Lambda(t)$, $\mathcal{E}(t))$ defined in Definition \ref{defn:dynamic-hypernetworks} subject to a constraint set $R_{est}(c_0,c_1,c_2,\dots ,c_m)$ with $m\geq 0$ at each time step $t\in [a,b]$, then the vertex-intersected network $N_{int}(t)$ admits a total set-coloring $F:V(N_{int}(t))\cup E(N_{int}(t))\rightarrow \mathcal{E}(t)$ such that

(i) the first constraint $c_0$ holds $F(uv)\supseteq F(u)\cap F(v)\neq \emptyset$ true; and

(ii)) for some $k\in [1,m]$, the $k$th constraint $c_k$ holds the constraint equation $c_k[a_u,c_{uv},b_v]=0$ for some $c_{uv}\in F(uv)$, $a_u\in F(u)$ and $b_v\in F(v)$.

Conversely, is a hyperedge $e(t)\in \mathcal{E}(t)$ corresponds anther hyperedge $e\,'(t)\in \mathcal{E}(t)$ holding $e(t)\cap e\,'(t)\neq \emptyset $, then there is an edge $xy\in E(N_{int}(t))$ holding $F(xy)\supseteq F(x)\cap F(y)=e(t)\cap e\,'(t)$.\qqed
\end{defn}

Similarly with Theorem \ref{thm:infinite-v-intersected-graphs}, we have the following result:

\begin{thm}\label{thm:666666}
A hypernetwork $\mathcal{N}_{yper}(t)=(\Lambda(t),\mathcal{E}(t))$ defined in Definition \ref{defn:dynamic-hypernetworks} has infinite vertex-intersected networks.
\end{thm}

\begin{rem}\label{rem:333333}
At each time step $t\in [a,b]$, some relations between vertex-intersected networks defined in Definition \ref{defn:hypernetwork-vertex-intersected graphs} are as follows: At each time step $t\in [a,b]$ for two integers $a$ and $b$ subject $0\leq a<b$,

(i) $N_{int}(t-1)\subset N_{int}(t)$, $\Lambda(t-1)\subset \Lambda(t)$ and $\mathcal{E}(t-1)\subset \mathcal{E}(t)$;

(ii) $N_{int}(t)=[O_{j_1},O_{j_2},\dots ,O_{j_{m_i}}]\langle N_{int}(t-1),N_{int}(t-2),\dots ,N_{int}(t-k)\rangle$ for operations $O_{j_i}\in \textbf{\textrm{O}}$ with $k\geq 1$, where $\textbf{\textrm{O}}=\{O_k\}^n_{k=1}$ is a \emph{network operation base};

(iii) We are interesting that each hypernetwork $N_{int}(t)$ obeys some topological structure properties, such as self-similarity, scale-free distribution, coefficient, small wold, homogeneousness, average, develop velocity \emph{etc}.\qqed
\end{rem}

\begin{defn} \label{defn:hypernetworks-11}
$^*$ \textbf{Multi-hyperedge hypernetworks, hyperedge-intersected networks.} Let $\mathcal{H}_{yper}(t)=(\Lambda(t),\mathcal{E}(t))$ for each time step $t\in [a,b]$ be a hypergraph and $\mathcal{E}(\Lambda^2(t))=\{\mathcal{E}_j(t):~j\in [1, n(\Lambda (t))]\}$ be the hypergraph set defined in Definition \ref{defn:hypergraph-basic-definition}. A \emph{hyperedge-set sequence}
$$\{\mathcal{E}_i(t)\}^m_{i=1}=\{\mathcal{E}_i(t)\in \mathcal{E}(\Lambda^2(t)):i\in [1,m]\}
$$ with $\mathcal{E}_i(t)\not \subseteq \mathcal{E}_j(t)$ if $i\neq j$ forms a \emph{hyperedge-intersected network} $N(t)$ of a \emph{multi-hyperedge hypernetwork} $\mathcal{C}_{hynet}(t)=\big (\Lambda(t),\{\mathcal{E}_i(t)\}^m_{i=1}\big )$ admitting a set-coloring
\begin{equation}\label{eqa:555555}
F:V(N(t))\cup E(N(t))\rightarrow \{\mathcal{E}_i(t)\}^m_{i=1}
\end{equation} such that $F(x)\neq F(y)$ for each edge $xy\in E(N(t))$, and each edge $uv\in E(N(t))$ is colored by
\begin{equation}\label{eqa:555555}
F(uv)=\mathcal{E}_{uv}(t)=\mathcal{E}_{u}(t)[\bullet]\mathcal{E}_{v}(t)=F(u)[\bullet]F(v)
\end{equation} under an operation ``$[\bullet]$'' on the hypergraph set $\mathcal{E}\big (\Lambda^2(t)\big )$. So, each hyperedge set $\mathcal{E}_i(t)\in \mathcal{E}\big (\Lambda^2(t)\big )$ is a \emph{local network} of the multi-hyperedge hypernetwork $\mathcal{C}_{hynet}(t)$ for each time step $t\in [a,b]$.\qqed
\end{defn}

\begin{problem}\label{question:444444}
Since the number of subsets of the set $\Lambda(t)=\{x_1,x_2,\dots,x_{p(t)}\}$ is $|\Lambda^2(t)|=2\,^{p(t)}-1$ for each time step $t\in [a,b]$, so we want to know the value $|\mathcal{E}(\Lambda^2(t))|$ and the topological structure of the hypergraph set $\mathcal{E}(\Lambda^2(t))$.
\end{problem}

\subsection{Scale-free vertex-intersected networks}

In 1999, Barabasi and Albert in \cite{Barabasi-Albert1999} have shown that a \emph{scale-free network} $N(t)$ has its own \emph{degree distribution}
\begin{equation}\label{eqa:Barabasi-Albert1999}
P(k)=\textrm{Pr}(x=k)\sim k^{-\lambda}, ~2<\lambda<3
\end{equation}
where $P(k)$ is the \emph{power-law distribution} of a vertex joined with $k$ vertices in the scale-free network $N(t)$ at each time step $t\in [a,b]$.

Since a vertex-intersected network $N_{int}(t)$ admits a total set-coloring
$$
F_t:V(N_{int}(t))\cup E(N_{int}(t))\rightarrow \mathcal{E}(t),~t\in [a,b]
$$ where $\mathcal{E}(t)$ is the hyperedge set of a hypernetwork $\mathcal{N}_{yper}(t)=(\Lambda(t),\mathcal{E}(t))$ at each time step $t\in [a,b]$. If $N_{int}(t)$ is \emph{scale-free}, then a vertex $w\in V(N_{int}(t))$ joined with $k$ vertices in $N_{int}(t)$ corresponds a hyperedge $e\in \mathcal{E}(t)$ holding $F_t(w)=\{e\}$, which means that the \emph{hyperedge-degree distribution} of the hypernetwork $\mathcal{N}_{yper}(t)$ obeys the \emph{power-law distribution} defined in Eq.(\ref{eqa:Barabasi-Albert1999}) too, thus, we say $\mathcal{N}_{yper}(t)$ to be a \emph{scale-free hypernetwork}.

According to Definition \ref{defn:hypernetwork-vertex-intersected graphs} and some fundamental characteristics of a \emph{scale-free network model} summarized in \cite{Yao-Ma-Su-Wang-Zhao-Yao-2016}, a \emph{scale-free vertex-intersected network} $N_{int}(t)$ at each time step $t\in [a,b]$ having $v_{net}(t)$ vertices and $e_{net}(t)$ edges holds the following topological properties:

\textbf{Fc-1.} \cite{Barabasi-Albert1999} Two \emph{common mechanisms}. \emph{Growth} is $v_{net}(t)>v_{net}(t-1)$ and $e_{net}(t)>e_{net}(t-1)$, and \emph{Preferential attachment} is defined by $\Pi_i(t)=k_i(t)/\sum_jk_j(t)$, where $k_i(t)$ and $k_ j(t)$ are hyperedge-degrees of hyperedges $e_i$ and $e_ j$ of the hyperedge set $\mathcal{E}(t)$.

\textbf{Fc-2.} \cite{Barabasi-Albert1999} There is a \emph{dynamic equation}
\begin{equation}\label{eqa:aa-dynamic-equations}
\frac{\partial k_i(t)}{\partial t}=m\Pi_i(t),~t\in [a,b]
\end{equation} Using the initial condition $k_i(t_i)=m$ solves the hyperedge-degree function $k_i(t)$ from the dynamic equation Eq.(\ref{eqa:aa-dynamic-equations}).

\textbf{Fc-3.} A \emph{sum} $\sum n_i(k_i(t))\frac{\partial k_i(t)}{\partial t}$, where $n_i(k_i(t))$ is the number of hyperedges having hyperedge-degree $k_i(t)$ in $N_{int}(t)$.

\textbf{Fc-4.} \cite{Barabasi-Albert1999} A \emph{hyperedge-degree distribution} $P(k)\sim k^{-\gamma}$ and a \emph{hyperedge cumulative distribution} $P_{cum}(k)\sim k^{1-\gamma}$ with $2<\gamma<3$ hold
\begin{equation}\label{eqa:Barabasi-hyperedge-degree-distribution}
P(k)=\frac{\partial P(k_i(t)<k)}{\partial k},\quad P(k)=-\frac{\partial P_{cum}(k)}{\partial k},~t\in [a,b]
\end{equation}The \emph{hyperedge cumulative distribution} is defined as
\begin{equation}\label{eqa:Dorogovstev-Goltsev2002}
P_{cum}(k)=\sum_{k\,'\geq k}\frac{|V(k\,',t)|}{v_{net}(t)}\sim k^{1-\lambda},~t\in [a,b]
\end{equation}
with $2<\lambda<3$, where $|V(k\,',t)|$ is the number of hyperedges of hyperedge-degree $k\,'$ (Ref. \cite{Dorogovstev-2002}).

The \emph{hyperedge hyperedge-cumulative distribution} is
\begin{equation}\label{eqa:edge-cumulative-distribution}
P_{ecum}(k) =\sum _{k\,'\geq k}\frac{E(k\,',t)}{e_{net}(t)}\sim k^{1-\delta},~t\in [a,b]
\end{equation} with $2<\delta<3$, where $E(k\,',t)$ is the number of hyperedges adjacent with hyperedges of hyperedge-degree $k\,'$ greater than $k$ at each time step $t\in [a,b]$ (Ref. \cite{Liu-Yao-Zhang-Chen-Liu-Yao-ICT2014}), and the \emph{$d$-hyperedge-cumulative distribution} is
\begin{equation}\label{eqa:d-edge-cumulative-distribution}
P^{deg}_{ecum}(k) =\sum _{k\,'\geq k}\frac{k\,'E(k\,',t)}{e_{net}(t)}\sim k^{1-\varepsilon},~t\in [a,b]
\end{equation} with $2<\varepsilon<3$ (Ref. \cite{Yao-Wang-Su-Ma-Yao-Zhang-Xie2016}).

\textbf{Fc-5.} We define the \emph{velocity} $V_{elo}(N_{int}(t))$ of a vertex-intersected network $N_{int}(t)$ as
\begin{equation}\label{eqa:555555}
V_{elo}(N_{int}(t))=\sqrt{\left [\frac{\partial v_{net}(t)}{\partial t}\right ]^2+ \left [\frac{\partial e_{net}(t)}{\partial t}\right ]^2},~t\in [a,b]
\end{equation}
where $\partial v_{net}(t)/\partial t$ is the \emph{v-velocity} and $\partial e_{net}(t)/\partial t$ is the \emph{e-velocity}. The \emph{speed ratio} of the vertex-intersected network $N_{int}(t)$ is defined by $\partial v_{net}(t)/\partial t\div \partial e_{net}(t)/\partial t$, so the \emph{average hyperedge-degree} $\langle k\rangle (t)$ of the hyperedge set $\mathcal{E}(t)$ holds
\begin{equation}\label{eqa:c3xxxxx}
\frac{\partial e_{net}(t)}{\partial t}\sim \langle k\rangle (t)\cdot \gamma \cdot \frac{\partial v_{net}(t)}{\partial t},~t\in [a,b]
\end{equation} where $\gamma$ is a constant (Ref. \cite{Yao-Wang-Su-Ma-Yao-Zhang-Xie2016}).

\textbf{Fc-6.} The scale-free hypernetwork $N_{int}(t)$ is \emph{sparse} (Ref. \cite{Genio-Gross-Bassler2011}), so its average hyperedge-degree $\langle k\rangle (t)$ is approximate to a constant, or
\begin{equation}\label{eqa:sparse-edge-vertex-numbers}
e_{net}(t) \sim \textrm{O}(v_{net}(t)\ln[v_{net}(t)]),~t\in [a,b]
\end{equation}

\textbf{Fc-7.} $^*$ Let $R_{eceive}(u,t)=R^{out}_{eceive}(u,t)+R^{inner}_{eceive}(u,t)$ be the amount of information received for a vertex $u$, where $R^{out}_{eceive}(u,t)$ is the amount of information received from out of $N_{int}(t)$, and $R^{inner}_{eceive}(u,t)$ is the amount of information received from inside of $N_{int}(t)$. And let $S_{end}(u,t)=S^{out}_{end}(u,t)+S^{inner}_{end}(u,t)$ be the amount of information sent out from a vertex $u$, where $S^{out}_{end}(u,t)$ is the amount of information sent out of $N_{int}(t)$, and $S^{inner}_{end}(u,t)$ is the amount of information sent to the vertices of $N_{int}(t)$.

Thereby, at each time step $t\in [a,b]$, the amount $R_{eceive}(N_{int}(t))$ of information received by the hypernetwork $N_{int}(t)$ is
\begin{equation}\label{eqa:555555}
R_{eceive}(N_{int}(t))=\sum _{u\in N_{int}(t)}R_{eceive}(u,t)=\sum _{u\in N_{int}(t)}R^{out}_{eceive}(u,t)+\sum _{u\in N_{int}(t)}R^{inner}_{eceive}(u,t)
\end{equation} and the amount $S_{end}(N_{int}(t))$ of information sent out from the hypernetwork $N_{int}(t)$ is
\begin{equation}\label{eqa:555555}
S_{end}(N_{int}(t))=\sum _{u\in N_{int}(t)}S_{end}(u,t)=\sum _{u\in N_{int}(t)}S^{out}_{end}(u,t)+\sum _{u\in N_{int}(t)}S^{inner}_{end}(u,t)
\end{equation} The amounts of information received and sent describes the active level of the hypernetwork $N_{int}(t)$. Such methods can be used to Internet of Things (Ref. \cite{Yao-Liu-Zhang-Chen-Zhang-Yao-Zhao-2013}).

\vskip 0.2cm

\textbf{Fc-8.} \textbf{DGN-vertex-intersected networks.} Suppose that a vertex-intersected network $N_{int}(t)$ is a \emph{deterministic growing vertex-intersected network } having $v_{net}(t)$ vertices and $e_{net}(t)$ edges at each time step $t\in [a,b]$, we call $N_{int}(t)$ \emph{DGN-vertex-intersected network}.

\vskip 0.2cm

In \cite{Yao-Ma-Su-Wang-Zhao-Yao-2016}, the authors have considered: A DGN-vertex-intersected network $N_{int}(t)$ satisfies a system of linear equations (\emph{linear growth})
\begin{equation}\label{eqa:linear-growth}
v_{net}(t)=a_v t+b_v,\quad e_{net}(t)=a_e t+b_e,~t\in [a,b]
\end{equation}
with $a_v>0$ and $a_e>0$. As known, Barabasi-Albert's model holds Eq.(\ref{eqa:linear-growth}) true (Ref. \cite{Barabasi-Albert1999}). Very often, a system of non-linear equations (\emph{exponential growth}) appeared in many literature is
\begin{equation}\label{eqa:Only-exponent}
v_{net}(t)=a_v r\,^t+b_v,\quad e_{net}(t)=a_e s\,^t+b_e,~t\in [a,b]
\end{equation}
with $a_v>0$, $a_e>0$, $|r|\neq 1$ and $|s|\neq 1$.

There are deterministic growing network models holding Eq.(\ref{eqa:Only-exponent}), such as the Sierpinski model $S(t)$ with $r=3$ (Ref. \cite{Zhang-Zhou-Fang-Guan-Zhang2007}), the Recursive tree model $R(t)$ with $r=q+1$ for $q\geq 2$ (Ref. \cite{Francesc-2004}), and the Apollonian model $A(t)$ with $r=m(d+1)$ (Ref. \cite{Zhang-Comellas-Fertin-Rong-2006}), we call them \emph{$r$-rank models}. For example, a tripartite model introduced in \cite{Lu-Su-Guo2013} and a 4-partite model appeared in \cite{Zhang-Rong-Guo2006} both are \emph{2-rank models}.

\vskip 0.2cm

\textbf{Fc-9.} \textbf{Random growth networks.} If a network $N(t)$ holds the vertex sets and edge sets $|V(N(t))|<|V(N(t+1))|$ and $|E(N(t))|<|E(N(t+1))|$ for each time step $t\in[a,b]$ for two integers $a$ and $b$ subject $1\leq a<b$, then we say that the network $N(t)$ is randomly \emph{ve-all increasing}.

In Fig.\ref{fig:increase-random}, each connected graph $G_i$ admits a graceful-difference total labeling $f_i$ holding $\big ||f_i(u)-f_i(v)| -f_i(uv)\big |=8$ for each edge $uv\in E(G_i)$ with $i\in [0,3]$. Since vertex numbers $|V(G_j)|<|V(G_{j+1})|$ for $j\in [0,2]$, each network $G_i$ is a random growing network. Notice that each growing network $G_i$ is \emph{colored graph homomorphism} to another network, for example, $G_3\rightarrow H$ shown in Fig.\ref{fig:increase-random}(e), so the total coloring $h$ admitted by the network $H$ is a felicitous-difference total labeling too, such that $\big ||h(x)-h(y)| -h(xy)\big |=8$ for each edge $xy\in E(H)$. Conversely, we can vertex-split the network $H$ (as a public-key) to obtain the network $G_3$ (as a private-key).

\begin{figure}[h]
\centering
\includegraphics[width=16.4cm]{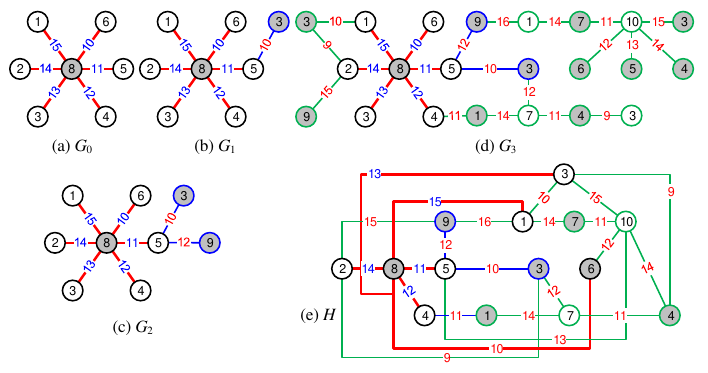}\\
\caption{\label{fig:increase-random}{\small A random growth network model from (a) to (d).}}
\end{figure}

\begin{problem}\label{question:444444}
Vertex-splitting a total colored connected graph $H$ to another total colored connected graph $G_m$, such that each total colored connected graph $G_k$ is the result of adding leaves to the total colored connected graph $G_{k-1}$ for $k\in [1,m]$ with $|V(G_{k-1})|<|V(G_k)|$ and $|V(G_{0})|=1$, refer to Fig.\ref{fig:increase-random}.
\end{problem}

\subsection{Lattices based on vertex-intersected networks}

In the article \cite{Yao-Ma-arXiv-2201-13354v1}, the authors have researched hypernetworks, here, we will to conduct more researching hypernetworks.

Let $\textbf{\textrm{O}}=\{O_1,O_2,\dots ,O_n\}$ be a set of graph operations \cite{Yao-Su-Sun-Wang-Graph-Operations-2021}. So, we call a subset $\textbf{\textrm{O}}_i=\{O_{i,1}$, $O_{i,2}$, $\dots ,O_{i,m_i}\}\subset \textbf{\textrm{O}}$ to be \emph{$W$-operation base} if each graph operation $O_{i,j}$ is not a compound of others, and obeys the \emph{$W$-property}. We use symbol $a_kO_{i,k}$ to indicate $a_k$ copies of $O_{i,k}\in \textbf{\textrm{O}}_i$ for $k\in [1,m_i]$, and $O_{j_1},O_{j_2},\dots ,O_{j_{B}}$ is a permutation of these graph operations $a_1O_{i,1}$, $a_2O_{i,2}$, $\dots$, $a_{m_i}O_{i,m_i}$, where $B=\sum ^{m_i}_{k=1}a_k$.

We implement these graph operations $O_{j_1},O_{j_2},\dots ,O_{j_{B}}$ to a vertex-intersected network $N_{int}(t)$ at each time step $t\in [a,b]$, one by one, in the following way: The first graph operation $O_{j_1}$ is implemented to a vertex-intersected network $N_{int}(t)$, the resultant vertex-intersected network is denoted as $G_1(t)=O_{j_1}(N_{int}(t))$, and we have $G_2(t)=O_{j_2}(G_1(t))$, $G_3(t)=O_{j_3}(G_2(t))$, $\dots$, $G_{B}(t)=O_{j_{B}}(G_{B-1}(t))$. We write the last vertex-intersected network $G_{B}(t)=N_{int}(t) \lhd |^{m_i}_{k=1} a_kO_{i,k}$ for integrity. Simulating with the lattices introduced in \cite{Yao-Su-Sun-Wang-Graph-Operations-2021}, we define an \emph{operation vertex-intersected network lattice} as follows
\begin{equation}\label{eqa:operation-base-graphic-latticess}
{
\begin{split}
\textbf{\textrm{L}}\big (\textbf{\textrm{F}}(t)\lhd Z^0\textbf{\textrm{O}}_i\big )=&\big \{N_{int}(t)\lhd |^{m_i}_{k=1} a_kO_{i,k} :~ a_k\in Z^0,O_{i,k}\in \textbf{\textrm{O}}_i,N_{int}(t)\in \textbf{\textrm{F}}(t)\big \}
\end{split}}
\end{equation} with $\sum ^{m_i}_{k=1}a_k\geq 1$, where $\textbf{\textrm{F}}(t)$ is the set of vertex-intersected networks at each time step $t\in [a,b]$, and $\textbf{\textrm{O}}_i$ is called \emph{$\Gamma$-operation base}.

A \emph{$\Gamma$-operation base} $\textbf{\textrm{O}}_i$ is called \emph{scale-free operation base} if each operation $O_{i,k}\in \textbf{\textrm{O}}_i$ is a \emph{scale-free operation} defined by $P(k) \sim k^{-\lambda_{i,k}}$ with $\lambda_{i,k}>2$. We call the following set
\begin{equation}\label{eqa:scale-free-network-lattice11}
{
\begin{split}
&\textbf{\textrm{L}}\big (\textbf{\textrm{F}}_{scale}(t)\lhd Z^0\textbf{\textrm{O}}_i\big )=\big \{N_{int}(t)\lhd |^{m_i}_{k=1} b_kO_{i,k}:~ b_k\in Z^0, O_{i,k}\in \textbf{\textrm{O}}_i,N_{int}(t)\in \textbf{\textrm{F}}_{scale}(t)\big \}
\end{split}}
\end{equation} to be \emph{scale-free operation vertex-intersected network meta-lattice} with $\sum ^{m_i}_{k=1}b_k\geq 1$, where $\textbf{\textrm{F}}_{scale}(t)$ is the set of scale-free vertex-intersected networks $N_{int}(t)$ at each time step $t\in [a,b]$, and $\textbf{\textrm{O}}_i$ is a \emph{scale-free operation base} such that each network of $\textbf{\textrm{L}}\big (\textbf{\textrm{F}}_{scale}(t)\lhd Z^0\textbf{\textrm{O}}_i\big )$ is scale-free, in other word, $\textbf{\textrm{L}}\big (\textbf{\textrm{F}}_{scale}(t)\lhd Z^0\textbf{\textrm{O}}_i\big )$ is really a \emph{scale-free network generator}.

\begin{problem}\label{question:444444}
Since $N_{int}(t)$ is a \emph{vertex-intersected network} of a hypernetwork $\mathcal{N}_{yper}(t)=(\Lambda(t)$, $\mathcal{E}(t))$ defined in Definition \ref{defn:dynamic-hypernetworks}, if the vertex-intersected network $N_{int}(t)$ holds a $W$-property of networks, we say that $N_{int}(t)$ is a $W$-property vertex-intersected network with $\Lambda(t-1)\subseteq \Lambda(t)$ and $\mathcal{E}(t-1)\subseteq \mathcal{E}(t)$ at each time step $t\in [a,b]$. \textbf{Investigate} various $W$-property vertex-intersected networks.
\end{problem}

\subsection{Meta-lattices based on hypernetworks}

Suppose that a connected graph $H(t)$ has its own vertex set $V(H)=\{\Lambda_i(t):i\in [1,m]\}$, where each $\Lambda_i(t)$ is a finite set, and each edge $\Lambda_i(t)\Lambda_j(t)$ of $H(t)$ holds $\Lambda_i(t)\cap \Lambda_j(t)\neq \emptyset$ for each time step $t\in [a,b]$, as well as $\Lambda_i(t)\not \subseteq \Lambda_j(t)$ for $i\neq j$. We have a finite set $\Lambda_{\cup}(t)=\bigcup^m_{i=1}\Lambda_i(t)$ and a multi-hyperedge set
\begin{equation}\label{eqa:multi-hyperedge-sets-11}
\mathcal{E}_{\cup}(t)=\big \{\mathcal{E}_{i,j}(t)\in \mathcal{E}\big (\Lambda^2_i(t)\big ):~j\in [1,n_i(t)],i\in [1,m]\big \}
\end{equation} where $\mathcal{E}_{i,j}(t)=\big \{e_{i,j,s}(t)\big \}^{n_{i,j}(t)}_{s=1}$ with $e_{i,j,s}(t)\in \Lambda^2_i(t)$, $n_i(t)\geq 1$ and $\mathcal{E}\big (\Lambda^2_i(t)\big )$ is the hypergraph set for $i\in [1,m]$ (Ref. Definition \ref{defn:hypergraph-basic-definition}). We call the following set

\begin{equation}\label{eqa:hypergraph-555555}
\textbf{\textrm{B}}(H(t),m)=(\Lambda_1(t),\Lambda_2(t),\dots ,\Lambda_m(t),H(t)),~\Lambda_i(t)\not \subset \Lambda_j(t), \Lambda_i(t)\not \cong \Lambda_j(t),i\neq j
\end{equation} \emph{hypernetwork meta-lattice base}, and we get a \emph{hypernetwork meta-lattice} based on $\textbf{\textrm{B}}(H(t),m)$ at each time step $t\in [a,b]$, denoted by $\textbf{\textrm{L}}(\Lambda_{\cup}(t),\mathcal{E}_{\cup}(t)\mid\textbf{\textrm{B}}(H(t),m))$.

Furthermore, we have the hyperedge set $\overline{\mathcal{E}}_{i,j}(t)\in \mathcal{E}\big (\Lambda^2_i(t)\big )$, where $\overline{\mathcal{E}}_{i,j}(t)=\big \{\overline{e}_{i,j,s}(t)\big \}^{n_{i,j}(t)}_{s=1}$ with $\overline{e}_{i,j,s}(t)\in \Lambda^2_i(t)$, such that $\overline{e}_{i,j,s}(t)=\Lambda^2_i(t)\setminus e_{i,j,s}(t)$ for $s\in [1,n_{i,j}(t)]$, $j\in [1,n_i(t)]$, $i\in [1,m]$ at each time step $t\in [a,b]$, and obtain another multi-hyperedge set
\begin{equation}\label{eqa:multi-hyperedge-sets-22}
\overline{\mathcal{E}}_{\cup}(t)=\big \{\overline{\mathcal{E}}_{i,j}(t)\in \mathcal{E}\big (\Lambda^2_i(t)\big ):~j\in [1,n_i(t)],i\in [1,m]\big \}
\end{equation} as well as obtain another hypernetwork meta-lattice $\overline{\textbf{\textrm{L}}}(\Lambda_{\cup}(t),\overline{\mathcal{E}}_{\cup}(t)\mid \textbf{\textrm{B}}(H(t),m))$ based on the hypernetwork meta-lattice base $\textbf{\textrm{B}}(H(t),m)$ shown in Eq.(\ref{eqa:hypergraph-555555}) for each time step $t\in [a,b]$.
In real applications, people can use two hypernetwork meta-lattices $\textbf{\textrm{L}}(\Lambda_{\cup}(t),\mathcal{E}_{\cup}(t)\mid\textbf{\textrm{B}}(H(t),m))$ and $\overline{\textbf{\textrm{L}}}(\Lambda_{\cup}(t),\overline{\mathcal{E}}_{\cup}(t)\mid \textbf{\textrm{B}}(H(t),m))$ as a pair of \emph{hypernetwork meta-lattice-keys} at each time step $t\in [a,b]$, and each $\Lambda_i(t)$ with $i\in [1,m]$ forms a community of hypernetwork meta-lattices.

If the connected graph $H(t)$ is a \emph{scale-free network}, then the vertex set $V(H(t))=\{\Lambda_i(t):i\in [1,p(t)]\}$ with $p(t)\geq 2$ is changing over all time step $t\in [a,b]$, we get a dynamic set $\Lambda^{p(t)}_{\cup}(t)=\bigcup^{p(t)}_{j=1}\Lambda_j(t)$ with $\Lambda_i(t)\not \subseteq \Lambda_j(t)$ for $i\neq j$, and the set
\begin{equation}\label{eqa:555555}
\textbf{\textrm{B}}(H(t),p(t))=(\Lambda_1(t),\Lambda_2(t),\dots ,\Lambda_{p(t)}(t),H(t)),~t\in [a,b]
\end{equation} is called \emph{scale-free hypernetwork meta-lattice base} for each time step $t\in [a,b]$ and we obtain two \emph{scale-free hypernetwork meta-lattices}
\begin{equation}\label{eqa:scale-free-meta-lattices-11}
{
\begin{split}
&\textbf{\textrm{L}}^{sf}\big (\Lambda^{p(t)}_{\cup}(t),\mathcal{E}_{\cup}(t)\mid \textbf{\textrm{B}}(H(t),p(t))\big ),\\
&\mathcal{E}_{\cup}(t)=\big \{\mathcal{E}_{i,j}(t)\in \mathcal{E}\big (\Lambda^2_i(t)\big ):j\in [1,n_i(t)],i\in [1,p(t)]\big \},~t\in [a,b]
\end{split}}
\end{equation} and
\begin{equation}\label{eqa:scale-free-meta-lattices-22}
{
\begin{split}
&\overline{\textbf{\textrm{L}}}^{sf}\big (\Lambda^{p(t)}_{\cup}(t),\overline{\mathcal{E}}_{\cup}(t)\mid \textbf{\textrm{B}}(H(t),p(t))\big ), \\
&\overline{\mathcal{E}}_{\cup}(t)=\big \{\overline{\mathcal{E}}_{i,j}(t)\in \mathcal{E}\big (\Lambda^2_i(t)\big ):j\in [1,n_i(t)],i\in [1,p(t)]\big \},~t\in [a,b]
\end{split}}
\end{equation} where $n_i(t)\geq 1$ for $i\in [1,p(t)]$.

\begin{rem}\label{rem:333333}
Since, at each time step $t\in [a,b]$, $1\leq n_i(t)\leq \big |\mathcal{E}\big (\Lambda^2_i(t)\big )\big |$ for each $i\in [1,p(t)]$ in Eq.(\ref{eqa:multi-hyperedge-sets-11}), Eq.(\ref{eqa:multi-hyperedge-sets-22}), Eq.(\ref{eqa:scale-free-meta-lattices-11}) and Eq.(\ref{eqa:scale-free-meta-lattices-22}), although we do not know the exact value of each number $\big |\mathcal{E}\big (\Lambda^2_i(t)\big )\big |$ with $i\in [1,p(t)]$, however, the (scale-free) hypernetwork meta-lattices introduced in this subsection are fundamentally different from other graphic lattices.\qqed
\end{rem}

\subsection{Graph networks from DeepMind and GoogleBrain}

\subsubsection{Definition, functions and questions of graph networks}

\begin{defn} \label{defn:graph-network-DeepMind}
\cite{Peter-Battaglia-Jessica-et-al-arXiv-2018} A \emph{graph network framework} (GNF) is a set of functions organized according to the graph structure in a topological space, used for relational reasoning and combinatorial generalization. \qqed
\end{defn}

\begin{defn} \label{defn:graph-network-DeepMind22}
\cite{Peter-Battaglia-Jessica-et-al-arXiv-2018} Within graph network framework (can be used to implement a wide variety of architectures), a \emph{graph} $G$ is defined as $G=(u,V,E)$, where $u$ is the global attributes, $V =\big \{v_i\big \}^{N^v}_{i=1}$ is the set of nodes ($|V|=N^v$), and $E =\big \{(e_k, r_k,s_k)\big \}^{N^e}_{k=1}$ is the set of edges ($|E|=N^e$).\qqed
\end{defn}

\begin{rem}\label{rem:333333}
About Definition \ref{defn:graph-network-DeepMind22}, the authors use ``\emph{graph}'' to mean a directed, attributed multi-graph with a global attribute, where a node is denoted as $v_i$, an edge as $e_k$, and the global attributes as $u$. The authors also use $s_k$ and $r_k$ to indicate the indices of the sender and receiver nodes (see below), respectively, for edge $k$. To be more precise, these terms are defined as:

\textbf{Directed}: one-way edges, from a ``sender'' node to a ``receiver'' node.

\textbf{Attribute}: properties that can be encoded as a vector, set, or even another graph.

\textbf{Attributed}: edges and vertices have attributes associated with them.

\textbf{Global attribute}: a graph-level attribute.

\textbf{Multi-graph}: there can be more than one edge between vertices, including self-edges.

The authors in \cite{Peter-Battaglia-Jessica-et-al-arXiv-2018} argue that \emph{combinatorial generalization} must be a top priority for AI to achieve human-like abilities, and that structured representations and computations are key to realizing this objective. And they present a new building block for the AI toolkit with a \emph{strong relational inductive bias} -- \emph{the graph network}, which generalizes and extends various approaches for neural networks that operate on graphs, and provides a straightforward interface for manipulating \emph{structured knowledge} and producing \emph{structured behaviors}. The \emph{principle} of combinatorial generalization supported by \emph{graph networks} is to construct new inferences, predictions, and behaviors from known building blocks.\qqed
\end{rem}

\begin{rem}\label{rem:relational-reasoning-generalization}
In \cite{Peter-Battaglia-Jessica-et-al-arXiv-2018}, the authors from DeepMind, GoogleBrain, MIT and University of Edinburgh have shown the following main functions of graph networks:
\begin{asparaenum}[\textbf{Func}-1. ]
\item Graph networks support \emph{relational reasoning} and \emph{combinatorial generalization}, becoming more complex, interpretable, and flexible reasoning patterns.
\item The graph network framework defines a class of relational inference functions for representing graphical structures, summarizes and extends various MPNN, graph neural network, and NLNN methods, and supports the construction of complex structures from simple building blocks.
\item The main computing unit of the graph network framework is the graph to graph module, which takes the graph as input, performs calculations on the structure, and returns the graph as output.
\item Graph networks generalizes and extends various approaches for neural networks that operate on graphs, and provides a straightforward interface for manipulating \emph{structured knowledge} and producing \emph{structured behaviors}.\qqed
\end{asparaenum}
\end{rem}

\begin{rem}\label{rem:graph-network-open-questions}
And moreover, the authors in \cite{Peter-Battaglia-Jessica-et-al-arXiv-2018} have asked for the solutions of the following open questions:
\begin{asparaenum}[\textbf{GNQ}-1. ]
\item \label{GNQ:graphs-come-from11} \textbf{Where} do the graphs come from that graph networks operate over? Examples of data with more explicitly specified entities and relations include knowledge graphs, social networks, parse trees, optimization problems, chemical graphs, road networks, and physical systems with known interactions.
\item \label{GNQ:graphs-come-from22} Many underlying graph structures are much more sparse than a fully connected graph, and it is an open question how to induce this sparsity.
\item \label{GNQ:graphs-come-from33} One of the hallmarks of deep learning has been its ability to perform complex computations over raw sensory data, such as images and text, yet it is unclear the best ways to convert sensory data into more structured representations like graphs.
\item \label{GNQ:modify-graph-structures} \textbf{How} to adaptively modify graph structures during the course of computation? For example, if an object fractures into multiple pieces, a node representing that object also ought to split into multiple nodes. Similarly, it might be useful to only represent edges between objects that are in contact, thus requiring the ability to add or remove edges depending on context. The question of how to support this type of adaptivity is also actively being researched, and in particular, some of the methods used for identifying the underlying structure of a graph may be applicable.\qqed
\end{asparaenum}
\end{rem}

\begin{problem}\label{qeu:444444}
Let $S_{hyper}=\bigcup ^m_{k=1}A_k$ be a set-set, where $A_k=\{e_{k,1},e_{k,2},\dots ,e_{k,c(k)}\}$ with $c(k)\geq 1$, such that there is no relational function $\varphi_{i,j}$ for any pair of sets $A_i$ and $A_j$ holding $A_j=\varphi_{i,j}(A_i)$, also there is no relational reasoning between elements of the set-set $S_{hyper}$. A graph $G$ admits a set-set-coloring $F:V(G)\rightarrow \mathcal{E}\in \mathcal{E}(S^2_{hyper})$ with $\bigcup_{e\in \mathcal{E}}e=S_{hyper}$, such that each edge $uv\in E(G)$ holds $F(u)\cap F(v)\neq \emptyset$ true. Conversely, each pair of subsets $e,e\,'\in \mathcal{E}$ with $e\cap e\,'\neq \emptyset$ corresponds an edge $xy\in E(G)$, so that $F(x)\cap F(y)\neq \emptyset$. Then the complementary graph $\overline{G}$ of the graph $G$ admits a set-set-coloring $\overline{F}:V(\overline{G})\rightarrow \mathcal{E}\in \mathcal{E}(S^2_{hyper})$, such that $\overline{F}(x)\cap \overline{F}(y)=\emptyset$ for each edge $xy\in E(\overline{G})$, and $\overline{F}(V(\overline{G}))=F(V(G))$.

For each edge $uv\in E(G)$, there is no relational function $\varphi_{a,b}$ for any pair of sets $A_a\in F(u)$ and $A_b\in F(v)$ holding $A_b=\varphi_{a,b}(A_a)$ if $A_a\neq A_b$.

For each edge $xy\in E(\overline{G})$, there is no relational function $\varphi_{s,t}$ for any pair of sets $A_s\in \overline{F}(x)$ and $A_t\in \overline{F}(y)$ holding $A_t=\varphi_{s,t}(A_s)$ if $A_s\neq A_t$.

\textbf{Consider} the impact of the graphs $G$ and $\overline{G}$ to some artificial intelligences.
\end{problem}

\subsubsection{Techniques of topology code theory for graph networks}

We present a definition of graph networks contrasting Definition \ref{defn:graph-network-DeepMind} as follows:

\begin{defn} \label{defn:coloring-define-graph-network}
$^*$ Let $S_{func}(t)=\{f_1(t),f_2(t),\dots ,f_m(t)\}$ be a reasoning function set for $t\in [a,b]$. A network $N(t)$ admits a total coloring $F: M(t)\rightarrow S_{func}(t)$, where $M(t)\subseteq V(N(t))\cup E(N(t))$ for each time step $t\in [a,b]$. Then there are the following constraints:

(i) As $M(t)=V(N(t))$, there is a relational reasoning transformation $\theta_{uv}$ for each edge $uv\in E(N(t))$, such that $F(u)=f_u(t)\in S_{func}(t)$ and $F(v)=f_v(t)\in S_{func}(t)$ hold $f_v(t)=\theta_{uv}[f_u(t)]$, and $F(V(N(t)))=S_{func}(t)$.

(ii) As $M(t)=E(N(t))$, there is a relational reasoning transformation $\theta_{uv,uw}$ for two adjacent edges $uv,uw\in E(N(t))$, such that $F(uv)=f_{uv}(t)\in S_{func}(t)$ and $F(uw)=f_{uw}\in S_{func}(t)$ hold $f_{uw}=\theta_{uv,uw}[f_{uv}(t)]$, and $F(E(N(t)))=S_{func}(t)$.

(iii) As $M(t)=V(N(t))\cup E(N(t))$, such that $F(u)=f_u(t)\in S_{func}(t)$, $F(uv)=f_{uv}(t)\in S_{func}(t)$ and $F(v)=f_v(t)\in S_{func}(t)$ hold
$$
F(uv)=f_{uv}(t)=\theta_u(f_u(t))=\theta_u(F(u)),~ F(uv)=f_{uv}(t)=\eta_v(f_v(t))=\eta_v(F(v))
$$ and $\theta_v(F(u))=\eta_u(F(v))$, as well as $F(N(t))\cup E(N(t))=S_{func}(t)$.

Conversely, each pair of functions $f_i(t)$ and $f_j(t)$ of $S_{func}(t)$ holding $f_i=\theta_{i,j}(f_j)$ corresponds an edge $xy\in E(N(t))$, such that $F(x)=f_i(t)$ and $F(y)=f_j(t)$, or $F(x)=f_i(t)$ and $F(xy)=f_{j}(t)$.

Then we call $N(t)$ \emph{reasoning graph network} based on the reasoning function set $S_{func}(t)$.\qqed
\end{defn}

\begin{defn} \label{defn:hypergraph-data-sequence}
$^*$ For a \emph{thing-data set} $D_a(t)=\big \{D^1_{a}(t),D^2_{a}(t),\dots ,D^{m(t)}_{a}(t)\big \}$, we have a hypergraph set $\mathcal{E}(D^2_a(t))=\{\mathcal{E}_i(t): i\in [1,n(t)]\}$, where each hyperedge set $\mathcal{E}_i(t)$ holds $\bigcup _{e\in \mathcal{E}_i(t)}e=D_a(t)$. Suppose that each graph $G_i(t)$ admits a total set-coloring $\theta _i:V(G_i(t))\cup E(G_i(t))\rightarrow \mathcal{E}_i(t)$, such that each edge $uv\in E(G_i(t))$ holds $\theta_i(u)\cap \theta_i(v)\neq \emptyset $ and $\theta(u)_i\cap \theta_i(v)\subseteq \theta_i(uv)$. Conversely, if there are hyperedges $e,e\,'\in \mathcal{E}_i(t)$ holding $e\cap e\,'\neq \emptyset $, then there is an edge $xy\in E(G_i(t))$ holds $\theta_i(x)=e$ and $\theta_i(y)=e\,'$, as well as $\theta_i(x)\cap \theta_i(y)\subseteq \theta_i(xy)$. The sequence $\{G_i(t)\}^{n(t)}_{i=1}$ is called \emph{hypergraph data-functional network sequence}.\qqed
\end{defn}

\begin{defn} \label{defn:coloring-define-graph-network22}
$^*$ Let $S_{func}(t)=\{f_1(t),f_2(t),\dots ,f_m(t)\}$ be a dynamic function set for $t\in [a,b]$. A hyperedge set $\mathcal{E}\in \mathcal{E}(S^2_{func}(t))$ forms a vertex-intersected graph $H$ of a hypergraph $\mathcal{H}_{yper}=(S_{func}(t),\mathcal{E})$ admitting a set-coloring $h:V(H)\rightarrow \mathcal{E}$, such that each edge $uv\in E(H)$ corresponds a hyperedge $h(u)=e_u\in \mathcal{E}$ and another hyperedge $h(v)=e_v\in \mathcal{E}$ and there is a \emph{relational reasoning transformation} $O_{uv}$ hold $f_k(t)=O_{uv}[f_{i_1}(t), f_{i_2}(t),\dots ,f_{i_a}(t)]$ with $i_a\geq 1$, where $f_k(t)\in e_v$ and $f_{i_j}(t)\in e_u$ for each $i_j\in [i_1,i_a]$.

If there are two subsets $e_x\in \mathcal{E}$ and $e_y\in \mathcal{E}$, such that $f_y(t)\in e_y$ and hold
$$f_y(t)=O_{xy}[f_{x_1}(t),f_{x_2}(t),\dots ,f_{x_b}(t)],~x_j\in [x_1,x_b]
$$ then there is an edge $xy\in E(H(t))$ holds $h(x)=e_x$ and $h(y)=e_y$. We say that the set $\mathcal{E}\in \mathcal{E}(S^2_{func}(t))$ determines \emph{vertex-relational reasoning graphs} of a hypergraph $\mathcal{H}_{yper}=(S_{func}(t),\mathcal{E}(t))$.

Suppose that $\mathcal{E}(S^2_{func}(t))=\{\mathcal{E}_k(t): k\in [1,n_{f}(t)]\}$ with $t\in [a,b]$, where each hyperedge set $\mathcal{E}_k(t)$ holds $\bigcup _{e\in \mathcal{E}_k(t)}e=S_{func}(t)$. We get the \emph{vertex-relational reasoning sequence} $\{H_k(t)\}^{n_{f}(t)}_{k=1}$.\qqed
\end{defn}

\begin{rem}\label{rem:333333}
If the vertex number $|V(N(t))|>m=|S_{func}(t)|$ in Definition \ref{defn:coloring-define-graph-network}, then there exists a proper subnetwork $H(t)$ of the network $N(t)$ holding the vertex color number $|F(V(H(t)))|=|V(H(t))|=m=|S_{func}(t)|$, and moreover we say that the proper subnetwork $H(t)$ is an \emph{adaptation graph} of the function set $S_{func}(t)$, denoted as $H(t)\sim S_{func}(t)$. Clearly, if there are two adaptation graphs $H(t)\sim S_{func}(t)$ and $T(t)\sim S_{func}(t)$, then we claim that $H(t)\cong T(t)$.

In Definition \ref{defn:coloring-define-graph-network22}, we get a one-vs-more relational reasoning $f_k(t)=O_{uv}[f_{i_1}(t), f_{i_2}(t),\dots ,f_{i_a}(t)]$ about the relational reasoning between functions of the dynamic function set $S_{func}(t)$.\qqed
\end{rem}

\begin{example}\label{exa:8888888888}
In the subsection $C_{olor}$-hypergraphs, a graph $G$ admits colorings $f_1,f_2,\dots ,f_n$, and each coloring $f_i$ corresponds to another coloring $f_j$ with $i\neq j$ such that there is a transformation $\theta_{i,j}$ holding $f_j=\theta_{i,j}(f_i)$. So, we have a non-dynamic function set $S_{func}=\{f_1,f_2,\dots ,f_n\}$ for a graph network $N(t)$ defined in Definition \ref{defn:coloring-define-graph-network}.

In Definition \ref{defn:graphs-vs-total-graphs}, a total coloring $f_t$ of a graph $G$ becomes a (proper) vertex coloring $g_v$ of its own total graph $T(G)$, that is, $f_t\sim g_v$. If a non-dynamic function set $S^T_{func}=\{f_{t,1},f_{t,2},\dots ,f_{t,n}\}$ is made by the total colorings $f_{t,1},f_{t,2},\dots ,f_{t,n}$ of a graph $G$, then $S^T_{func}$ corresponds to another non-dynamic function set $S^v_{func}=\{g_{v,1},g_{v,2},\dots ,v_{t,n}\}$, in which each $g_{v,i}$ is a (proper) vertex coloring of the total graph $T(G)$. Thereby, this non-dynamic function set $S^T_{func}$ forms a graph network, so does $S^v_{func}$.\qqed
\end{example}

\subsubsection{Dynamic networks and active subnetworks for graph networks}

With the help of techniques of topology code theory, we try to approaches, or indirectly studies the functionality of graph networks proposed by \cite{Peter-Battaglia-Jessica-et-al-arXiv-2018}.

\textbf{A. }\textbf{Graph network framework and graph network block.} By means of Definition \ref{defn:pan-total-coloring-thing-set}, a graph $\phi(t)$ for each time step $t\in [a,b]$ admits a \emph{total operation-coloring} $F_t:V(\phi(t))\cup E(\phi(t))\rightarrow S_{thing}$ if $F(uv)=F(u)[\bullet_{W}]F(v)$ for each edge $uv \in E(\phi(t))$, where $[\bullet_{W}]$ is a $W$-constraint operation on a \emph{thing set} $S_{thing}$, and the graph $\phi(t)$ is called \emph{topological encoding graph} (topen-graph) at each time step $t\in [a,b]$. The topen-graph $\phi(t)$ is the graph network framework, and each topen-graph $\pi_{\varphi}(t)$ admits a $W$-constraint total coloring $h_i$, and operation-graph homomorphism $\pi_{\varphi}(t)\rightarrow _{oper}\phi(t)$, where $\varphi\in \{data, func, adapt, gener\}$ as follows:
\begin{asparaenum}[\textbf{\textrm{GNs}}-1.]
\item the topen-graph $\pi_{data}(t)$ is a data graph block used for assigning values to vertices of $\phi(t)$;

\item the topen-graph $\pi_{func}(t)$ is a function graph block used for the one-vs-one or more-vs-one relational reasoning of $\phi(t)$ (Ref. Definition \ref{defn:coloring-define-graph-network22});

\item the topen-graph $\pi_{adapt}(t)$ is an adjusting structure graph block used for the adaptive adjustment of $\phi(t)$;

\item the topen-graph $\pi_{gener}(t)$ is a generalized graph block used for the combinatorial generalization of $\phi(t)$.
\end{asparaenum}

Here, two or more graph networks can be stacked together (each network is isomorphic, but they represent different information), allowing us to construct multi-layer graph networks and fit complex computational processes.

\textbf{B. }\textbf{Dynamic networks and active subnetworks.} At each time step $t\in [a,b]$, a dynamic network $N(t)$ has its own vertex set $V(N(t))=V_+(t)\cup V_0(t)$, each node $x\in V_+(t)$ is active, each node $y\in V_0(t)$ is sleeping, and its own edge set $E(N(t))=E_+(t)\cup E_0(t)\cup E_{0+}(t)$, each edge $uv\in E_+(t)$ has its own ends $u,v\in V_+(t)$, each edge $st\in E_0(t)$ has its own ends $s,t\in V_o(t)$, $E_{0+}(t)$ contains those edges with one end in $V_+(t)$ and anther end in $V_0(t)$.

After doing a series of \emph{operation-graph homomorphisms} (\emph{weighted graph homomorphisms})
\begin{equation}\label{eqa:operation-graph-homomorphisms-GNQs}
\pi_{\varphi}(t)\rightarrow_{oper} N(t),\quad \varphi\in \{data, func, adapt, gener\},~t\in [a,b]
\end{equation} including adjusting structures for adaptivity, assigning data values to vertices and edges, realizing relational reasoning (Ref. Definition \ref{defn:coloring-define-graph-network} and Definition \ref{defn:coloring-define-graph-network22}), doing combinatorial generalization, so as to success \textbf{Func-1} in Remark \ref{rem:relational-reasoning-generalization}, thereby, we get an active network $N_+(t)=(u(t),V_+(t),E_+(t))$, with the global attribute $u(t)$ which is the composition of multiple-operation graph homomorphisms.

The properties of the graph network are updated in time steps during computation, including synchronous and asynchronous methods. When updating synchronously, the properties of all nodes in one time step are updated, while when updating asynchronously, only some nodes in one time step have their properties updated.

\textbf{C. }\textbf{Construction of topological structures.} We are able to construct complex topological structures $N(t)$ from simple building blocks $G_k(t)$ by various graph operations of topology code theory, such as $N(t)=\big [\bullet_{W}\big ]^m_{k=1}G_k(t)$, $N(t)=\big [\bullet_{W}\big ]^m_{k=1}a_kG_k(t)$, so as to implement \textbf{Func-2} in Remark \ref{rem:relational-reasoning-generalization}.

\textbf{D. }\textbf{Computation.} Notice that each topen-graph $G$ can be imputed into computer by its own adjacent matrix and its own Topcode-matrix, so as to achieve \textbf{Func-3} in Remark \ref{rem:relational-reasoning-generalization}.

\subsubsection{Discussing questions from graph networks}

We have noticed that no one of hyper-node, hyper-edge and hyergraph was mentioned in \cite{Peter-Battaglia-Jessica-et-al-arXiv-2018}, however, we have used the hypergraph $\mathcal{H}_{yper}=(S_{func},\mathcal{E})$ defined in Definition \ref{defn:coloring-define-graph-network22} for building up the active network $N_+(t)=(u(t),V_+(t),E_+(t))$ with the global attribute $u(t)$ at each time step $t\in [a,b]$ (Ref. Definition \ref{defn:graph-network-DeepMind22} and Eq.(\ref{eqa:operation-graph-homomorphisms-GNQs})).

\begin{example}\label{exa:8888888888}
According to Definition \ref{defn:coloring-define-graph-network}, the $n$ discrete entities in the real world can be related to each other according to a certain rule $P$, naturally, they themselves form a $(n,q)$-graph $G\subseteq K_n$ with this rule $P$, so we answer fully \textbf{GNQ}-\ref{GNQ:graphs-come-from11} of Remark \ref{rem:graph-network-open-questions}.

We can provide a part solution for \textbf{GNQ}-\ref{GNQ:graphs-come-from22} of Remark \ref{rem:graph-network-open-questions} as follows: Many underlying graph structures of graph networks $N(t)$ are \emph{scale-free networks} (\emph{scale-free vertex-intersected networks}) obeying a \emph{degree distribution} $P(k)\sim k^{-\gamma}$ and a \emph{cumulative distribution} $P_{cum}(k)\sim k^{1-\gamma}$ with $2<\gamma<3$ hold Eq.(\ref{eqa:Barabasi-hyperedge-degree-distribution}), since many dynamic networks are scale-free \cite{Barabasi-Albert1999}.

Moreover, since a scale-free network $N(t)$ is sparse, then the edge number $e_{net}(t)$ and vertex number $v_{net}(t)$ of the scale-free network $N(t)$ hold Eq.(\ref{eqa:sparse-edge-vertex-numbers}).
\end{example}

\begin{example}\label{exa:8888888888}
\textbf{Scale-free graph network lattice. }We can design the scale-free graph network lattice based on non-multi-edge vertex-coinciding operation.

\textbf{Input: } A \emph{scale-free graph network block base} is $\textbf{\textrm{S}}_{sf}(t)=\{S_1(t)$, $S_2(t)$, $\dots$, $S_m(t)\}$ ($t\in [a,b]$) with $S_i(t)\not \subset S_j(t)$ and $S_i(t)\not \cong S_j(t)$ if $i\neq j$, where each graph $S_k(t)$ is a non-data connected scale-free graph network block.

\textbf{Output: } A non-data scale-free graph network lattice $\textbf{\textrm{L}}_{sf}\big (Z^0[\bullet_{\pi}]\textbf{\textrm{S}}_{sf}(t)\big )$ shown in Eq.(\ref{eqa:v-co-oper-scale-free-net-graph-lattice}).

\textbf{Step-1} Suppose that $Q_1(t),Q_2(t),\dots, Q_I(t)$ is a permutation of $I$ scale-free graph network blocks $e_1S_1(t),e_2S_2(t),\dots, e_mS_m(t)$, where $I=\sum^m_{i=1} e_i\geq 1$.

\textbf{Step-2} Do the non-common neighbor vertex-coinciding operation defined in Definition \ref{defn:vertex-split-coinciding-operations} to the permutation $Q_1(t),Q_2(t),\dots, Q_I(t)$. We vertex-coincide a vertex $u_{i,j}\in V(Q_i(t))$ having higher degree with a vertex $v_{k,j}\in V(Q_k(t))$ having higher degree
into a vertex $w_{i,k,j}=u_{i,j}\bullet v_{k,j}$ for $j\in [1,k_{i,j}]$, such that the resultant graph $Q_i(t)[\bullet_{k_{i,j}}]Q_j(t)$ is a non-multi-edge scale-free graph network block holding
\begin{equation}\label{eqa:555555}
\big |E\big (Q_i(t)[\bullet_{k_{i,j}}]Q_j(t)\big )\big |=|E(Q_i(t))|+|E(Q_j(t))|
\end{equation}

Thereby, we get a non-multi-edge scale-free graph network block $R_1(t)=Q_1(t)[\bullet_{k_1}]Q_2(t)$, such that $|E(R_1(t))|=|E(Q_1(t))|+|E(Q_2(t))|$ after doing the non-common neighbor vertex-coinciding operation. Again, we get another non-multi-edge scale-free graph network block $R_2(t)=R_1(t)[\bullet_{k_2}]Q_3(t)$ holding $|E(R_2(t))|=|E(R_1(t))|+|E(Q_3)|$. Go on in this way, we have scale-free graph network blocks
\begin{equation}\label{eqa:555555}
R_i(t)=R_{i-1}(t)[\bullet_{k_{i}}]Q_{i+1}(t),~|E(R_i(t))|=|E(R_{i-1}(t))|+|E(Q_{i+1}(t))|,~i\in [1,I-1]
\end{equation} as well as $R_{0}(t)=Q_1(t)$. For simplicity's sake, we write
\begin{equation}\label{eqa:block-base-v-coinciding-oper11}
{
\begin{split}
R_{I-1}(t)&=R_{I-2}(t)[\bullet_{k_{I-1}}]Q_{I}(t)\\
&=Q_1(t)[\bullet_{k_1}]Q_2(t)[\bullet_{k_2}]\cdots[\bullet_{k_{I-1}}]Q_I(t)\\
&=\big [\bullet_{\pi}\big ]^m_{k=1}e_kS_k(t),~t\in [a,b]
\end{split}}
\end{equation} where $\pi=(k_1,k_2,\dots, k_{I-1})$.

\textbf{Step-3} The following set of scale-free graph network blocks
\begin{equation}\label{eqa:v-co-oper-scale-free-net-graph-lattice}
\textbf{\textrm{L}}_{sf}\big (Z^0[\bullet_{\pi}]\textbf{\textrm{S}}_{sf}(t)\big )=\Big \{ \big [\bullet_{\pi}\big ]^m_{k=1}e_kS_k(t):~e_k\in Z^0,S_k(t)\in \textbf{\textrm{S}}_{sf}(t)\Big \},~\sum^m_{k=1} e_k\geq 1,~t\in [a,b]
\end{equation} is called \emph{non-data scale-free graph network lattice} based on the scale-free graph network block base $\textbf{\textrm{S}}_{sf}(t)$, or abbreviated as \emph{scale-free graph network lattice}.

It has been confirmed: The use of \emph{linear preferential attachment} method will result in a scale-free network, and moreover Eq.(\ref{eqa:block-base-v-coinciding-oper11}) is just a linear preferential attachment. The preferential attachment plays a leading role in the development of scale-free networks.
\end{example}

\begin{problem}\label{question:answer-GNQ-1234}
\textbf{Is} a graph network defined in Definition \ref{defn:graph-network-DeepMind} a hypernetwork, or a scale-free vertex-intersected network? If it is so, some problems proposed in \textbf{GNQ-$k$} of Remark \ref{rem:graph-network-open-questions} for $k\in [1,4]$ can be partly answered.
\end{problem}

About Problem \ref{question:answer-GNQ-1234}, we have partly answered \textbf{GNQ-3} of Remark \ref{rem:graph-network-open-questions} by the technique of the operation-graph homomorphisms $\pi_{\varphi}(t)\rightarrow_{oper} N(t)$ in Eq.(\ref{eqa:operation-graph-homomorphisms-GNQs}).

\begin{problem}\label{question:444444}
Since graph networks support relational reasoning and combinatorial generalization, becoming more complex, interpretable, and flexible reasoning patterns in \cite{Peter-Battaglia-Jessica-et-al-arXiv-2018}. \textbf{Can} hypernetworks, or scale-free vertex-intersected networks be used for relational reasoning and combinatorial generalization?
\end{problem}

\begin{rem}\label{rem:333333}
Replacing ``graph'' with ``hypergraph'' yields possible hypergraph neural networks, hypergraph convolutional networks, hypergraph attention networks, hypergraph embedding, hypergraph generation networks, hypergraph spatiotemporal networks, etc.\qqed
\end{rem}

\textbf{What} is the accurate definition of a graph network? We did not find it in \cite{Peter-Battaglia-Jessica-et-al-arXiv-2018}.

The team led by Yu Shilun from Tsinghua University has summarized the many advances in deep learning graph processing into five sub-directions:

1) graph convolution networks;

2) graph attention networks;

3) graph embedding;

4) graph generative networks; and

5) graph spatialtemporal networks.

The DeepMind team focuses on solving the last four directions in five sub directions, namely graph attention network, graph embedding, graph generation network, and graph spatiotemporal network. They integrated the results of these four directions into a unified framework and named it as \emph{graph network}.

Based on numerous facts, we provide the definition of a graph network as follows:

\begin{defn} \label{defn:graph-network-definition}
$^*$ A \emph{graph network} is a network that can organically associate research objects according to attribute rules by using a topological structure $G$ and a mathematical (attribute) constraint set $R_{est}(c_1,c_2,\dots, c_m)$, such that each research object is a vertex of the topological structure $G$, the attributes of two ends $x$ and $y$ of each edge $xy\in E(G)$ and the attribute of the edge $xy$ hold the constraint set $R_{est}(c_1,c_2,\dots, c_m)$. \qqed
\end{defn}

About definition \ref{defn:graph-network-definition}, the specific manifestation in topology code theory is a graph pan-coloring, as defined in the definition \ref{defn:pan-all-thing-coloringss}.

\begin{defn} \label{defn:pan-all-thing-coloringss}
\cite{Bing-et-al-arXiv-asymmetric-4520331} A graph $G$ admits a coloring $F:S\rightarrow X$, where $S\subseteq V(G)\cup E(G)$, such that $F$ holds a group of $W_1$-constraint, $W_2$-constraint, $\dots$, $W_n$-constraint, denoted as $\{W_i\}^n_{i=1}$-constraint with $n\geq 1$. Then
\begin{asparaenum}[\textbf{\textrm{Pancolor}}-1.]
\item $F$ is called \emph{$\{W_i\}^n_{i=1}$-constraint string-coloring} if $X$ is a set of strings.
\item $F$ is called \emph{$\{W_i\}^n_{i=1}$-constraint coloring-coloring} if $X$ is a set of colorings.
\item $F$ is called \emph{$\{W_i\}^n_{i=1}$-constraint set-coloring} if $X$ is a set of sets.
\item $F$ is called \emph{$\{W_i\}^n_{i=1}$-constraint vector-coloring} if $X$ is a set of vectors.
\item $F$ is called \emph{$\{W_i\}^n_{i=1}$-constraint matrix-coloring} if $X$ is a set of matrices.
\item $F$ is called \emph{$\{W_i\}^n_{i=1}$-constraint graph-coloring} if $X$ is a set of graphs.
\item $F$ is called \emph{$\{W_i\}^n_{i=1}$-constraint string-lattice coloring} if $X$ is a string lattice.
\item $F$ is called \emph{$\{W_i\}^n_{i=1}$-constraint graphic-lattice coloring} if $X$ is a graphic lattice.
\item $F$ is called \emph{$\{W_i\}^n_{i=1}$-constraint graph set-coloring} if $X$ is a set of graph sets.
\item $F$ is called \emph{$\{W_i\}^n_{i=1}$-constraint group-coloring} if $X$ is an every-zero graphic group.
\item $F$ is called \emph{$\{W_i\}^n_{i=1}$-constraint string-group coloring} if $X$ is an every-zero string group.
\item $F$ is called \emph{$\{W_i\}^n_{i=1}$-constraint graphic-group coloring} if $X$ is an every-zero graphic group.
\item $F$ is called \emph{$\{W_i\}^n_{i=1}$-constraint thing-coloring} if $X$ is a set of things having a particular property or a group of particular properties.\qqed
\end{asparaenum}
\end{defn}

We have summarized the research points of the literature \cite{Peter-Battaglia-Jessica-et-al-arXiv-2018} as follows:
\begin{asparaenum}[\textbf{Point-}1. ]
\item \textbf{Graph networks attempt to unify various networks in deep learning.} Graph networks are the generalization of graph neural networks (GNN) in deep learning theory and probabilistic graphical model (PGM). A graph network is composed of graph network blocks and has a flexible topology structure, and it can be transformed into various forms of connectionism models including feedforward neural networks (FNN) and recursive neural network (RNN) \emph{etc}.

\item \textbf{Undirected graph networks and directed graph networks of graph networks.} The properties of nodes and edges in a graph network are the same as the graph structure, which can be divided into directed graph and undirected graph. The examples of directed graphs are recurrent neural networks, the examples of directed graphs are Hopfield neural networks, Markov networks \emph{etc}. More general graph networks are suitable for processing data with graph structures, such as knowledge graphs, social networks, molecular networks, \emph{etc}.

\item \textbf{Graph network framework.} A graph network framework based on graph network blocks defines a class of functions for relational reasoning on graph structure representations. The graph network framework summarizes and extends various graph neural networks, MPNN, and NLNN methods, and supports building complex architectures from simple building blocks.

\item \textbf{Relational reasoning in intelligent evolution.} The ability of humans to summarize combinations mainly depends on their \emph{cognitive mechanisms} for expressing structures and reasoning relationships, structured knowledge, and structured behavior. The graph network precisely reflects the principle of combinatorial induction, which constructs new inferences, predictions, and behaviors from known building blocks.

\item \textbf{Inductive bias.} In model learning, inductive biases make parameters tend to adjust to a certain state, and the model tends to learn to a certain standard format.

\item \textbf{Combinatorial generalization.} Each node in the graph network has internal and system states, called \emph{attribute}. The attributes of the graph network are updated in time-steps during computation, including synchronous and asynchronous methods. When updating synchronously, the attributes of all nodes in one time step are updated, while when updating asynchronously, only some nodes in one time step have their attributes updated.
\item \textbf{Structured intelligent technology.} Graph networks are suitable for processing data with graph structures. To evolve AI from ``perceptual intelligence'' to ``cognitive intelligence'', graph networks advocate for actionable structured knowledge, generated structured behavior, structured representations, and structured computing.
\item \textbf{Artificial intelligence implementation and intelligent evolution of graph networks.} Combination generalization is the primary task for artificial intelligence to achieve similar abilities to humans, and the structured representation and computation are the key to achieving this goal, and the key to achieving this goal is to represent data in a structured manner, as well as structured computation.
\end{asparaenum}

\section{Conclusion}

The research works of this article have the following key-points:
\begin{asparaenum}[\textbf{\textrm{Point}}-1 ]
\item Try finding new objects, ideas, problems, and theories for topology code theory, such as, proposing edge-hamiltonian graph; constructing hypergraphs; exploring hypernetwork; considering some problems similarly with Kelly-Ulam's Reconstruction Conjecture: edge-hamiltonian problem, 4-colorable problem on maximal planar graphs, hypergraph isomorphism conjecture.
\item Algebraic methods, including topological knowledge, topological action, topological expression and topological computation, are used in researching hypergraphs and their vertex-intersected graphs, and enrich

\qquad (i) topological groups including every-zero graphic group, every-zero Topcode-matrix group, every-zero parameterized Topcode-matrix group, every-zero adjacent-matrix group, every-zero topological string group, every-zero topen-graph set group, every-zero mixed-graphic group, every-zero hyperedge-set group, every-zero graphic group based on hypergraph, every-zero hypergraph group and pan-group.

\qquad (ii) tree-graph lattice, edge-coincided vertex-intersected graph lattice, hyperedge-coincided hypergraph lattice, vertex-coincided vertex-intersected graph lattice, $K_{tree}$-spanning lattice, mixed vertex-intersected graph lattice, operation vertex-intersected network lattice, scale-free operation vertex-intersected network meta-lattice and hypernetwork meta-lattice, edge-hamiltonian lattice and maximal planar graphic lattice.
\item Topological groups in topology code theory are applied to network overall topology encryption.
\item Try solving some difficult problems by Topological lattices of topology code theory.
\item Applying the technology of topology code theory to the investigation of graph networks from DeepMind and GoogleBrain is a new attempt.

\item Although a hypergraph is a subset system of a finite set, however, our research findings confirm: A key signature of human intelligence is the ability to make ``infinite use of finite means'' (Humboldt, 1836; Chomsky, 1965), in which a small set of elements (such as words) can be productively composed in limitless ways (such as into new sentences) in \cite{Peter-Battaglia-Jessica-et-al-arXiv-2018}.
\item Inspired by the functions of graph networks proposed by DeepMind and GoogleBrain, and ``\textbf{how} to adaptively modify graph structures during the course of computation in \textbf{GNQ}-\ref{GNQ:modify-graph-structures} of Remark \ref{rem:graph-network-open-questions}?'', we \textbf{raise} the following questions:

\qquad What is an intelligent (dynamic) graph network?

\qquad What is an intelligent Internet of Things?

\qquad What is an intelligent dynamic network?

\qquad What is an intelligent dynamic hypernetwork?

since AI technology will surpass ordinary researchers in many fields of scientific research in the near future.
\end{asparaenum}

\vskip 0.2cm

\textbf{We hope:} The FCGSC-problem, the Hypergraph-string problem and the PCTSMGHS-string problem proposed in this article can resist attacks equipped with AI technology and quantum computing in the future.

\textbf{We know:} There is a lack of applications of hypergraphs in other fields, such as graph theory, information security, privacy protection. And topology code theory requires more profound theories and powerful applications, and develops topological knowledge, topological representation, topological behavior and topological computation.

\textbf{We hope:} The vertex-intersected graphs mentioned here can be applied to real situations, such as network security, privacy protection, as well as anti-quantum computing in asymmetric cryptography. Predictably, hypergraph theory will be an important application in the future resisting AI attacks equipped quantum computer, although the research achievements and application reports of hypergraphs are far less than that of popular graphs.

\textbf{We think:} In future research, metaverse will be an important scene for the theoretical and application research of hypergraphs and hypernetworks. It can also be said that the development of Metaverse will greatly promote the theoretical and application research of hypergraphs and hypernetworks.

\textbf{We imagine:} It is not difficult to imagine that people use quantum computers, quantum phones, and quantum cryptography for communication and work in a short period of time in future bioelectronic devices. The human body can embed bioelectronic computer chips and other bioelectronic devices, and the DNA of the human body stores information data and provides power for embedded bioelectronic devices.

\textbf{We declare:} Part of contents mentioned in this article have been applied the invention patents of CHINA.

\emph{Graphs are codes, and codes are graphs. There's endless stuff to explore, it's interesting and beautiful, and a source of a lot of great questions} \cite{how-big-data-graph-theory-20210819}.

\section*{Acknowledgment}

The author, \emph{Bing Yao}, was supported by the National Natural Science Foundation of China under grants No. 61163054, No. 61363060 and No. 61662066.

\newpage

\textbf{Appendix A}

\begin{center}
\textbf{Table-1.} The numbers of trees of $p$ vertices \cite{Harary-Palmer-1973}
\end{center}
\begin{center}
\begin{tabular}{c|ll}
$p$&$t_p$&$T_p$\\
\hline
7&11&48\\
8&23&115\\
9&47&286\\
10&106&719\\
11&235&1,842\\
12&551&4,766\\
13&1,301&12,486\\
14&3,159&32,973\\
15&7,741&87,811\\
16&19,320&235,381\\
\end{tabular}\qquad
\begin{tabular}{c|ll}
$p$&$t_p$&$T_p$\\
\hline
17&48,629&634,847\\
18&123,867&1,721,159\\
19&317,955&4,688,676\\
20&823,065&12,826,228\\
21&2,144,505&35,221,832\\
22&5,623,756&97,055,181\\
23&14,828,074&268,282,855\\
24&39,299,897&743,724,984\\
25&104,636,890&2,067,174,645\\
26&279,793,450&5,759,636,510\\
\end{tabular}
\end{center}
where $t_p$ is the number of trees of $p$ vertices, and $T_p$ is the number of rooted trees of $p$ vertices.


\textbf{Appendix B}

\begin{figure}[h]
\centering
\includegraphics[width=16.4cm]{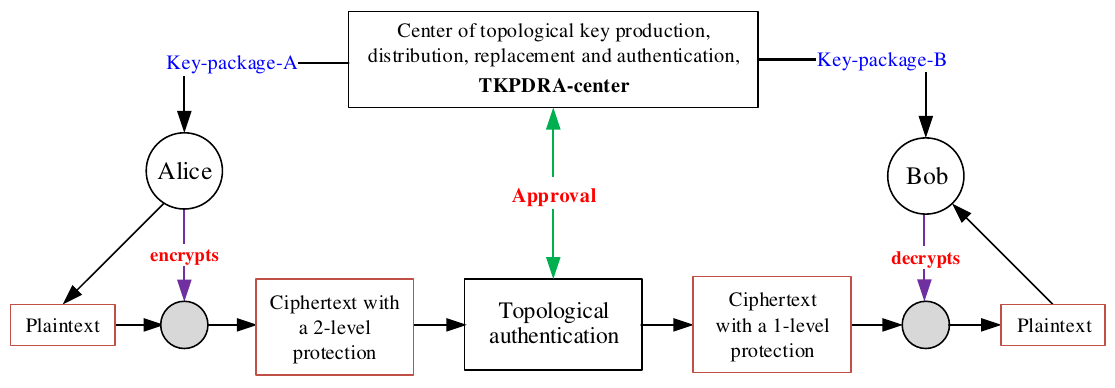}\\
\caption{\label{fig:string-graph-group-center-11}{\small The single topological authentication mechanism of TKPDRA-center, cited from \cite{Bing-et-al-arXiv-asymmetric-4520331}.}}
\end{figure}

\begin{figure}[h]
\centering
\includegraphics[width=16.4cm]{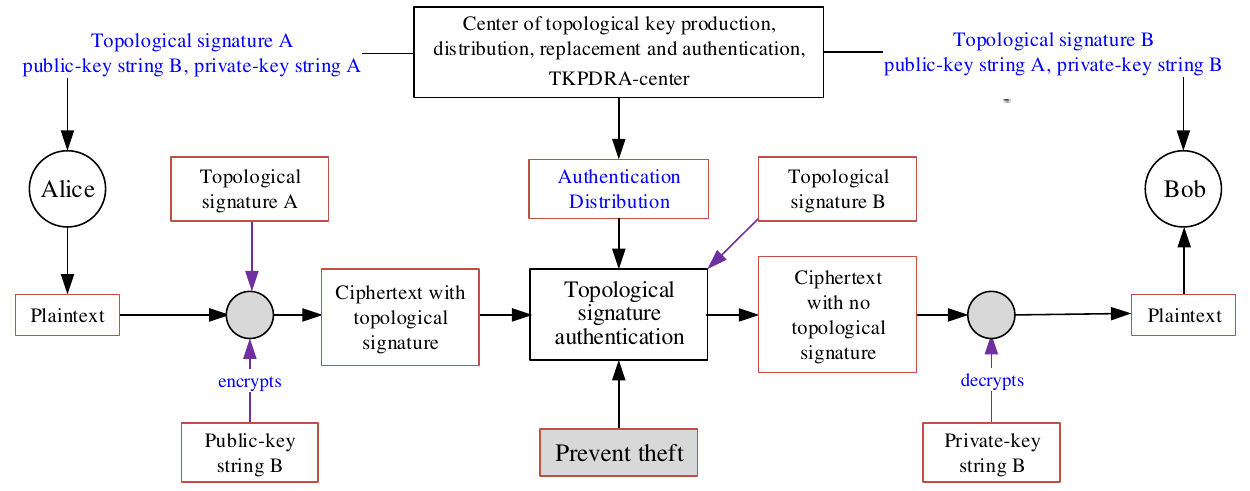}\\
\caption{\label{fig:Topological-signature-22}{\small The center of key-encryption distribution for local area networks and communities, cited from \cite{Bing-et-al-arXiv-asymmetric-4520331}.}}
\end{figure}

\begin{figure}[h]
\centering
\includegraphics[width=16.4cm]{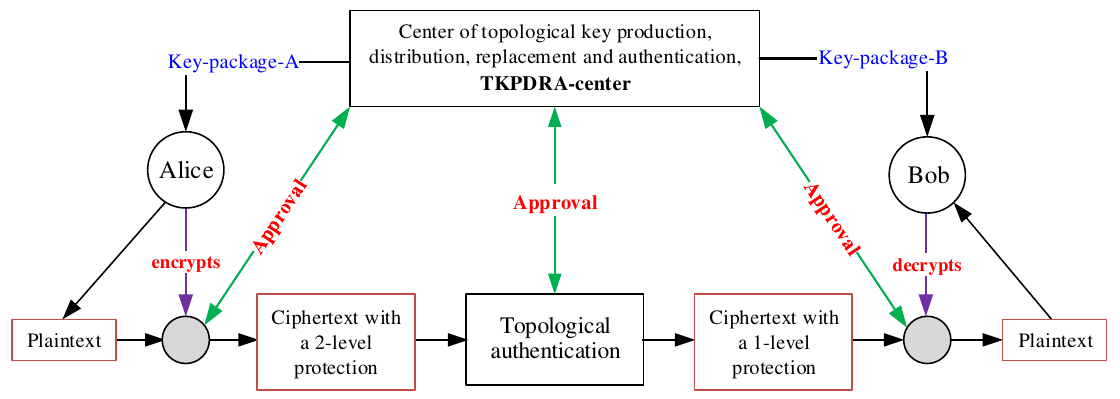}\\
\caption{\label{fig:string-graph-group-center-2}{\small The multiple topological authentication mechanism of TKPDRA-center, cited from \cite{Bing-et-al-arXiv-asymmetric-4520331}.}}
\end{figure}


\end{document}